\documentclass[11pt,a4paper,austrian,titlepage,chapterprefix,headsepline,parskip,pdftex,twoside]{scrreprt}

\makeatletter
\renewcommand{\paragraph}{\@startsection
   {paragraph} 
   {4} 
   {0mm} 
   {-\baselineskip} 
   {0.1\baselineskip} 
   {\normalfont\normalsize\bfseries}} 
\makeatother 

\makeatletter 
\renewcommand{\subparagraph}{\@startsection
   {subparagraph} 
   {5} 
   {0mm} 
   {-\baselineskip} 
   {0.1\baselineskip} 
   {\normalfont\normalsize\bfseries}} 
\makeatother 

\usepackage{setspace}
\usepackage[pdftex]{graphicx}

\usepackage[pdftex]{color}

\usepackage[colorlinks=true,
    linkcolor=black,
    citecolor=black,
    urlcolor=black,
    breaklinks=true,
    bookmarksnumbered=true,
    hypertexnames=false,
    pdfpagemode=UseOutlines,
    pdfview=FitH,
    plainpages=false,
    pdfpagelabels,
    bookmarks=true,
    linktocpage=true]{hyperref}

\hypersetup{pdfauthor={Jakob Ablinger},
    pdftitle={Computer Algebra Algorithms for Special Functions in Particle Physics},
    pdfsubject={Dissertation},
    pdfkeywords={},
    pdfcreator={pdfLaTeX with hyperref (\today})}

\setkomafont{chapter}{\Huge}

\usepackage[automark]{scrpage2}

\clearscrheadings \clearscrplain \clearscrheadfoot
\ohead{\pagemark}
\ihead{\headmark}
\usepackage{array}

\addtokomafont{chapter}{\sffamily}
\addtokomafont{sectioning}{\rmfamily}

\usepackage[english]{babel}

\usepackage[latin1]{inputenc}

\usepackage[T1]{fontenc}
\usepackage{ae}

\usepackage{url}

\usepackage{amsmath, amsthm, amssymb}
\newtheorem{thm}{Theorem}[section]

\newtheorem{lemma}[thm]{Lemma}
\newtheorem{prop}[thm]{Proposition}

\newtheorem{definition}[thm]{Definition}
\newtheorem{example}[thm]{Example}

\newtheorem{remark}[thm]{Remark}
\newtheorem{notation}[thm]{Notation}

\newcommand{\ve}[1]{\textit{\textbf{#1}}}
\newcommand{\abs}[1]{\left|{#1}\right|}

\newcommand{\sign}[1]{\textnormal{sign}\hspace{-0.1em}\left(#1\right)}

\renewcommand{\H}[2]{\textnormal{H}_{#1}\hspace{-0.2em}\left(#2\right)}
\newcommand{\Hma}[2]{\textnormal{H}_{#1}\left(#2\right)}%H mit abstand
\newcommand{\Hs}[2]{\widehat{\textnormal{H}}_{#1}\hspace{-0.2em}\left(#2\right)}
\newcommand{\Hsma}[2]{\widehat{\textnormal{H}}_{#1}\left(#2\right)}
\newcommand{\M}[2]{\textnormal{M}\hspace{-0.2em}\left(#1,#2\right)}
\newcommand{\Mma}[2]{\textnormal{M}\left(#1,#2\right)}%M mit abstand
\newcommand{\Mp}[2]{\textnormal{M}^+\hspace{-0.2em}\left(#1,#2\right)}
\renewcommand{\S}[2]{\textnormal{S}_{#1}\hspace{-0.2em}\left(#2\right)}
\renewcommand{\SS}[3]{\textnormal{S}_{#1}\hspace{-0.2em}\left(#2;#3\right)}

\newcommand{\R}{\mathbb R}
\newcommand{\Q}{\mathbb Q}

\newcommand{\C}{\mathbb C}

\newcommand{\Z}{\mathbb Z}
\newcommand{\N}{\mathbb N}

\newcommand{\ds}{\displaystyle}
\newcommand{\Li}{{\rm Li}}

\newcommand{\ie}{i.e.,\ }
\newcommand{\eg}{e.g.,\ }
\renewcommand{\Re}{\operatorname{Re}}

\newenvironment{ProblemSpec}[1]{\noindent#1\\}{}
\newcommand{\SigmaP}{\texttt{Sigma}}

\newcommand{\ProblemRS}{\textsf{RS}}
\newcommand{\FLSR}{\textsf{FLSR}}
\newcommand{\ep}{\varepsilon}

\let\set\mathbb

\usepackage{rotating}
\newcommand{\shuffle}{\, \raisebox{1.2ex}[0mm][0mm]{\rotatebox{270}{$\exists$}} \,}
\usepackage{algorithm}
\usepackage{algorithmicx}
\usepackage{algpseudocode}
\usepackage{graphicx}
\usepackage{pdfpages}

\newcounter{mmacnt}
\def\restartmma{\setcounter{mmacnt}{0}}
\restartmma \catcode`|=\active
\def|#1|{\mathrm{#1}}
\catcode`|=12
\newenvironment{mma}{
 \par\smallskip
 \catcode`|=\active
 \parskip=0pt\parindent=0pt % locally
 \small
 \def\In##1\\{%
   \def\linebreak{\hfill\break\null\qquad}%
   \refstepcounter{mmacnt}
   \hangindent=2.5em\hangafter=0
   \leavevmode
   \llap{\tiny\sffamily In[\arabic{mmacnt}]:=\kern.5em}%
   \mathversion{bold}\footnotesize$\displaystyle##1$\normalsize
   \mathversion{normal}\par
 }%
 \def\Print##1\\{%
   \def\linebreak{\hfill\break}%
   \hangindent=2.5em\hangafter=0
   \leavevmode ##1\par}%
 \def\Out##1\\{%
   \def\linebreak{$\hfill\break\null\hfill$}%
   \kern\abovedisplayskip\par
   \hangindent=2.5em\hangafter=0
   \leavevmode
   \llap{\tiny\sffamily Out[\arabic{mmacnt}]=\kern.5em}
   \footnotesize$\displaystyle##1$\normalsize\hfill\null\par
   \kern\belowdisplayskip
 }%
 \def\Warning##1##2\\{%
   \def\linebreak{\hfill\break}%
   \hangindent=2.5em\hangafter=0
   \leavevmode
   {\scriptsize##1 : ##2}\par}%
}{%
 \par\smallskip
}

\usepackage{qtree}

\usepackage{color}
\usepackage{framed}

\newenvironment{fshaded}{%
\MakeFramed {\FrameRestore}
}%
{\endMakeFramed}

\newenvironment{fmma}[1]{\definecolor{shadecolor}{rgb}{1,1,1}%
\definecolor{framecolor}{rgb}{0,0,0}%
\begin{fshaded}\text{\bf HarmonicSums session.}\begin{mma}#1}{\end{mma}\end{fshaded}}

\newenvironment{fmma2}[1]{\definecolor{shadecolor}{rgb}{1,1,1}%
\definecolor{framecolor}{rgb}{0,0,0}%
\begin{fshaded}\text{\bf MultiIntegrate session.}\begin{mma}#1}{\end{mma}\end{fshaded}}

\usepackage{tikz}
\usetikzlibrary{matrix}

\allowdisplaybreaks[4]

\begin{document}

\includepdf{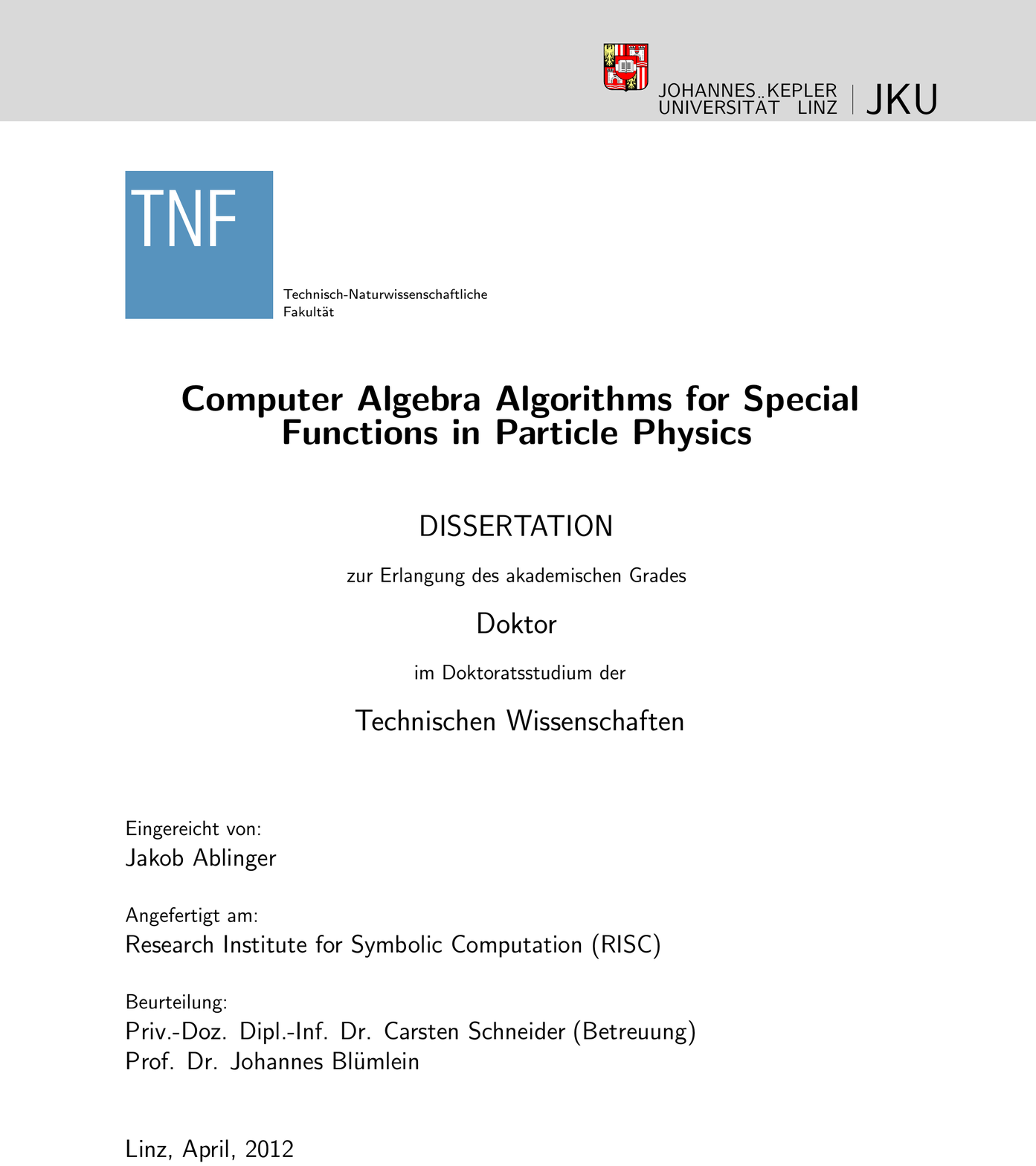}
\thispagestyle{empty}  %fuer beidseitige Version
\cleardoublepage  %fuer beidseitige Version
\pagenumbering{Roman}

\chapter*{Eidesstattliche Erkl\"arung}
\markright{Eidesstattliche Erkl\"arung}

Ich erkl\"are an Eides statt, dass ich die vorliegende Dissertation selbstst\"andig und ohne fremde Hilfe verfasst, 
andere als die angegebenen Quellen und Hilfsmittel nicht benutzt bzw.\ die w\"ortlich oder sinngem\"a{\ss} entnommenen Stellen als solche kenntlich gemacht habe.

Die vorliegende Dissertation ist mit dem elektronisch \"ubermittelten Textdokument identisch.

Linz, April 2012,

\hfill Jakob Ablinger

\chapter*{Kurzfassung}
\markright{Kurzfassung}

\hyphenation{
  zyclo-tom-ische
  ha-rmon-ischen
}

Diese Arbeit behandelt spezielle verschachtelte Objekte, die in massiven St\"orungsberechnungen h\"oherer Ordnung renormierbarer Quantenfeldtheorien auftreten. Einerseits bearbeiten wir
verschachtelte Summen, wie harmonische Summen und deren Verallgemeinerungen (S-Summen, zyclotomische harmonische Summen, zyclotomische S-Summen) und
andererseits arbeiten wir mit iterierten Integralen nach Poincar\'e und Chen, wie harmonischen Polylogarithmen und deren Verallgemeinerungen (multiple Polylogarithmen, zyclotomische
harmonische Polylogarithmen). Die iterierten Integrale sind \"uber die Mellin-Transformation (und Erweiterungen der Mellin-Transformation) mit den
verschachtelten Summen verkn\"upft und wir zeigen wie diese Transformation berechnet werden kann. Wir leiten algebraische und stukturelle Relationen zwischen
den verschachtelten Summen und zus\"atzlich zwischen den Werten der Summen bei Unendlich und den damit verbundenen Werten der iterierten Integrale ausgewertet bei
speziellen Konstanten her. Dar\"uber hinaus pr\"asentieren wir Algorithmen zur Berechnung der asymptotischen Entwicklung dieser Objekte und wir beschreiben
einen Algorithmus der bestimmte verschachtelte Summen in Ausdr\"ucke bestehend aus zyclotomische S-Summen umwandelt. Des Weiteren fassen wir die wichtigsten Funktionen
des Computeralgebra Pakets \ttfamily HarmonicSums \rmfamily zusammen, in welchem alle diese Algorithmen und Tranformationen impementiert sind.
Ferner pr\"asentieren wir Anwendungen des multivariaten Almkvist-Zeilberger Algorithmuses und Erweiterungen dieses Algorithmuses auf spezielle Typen von Feynman-Integralen und wir stellen das dazugh\"orige
Computeralgebra Paket \ttfamily MultiIntegrate \rmfamily vor. 

\chapter*{Abstract}
\markright{Abstract}

This work deals with special nested objects arising in massive higher order perturbative calculations in renormalizable quantum field theories.
On the one hand we work with nested sums such as harmonic sums and their generalizations (S-sums, cyclotomic harmonic sums, cyclotomic S-sums) and
on the other hand we treat iterated integrals of the Poincar\'e and Chen-type, such as harmonic polylogarithms and their generalizations (multiple polylogarithms, cyclotomic harmonic polylogarithms).
The iterated integrals are connected to the nested sums via (generalizations of) the Mellin-transformation and we show how this transformation can be computed.
We derive algebraic and structural relations between the nested sums as well as relations between the values of the sums at infinity and connected to it the values
of the iterated integrals evaluated at special constants. In addition we state algorithms to compute asymptotic expansions of these nested objects and we state an algorithm which
rewrites certain types of nested sums into expressions in terms of cyclotomic S-sums. Moreover we summarize the main functionality of the computer algebra package
\ttfamily HarmonicSums \rmfamily in which all these algorithms and transformations are implemented.
Furthermore, we present application of and enhancements of the multivariate Almkvist-Zeilberger algorithm to certain types of Feynman integrals and the corresponding
computer algebra package \ttfamily MultiIntegrate\rmfamily. 

\tableofcontents

\chapter*{Notations}
\markright{Notations}

  \begin{tabbing}
    \hspace*{3.5cm}\= \kill
    $\N$ \> $\N=\{1,2,\ldots\},$ natural numbers\\
    $\N_0$ \> $\N_0=\{0,1,2,\ldots\}$\\
    $\Z$ \> integers \\
    $\Z^*$ \> $\Z\setminus\{0\}$ \\
    $\Q$ \> rational numbers\\
    $\R$ \> real numbers\\
    $\R^*$ \> $\R\setminus\{0\}$\\
    $\S{a_1,a_2,\ldots}{n}$  \> harmonic sum; see page \pageref{HSdefHsum}\\
    $\S{\ve a}{\ve b; n}$ \> S-sum; see page \pageref{SSchapter}\\
    $\S{(a_1,b_1,c_1),\ldots}{n}$  \> cyclotomic harmonic sum; see page \pageref{CSdef}\\
    $\S{(a_1,b_1,c_1),\ldots}{\ve b; n}$\>  cyclotomic S-sum; see page \pageref{CSSdef}\\
    $\mathcal{S}(n)$ \>the set of polynomials in the harmonic sums; see page \pageref{S}\\
    $\mathcal{C}(n)$ \>the set of polynomials in the cyclotomic harmonic sums; see page \pageref{C}\\
    $\mathcal{CS}(n)$ \>the set of polynomials in the cyclotomic S-sums; see page \pageref{CS}\\
    $\H{m_1,m_2,\ldots}{x}$ \> harmonic/multiple polylogarithm; see pages \pageref{HShlogdef} and \pageref{SShlogdef}\\
    $\H{(m_1,n_1),(m_2,n_2),\ldots}{x}$ \> cyclotomic harmonic polylogarithm; see page \pageref{CShlogdef}\\
    $\shuffle$ \> the shuffle product; see page \pageref{HShpro}\\
    $\textnormal{sign}$ \> $\sign{a}$ gives the sign of the number $a$\\
    $\abs{\ }$ \> $\abs{\ve w}$ gives the degree of a word $\ve w$; $\abs{a}$ is the absolute value of the number $a$\\ 
    $\wedge$ \> $a \wedge b = \sign{a}\sign{b}(\abs{a}+\abs{b});$ see page \\%{abwedge}\\
    $\mu(n)$ \>the M\"obius function; see page \pageref{abmue}\\
    $\zeta_n$\> $\zeta_n$ is the value of the Riemann zeta-function at $n$\\
    $\M{f(x)}{n}$ \> the Mellin transform of f(x); see pages \pageref{HSabmell}, \pageref{SSabmellplus} and \pageref{CSabmellplus} \\
    $\Mp{f(x)}{n}$ \> the extended Mellin transform of f(x); see page \pageref{HSabmellplus} \\
    $\ve{0}_k$ \> $(\underbrace{0,0,\ldots,0}_{k\times})$ \\
  \end{tabbing}

\newpage

\cleardoublepage 

\pagenumbering{arabic}

\chapter{Introduction}
\label{Introduction}
Harmonic sums and harmonic polylogarithms associated to them by a Mellin transform \cite{Vermaseren1998,Bluemlein1999}, emerge in perturbative 
calculations of massless or massive single scale problems in quantum field theory \cite{Yndurain1979,GonzalezArroyo,Mertig,Bluemlein2009a}. In order to facilitate these computations, many properties and methods have been worked out for 
these special functions. In particular, that the harmonic polylogarithms form a shuffle algebra \cite{Remiddi2000}, while the harmonic sums form a quasi 
shuffle algebra \cite{Hoffman1992,Hoffman,Hoffman1997,Radford1979,Bluemlein2004,Vermaseren1998,Ablinger2009}. Various relations between harmonic sums, \ie algebraic and structural relations have already been 
considered \cite{Bluemlein2004,Bluemlein2008,Bluemlein2009,Bluemlein2009a}, and the asymptotic behavior and the analytic continuation of these sums was worked out up to certain 
weights \cite{Bluemlein2000,Bluemlein2009,Bluemlein2009a}. Moreover, algorithms to compute the Mellin transform of harmonic polylogarithms and the inverse Mellin transform \cite{Remiddi2000} as well 
as algorithms to compute the asymptotic expansion of harmonic sums (with positive indices) are known \cite{Minh2000}. In addition, argument transforms and power series expansions of harmonic 
polylogarithms were worked out~\cite{Remiddi2000,Ablinger2009}. Note that harmonic sums at infinity and harmonic polylogarithms at one are closely related to the multiple zeta values \cite{Remiddi2000}; 
relations between them and basis representation of them are considered, \eg in  \cite{Broadhurst2010,Vermaseren1998,MZV1,MZV2,MZV3}.
For the exploration of related objects, like Euler sums, and their applications with algorithms  see, \eg \cite{Borwein95,Borwein96,Flajolet1998,SchneiderPemantle}.

In recent calculations generalizations of harmonic sums and harmonic polylogarithms, \ie S-sums (see e.g., \cite{Moch2002,Ablinger2011c,Ablinger2010b,Ablinger2010}) and cyclotomic harmonic 
sums (see e.g., \cite{Ablinger2011,Ablinger2012}) on the one hand and multiple polylogarithms and cyclotomic polylogarithms (see e.g., \cite{Ablinger2011,Ablinger2012}) on the other 
hand emerge. Both generalizations of the harmonic sums, \ie cyclotomic harmonic sums and S-sums can be viewed as subsets of the more general class of cyclotomic S-sums
$\S{(a_1,b_1,c_1),(a_2,b_2,c_2),\ldots,(a_k,b_k,c_k)}{x_1,x_2,\ldots,x_k;n}$ defined as
\begin{eqnarray*}
&&\S{(a_1,b_1,c_1),(a_2,b_2,c_2),\ldots,(a_k,b_k,c_k)}{x_1,x_2,\ldots,x_k;n}=\\
&&\hspace{4cm}\sum_{i_1 \geq i_2,\cdots i_k \geq 1}\frac{x_1^{i_1}}{(a_1 i_1+b_1)^{c_1}}\frac{x_2^{i_2}}{(a_2 i_2+b_2)^{c_2}}\cdots\frac{x_k^{i_1}}{(a_k i_k+b_k)^{c_k}},
\end{eqnarray*}
which form a quasi-shuffle algebra. In addition, the corresponding generalizations of the harmonic polylogarithms to multiple polylogarithms and cyclotomic polylogarithms fall into 
the class of Poincar\'{e} iterated integrals (see \cite{Poincare1884,Lappo-Danielevsky1953}).\\
We conclude that so far only some basic properties of S-sums \cite{Moch2002} are known and these generalized objects need further investigations for, \eg ongoing calculations.\\
One of the main goals of this theses is to extend known properties of harmonic sums and harmonic polylogarithms to their generalizations and to provide algorithms to deal with these quantities, in 
particular, to find relations between them, to compute series expansions, to compute Mellin transforms, etc.\\
Due to the complexity of 
higher order calculations the knowledge of as many as possible relations between the finite harmonic sums, S-sums and cyclotomic harmonic sums is of importance to simplify the calculations and to obtain as 
compact as possible analytic results. Therefore we derive not only algebraic but also structural relations originating from differentiation and multiplication of the upper summation limit. Concerning 
harmonic sums these relations have already been considered in \cite{Ablinger2009,Bluemlein2009,Bluemlein2009a} and for the sake of completeness we briefly summarize these results here and extend them 
to S-sums and cyclotomic harmonic sums.\\
In physical applications (see e.g.,\cite{Bluemlein1999,Bluemlein2000,Bluemlein2004,Moch2004,Moch2004b,Moch2005,Bluemlein2005,Bluemlein2008,Ablinger2012,Ablinger2011c,Ablinger2011b,Ablinger2010}) an analytic continuation of 
these nested sums is required (which is already required in order to establish differentiation) and eventually one  has to derive the complex analysis for these nested sums. Hence we look at integral 
representations of them which leads to Mellin type transformations of harmonic polylogarithms, multiple polylogarithms and cyclotomic polylogarithms.\\
Starting form the integral representation of the nested sums we derive algorithms to determine the asymptotic behavior of harmonic sums, S-sums and cyclotomic harmonic sums, \ie we are able to compute 
asymptotic expansions up to arbitrary order of these sums.\\ 
In the context of Mellin transforms and asymptotic expansions the values of the nested sums at infinity and connected to them the values of harmonic polylogarithms, 
multiple polylogarithms and cyclotomic polylogarithms at certain constants are of importance. For harmonic sums this leads to multiple zeta values and for example in \cite{Broadhurst2010} relations between 
these constants have been investigated. Here we extend these consideration to S-sums and cyclotomic harmonic sums (compare \cite{Ablinger2011}).\\
Besides several argument transformations of harmonic polylogarithms, multiple polylogarithms and cyclotomic polylogarithms we derive algorithms to compute power series expansions about
zero as well as algorithms to determine the asymptotic behavior of these iterated integrals.\\
The various relations between the nested sums together with the values at infinity on the one hand and the nested integrals together with the values at special constants on the other hand 
are summarized by Figure \ref{Ifig1}.

\def\circlea{(-8.3cm,6cm) circle (1.8cm and 1cm)}
\def\circleb{(-6.5cm,6cm) circle (1.8cm and 1cm)}
\def\circlec{(-1.8cm,6cm) circle (1.8cm and 1cm)}
\def\circled{(-0.0cm,6cm) circle (1.8cm and 1cm)}
\def\circlee{(-8.3cm,2cm) circle (1.8cm and 1cm)}
\def\circlef{(-6.5cm,2cm) circle (1.8cm and 1cm)}
\def\circleg{(-1.8cm,2cm) circle (1.8cm and 1cm)}
\def\circleh{(-0.0cm,2cm) circle (1.8cm and 1cm)}

\tikzset{filled/.style={fill=circle area, draw=circle edge, thick}, outline/.style={draw=circle edge, thick}}

\colorlet{circle edge}{black!100}
\colorlet{circle area}{black!20}
\setlength{\parskip}{5mm}
\begin{figure}[ht!]
\centering
\begin{tikzpicture}
\centering
     \begin{scope}
         \clip \circlea;
         \fill[filled] \circleb;
     \end{scope}
     \draw[outline] \circlea;
     \draw[outline] \circleb;
     \draw (-7.40cm,6.3cm) node[font=\small] {H-Sums};
     \draw (-7.35cm,5.7cm) node[font=\tiny] {$\S{-1,2}n$};
     \draw (-9.05cm,6.3cm) node[font=\small] {S-Sums};
     \draw (-9.05cm,5.7cm) node[font=\tiny] {$\S{1,2}{\frac{1}{2},1;n}$};
     \draw (-5.70cm,6.3cm) node[font=\small] {C-Sums};
     \draw (-5.65cm,5.7cm) node[font=\tiny] {$\S{(2,1,-1)}n$};
     \begin{scope}
         \clip \circlec;
         \fill[filled] \circled;
     \end{scope}
     \draw[outline] \circlec;
     \draw[outline] \circled;
     \draw (-0.90cm,6.3cm) node[font=\small] {H-Logs};
     \draw (-0.85cm,5.7cm) node[font=\tiny] {$\H{-1,1}x$};
     \draw (-2.55cm,6.3cm) node[font=\small] {C-Logs};
     \draw (-2.60cm,5.7cm) node[font=\tiny] {$\textnormal{H}_{(4,1),(0,0)}\hspace{-0.05cm}(x)$};
     \draw (+0.80cm,6.3cm) node[font=\small] {M-Logs};
     \draw (+0.85cm,5.7cm) node[font=\tiny] {$\H{2,3}x$};

     \draw[thick,->] (-9cm,6.5cm) .. controls +(80:1cm) and +(100:1cm) .. (1cm,6.5cm) node[midway,sloped,above,font=\small] {integral representation (inv. Mellin transform)};
     \draw[thick,->] (-7.25cm,6.5cm) .. controls +(60:1cm) and +(120:1cm) .. (-0.75cm,6.5cm);
     \draw[thick,->] (-5.50cm,6.5cm) .. controls +(40:1cm) and +(140:1cm) .. (-2.5cm,6.5cm);

     \draw[thick,<-] (-9cm,5.5cm) .. controls +(280:1cm) and +(260:1cm) .. (1cm,5.5cm) node[midway,sloped,below,font=\small] {Mellin transform};
     \draw[thick,<-] (-7.25cm,5.5cm) .. controls +(300:1cm) and +(240:1cm) .. (-0.75cm,5.5cm);
     \draw[thick,<-] (-5.5cm,5.5cm) .. controls +(320:1cm) and +(220:1cm) .. (-2.5cm,5.5cm);

     \begin{scope}
         \clip \circlee;
         \fill[filled] \circlef;
     \end{scope}
     \draw[outline] \circlee;
     \draw[outline] \circlef;
     \draw (-7.40cm,2cm) node[font=\tiny] {$\S{-1,2}\infty$};
     \draw (-9.15cm,2cm) node[font=\tiny] {$\S{1,2}{\frac{1}{2},1;\infty}$};
     \draw (-5.60cm,2cm) node[font=\tiny] {$\S{(2,1,-1)}\infty$};
     \draw[thick,->] (-7.4cm,4.7cm) -- (-7.4cm,3.3cm) node[midway,sloped,above,font=\small] {$n\rightarrow\infty$};

     \begin{scope}
         \clip \circleg;
         \fill[filled] \circleh;
     \end{scope}
     \draw[outline] \circleg;
     \draw[outline] \circleh;
     \draw (-0.85cm,2cm) node[font=\tiny] {$\H{-1,1}1$};
     \draw (-2.65cm,2cm) node[font=\tiny] {$\textnormal{H}_{(4,1),(0,0)}\hspace{-0.05cm}(1)$};
     \draw (0.85cm,2cm) node[font=\tiny] {$\H{2,3}c$};
     \draw[thick,->] (-0.9cm,5.3cm) -- (-0.9cm,2.7cm) node[midway,sloped,above,font=\small] {$x\rightarrow 1$};
     \draw[thick,->] (-2.0cm,5.3cm) -- (-2.0cm,2.7cm) node[midway,sloped,above,font=\small] {$x\rightarrow 1$};
     \draw[thick,->] (0.2cm,5.3cm) -- (0.2cm,2.7cm) node[midway,sloped,above,font=\small] {$x\rightarrow c\in \R$};

     \draw[thick,<-] (-9cm,1.5cm) .. controls +(280:1cm) and +(260:1cm) .. (1cm,1.5cm) node[midway,sloped,below,font=\small] {power series expansion};
     \draw[thick,<-] (-7.25cm,1.5cm) .. controls +(300:1cm) and +(240:1cm) .. (-0.75cm,1.5cm);
     \draw[thick,<-] (-5.5cm,1.5cm) .. controls +(320:1cm) and +(220:1cm) .. (-2.5cm,1.5cm);

     \draw[thick,->] (-9cm,2.5cm) .. controls +(80:1cm) and +(100:1cm) .. (1cm,2.5cm);
     \draw[thick,->] (-7.25cm,2.5cm) .. controls +(60:1cm) and +(120:1cm) .. (-0.75cm,2.5cm);
     \draw[thick,->] (-5.50cm,2.5cm) .. controls +(40:1cm) and +(140:1cm) .. (-2.5cm,2.5cm);

\end{tikzpicture}
\caption{\label{Ifig1}Connection between harmonic sums (H-Sums), S-sums (S-Sums) and cyclotomic harmonic sums (C-Sums), their values at infinity and harmonic polylogarithms (H-Logs), multiple polylogarithms (M-Logs) and 
cyclotomic harmonic polylogarithms (C-Logs) and their values at special constants.}
\end{figure}
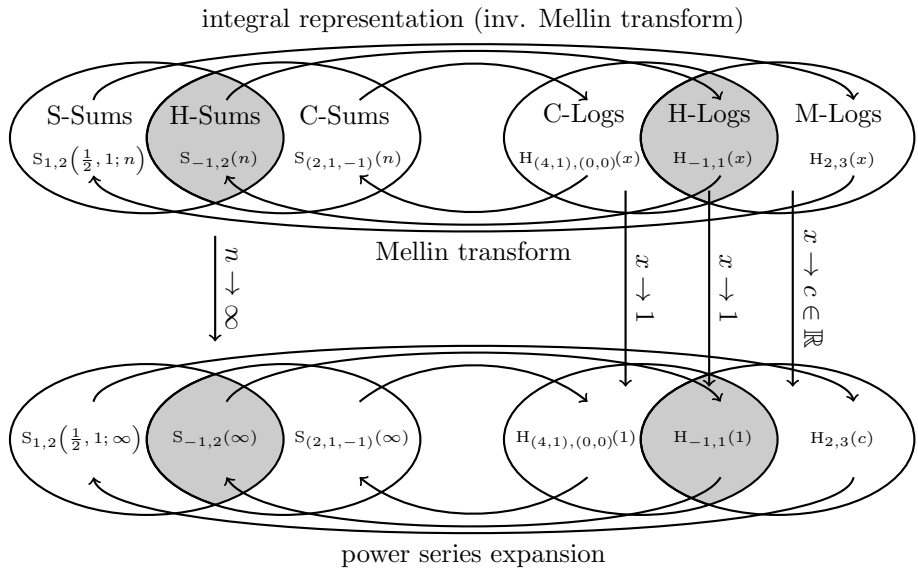
All the algorithms for the harmonic sums, S-sums, cyclotomic harmonic sums, cyclotomic S-sums, harmonic polylogarithms, multiple polylogarithms and cyclotomic harmonic polylogarithms which are presented 
in this theses are implemented in the Mathematica package \ttfamily HarmonicSums \rmfamily which was developed in \cite{Ablinger2009} originally for harmonic sums and which was extended and generalized 
in the frame of this thesis.

In the last chapter we extend the multivariate Almkvist-Zeilberger algorithm \cite{AlmZeil} in 
several directions using ideas and algorithms from \cite{Bluemlein2011} in order to apply it to special Feynman integrals emerging in renormalizable Quantum field Theories (see \cite{Bluemlein2011}).
We will look on integrals of the form
\begin{eqnarray*}
{\cal I}(\ep,N) = \int_{u_d}^{o_d} \dots\int_{u_1}^{o_1}F(n;x_1, \dots, x_d;\ep) dx_1 \dots dx_d,
\end{eqnarray*}
with $d,N \in \N$, $F(n;x_1, \dots, x_d;\ep)$ a hyperexponential term (for details see \cite{AlmZeil} or Theorem~\ref{mAZthm}), $\ep>0$ a real parameter and $u_i,o_i \in \R\cup \{-\infty,\infty\}.$
As indicated by several examples, the solution of these integrals may lead to harmonic sums, cyclotomic harmonic sums and S-sums.\\
Our Mathematica package \ttfamily MultiIntegrate, \rmfamily which was developed in the frame of this theses, can be considered as an enhanced implementation of the multivariate Almkvist Zeilberger algorithm to compute 
recurrences for the integrands and integrals. Together with the summation package \SigmaP\ \cite{Schneider2005,Schneider2001,Schneider2007,Schneider2008a,Schneider2008,Schneider2007a} our package provides methods to compute representations of ${\cal I}(\ep,N)$ in terms of harmonic sums, 
cyclotomic harmonic sums and S-sums and to compute Laurent series expansions of ${\cal I}(\ep,N)$ in the form
\begin{equation*}
{\cal I}(\ep, N) = I_t(N)\ep^t+I_{t+1}(N)\ep^{t+1}+I_{t+2}(N)\ep^{t+2}+\dots
\end{equation*}
where $t \in \Z$ and the $I_k$ are expressed in terms of harmonic sums, cyclotomic harmonic sums and S-sums, and even more generally in terms of indefinite nested sums and products \cite{Schneider2007a}. 

The remainder of the thesis is structured as follows: In Chapter~\ref{HSchapter} harmonic sums and harmonic polylogarithms are introduced. We seek power series expansions  
and the asymptotic behavior of harmonic polylogarithms and harmonic sums. We look for integral representations of harmonic sums and Mellin-transforms of harmonic polylogarithms. In addition 
we consider algebraic and structural relations between harmonic sums and we consider values of harmonic sums at infinity and values of harmonic polylogarithms at one.
In Chapter~\ref{SSchapter} S-sums and multiple polylogarithms are introduced and we generalize the theorems and algorithms from Chapter~\ref{HSchapter} to S-sums and multiple polylogarithms,
while in Chapter~\ref{CSchapter} cyclotomic harmonic sums and cyclotomic harmonic polylogarithms are defined and we seek to generalize the theorems and algorithms from Chapter~\ref{HSchapter} to these 
structures. Chapter~\ref{CSSchapter} introduces cyclotomic S-sums and provides mainly an algorithms that transforms indefinite nested sums and products to cyclotomic S-sums whenever this is possible.
In Chapter~\ref{Packagechapter} we briefly summarize the main features of the packages \ttfamily HarmonicSums \rmfamily in which the algorithms and procedures presented in the foregoing chapters are 
implemented. Finally, Chapter~\ref{AZchapter} deals with the application of (extensions of) the multivariate Almkvist-Zeilberger to certain types of Feynman integrals and the corresponding 
package \ttfamily MultiIntegrate\rmfamily.\\

\text{\bf Acknowledgments.} This research was supported by the Austrian Science Fund (FWF): grant P20162-N18 and the Research Executive Agency (REA) of the European Union under 
the Grant Agreement number PITN-GA-2010-264564 (LHCPhenoNet).

I would like to thank C. Schneider and J. Bl\"umlein for supervising this theses and their support throughout the last three years.

\cleardoublepage  

\chapter{Harmonic Sums}
\label{HSchapter}

One of the main goals of this chapter is to present a new algorithm which computes the asymptotic expansions of harmonic sums (see, e.g., \cite{Bluemlein2000,Bluemlein2009,Bluemlein2009a,Vermaseren1998}). 
Asymptotic expansions of harmonic sums with positive indices have already been considered 
in \cite{Minh2000}, but here we will derive an algorithm which is suitable for harmonic sums in general (so negative indices are allowed as well). Our approach is inspired by \cite{Bluemlein2009,Bluemlein2009a}
and uses an integral representation of harmonic sums and in addition makes use of several properties of harmonic sums. Therefore we will state first several known and needed facts about harmonic sums.\\
Note that the major part of this chapter up to Section \ref{HSExpansion} can already be found in \cite{Ablinger2009}, however there are some 
differences and some complemental facts are added.\\
The integral representation which is used in the algorithm is based on the Mellin transformation of harmonic polylogarithms defined in \cite{Remiddi2000}; note that the Mellin transformation 
defined here slightly differs. We will show how the Mellin transformation of harmonic polylogarithms can be computed using a new approach different form \cite{Remiddi2000}.\\
The Mellin transformation also leads to harmonic polylogarithms at one and connected to them to harmonic sums at infinity. Since this relation is of importance for further considerations, we repeat how these quantities are related (see \cite{Remiddi2000}).
Since the algorithm utilizes algebraic properties of harmonic sums, we will shortly comment on the structure of harmonic sums (for details see, \eg \cite{Hoffman}) and we will look at algebraic and structural 
relations (compare \cite{Ablinger2009,Bluemlein2004,Bluemlein2008,Bluemlein2009,Bluemlein2009a}). In addition we derive new formulas that count the number of basis sums at a specified weight.\\
We state some of the argument transformations of harmonic polylogarithms from \cite{Remiddi2000}, which will for example allow us to determine the asymptotic behavior of harmonic polylogarithms.
One of this argument transformations will be used in the algorithm which computes asymptotic expansions of harmonic sums. But since harmonic polylogarithms are not closed under this transformation, we 
will extend harmonic polylogarithms by adding a new letter to the index set. This new letter is added in a way such that the extended harmonic polylogarithms are closed under the required transformation.\\
The stated material will finally allow us to compute asymptotic expansion of the harmonic sums.

\def\circlei{(-8.3cm,6cm) circle (2.1cm and 1.2cm)}
\def\circlej{(-0.0cm,6cm) circle (2.1cm and 1.2cm)}
\def\circlek{(-8.3cm,2cm) circle (2.1cm and 1.2cm)}
\def\circlel{(-0.0cm,2cm) circle (2.1cm and 1.2cm)}

\tikzset{filled/.style={fill=circle area, draw=circle edge, thick},
    outline/.style={draw=circle edge, thick}}

\colorlet{circle edge}{black!100}
\colorlet{circle area}{black!20}
\setlength{\parskip}{5mm}
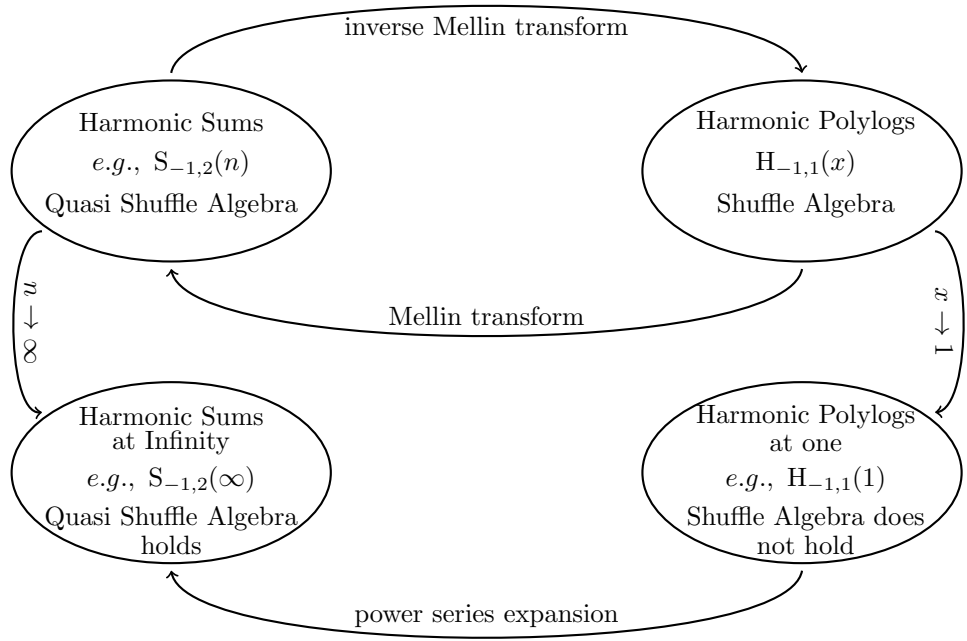
\begin{figure}
\begin{tikzpicture}
\centering

     \draw[outline] \circlei;
     \draw (-8.30cm,6.65cm) node[font=\small] {Harmonic Sums};
     \draw (-8.30cm,6.10cm) node[font=\small] {$\eg \S{-1,2}n$};
     \draw (-8.30cm,5.55cm) node[font=\small] {Quasi Shuffle Algebra};

     \draw[outline] \circlej;
     \draw (0.05cm,6.65cm) node[font=\small] {Harmonic Polylogs};
     \draw (0.05cm,6.10cm) node[font=\small] {$\H{-1,1}x$};
     \draw (0.05cm,5.55cm) node[font=\small] {Shuffle Algebra};

     \draw[thick,->,black] (-8.3cm,7.3cm) .. controls +(80:1.2cm) and +(100:1.2cm) .. (0cm,7.3cm) node[midway,sloped,below,font=\small] {inverse Mellin transform};
     \draw[thick,<-,black] (-8.3cm,4.7cm) .. controls +(280:1.2cm) and +(260:1.2cm) .. (0cm,4.7cm) node[midway,sloped,above,font=\small] {Mellin transform};

     \draw[outline] \circlek;
     \draw (-8.30cm,2.75cm) node[font=\small] {Harmonic Sums};
     \draw (-8.30cm,2.40cm) node[font=\small] {at Infinity};
     \draw (-8.30cm,1.90cm) node[font=\small] {$\eg \S{-1,2}\infty$};
     \draw (-8.30cm,1.40cm) node[font=\small] {Quasi Shuffle Algebra};
     \draw (-8.30cm,1.05cm) node[font=\small] {holds};
     \draw[thick,->] (-10cm,5.2cm) .. controls +(180:0.5cm) and +(185:0.5cm) .. (-10cm,2.8cm) node[midway,sloped,above,font=\small] {$n\rightarrow\infty$};

     \draw[outline] \circlel;
     \draw (0.05cm,2.75cm) node[font=\small] {Harmonic Polylogs};
     \draw (0.05cm,2.40cm) node[font=\small] {at one};
     \draw (0.05cm,1.90cm) node[font=\small] {$\eg \H{-1,1}1$};
     \draw (0.05cm,1.40cm) node[font=\small] {Shuffle Algebra does};
     \draw (0.05cm,1.05cm) node[font=\small] {not hold};
     \draw[thick,->] (1.75cm,5.2cm) .. controls +(355:0.5cm) and +(5:0.5cm) .. (1.75cm,2.8cm) node[midway,sloped,below,font=\small] {$x\rightarrow 1$};
     \draw[thick,<-,black] (-8.3cm,0.7cm) .. controls +(280:1.2cm) and +(260:1.2cm) .. (0cm,0.7cm) node[midway,sloped,above,font=\small] {power series expansion};

\end{tikzpicture}
\caption[height=5cm]{\label{HSconnection}Connection between harmonic sums and harmonic polylogarithms (compare~\cite{Ablinger2009}).}
\end{figure}

The connection between harmonic sums, there values at infinity, harmonic polylogarithms and there values at one are summarized in Figure \ref{HSconnection}.
Note that this chapter can be viewed as a model for Chapter \ref{SSchapter} and Chapter \ref{CSchapter}, since there we will extend the properties of harmonic sums and harmonic polylogarithms, 
which will lead to asymptotic expansion of generalizations of harmonic sums. 

\section{Definition and Structure of Harmonic Sums}
\label{HSdef}
We start by defining multiple harmonic sums; see, e.g., \cite{Bluemlein2004,Vermaseren1998}.
\begin{definition}[Harmonic Sums] For $k \in \N$, $n \in \N_0$ and $a_i\in \Z^*$ with $1\leq i\leq k$ we define
\begin{eqnarray*}
	\S{a_1,\ldots ,a_k}n= \sum_{n\geq i_1 \geq i_2 \geq \cdots \geq i_k \geq 1} \frac{\sign{a_1}^{i_1}}{i_1^{\abs {a_1}}}\cdots
	\frac{\sign{a_k}^{i_k}}{i_k^{\abs {a_k}}}.	
\end{eqnarray*}
$k$ is called the depth and $w=\sum_{i=0}^k\abs{a_i}$ is called the weight of the harmonic sum $\S{a_1,\ldots ,a_k}n$.
\label{HSdefHsum}
\end{definition}

One of the properties of harmonic sums is the following identity: for $n\in \N,$
\begin{eqnarray}
	\S{a_1,\ldots ,a_k}n\S{b_1,\ldots ,b_l}n&=&
	\sum_{i=1}^n \frac{\sign{a_1}^i}{i^{\abs {a_1}}}\S{a_2,\ldots ,a_k}i\S{b_1,\ldots ,b_l}i \nonumber\\
	&&+\sum_{i=1}^n \frac{\sign{b_1}^i}{i^{\abs {b_1}}}\S{a_1,\ldots ,a_k}i\S{b_2,\ldots ,b_l}i \nonumber\\
	&&-\sum_{i=1}^n \frac{\sign{a_1*b_1}^i}{i^{\abs{a_1}+\abs{b_1}}}\S{a_2,\ldots ,a_k}i\S{b_2,\ldots ,b_l}i.
\label{HShsumproduct}
\end{eqnarray}
The proof of equation (\ref{HShsumproduct}) follows immediately from 
$$
\sum_{i=1}^n \sum_{j=1}^n a_{ij}=\sum_{i=1}^n \sum_{j=1}^i a_{ij}+\sum_{j=1}^n \sum_{i=1}^j a_{ij}-\sum_{i=1}^n a_{ii};
$$
for a graphical illustration of the proof see Figure \ref{HSfig1}. As a consequence, any product of harmonic sums with the same upper summation limit can be written as a linear combination of harmonic sums by iterative 
application of (\ref{HShsumproduct}). Together with this product, harmonic sums form a quasi shuffle algebra (for further details see Section \ref{HSalgrel} and, \eg \cite{Hoffman1992,Hoffman1997,Hoffman}).

\begin{figure}
\centering
\parbox{.15\textwidth}{\includegraphics[width=2cm]{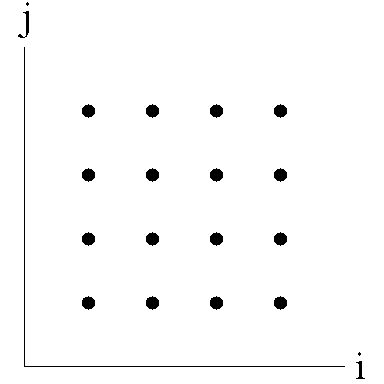}}
\centering
\parbox{.05\textwidth}{$=$}
\centering
\parbox{.15\textwidth}{\includegraphics[width=2cm]{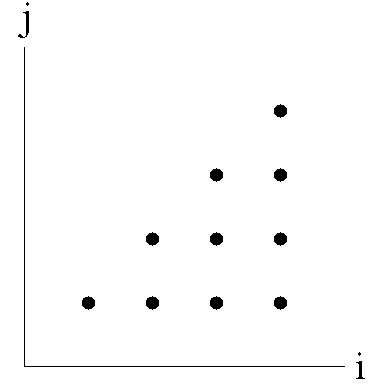}}
\centering
\parbox{.05\textwidth}{$+$}
\centering
\parbox{.15\textwidth}{\includegraphics[width=2cm]{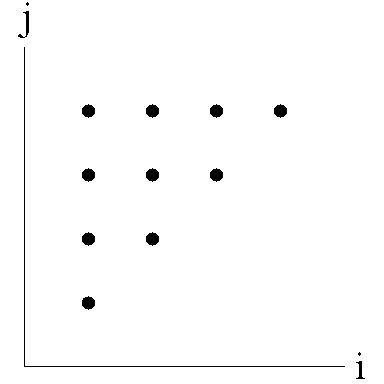}}
\centering
\parbox{.05\textwidth}{$-$}
\centering
\parbox{.15\textwidth}{\includegraphics[width=2cm]{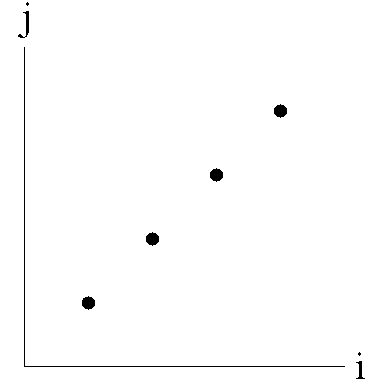}}
\caption{\label{HSfig1}Sketch of the proof of the multiplication of harmonic sums (compare \cite{Moch2002}).}
\end{figure}

\begin{example}
\begin{eqnarray*}
\S{2,3}n\S{2,2}n&=&\S{4,5}n-2\, \S{2,2,5}n-\S{2,4,3}n-\S{2,5,2}n-\S{4,2,3}n-\\&&\S{4,3,2}n+3\, \S{2,2,2,3}n+2\, \S{2,2,3,2}n+\S{2,3,2,2}n.
\end{eqnarray*}
\end{example}

\begin{definition}
We say that a harmonic sum $\S{a_1, a_2,\ldots,a_k}{n}$ has trailing ones if there is an $i \in \left\{1,2,\ldots,k\right\}$ such that for all $j \in \N$ with $i\leq j \leq k$,  $a_j=1.$ \\
Likewise, we say that a harmonic sum $\S{a_1, a_2,\ldots,a_k}{n}$ has leading ones if there is an $i \in \left\{1,2,\ldots,k\right\}$ such that for all $j \in \N$ with $1 \leq j \leq i$,  $a_j=1.$  
\end{definition}
Choosing one of the harmonic sums in (\ref{HShsumproduct}) with depth one we get
\begin{eqnarray}
\S{a}n\S{b_1,\ldots ,b_l}n&=&\S{a,b_1,\ldots ,b_l}n+\S{b_1,a,b_2,\ldots ,b_l}n+\cdots+\S{b_1,b_2,\ldots ,b_l,a}n\nonumber\\
&&-\S{a\wedge b_1,b_2,\ldots,b_l}n-\cdots-\S{b_1,b_2,\ldots,a\wedge b_m}n.
\label{HShsumpro1}
\end{eqnarray}
Here the symbol $\wedge$ is defined as
\begin{equation}
a \wedge b = \sign{a}\sign{b}(\abs{a}+\abs{b}).\nonumber\\
\label{HSabwedge}
\end{equation}

We can use (\ref{HShsumpro1}) now to single out terms of $\S1n$ from harmonic sums whose indices have leading ones. Taking $a=1$ in (\ref{HShsumpro1}), we get:
\begin{eqnarray*}
\S{1}n\S{b_1,\ldots ,b_l}n&=&\S{1,b_1,\ldots ,b_l}n+\S{b_1,1,b_2,\ldots ,b_l}n+\cdots+\S{b_1,b_2,\ldots ,b_l,1}n\\
&&-\S{1\wedge b_1,b_2,\ldots,b_l}n-\cdots-\S{b_1,b_2,\ldots,1\wedge b_m}n,
\end{eqnarray*} 
or equivalently:
\begin{eqnarray*}
\S{1,b_1,\ldots ,b_l}n &=& \S{1}n\S{b_1,\ldots ,b_l}n-\S{b_1,1,b_2,\ldots ,b_l}n-\cdots-\S{b_1,b_2,\ldots ,b_l,1}n\\
&&+\S{1\wedge b_1,b_2,\ldots,b_l}n+\cdots+\S{b_1,b_2,\ldots,1\wedge b_m}n.
\end{eqnarray*} 
If $b_1$ is $1$ as well, we can move $\S{b_1,1,b_2,\ldots ,b_l}n$ to the left and can divide by two. This leads to
\begin{eqnarray*}
\S{1,b_1,\ldots ,b_l}n &=& \frac{1}{2}\left(\S{1}n\S{b_1,\ldots ,b_l}n-\S{b_1,b_2, 1,\ldots ,b_l}n-\cdots-\S{b_1,b_2,\ldots ,b_l,1}n \right.\\ 
&&\left.+\S{1\wedge b_1,b_2,\ldots,b_l}n+\cdots+\S{b_1,b_2,\ldots,1\wedge b_m}n \right).
\end{eqnarray*} 
Now we can use (\ref{HShsumpro1}) for all the other terms, and we get an identity which extracts two leading ones. We can repeat this strategy as often as 
needed in order to extract all the powers of $\S{1}n$ from a harmonic sum.
\begin{example} For $x \in \N$,
\begin{eqnarray*}
\S{1,1,2,3}n&=&\frac{1}{2} \S{2,3}n \S{1}n^2+(\S{2,4}n+\S{3,3}n-\S{2,1,3}n-\S{2,3,1}n) \S{1}n\\
		&&+\frac{1}{2} \S{2,5}n+\S{3,4}n+\frac{1}{2} \S{4,3}n-\S{2,1,4}n-\frac{1}{2}\S{2,3,2}n-\S{2,4,1}n\\
		&&-\S{3,1,3}n-\S{3,3,1}n+\S{2,1,1,3}n+\S{2,1,3,1}n+\S{2,3,1,1}n.
\end{eqnarray*}
\end{example}
\begin{remark}
We can decompose a harmonic sum $\S{b_1,\ldots ,b_l}n$ in a univariate polynomial in $\S1n$ with coefficients in the harmonic sums without leading ones. 
If the harmonic sum has exactly $r$ leading ones, the highest power of $\S1n,$ which will appear, is $r.$\\
Note, that in a similar way it is possible to extract trailing ones. Hence we can also decompose a harmonic sum $\S{b_1,\ldots ,b_l}n$ in a univariate 
polynomial in $\S1n$ with coefficients in the harmonic sums without trailing ones.
\label{HSextractleading1rem}
\end{remark}

\section{Definition and Structure of Harmonic Polylogarithms}
\label{HShpldefsec}
In this section we define harmonic polylogarithms as introduced in \cite{Remiddi2000}. In addition we will state some basic properties, which can already be found in \cite{Remiddi2000}.
We start with defining the auxiliary functions $f_a$ for $a\in\{-1,0,1\}$ as follows:
\begin{eqnarray}
&&f_a:(0,1)\mapsto \R\nonumber\\
&&f_a(x)=\left\{ 
		\begin{array}{ll}
				\frac{1}{x}, &  \textnormal{if }a=0  \\
				\frac{1}{\abs{a}-\sign{a}\;x}, & \textnormal{otherwise}.  
		\end{array} 
		\right.  \nonumber
\end{eqnarray}
\begin{definition}[Harmonic Polylogarithms]
For  $m_i \in \{-1,0,1\},$ we define harmonic polylogarithms for $x\in (0,1):$
\begin{eqnarray}
\H{}{x}&=&1,\nonumber\\
\H{m_1,m_{2},\ldots,m_w}{x} &=&\left\{ 
		  	\begin{array}{ll}
						\frac{1}{w!}(\log{x})^w,&  \textnormal{if }(m_1,\ldots,m_w)=(0,\ldots,0)\\
						\int_0^x{f_{m_1}(y) \H{m_{2},\ldots,m_w}{y}dy},& \textnormal{otherwise}. 
					\end{array} \right.  \nonumber
\end{eqnarray}
The length $w$ of the vector $\ve m=(m_1,\ldots,m_w)$ is called the weight of the harmonic polylogarithm $\H{\ve m}x.$
\label{HShlogdef}
\end{definition}
\begin{example} For $x\in(0,1),$
\begin{eqnarray}
\H{1}x &=& \int_0^x{\frac{dy}{1-y}}=-\log{(1-x)}, \nonumber\\
\H{-1}x&=& \int_0^x{\frac{dy}{1+y}}=\log{(1+x)}, \nonumber\\
\H{-1,1}x&=& \int_0^x{\frac{\H{1}{y}}{1+y}dy}= \int_0^x{\frac{\log{(1-y)}}{1+y}dy}. \nonumber
\end{eqnarray}
\end{example}
A harmonic polylogarithm $\H{\ve m}x=\H{m_1,\ldots,m_w}x$ is an analytic functions for $x\in(0,1).$ For the limits $x\rightarrow 0$ and  $x\rightarrow 1$ we have the following facts (compare \cite{Remiddi2000}):
\begin{itemize}
 \item It follows from the definition that if $\ve m\neq \ve 0_w$, $\H{\ve m}0~=~0.$
 \item If $m_1\neq 1$ or if $m_1=1$ and $m_v=0$ for all $v$ with $1<v\leq w$ then $\H{\ve m}1$ is finite.
 \item If $m_1=1$ and $m_v\neq0$ for some $v$ with $1<v\leq w$, $\lim_{x\rightarrow 1^-} \H{\ve m}x$ behaves as a combination of powers of $\log(1-x)$ as we will see later in more detail.
\end{itemize}
We define $\H{\ve m}0:=\lim_{x\rightarrow 0^+} \H{\ve m}x$ and $\H{\ve m}1:=\lim_{x\rightarrow 1^-} \H{\ve m}x$ if the limits exist.
\begin{remark}
For the derivatives we have for all $x\in (0,1)$ that $$ \frac{d}{d x} \H{\ve m}{x}=f_{m_1}(x)\H{m_{2},m_{3},\ldots,m_w}{x}. $$ 
\end{remark}
The product of two harmonic polylogarithms of the same argument can be expressed using the formula
\begin{equation}
\label{HShpro}
\H{\ve p}x\H{\ve q}x=\sum_{\ve r= \ve p \shuffle \ve q}\H{\ve r}x
\end{equation}
in which $\ve p \shuffle \ve q$ represent all merges of $\ve p$ and $\ve q$ in which the relative orders of the elements of $\ve p$ and $\ve q$ are preserved. For details see \cite{Remiddi2000,Ablinger2009}.
The following remarks are in place.
\begin{remark}
The product of two harmonic polylogarithms of weights $w_1$ and $w_2$ can be expressed as a linear combination of $(w_1+w_2)!/(w_1!w_2!)$ harmonic polylogarithms of weight $w=w_1+w_2$.
\end{remark}
\begin{definition}
We say that a harmonic polylogarithm $\H{m_1, m_2,\ldots,m_w}{x}$ has trailing zeros if there is an $i \in \left\{1,2,\ldots,w\right\}$ such that for all $j \in \N$ with $i\leq j \leq w$,  $m_j=0.$ \\
Likewise, we say that a harmonic polylogarithm $\H{m_1, m_2,\ldots,m_w}{x}$ has leading ones if there is an $i \in \left\{1,2,\ldots,w\right\}$ such that for all $j \in \N$ with $1 \leq j \leq i$,  $a_j=1.$  
\end{definition}
Analogously to harmonic sums, we proceed as follows; compare \cite{Remiddi2000}. Choosing one of the harmonic polylogarithms is in (\ref{HShpro}) with weight one yields:
\begin{eqnarray}
\label{HShpro1}
\H{a}{x}\H{m_1,\ldots,m_w}{x}&=&\H{a,m_1,m_{2},\ldots,m_w}{x}\nonumber\\
&&+\H{m_1,a,m_{2},\ldots,m_w}{x}\nonumber\\
&&+\H{m_1,m_{2},a,m_{3},\ldots,m_w}{x}\nonumber\\
&&+\cdots +\nonumber\\
&&+\H{m_1,m_{2},\ldots,m_w,a}{x}.
\end{eqnarray}
We can use (\ref{HShpro1}) now to single out terms of $\log{x}=\H{0}x$ from harmonic polylogarithms whose indices have trailing zeros. Let $a=0$ in (\ref{HShpro1}); then after using the definition of harmonic polylogarithms we get:
\begin{eqnarray}
\H0x\H{m_1,\ldots,m_w}{x}&=&\H{0,m_1,\ldots,m_w}{x} + \H{m_1,0,m_{2},\ldots,m_w}{x}\nonumber\\
&&+\H{m_1,m_{2},0,m_{3},\ldots,m_w}{x} + \cdots + \H{m_1,\ldots,m_w,0}{x},\nonumber
\end{eqnarray} 
or equivalently:
\begin{eqnarray}
\H{m_1,\ldots,m_w,0}{x} &=& \H0x\H{m_1,\ldots,m_w}{x} - \H{0,m_1,\ldots,m_w}{x}\nonumber\\ 
&&-\H{m_1,0,m_{2},\ldots,m_w}{x} -\cdots- \H{m_1,\ldots,0,m_w}{x}.\nonumber
\end{eqnarray} 
If $m_w$ is $0$ as well, we can move the last term to the left and can divide by two. This leads to
\begin{eqnarray}
\H{m_1,\ldots,0,0}{x} &=& \frac{1}{2}\left(\H0x\H{m_1,\ldots,m_{w-1},0}{x} - \H{0,m_1,\ldots,m_{w-1},0}{x}\right.\nonumber\\ 
&&\left.-\H{m_1,0,m_{2},\ldots,m_{w-1},0}{x} -\cdots- \H{m_1,\ldots,0,m_{w-1}}{x}\right).\nonumber
\end{eqnarray} 
Now we can use (\ref{HShpro1}) for all the other terms, and we get an identity which extracts the logarithmic singularities due to two trailing zeros. We can repeat this strategy as often as needed in order to extract all the powers of $\log(x)$ or equivalently $\H{0}{x}$ from a harmonic polylogarithm.
\begin{example} For $x \in (0,1)$,
\begin{eqnarray}
\H{1, -1, 0, 0}{x}&=&\frac{1}{2} \H{0}{x}^2 \H{1, -1}{x} - \H{0}{x} \H{0, 1, -1}{x} \nonumber\\
										&&- \H{0}{x} \H{1, 0, -1}{x} +\H{0, 0, 1, -1}{x} + \H{0, 1, 0, -1}{x}\nonumber\\
										&&+ \H{1, 0, 0, -1}{x}\nonumber\\
										&=&\frac{1}{2} \H{0}{x}^2 \H{1, -1}{x} - \H{0}{x} \H{0, 1, -1}{x} \nonumber\\
										&&- \H{0}{x} \H{1, 0, -1}{x} +\H{0, 0, 1, -1}{x} + \H{0, 1, 0, -1}{x}\nonumber\\
										&&+ \H{1, 0, 0, -1}{x}.\nonumber
\end{eqnarray}
\end{example}
\begin{remark}
We can decompose a harmonic polylogarithm $\H{m_1, m_{2},\ldots,m_w}{x}$ in a univariate polynomial in $\H{0}{x}$ with coefficients in the harmonic polylogarithms without trailing zeros. If the harmonic polylogarithm has exactly $r$ trailing zeros, the highest power of $\H{0}{x},$ which will appear, is $r.$
\label{HSremarktrailing}
\end{remark}
Similarly, we can use (\ref{HShpro1}) to extract powers of $\log(1-x)$, or equivalently $-\H{1}{x}$ from harmonic polylogarithms whose indices have leading ones and hence are singular around $x=1$. 

\begin{example} For $x \in (0,1)$,
\begin{eqnarray}
\H{1, 1, 0, -1}{x}&=&-\H{0, -1}{x} \H{1, 1}{x} - \H{1}{x} \H{0, -1, 1}{x}\nonumber\\
										&&- \H{1}{x} \H{0, 1, -1}{x} + \H{0, -1, 1, 1}{x} \nonumber\\
										&&+ \H{0, 1, -1, 1}{x} + \H{0, 1, 1, -1}{x} \nonumber\\
										&=&-\frac{1}{2}\H{0, -1}{x} \H{1}{x}^2 -\H{1}{x} \H{0, -1, 1}{x}\nonumber\\
										&&- \H{1}{x} \H{0, 1, -1}{x} + \H{0, -1, 1, 1}{x} \nonumber\\
										&&+ \H{0, 1, -1, 1}{x} + \H{0, 1, 1, -1}{x}. \nonumber
\end{eqnarray}
\end{example}

\begin{remark}\label{HSremlead1}
We can decompose a harmonic polylogarithm $\H{m_1, m_{2},\ldots,m_w}{x}$ in a univariate polynomial in $\H{1}{x}$ with coefficients in the harmonic 
polylogarithms without leading ones. If the harmonic polylogarithm has exactly $r$ leading ones, the highest power of $\H{1}{x},$ which will appear, 
is $r.$ So all divergences of harmonic polylogarithms for $x\rightarrow 1$ can be traced back to the basic divergence of $\H{1}{x}=-\log(1-x)$ for $x\rightarrow 1$.
\end{remark}

\begin{remark}\label{HSremlead1trail0}
We can combine these two strategies to extract both, leading ones and trialling zeros. Hence, it is always possible to express a harmonic polylogarithm in a 
bivariate polynomial in $\H{0}{x}=\log(x)$ and $\H{1}{x}=-\log(1-x)$ with coefficients in the harmonic polylogarithms without leading ones or trailing zeros, 
which are continuous on $[0,1]$ and finite at $x=0$ and $x=1.$ 
\end{remark}

\begin{example} For $x \in (0,1)$,
\begin{eqnarray}
\H{1, -1, 0}{x}&=&-\H{0}{x} \H{-1, 1}{x} + \H{1}{x} (\H{-1}{x} \H{0}{x} \nonumber\\
								 &&- \H{0, -1}{x}) + \H{0, -1, 1}{x} \nonumber\\
								 &=&-\log(x) \H{-1, 1}{x} + \log(1-x) (\H{-1}{x} \log(x) \nonumber\\
								 &&- \H{0, -1}{x}) + \H{0, -1, 1}{x}. \nonumber
\end{eqnarray}
\end{example}

\section{Identities between Harmonic Polylogarithms of Related Arguments}
\label{HSRelatedArguments}
In this section we look at special transforms of the argument of harmonic polylogarithms. Again they can be already found in \cite{Remiddi2000}.

\subsection{\texorpdfstring{$1-x\rightarrow x$}{1-x->x}}
\label{HS1xx}
For the transformation $1-x\rightarrow x$ we restrict to harmonic polylogarithms with index sets that do not contain $-1$ (we will include the index $-1$ later in Section \ref{HS1xxextended}). 
We proceed recursively on the weight $w$ of the harmonic polylogarithm, For the base cases we have for $x \in (0,1)$
\begin{eqnarray} 
\H{0}{1-x}&=&-\H{1}{x}\label{HStrafo1xx0}\\
\H{1}{1-x}&=&-\H{0}{x}\label{HStrafo1xx}.
\end{eqnarray}
Now let us look at higher weights $w>1.$ We consider $\H{m_1,m_2,\ldots,m_w}{1-x}$ with $m_i\in \{0,1\}$ and suppose that we can already apply the transformation for harmonic polylogarithms of weight $<w.$ 
If $m_1=1,$ we can remove leading ones and end up with harmonic polylogarithms without leading ones and powers of $\H{1}{1-x}.$ For the powers of $\H{1}{1-x}$ we can use (\ref{HStrafo1xx}); therefore, 
only the cases in which the first index $m_1\neq 1$ are to be considered. If $m_i=0$ for all $1<i\leq w,$ we are in fact dealing with a power of $\H{0}{1-x}$ and hence we can use (\ref{HStrafo1xx0});
if $m_1=0$ and there exists an $i$ such that $m_i\neq 0$ we get for $x\in (0,1)$(see \cite{Remiddi2000}):
\begin{eqnarray}
\H{0,m_2,\ldots,m_w}{1-x}&=&\int_0^{1-x}{\frac{\H{m_2,\ldots,m_w}y}{y}dy}\nonumber\\
		&=&\int_0^{1}{\frac{\H{m_2,\ldots,m_w}y}{y}dy}-\int_{1-x}^{1}{\frac{\H{m_2,\ldots,m_w}y}{y}dy}\nonumber\\
		&=&\H{0,m_2,\ldots,m_w}1-\int_{0}^{x}{\frac{\H{m_2,\ldots,m_w}{1-t}}{1-t}dt},
\end{eqnarray}
where the constant $\H{0,m_2,\ldots,m_w}1$ is finite.  Since we know the transform for weights $<w$ we can apply it to $\H{m_2,\ldots,m_w}{1-t}$ and finally we obtain the required weight $w$ identity.

\subsection{\texorpdfstring{$\frac{1-x}{1+x} \rightarrow x$}{(1-x)/(1+x)->x}}
\label{HS1x1x}
Proceeding recursively on the weight $w$ of the harmonic polylogarithm, the base case is for $x\in (0,1)$
\begin{eqnarray}
\H{-1}{\frac{1-x}{1+x}}&=&\H{-1}{1}-\H{-1}x\\
\H{0}{\frac{1-x}{1+x}}&=&-\H{1}{x}+\H{-1}x\\
\H{1}{\frac{1-x}{1+x}}&=&-\H{-1}{1}-\H{0}{x}+\H{-1}{x}\label{HStrafo1x1x}.
\end{eqnarray}
Now let us look at higher weights $w>1.$ We consider $\H{m_1,m_2,\ldots,m_w}{\frac{1-x}{1+x}}$ with $m_i\in \{-1,0,1\}$ and we suppose that we can already apply the transformation for harmonic 
polylogarithms of weight $<w.$ If $m_1=1,$ we can remove leading ones and end up with harmonic polylogarithms without leading ones and powers of $\H{1}{\frac{1-x}{1+x}}.$ For the powers 
of $\H{1}{\frac{1-x}{1+x}}$ we can use (\ref{HStrafo1x1x}); therefore, only the cases in which the first index $m_1\neq 1$ are to be considered. For $m_1\neq 1$ we get for $x\in (0,1)$ (see \cite{Remiddi2000}):
\begin{eqnarray*}
\H{-1,m_2,\ldots,m_w}{\frac{1-x}{1+x}}&=&\H{-1,m_2,\ldots,m_w}1-\int_0^x\frac{1}{1+t}\H{m_2,\ldots,m_w}{\frac{1-t}{1+t}}dt\\
\H{0,m_2,\ldots,m_w}{\frac{1-x}{1+x}}&=&\H{0,m_2,\ldots,m_w}1-\int_0^x\frac{1}{1-t}\H{m_2,\ldots,m_w}{\frac{1-t}{1+t}}dt\\
		&&-\int_0^x\frac{1}{1-t}\H{m_2,\ldots,m_w}{\frac{1+t}{1+t}}dt.
\end{eqnarray*}
Again we have to consider the case for $\H{0,\ldots,0}{\frac{1-x}{1+x}}$ separately, however we can handle this case since we can handle $\H{0}{\frac{1-x}{1+x}}.$
Since we know the transform for weights $<w,$ we can apply it to $\H{m_2,\ldots,m_w}{\frac{1+t}{1+t}},$ and finally we obtain the required weight $w$ identity by using the definition of the extended harmonic polylogarithms.

\subsection{\texorpdfstring{$\frac{1}{x}\rightarrow x$}{1/x->x}}
\label{HS1dxx}
We first restrict this transformation to harmonic polylogarithms with index sets that do not contain $1.$
Proceeding recursively on the weight $w$ of the harmonic polylogarithm we use the base case for $x\in (0,1):$
\begin{eqnarray*}
\H{-1}{\frac{1}{x}}&=&\H{-1}{x}-\H{0}x\\
\H{0}{\frac{1}{x}}&=&-\H{0}{x}.%\label{HStrafo1dx1}
\end{eqnarray*}
Now let us look at higher weights $w>1.$ We consider $\H{m_1,m_2,\ldots,m_w}{\frac{1}{x}}$ with $m_i\in \{-1,0\}$ and suppose that we can already apply the transformation for harmonic polylogarithms 
of weight $<w.$ For $m_1\neq 1$ we get for $x\in (0,1)$ (compare \cite{Remiddi2000}):
\begin{eqnarray*}
\H{-1,m_2,\ldots,m_w}{\frac{1}{x}}&=&\H{-1,m_2,\ldots,m_w}1+\int_x^1\frac{1}{t^2(1+1/t)}\H{m_2,\ldots,m_w}{\frac{1}{t}}dt\\
				  &=&\H{-1,m_2,\ldots,m_w}1+\int_x^1\frac{1}{t}\H{m_2,\ldots,m_w}{\frac{1}{t}}dt\\
				  & &-\int_x^1\frac{1}{1+t}\H{m_2,\ldots,m_w}{\frac{1}{t}}dt\\
\H{0,m_2,\ldots,m_w}{\frac{1}{x}}&=&\H{0,m_2,\ldots,m_w}1+\int_x^1\frac{1}{t^2(1/t)}\H{m_2,\ldots,m_w}{\frac{1}{t}}dt\\
				  &=&\H{0,m_2,\ldots,m_w}1+\int_x^1\frac{1}{t}\H{m_2,\ldots,m_w}{\frac{1}{t}}dt.\\
\end{eqnarray*}
Again we have to consider the case for $\H{0,\ldots,0}{\frac{1}{x}}$ separately, however we can handle this case since we can handle $\H{0}{\frac{1}{x}}.$
Since we know the transformation for weights $<w$ we can apply it to $\H{m_2,\ldots,m_w}{\frac{1}{t}}$ and finally we obtain the required weight $w$ identity by using the definition of the
harmonic polylogarithms.\\
The index $1$ in the index set leads to a branch point at $1$ and a branch cut discontinuity in the complex plane for $x\in(1,\infty).$ This corresponds to the branch point at $x=1$ 
and the branch cut discontinuity in the complex plane for $x\in(1,\infty)$ of $-\log(1-x)=\H{1}x.$ However the analytic properties of the logarithm are well known and we
can set for $0<x<1$ for instance
\begin{eqnarray}
\H{1}{\frac{1}{x}}&=&\H{1}{x}+\H{0}{x}-i\pi, \label{HStrafo1dx11}
\end{eqnarray}
by approaching the argument $\frac{1}{x}$ form the lower half complex plane.
The strategy now is as follows. If a harmonic polylogarithm has leading ones, we remove them and end up with harmonic polylogarithms without leading ones 
and powers of $\H{1}{\frac{1}{x}}.$ We know how to deal 
with the harmonic polylogarithms without leading ones  due to the previous part of this section and for the powers of $\H{1}{\frac{1}{x}}$ we can 
use~(\ref{HStrafo1dx11}).

\section{Power Series Expansion of Harmonic Polylogarithms}
\label{HSPowerExp}
Due to trailing zeros in the index, harmonic polylogarithms in general do not have regular Taylor series expansions. To be more precise, the expansion is separated into two parts, 
one in $x$ and one in $\log(x);$ see also Remark \ref{HSremlead1trail0}. E.g., it can be seen easily that trailing zeros are responsible for powers of $\log(x)$ in the expansion; see \cite{Remiddi2000}.
Subsequently, we only consider the case without trailing zeros in more detail. For weight one we get the following well known expansions.
\begin{lemma} For $x\in \left[0\right.,\left.1\right),$
\begin{eqnarray}
\H{1}{x}&=&-\log(1-x)=\sum_{i=1}^\infty{\frac{x^i}{i}},\nonumber\\
\H{-1}{x}&=&\log(1+x)=\sum_{i=1}^\infty{-\frac{(-x)^i}{i}}.\nonumber
\end{eqnarray}
\label{HSpowexp1}
\end{lemma}
For higher weights we proceed inductively (see \cite{Remiddi2000}). If we assume that $\ve m$ has no trailing zeros and that for $x\in[0,1]$ we have
$$
\H{\ve m}{x}=\sum_{i=1}^\infty{\frac{\sigma^i x^i}{i^a}\S{\ve n}i},
$$
then for $x\in \left[0\right.,\left.1\right)$ the following theorem holds.
\begin{thm} For $x\in \left[0\right.,\left.1\right),$
\begin{eqnarray}
\H{0,\ve m}{x}&=&\sum_{i=1}^\infty{\frac{\sigma^ix^i}{i^{a+1}}\S{\ve n}i},\nonumber\\
\H{1,\ve m}{x}&=&\sum_{i=1}^\infty{\frac{x^i}{i}\S{\sigma a ,\ve n}{i-1}}=\sum_{i=1}^\infty{\frac{x^i}{i}\S{\sigma a ,\ve n}i}-
							\sum_{i=1}^\infty{\frac{\sigma^ix^i}{i^{a+1}}\S{\ve n}i},\nonumber\\
\H{-1,\ve m}{x}&=&-\sum_{i=1}^\infty{\frac{(-1)^ix^i}{i}\S{-\sigma a ,\ve n}{i-1}}=-\sum_{i=1}^\infty{\frac{(-1)^ix^i}{i}\S{-\sigma a ,\ve n}i}+
							\sum_{i=1}^\infty{\frac{\sigma^ix^i}{i^{a+1}}\S{\ve n}i}.\nonumber
\end{eqnarray}
\label{HSpowexp}
\end{thm}
\begin{example} For $x\in \left[0,1\right],$
$$
\H{-1,1}x=\sum_{i = 1}^{\infty}\frac{x^i}{i^2} - \sum_{i = 1}^{\infty}\frac{{\left( -x \right) }^i\,\S{-1}{i}}{i}.
$$
\end{example}

Note that even the reverse direction is possible, \ie given a sum of the form
$$
\sum_i^{\infty}(\pm x)^i\frac{\S{\ve n}{i}}{i^k}
$$
with $k\in \N_0,$ we can find a linear combination of (possibly weighted) harmonic polylogarithms $h(x)$ such that for $x\in(0,1):$
$$
h(x)=\sum_i^{\infty}x^i\frac{\S{\ve n}{i}}{i^k}.
$$
\begin{example}For $x\in(0,1)$ we have
\begin{eqnarray*}
\sum_i^{\infty}x^i\frac{\S{1,-2}{i}}{i}=-\textnormal{H}_{0,0,0,-1}(x)-\textnormal{H}_{0,1,0,-1}(x)-\textnormal{H}_{1,0,0,-1}(x)-\textnormal{H}_{1,1,0,-1}(x).
\end{eqnarray*}
\end{example}

\subsection{Asymptotic Behavior of Harmonic Polylogarithms}
\label{HSasybeh}
Combining Section \ref{HS1dxx} together with the power series expansion of harmonic polylogarithms we can determine the asymptotic behavior of harmonic polylogarithms. Let us look at the harmonic polylogarithm 
$\H{\ve m}x$ and define $y:=\frac{1}{x}.$ Using Section \ref{HS1dxx} on $\H{\ve m}{\frac{1}{y}}=\H{\ve m}x$ we can rewrite $\H{\ve m}x$ in terns of harmonic polylogarithms at argument $y$ together with some constants.
Now we can get the power series expansion of the harmonic polylogarithms at argument $y$ about 0 easily using the previous part of this section. Since sending $x$ to infinity corresponds to sending $y$ to zero,
 we get the asymptotic behavior of $\H{\ve m}x.$
\begin{example}For $x \in (0, \infty)$ we have
 \begin{eqnarray*}
  \H{-1,0}x&=&2 \H{-1,0}1-\H{-1,0}{\frac{1}x}-\H{0,0}1+\H{0,0}{\frac{1}x}\\
	   &=&\frac{1}{2}\, \H0x^2-\H0x \left(\sum _{i=1}^{\infty } \frac{\left(-\frac{1}{x}\right)^{i}}{i}\right)-\sum _{i=1}^{\infty}\frac{\left(-\frac{1}{x}\right)^{i}}{i^2}\\
	   & &+2\, \H{-1,0}1.
 \end{eqnarray*}
\end{example}

\subsection{Values of Harmonic Polylogarithms at 1 Expressed by Harmonic Sums at Infinity}
\label{HSinfval}
As worked out earlier, the expansion of the harmonic polylogarithms without trailing zeros is a combination of sums of the form:
$$
\sum_{i=1}^\infty x^i \frac{\sigma^i \S{\ve n}i}{i^a},\ \sigma \in \{-1,1\}, \ a \in \N.
$$
For $x\rightarrow1$ these sums turn into harmonic sums at infinity if $\sigma a\neq1$:
$$
\sum_{i=1}^\infty x^i \frac{\sigma^i \S{\ve n}i}{i^a}\rightarrow \S{\sigma a,\ve n}{\infty}.
$$
\begin{example}
$$
\H{-1,1,0}{1}=-2\S{3}{\infty}+\S{-1,-2}{\infty}+\S{-2,-1}{\infty}.
$$
\end{example}
If $\sigma a=1,$ these sums turn into
$$
\sum_{i=1}^\infty x^i \frac{\S{\ve n}i}{i},
$$
and sending $x$ to one gives:
$$
\lim_{x\rightarrow1}\sum_{i=1}^\infty x^i \frac{\S{\ve n}i}{i}=\infty.
$$
We see that these limits do not exist: this corresponds to the infiniteness of the harmonic sums with leading ones: $\lim_{k\rightarrow \infty}\S{1,\ve n}{k}=\infty.$
Hence the values of harmonic polylogarithms at one are related to the values of the multiple harmonic sums at infinity (see \cite{Remiddi2000}). For further details see Section \ref{HSinfrelsec}.

\section{Integral Representation of Harmonic Sums}
\label{HSIntegralrep}
In this section we want to derive an integral representation of harmonic sums. First we will find a multidimensional integral which can afterwards be 
rewritten as a sum of one-dimensional integrals in terms of harmonic polylogarithms using integration by parts. We start with the base case of a harmonic sum of depth one.

\begin{lemma}
Let $m\in\N,$ and $n\in\N$ then
\begin{eqnarray*}
\S{1}n&=&\int_0^{1}{\frac{x_1^n-1}{x_1-1}dx_1}\\
\S{-1}n&=&\int_0^{1}{\frac{(-x_1)^n-1}{x_1+1}dx_1}\\
\S{m}n&=&\int_0^1{\frac{1}{x_m}\int_0^{x_m}{\frac{1}{x_{m-1}} \cdots \int_0^{x_3}{\frac{1}{x_2}\int_0^{x_2}{\frac{x_1^n-1}{x_1-1}dx_1}dx_2}\cdots dx_{m-1}}dx_m}\\
\S{-m}n&=&\int_0^1{\frac{1}{x_m}\int_0^{x_m}{\frac{1}{x_{m-1}} \cdots \int_0^{x_3}{\frac{1}{x_2}\int_0^{x_2}{\frac{(-x_1)^n-1}{x_1+1}dx_1}dx_2}\cdots dx_{m-1}}dx_m}.
\end{eqnarray*}
\label{HSintrep1}
\end{lemma}
\begin{proof}
We proceed by induction on the weight $m.$ For $m=1$ we get:
\begin{eqnarray*}
\int_0^{1}{\frac{x_1^n-1}{x_1-1}dx_1}&=&\int_0^{1}{\sum_{i=0}^{n-1}{{x_1}^i}dx_1}=\sum_{i=0}^{n-1}\int_0^{1}{{{x_1}^i}dx_1}=\sum_{i=0}^{n-1}{\frac{1}{i+1}}
	=\sum_{i=1}^{n}{\frac{1}{i}}=\S{1}n
\end{eqnarray*}
and
\begin{eqnarray*}
\int_0^{1}{\frac{(-x_1)^n-1}{x_1+1}dx_1}&=&-\int_0^{1}{\sum_{i=0}^{n-1}{(-x_1)^i}dx_1}=-\sum_{i=0}^{n-1}\int_0^{1}{{(-x_1)^i}dx_1}=-\sum_{i=0}^{n-1}{\frac{(-1)^i}{i+1}}\\
	&=&\sum_{i=1}^{n}{\frac{(-1)^i}{i}}=\S{-1}n.
\end{eqnarray*}
Now suppose the theorem holds for weight $m-1$. We get
\begin{eqnarray*}
&&\int_0^1{\frac{1}{x_m}\int_0^{x_m}{\frac{1}{x_{m-1}} \cdots \int_0^{x_3}{\frac{1}{x_2}\int_0^{x_2}{\frac{x_1^n-1}{x_1-1}dx_1}dx_2}\cdots dx_{m-1}}dx_m}=\\
&&\hspace{0.5cm}\int_0^1{\frac{\SS{m-1}{x_m}n}{x_m}dx_m}=\int_0^1{\sum_{i=1}^n\frac{{x_m}^{i-1}}{i^{m-1}}dx_m}=\sum_{i=1}^n\int_0^1{\frac{{x_m}^{i-1}}{i^{m-1}}dx_m}=\\
&&\hspace{0.5cm}\sum_{i=1}^n\frac{1}{i}\frac{1}{i^{m-1}}=\S{m}n.
\end{eqnarray*}
For $-m$ the theorem follows similarly.
\end{proof}
Now we can look at the general cases:
\begin{thm}
For $m_i\in\N,$ $b_i\in \{-1,1\}$ and $n\in\N,$ then
\begin{eqnarray*}
&&\S{b_1m_1,b_2m_2,\ldots, b_km_k}n=\\
&&\hspace{0.4cm}\int_0^{b_1\cdots b_k}{\frac{dx_{k}^{m_k}}{x_{k}^{m_k}}\int_0^{x_{k}^{m_k}}{\frac{dx_{k}^{m_k-1}}{x_{k}^{m_k-1}} \cdots
\int_0^{x_{k}^3}{\frac{dx_{k}^2}{x_{k}^2}\int_0^{x_{k}^2}{\frac{dx_{k}^1}{x_{k}^1-b_1\cdots b_{k-1}}}}}}\\
&&\hspace{0.4cm}\int^{x_{k}^1}_{b_1\cdots b_{k-1}}{\frac{dx_{k-1}^{m_{k-1}}}{x_{{k-1}}^{m_{k-1}}}\int_0^{x_{{k-1}}^{m_{k-1}}}{\frac{dx_{{k-1}}^{m_{k-1}-1}}{x_{{k-1}}^{m_{k-1}-1}} \cdots
\int_0^{x_{{k-1}}^3}{\frac{dx_{{k-1}}^2}{x_{{k-1}}^2}}}}\int_0^{x_{{k-1}}^2}\hspace{-0.4em}{\frac{dx_{{k-1}}^1} {x_{{k-1}}^1-b_1\cdots b_{k-2}}}\\
&&\hspace{0.4cm}\int^{x_{{k-1}}^1}_{b_1\cdots b_{k-2}}{\frac{dx_{{k-2}}^{m_{k-2}}}{x_{{k-2}}^{m_{k-2}}}\int_0^{x_{{k-2}}^{m_{k-2}}}{\frac{dx_{{k-2}}^{m_{k-2}-1}}{x_{{k-2}}^{m_{k-2}-1}} \cdots
\int_0^{x_{{k-2}}^3}{\frac{dx_{{k-2}}^2}{x_{{k-2}}^2}}}}\int_0^{x_{{k-2}}^2}\hspace{-0.4em}{\frac{dx_{{k-2}}^1} {x_{{k-2}}^1-b_1\cdots b_{k-3}}}\\
&&\vspace{0.1cm}\\
&&\hspace{0cm}\hbox to 0.4\textwidth{}\vdots\\
&&\vspace{0.1cm}\\
&&\hspace{0.4cm}\int^{x_{4}^1}_{b_1b_2b_3}{\frac{dx_{3}^{m_3}}{x_{3}^{m_3}}\int_0^{x_{3}^{m_3}}{\frac{dx_{3}^{m_3-1}}{x_{3}^{m_3-1}} \cdots \int_0^{x_{3}^3}{\frac{dx_{3}^2}{x_{3}^2}\int_0^{x_{3}^2}{\frac{dx_{3}^1}{x_{3}^1-b_1b_2}}}}}\\
&&\hspace{0.4cm}\int^{x_{3}^1}_{b_1b_2}{\frac{dx_{2}^{m_2}}{x_{2}^{m_2}}\int_0^{x_{2}^{m_2}}{\frac{dx_{2}^{m_2-1}}{x_{2}^{m_2-1}} \cdots \int_0^{x_{2}^3}{\frac{dx_{2}^2}{x_{2}^2}\int_0^{x_{2}^2}{\frac{dx_{2}^1}{x_{2}^1-b_1}}}}}\\
&&\hspace{0.4cm}\int^{x_{2}^1}_{b_1}{\frac{dx_{1}^{m_1}}{x_{1}^{m_1}}\int_0^{x_{1}^{m_1}}{\frac{dx_{1}^{m_1-1}}{x_{1}^{m_1-1}} \cdots \int_0^{x_{1}^3}{\frac{dx_{1}^2}{x_{1}^2}\int_0^{x_{1}^2}{\frac{\left({x_{1}^1}\right)^n-1}{x_{1}^1-1}dx_{1}^1}}}}.
\end{eqnarray*}
\label{HSintrep}
\end{thm}
We will not give a proof of this theorem at this point, since it is a special case of Theorem~\ref{SSintrep}, which is proven in Chapter~\ref{SSchapter}. Applying integration 
by parts to this integral representation it is possible to arrive at one-dimensional integral representations containing harmonic polylogarithms; for details we again 
refer to Chapter~\ref{SSchapter}. A different way to arrive at a one-dimensional integral representation in terms of harmonic polylogarithms of harmonic sums is shown in the following subsection.

\subsection{Mellin Transformation of Harmonic Polylogarithms}
\label{HSMelTrafo}
In this section we look at the so called Mellin-transform of harmonic polylogarithms; compare~\cite{Paris2001}.
\begin{definition}[Mellin-Transform]
Let $f(x)$ be a locally integrable function on $(0,1)$ and $n \in \N$. Then the Mellin transform is defined by:
\begin{eqnarray}
\M{f(x)}{n}&=&\int_0^1{x^nf(x)dx}, \ \textnormal{when the integral converges.}\nonumber
\label{HSabmell}
\end{eqnarray}
\end{definition}

\begin{remark}
A locally integrable function on $(0,1)$ is one that is absolutely integrable on all closed subintervals of $(0,1)$.
\end{remark}

For $f(x)=1/(1-x)$ the Mellin transform is not defined since the integral $\int_0^1\frac{x^n}{1-x}$ does not converge. We modify the definition of the Mellin transform to include 
functions like $1/(1-x)$ as follows (compare \cite{Bluemlein1999}).

\begin{definition}
Let $h(x)$ be a harmonic polylogarithm or $h(x)=1$ for $x\in [0,1]$. Then we define the extended and modified Mellin-transform as follows:
\begin{eqnarray}
\Mp{h(x)}{n}&=&\M{h(x)}{n}=\int_0^1{x^nh(x)dx},\nonumber\\
\Mp{\frac{h(x)}{1+x}}n&=&\M{\frac{h(x)}{1+x}}{n}=\int_0^1{\frac{x^nh(x)}{1+x}dx},\nonumber\\
\Mp{\frac{h(x)}{1-x}}n&=&\int_0^1{\frac{(x^n-1)h(x)}{1-x}dx}.\nonumber
\end{eqnarray}
\label{HSabmellplus}
\end{definition}

\begin{remark}
Note that in \cite{Ablinger2009} and \cite{Remiddi2000} the Mellin-transform is defined slightly differently.
\end{remark}

\begin{remark}
\label{HSMelRemark}
It is possible that for a harmonic polylogarithm the original Mellin transform $\M{\frac{\H{m_1,m_2,...,m_k}x}{1-x}}n$ exists, but also in this case we modified the definition. 
Namely, we get
\begin{equation}
\M{\frac{\H{\ve m}x}{1-x}}n=\Mp{\frac{\H{\ve m}x}{1-x}}n+\H{1,\ve m}1.\nonumber
\end{equation}
Since $\Mp{\frac{\H{\ve m}x}{1-x}}n$ converges as well, $\H{1,\ve m}1$ has to be finite. But this is only the case if $\H{1,\ve m}x=\H{1,0,\ldots,0}x$. Hence we modified the definition only 
for functions of the form $\frac{\H{0,0,\ldots,0}x}{1-x}$ all the other cases are extensions.
\end{remark}

\begin{remark}
From now on we will call the extended and modified Mellin transform $M^+$ just Mellin transform and we will write $M$ instead of $M^+.$
\end{remark}

In the rest of this section we want to study how we can actually calculate the Mellin transform of harmonic polylogarithms weighted by $1/(1 + x)$ or $1/(1-x)$, \ie we want to represent the Mellin transforms 
in terms of harmonic sums at $n,$ harmonic sums at infinity and harmonic polylogarithms at one. This will 
be possible due to the following lemmas. We will proceed recursively on the depth of the harmonic polylogarithms.

\begin{lemma} For $n\in \N$ and $\ve m\in \left\{0,-1,1 \right\}^k$,
\begin{eqnarray}
\M{\frac{\H{\ve m}{x}}{1-x}}n &=&-n\cdot\M{\H{1,\ve m}{x}}{n-1}, \nonumber\\
\M{\frac{\H{\ve m}{x}}{1+x}}n &=&-n\cdot\M{\H{-1,\ve m}{x}}{n-1}+\H{-1,\ve m}{1}.\nonumber
\end{eqnarray}
\label{HSmelweighted}
\end{lemma}
\begin{proof}
We have
\begin{eqnarray*}
\int_0^1 x^n \H{1,\ve m}{x}dx
		&=&\frac{1}{n+1}\lim_{\epsilon \rightarrow 1^-}\left(\epsilon^{n+1}\H{1,\ve m}{\epsilon}-\int_0^{\epsilon}
			{\frac{x^{n+1}-1}{1-x}\H{\ve m}x dx}-\H{1,\ve m}{\epsilon} \right)\nonumber\\
		&=&-\frac{1}{n+1}\int_0^{1}{\frac{x^{n+1}-1}{1-x}\H{\ve m}x dx}\nonumber\\
		&=&-\frac{1}{n+1}\M{\frac{\H{1,\ve m}{x}}{1-x}}{n+1}\nonumber
\end{eqnarray*}
and
\begin{eqnarray*}
\int_0^1 x^n \H{-1,\ve m}{x}dx
		&=&\frac{\H{-1,\ve m}1}{n+1}-\frac{1}{n+1}\int_0^1 \frac{x^{n+1}}{(1+x)}\H{\ve m}{x}dx\nonumber\\
		&=&\frac{\H{-1,\ve m}1}{n+1}-\frac{1}{n+1}\M{\frac{\H{\ve m}{x}}{1+x}}{n+1}.\nonumber
\end{eqnarray*}
Hence we get 
\begin{eqnarray}
\M{\frac{\H{\ve m}{x}}{1-x}}{n+1} &=&-(n+1)\M{\H{1,\ve m}{x}}{n}, \nonumber\\
\M{\frac{\H{\ve m}{x}}{1+x}}{n+1} &=&-(n+1)\M{\H{-1,\ve m}{x}}{n}+\H{-1,\ve m}{1}.\nonumber
\end{eqnarray}
This is equivalent to the lemma.
\end{proof}

Due to Lemma \ref{HSmelweighted} we are able to reduce the calculation of the Mellin transform of harmonic polylogarithms weighted by $1/(1 + x)$ or $1/(1-x)$ to the calculation of the Mellin transform 
which are not weighted, \ie to the calculation of expressions of the form $\M{\H{\ve m}{x}}n.$ To calculate $\M{\H{\ve m}{x}}n$ we proceed by recursion. First let us state the base cases \ie the Mellin 
transforms of extended harmonic polylogarithms with weight one \cite{Remiddi2000}; see also \cite{Ablinger2009}.
\begin{lemma} For $n\in \N$ we have
\begin{eqnarray}
\M{\H{0}{x}}n&=& -\frac{1}{(n+1)^2}, \nonumber\\
\M{\H{1}{x}}n&=& \frac{\S1{n+1}}{n+1}, \nonumber\\
\M{\H{-1}{x}}n&=& (-1)^n\frac{\S{-1}{n+1}}{n+1}+\frac{\H{-1}1}{n+1}(1+(-1)^n). \nonumber
\end{eqnarray}
\label{HSweight1mel}
\end{lemma}
The higher depth results can now be obtained by recursion.
\begin{lemma} For $n\in \N$ and $\ve m\in \left\{0,-1,1 \right\}^k$,
\begin{eqnarray}
\M{\H{0,\ve m}{x}}n &=&\frac{\H{0,\ve m}1}{n+1}-\frac{1}{n+1} \M{\H{\ve m}{x}}n, \nonumber\\
\M{\H{1,\ve m}{x}}n &=&\frac{1}{n+1} \sum_{i=0}^n{\M{\H{\ve m}{x}}n}, \nonumber\\
\M{\H{-1,\ve m}{x}}n &=&\frac{1+(-1)^n}{n+1}\H{-1,\ve m}{1}-\frac{(-1)^n}{n+1} \sum_{i=0}^n{(-1)^i\M{\H{\ve m}{x}}n}. \nonumber	
\end{eqnarray}
\label{HSmelnotweighted}
\end{lemma}
\begin{proof}
We get the following results by integration by parts:
\begin{eqnarray}
\int_0^1 x^n \H{0,\ve m}{x}dx &=& \left. \frac{x^{n+1}}{n+1}\H{0,\ve m}{x}\right|_0^1-\int_0^1 \frac{x^n}{n+1}\H{\ve m}{x}dx\nonumber\\
				&=&\frac{\H{0,\ve m}1}{n+1}-\frac{1}{n+1} \M{\H{\ve m}{x}}n, \nonumber\\
\int_0^1 x^n \H{1,\ve m}{x}dx 
		&=&\lim_{\epsilon \rightarrow 1^-} \int_0^{\epsilon} x^n \H{1,\ve m}{x}dx \nonumber\\
		&=&\lim_{\epsilon \rightarrow 1^-}\left( \left. \frac{x^{n+1}}{n+1}\H{1,\ve m}{x}\right|_0^{\epsilon}-\int_0^{\epsilon}
			\frac{x^{n+1}}{(n+1)(1-x)}\H{\ve m}{x}dx \right) \nonumber\\
		&=&\frac{1}{n+1}\lim_{\epsilon \rightarrow 1^-}\left(\epsilon^{n+1}\H{1,\ve m}{\epsilon}-\int_0^{\epsilon}
			{\frac{x^{n+1}-1}{1-x}\H{\ve m}x dx}-\H{1,\ve m}{\epsilon} \right) \nonumber\\
		&=&\frac{1}{n+1}\left(\lim_{\epsilon \rightarrow 1^-}(\epsilon^{n+1}-1)\H{1,\ve m}{\epsilon}+\lim_{\epsilon \rightarrow 1^-}
			{\int_0^{\epsilon}\sum_{i=0}^n{x^i\H{\ve m}x}dx}\right)\nonumber\\
		&=&\frac{1}{n+1}\left(0+\int_0^{1}\sum_{i=0}^n{x^i\H{\ve m}x}dx\right)\nonumber\\
		&=&\frac{1}{n+1} \sum_{i=0}^n{\M{\H{\ve m}{x}}n}, \nonumber\\
\int_0^1 x^n \H{-1,\ve m}{x}dx
		&=&\left. \frac{x^{n+1}}{n+1}\H{-1,\ve m}{x}\right|_0^1-\int_0^1 \frac{x^{n+1}}{(n+1)(1+x)}\H{\ve m}{x}dx\nonumber\\
		&=&\frac{\H{-1,\ve m}1}{n+1}-\frac{1}{n+1}\left( \int_0^1 \frac{x^{n+1}+(-1)^n}{1+x}\H{\ve m}{x}dx\right.\nonumber\\
		& &\left.-(-1)^n\int_0^1{\frac{\H{\ve m}x}{1+x}dx}\right)=\frac{\H{-1,\ve m}1+(-1)^n\H{-1,\ve m}1}{n+1}\nonumber\\
		& &-\frac{(-1)^n}{n+1}\int_0^1{\sum_{i=0}^n{(-1)^ix^i\H{\ve m}x}dx}\nonumber\\
		&=&\frac{1+(-1)^n}{n+1}\H{-1,\ve m}{1}-\frac{(-1)^n}{n+1} \sum_{i=0}^n{(-1)^i\M{\H{\ve m}{x}}n}. \nonumber
\end{eqnarray}
\end{proof}

\begin{remark}
Combining the Lemmas \ref{HSmelweighted}, \ref{HSweight1mel} and \ref{HSmelnotweighted} we can represent the Mellin transforms of Definition \ref{HSabmellplus} 
in terms of harmonic sums at $n,$ harmonic sums at infinity and harmonic polylogarithms at one.
\end{remark}

\begin{remark}
As already mentioned (see Remark \ref{HSMelRemark}), our definition of the Mellin transform slightly differs from the definitions in \cite{Ablinger2009} and \cite{Remiddi2000}, however we can easily 
connect these different definitions. In fact the difference of the results is just a constant; this constant can be computed and we can use the strategy presented here 
as well to compute the Mellin-transform as defined in \cite{Ablinger2009} and \cite{Remiddi2000}.
\end{remark}

\subsection{The Inverse Mellin Transform}
\label{HSInvMellin}
In this subsection we want to summarize briefly how the inverse Mellin transform of harmonic sums can be computed. For details we refer to \cite{Ablinger2009,Remiddi2000}.
First we define as in \cite{Ablinger2009} an order on harmonic sums.
\begin{definition}[Order on harmonic sums]
Let $\S{\ve m_1}n$ and $\S{\ve m_2}n$ be harmonic sums with weights $w_1$, $w_2$ and depths $d_1$ and $d_2,$ respectively. Then
$$
		  	\begin{array}{ll}
						\S{\ve m_1}n \prec \S{\ve m_2}n, \ \textnormal{if} \ w_1<w_2, & \textnormal{or } \ (w_1=w_2 \ \textnormal{and} \ d_1<d_2).
			\end{array}
$$
We say that a harmonic sum $s_1$ is more complicated than a harmonic sum $s_2$ if $s_2 \prec s_1$.
For a set of harmonic sums we call a harmonic sum {\upshape most complicated} if it is a largest function with respect to $\prec$. 
\label{HSsord} 
\end{definition}
From \cite{Remiddi2000} (worked out in detail in \cite{Ablinger2009}) we know that there is only one {\itshape most complicated} harmonic sum in the the Mellin transform of a harmonic polylogarithm and this single {\itshape most complicated} 
harmonic sum can be calculated using, \eg Algorithm 2 of \cite{Ablinger2009}. But even the reverse direction is possible: given a harmonic sum, we can 
find a harmonic polylogarithm 
weighted by $1/(1-x)$ or $1/(1+x)$ such that the harmonic sum is the {\itshape most complicated} harmonic sum in the Mellin transform of this weighted harmonic polylogarithm. 
We will call this weighted harmonic polylogarithm the {\itshape most complicated} weighted harmonic polylogarithm in the inverse Mellin transform of the harmonic sum.
Equipped with these facts the computation of the inverse Mellin transform is relatively easy. As stated in \cite{Remiddi2000} we can proceed as follows:
\begin{itemize}
	\item Locate the most complicated harmonic sum.
	\item Construct the corresponding most complicated harmonic polylogarithm.
	\item Add it and subtract it.
	\item Perform the Mellin transform to the subtracted version. This will cancel the original harmonic sum.
	\item Repeat the above steps until there are no more harmonic sums.
	\item Let $c$ be the remaining constant term, replace it by $\textnormal{M}^{-1}(c)$, or equivalently, multiply $c$ by $\delta(1-x)$.
\end{itemize}
Here $\delta$ is the Dirac-$\delta$-distribution \cite{DeltaDist}.

\begin{remark}
 A second way to compute the inverse Mellin transform of a harmonic sum uses the integral representation of Theorem \ref{HSIntegralrep}. After repeated suitable integration by parts we
 can find the inverse Mellin transform.
\end{remark}

\begin{example}
\begin{eqnarray*}
 \S{-2}{n}&=&\int_0^1\frac{1}{y}\int_0^y\frac{(-x)^n-1}{x+1}dydx\\
	  &=&-\H{0,1}1+\left.\H{0}y\int_0^y\frac{(-x)^n}{x+1}dx\right|_0^y-\int_0^1\frac{(-x)^n \H{0}x}{1+x}dx\\
	  &=&-\H{0,1}1-(-1)^n\M{\frac{\H{0}x}{1+x}}{n}.
\end{eqnarray*}
\end{example}

\subsection{Differentiation of Harmonic Sums}
\label{HSdifferentiation}
Due to the previous section we are able to calculate the inverse Mellin transform of harmonic sums. It turned out that they are usually linear combinations of harmonic polylogarithms 
weighted by the factors $1/(1\pm x)$ and they can be distribution-valued (compare \cite{Ablinger2009}). By computing the Mellin transform of a harmonic sum we find in fact an 
analytic continuation of the sum to $n\in\R.$ For explicitly given analytic continuations see \cite{Bluemlein1999,Bluemlein2000,Bluemlein2004,Bluemlein2005,Bluemlein2008}. The 
differentiation of harmonic sums in the physic literature has been considered the first time in \cite{Bluemlein1999}. Worked out in \cite{Bluemlein2008,Bluemlein2009,Bluemlein2009a} 
this allows us to consider differentiation with respect to $n$, since we can differentiate the analytic continuation. Afterwards we may transform back to harmonic sums with 
the help of the Mellin transform. Differentiation turns out to be relatively easy if we represent the harmonic sum using its inverse Mellin transform as the following 
example suggests (see \cite[p.81]{Paris2001} and \cite{Ablinger2009}):
\begin{example}For $n \in \R$ we have
$$\frac{d}{d n}\int_0^1{x^n\H{-1}{x}dx}=\int_0^1{x^n\log{(x)}\H{-1}{x}dx}=\int_0^1{x^n\H{-1,0}{x}dx}+\int_0^1{x^n\H{0,-1}{x}dx}.$$
\end{example}

Based on this example, if we want to differentiate $\S{\ve a}n$ with respect to $n,$ we can proceed as follows:
\begin{itemize}
	\item Calculate the inverse Mellin transform of $\S{\ve a}n.$
	\item Set the constants to zero and multiply the terms of the inverse Mellin transform where $x^n$ is present by $\H0x$. This is
	 in fact differentiation with respect to $n$. 
	\item Use the Mellin transform to transform back the multiplied inverse Mellin transform of $\S{\ve a}n$ to harmonic sums. 
\end{itemize}

\begin{example}
Let us differentiate $\S{2,1}n.$ We have:
\begin{eqnarray*}
\S{2,1}n&=&\int_0^1{\frac{(x^n-1)\H{1,0}x}{1-x}dx}.
\end{eqnarray*}
Differentiating the right hand side with respect to $n$ yields:
\begin{eqnarray*}
&&\frac{d}{d n}\int_0^1{\frac{(x^n-1)\H{1,0}x}{1-x}dx}=\int_0^1{\frac{x^n\H{0}x\H{1,0}x}{1-x}dx}\\
&&\hspace{1cm}=\int_0^1{\frac{(x^n-1)\H{0}x\H{1,0}x}{1-x}dx}+\int_0^1{\frac{\H{0}x\H{1,0}x}{1-x}dx}\\
&&\hspace{1cm}= \M{\frac{\H{0,1,0}x}{1-x}}n+2\,\M{\frac{\H{1,0,0}x}{1-x}}n-\H{0,1,1,0}1\\
&&\hspace{1cm}=-\S{2,2}n-2\,\S{1,3}n+2\, \H{0,0,1}1\S{1}n-\H{0,1,1,0}1.
\end{eqnarray*}
\end{example}
To evaluate $\int_0^1\frac{\H{0}x\H{1,0}x}{1-x}dx$ we use the following lemma:
\begin{lemma}
For a harmonic polylogarithm $\H{\ve m}x$ the integral
$$
\int_0^1\frac{\H{0}x\H{\ve m}x}{1-x}dx
$$ 
exists and equals $-\H{0,1,\ve m}1$.
\label{HSdifflem}
\end{lemma}
\begin{proof}
Using integration by parts we get
\begin{eqnarray*}
\int_0^1\frac{\H{0}x\H{\ve m}x}{1-x}dx&=&\lim_{b \rightarrow 1^-}\lim_{a \rightarrow 0^+}{\int_a^{b}\frac{\H{0}x\H{\ve m}x}{1-x}dx}\\
	&=&\lim_{b \rightarrow 1^-}\lim_{a \rightarrow 0^+}\left(\left.\H{0}x\H{1,\ve m}x\right|_a^{b}-
		{\int_a^{b}\frac{\H{1,\ve m}x}{x}dx}\right)\\
	&=&\lim_{b \rightarrow 1^-}\left(\H{0}b\H{1,\ve m}b-\lim_{a \rightarrow 0^+}\H{0}a\H{1,\ve m}a
		-\H{0,1,\ve m}b\right)\\
	&=&\lim_{b \rightarrow 1^-}\left(\H{0}b\H{1,\ve m}b\right)-\H{0,1,\ve m}1=-\H{0,1,\ve m}1.
\end{eqnarray*}
\end{proof}

\section{Relations between Harmonic Sums}
In this section we look for the sake of completeness at various relations between harmonic sums as already done in \cite{Ablinger2009,Bluemlein2004,Bluemlein2008,Bluemlein2009,Bluemlein2009a}.
\subsection{Algebraic Relations}
\label{HSalgrel}
We already mentioned in Section \ref{HSdef} that harmonic sums form a quasi shuffle algebra. We will now take a closer look at the quasi shuffle algebra property (for details see
for example \cite{Hoffman} or \cite{Ablinger2009}). 
Subsequently, we define the following sets:
\begin{eqnarray}
\mathcal{S}(n)&=&\bigl\{q(s_1,\ldots,s_r)\bigl|r\in \N;  s_i \textnormal{ a harmonic sum at }n ;\ q\in \R[x_1,\ldots,x_r]\bigr\} \label{S}\nonumber\\
\mathcal{S}_p(n)&=&\bigl\{q(s_1,\ldots,s_r)\bigl|r\in \N;  s_i \textnormal{ a harmonic sum with positive indices at }n ;\nonumber\\\
&&\hspace{2.3cm}\ q\in \R[x_1,\ldots,x_r]\bigr\}. \nonumber\label{Sp}
\end{eqnarray}
From Section \ref{HSdef} we know that we can always expand a product of harmonic sums with the same upper summation limit into a linear combination of harmonic sums. Let $L(s_1,s_2)$ denote the expansion 
of $s_1s_2$ into a linear combination of harmonic sums.
Now we can define the ideals $\mathcal{I}$ and $\mathcal{I}_p$ on $\mathcal{S}(n)$ respectively $\mathcal{S}_p(n)$:
\begin{eqnarray*}
\mathcal{I}(n)&=&\left\{s_1s_2-L(s_1,s_2)\left| s_i \textnormal{ a harmonic sum at }n\right.\right\} \label{HSI}\\
\mathcal{I}_p(n)&=&\left\{s_1s_2-L(s_1,s_2)\left|s_i \textnormal{ a harmonic sum at }n \textnormal{ with positive indices}\right.\right\}. \label{HSIp}
\end{eqnarray*}
By construction $\mathcal{S}(n)/\mathcal{I}$ and $\mathcal{S}_p(n)/\mathcal{I}_p$ are quasi shuffle algebras (for further details we again refer to \cite{Ablinger2009}).
We remark that it has been shown in \cite{Minh2000} that 
$$\mathcal{S}_p(n)/\mathcal{I}_p \cong \mathcal{S}_p(n)/\sim $$ 
where $a(n)\sim b(n)\Leftrightarrow \forall k\in \N \ a(k)=b(k),\ie$ the quasi-shuffle algebra is equivalent to the harmonic sums considered as sequences.
To our knowledge it has not been shown so far that 
$$\mathcal{S}(n)/\mathcal{I} \cong \mathcal{S}(n)/\sim $$ 
where $a(n)\sim b(n)\Leftrightarrow \forall k\in \N \ a(k)=b(k)$. Nevertheless, we strongly believe in this fact; see also Remark \ref{X}.

For further considerations let $A$ be a totally ordered, graded set. The degree of $a\in A$ is denoted by $\abs{a}.$
Let $A^*$ denote the free monoid over $A$, i.e.,
		$$A^*=\left\{a_1 \cdots a_n | a_i \in A, n\geq 1\right\}\cup\left\{\epsilon \right\},$$
where $\epsilon$ is the empty word. We extend the degree function to $A^*$ by $\abs{a_1 a_2 \cdots a_n}=\abs{a_1}+\abs{a_2}+\cdots+\abs{a_n}$ for $a_i \in A$ and $\abs{\epsilon}=0.$
Suppose now that the set $A$ of letters is totally ordered by $<$.
We extend this order to words with the following lexicographic order:
$$
\left\{ 
		  	\begin{array}{ll}
						\ve u <\ve uv \ \textnormal{for}  \ v \in A^+\ \textnormal{and } \ve u \in A^*,& \\
						\ve w_1a_1\ve w_2<\ve w_1a_2\ve w_3, \ \textnormal{if} \, a_1<a_2 \ \textnormal{for}\,  a_1,a_2 \in A \textnormal{ and } w_1,w_2,w_3 \in A^*. & 
					\end{array} \right.
$$

\begin{definition}
We refer to elements of $A$ as \textit{letters} and to elements of $A^*$ as \textit{words}.
For a word $\ve w=pxs$ with $p,x,s \in A^*,$ $p$ is called a \textit{prefix} of $w$ if $p\neq \epsilon$ and any $s$ is called a \textit{suffix} of $\ve w$ if $s\neq \epsilon$.
A word $\ve w$ is called a \textit{Lyndon} word if it is smaller than any of its suffixes. The set of all \textit{Lyndon} words on $A$ is denoted by Lyndon($A$).
\label{ablyn}
\end{definition}
Inspired by \cite{Hoffman} it has been shown in \cite{Ablinger2009} directly that the quasi shuffle algebras $\mathcal{S}(n)/\mathcal{I}$ and $\mathcal{S}_p(n)/\mathcal{I}_p$
are the free polynomial algebras on the \textit{Lyndon} words with alphabet $A=\Z^*$ and $A=\N$ respectively, where we define the degree of a 
letter $\abs{\pm n}=n$ for all $n \in \N$ and we order the letters by $-n\prec n$ and $n \prec m$ for all $n, m \in \mathbb N$ with $\ n < m$, and extend this order 
lexicographically (see \cite{Hoffman,Ablinger2009}).
Hence the number of algebraic independent sums in $\mathcal{S}_p(n)/\mathcal{I}$, respectively $\mathcal{S}(n)/\mathcal{I}$ which we also call basis sums 
is connected to the number of \textit{Lyndon} words. 
In order to count the number of \textit{Lyndon} words belonging to a specific index set we use the second Witt formula \cite{Witt1937,Witt1956,Reutenauer1969}:
\begin{equation}
\label{HSWitt2}
l_n(n_1,\ldots,n_q)=\frac{1}{n}\sum_{d|n}{\mu(d)\frac{(\frac{n}{d})!}{(\frac{n_1}{d})!\cdots(\frac{n_q}{d})!}}, \ \ n=\sum_{k=1}^q{n_k};
\end{equation} 
here $n_i$ denotes the multiplicity of the indices that appear in the index set and
\begin{equation}
	\mu(n)=\left\{ 
		  	\begin{array}{ll}
						1\  & \textnormal{if } n = 1  \\
						0\  & \textnormal{if } n \textnormal{ is divided by the square of a prime}  \\
						(-1)^s\  & \textnormal{if } n \textnormal{ is the product of } s \textnormal{ different primes} 
					\end{array} \right. \nonumber\\
\label{abmue}
\end{equation}
is the M\"obius function.
The number of \textit{Lyndon} words of length $n$ over an alphabet of length $q$ is given by the first Witt 
formula \cite{Witt1937,Witt1956,Reutenauer1969}:
\begin{equation}
\label{HSWitt1}
l_n(q)=\frac{1}{n}\sum_{d|n}{\mu(d)q^{n/d}}.
\end{equation}
In order to use this formula to count the number of algebraic basis sums at a specific weight, we introduce a modified notation of harmonic sums compare \cite{Vermaseren1998} and \cite{Bluemlein2004}. Of course we can identify
a harmonic sum at $n$ by its index set viewed as a word:
$$
\S{a_1,a_2,\ldots,a_k}n \rightarrow a_1a_2\cdots a_k.
$$
In addition, we can represent each such word using the alphabet $\{-1,0,1\}$ where a letter zero which is followed by an letter that is
nonzero indicates that one should be added to the absolute value of the nonzero letter.
\begin{eqnarray*}
\S{a_1,a_2,\ldots,a_k}n &\rightarrow& a_1a_2\cdots a_k\rightarrow \underbrace{0\cdots0}_{\abs{a_1}-1 \times}\sign{a_1}\underbrace{0\cdots0}_{\abs{a_2}-1 \times}\sign{a_2}\cdots\underbrace{0\cdots0}_{\abs{a_k}-1 \times}\sign{a_k}\\
0100{-1}1&\rightarrow& \S{2,-3,1}n.
\end{eqnarray*}
Note that the last letter is either $-1$ or $1$,
The advantage of this notation is that a harmonic sum of weight $w$ is expressed by a word of length $w$ in the alphabet $\{-1,0,1\}$ and hence we can use (\ref{HSWitt1}) to count the number 
of basis harmonic sums at specified weight $w\geq 2:$
\begin{eqnarray}\label{HSalgbasnum1}
\frac{1}{w}\sum_{d|w}{\mu\left(\frac{w}{d}\right)3^d}.
\end{eqnarray}
For $w=3$ this formula gives 8 and hence there are 8 algebraic basis harmonic sums spanning the quasi shuffle algebra at weight 3. 

A method to find the basis sums together with the relations for the dependent sums was presented in \cite{Bluemlein2004}. Efficient implementations 
and further variations are developed in \cite{Ablinger2009}. Here we want to give an 
example for harmonic sums at weight $w=3$. 
\begin{example} At weight $w=3$ we have for instance the 8 basis sums:
$$\S{-3}{n},\S{3}{n},\S{-2,-1}{n},\S{-2,1}{n},\S{2,-1}{n},\S{2,1}{n},\S{-1,1,-1}{n},\S{-1,1,1}{n}$$
together with the relations for remaining sums:
\begin{eqnarray*}
\S{1,-2}{n}&=& \S{-3}{n}+\S{-2}{n} \S{1}{n}-\S{-2,1}{n}\\
\S{1,2}{n}&=& \S{1}{n}\S{2}{n}+\S{3}{n}-\S{2,1}{n}\\
\S{1,-1,-1}{n}&=& \frac{1}{2} \S{-2}{n} \S{-1}{n}-\frac{1}{2} \S{-1,1}{n} \S{-1}{n}+\frac{1}{2} \S{3}{n}+\frac{1}{2} \S{-2,-1}{n}\\&&+\S{1}{n} \S{-1,-1}{n}-\frac{1}{2}
   \S{2,1}{n}-\frac{1}{2} \S{-1,1,-1}{n}\\
\S{1,-1,1}{n}&=& \S{-3}{n}+\S{-1}{n} \S{2}{n}+\S{-2,1}{n}+\S{1}{n} \S{-1,1}{n}-\S{2,-1}{n}\\&&-2 \S{-1,1,1}{n}\\
\S{1,1,-1}{n}&=& \frac{1}{2} \S{-2}{n} \S{1}{n}-\frac{1}{2}\S{-1,1}{n} \S{1}{n}+\frac{1}{2} \S{1,-1}{n} \S{1}{n}\\&&-\frac{1}{2} \S{-1}{n} \S{2}{n}-\S{-2,1}{n}+\S{2,-1}{n}+\S{-1,1,1}{n}\\
\S{1,1,1}{n}&=& \frac{1}{3} \S{1}{n} \S{2}{n}+\frac{1}{3} \S{3}{n}+\frac{1}{3}\S{1}{n} \S{1,1}{n}\\
\S{-1,-2}{n}&=& \S{-2}{n} \S{-1}{n}+\S{3}{n}-\S{-2,-1}{n}\\
\S{-1,2}{n}&=& \S{-3}{n}+\S{-1}{n} \S{2}{n}-\S{2,-1}{n}\\
\S{-1,-1,-1}{n}&=& \frac{1}{3} \S{-3}{n}+\frac{1}{3} \S{-1}{n}\S{2}{n}+\frac{1}{3} \S{-1}{n} \S{-1,-1}{n}\\
\S{-1,-1,1}{n}&=& \frac{1}{2} \S{-2}{n} \S{-1}{n}+\frac{1}{2} \S{-1,1}{n} \S{-1}{n}+\frac{1}{2} \S{3}{n}-\frac{1}{2} \S{-2,-1}{n}\\&&+\frac{1}{2}\S{2,1}{n}-\frac{1}{2} \S{-1,1,-1}{n}.
\end{eqnarray*}
Hence we can use the 8 basis sums together with sums of lower weight to express all sums of weight $w=3$. Note that the sums of lower weight in this example are not yet reduced to a basis.
\label{HSRelationEx1}
\end{example}

\begin{remark}[compare \cite{Ablinger2009}]
\label{X}
In the difference field setting of $\Pi \Sigma$-fields one can verify algebraic independence of sums algorithmically for a particular given finite set of sums; 
see \cite{Schneider2008}. We could verify up to weight 8 that the figures ${\sf N_A}$ in Table \ref{HSnumberofbasissumstab} are correct, interpreting the objects 
in $\mathcal{S}(n)/\sim$ ,$\ie$ as sequences (see \cite{HSums2012}). Nevertheless, unless we do not have a rigorous proof for $\mathcal{S}(n)/\mathcal{I} \cong \mathcal{S}(n)/\sim$, 
we can only assume that the figures in Table \ref{HSnumberofbasissumstab} give an upper bound.\\
Using the difference field setting of $\Pi \Sigma$-fields it is possible to verify the algebraic independence of a special set of harmonic sums. As an example we could prove
that for $m \in \N$ the harmonic sums $\{S_{i,m}(n)\left| \right. i\geq 1\}$ are algebraically independent over $\Q(n)$ (for details we refer to \cite{JACS2012}).
\end{remark}

\subsection{Differential Relations}
\label{HSdiffrel}
In Section \ref{HSdifferentiation} we described the differentiation of harmonic sums with respect to the upper summation limit. 
The differentiation leads to new relation, for instance we find
$$
\frac{\partial}{\partial n}\S{2,1}n=\frac{7\zeta_2^2}{10}+\zeta_2\S2n-\S{2,2}n-2\S{3,1}n.
$$
\begin{example}[Example \ref{HSRelationEx1} continued]From differentiation we get the additional relations
\begin{eqnarray*}
\S{-3}{n}&=& \S{-3}{\infty}-\frac{1}{2} \frac{\partial}{\partial n}\S{-2}{n}\\
\S{3}{n}&=& \S{3}{\infty}-\frac{1}{2} \frac{\partial}{\partial n}\S{2}{n}\\
\S{-2,1}{n}&=& \frac{1}{2}\frac{\partial}{\partial n}\S{-2}{n}-\frac{\partial}{\partial n}\S{-1,1}{n}+\S{-1}{n} \S{2}{\infty}-\S{-1}{n} \S{2}{n}\\
	      &&+\S{2,-1}{n}-\S{-3}{\infty}-\S{3}{\infty}+\S{-2,-1}{\infty}.
\end{eqnarray*}
Hence we could reduce the number of basis sums at weight $w=3$ to $5$ by introducing differentiation. The basis sums are:
$$\S{-2,-1}{n},\S{2,-1}{n},\S{2,1}{n},\S{-1,1,-1}{n},\S{-1,1,1}{n}.$$
\label{HSRelationEx2}
\end{example}

Note that we collect the derivatives in 
%------------------------------------------------------------------------------------------------------------
\begin{eqnarray*}
\S{a_1,\ldots,a_k}{n}^{(D)} = \left\{\frac{\partial^N}{\partial n^N}\S{a_1,\ldots,a_k}{n}; N \in \N\right\},
\end{eqnarray*}
and identify an appearance of a derivative of a harmonic sum with the harmonic sum itself.
This makes sense, since a given finite harmonic sum $\S{a_1,\ldots,a_k}n$ is determined for $n \in \mathbb{C}$ by its asymptotic 
representation (see Section \ref{HSExpansion}) and 
the corresponding recursion from $n \rightarrow (n-1)$ and both, the asymptotic representation and the recursion
can be easily differentiated analytically. Hence the differentiation of a harmonic sum with respect to $n$
is closely related to the original sum. 

\subsection{Duplication Relations}
\label{HShalfrel}
Considering harmonic sums up to $n$ together with harmonic sums up to $2\cdot n$ we find new relations, which are summarized in the following theorem.
\begin{thm}(see \cite{Vermaseren1998,Ablinger2009})
Let $n,m\in\N$ and $a_i \in \N$ for $i \in \N.$ Then we have the following relation:
\begin{eqnarray}
\sum{\S{\pm a_m, \pm a_{m-1},\ldots,\pm a_1}{2\cdot n}}=\frac{1}{2^{\sum_{i=1}^m a_i-m}}\S{a_m,a_{m-1},\ldots,a_1}n
\label{HShalfint}
\end{eqnarray}
where we sum on the left hand side over the $2^m$ possible combinations concerning $\pm$.
\label{HSDuprel}
\end{thm}
As for differentiation, when we are counting basis elements, we identify an appearance of a harmonic sum $\S{\ve a}{2\cdot n}$ with the harmonic sum $\S{\ve a}{n}$. We can continue Example \ref{HSRelationEx2}:
\begin{example}[Example \ref{HSRelationEx2} continued]
\begin{eqnarray*}
\S{-2,-1}{n}&=& -\frac{1}{2} \frac{\partial}{\partial n}\S{-2}{n}+\frac{\partial}{\partial n}\S{-1,1}{n}+\frac{1}{2} \S{2,1}{\frac{n}{2}}-\S{-1}{n} \S{2}{\infty }-2\, \S{2,-1}{n}\\
	      &&+\S{-1}{n} \S{2}{n}-\S{2,1}{n}+\S{-3}{\infty}+\S{3}{\infty }-\S{-2,-1}{\infty}.
\end{eqnarray*}
Hence we could reduce the number of basis sums at weight $w=3$ to $4$ by introducing duplication. The basis sums are:
$$\S{2,-1}{n},\S{2,1}{n},\S{-1,1,-1}{n},\S{-1,1,1}{n}.$$
\label{HSRelationEx3}
\end{example}

\subsection{Number of Basis Elements}
\label{HSnumberofbasissums}
We know that the number of harmonic sums ${\sf N_S}(w)$ at a specified weight $w\geq 1$ is given by $${\sf N_S}(w) = 2\cdot 3^{w-1}.$$ This can be seen using the representation of harmonic sums
as words in the alphabet $\{-1,0,1\}$ as explained on page~\pageref{HSalgbasnum1}: the last letter is either $-1$ or $1,$ the other $w-1$ are chosen out of $\{-1,0,1\}$. 
Of course we have ${\sf N_S}(0)=0.$

On page~\pageref{HSalgbasnum1} we have already seen that the number of algebraic basis harmonic sums for weight $w\geq 2$ is given by
 $${\sf N_A}(w) = \frac{1}{w}\sum_{d|w}{\mu\left(\frac{w}{d}\right)3^d}.$$
Again we have ${\sf N_A}(0)=0$ and ${\sf N_A}(1)=2.$
This means we can express all the sums of weight $w$ by algebraic relations using ${\sf N_A}(w)$ harmonic sums of weight $w$ (of course we can not choose them randomly, 
but there is at least one choice) and sums of lower weight.

Using other types of relations (\ie those emerging from differentiation, or duplication relations) we can also count the numbers concerning these relations. Of course 
we can also combine the different types of relations and hence get a different number of basis harmonic sums (as we did in the previous sections for harmonic sums at weight 3). 
Here we want to present all the different formulas which arise using algebraic, differential and duplication relations.\\
We start with the formulas which arise if we do not mix the different relations. We already stated the number of algebraic basis harmonic sums, so it 
remains to give respective numbers concerning differential (${\sf N_D}(w)$) and duplication (${\sf N_H}(w)$) relations:
\begin{prop}
\begin{eqnarray}
{\sf N_D}(w) &=&\left\{ 
		\begin{array}{ll}
			4\cdot 3^{w-2},\ & w\geq 2  \\
			2,\ & w=1   
		\end{array} \right.\label{HSnd}\\
 {\sf N_H}(w) &=& 2\cdot 3^{w-1}-2^{w-1}. \ w\geq 1 \label{HSnh}
\end{eqnarray}
\end{prop}
\begin{proof}
 Let us start to prove (\ref{HSnh}). From Theorem \ref{HSDuprel} we know that there are as many duplication relations at weight $w$ as there are harmonic sums with only
 positive indices at weight $w.$ The harmonic sums of weight $w$ with positive indices can be viewed as the words of length $w$ out of an alphabet with two letters \ie $\{0,1\}$ 
 (compare page~\pageref{HSalgbasnum1}). Hence there are $2^{w-1}$ such relations. Since each sum of weight $w$ appears in exactly one of these relations, we can 
 express $2^{w-1}$ of the ${\sf N_S}(w)$ sums by using the remaining sums. Thus ${\sf N_H}(w)\leq 2\cdot 3^{w-1}-2^{w-1}$ and since we used all possible 
 relations we have ${\sf N_H}(w)= 2\cdot 3^{w-1}-2^{w-1}.$\\
In order to prove (\ref{HSnd}), we first take a look at the action of the differential operator $\frac{\partial}{\partial n}$ on a harmonic sum of weight $w-1$:
\begin{eqnarray*}
 \frac{\partial}{\partial n}(\S{a_1,a_2,\ldots,a_k}n)&=&-\abs{a_1}\S{a_1\wedge 1,a_2,\ldots,a_k}n-\abs{a_2}\S{a_1,a_2\wedge 1,\ldots,a_k}n\\ &&-\cdots-\abs{a_k}\S{a_1,a_2,\ldots,a_k\wedge 1}n+R(n)\\
 &=&-\sum_{i=1}^k{\abs{a_i}\S{a_1,\ldots,a_{i-1},a_i\wedge 1,a_{i+1},\ldots,a_k}n}+R(n);
\end{eqnarray*}
the harmonic sums popping up in $R(n)$ are of weight less then $w$ and $$a\wedge b=\sign a\sign b(\abs{a}+\abs{b}).$$ Note that differentiation increases the weight 
by $1.$ Hence only the differentiation of 
the ${\sf N_S}(w-1)$ sums of weight $w-1$ leads to relations including harmonic sums of weight $w.$
Let $\S{a_1,a_2,\ldots,a_k}n$ be a sum of weight $w-1.$ We use the differentiation of this sum to express, \eg the sum $\S{\sign{a_1}(\abs{a_1}+1),a_2,\ldots,a_k}n$ of 
weight $w.$ In this way we can express ${\sf N_S}(w-1)$ sums (note that we cannot express more sums because the number of relations equals ${\sf N_S}(w-1)$), hence
\begin{eqnarray*}
 {\sf N_D}(w)={\sf N_S}(w)-{\sf N_S}(w-1)&=&\left\{
		\begin{array}{ll}
			2\cdot 3^{w-1}-2\cdot 3^{w-2}=4\cdot 3^{w-2},\ & w\geq 2  \\
			2-0=2,\ & w=1.   
		\end{array} \right.
\end{eqnarray*}
\end{proof}
Note that we have of course ${\sf N_D}(0)=0$ and ${\sf N_H}(0)=0.$
Let us consider the formulas which arise if we mix two different types of relations. There are three different ways to combine 2 types of relations. We give the numbers 
using algebraic and differential (${\sf N_{AD}}(w)$), using algebraic and duplication (${\sf N_{AH}}(w)$) and using differential and duplication (${\sf N_{DH}}(w)$) relations:
\begin{prop}
\begin{eqnarray}
 {\sf N_{AD}}(w) &=&\left\{
		\begin{array}{ll} 
			\frac{1}{w}\displaystyle{\sum_{d|w}}{\mu\left(\frac{w}{d}\right)3^d}-\frac{1}{w-1}\sum_{d|(w-1)}{\mu\left(\frac{w-1}{d}\right)3^d},\ & w\geq 2  \\
			2,\ & w=1
		 \end{array} \right. \label{HSnad}\\
 {\sf N_{AH}}(w) &=& \frac{1}{w}\sum_{d|w}{\mu\left(\frac{w}{d}\right)\left[3^d-2^d\right]}\label{HSnah}\\
 {\sf N_{DH}}(w) &=&\left\{
		\begin{array}{ll} 
			4\cdot 3^{w-2}-2^{w-2},\ & w\geq 2  \\
			1,\ & w=1.
		\end{array} \right.\label{HSndh}
\end{eqnarray}
\end{prop}
\begin{proof}
 In order to prove (\ref{HSnad}), we first use the algebraic relations to get a basis out of ${\sf N_A}(w)$ harmonic sums for the sums at weight $w$. From the 
differentiation of the harmonic sums of weight $w-1$ we get ${\sf N_S}(w-1)$ new relations (indeed these relations are different and independent from the 
algebraic ones since in each relation there is a new symbol; the differential operator applied to a sum of weight $w-1,$ \ie  $\frac{\partial}{\partial n}\S{a_1,\ldots,a_k}n$). 
However not all of these new relations are independent since the sums of weight $w-1$ fulfill algebraic relations and so do the new symbols. Since there 
are ${\sf N_A}(w-1)$ algebraic basis sums of weight $w-1,$ we get the same number of independent relations. We can use these ${\sf N_A}(w-1)$ relations (these are all relations we can use). This reduces 
the number of algebraic harmonic sums and we get 
\begin{eqnarray*}
{\sf N_{AD}}(w)&=&{\sf N_{A}}(w)-{\sf N_{A}}(w-1)\\
		&=&\left\{
		\begin{array}{ll}
			\frac{1}{w}\sum_{d|w}{\mu\left(\frac{w}{d}\right)3^d}-\frac{1}{w-1}\sum_{d|w-1}{\mu\left(\frac{w-1}{d}\right)3^d},\ & w\geq 2  \\
			2-0=2,\ & w=1.
		\end{array} \right.
\end{eqnarray*}
In order to prove (\ref{HSnah}), we first use the algebraic relations to get a basis out of ${\sf N_A}(w)$ harmonic sums for the sums at weight $w$. From the 
duplication relations of the harmonic sums we get ${\sf N_H}(w)$ new relations (indeed these relations are different and independent from the algebraic ones 
since in each relation there is a new symbol; the operator $H$ applied to a sum of weight $w$ with positive indices, \ie  $H(\S{a_1,\ldots,a_k}n)$). However 
not all of these new relations are independent since the sums with positive indices of weight $w$ fulfill algebraic relations and so do the new symbols. 
The harmonic sums of weight $w$ with positive indices can be viewed as the words of length $w$ out of an alphabet with two letters \ie $\{0,1\}$ (compare page~\pageref{HSalgbasnum1}), hence there 
are $\frac{1}{w}\sum_{d|w}{\mu\left(\frac{w}{d}\right)2^d}$ algebraic basis sums for harmonic sums with positive indices (compare this to the explanation 
for ${\sf N_{A}}(w)$). Therefore we get $\frac{1}{w}\sum_{d|w}{\mu\left(\frac{w}{d}\right)2^d}$ new independent relations. We can use 
these $\frac{1}{w}\sum_{d|w}{\mu\left(\frac{w}{d}\right)2^d}$ relations (but not more) and reduce the number of algebraic harmonic sums and we get
\begin{eqnarray*}
 {\sf N_{AH}}(w)={\sf N_{A}}(w)-\frac{1}{w}\sum_{d|w}{\mu\left(\frac{w}{d}\right)2^d}=\frac{1}{w}\sum_{d|w}{\mu\left(\frac{w}{d}\right)\left[3^d-2^d\right]}.
\end{eqnarray*}
 In order to prove (\ref{HSndh}) we first use the duplication relations to get a basis out of ${\sf N_H}(w)$ harmonic sums for the sums at weight $w$. From the 
differentiation of the harmonic sums of weight $w-1$ we get ${\sf N_S}(w-1)$ new relations (as earlier these relations are different and independent from the 
duplication relations since in each relation there is a new symbol; the differential operator applied to a sum of weight $w-1,$ \ie  $\frac{\partial}{\partial n}\S{a_1,\ldots,a_k}n$). 
However not all of these new relations are independent since the sums of weight $w-1$ fulfill duplication relations and so do the new symbols. Since there 
are ${\sf N_H}(w-1)$ duplication basis sums of weight $w-1$ we get the same number of independent relations. We can use these ${\sf N_H}(w-1)$ relations (but not more). This reduces the 
number of basis duplication harmonic sums further and we get 
\begin{eqnarray*}
{\sf N_{DH}}(w)&=&{\sf N_{H}}(w)-{\sf N_{H}}(w-1)\\
		&=&\left\{
		\begin{array}{ll}
			(2\cdot 3^{w-1}-2^{w-1})-(2\cdot 3^{w-2}-2^{w-2})=4\cdot 3^{w-2}-2^{w-2} ,\ & w\geq 2  \\
			1-0=1,\ & w=1.
		\end{array}\right.
\end{eqnarray*}
\end{proof}
Note, obviously we have ${\sf N_{AD}}(0)=0,$ ${\sf N_{AH}}(0)=0$ and ${\sf N_{DH}}(0)=0.$
Let us now consider consider all three types of relations together:
\begin{prop}
\begin{eqnarray*}
{\sf N_{ADH}}(w) &=&\left\{ 
		\begin{array}{ll}
		\frac{1}{w}\displaystyle{\sum_{d|w}}{\mu\left(\frac{w}{d}\right)\left[3^d-2^d\right]}-\frac{1}{w-1}\sum_{d|w-1}{\mu\left(\frac{w-1}{d}\right)\left[3^d-2^d\right]},\ & w\geq 2  \\
		1,\ & w=1.
	\end{array} \right.%\label{HSnadh}
\end{eqnarray*}
\end{prop}
\begin{proof}
 In order to prove this formula, we first use the algebraic and the duplication relations to get a basis out of ${\sf N_{AH}}(w)$ harmonic sums for the sums at 
weight $w$. From the differentiation of the harmonic sums of weight $w-1$ we get ${\sf N_S}(w-1)$ new relations (indeed these relations are different and 
independent from the algebraic together with the duplication relations since in each relation there is a new symbol; the differential operator applied 
to a sum of weight $w-1,$ \ie  $\frac{\partial}{\partial n}\S{a_1,\ldots,a_k}n$). However not all of these new relations are independent since the sums of weight $w-1$ fulfill 
algebraic and duplication relations and so do the new symbols. Since there are ${\sf N_{AH}}(w-1)$ algebraic-duplication basis sums of weight $w-1,$ we get the 
same number of independent relations. We can use these ${\sf N_{AH}}(w-1)$ relations (but not more). This reduces the number of algebraic-duplication basis harmonic sums 
further and we get 
\begin{eqnarray*}
{\sf N_{ADH}}(w)&=&{\sf N_{AH}}(w)-{\sf N_{AH}}(w-1)\\
		&=&\left\{
		\begin{array}{ll}\frac{1}{w}\sum_{d|w}{\mu\left(\frac{w}{d}\right)\left[3^d-2^d\right]}
				  -\frac{1}{w-1}\sum_{d|(w-1)}{\mu\left(\frac{w-1}{d}\right)\left[3^d-2^d\right]},\ & w\geq 2  \\
		1-0=1,\ & w=1.
		\end{array}\right.
\end{eqnarray*}
\end{proof}

Let us finally summarize all these formulas:
\begin{eqnarray*}
 {\sf N_S}(w) &=& 2\cdot 3^{w-1}\\
 {\sf N_A}(w) &=& \frac{1}{w}\sum_{d|w}{\mu\left(\frac{w}{d}\right)3^d}\\
 {\sf N_D}(w) &=& 4\cdot 3^{w-2}\\
 {\sf N_H}(w) &=& 2\cdot 3^{w-1}-2^{w-1}\\
 {\sf N_{AD}}(w) &=& \frac{1}{w}\sum_{d|w}{\mu\left(\frac{w}{d}\right)3^d}-\frac{1}{w-1}\sum_{d|w-1}{\mu\left(\frac{w-1}{d}\right)3^d}\\
 {\sf N_{AH}}(w) &=& \frac{1}{w}\sum_{d|w}{\mu\left(\frac{w}{d}\right)\left[3^d-2^d\right]}\\
 {\sf N_{DH}}(w) &=& 4\cdot 3^{w-2}-2^{w-2}\\
 {\sf N_{ADH}}(w)&=&\frac{1}{w}\sum_{d|w}{\mu\left(\frac{w}{d}\right)\left[3^d-2^d\right]}-\frac{1}{w-1}\sum_{d|w-1}{\mu\left(\frac{w-1}{d}\right)\left[3^d-2^d\right]}.
\end{eqnarray*}
Concrete values up to weight 8 are given in Table \ref{HSnumberofbasissumstab}; note that the corresponding relations that lead to this amount of basis sums have been explicitely computed  using the 
package \ttfamily HarmonicSums \rmfamily and are now available within it.

\begin{table}\centering
\begin{tabular}{| r | r | r | r | r | r | r | r | r|}
\hline	
&  \multicolumn{8}{|c|}{Number of} \\
%\cline{2-9}
Weight& ${\sf N_S}$ & ${\sf N_A}$ & ${\sf N_D}$ & ${\sf N_H}$ & ${\sf N_{AD}}$ & ${\sf N_{AH}}$ & ${\sf N_{DH}}$ & ${\sf N_{ADH}}$\\
\hline	
  1 &    2 &   2 &   2 &   1 &   2 &   1 &   1 &  1 \\
  2 &    6 &   3 &   4 &   4 &   1 &   2 &   3 &  1 \\
  3 &   18 &   8 &  12 &  14 &   5 &   6 &  10 &  4 \\
  4 &   54 &  18 &  36 &  46 &  10 &  15 &  32 &  9 \\
  5 &  162 &  48 & 108 & 146 &  30 &  42 & 100 & 27 \\ 
  6 &  486 & 116 & 324 & 454 &  68 & 107 & 308 & 65 \\
  7 & 1458 & 312 & 972 &1394 & 196 & 294 & 940 &187 \\
  8 & 4374 & 810 &2916 &4246 & 498 & 780 &2852 &486 \\
\hline
\end{tabular}
\caption{\label{HSnumberofbasissumstab}Number of basis sums concerning different relations up to weight 8.}
\end{table}

\section{Harmonic Sums at Infinity}
\label{HSinfrelsec}
Not all harmonic sums are finite at infinity, since for example $\lim_{n\rightarrow \infty} \S1n$ does not exist.
In fact, we have the following lemma, compare \cite{Minh2000}:
\begin{lemma}
Let $a_1, a_2, \ldots a_p \in \Z^*$ for $p \in \N.$
The harmonic sum $\S{a_1,a_2,\ldots,a_p}n$ is convergent, when $n\rightarrow \infty$, if and only if $a_1 \neq 1.$
\label{HSconsumlem}
\end{lemma} 
Note that a refined (and generalized) version of this lemma is given in Theorem \ref{SSconsumthm}.
As already mentioned in Section \ref{HSdef} we are always able to express a harmonic sum as an expression consisting only of sums without leading ones and 
sums of the type $\S{\ve 1_w}n$ with $w\in \N.$ Since sums without leading ones are convergent, they do not produce any problems. For sums of the 
type $\S{\ve 1_w}n$ we have the following proposition which is a direct consequence of Corollary 3 of \cite{Dilcher1995} and Proposition 2.1 of \cite{Kirschenhofer1996}.

\begin{prop} For $w\geq 1$ and $n\in\N$, we have
$$\S{\ve 1_w}n=\sum_{i=1}^n(-1)^{i-1}\binom{n}{i}\frac{1}{i^w}=-\frac{1}{w}Y_w(\ldots,(j-1)!\S{j}n,\ldots),$$
where $Y_w(\ldots,x_i,\ldots)$ are the Bell polynomials.
\end{prop}
Note that the explicit formulas were found empirically in \cite{Bluemlein1999}; for a determinant evaluation formula we refer to \cite{Bluemlein1999,Bruno1881,Berndt1985}.
Using this proposition we can express each sum of the type $\S{\ve 1_w}n$ by sums of the form $\S{i}n,$ where $i\in\Z$. Hence we can decompose each 
harmonic sum $\S{a_1, a_2,\ldots,a_p}n$ in a univariate polynomial in $\S1n$ with coefficients in the convergent harmonic sums. So all divergences can be 
traced back to the basic divergence of $\S1n$ when $n\rightarrow\infty$.
\begin{notation}
We will define the symbol $\S{1}{\infty}:=\lim_{n\rightarrow\infty}\S{1}n,$ and every time it is present, we are in fact dealing with limit 
processes. For $k\in N$ with $k>1$ the harmonic sums $\S{k}{\infty}$ turn into zeta-values and we will sometimes write $\zeta_k$ for $\S{k}{\infty}.$
\end{notation} 
  
\begin{example}
\begin{eqnarray}
\lim_{n\rightarrow\infty}{\S{1,1,2}n}&=&\lim_{n\rightarrow\infty}\left\{\frac 1 2 \S{1}n^2 \S{2}n + \S{1}n (\S{3}n-\S{2, 1}n)\right.\nonumber\\
					& &- \left. \S{3, 1}n + \frac 1 2 \S{4}n + \S{2, 1, 1}n\right\}\nonumber \\
					&=& \S{1}{\infty}^2\frac{\zeta_2}2-\S{1}{\infty}\zeta_3+\frac{9\zeta_2^2}{10}.\nonumber
\end{eqnarray}
\end{example} 
For the computation of the actual values of the harmonic sums at infinity see \cite{Vermaseren1998}.

\subsection{Relations between Harmonic Sums at Infinity}\label{HSInfRelations}
In this section we just want to give several types of relations between the values of harmonic sums at infinity which are of importance in the following chapters. Note that harmonic sums at infinity are 
closely related to the multiple zeta values; for details 
we refer to \cite{Broadhurst2010,Vermaseren1998,MZV1,MZV2,MZV3,Zagier1994}.\\
\begin{itemize}
\item The first type of relations  originate from the algebraic relations of harmonic sums, see Section \ref{HSalgrel}. These relations remain valid when we consider 
 them at infinity. We will refer to this relations as the stuffle relations.
\item The duplication relations from Section \ref{HShalfrel} remain valid if we consider sums which are finite at infinity (\ie do not have leading ones), since it makes 
 no difference whether the argument is $\infty$ or $\frac{\infty}{2}.$
\item In \cite{Vermaseren1998} the following relation for not both $m_1=1$ and $k_1=1$ can be found:
\begin{eqnarray*}
\S{m_1,\ldots,m_p}{\infty}\S{k_1,\ldots,k_q}{\infty}&=&
	\lim_{n \rightarrow \infty}\sum_{i=1}^n\frac{\sign{k_1}^i \S{m_1,\ldots,m_p}{n-i}\S{k_2,\ldots,k_q}i} {i^{\abs{k_1}}}.
\end{eqnarray*}
Using
\begin{eqnarray*}
\sum_{i=1}^n\frac{\sign{k_1}^i \S{m_1,\ldots,m_p}{n-i}\S{k_2,\ldots,k_q}i} {i^{\abs{k_1}}}\\
&&\hspace{-3cm}=\sum_{a=1}^k \binom{\abs{k_1}+\abs{m_1}-1-a}{\abs{m_1}-1}\sum_{i=1}^n\frac{\sign{m_1}^i} {i^{\abs{m_1}+\abs{k_1}-a}}\\
	&&\hspace{-2cm}\sum_{j=1}^i\frac{\sign{\frac{k_1}{m_1}}^j\S{m_2,\ldots,m_p}{i-j}\S{k_2,\ldots,k_q}j} {j^a}\\
&&\hspace{-3cm}+\sum_{a=1}^{m_1} \binom{\abs{k_1}+\abs{m_1}-1-a}{\abs{m_1}-1}\sum_{i=1}^n\frac{\sign{m_1}^i} {i^{\abs{m_1}+\abs{k_1}-a}}\\
	&&\hspace{-2cm}\sum_{j=1}^i\frac{\sign{\frac{k_1}{m_1}}^j\S{m_2,\ldots,m_p}{j}\S{k_2,\ldots,k_q}{i-j}}{j^a},
\end{eqnarray*}
we can rewrite the right hand side in terms of harmonic sums. We will refer to these relations as the shuffle relations since one could also obtain them from the shuffle algebra of harmonic polylogarithms.
\item In \cite{Broadhurst2010} in addition so called generalized doubling relations are considered.
\end{itemize}

All these relations enable us to express harmonic sums at infinity up to weight 8 using the following 19 constants:
\begin{eqnarray*}
&&\{\S{-1}{\infty},\S{2}{\infty},\S{3}{\infty},\S{5}{\infty},\S{7}{\infty},\S{-1,1,1,1}{\infty},\S{-1,1,1,1,1}{\infty},\\
&&\hspace{0.1cm}\S{-1,1,1,1,1,1}{\infty},\S{-5,-1}{\infty},\S{-1,1,1,1,1,1,1}{\infty},\S{-5,1,1}{\infty},\S{5,-1,-1}{\infty},\\
&&\hspace{0.1cm}\S{-1,1,1,1,1,1,1,1}{\infty},\S{5,3}{\infty},\S{-7,-1}{\infty},\S{-5,-1,-1,-1}{\infty},\S{-5,-1,1,1}{\infty}\}.
\end{eqnarray*}
The way we chose these constants is not unique, one could look for other basis constants. Note, that the corresponding relations, that enables one to rewrite an expression with infinite harmonic sums up to 
weight 8 to an expression with these basis sums, are available within the \ttfamily HarmonicSums \rmfamily package.

\section{Asymptotic Expansion of Harmonic Sums}
\label{HSExpansion}
In \cite{Minh2000} an algorithm to compute asymptotic expansions of harmonic sums with positive indices was introduced, while in \cite{Albino2009} a different approach was used to compute asymptotic expansions 
of harmonic sums including negative indices.
In this section we will use a new approach and provide an efficient algorithm to compute asymptotic expansions of harmonic sums (including negative indices); this approach is 
inspired by \cite{Bluemlein2009,Bluemlein2009a}.

In general we say that the function $f:\R \rightarrow \R$ is expanded in an asymptotic series
$$
f(x) \sim \sum_{n=1}^{\infty}{\frac{a_n}{x^n}}, \ x \rightarrow \infty,
$$
where $a_n$ are constants from $\R$, if for all $N\geq 0$
$$
R_N(x)=f(x)-\sum_{n=0}^N{\frac{a_n}{x^n}}\in o\left(\frac{1}{x^N}\right), \ x \rightarrow \infty;
$$
note, here a function $g(x)\in o(G(x))\Leftrightarrow \lim\limits_{x\to \infty} \left |\frac{g(x)}{G(x)}\right| =0.$

In order to compute asymptotic series for harmonic sums, we will use their representations as Mellin transforms of harmonic polylogarithms. But 
first we have to state several details.

\subsection{Extending Harmonic Polylogarithms}
\label{HSExtended}
In the following we look again at the change of the variable $1-x\rightarrow x$ for harmonic polylogarithms (compare Section \ref{HS1xx}). This transformation can as well be found 
in \cite{Remiddi2000}, however we will generalize it to index sets including $-1$. In order to accomplish this extension, we will first introduce an extension 
for the harmonic polylogarithms. We introduce $2$ as a member of the index set as follows.

\begin{definition}
Let $m_i\in \{-1,0,1,2\}$; for $x\in (0,1)$ we extend the harmonic polylogarithms from Definition \ref{HShlogdef} by introducing in addition
$$\H{2}x=\int_0^x\frac{1}{2-y}dy$$
and
$$\H{2,m_2,\ldots,m_w}x=\int_0^x\frac{\H{m_2,\ldots,m_w}x}{2-y}dy.$$
\end{definition}

\begin{remark}
The extended harmonic polylogarithms are still analytic functions for $x\in (0,1)$ and the differentiation extends as follows:
$$\frac{d}{dx}\H{2,m_2,\ldots,m_w}x=\frac{\H{m_2,\ldots,m_w}x}{2-x}.$$
The product of these extended harmonic polylogarithms is still the shuffle product and still only harmonic polylogarithms with leading 1 and not followed 
only by zeroes are not finite at 1.\\
\end{remark}

\begin{example}
\begin{eqnarray*}
\H{1,2}x\H{0,2,-1}x&=&\H{0,1,2,-1,2}x+2\, \H{0,1,2,2,-1}x+\H{0,2,-1,1,2}x+\H{0,2,1,-1,2}x\\
			&&+\H{0,2,1,2,-1}x+\H{1,0,2,-1,2}x+2\, \H{1,0,2,2,-1}x+\H{1,2,0,2,-1}x.
\end{eqnarray*}
\end{example}

\begin{remark}
For further extensions of the index set, we refer to the later chapters.
\end{remark}

\subsubsection*{Extension of the \texorpdfstring{$1-x\rightarrow x$}{1-x->x} Transform}
\label{HS1xxextended}
Let us again look at the transformation $1-x\rightarrow x$ (compare Section \ref{HS1xx}). Proceeding recursively on the weight $w$ of the harmonic polylogarithm we have for $x\in(0,1)$
\begin{eqnarray}
\H{0}{1-x}&=&-\H{1}{x}\\
\H{1}{1-x}&=&-\H{0}{x}\label{HStrafo2}\\
\H{-1}{1-x}&=&\H{-1}1-\H{2}{x}\\
\H{2}{1-x}&=&\H{2}1-\H{-1}{x}.
\end{eqnarray}
Now let us look at higher weights $w>1.$ We consider $\H{m_1,m_2,\ldots,m_w}{1-x}$ and suppose that we can already apply the transform for harmonic 
polylogarithms of weight $<w.$ If $m_1=1$ we can remove leading ones and end up with harmonic polylogarithms without leading ones and powers 
of $\H{1}{1-x}.$ For the powers of $\H{1}{1-x}$ we can use (\ref{HStrafo2}); therefore, only the cases in which the first index $m_1\neq 1$ are to 
be considered. For $m_1=0$ we get (see \cite{Remiddi2000}):
\begin{eqnarray}
\H{0,m_2,\ldots,m_w}{1-x}&=&\int_0^{1-x}{\frac{\H{m_2,\ldots,m_w}y}{y}dy}\nonumber\\
		&=&\int_0^{1}{\frac{\H{m_2,\ldots,m_w}y}{y}dy}-\int_{1-x}^{1}{\frac{\H{m_2,\ldots,m_w}y}{y}dy}\nonumber\\
		&=&\H{0,m_2,\ldots,m_w}1-\int_{0}^{x}{\frac{\H{m_2,\ldots,m_w}{1-t}}{1-t}dt},
\end{eqnarray}
where the constant $\H{0,m_2,\ldots,m_w}1$ is finite. For $m_1=-1$ we get:
\begin{eqnarray}
\H{-1,m_2,\ldots,m_w}{1-x}&=&\int_0^{1-x}{\frac{\H{m_2,\ldots,m_w}y}{1+y}dy}\nonumber\\
		&=&\int_0^{1}{\frac{\H{m_2,\ldots,m_w}y}{1+y}dy}-\int_{1-x}^{1}{\frac{\H{m_2,\ldots,m_w}y}{1+y}dy}\nonumber\\
		&=&\H{-1,m_2,\ldots,m_w}1-\int_{0}^{x}{\frac{\H{m_2,\ldots,m_w}{1-t}}{2-t}dt},
\end{eqnarray}
where the constant $\H{-1,m_2,\ldots,m_w}1$ is finite. Finally for $m_1=2$ we get:
\begin{eqnarray}
\H{2,m_2,\ldots,m_w}{1-x}&=&\int_0^{1-x}{\frac{\H{m_2,\ldots,m_w}y}{2-y}dy}\nonumber\\
		&=&\int_0^{1}{\frac{\H{m_2,\ldots,m_w}y}{2-y}dy}-\int_{1-x}^{1}{\frac{\H{m_2,\ldots,m_w}y}{2-y}dy}\nonumber\\
		&=&\H{2,m_2,\ldots,m_w}1-\int_{0}^{x}{\frac{\H{m_2,\ldots,m_w}{1-t}}{1+t}dt},
\end{eqnarray}
where the constant $\H{2,m_2,\ldots,m_w}1$ is finite. Since we know the transform for weights $<w,$ we can apply it to $\H{m_2,\ldots,m_w}{1-t}$ and 
finally we obtain the required weight $w$ identity by using the definition of the extended harmonic polylogarithms.

\begin{example}For $x \in (0,1)$
 \begin{eqnarray*}
\H{-1,1}{1-x}&=&\H{-1,1}1+\H{2,0}x\\
\H{1,-1,2}{1-x}&=&-\H{0}x\H{-1,2}1+\H{2}1\H{0,2}x-\H{-1,1,2}1\\
	&&-\H{-1,2,1}1-\H{0,2,-1}x.
 \end{eqnarray*}
\end{example}

\begin{remark}
If we apply the transform $1-x\rightarrow x$ to the harmonic polylogarithm $\H{m_1,m_2,\ldots,m_w}{1-x}$ and if we afterwards expand all the 
products of polylogarithms at $x$ it is easy to see, that there will be just one harmonic polylogarithm $\H{n_1,n_2,\ldots,n_w}{x}$ at $x$ with 
weight $w.$ The index set of this harmonic polylogarithm changed as follows: 
\begin{itemize}
\item if $m_i=1$ then $n_i=0$
\item if $m_i=0$ then $n_i=1$
\item if $m_i=-1$ then $n_i=2$
\item if $m_i=2$ then $n_i=-1.$
\end{itemize}
\label{HStraforem}
\end{remark}

If we want to apply the transform $x\rightarrow 1-x$ to the harmonic polylogarithm $\H{n_1,n_2,\ldots,n_w}{x},$ we have two possibilities:
\begin{enumerate}
\item Set $x=1-y$; apply the transform $1-y\rightarrow y$; replace $y$ by $1-x.$
\item By Remark \ref{HStraforem} we construct $\H{m_1,m_2,\ldots,m_w}{1-x}$; we apply the transform to $\H{m_1,m_2,\ldots,m_w}{1-x}$; in the result we 
find $\H{n_1,n_2,\ldots,n_w}{x}.$ Note that there might still be other harmonic polylogarithms at $x$ but their weight is smaller than $w$; for all these harmonic 
polylogarithms we can apply the same strategy.
\end{enumerate}
The advantage of the second strategy is that if we apply first the transform $1-x\rightarrow x$ to $\H{m_1,m_2,\ldots,m_w}{1-x}$ and then $x\rightarrow 1-x$ to 
the result, we get back $\H{m_1,m_2,\ldots,m_w}{1-x}$ immediately. Using the first strategy we might have to know various relations to get back 
to $\H{m_1,m_2,\ldots,m_w}{1-x}.$ We will illustrate this with an example:

\begin{example}
Transforming $\H{1,0,1}{1-x}$ we get $\H{0}x\H{0,1}1-\H{0,1,0}x-2\,\H{0,1,1}1.$ If we transform back using the first strategy, we end up 
with $-\H{0,1,0}1-2\,\H{0,1,1}1+\H{1,0,1}{1-x};$ note that we introduced $-\H{0,1,0}1-2\,\H{0,1,1}1$ which turns out to be zero. But using the second strategy we do not 
need this additional knowledge to remove possibly introduced zeros.
\end{example}

\begin{lemma}
Let $x\in(0,1).$ If we perform the the transformation $1-x\rightarrow x$ to the harmonic polylogarithm $\H{m_1,m_2,\ldots,m_w}{1-x},$ where $m_i\neq -1$ and express it by 
harmonic polylogarithms with argument $x,$ then the 
index $2$ does not appear in the resulting harmonic polylogarithms with argument $x.$ It is however possible that $2$ appears in the resulting 
harmonic polylogarithms at $1$.
\end{lemma}
\begin{proof}
For $w=1$ we have:
\begin{eqnarray}
\H{0}{1-x}&=&-\H{1}{x}\nonumber\\
\H{1}{1-x}&=&-\H{0}{x}\nonumber\\
\H{2}{1-x}&=&\H{2}1-\H{-1}{x}.\label{HStrafolem1}
\end{eqnarray}
Assume the lemma holds for weights $<w.$ We have:
\begin{eqnarray}
\H{2,m_2,\ldots,m_w}{1-x}&=&\H{2,m_2,\ldots,m_w}1-\int_{0}^{x}{\frac{\H{m_2,\ldots,m_w}{1-t}}{1+t}dt}\label{HStrafolem2}\\
\H{0,m_2,\ldots,m_w}{1-x}&=&\H{0,m_2,\ldots,m_w}1-\int_{0}^{x}{\frac{\H{m_2,\ldots,m_w}{1-t}}{1-t}dt}\label{HStrafolem3}.
\end{eqnarray}
Since the lemma holds for weights $<w,$ it holds for $\H{m_2,\ldots,m_w}{1-t},$ thus it can be written in the right format.
If $m_1=1$ then we extract the leading ones (this does not generate the index -1) and afterwards we can use (\ref{HStrafolem1}), (\ref{HStrafolem2}) and (\ref{HStrafolem3}) to obtain the desired format.
\end{proof}

\begin{lemma}
If we perform the the transformation $x\rightarrow 1-x$ (no matter which strategy of the two we use) to the harmonic polylogarithm $\H{m_1,m_2,\ldots,m_w}{x},$ where $m_i~\neq~2,$ and express 
it by harmonic polylogarithms with argument $1-x,$ then the index $-1$ does not appear in the resulting harmonic polylogarithms with argument $1-x.$ It is however possible that $-1$ 
appears in the resulting harmonic polylogarithms at $1$.
\end{lemma}
\begin{proof}
We just give a proof for the second strategy.
For $w=1$ we have:
\begin{eqnarray}
\H{0}{x}&=&-\H{1}{1-x}\nonumber\\
\H{1}{x}&=&-\H{0}{1-x}\nonumber\\
\H{-1}{x}&=&\H{2}1-\H{2}{1-x}.
\end{eqnarray}
Assume the lemma holds for weights $<w$.
$\H{m_1,\ldots,m_w}{x}$ arises from $\H{n_1,\ldots,n_w}{1-x}$ where $n_i=2,0,1$ if $m_i=-1,1,0,$ so $n_i\neq -1.$ Hence, according to the previous lemma in the 
transform of $\H{n_1,\ldots,n_w}{1-x}$ there is no index $2$ in the harmonic polylogarithms at $x$. All the the harmonic polylogarithms in the transform 
at $x$ (except $\H{m_1,\ldots,m_w}{x}$ itself) have weight less than $w$ and therefore the index $-1$ does not appear in their transform (except in harmonic 
sums at $1$).
\end{proof}

\begin{remark}
If we want to remove the index $2$ from $\H{m_1,\ldots,m_w}{1},$ we can use the transform $x\rightarrow 1-x$ for $\H{m_1,\ldots,m_w}{x}$ and afterwards we 
replace $x$ by $1$ in the result. All the indices $2$ are removed.
\end{remark}

\subsection{Asymptotic Representations of Factorial Series}
\label{HShexp}
We can use the Mellin transform of harmonic polylogarithms to represent harmonic sums. Hence any partial result to derive asymptotic expansions of integrals of the form 
$$\int_0^1{\frac{x^n\H{\ve m}x}{1 + x}dx} \textnormal{ and } \int_0^1{\frac{(x^n-1)\H{\ve m}x}{1 - x}dx}$$
will support us to obtain asymptotic expansions of the harmonic sums. 
For the following we refer to \cite{Bluemlein2009,Nielsen1906,Landau}.
Let us look at the Mellin transform
$$
\M{\varphi(x)}{n}=\int_0^1{x^{n}\varphi(x)dx}
$$
where $\varphi(x)$ is analytic at $x=1$ and has the following Taylor series expansion about $x=1$,
$$
\varphi(1-x)=\sum_{k=0}^{\infty}{a_kx^k.}
$$
For $\Re(x)>0,$ $\M{\varphi(x)}{n}$ is given by the factorial series
$$
\M{\varphi(x)}{n}=\sum_{k=0}^{\infty}{\frac{a_{k}k!}{(n+1)(n+2)\ldots(n+k+1)}}.
$$
$\M{\varphi(x)}{n}$ has poles at the negative integers and one may continue $\M{\varphi(x)}{n}$ analytically to values of $n \in \C$ as a meromorphic function. 
The asymptotic representation is
\begin{eqnarray}\label{HSbk}
\M{\varphi(x)}{n} \sim \sum_{k=0}^{\infty}{\frac{b_k}{n^{k+1}}}, \ n \rightarrow \infty,
\end{eqnarray}
where 
\begin{eqnarray*}
b_0&=&a_0\\
b_k&=&\sum_{l=0}^{k-1}(-1)^{l+1}S_{k-l}^ka_{k-l}(k-l)!
\end{eqnarray*}
and $S_l^k$ are the Stirling numbers of 2nd kind.\\
\begin{example}[Asymptotic expansion of $\S{-1}n$]
First we express $\S{-1}n$ using the Mellin transformation:
\begin{eqnarray*}
\S{-1}n&=&\int_0^1{x^n\overbrace{\frac{1}{1+x}}^{\varphi(x)}dx}+\log{(2)}=\M{\varphi(x)}{n}+\log{(2)}.
\end{eqnarray*}
The Taylor series expansion of $\varphi{(1-x)}$ yields:
\begin{eqnarray*}
\varphi{(1-x)}&=&\frac{1}{2-x}=\sum_{k=0}^{\infty}{\frac{1}{2^{k+1}}x^k}.
\end{eqnarray*}
Hence the factorial series of $\M{\varphi(x)}{n}$ is given by
\begin{eqnarray*}
\M{\varphi(x)}{n}&=&\frac{1}{2(1+n)}+\frac{1}{4(1+n)(1+n)}+\frac{2}{8n(1+n)(2+n)(3+n)}\\
		&&\ +\;\frac{6}{16n(1+n)(2+n)(3+n)(4+n)}+\;\cdots,
\end{eqnarray*}
which leads to the following asymptotic representation of $\M{\varphi(x)}{n}:$ 
\begin{eqnarray*}
\M{\varphi(x)}{n}&=&\frac{1}{2n}-\frac{1}{4n^2}+\frac{1}{8n^4}+o\left(\frac{1}{n^6}\right).
\end{eqnarray*}
Combining the results yields
\begin{eqnarray*}
\S{-1}n&=&\log{(2)}+\frac{1}{2n}-\frac{1}{4n^2}+\frac{1}{8n^4}+o\left(\frac{1}{n^6}\right).
\end{eqnarray*}
\end{example}

If we are interested in accuracy of this expansion, the following considerations may help. From repeated integration by parts we get:
\begin{eqnarray*}
\M{\varphi(x)}{n}&=&\sum_{s=0}^{k}{\frac{a_{s}s!}{(n+1)(n+2)\ldots(n+k+1)}}+\frac{(-1)^{k+1}R_{k+1}(n)}{(n+1)(n+2)\cdots(n+k+1)}
\end{eqnarray*}
with 
$$R_{k+1}(n)=\int_0^1x^{n+k}\;\varphi^{(k)}(x)dx,$$
or similarly 
\begin{eqnarray*}
\M{\varphi(x)}{n}&=&\sum_{s=0}^{k}\frac{b_{s}}{n^{s+1}}+\frac{(-1)^{k+1}}{n^{k+1}}\underbrace{\int_0^1\;x^{n-1}(D_x)^k\left(x\;\varphi(x)\right)dx}_{\bar{R}_{k+1}(n):=}
\end{eqnarray*}
with $$D_x(f(x)):=x\;f'(x);$$
for the estimate $\bar{R}_{k+1}$ we can use in addition
$$
\hspace{10mm}(D_x)^k\left(x\;f(x)\right)=\sum_{i=0}^k g(i+1,k)x^{i+1}f^{(i)}(x),\hspace{5mm} \textnormal{with}\ g(i,k):=\sum_{j=1}^i(-1)^j \frac{j^{k+1}}{(i-j)!j!}.
$$

\begin{example}[Asymptotic expansion of $\S{-1}n$ continued] We get
\begin{eqnarray*}
\int_0^1{\frac{x^n}{1+x}}dx&=&\int_0^1{x^{n-1}\frac{x}{1+x}}dx=\left.\frac{x^{n}}{n}\frac{x}{1+x}\right|_0^1-\int_{0}^1\frac{x^{n}}{n}\frac{1}{(1+x)^2}dx\\
	&=&\frac{1}{2n}-\int_{0}^1 x^{n-1} \frac{x}{n(1+x)^2}dx\\
	&=&\frac{1}{2n}-\left.\frac{x^{n}}{n}\frac{x}{n(1+x)^2}\right|_0^1+\int_{0}^1 \frac{x^{n}}{n} \frac{(1-x)}{n(1+x)^3}dx\\
	&=&\frac{1}{2n}-\frac{1}{4n^2}+\int_{0}^1 x^{n-1} \frac{x(1-x)}{n^2(1+x)^3}dx\\
	&=&\cdots=\frac{1}{2n}-\frac{1}{4n^2}+\frac{1}{8n^4}-\frac{1}{4n^6}+\frac{1}{n^6}\bar{R}_6(n)
\end{eqnarray*}
with $\abs{\bar{R}_6(n)}\leq 2.$
Hence
$$\left|\int_0^1{\frac{x^{20}}{1+x}}dx-\frac{6240199}{256000000}\right|\leq\frac{1}{32000000}.$$
\end{example}
As the previous example suggests, we can use repeated integration by parts to get the $b_k$ of equation (\ref{HSbk}) directly.
Unfortunately not all integrals of the form 
$$\int_0^1{\frac{x^n\H{m_1,m_2,\ldots,m_{k-1},m_k}x}{1\pm x}dx}$$
can be expanded in that way up to arbitrary order, since $\frac{\H{m_1,m_2,\ldots,m_{k-1},m_k}x}{1\pm x}$ does not have to be analytic at $x=1$. 
In the following we will explore certain functions (see Lemma \ref{HSanalytic1} and Lemma \ref{HSanalytic2}), which are analytic at $1$ (and hence infinitely times differentiable). We will try afterwards to transform all 
the integrals $\int_0^1{\frac{x^n\H{m_1,m_2,\ldots,m_{k-1},m_k}x}{1\pm x}dx}$ 
as far as possible into a form given in terms of these analytic functions at $1$ and therefore we can use the method mentioned above to find an 
asymptotic representation.
In the following lemma we will summarize some properties of the $n$-th derivatives of functions that we will come across.
\begin{lemma}Let $f:D\rightarrow \R$ be an analytic function on $D$ with $D\subset\R.$ For $x\in D$ and $n\in \N$ we have,
 \begin{eqnarray*}
 \left(\frac{f(x)}{1-x}\right)^{(n)}&=&\frac{1}{(1-x)^{n+1}}\sum_{k=0}^n \frac{n!}{k!}f(x)^{(k)}(1-x)^k\\
 \left(\frac{f(x)}{1+x}\right)^{(n)}&=&\frac{(-1)^n}{(1+x)^{n+1}}\sum_{k=0}^n(-1)^k \frac{n!}{k!}f^{(k)}(x)(1+x)^k\\
 \left(\frac{f(x)}{x}\right)^{(n)}&=&\frac{(-1)^n}{x^{n+1}}\sum_{k=0}^n(-1)^k \frac{n!}{k!}f^{(k)}(x)x^k.\\
\end{eqnarray*}
\label{HSderivatives}
\end{lemma}
\begin{proof}
All three identities follow from the  Leibniz rule; for instance:
\begin{eqnarray*}
\left(\frac{f(x)}{1-x}\right)^{(n)}&=&\sum_{k=0}^n \binom{n}{k}f^{(k)}(x)\left(\frac{1}{1-x}\right)^{(n-k)}=\sum_{k=0}^n \binom{n}{k}f^{(k)}(x)\frac{(n-k)!}{(1-x)^{n-k+1}}\\
      &=&\frac{1}{(1-x)^{n+1}}\sum_{k=0}^n \frac{n!}{k!}f^{(k)}(x)(1-x)^k.
\end{eqnarray*}
\end{proof}

\begin{lemma}
Let $\H{m_1,m_2,\ldots,m_k}x$ be a harmonic polylogarithm with $m_i\in\{0,-1\}$ for $1\leq i\leq k.$ Then 
$$\H{m_1,m_2,\ldots,m_k}x,\; \frac{\H{m_1,m_2,\ldots,m_k}x}{1+x}$$
and
$$\frac{\H{m_1,m_2,\ldots,m_k}x-\H{m_1,m_2,\ldots,m_k}1}{1-x}$$
are analytic for $x \in (0,\infty).$
\label{HSanalytic1}
\end{lemma}
\begin{proof}
We know that $\H0x$ and $\H{-1}x$ are analytic $\in (0,\infty).$ Now we suppose the harmonic polylogarithms with indices $\in\{0,-1\}$ of depth $k-1$ are analytic and
consider the the $n$-th derivative of $\H{m_1,m_2,\ldots,m_k}x:$
If $m_1=-1,$ we get 
\begin{eqnarray*}
\H{-1,m_2,\ldots,m_k}x^{(n)}&=&\left(\frac{\H{m_2,\ldots,m_k}x}{1+x}\right)^{(n-1)}\\
    &=&\frac{(-1)^{n-1}}{(1+x)^{n}}\sum_{k=0}^{n-1}(-1)^k \frac{(n-1)!}{k!}\H{m_2,\ldots,m_k}x^{(k)}(1+x)^k
\end{eqnarray*}
using Lemma \ref{HSderivatives}. Similarly, if $m_1=0,$ we get
$$
\H{-1,m_2,\ldots,m_k}x^{(n)}=\frac{(-1)^{n-1}}{x^{n}}\sum_{k=0}^{n-1}(-1)^k \frac{(n-1)!}{k!}\H{m_2,\ldots,m_k}x^{(k)}x^k.
$$
In both cases $\H{m_1,m_2,\ldots,m_k}x$ is analytic for $x \in (0,\infty)$ since $\H{m_2,\ldots,m_k}x$ is analytic for $x \in (0,\infty).$
Moreover, $\frac{\H{m_1,m_2,\ldots,m_k}{x}}{1+x}$ is analytic since  $\H{m_1,m_2,\ldots,m_k}x$ and $\frac{1}{1+x}$ are analytic. Finally we look at
$$\frac{\H{m_1,m_2,\ldots,m_k}x-\H{m_1,m_2,\ldots,m_k}1}{1-x};$$
for $x \in (0,1)\cup (1,\infty)$ the analyticity can be seen as in the previous case. In the following we will use the abbreviation $f(x):=\H{m_1,m_2,\ldots,m_k}x-\H{m_1,m_2,\ldots,m_k}1.$
 For $x=1$ we get using Lemma \ref{HSderivatives} and by using de l'Hospital's rule $n+1$ times (note that $\frac{\partial^{n+1}}{\partial x^{n+1}}(1-x)^{n+1}=(-1)^{n+1}(n+1)!$):
\begin{eqnarray*}
 \lim_{x\rightarrow 1}\left(\frac{f(x)}{1-x}\right)^{(n)}&=&\lim_{x\rightarrow 1}\frac{1}{(1-x)^{n+1}}\sum_{k=0}^n \frac{n!}{k!}f(x)^{(k)}(1-x)^k\\
	&=&\lim_{x\rightarrow 1}\frac{(-1)^{n+1}}{(n+1)!}\sum_{k=0}^n \frac{n!}{k!}\left(f(x)^{(k)}(1-x)^k\right)^{(n+1)}\\
	&=&\lim_{x\rightarrow 1}\frac{(-1)^{n+1}}{n+1}\sum_{k=0}^n \frac{1}{k!}\sum_{i=0}^{n+1}\binom{n+1}{i}f(x)^{(k+i)}((1-x)^k)^{(n+1-i)}\\
	&=&\lim_{x\rightarrow 1}\frac{(-1)^{n+1}}{n+1}\sum_{k=0}^n \frac{1}{k!}\sum_{i=n-k+1}^{n+1}\binom{n+1}{i}f(x)^{(k+i)}\\
	    &&\hspace{4.5cm}\times\frac{k!(-1)^{n+1-i}}{(k-n-1+i)!}(1-x)^{k-n-1+i}\\
	&=&\frac{1}{n+1}\sum_{k=0}^n\binom{n+1}{n-k+1}f(1)^{(n+1)}(-1)^{n-k+1}=-\frac{f(1)^{(n+1)}}{n+1}.
\end{eqnarray*}
Hence $\frac{\H{m_1,m_2,\ldots,m_k}x-\H{m_1,m_2,\ldots,m_k}1}{1-x}$ is analytic for $x \in (0,\infty).$
\end{proof}
Using similar arguments as in the proof of the previous lemma we get the following lemma.
\begin{lemma}
Let $\H{m_1,m_2,\ldots,m_k}x$ be a harmonic polylogarithm with $m_i\in\{1, 0,-1,2\}$ for $1\leq i \leq k,$ and $m_k\neq 0.$ Then 
$$\frac{\H{m_1,m_2,\ldots,m_k}{1-x}}{1+x}$$
and
$$\frac{\H{m_1,m_2,\ldots,m_k}{1-x}}{1-x}$$
are analytic for $x \in (0,2).$
\label{HSanalytic2}
\end{lemma}

\begin{lemma}
Let $\H{\ve m}x=\H{m_1,m_2,\ldots,m_k}x$ and $\H{\ve b}x=\H{b_1,b_2,\ldots,b_l}x$ be harmonic polylogarithms with $m_i,b_i\in~\{1, 0,-1,2\}, m_1\neq~1$ and $b_l\neq~0.$ Then we have
$$
\M{\frac{\H{\ve m}x}{1-x}}{n} =\int_0^1{\frac{x^n(\H{\ve m}x-\H{\ve m}1)}{1-x}dx}-\int_0^1{\frac{\H{\ve m}x-\H{\ve m}1}{1-x}}dx-\S{1}n\H{\ve m}{1},
$$
and
$$
\M{\frac{\H{\ve b}{1-x}}{1-x}}{n}=\int_0^1{\frac{x^n\H{\ve b}{1-x}}{1-x}dx}-\H{0,\ve b}1
$$
where
$$\int_0^1{\frac{\H{\ve m}x-\H{\ve m}1}{1-x}}dx, \ \H{\ve m}{1} \textnormal{ and } \H{0,\ve b}1 $$
are finite constants.
\label{HSmelexpconst}
\end{lemma}
\begin{proof}
By integration by parts we get
\begin{eqnarray*}
&& \int_0^1{\frac{\H{m_1,m_2,\ldots,m_k}x-\H{m_1,m_2,\ldots,m_k}1}{1-x}dx} \\
&& \ \ =\lim_{\epsilon \rightarrow 1^-} \int_0^{\epsilon}{\frac{\H{m_1,m_2,\ldots,m_k}x-\H{m_1,m_2,\ldots,m_k}1}{1-x}dx}\\
&& \ \ =\lim_{\epsilon \rightarrow 1^-} \Bigl( \left.\left(\H{m_1,m_2,\ldots,m_k}x-\H{m_1,m_2,\ldots,m_k}1\right)\H{1}x\right|_0^{\epsilon}\Bigr.\\
&& \ \	\Bigl.\hspace{2cm}-\int_0^{\epsilon}{\frac{\H{m_2,m_3,\ldots,m_k}x\H{1}x}{x-m_1}dx}\Bigr)\\
&& \ \ =0-\lim_{\epsilon \rightarrow 1^-}\int_0^{\epsilon}{\frac{\H{1,m_2,\ldots,m_k}x+\H{m_2,1,m_3,\ldots,m_k}x+\ldots +\H{m_2,\ldots,m_k,1}x}{x-m_1}dx}\\
&& \ \ =-\lim_{\epsilon \rightarrow 1^-}{\left( \H{m_1,1,m_2,\ldots,m_k}{\epsilon}+\H{m_1,m_2,1,m_3,\ldots,m_k}{\epsilon}+\ldots +\H{m_1,m_2,\ldots,m_k,1}{\epsilon} \right)} \\
&& \ \ =-\H{m_1,1,m_2,\ldots,m_k}1-\H{m_1,m_2,1,m_3,\ldots,m_k}1-\ldots -\H{m_1,m_2,\ldots,m_k,1}1.
\end{eqnarray*}
Since $m_1\neq 1$ all these harmonic polylogarithms are finite. The same is true for $\H{\ve m}{1}$ and $\H{0,\ve b}1.$ Now we get
\begin{eqnarray*}
&&\M{\frac{\H{\ve m}x}{1-x}}{n}=\int_0^1{\frac{(x^n-1)\H{\ve m}x}{1-x}dx}\\
	&& \ \ =\int_0^1{\frac{x^n \H{\ve m}x-\H{\ve m}x-x^n\H{\ve m}1+x^n\H{\ve m}1-\H{\ve m}1+\H{\ve m}1}{1-x}dx}\\
	&& \ \ =\int_0^1{\frac{x^n(\H{\ve m}x-\H{\ve m}1)}{1-x}dx}-\int_0^1{\frac{\H{\ve m}x-\H{\ve m}1}{1-x}}dx+\H{\ve m}{1}\int_0^1{\frac{x^n-1}{1-x}dx}\\
	&& \ \ =\int_0^1{\frac{x^n(\H{\ve m}x-\H{\ve m}1)}{1-x}dx}-\int_0^1{\frac{\H{\ve m}x-\H{\ve m}1}{1-x}}dx-\S{1}n\H{\ve m}{1}
\end{eqnarray*}
and
\begin{eqnarray*}
\M{\frac{\H{\ve b}{1-x}}{1-x}}{n}&=&\int_0^1{\frac{(x^n-1)\H{\ve b}{1-x}}{1-x}dx}\\
    &=&\int_0^1{\frac{x^n \H{\ve b}{1-x}}{1-x}dx}-\int_0^1{\frac{\H{\ve b}{1-x}}{1-x}dx}\\
    &=&\int_0^1{\frac{x^n \H{\ve b}{1-x}}{1-x}dx}-\H{0,\ve b}1.
\end{eqnarray*}
\end{proof}

\begin{remark}
\label{HSExpandableIntegrals}
Combining Lemma \ref{HSanalytic1}, Lemma \ref{HSanalytic2} and \ref{HSmelexpconst} we are able to compute the asymptotic expansion of Mellin transforms of the form
$$\M{\frac{\H{m_1,m_2,\ldots,m_k}x}{1 \pm x}}{n} \ \textnormal{ and } \M{\frac{\H{b_1,b_2,\ldots,b_k}{1-x}}{1 \pm x}}{n}$$
where $m_i\in \{-1,0\}, b_i\in \{-1,0,1,2\}$ and $b_l\neq 0$ using the method presented in the beginning of this section, or using repeated integration by parts.
\end{remark}

\subsection{Computation of Asymptotic Expansions of Harmonic Sums}
\label{HScompasy}
The following observation is crucial.
\begin{lemma}
\begin{enumerate}
\item If a harmonic sum $\S{a_1,a_2,\ldots,a_k}n$ has no trailing ones, \ie $a_k\neq1$ then the most complicated harmonic polylogarithm in the inverse Mellin 
transform of $\S{a_1,a_2,\ldots,a_k}n$ has no trailing ones.
\item If a harmonic polylogarithm $\H{m_1,m_2,\ldots,m_k}x$ has no trailing ones, \ie $m_k\neq1$ then the most complicated harmonic sum in the Mellin 
transform of $\H{m_1,m_2,\ldots,m_k}x$ has no trailing ones.
\end{enumerate}
\label{HSnotrail1}
\end{lemma}
\begin{proof}
Due to \cite[Remark 3.5.18.]{Ablinger2009} 1. and 2. are equivalent. The proof of 2. follows by using Algorithm 2 of \cite{Ablinger2009} for each of the possible values of $m_k.$
\end{proof}

According to the previous lemma we know that if $\H{m_1,m_2,\ldots,m_{k}}x$ is the most complicated harmonic polylogarithm in the inverse Mellin 
transform of a harmonic sum without trailing ones then $m_k\neq 1.$ Due to Remark \ref{HSExpandableIntegrals} we can expand the resulting integrals (either we can handle it directly,
or we have to transform the argument of the harmonic polylogarithm).\\ 
Moreover, a harmonic sum with trailing ones might lead to a harmonic polylogarithm with trailing ones. But $\frac{\H{m_1,m_2,\ldots,m_{k-1},1}x}{1\pm x}$ is not analytic at $1$ and hence we cannot use the strategy mentioned 
in Section \ref{HShexp} to find an asymptotic representation of such integrals. Fortunately we can always extract trailing ones such that we end up in a univariate polynomial 
in $\S1n$ with coefficients in the harmonic sums without trailing ones (see Section \ref{HSdef}). Hence we only need to deal with powers of $\S1n$ and harmonic sums without trailing ones.

First we consider $\S{1}n.$ Since $\frac{1}{1-x}$ is not analytic at $1,$ we cannot use the strategy mentioned in Section \ref{HShexp} to find an asymptotic representation 
of $\int_0^1{\frac{x^n-1}{1\pm x}dx}=\S{1}n,$ however we have the following well known lemma.

\begin{lemma}
The asymptotic representation of $\S1n$ is given by
$$
\gamma+\log(n)+\frac{1}{2\,n}-\sum_{k=1}^{\infty}{\frac{B_{2k}}{2\,k\,n^{2k}}}
$$
where $\gamma$ is the Euler-Mascheroni constant and $B_{2k}$ are Bernoulli numbers.
\label{HSexpandS1}
\end{lemma}

\begin{remark}
To compute the asymptotic representation of $\S1n,$ we can as well proceed in the following way: we compute the expansion of $-\int_0^1{\frac{x^n\H{0}x}{1-x}}dx;$ integrate the result and 
add the constant $\gamma.$ This is possible since $\S1n=\int_0^1{\frac{x^n-1}{1-x}}dx$ and $\frac{d}{dn}\int_0^1{\frac{x^n-1}{1-x}}dx=-\int_0^1{\frac{x^n\H{0}x}{1-x}}dx.$
\end{remark}

Now we are ready to state an algorithm to compute asymptotic expansions of harmonic sums.
If we want to find the asymptotic expansions of a harmonic sum $\S{a_1,a_2,\ldots,a_k}n$ we can proceed as follows:

\begin{itemize}
	\item If $\S{a_1,a_2,\ldots,a_k}n$ has trailing ones, \ie $a_k=1,$ we first extract them such that we end up in a univariate 
	polynomial in $\S1n$ with coefficients in the harmonic sums without trailing ones; apply the 
	following items to each of the harmonic sums without trailing ones.
	\item Suppose now $\S{a_1,a_2,\ldots,a_k}n$ has no trailing ones, \ie $a_k\neq 1;$ let $\frac{\H{m_1,m_2,\ldots,m_l}x}{1+sx}$ be the most complicated weighted harmonic 
	 polylogarithm in the inverse Mellin transform of $\S{a_1,a_2,\ldots,a_k}n;$ express $\S{a_1,a_2,\ldots,a_k}n$ as
	\begin{equation}\label{HSasyalg1}
		\S{a_1,a_2,\ldots,a_k}n=\M{\frac{\H{m_1,m_2,\ldots,m_l}x}{1+sx}}{n}+T.
	\end{equation}
	Note that $T$ is an expression in harmonic sums (which are less complicated than $\S{a_1,a_2,\ldots,a_k}n$) and constants.
	\item We proceed by expanding $\M{\frac{\H{m_1,m_2,\ldots,m_l}x}{1+sx}}{n}$: 
	\begin{description}
		\item[all $m_i\neq 1$:] expand $\M{\frac{\H{m_1,m_2,\ldots,m_l}x}{1+sx}}{n}$ directly see Remark \ref{HSExpandableIntegrals}

		\item[not all $m_i\neq 1$:]
		\begin{itemize}\item[]
			\item transform $x\rightarrow 1-x$ in $\H{\ve m}x$; expand all products; now we can rewrite
				\begin{equation}\label{HSasyalg2}
					\M{\frac{\H{\ve m}x}{1+sx}}{n}=\sum_{i=1}^pc_i\M{\frac{\H{\ve b_i}{1-x}}{1+sx}}{n}+c \ \textnormal{  with } c,c_i\in\R
				\end{equation}
			 \item for each Mellin transform $\M{\frac{\H{b_1,\ldots,b_j}{1-x}}{1+sx}}{n}$ do
				\begin{description}\item[]
					\item[$b_j\neq 0:$] expand $\int_0^1{\frac{x^n\H{\ve b}{1-x}}{1+sx}}dx$ see Remark \ref{HSExpandableIntegrals}
					\item[$b_j=0:$] transform back $1-x\rightarrow x$ in $\H{\ve b}{1-x}$; expand all products by shuffling; write
						\begin{equation}\label{HSasyalg3}
							\M{\frac{\H{\ve b}{1-x}}{1+sx}}{n}=\sum_{i=1}^pd_i\M{\frac{\H{\ve g_i}{x}}{1+sx}}{n}+d
						\end{equation}
						with $d,d_i\in\R$ and perform the Mellin transforms $\M{\frac{\H{\ve g_i}{x}}{1-x}}{n}.$
				\end{description}
		\end{itemize}
	\end{description}
	\item Replace $\M{\frac{\H{m_1,m_2,\ldots,m_l}x}{1+sx}}{n}$ in equation (\ref{HSasyalg1}) by the result of this process.
	\item Expand the powers of $\S{1}n$ using Lemma \ref{HSexpandS1} and perform the Cauchy product to the expansions. 
	\item For all harmonic sums that remain in equation (\ref{HSasyalg1}) apply the above points; since these harmonic sums are less complicated this process will terminate.
\end{itemize}

Some remarks are in place: Since $a_k\neq 1$ in equation (\ref{HSasyalg1}), we know due to Lemma \ref{HSnotrail1} that $m_l\neq 0$ in equation (\ref{HSasyalg1}).
If not all $m_i\neq 1$ in equation  (\ref{HSasyalg1}), we have to transform $x\rightarrow 1-x$ in $\H{m_1,m_2,\ldots,m_l}x.$ From Remark \ref{HStraforem} we know that the s
harmonic polylogarithm at argument $x$ with weight $l$ which will pop up will not have trailing zeroes since $m_l\neq 1$. Therefore the harmonic sums which will appear
in equation (\ref{HSasyalg3}) are less complicated (\ie smaller in the sense of Definition \ref{HSsord}) than $\S{a_1,a_2,\ldots,a_k}n$ of (\ref{HSasyalg1}) and hence this algorithm will eventually terminate.

\begin{example}
 Let us now consider the asymptotic expansion of $\S{2,1}n.$ First we remove the trailing ones:
$$
\S{2,1}n=\S{1}n\S{2}n+\S{3}n-\S{1,2}n.
$$
We can expand $\S{1}n$ using Lemma \ref{HSexpandS1}. For $\S{2}n$ we get:
\begin{eqnarray*}
\S{2}n=\M{\frac{\H{0}x}{1-x}}{n}=\int_0^1\frac{x^n(\H{0}x-\H{0}1)}{1-x}dx-\int_0^1\frac{\H{0}x-\H{0}1}{1-x}dx.
\end{eqnarray*}
The second integral is a constant. Expanding the first integral leads to 
$$
\S{2}n \sim -\H{1,0}1-\frac{1}{n}+\frac{1}{2 n^2}-\frac{1}{6 n^3}+\frac{1}{30 n^5}.
$$
Analogously for $\S{3}n$ we get:
\begin{eqnarray*}
\S{3}n&=&-\M{\frac{\H{0,0}x}{1-x}}{n}\\
      &=&-\int_0^1\frac{x^n(\H{0,0}x-\H{0,0}1)}{1-x}dx+\int_0^1\frac{\H{0,0}x-\H{0,0}1}{1-x}dx,
\end{eqnarray*}
and hence:
$$
\S{3}n \sim \H{1,0,0}1-\frac{1}{2 n}+\frac{1}{2 n^3}-\frac{1}{4 n^4}.
$$
Finally for $\S{1,2}n$ we have
\begin{eqnarray*}
\S{1,2}n=\M{\frac{\H{1,0}x}{1-x}}{n}&=&\int_0^1\frac{x^n(\H{1,0}x-\H{1,0}1)}{1-x}dx\\
	  &&-\int_0^1\frac{\H{1,0}x-\H{1,0}1}{1-x}dx-\S{1}n\H{1,0}1,
\end{eqnarray*}
and consequently:
$$
\S{1,2}n \sim -\H{1,0,0}1+\frac{1}{n}-\frac{3}{4 n^2}+\frac{17}{36 n^3}-\frac{5}{24 n^4}+\frac{7}{450 n^5}-\H{1,0}1\S{1}n.
$$
Combining these results we find the asymptotic expansion of $\S{2,1}n$:
\begin{eqnarray*}
\S{2,1}n &\sim& 2\,\H{1,0,0}1-\frac{1}{n}-\frac{1}{4 n^2}+\frac{13}{36 n^3}-\frac{1}{6 n^4}-\frac{1}{100 n^5}\\
	&&+(\log(n)+\gamma)\left( -\frac{1}{n}+\frac{1}{2 n^2}-\frac{1}{6 n^3}+\frac{1}{30 n^5}\right).
\end{eqnarray*}
\end{example}

\cleardoublepage  

\chapter{S-Sums}
\label{SSchapter}
The main goal of this chapter is to extend the ideas form Chapter \ref{HSchapter} to a generalization of harmonic sums (see Definition \ref{HSdefHsum}), \ie we generalize 
the harmonic sums to the S-sums of \cite{Moch2002}.
\section{Definition and Structure of S-Sums and Multiple Polylogartihms}
\label{SSdef}
\begin{definition}[S-Sums]
For $a_i\in \N$ and $x_i\in \R^*$ we define
\begin{equation}
	\S{a_1,\ldots ,a_k}{x_1,\ldots ,x_k;n}= \sum_{n\geq i_1 \geq i_2 \geq \cdots \geq i_k \geq 1} \frac{x_1^{i_1}}{i_1^{a_1}}\cdots
	\frac{x_k^{i_k}}{i_k^{a_k}}.
\end{equation}
$k$ is called the depth and $w=\sum_{i=0}^ka_i$ is called the weight of the S-sum $\S{a_1,\ldots ,a_k}{x_1,\ldots ,x_k;n}$.
\end{definition}
As for harmonic sums a product of two S-sums with the same upper summation limit can be written in terms of single S-sums (see \cite{Moch2002}): for $n\in \N,$
\begin{eqnarray}
	&&\S{a_1,\ldots ,a_k}{x_1,\ldots ,x_k;n}\S{b_1,\ldots ,b_l}{y_1,\ldots ,y_l;n}=\nonumber\\
	&&\hspace{2cm}\sum_{i=1}^n \frac{x_1^i}{i^{a_1}}\S{a_2,\ldots ,a_k}{x_2,\ldots ,x_k;i}\S{b_1,\ldots ,b_l}{y_1,\ldots ,y_l,i} \nonumber\\
	&&\hspace{2cm}+\sum_{i=1}^n \frac{y_1^i}{i^{b_1}}\S{a_1,\ldots ,a_k}{x_1,\ldots ,x_k,i}\S{b_2,\ldots ,b_l}{y_2,\ldots ,y_l;i} \nonumber\\
	&&\hspace{2cm}-\sum_{i=1}^n \frac{(x_1\cdot y_1)^i}{i^{a_1+b_1}}\S{a_2,\ldots ,a_k}{x_2,\ldots ,x_k,i}\S{b_2,\ldots ,b_l}{y_2,\ldots ,y_l;i}.
	\label{SSsumproduct}
\end{eqnarray}
Recursive application of (\ref{SSsumproduct}) leads to a linear combination of single S-sums. Together with this product S-sums form a quasi shuffle algebra.\\
The quasi shuffle algebra property will allow us again to look for algebraic relations, however we are as well interested in structural relations and asymptotic 
expansions of S-sums. Therefore we will again use an integral representation of these sums, which leads to an extension of the harmonic polylogarithms. 
In Section \ref{HSExtended} we already extended harmonic polylogarithms which were originally defined on the alphabet $\{-1,0,1\}$ (see Definition \ref{HShlogdef}) to the 
alphabet $\{-1,0,1,2\}.$ In this section we will extend harmonic polylogarithms even further (for an equivalent defintion see \cite{GON1}).
\begin{definition}[Multiple Polylogarithms]
Let $a\in\R$ and
\begin{eqnarray*}
q=\left\{ 
		\begin{array}{ll}
				a, &  \textnormal{if }a>0  \\
				\infty, & \textnormal{otherwise}. 
		\end{array} \right.
\end{eqnarray*}
We define the auxiliary function as follows:
\begin{eqnarray}
&&f_a:(0,q)\mapsto \R\nonumber\\
&&f_a(x)=\left\{ 
		\begin{array}{ll}
				\frac{1}{x}, &  \textnormal{if }a=0  \\
				\frac{1}{\abs{a}-\sign{a}\,x}, & \textnormal{otherwise}.  
		\end{array} 
		\right.  \nonumber
\end{eqnarray}
Now we can define the \textit{multiple polylogarithms} \textnormal{H}.
Let $m_i \in \R$ and let $q=\displaystyle{\min_{m_i>0}{m_i}};$ we define for $x\in (0,q):$
\begin{eqnarray}
\H{}{x}&=&1,\nonumber\\
\H{m_1,m_{2},\ldots,m_w}{x} &=&\left\{ 
		  	\begin{array}{ll}
						\frac{1}{w!}(\log{x})^w,&  \textnormal{if }(m_1,\ldots,m_k)=(0,\ldots,0)\\
						\int_0^x{f_{m_1}(y) \H{m_{2},\ldots,m_k}{y}dy},& \textnormal{otherwise}. 
					\end{array} \right.  \nonumber
\end{eqnarray}
The length $w$ of the vector $\ve m$ is called the weight of the shifted multiple polylogarithm $\H{\ve m}x.$
\label{SShlogdef}
\end{definition}

\begin{example}
\begin{eqnarray*}
\H{1}x&=&\int_{0}^x\frac{1}{1-x_1}dx_1=-\log(1-x)\\ 
\H{-2}x&=&\int_{0}^x\frac{1}{2+x_1}dx_1=\log(2+x)-\log(2)\\ 
\H{2,0,-\frac{1}{2}}x&=&\int_{0}^x\frac{1}{2-x_1}\int_{0}^{x_1}\frac{1}{x_2}\int_{0}^{x_2}\frac{1}{\frac{1}{2}+x_3}dx_3dx_2dx_1.
\end{eqnarray*}
\end{example}

A multiple polylogarithm $\H{\ve m}x=\H{m_1,\ldots,m_w}x$ with $q:=\min_{m_i>0}{m_i}$ is an analytic functions for $x\in (0,q).$ For the limits $x\rightarrow 0$ and  $x\rightarrow q$ 
we have:
\begin{itemize}
 \item It follows from the definition that if $\ve m\neq \ve 0_w$, $\H{\ve m}0~=~0.$
 \item If $m_1\neq q$ or if $m_1=1$ and $m_v=0$ for all $v$ with $1<v\leq w$ then $\H{\ve m}q$ is finite.
\end{itemize}
We define $\H{\ve m}0:=\lim_{x\rightarrow 0^+} \H{\ve m}x$ and $\H{\ve m}1:=\lim_{x\rightarrow 1^-} \H{\ve m}x$ if the limits exist.\\ 
Note that for multiple polylogarithms of the form $$\H{m_1,\ldots,m_k,1,0\ldots,0}x$$ with $c:=\min_{m_i>0}{m_i}$ we can extend the definition range from 
$x\in (0,1)$ to $x \in (0,c).$ These multiple polylogarithms are analytic functions for $x \in (0,c).$ The limit $x\rightarrow q$ exists if and only if $m_1\neq c.$

\begin{remark}
In the following we will always assume that the arguments of the considered multiple polylogarithms are chosen in a way such that multiple polylogarithms are defined and finite.
\end{remark}

\begin{remark}
For the derivatives we have for all $x\in (0,\min_{m_i>0}{m_i})$ that $$ \frac{d}{d x} \H{\ve m}{x}=f_{m_1}(x)\H{m_{2},m_{3},\ldots,m_w}{x}. $$ 
\end{remark}

\begin{remark}
Again the product of two multiple polylogarithms of the same argument can be expressed using the formula (compare (\ref{HShpro}))
\begin{equation}
\label{SShpro}
\H{\ve p}x\H{\ve q}x=\sum_{\ve r= \ve p \shuffle \ve q}\H{\ve r}x
\end{equation}
in which $\ve p \shuffle \ve q$ represent all merges of $\ve p$ and $\ve q$ in which the relative orders of the elements of $\ve p$ and $\ve q$ are preserved.
\end{remark}

Note that we can use the shuffle algebra property to extract trailing zeroes using the ideas of Section \ref{HShpldefsec}. There is one situation we have to take care of: given a 
multiple polylogarithm in the form $$\H{m_1,\ldots,m_k,1,0\ldots,0}x$$ with $1<q:=\min_{m_i>0}{m_i}$ then we cannot use the strategy to extract trailing zeroes if $x\geq1$ since this would lead
to infinities. The shuffle algebra would suggest to rewrite $\H{2,1,0}{x}$ as
$$\H{0}x\H{2,1}{x}-\H{0,2,1}{x}-\H{2,0,1}{x},$$ 
however for $x>1$ for instance $\H{2,0,1}{x}$ is not defined.
Thus, if we consider the extraction of trailing zeroes in the following then trailing zeroes of multiple polylogarithm in the form 
$\H{m_1,\ldots,m_k,1,0\ldots,0}x$ will not be extracted and these multiple polylogarithms will remain untouched if $x\geq1.$

Along with the integral representation multiple polylogarithms at special constants and conected to them S-sums at infinity will pop up. Hence we will take a closer look at relations 
between these quantities. This will lead to power series expansions of multiple polylogarithms and identities between multiple polylogarithms of related arguments.

\section{Identities between Multiple Polylogarithms of Related Arguments}
\label{SSRelatedArguments}
In the following we want to look at several special transforms of the argument of multiple polylogarithms which will be useful in later considerations.

\subsection{\texorpdfstring{$x+b\rightarrow x$}{x+b->x}}
\label{SStransformplusxplusb}
\begin{lemma}
 Let $m_i \in \R\setminus (0,1)$ and $x>0$ such that the multiple polylogarithm $\H{m_1,\ldots, m_l}{x+1}$ is defined. If $(m_1,\ldots,m_l)\neq (1,0,\ldots,0),$
\begin{eqnarray*}
\H{m_1,\ldots, m_l}{x+1}&=&\\
&&\hspace{-2cm}\H{m_1,\ldots, m_l}{1}+\H{m_1-1}x\H{m_2,\ldots, m_l}{1}+\H{m_1-1,m_2-1}x\H{m_3,\ldots, m_l}{1}\\
&&\hspace{-2cm}+\cdots+\H{m_1-1,\ldots, m_{l-1}-1}{x}\H{m_l}1+\H{m_1-1,\ldots, m_l-1}{x}.
\end{eqnarray*}
If $(m_1,\ldots,m_l)=(1,0,\ldots,0),$
\begin{eqnarray*}
\H{1,0,\ldots, 0}{x+1}&=&\H{1,0,\ldots, 0}{1}-\H{0,-1,\ldots, -1}{x}.
\end{eqnarray*}
\label{SSeinplustrafo}
\end{lemma}
\begin{proof}
 We proceed by induction on $l.$ For $l=1$  and $m_1>1$ we have:
\begin{eqnarray*}
\H{m_1}{x+1}&=&\int_0^{x+1}\frac{1}{m_1-u}du=\H{m_1}1+\int_1^{x+1}\frac{1}{m_1-u}du\\
&=&\H{m_1}1+\int_0^{x}\frac{1}{m_1-(u+1)}du=\H{m_1}1+\H{m_1-1}x.
\end{eqnarray*}
For $l=1$  and $m_1<0$ we have:
\begin{eqnarray*}
\H{m_1}{x+1}&=&\int_0^{x+1}\frac{1}{\abs{m_1}+u}du=\H{m_1}1+\int_1^{x+1}\frac{1}{\abs{m_1}+u}du\\
&=&\H{m_1}1+\int_0^{x}\frac{1}{\abs{m_1}+(u+1)}du=\H{m_1}1+\H{m_1-1}x.
\end{eqnarray*}
For $l=1$  and $m_1 = 0$ we have:
\begin{eqnarray*}
\H{0}{x+1}&=&\H{0}1+\int_1^{x+1}\frac{1}{u}du=\int_0^{x}\frac{1}{1+u}du=\H{0}1+\H{-1}{x}.
\end{eqnarray*}
Suppose the theorem holds for $l-1.$ If $m_1\neq 0$ we get
\begin{eqnarray*}
\H{m_1,\ldots, m_l}{x+1}&=&\H{m_1,\ldots, m_l}{1}+\int_1^{x+1}{\frac{\H{m_2,\ldots,m_l}{u}}{\abs{m_1}-\sign{m_1}u}}du\\
&=&\H{m_1,\ldots, m_l}{1}+\int_0^{x}{\frac{\H{m_2,\ldots,m_l}{u+1}}{\abs{m_1}-\sign{m_1}{(u+1)}}}du\\
&=&\H{m_1,\ldots, m_l}{1}+\int_0^{x}\frac{1}{\abs{m_1}-\sign{m_1}{(u+1)}}\biggl(\H{m_2,\ldots, m_l}{1}\biggr.\\
&&+\H{m_2-1}x\H{m_3,\ldots, m_l}{1}+\H{m_2-1,m_3-1}x\H{m_4,\ldots, m_l}{1}+\cdots\\
&&\biggl.+\H{m_2-1,\ldots, m_{l-1}-1}{x}\H{m_l}1+\H{m_2-1,\ldots, m_l-1}{x}\biggr)du\\
&=&\H{m_1,\ldots, m_l}{1}+\H{m_1-1}x\H{m_2,\ldots, m_l}{1}\\
&&+\H{m_1-1,m_2-1}x\H{m_3,\ldots, m_l}{1}+\cdots+\\
&&+\H{m_1-1,\ldots, m_{l-1}-1}{x}\H{m_l}1+\H{m_1-1,\ldots, m_l-1}{x}.
\end{eqnarray*}
If $m_1= 0$ we get
\begin{eqnarray*}
\H{0,m_2,\ldots, m_l}{x+1}&=&\H{0,m_2,\ldots, m_l}{1}+\int_1^{x+1}{\frac{\H{m_2,\ldots,m_l}{u}}{u}}du\\
&=&\H{0,m_2,\ldots, m_l}{1}+\int_0^{x}{\frac{\H{m_2,\ldots,m_l}{u+1}}{(1+u)}}du\\
&=&\H{0,m_2,\ldots, m_l}{1}+\int_0^{x}\frac{1}{1+u}\biggl(\H{m_2,\ldots, m_l}{1}+\biggr.\\
&&+\H{m_2-1}x\H{m_3,\ldots, m_l}{1}+\H{m_2-1,m_3-1}x\H{m_4,\ldots, m_l}{1}+\\
&&\biggl.+\cdots+\H{m_2-1,\ldots, m_{l-1}-1}{x}\H{m_l}1+\H{m_2-1,\ldots, m_l-1}{x}\biggr)du\\
&=&\H{0,m_2,\ldots, m_l}{1}+\H{-1}x\H{m_2,\ldots, m_l}{1}\\
&&+\H{-1,m_2-1}x\H{m_3,\ldots, m_l}{1}+\cdots+\\
&&+\H{-1,\ldots, m_{l-1}-1}{x}\H{m_l}1+\H{-1,\ldots, m_l-1}{x}.
\end{eqnarray*}
\end{proof}

The proofs of the following three lemmas are similar to the proof of the previous one, hence we will omit them here.
\begin{lemma}
Let $b>0$, $m_1\in \R \setminus (0,b]$, $m_i \in \R \setminus (0,b)$ for $i\in \{2,3,\ldots l\}$ such that $(m_j,\ldots,m_l)\neq(1,0,\ldots,0)$ for all $j\in \{1,\ldots l\}.$ 
For $x>0$ we have
\begin{eqnarray*}
\H{m_1,\ldots, m_l}{x+b}&=&\H{m_1,m_2,\ldots, m_l}{b}+\H{m_1-b}x\H{m_2,\ldots, m_l}{b}\\
&&+\H{m_1-b,m_2-b}x\H{m_3,\ldots, m_l}{b}+\cdots+\\
&&+\H{m_1-b,\ldots, m_{l-1}-b}{x}\H{m_l}b+\H{m_1-b,\ldots, m_l-b}{x}.
\end{eqnarray*}
\label{SStransformplusa1}
\end{lemma}

\begin{lemma}
Let $(m_1,\ldots,m_l)=(1,0,\ldots,0).$ For $1>b>0$ and $1-b>x>0$ we have
\begin{eqnarray*}
\H{m_1,\ldots, m_l}{x+b}&=&\\
&&\hspace{-2cm}\H{m_1,m_2,\ldots, m_l}{b}+\H{m_1-b}x\H{m_2,\ldots, m_l}{b}+\H{m_1-b,m_2-b}x\H{m_3,\ldots, m_l}{b}\\
&&\hspace{-2cm}+\cdots+\H{m_1-b,\ldots, m_{l-1}-b}{x}\H{m_l}b+\H{m_1-b,\ldots, m_l-b}{x}.
\end{eqnarray*}
For $x>0$ and $b\geq 1$ we have
\begin{eqnarray*}
\H{m_1,\ldots, m_l}{x+b}&=&\\
&&\hspace{-2cm}\H{m_1,m_2,\ldots, m_l}{b}-\H{m_1-b}x\H{m_2,\ldots, m_l}{b}-\H{m_1-b,m_2-b}x\H{m_3,\ldots, m_l}{b}\\
&&\hspace{-2cm}-\cdots-\H{m_1-b,\ldots, m_{l-1}-b}{x}\H{m_l}b-\H{m_1-b,\ldots, m_l-b}{x}.
\end{eqnarray*}
\label{SStransformplusa2}
\end{lemma}

\begin{lemma}
Let $b>0$, $m_1\in \R \setminus (0,b]$, $m_i \in \R \setminus (0,b)$ for $i\in \{2,3,\ldots k\}$ and $(m_{k+1},\ldots,m_l)= (1,0,\ldots,0).$ For $1-b>x>0$ and $1>b>0$ we have
\begin{eqnarray*}
\H{m_1,\ldots, m_l}{x+b}&=&\\
&&\hspace{-2cm}\H{m_1,m_2,\ldots, m_l}{b}+\H{m_1-b}x\H{m_2,\ldots, m_l}{b}+\H{m_1-b,m_2-b}x\H{m_3,\ldots, m_l}{b}\\
&&\hspace{-2cm}+\cdots+\H{m_1-b,\ldots, m_{k}-b}{x}\H{m_{k+1},\ldots,m_l}b+\H{m_1-b,\ldots, m_{k+1}-b}{x}\H{m_{k+2},\ldots,m_l}b\\
&&\hspace{-2cm}+\cdots+\H{m_1-b,\ldots, m_{l-1}-b}{x}\H{m_l}b+\H{m_1-b,\ldots, m_l-b}{x}.
\end{eqnarray*}
For $x>0$ and $b\geq 1$ we have
\begin{eqnarray*}
\H{m_1,\ldots, m_l}{x+b}&=&\\
&&\hspace{-2cm}\H{m_1,m_2,\ldots, m_l}{b}+\H{m_1-b}x\H{m_2,\ldots, m_l}{b}+\H{m_1-b,m_2-b}x\H{m_3,\ldots, m_l}{b}\\
&&\hspace{-2cm}+\cdots+\H{m_1-b,\ldots, m_{k}-b}{x}\H{m_{k+1},\ldots,m_l}b - \H{m_1-b,\ldots, m_{k+1}-b}{x}\H{m_{k+2},\ldots,m_l}b\\
&&\hspace{-2cm} -\cdots-\H{m_1-b,\ldots, m_{l-1}-b}{x}\H{m_l}b-\H{m_1-b,\ldots, m_l-b}{x}.
\end{eqnarray*}
\label{SStransformplusa3}
\end{lemma}

\subsection{\texorpdfstring{$b-x\rightarrow x$}{b-x->x}}
\label{SSbxx}
In this subsection we assume that $b>0$ and we consider indices $m_i\in \R\setminus (0,b).$  
Proceeding recursively on the weight $w$ of the multiple polylogarithm we have for $m_1\neq b$:
\begin{eqnarray}
 \H{m_1}{b-x}&=&\H{m_1}b-\H{b-m_1}x\\
 \H{b}{b-x}&=&\H{0}x-\H{0}b. \label{SStrafobx1}
\end{eqnarray}
Now let us look at higher weights $w>1.$ We consider $\H{m_1,m_2,\ldots,m_w}{b-x}$ with $m_i\in \R\setminus(0,b)$ and suppose that we can already apply the transformation for multiple polylogarithms
 of weight $<w,$ If $m_1=b,$ we can remove leading $b's$ (similar to the extraction of leading ones) and end up with multiple polylogarithms without leading $b's$ and powers of $\H{b}{b-x}.$ For the powers of $\H{b}{b-x}$ we
 can use (\ref{SStrafobx1}); therefore, only the cases in which the first index $m_1\neq 1$ are to be considered. For $b\neq 0$:
\begin{eqnarray*}
\H{m_1,m_2,\ldots,m_w}{b-x}&=&\H{m_1,m_2,\ldots,m_w}b-\int_0^x\frac{\H{m_2,\ldots,m_w}{b-t}}{\abs{-m_1+b}-\sign{-m_1+b}t}dt.
\end{eqnarray*}
Since we know the transform for weights $<w,$ we can apply it to $\H{m_2,\ldots,m_w}{b-t}$ and finally we obtain the required weight $w$ identity by using the definition of the 
multiple polylogarithms.

\subsection{\texorpdfstring{$\frac{1-x}{1+x} \rightarrow x$}{(1-x)/(1+x)->x}}
\label{SS1x1x}
In this subsection we restrict the indices again. We consider indices that are greater equal 1 or less equal 0. 
Proceeding recursively on the weight $w$ of the multiple polylogarithm we have for $0<\frac{1-x}{1+x}<1$, $y_1<-1$, $-1<y_2<0$ and $y_3>1:$
\begin{eqnarray}
\H{y_1}{\frac{1-x}{1+x}}&=&-\H{0}{-y_1}-\H{-1}{x}+\H{0}{1-y_1}+\H{\frac{1-y_1}{1+y_1}}x\\
\H{-1}{\frac{1-x}{1+x}}&=&\H{-1}{1}-\H{-1}x\\
\H{y_2}{\frac{1-x}{1+x}}&=&-\H{0}{-y_2}-\H{-1}{x}+\H{0}{1-y_2}-\H{\frac{1-y_2}{1+y_2}}x\\
\H{0}{\frac{1-x}{1+x}}&=&-\H{1}{x}+\H{-1}x\label{SStrafo1x1x0}\\
\H{1}{\frac{1-x}{1+x}}&=&-\H{-1}{1}-\H{0}{x}+\H{-1}{x}\label{SStrafo1x1x}\\
\H{y_3}{\frac{1-x}{1+x}}&=&\H{0}{y_3}+\H{-1}{x}-\H{0}{y_3-1}-\H{\frac{1-y_3}{1+y_3}}x.
\end{eqnarray}
Now let us look at higher weights $w>1.$ We consider $\H{m_1,m_2,\ldots,m_w}{\frac{1-x}{1+x}}$ with $m_i\in \R\setminus(0,1)$ and suppose that we can already apply the transformation for multiple polylogarithms
 of weight $<w.$ If $m_1=1,$ we can remove leading ones and end up with multiple polylogarithms without leading ones and powers of $\H{1}{\frac{1-x}{1+x}}.$ For the powers of $\H{1}{\frac{1-x}{1+x}}$ we
 can use (\ref{SStrafo1x1x}); therefore, only the cases in which the first index $m_1\neq 1$ are to be considered. For $y_1<0$, $ y_1\neq-1$ and $y_3>1$ we get (compare Section~\ref{HS1x1x}):
\begin{eqnarray*}
\H{y_1,m_2,\ldots,m_w}{\frac{1-x}{1+x}}&=&\H{y_2,m_2,\ldots,m_w}1-\int_0^x\frac{1}{1+t}\H{m_2,\ldots,m_w}{\frac{1-t}{1+t}}dt\\
					& &+\int_0^x\frac{1}{\frac{1-y_1}{1+y_1}+t}\H{m_2,\ldots,m_w}{\frac{1-t}{1+t}}dt\\
\H{-1,m_2,\ldots,m_w}{\frac{1-x}{1+x}}&=&\H{-1,m_2,\ldots,m_w}1-\int_0^x\frac{1}{1+t}\H{m_2,\ldots,m_w}{\frac{1-t}{1+t}}dt\\
\H{0,m_2,\ldots,m_w}{\frac{1-x}{1+x}}&=&\H{0,m_2,\ldots,m_w}1-\int_0^x\frac{1}{1-t}\H{m_2,\ldots,m_w}{\frac{1-t}{1+t}}dt\\
		&&-\int_0^x\frac{1}{1-t}\H{m_2,\ldots,m_w}{\frac{1+t}{1+t}}dt\\
\H{y_2,m_2,\ldots,m_w}{\frac{1-x}{1+x}}&=&\H{y_2,m_2,\ldots,m_w}1+\int_0^x\frac{1}{1+t}\H{m_2,\ldots,m_w}{\frac{1-t}{1+t}}dt\\
					& &-\int_0^x\frac{1}{\frac{1-y_2}{1+y_2}+t}\H{m_2,\ldots,m_w}{\frac{1-t}{1+t}}dt.\\
\end{eqnarray*}
We have to consider the case for $\H{0,\ldots,0}{\frac{1-x}{1+x}}$ separately, however we can deal with this case since it reduces by Definition \ref{SShlogdef} to $\H{0}{\frac{1-x}{1+x}},$ which 
we handled in (\ref{SStrafo1x1x0}).
Since we know the transform for weights $<w,$ we can apply it to $\H{m_2,\ldots,m_w}{\frac{1+t}{1+t}}$ and finally we obtain the required weight $w$ identity by using the definition of 
the multiple polylogarithms.

\subsection{\texorpdfstring{$k x\rightarrow x$}{k x->x}}
\begin{lemma}
 If $m_i \in \R, m_l\neq 0, x\in\R^+$ and $k\in\R^+$ such that the multiple polylogarithm $\H{m_1,\ldots, m_l}{k\cdot x}$ is defined then we have
\begin{eqnarray}
\H{m_1,\ldots, m_l}{k\cdot x}=\H{\frac{m_1}{k},\ldots,\frac{m_l}{k}}{x}.
\end{eqnarray}
\end{lemma}
\begin{proof}
 We proceed by induction on $l.$ For $l=1$ we have:
$$
\H{m_1}{k\cdot x}=\int_0^{kx}\frac{1}{\abs{m_1}-\sign{m_1}u}du=\int_0^x\frac{k}{\abs{m_1}-\sign{m_1}ku}du=\H{\frac{m_1}{k}}x.
$$
Suppose the lemma holds for $l.$ For $m_1\neq 0$ we get
\begin{eqnarray*}
\H{m_1,\ldots, m_{l+1}}{k\cdot x}&=&\int_0^{kx}\frac{\H{m_2,\ldots,m_{l+1}}{u}}{\abs{m_1}-\sign{m_1}u}du
	=\int_0^{x}\frac{\H{m_2,\ldots,m_{l+1}}{ku}}{\abs{m_1}-\sign{m_1}ku}kdu\\
&=&\int_0^{x}\frac{\H{\frac{m_2}{k},\ldots,\frac{m_{l+1}}{k}}{u}}{\frac{\abs{m_1}}{k}-\sign{m_1}u}du=\H{\frac{m_1}{k},\ldots,\frac{m_{l+1}}{k}}{x}.
\end{eqnarray*}
For $m_1 = 0$ we get
\begin{eqnarray*}
\H{0,m_2,\ldots, m_{l+1}}{k\cdot x}&=&\int_0^{kx}\frac{\H{m_2,\ldots,m_{l+1}}{u}}{u}du =\int_0^{x}\frac{\H{m_2,\ldots,m_{l+1}}{ku}}{ku}kdu\\
&=&\int_0^{x}\frac{\H{\frac{m_2}{k},\ldots,\frac{m_{l+1}}{k}}{u}}{u}du=\H{\frac{0}{k},\frac{m_2}{k},\ldots,\frac{m_{l+1}}{k}}{x}.
\end{eqnarray*}
\end{proof}

\subsection{\texorpdfstring{$-x\rightarrow x$}{-x->x}}
The proof of the following lemma is straightforward (compare \cite{Remiddi2000}, where it is stated for harmonic polylogarithms).
\begin{lemma}
 If $m_i \in \R, m_l\neq 0$ and $x\in\R^-$ such that the multiple polylogarithm $\H{m_1,\ldots, m_l}{-x}$ is defined then we have
\begin{eqnarray}
\H{m_1,\ldots, m_l}{-x}=(-1)^{l-k}\H{-m_1,\ldots, -m_l}{x},
\end{eqnarray}
here $k$ denotes the number of $m_i$ which equal zero.
\label{SStransformminusx}
\end{lemma}

\subsection{\texorpdfstring{$\frac{1}{x}\rightarrow x$}{1/x->x}}
\label{SS1dxx}
In this subsection we restrict the indices again. We consider negative indices together with the index $0.$ 
Proceeding recursively on the weight $w$ of the multiple polylogarithm we have for $0<x<1$ and $y<0:$
\begin{eqnarray*}
\H{y}{\frac{1}{x}}&=&\H{\frac{1}{y}}{x}-\H{0}x-\H{0}{-y}\\
\H{0}{\frac{1}{x}}&=&-\H{0}{x}.% \label{SStrafo1dx1}
\end{eqnarray*}
Now let us look at higher weights $w>1.$ We consider $\H{m_1,m_2,\ldots,m_w}{\frac{1}{x}}$ with $m_i\leq 0$ and suppose that we can already apply the transformation for multiple polylogarithms 
of weight $<w.$ For $m_1<0$ we get (compare \cite{Remiddi2000} and Section \ref{HS1dxx}):
\begin{eqnarray*}
\H{m_1,m_2,\ldots,m_w}{\frac{1}{x}}&=&\H{m_1,m_2,\ldots,m_w}1+\int_x^1\frac{1}{t^2(-m_1+1/t)}\H{m_2,\ldots,m_w}{\frac{1}{t}}dt\\
				  &=&\H{m_1,m_2,\ldots,m_w}1+\int_x^1\frac{1}{t}\H{m_2,\ldots,m_w}{\frac{1}{t}}dt\\
				  & &-\int_x^1\frac{1}{-\frac{1}{m_1}+t}\H{m_2,\ldots,m_w}{\frac{1}{t}}dt\\
\H{0,m_2,\ldots,m_w}{\frac{1}{x}}&=&\H{0,m_2,\ldots,m_w}1+\int_x^1\frac{1}{t^2(1/t)}\H{m_2,\ldots,m_w}{\frac{1}{t}}dt\\
				  &=&\H{0,m_2,\ldots,m_w}1+\int_x^1\frac{1}{t}\H{m_2,\ldots,m_w}{\frac{1}{t}}dt.\\
\end{eqnarray*}
Since we know the transform for weights $<w,$ we can apply it to $\H{m_2,\ldots,m_w}{\frac{1}{t}}$ and finally we obtain the required weight $w$ identity by using the definition of 
the multiple polylogarithms.\\
An index $m>0$ in the index set leads to a branch point at $m$ and a branch cut discontinuity in the complex plane for $x\in(m,\infty).$ This corresponds to the branch point at $x=m$ 
and the branch cut discontinuity in the complex plane for $x\in(m,\infty)$ of $\log(m-x)=\log(m)-\H{m}x.$ However, the analytic properties of the logarithm are well known and we
can set for $0<x<1$ for instance
\begin{eqnarray}
\H{m}{\frac{1}{x}}&=&\H{\frac{1}{m}}{x}+\H{0}{m}+\H{0}{x}-i\pi, \label{SStrafo1dx11}
\end{eqnarray}
by approaching the argument $\frac{1}{x}$ form the lower half complex plane.
We can now consider an alphabet with at most one letter $m\geq1$ (note that we could as well assume $m>0$ however for simplicity we restrict to $m\geq1$).
For such an alphabet the strategy is as follows: if a multiple polylogarithm has leading $m$'s, we remove them and end up with multiple polylogarithms without leading $m$'s and powers 
of $\H{m}{\frac{1}{x}}.$ To deal with the multiple polylogarithms without leading $m$, we can slightly modify the previous part of this section and for the powers of $\H{m}{\frac{1}{x}}$ we can 
use~(\ref{SStrafo1dx11}).

\section{Power Series Expansion of Multiple Polylogarithms}
As stated in Section \ref{HSPowerExp} in general, the harmonic polylogarithms $\H{\ve m}x$ do not have a regular Taylor series expansion. This is due to the
effect that trailing zeroes in the index field may cause powers of $\log(x).$ Hence the proper expansion
is one in terms of both $x$ and $\log{x}$. The same is true for multiple polylogarithms with indices in $\R$. In this case we have for $j\in \R$ and $x\in(0,\abs{j})$:
\begin{eqnarray*}
\H{j}x=\left\{ 
		\begin{array}{ll}
				-\sum_{i=1}^\infty \frac{j^{-i}(-x)^i}{i},&  \textnormal{if }j< 0\\
\\
				\sum_{i=1}^\infty \frac{j^{-i}x^i}{i},& \textnormal{if }j > 0. 
		\end{array} \right.
\end{eqnarray*}
Set $\ve m=(m_1,\ldots,m_l)$ and $q=\min_{m_i\neq 0}{\abs{m_i}}$. Assuming that for $x \in (0,q)$
$$\H{\ve m}x=\sum_{i=1}^{\infty}\frac{\sigma^ix^i}{i^a}\S{\ve n}{\ve b;i}$$ one can write the relations for $j>0$ and $x \in (0, \min(j,q))$
\begin{eqnarray*}
\H{0,\ve m}x&=&\sum_{i=1}^{\infty}\frac{\sigma^ix^i}{i^{a+1}}\S{\ve n}{\ve b;i}\\
\H{j,\ve m}x&=&\sum_{i=1}^{\infty}\frac{x^i}{ij^i}\S{a,\ve n}{\sigma j,\ve b;i-1}\\
&=&\sum_{i=1}^{\infty}\frac{x^i}{ij^i}\S{a,\ve n}{\sigma j,\ve b;i}-\sum_{i=1}^{\infty}\frac{\sigma^ix^i}{i^{a+1}}\S{\ve n}{\ve b;i}\\
\H{-j,\ve m}x&=&-\sum_{i=1}^{\infty}\frac{x^i}{i(-j)^i}\S{a,\ve n}{-\sigma j,\ve b;i-1}\\
&=&-\sum_{i=1}^{\infty}\frac{x^i}{i(-j)^i}\S{a,\ve n}{-\sigma j,\ve b;i}+\sum_{i=1}^{\infty}\frac{\sigma^ix^i}{i^{a+1}}\S{\ve n}{\ve b;i}.
\end{eqnarray*}
\begin{proof}
We just prove the case that $j>0$. The other cases follow analogously.
 \begin{eqnarray*}
\H{j,\ve m}x&=&\int_0^x{\frac{1}{j-y}\sum_{i=1}^{\infty}\frac{\sigma^iy^i}{i^a}\S{\ve n}{\ve b;i}}dy\\
&=&\int_0^x{\frac{1}{j}\sum_{k=0}^\infty\left(\frac{y}{j}\right)^k\sum_{i=1}^{\infty}\frac{\sigma^iy^i}{i^a}\S{\ve n}{\ve b;i}}dy\\
&=&\frac{1}{j}\int_0^x{\sum_{k=0}^\infty\left(\frac{y}{j}\right)^k\sum_{i=0}^{\infty}\frac{\sigma^{i+1}y^{i+1}}{(i+1)^a}\S{\ve n}{\ve b;i+1}}dy\\
&=&\frac{1}{j}\int_0^x{\sum_{i=0}^\infty\sum_{k=0}^{i}\left(\frac{y}{j}\right)^{i-k}\frac{\sigma^{k+1}y^{k+1}}{(k+1)^a}\S{\ve n}{\ve b;k+1}}dy\\
&=&\int_0^x{\sum_{i=0}^\infty\frac{y^{i+1}}{j^{i+2}}\S{a,\ve n}{\sigma j,\ve b;i+1}}dy\\
&=&\sum_{i=1}^\infty\frac{x^{i+1}}{j^{i+1}(i+1)}\S{a,\ve n}{\sigma j,\ve b;i}\\
&=&\sum_{i=1}^\infty\frac{x^{i}}{ij^{i}}\S{a,\ve n}{\sigma j,\ve b,i-1}.
 \end{eqnarray*}
\end{proof}
Note that again the reverse direction is possible, \ie given a sum of the form
$$
\sum_i^{\infty}(c x)^i\frac{\S{\ve n}{\ve b; i}}{i^k}
$$
with $k\in \N_0$ and $c\in\R^*,$ we can find an linear combination of (possibly weighted) multiple polylogarithms $h(x)$ such that for $x\in(0,j)$ with $j>0$ small enough 
$$
\sum_i^{\infty}(c x)^i\frac{\S{\ve n}{\ve b; i}}{i^k}=h(x).
$$
\begin{example}For $x\in(0,\frac{1}{2})$ we have
\begin{eqnarray*}
\sum_i^{\infty}\frac{2^i x^i S_{1,1}\left(3,\frac{1}{2};i\right)}{i}&=&
   \textnormal{H}_{0,0,\frac{1}{3}}(x)+\textnormal{H}_{0,\frac{1}{6},\frac{1}{3}}(x)+\textnormal{H}_{\frac{1}{2},0,\frac{1}{3}}(x)+\textnormal{H}_{\frac{1}{2},\frac{1}{6},\frac{1}{3}}(x).
\end{eqnarray*}
\end{example}

\subsection{Asymptotic Behavior of Multiple Polylogarithms}
\label{SSasybeh}
Similar to Section \ref{HSasybeh} we can combine Section \ref{SS1dxx} together with the power series expansion of multiple polylogarithms to determine the 
asymptotic behavior of multiple polylogarithms. 
Let us look at the multiple polylogarithm 
$\H{\ve m}x$ and define $y:=\frac{1}{x}.$ Using Section \ref{SS1dxx} on $\H{\ve m}{\frac{1}{y}}=\H{\ve m}x$ we can rewrite $\H{\ve m}x$ in terms of multiple polylogarithms at argument $y$ together with some constants.
Now we can get the power series expansion of the multiple polylogarithms at argument $y$ around 0 easily using the previous part of this Section. Since sending $x$ to infinity corresponds to sending $y$ to 
zero, we get the asymptotic behavior of $\H{\ve m}x.$
\begin{example}
 \begin{eqnarray*}
  \H{-2,0}x&=& \H{-2,0}1+\H{-\frac{1}{2},0}1-\H{-\frac{1}{2},0}{\frac{1}x}-\H{0,0}1+\H{0,0}{\frac{1}x}\\
	   &=&\frac{1}{2}\; \H0x^2-\H0x \left(\sum _{i=1}^{\infty } \frac{\left(-\frac{2}{x}\right)^{i}}{i}\right)-\sum _{i=1}^{\infty}\frac{\left(-\frac{2}{x}\right)^{i}}{i^2}\\
	   & &+\H{-2,0}1+\H{-\frac{1}{2},0}1-\H{0,0}1. 
 \end{eqnarray*}
\end{example}

\subsection{Values of Multiple Polylogarithms Expressed by S-Sums at Infinity}
\label{SSinfval}

For $x = c, c\in (0,1]$ we have that
$$
\sum_{i=1}^{\infty}\frac{\sigma^ix^i}{i^a}\S{\ve n}{\ve b;i} \rightarrow \S{a,\ve n}{c\sigma,\ve b;\infty}
$$
and hence (due to the previous part of this section) the values of the multiple polylogarithms in $x = c$ are related to the values of the generalized multiple 
harmonic sums at infinity. In this section we look at the problem to rewrite a multiple polylogarithm $\H{\ve{m}}c$ ($c$ such that $\H{\ve{m}}c$ is finite) in 
terms of finite S-sums at infinity.

If $c<0,$ we use Lemma \ref{SStransformminusx} to transform $\H{\ve{m}}c$ to an expression out of multiple polylogarithms at argument $-c.$ Hence we can assume $c>0$.
If the multiple polylogarithm $\H{\ve{m}}c$ has trailing zeroes we first extract them. If $0<c\leq 1,$ we end up in an expression of multiple polylogarithms without 
trailing zeroes and powers of $\H{0}{c},$ while there can also contribute multiple polylogarithms of the form $\H{a_1,a_2,\ldots,a_k,1,0,\ldots,0}c$ if $c>1$ (see Section~\ref{SSdef}).

To rewrite multiple polylogarithms of the form $\H{a_1,a_2,\ldots,a_k}c$ where $a_k\neq 0$ and $\min_{a_i\neq 0}{\abs{a_i}}>c$ we use the power series expansion 
of $\H{a_1,a_2,\ldots,a_k}x$ and finally sending $x\rightarrow c.$

We can use the following lemma to rewrite powers of $\H{0}{c}.$
\begin{lemma}
 Let $c>0.$ We have
\begin{eqnarray*}
  0<c< 1:&&\H{0}c=-\S{1}{1-c;\infty}\\
  c=1:&&\H{0}c=0\\
  1<c\leq 2:&&\H{0}c=-\S{1}{1-c;\infty}\\
  2<c:&&\H{0}c=\H{0}2-\H{0}{\frac{c}{2}}.
\end{eqnarray*}
\end{lemma}

To rewrite $\H{a_1,a_2,\ldots,a_k}c$ with $(a_1,a_2,\ldots,a_k)=(1,0,\ldots,0)$ and $c>1$ we use the following lemma.
\begin{lemma}
 For $c>0$ and  $(a_1,a_2,\ldots,a_k)=(1,0,\ldots,0)$ we have
$$
\H{a_1,a_2,\ldots,a_k}c=\H{a_1,a_2,\ldots,a_k}1-\H{0,a_1-1,\ldots,a_k-1}{c-1}.
$$
\end{lemma}

To rewrite multiple polylogarithms of the form $\H{a_1,a_2,\ldots,a_k}c$ where $a_1\in \R \setminus (0, c]$,  $a_k\neq 0,$ $\min_{a_i > 0}{\abs{a_i}}\geq c$ and 
$s:=\max_{a_i<0}{a_i}>-c$ we use Lemma \ref{SStransformplusa1} to rewrite $\H{a_1,a_2,\ldots,a_k}{x-s}$ and afterwards send $x\rightarrow c+s.$ Now there are 
two cases: Either we can already handle the arising polylogarithms since we arrive in a previous case or we have to apply the same strategy once more.

Finally, to rewrite multiple polylogarithms of the form $\H{a_1,a_2,\ldots,a_k,1,0,\ldots,0}c$ where $c>1$ we can proceed similar to the previous case: we use 
Lemma \ref{SStransformplusa3} to rewrite $\H{a_1,a_2,\ldots,a_k}{x+1}$ and afterwards send $x\rightarrow c-1.$ Now we arrive at multiple polylogarithms 
at values we can already handle.

Summarizing, we can rewrite all finite multiple polylogarithms $\H{a_1,a_2,\ldots,a_k}c$ in terms of S-sums at infinity.

\begin{example}
\begin{eqnarray*}
 \H{4,1,0}{3}&=&-\S{-3}{\infty} + 2 \S{3}{\infty} + \S{-2}{\infty} \S{1}{\frac{1}{2};\infty} - \S{2}{\infty} \S{1}{\frac{2}{3};\infty}\\
	      &&- \S{3}{-\frac{1}{2};\infty} + \S{-1}{\infty} \left(\S{-2}{\infty} - \S{1, 1} {\frac{1}{2}, -2;\infty}\right) \\
	      && - \S{1, 2} {\frac{1}{4}, 4;\infty}+ \S{1, 2} {\frac{1}{3}, -3;\infty} + \S{1, 2}{\frac{1}{2}, -1;\infty} \\
	      && + \S{2, 1}{-1, \frac{1}{2};\infty} - \S{2, 1}{\frac{1}{4}, 4;\infty}- \S{1, 1, 1}{\frac{1}{2}, -2, \frac{1}{2};\infty}.
\end{eqnarray*}
\end{example}

Note that this process can be reversed. In particular, using Algorithm \ref{SSsinftoh} also the other direction is possible; given a finite S-sum at infinity we can rewrite it in terms of 
multiple polylogarithms at one.

\begin{example}
\begin{eqnarray*}
 \SS{2,3,1}{-\frac{1}{2},\frac{1}{3},2}{\infty}&=&-\H{0, -2, 0, 0, -6, -3}1 + \H{0, -2, 0, 0, 0, -3}1 + \H{0, 0, 0, 0, -6, -3}1\\
						&& - \H{0, 0, 0, 0, 0, -3}1.
\end{eqnarray*}
\end{example}

\begin{algorithm}
\caption{}
\label{SSsinftoh}
\begin{algorithmic}
\State \bfseries input:\ \normalfont a finite S-sum $\S{a_1,\ldots,a_k}{b_1,\ldots,b_k;\infty},$ with $a_i \in \N$ and $b_i \in \R^*$
\State \bfseries output:\ \normalfont expression in terms of multiple polylogarithms evaluating to finite constants
\Procedure{SinfToH}{$\S{a_1,\ldots,a_k}{b_1,\ldots,b_k;n}$}
\State $c_1=\frac{1}{b_1}$
\For{$i=2$ to $k$}
	\State $c_i=\frac{b_{i-1}}{b_i}$
\EndFor
\State $l=$ number of entries $<0$ in $(c_1,\ldots,c_k)$
\State $\bar{k}=\sum_{i=1}^k{a_i}$
\State $j=1$
\For{$i=1$ to $k$}
	\While{$a_i>1$}
		\State $a_i=a_i-1$
		\State $d_j=0$
		\State $j=j+1$
	\EndWhile
	\State $d_j=c_i$
	\State $j=j+1$
\EndFor
\State $h=(-1)^l\cdot \H{d_1,d_2,\ldots,d_{\bar{k}}}{1}$
\State rewrite $h$ in terms of S-sums $\rightarrow \bar{h}$
\State \textbf{return} $h+\Call{SinfToH}{\S{a_1,\ldots,a_k}{b_1,\ldots,b_k;n}-\bar{h}}$
\EndProcedure
\end{algorithmic}
\end{algorithm}

\section{Integral Representation of S-Sums}
\label{SSIntegralRep}
In this section we want to derive an integral representation for S-sums. First we will find a multidimensional integral which can afterwards be rewritten as a sum 
of one-dimensional integrals over multiple polylogarithms using integration by parts. We start with the base case of a S-sum of depth one and then give a detailed example.

\begin{lemma}[Compare Lemma \ref{HSintrep1}]
Let $m\in\N,$ $b\in\R^*$ and $n\in\N,$ then
\begin{eqnarray*}
\S{1}{b;n}&=&\int_0^{b}{\frac{x_1^n-1}{x_1-1}dx_1}\\
\S{2}{b;n}&=&\int_0^{b}{\frac{1}{x_2}\int_0^{x_2}{\frac{x_1^n-1}{x_1-1}dx_1}dx_2}\\
\S{m}{b;n}&=&\int_0^b{\frac{1}{x_m}\int_0^{x_m}{\frac{1}{x_{m-1}} \cdots \int_0^{x_3}{\frac{1}{x_2}\int_0^{x_2}{\frac{x_1^n-1}{x_1-1}dx_1}dx_2}\cdots dx_{m-1}}dx_m}.
\end{eqnarray*}
\label{SSintrep1}
\end{lemma}
\begin{proof}
We proceed by induction on the weight $m.$ For $m=1$ we get:
\begin{eqnarray*}
\int_0^{b}{\frac{x_1^n-1}{x_1-1}dx_1}&=&\int_0^{b}{\sum_{i=0}^{n-1}{{x_1}^i}dx_1}=\sum_{i=0}^{n-1}\int_0^{b}{{{x_1}^i}dx_1}=\sum_{i=0}^{n-1}{\frac{b^{i+1}}{i+1}}
	=\sum_{i=1}^{n}{\frac{b^{i}}{i}}=\S{1}{b;n}.
\end{eqnarray*}
Now suppose the theorem holds for weight $m-1$, we get
\begin{eqnarray*}
&&\int_0^b{\frac{1}{x_m}\int_0^{x_m}{\frac{1}{x_{m-1}} \cdots \int_0^{x_3}{\frac{1}{x_2}\int_0^{x_2}{\frac{x_1^n-1}{x_1-1}dx_1}dx_2}\cdots dx_{m-1}}dx_m}=\\
&&\hspace{0.5cm}\int_0^b{\frac{\S{m-1}{x_m;n}}{x_m}dx_m}=\int_0^b{\sum_{i=1}^n\frac{{x_m}^{i-1}}{i^{m-1}}dx_m}=\sum_{i=1}^n\int_0^b{\frac{{x_m}^{i-1}}{i^{m-1}}dx_m}=\\
&&\hspace{0.5cm}\sum_{i=1}^n\frac{b^{i}}{i}\frac{1}{i^{m-1}}=\S{m}{b;n}.
\end{eqnarray*}

\end{proof}

\begin{example}
We consider the sum $\S{1,2,1}{a,b,c;n}$ with $a,b,c\in\R^*.$ Due to the definition and Lemma \ref{SSintrep1} we have
\begin{eqnarray*}
 \S{1,2,1}{a,b,c;n}&=&\sum_{k=1}^n\frac{a^k}{k}\S{2,1}{b,c;k}=\sum_{k=1}^n{\frac{a^k}{k}\sum_{j=1}^k\frac{b^j}{j^2}\S{1}{c;j}}\\
	&=&\sum_{k=1}^n{\frac{a^k}{k}\sum_{j=1}^k\frac{b^j}{j^2}\int_{0}^c\frac{x^j-1}{x-1}dx}\\
	&=&\sum_{k=1}^n{\frac{a^k}{k}\int_{0}^c{\frac{1}{x-1}\sum_{j=1}^k\frac{b^j}{j^2}\left(x^j-1\right)dx}}\\
	&=&\sum_{k=1}^n{\frac{a^k}{k}\int_{0}^c{\frac{1}{x-1}\left(\S{2}{bx;k}-\S{2}{b;k}\right)dx}}\\
	&=&\sum_{k=1}^n{\frac{a^k}{k}\int_{0}^c{\frac{1}{x-1}\int_b^{bx}{\frac{1}{y}\int_0^y{\frac{z^k-1}{z-1}dz}dy}dx}}\\
	&=&\sum_{k=1}^n{\frac{a^k}{k}\int_0^{bc}{\frac{1}{x-b}\int_b^{x}{\frac{1}{y}\int_0^y{\frac{z^k-1}{z-1}dz}dy}dx}}\\
	&=&\int_0^{bc}{\frac{1}{x-b}\int_b^{x}{\frac{1}{y}\int_0^y{\frac{1}{z-1}\sum_{k=1}^n{\frac{a^k(z^k-1)}{k}}dz}dy}dx}\\
	&=&\int_0^{bc}{\frac{1}{x-b}\int_b^{x}{\frac{1}{y}\int_0^y{\frac{a}{z-1}\int_1^z\sum_{k=1}^n{(aw)^{k-1}}dwdz}dy}dx}\\
	&=&\int_0^{bc}{\frac{1}{x-b}\int_b^{x}{\frac{1}{y}\int_0^y{\frac{a}{z-1}\int_1^z{\frac{(aw)^{n}-1}{aw-1}}dwdz}dy}dx}\\
	&=&\int_0^{bc}{\frac{1}{x-b}\int_b^{x}{\frac{1}{y}\int_0^y{\frac{1}{z-1}\int_a^{za}{\frac{w^{n}-1}{w-1}}dwdz}dy}dx}\\
	&=&\int_0^{bc}{\frac{1}{x-b}\int_b^{x}{\frac{1}{y}\int_0^{ya}{\frac{1}{\frac{z}{a}-1}\frac{1}{a}\int_a^{z}{\frac{w^{n}-1}{w-1}}dwdz}dy}dx}\\
	&=&\int_0^{bc}{\frac{1}{x-b}\int_b^{x}{\frac{1}{y}\int_0^{ya}{\frac{1}{z-a}\int_a^{z}{\frac{w^{n}-1}{w-1}}dwdz}dy}dx}\\
	&=&\int_0^{bc}{\frac{1}{x-b}\int_{ba}^{xa}{\frac{1}{y}\int_0^{y}{\frac{1}{z-a}\int_a^{z}{\frac{w^{n}-1}{w-1}}dwdz}dy}dx}\\
	&=&\int_0^{abc}{\frac{1}{x-ab}\int_{ab}^{x}{\frac{1}{y}\int_0^{y}{\frac{1}{z-a}\int_a^{z}{\frac{w^{n}-1}{w-1}}dwdz}dy}dx}.
\end{eqnarray*}
\end{example}
Inspired by the previous example we arrive at the following theorem.
\begin{thm}[Compare Theorem \ref{HSintrep}]
Let $m_i\in\N,$ $b_i\in\R^*$ and $n\in\N,$ then
\begin{eqnarray*}
&&\S{m_1,m_2,\ldots,m_k}{b_1,b_2,\ldots,b_k;n}=\\
&&\hspace{0.4cm}\int_0^{b_1\cdots b_k}{\frac{dx_{k}^{m_k}}{x_{k}^{m_k}}\int_0^{x_{k}^{m_k}}{\frac{dx_{k}^{m_k-1}}{x_{k}^{m_k-1}} \cdots
\int_0^{x_{k}^3}{\frac{dx_{k}^2}{x_{k}^2}\int_0^{x_{k}^2}{\frac{dx_{k}^1}{x_{k}^1-b_1\cdots b_{k-1}}}}}}\\
&&\hspace{0.4cm}\int^{x_{k}^1}_{b_1\cdots b_{k-1}}{\frac{dx_{k-1}^{m_{k-1}}}{x_{{k-1}}^{m_{k-1}}}\int_0^{x_{{k-1}}^{m_{k-1}}}{\frac{dx_{{k-1}}^{m_{k-1}-1}}{x_{{k-1}}^{m_{k-1}-1}} \cdots
\int_0^{x_{{k-1}}^3}{\frac{dx_{{k-1}}^2}{x_{{k-1}}^2}}}}\int_0^{x_{{k-1}}^2}\hspace{-0.4em}{\frac{dx_{{k-1}}^1} {x_{{k-1}}^1-b_1\cdots b_{k-2}}}\\
&&\hspace{0.4cm}\int^{x_{{k-1}}^1}_{b_1\cdots b_{k-2}}{\frac{dx_{{k-2}}^{m_{k-2}}}{x_{{k-2}}^{m_{k-2}}}\int_0^{x_{{k-2}}^{m_{k-2}}}{\frac{dx_{{k-2}}^{m_{k-2}-1}}{x_{{k-2}}^{m_{k-2}-1}} \cdots
\int_0^{x_{{k-2}}^3}{\frac{dx_{{k-2}}^2}{x_{{k-2}}^2}}}}\int_0^{x_{{k-2}}^2}\hspace{-0.4em}{\frac{dx_{{k-2}}^1} {x_{{k-2}}^1-b_1\cdots b_{k-3}}}\\
&&\vspace{0.1cm}\\
&&\hspace{0cm}\hbox to 0.4\textwidth{}\vdots\\
&&\vspace{0.1cm}\\
&&\hspace{0.4cm}\int^{x_{4}^1}_{b_1b_2b_3}{\frac{dx_{3}^{m_3}}{x_{3}^{m_3}}\int_0^{x_{3}^{m_3}}{\frac{dx_{3}^{m_3-1}}{x_{3}^{m_3-1}} \cdots \int_0^{x_{3}^3}{\frac{dx_{3}^2}{x_{3}^2}\int_0^{x_{3}^2}{\frac{dx_{3}^1}{x_{3}^1-b_1b_2}}}}}\\
&&\hspace{0.4cm}\int^{x_{3}^1}_{b_1b_2}{\frac{dx_{2}^{m_2}}{x_{2}^{m_2}}\int_0^{x_{2}^{m_2}}{\frac{dx_{2}^{m_2-1}}{x_{2}^{m_2-1}} \cdots \int_0^{x_{2}^3}{\frac{dx_{2}^2}{x_{2}^2}\int_0^{x_{2}^2}{\frac{dx_{2}^1}{x_{2}^1-b_1}}}}}\\
&&\hspace{0.4cm}\int^{x_{2}^1}_{b_1}{\frac{dx_{1}^{m_1}}{x_{1}^{m_1}}\int_0^{x_{1}^{m_1}}{\frac{dx_{1}^{m_1-1}}{x_{1}^{m_1-1}} \cdots \int_0^{x_{1}^3}{\frac{dx_{1}^2}{x_{1}^2}\int_0^{x_{1}^2}{\frac{\left({x_{1}^1}\right)^n-1}{x_{1}^1-1}dx_{1}^1}}}}.
\end{eqnarray*}
\label{SSintrep}
\end{thm}
\begin{proof}
 We proceed by induction on depth $k.$ For $k=1$ see Lemma \ref{SSintrep1}. Now let $k>1.$ Due to the induction hypothesis we get  
\small
\begin{eqnarray*}
&&\S{m_1,m_2,\ldots, m_k}{b_1,b_2,\ldots,b_k;n}=\sum_{i=1}^n\frac{b_1^i}{i^{m_1}}\S{m_2,\ldots, m_k}{b_2,\ldots,b_k;i}\\
&&\\
&&\hspace{0.5cm}=\sum_{i=1}^n\frac{b_1^i}{i^{m_1}}\int_0^{b_2\cdots b_{k}}{\frac{dx_{k}^{m_{k}}}{x_{k}^{m_{k}}} \hspace{0.3cm}\cdots \hspace{0.3cm}
\int_0^{x_k^3}{\frac{dx_k^2}{x_k^2}\int_0^{x_k^2}{\frac{dx_k^1} {x_k^1-b_2\cdots b_{k-1}}}}}\\
&&\hspace{1cm}\int^{x_{k_1}}_{b_2\cdots b_{k-1}}{\frac{dx_{k-1}^{m_{k-1}}}{x_{k-1}^{m_{k-1}}} \hspace{0.3cm}\cdots \hspace{0.3cm}
\int_0^{x_{k-1}^3}{\frac{dx_{k-1}^2}{x_{k-1}^2}\int_0^{x_{k-1}^2}{\frac{dx_{k-1}^1} {x_{k-1}^1-b_2\cdots b_{k-2}}}}}\\
&&\hspace{1cm}\hbox to 0.2\textwidth{}\vdots\\
&&\hspace{1cm}\int^{x_4^1}_{b_2b_3}{\frac{dx_{3}^{m_3}}{x_{3}^{m_3}}\hspace{0.3cm}\cdots \hspace{0.3cm} \int_0^{x_3^3}{\frac{dx_3^2}{x_3^2}\int_0^{x_3^2}{\frac{dx_3^1}{x_3^1-b_2}}}}\\
&&\hspace{1cm}\int^{x_3^1}_{b_2}{\frac{dx_{2}^{m_2}}{x_{2}^{m_2}} \hspace{0.3cm}\cdots \hspace{0.3cm}\int^{x_1^1}_{x_2^3}{\frac{dx_2^2}{x_2^2}\int_0^{x_2^2}{\frac{\left({x_2^1}\right)^i-1}{x_2^1-1}dx_2^1}}}\\
&&\hspace{0.5cm}=\int_0^{b_2\cdots b_{k}}{\frac{dx_{k}^{m_{k}}}{x_{k}^{m_{k}}} \hspace{0.3cm}\cdots \hspace{0.3cm}
\int_0^{x_k^3}{\frac{dx_k^2}{x_k^2}\int_0^{x_k^2}{\frac{dx_k^1} {x_k^1-b_2\cdots b_{k-1}}}}}\\
&&\hspace{1cm}\int^{x_{k_1}}_{b_2\cdots b_{k-1}}{\frac{dx_{k-1}^{m_{k-1}}}{x_{k-1}^{m_{k-1}}} \hspace{0.3cm}\cdots \hspace{0.3cm}
\int_0^{x_{k-1}^3}{\frac{dx_{k-1}^2}{x_{k-1}^2}\int_0^{x_{k-1}^2}{\frac{dx_{k-1}^1} {x_{k-1}^1-b_2\cdots b_{k-2}}}}}\\
&&\hspace{1cm}\hbox to 0.2\textwidth{}\vdots\\
&&\hspace{1cm}\int^{x_4^1}_{b_2b_3}{\frac{dx_{3}^{m_3}}{x_{3}^{m_3}}\hspace{0.3cm}\cdots \hspace{0.3cm} \int_0^{x_3^3}{\frac{dx_3^2}{x_3^2}\int_0^{x_3^2}{\frac{dx_3^1}{x_3^1-b_2}}}}\\
&&\hspace{1cm}\int^{x_3^1}_{b_2}{\frac{dx_{2}^{m_2}}{x_{2}^{m_2}} \hspace{0.3cm}\cdots \hspace{0.3cm}\int^{x_2^1}_{x_2^3}{\frac{dx_2^2}{x_2^2}\int_0^{x_2^2}{\frac{dx_2^1}{x_2^1-1}\sum_{i=1}^n\frac{b_1^i\left(\left({x_2^1}\right)^i-1\right)}{i^{m_1}}}}}\\
&&\\
&&\hspace{0.5cm}=\int_0^{b_2\cdots b_{k}}{\frac{dx_{k}^{m_{k}}}{x_{k}^{m_{k}}} \hspace{0.3cm}\cdots \hspace{0.3cm}
\int_0^{x_k^3}{\frac{dx_k^2}{x_k^2}\int_0^{x_k^2}{\frac{dx_k^1} {x_k^1-b_2\cdots b_{k-1}}}}}\\
&&\hspace{1cm}\int^{x_{k_1}}_{b_2\cdots b_{k-1}}{\frac{dx_{k-1}^{m_{k-1}}}{x_{k-1}^{m_{k-1}}} \hspace{0.3cm}\cdots \hspace{0.3cm}
\int_0^{x_{k-1}^3}{\frac{dx_{k-1}^2}{x_{k-1}^2}\int_0^{x_{k-1}^2}{\frac{dx_{k-1}^1} {x_{k-1}^1-b_2\cdots b_{k-2}}}}}\\
&&\hspace{1cm}\hbox to 0.2\textwidth{}\vdots\\
&&\hspace{1cm}\int^{x_4^1}_{b_2b_3}{\frac{dx_{3}^{m_3}}{x_{3}^{m_3}}\hspace{0.3cm}\cdots \hspace{0.3cm} \int_0^{x_3^3}{\frac{dx_3^2}{x_3^2}\int_0^{x_3^2}{\frac{dx_3^1}{x_3^1-b_2}}}}\\
&&\hspace{1cm}\int^{x_3^1}_{b_2}{\frac{dx_{2}^{m_2}}{x_{2}^{m_2}} \hspace{0.3cm}\cdots \hspace{0.3cm}\int^{x_2^1}_{x_2^3}{\frac{dx_2^2}{x_2^2}\int_0^{x_2^2}{\frac{dx_2^1}{x_2^1-1}\left(\SS{m_1}{b_1x_2^1}i-\SS{m_1}{b_1}i\right)}}}\\
&&\\
&&\hspace{0.5cm}=\int_0^{b_2\cdots b_{k}}{\frac{dx_{k}^{m_{k}}}{x_{k}^{m_{k}}} \hspace{0.3cm}\cdots \hspace{0.3cm}
\int_0^{x_k^3}{\frac{dx_k^2}{x_k^2}\int_0^{x_k^2}{\frac{dx_k^1} {x_k^1-b_2\cdots b_{k-1}}}}}\\
&&\hspace{1cm}\int^{x_{k_1}}_{b_2\cdots b_{k-1}}{\frac{dx_{k-1}^{m_{k-1}}}{x_{k-1}^{m_{k-1}}} \hspace{0.3cm}\cdots \hspace{0.3cm}
\int_0^{x_{k-1}^3}{\frac{dx_{k-1}^2}{x_{k-1}^2}\int_0^{x_{k-1}^2}{\frac{dx_{k-1}^1} {x_{k-1}^1-b_2\cdots b_{k-2}}}}}\\
&&\hspace{1cm}\hbox to 0.2\textwidth{}\vdots\\
&&\hspace{1cm}\int^{x_4^1}_{b_2b_3}{\frac{dx_{3}^{m_3}}{x_{3}^{m_3}}\hspace{0.3cm}\cdots \hspace{0.3cm} \int_0^{x_3^3}{\frac{dx_3^2}{x_3^2}\int_0^{x_3^2}{\frac{dx_3^1}{x_3^1-b_2}}}}\\
&&\hspace{1cm}\int^{x_3^1}_{b_2}{\frac{dx_{2}^{m_2}}{x_{2}^{m_2}} \hspace{0.3cm}\cdots \hspace{0.3cm}\int_0^{x_2^3}{\frac{dx_2^2}{x_2^2}\int_0^{x_2^2}{\frac{dx_2^1}{x_2^1-1}}}}\\
&&\hspace{1cm}\int_{b_1}^{b_1x_2^1}{\frac{dx_{1}^{m_1}}{x_{1}^{m_1}} \cdots \int_0^{x_{1_3}}{\frac{1}{x_1^2}\int_0^{x_1^2}{\frac{\left(x_1^1\right)^n-1}{x_1^1-1}}}}\\
&&\\
&&\hspace{0.5cm}=\int_0^{b_1b_2\cdots b_{k}}{\frac{dx_{k}^{m_{k}}}{x_{k}^{m_{k}}} \hspace{0.3cm}\cdots \hspace{0.3cm}
\int_0^{x_k^3}{\frac{dx_k^2}{x_k^2}\int_0^{x_k^2}{\frac{dx_k^1} {x_k^1-b_1b_2\cdots b_{k-1}}}}}\\
&&\hspace{1cm}\int^{x_{k_1}}_{b_1b_2\cdots b_{k-1}}{\frac{dx_{k-1}^{m_{k-1}}}{x_{k-1}^{m_{k-1}}} \hspace{0.3cm}\cdots \hspace{0.3cm}
\int_0^{x_{k-1}^3}{\frac{dx_{k-1}^2}{x_{k-1}^2}\int_0^{x_{k-1}^2}{\frac{dx_{k-1}^1} {x_{k-1}^1-b_1b_2\cdots b_{k-2}}}}}\\
&&\hspace{1cm}\hbox to 0.2\textwidth{}\vdots\\
&&\hspace{1cm}\int^{x_4^1}_{b_1b_2b_3}{\frac{dx_{3}^{m_3}}{x_{3}^{m_3}}\hspace{0.3cm}\cdots \hspace{0.3cm} \int_0^{x_3^3}{\frac{dx_3^2}{x_3^2}\int_0^{x_3^2}{\frac{dx_3^1}{x_3^1-b_1b_2}}}}\\
&&\hspace{1cm}\int^{x_3^1}_{b_1b_2}{\frac{dx_{2}^{m_2}}{x_{2}^{m_2}} \hspace{0.3cm}\cdots \hspace{0.3cm}\int_0^{x_2^3}{\frac{dx_2^2}{x_2^2}\int_0^{x_2^2}{\frac{dx_2^1}{x_2^1-b_1}}}}\\
&&\hspace{1cm}\int_{b_1}^{x_2^1}{\frac{dx_{1}^{m_1}}{x_{1}^{m_1}} \cdots \int_0^{x_{1_3}}{\frac{1}{x_1^2}\int_0^{x_1^2}{\frac{\left(x_1^1\right)^n-1}{x_1^1-1}}}}.
\end{eqnarray*}
\normalsize
\end{proof}

\subsection{Mellin Transformation of Multiple Polylogarithms with Indices in \texorpdfstring{$\R\setminus(0,1)$}{R-(0,1)}}
\label{SSMellin}
In this section we look at the Mellin-transform of multiple polylogarithms with indices in $\R\setminus(0,1)$. It will turn out that these transforms can be expressed using a 
subclass of the S-sums together with certain constants.

For $f(x)=1/(a-x)$ with $a \in (0,1),$ the Mellin transform is not defined since the integral $\int_0^1\frac{x^n}{a-x}$ does not converge. We modify the definition of the 
Mellin transform form Definition \ref{HSabmellplus} to include these functions, like $1/(a-x),$ as follows (note that all the arising integrals are well defined).

\begin{definition}[Compare Definition \ref{HSabmellplus}]
Let $h(x)$ be a multiple polylogarithm with indices in $\R\setminus(0,1)$ or $h(x)=1$ for $x\in [0,1]$; let $a \in (0,\infty)$, $a_1\in (1,\infty)$, $a_2 \in (0,1]$. Then we define the 
extended and modified Mellin-transform as follows:
\begin{eqnarray}
\Mp{h(x)}{n}&=&\M{h(x)}{n}=\int_0^1{x^nh(x)dx},\label{SSmeldef1}\\
\Mp{\frac{h(x)}{a+x}}n&=&\M{\frac{h(x)}{a+x}}{n}=\int_0^1{\frac{x^nh(x)}{a+x}dx},\label{SSmeldef2}\\
\Mp{\frac{h(x)}{a_1-x}}n&=&\M{\frac{h(x)}{a_1-x}}{n}=\int_0^1{\frac{x^nh(x)}{a_1-x}dx},\label{SSmeldef3}\\
\Mp{\frac{h(x)}{a_2-x}}n&=&\int_0^{\frac{1}{a_2}}{\frac{(x^n-1)h(a_2\; x)}{1-x}}dx=\int_0^1{\frac{((\frac{x}{a_2})^n-1)h(x)}{a_2-x}dx}.\label{SSmeldef4}
\end{eqnarray}
\label{SSabmellplus}
\end{definition}

\begin{remark}
The definitions (\ref{SSmeldef1}),(\ref{SSmeldef2}),(\ref{SSmeldef3}) are just the original Mellin transform (see (\ref{HSabmell})). Only in (\ref{SSmeldef4}) we extended 
the original Mellin transform; besides that an additional difference was already mentioned in Remark \ref{HSMelRemark}.
\end{remark}
\begin{remark}
In (\ref{SSmeldef4}) both extensions $\int_0^{\frac{1}{a_2}}{\frac{x^n-1}{1-x}}dx$ and $\int_0^1{\frac{((\frac{x}{a_2})^n-1)h(x)}{a_2-x}dx}$ are equivalent, which can be seen 
using a simple substitution. Since from an algorithmic point of view it is easier to handle the second integral, we prefer this representation.
\end{remark}
\begin{remark}
From now on we will call the extended and modified Mellin transform $M^+$ just Mellin transform and we will write $M$ instead of $M^+.$
\end{remark}
Subsequently, we want to study how we can actually calculate the Mellin transform of multiple polylogarithms with indices in $\R\setminus(0,1)$ weighted by $1/(a + x)$ or $1/(a-x)$ 
for $a\in \R$. This will be possible due to the following lemma which is an extension of Lemma~\ref{HSmelweighted}. 
\begin{lemma}[Compare Lemma \ref{HSmelweighted}] For $n\in \N,$ $\ve m~=~(m_1,\overline{\ve{m}})~=~(m_1,m_2,\ldots,m_k)$ with $m_i\in \R \setminus (0,1)$ and 
$a \in (0,\infty)$, $a_1\in (1,\infty)$, $a_2 \in (0,1)$, we have
\begin{eqnarray*}
\M{\frac{1}{a+x}}n &=&\left\{
	\begin{array}{ll}
		(-a)^n\left(\S{1}{-\frac{1}{a};n}+\H{-a}{1}\right),& \textnormal{if } 0< a< 1\\
		(-a)^n\left(\S{1}{-\frac{1}{a};n}-\S{1}{-\frac{1}{a};\infty}\right),& \textnormal{if } 1\geq a\\
		 \end{array} \right.\\
\M{\frac{1}{a-x}}n &=&\left\{
	\begin{array}{ll}
		-\S{1}{\frac{1}{a};n},& \textnormal{if } 0< a\leq 1\\
		a^n\left(-\S{1}{\frac{1}{a};n}+\S{1}{\frac{1}{a};\infty}\right),& \textnormal{if } 1<a
		 \end{array} \right.\\
\M{\frac{\H{\ve m}{x}}{a+x}}n &=&-n\;\M{\H{-a,\ve m}{x}}{n-1}+\H{-a,\ve m}{1},\\
\M{\frac{\H{\ve m}{x}}{1-x}}n &=&-n\;\M{\H{1,\ve m}{x}}{n-1},\\
\M{\frac{\H{\ve m}{x}}{a_1-x}}n &=&-n\;\M{\H{a_1,\ve m}{x}}{n-1}+\H{a_1,\ve m}1,\\
\M{\frac{\H{m_1,\overline{\ve{m}}}{x}}{a_2-x}}n &=&\left\{ 
		\begin{array}{ll}
			-\H{m_1,\overline{\ve{m}}}1\S{1}{\frac{1}{a_2};n}-\sum_{i=1}^{n}{\frac{\left(\frac{1}{a_2}\right)^i\Mma{\frac{\Hma{\overline{\ve{m}}}{x}}{\abs{m_1}+x}}i}{i}},& \textnormal{if } m_1< 0\\
			-\H{0,\overline{\ve{m}}}1\S{1}{\frac{1}{a_2};n}+\sum_{i=1}^{n}{\frac{\left(\frac{1}{a_2}\right)^i\Mma{\frac{\Hma{\overline{\ve{m}}}{x}}{x}}i}{i}},&  \textnormal{if } m_1=0  \\
			\sum_{i=1}^{n}{\frac{\left(\frac{1}{a_2}\right)^i\Mma{\frac{\Hma{\overline{\ve{m}}}{x}}{1-x}}i}{i}},& \textnormal{if } m_1 = 1, \\
			-\H{m_1,\overline{\ve{m}}}1\S{1}{\frac{1}{a_2};n}+\sum_{i=1}^{n}{\frac{\left(\frac{1}{a_2}\right)^i\Mma{\frac{\Hma{\overline{\ve{m}}}{x}}{m_1-x}}i}{i}},& \textnormal{if } m_1 > 1,
		 \end{array} \right.
\end{eqnarray*}
where the arising constants on the right hand side are finite.
\label{SSmelweighted}
\end{lemma}
\begin{proof}
For $a>0$ we get
\begin{eqnarray*}
\M{\frac{1}{a+x}}n
	&=&\int_0^1{\frac{x^n}{a+x}dx}=\int_0^1{\frac{x^n-(-a)^n}{a+x}dx}+\int_0^1{\frac{(-a)^n}{a+x}dx}\\
	&=&\int_0^1{(-a)^{n-1}\sum_{i=0}^{n-1}{\frac{x^i}{(-a)^i}dx}}+(-a)^n\H{-a}{1}\\
	&=&(-a)^{n-1}\sum_{i=0}^{n-1}{\frac{1}{(-a)^i(i+1)}}+(-a)^n\H{-a}{1}\\
	&=&(-a)^n\S{1}{-\frac{1}{a};n}+(-a)^n\H{-a}{1}.
\end{eqnarray*}
If $a \geq 1$ this is equal to $(-a)^n\S{1}{-\frac{1}{a};n}-\S{1}{-\frac{1}{a};\infty}.$\\
For $0<a\leq1$ we get
\begin{eqnarray*}
\M{\frac{1}{a-x}}n
	&=&\frac{1}{a^n}\int_0^1{\frac{x^n-a^n}{a-x}dx}=-\frac{1}{a^n}\int_0^1{a^{n-1}\sum_{i=0}^{n-1}\left(\frac{x}{a}\right)^idx}\\
	&=&-\frac{1}{a}\sum_{i=0}^{n-1}\frac{1}{a^i}\int_0^1{x^idx}=-\frac{1}{a}\sum_{i=0}^{n-1}\frac{1}{a^i(1+i)}\\
	&=&-\S{1}{\frac{1}{a};n}.
\end{eqnarray*}
For $a>1$ we get
\begin{eqnarray*}
\M{\frac{1}{a-x}}n
	&=&\int_0^1{\frac{x^n}{a-x}dx}=\int_0^1{\frac{x^n-a^n}{a-x}dx}+a^n\int_0^1{\frac{1}{a-x}dx}\\
	&=&-a^n\S{1}{\frac{1}{a};n}+a^n\H{a}{1}=-a^n\S{1}{\frac{1}{a};n}+a^n\S{1}{\frac{1}{a};\infty}.
\end{eqnarray*}
For $a>0$ we get
\begin{eqnarray*}
\int_0^1 x^n \H{-a,\ve m}{x}dx
		&=&\frac{\H{-a,\ve m}1}{n+1}-\frac{1}{n+1}\int_0^1 \frac{x^{n+1}}{(a+x)}\H{\ve m}{x}dx\\
		&=&\frac{\H{-a,\ve m}1}{n+1}-\frac{1}{n+1}\M{\frac{\H{\ve m}{x}}{a+x}}{n+1}.
\end{eqnarray*}
Hence we get 
\begin{eqnarray*}
\M{\frac{\H{\ve m}{x}}{a+x}}{n+1} &=&-(n+1)\M{\H{-a,\ve m}{x}}{n}+\H{-a,\ve m}{1}.
\end{eqnarray*}
Similarly we get
\begin{eqnarray*}
\int_0^1 x^n \H{1,\ve m}{x}dx
		&=&\frac{1}{n+1}\lim_{\epsilon \rightarrow 1}\left(\epsilon^{n+1}\H{1,\ve m}{\epsilon}-\int_0^{\epsilon}
			{\frac{x^{n+1}-1}{1-x}\H{\ve m}x dx}+\H{1,\ve m}{\epsilon} \right)\\
		&=&\frac{1}{n+1}\int_0^{1}{\frac{x^{n+1}-1}{1-x}\H{\ve m}x dx}\\
		&=&\frac{1}{n+1}\M{\frac{\H{1,\ve m}{x}}{1-x}}{n+1}.
\end{eqnarray*}
And hence
\begin{eqnarray}
\M{\frac{\H{\ve m}{x}}{1-x}}{n+1} &=&-(n+1)\M{\H{1,\ve m}{x}}{n}.
\end{eqnarray}
For  $a_1\in (1,\infty)$ we get
\begin{eqnarray*}
\int_0^1 x^n \H{a_1,\ve m}{x}dx
		&=&\frac{\H{a_1,\ve m}1}{n+1}-\frac{1}{n+1}\int_0^1 \frac{x^{n+1}}{(a_1-x)}\H{\ve m}{x}dx\\
		&=&\frac{\H{a_1,\ve m}1}{n+1}-\frac{1}{n+1}\M{\frac{\H{\ve m}{x}}{a_1-x}}{n+1}.
\end{eqnarray*}
Hence we get 
\begin{eqnarray*}
\M{\frac{\H{\ve m}{x}}{a_1-x}}{n+1} &=&-(n+1)\M{\H{a_1,\ve m}{x}}{n}+\H{a_1,\ve m}{1}.
\end{eqnarray*}
For $a_2 \in (0,1)$ and $m_1<0$ we get
\begin{eqnarray*}
\M{\frac{\H{m_1,\ve m}{x}}{a_2-x}}n &=&\frac{1}{a_2^n}\int_0^1{\frac{(x^n-a_2^n)\H{m_1,\ve m}{x}}{a_2-x}dx}\\
	&=&\left.-\H{m_1,\ve m}{x}\sum_{i=1}^n{\frac{\left(\frac{x}{a_2}\right)^i}{i}}\right|_0^1+\int_0^1{\frac{\H{\ve m}x}{\abs{m_1}+x}\sum_{i=1}^n{\frac{\left(\frac{x}{a_2}\right)^i}{i}}dx}\\
	&=&-\H{m_1,\ve m}1\S{1}{\frac{1}{a_2};n}+\sum_{i=1}^n{\frac{\left(\frac{1}{a_2}\right)^i}{i}}\int_0^1{\frac{x^i\H{\ve m}x}{\abs{m_1}+x}dx}\\
	&=&-\H{m_1,\ve m}1\S{1}{\frac{1}{a_2};n}+\sum_{i=1}^n{\frac{\left(\frac{1}{a_2}\right)^i}{i}}\M{\frac{\H{\ve m}x}{\abs{m_1}+x}}{i},
\end{eqnarray*}
\begin{eqnarray*}
\M{\frac{\H{0,\ve m}{x}}{a_2-x}}n &=&\frac{1}{a_2^n}\int_0^1{\frac{(x^n-a_2^n)\H{0,\ve m}{x}}{a_2-x}dx}\\
	&=&\left.-\H{0,\ve m}{x}\sum_{i=1}^n{\frac{\left(\frac{x}{a_2}\right)^i}{i}}\right|_0^1+\int_0^1{\frac{\H{\ve m}x}{x}\sum_{i=1}^n{\frac{\left(\frac{x}{a_2}\right)^i}{i}}dx}\\
	&=&-\H{0,\ve m}1\S{1}{\frac{1}{a_2};n}+\sum_{i=1}^n{\frac{\left(\frac{1}{a_2}\right)^i}{i}}\int_0^1{\frac{x^i\H{\ve m}x}{x}dx}\\
	&=&-\H{0,\ve m}1\S{1}{\frac{1}{a_2};n}+\sum_{i=1}^n{\frac{\left(\frac{1}{a_2}\right)^i}{i}}\M{\frac{\H{\ve m}x}{x}}{i},
\end{eqnarray*}
and
\begin{eqnarray*}
\M{\frac{\H{1,\ve m}{x}}{a_2-x}}n &=&\frac{1}{a_2^n}\lim_{\epsilon \rightarrow 1} \int_0^{\epsilon}{\frac{(x^n-a_2^n)\H{m_1,\ve m}{x}}{a_2-x}dx}\\
	&=&\lim_{\epsilon \rightarrow 1}\left(\left.-\H{1,\ve m}{x}\sum_{i=1}^n{\frac{\left(\frac{x}{a_2}\right)^i}{i}}\right|_0^\epsilon+\int_0^\epsilon{\frac{\H{\ve m}x}{1-x}\sum_{i=1}^n{\frac{\left(\frac{x}{a_2}\right)^i}{i}}dx}\right)\\
	&=&\lim_{\epsilon \rightarrow 1}\left(-\H{1,\ve m}{\epsilon}\S{1}{\frac{\epsilon}{a_2};n}+\sum_{i=1}^n{\frac{\left(\frac{1}{a_2}\right)^i}{i}}\int_0^\epsilon{\frac{x^i\H{\ve m}x}{1-x}dx}\right)\\
	&=&\lim_{\epsilon \rightarrow 1}\Biggl(-\H{1,\ve m}{\epsilon}\S{1}{\frac{\epsilon}{a_2};n}+\S{1}{\frac{1}{a_2};n}\H{1,\ve m}{\epsilon}\\&&+\sum_{i=1}^n{\frac{\left(\frac{1}{a_2}\right)^i}{i}}\int_0^\epsilon{\frac{(x^i-1)\H{\ve m}x}{1-x}dx}\Biggr)\\
	&=&\sum_{i=1}^n{\frac{\left(\frac{1}{a_2}\right)^i}{i}}\M{\frac{\H{\ve m}x}{1-x}}{i}.
\end{eqnarray*}
For $a_2 \in (0,1)$ and $m_1>0$ we get
\begin{eqnarray*}
\M{\frac{\H{m_1,\ve m}{x}}{a_2-x}}n &=&\frac{1}{a_2^n}\int_0^1{\frac{(x^n-a_2^n)\H{m_1,\ve m}{x}}{a_2-x}dx}\\
	&=&\left.-\H{m_1,\ve m}{x}\sum_{i=1}^n{\frac{\left(\frac{x}{a_2}\right)^i}{i}}\right|_0^1+\int_0^1{\frac{\H{\ve m}x}{m_1-x}\sum_{i=1}^n{\frac{\left(\frac{x}{a_2}\right)^i}{i}}dx}\\
	&=&-\H{m_1,\ve m}1\S{1}{\frac{1}{a_2};n}+\sum_{i=1}^n{\frac{\left(\frac{1}{a_2}\right)^i}{i}}\int_0^1{\frac{x^i\H{\ve m}x}{m_1-x}dx}\\
	&=&-\H{m_1,\ve m}1\S{1}{\frac{1}{a_2};n}+\sum_{i=1}^n{\frac{\left(\frac{1}{a_2}\right)^i}{i}}\M{\frac{\H{\ve m}x}{m_1-x}}{i}.
\end{eqnarray*}
\end{proof}

Due to Lemma \ref{SSmelweighted} we are able to reduce the calculation of the Mellin transform of multiple polylogarithms with indices in 
$\R\setminus(0,1)$ weighted by $1/(a + x)$ or $1/(a-x)$ to the calculation of the Mellin transform of multiple polylogarithms with indices 
in $\R\setminus(0,1)$ which are not weighted, \ie to the calculation of expressions of the form $\M{\H{\ve m}{x}}n.$ To calculate $\M{\H{\ve m}{x}}n$ 
we proceed by recursion. First let us state the base cases, \ie the Mellin transforms of multiple polylogarithms with depth 1.
\begin{lemma}[Compare Lemma \ref{HSweight1mel}] For $n\in \N$, $a \in \R \setminus (0,1)$  we have
\begin{eqnarray*}
\M{\H{a}{x}}n=\left\{ 
		 \begin{array}{ll}
			\frac{-1}{(n+1)^2}\left(1+a^{n+1}(n+1)\S{1}{\frac{1}{a};n}\right.\\
				\hspace{1.2cm}\left.-(a^{n+1}-1)(n+1)\S{1}{\frac{1}{a};\infty}  \right),& \textnormal{if } a \leq -1
\vspace{0.3cm}\\
			\frac{-1}{(n+1)^2}\left(1+a^{n+1}(n+1)\S{1}{\frac{1}{a};n}\right.\\
				\hspace{1.2cm}\left.+(a^{n+1}-1)(n+1)\H{a}1\right),&  \textnormal{if } -1 < a < 0
\vspace{0.3cm}\\
			\frac{-1}{(n+1)^2},&  \textnormal{if } a=0
\vspace{0.3cm}\\
			\frac{1}{(n+1)^2}\left(1+(n+1)\S{1}{n}\right),&  \textnormal{if } a=1
\vspace{0.3cm}\\
			\frac{1}{(n+1)^2}\left(1+a^{n+1}(n+1)\S{1}{\frac{1}{a};n}\right.\\
				\hspace{1.2cm}\left.-(a^{n+1}-1)(n+1)\S{1}{\frac{1}{a};\infty}  \right),& \textnormal{if } a > 1 
		 \end{array} \right.
\end{eqnarray*}
where the arising constants are finite.
\label{SSweight1mel}
\end{lemma}
\begin{proof}
First let $a\leq-1.$ By integration by parts we get:
\begin{eqnarray*}
\M{\H{a}{x}}n&=&\int_0^1{x^n\H{a}x dx}=\left.\frac{x^{n+1}}{n+1}\H{a}{x}\right|_0^1-\int_0^1{\frac{x^{n+1}}{n+1}\frac{1}{\abs{a}+x}dx}\\
&=&\frac{\H{a}1}{n+1}-\frac{1}{n+1}\left(\int_0^1{\frac{x^{n+1}-a^{n+1}}{\abs{a}+x}dx}+a^{n+1}\int_0^1{\frac{1}{\abs{a}+x}dx}\right)\\
&=&\frac{\H{a}1}{n+1}(1-a^{n+1})-\frac{1}{n+1}\int_0^1{a^n\sum_{i=0}^n\frac{x^{i}}{a^i}dx}\\
&=&\frac{\H{a}1}{n+1}(1-a^{n+1})-\frac{a^n}{n+1}\sum_{i=0}^n\frac{1}{a^i(i+1)}\\
&=&\frac{-\S{1}{\frac{1}{a};\infty}}{n+1}(1-a^{n+1})-\frac{a^{n+1}}{n+1}\S{1}{\frac{1}{a};n+1}\\
&=&\frac{-1}{n+1}\left(\frac{1}{n+1}+a^{n+1}\S{1}{\frac{1}{a};n}-(a^{n+1}-1)\S{1}{\frac{1}{a};\infty}\right).
\end{eqnarray*}
For $-1<a<0$ we obtain:
\begin{eqnarray*}
\M{\H{a}{x}}n&=&\int_0^1{x^n\H{a}x dx}=\left.\frac{x^{n+1}}{n+1}\H{a}{x}\right|_0^1-\int_0^1{\frac{x^{n+1}}{n+1}\frac{1}{\abs{a}+x}dx}\\
&=&\frac{\H{a}1}{n+1}-\frac{1}{n+1}\left(\int_0^1{\frac{x^{n+1}-a^{n+1}}{\abs{a}+x}dx}+a^{n+1}\int_0^1{\frac{1}{\abs{a}+x}dx}\right)\\
&=&\frac{\H{a}1}{n+1}(1-a^{n+1})-\frac{1}{n+1}\int_0^1{a^n\sum_{i=0}^n\frac{x^{i}}{a^i}dx}\\
&=&\frac{\H{a}1}{n+1}(1-a^{n+1})-\frac{a^n}{n+1}\sum_{i=0}^n\frac{1}{a^i(i+1)}\\
&=&\frac{\H{a}1}{n+1}(1-a^{n+1})-\frac{a^{n+1}}{n+1}\S{1}{\frac{1}{a};n+1}\\
&=&\frac{-1}{(n+1)^2}\left(1+a^{n+1}(n+1)\S{1}{\frac{1}{a};n}+(a^{n+1}-1)(n+1)\H{a}1\right).
\end{eqnarray*}
For $a=0$ it follows that
\begin{eqnarray}
\M{\H{0}{x}}n&=&\int_0^1{x^n\H{0}x dx}=\left.\frac{x^{n+1}}{n+1}\H{0}{x}\right|_0^1-\int_0^1{\frac{x^{n+1}}{n+1}\frac{1}{x}dx}\nonumber\\
&=&-\int_0^1{\frac{x^n}{n+1}dx}=\left.-\frac{x^{n+1}}{(n+1)^2}\right|_0^1=-\frac{1}{(n+1)^2}.\nonumber
\end{eqnarray}
For $a=1$ we get:
\begin{eqnarray}
\int_0^1 x^n \H{1}{x}dx 
		&=&\lim_{\epsilon \rightarrow 1} \int_0^{\epsilon} x^n \H{1}{x}dx \nonumber\\
		&=&\lim_{\epsilon \rightarrow 1}\left( \left. \frac{x^{n+1}}{n+1}\H{1}{x}\right|_0^{\epsilon}-\int_0^{\epsilon}
			\frac{x^{n+1}}{(n+1)(1-x)}dx \right) \nonumber\\
		&=&\frac{1}{n+1}\lim_{\epsilon \rightarrow 1}\left(\epsilon^{n+1}\H{1}{\epsilon}-\int_0^{\epsilon}
			{\frac{x^{n+1}-1}{1-x}x dx}-\H{1}{\epsilon} \right) \nonumber\\
		&=&\frac{1}{n+1}\left(\lim_{\epsilon \rightarrow 1}(\epsilon^{n+1}-1)\H{1}{\epsilon}+\lim_{\epsilon \rightarrow 1}
			{\int_0^{\epsilon}\sum_{i=0}^n{x^i}dx}\right)\nonumber\\
		&=&\frac{1}{n+1}\left(0+\int_0^{1}\sum_{i=0}^n{x^i}dx\right)\nonumber\\
		&=&\frac{1}{n+1} \sum_{i=0}^n{\frac{1}{i+1}}=\frac{1}{(n+1)^2}\left(1+(n+1)\S{1}{n}\right). \nonumber
\end{eqnarray}
Finally for $a > 1$ we conclude that
\begin{eqnarray*}
\M{\H{a}{x}}n&=&\int_0^1{x^n\H{a}x dx}=\left.\frac{x^{n+1}}{n+1}\H{a}{x}\right|_0^1-\int_0^1{\frac{x^{n+1}}{n+1}\frac{1}{a-x}dx}\\
&=&\frac{\H{a}1}{n+1}-\frac{1}{n+1}\left(\int_0^1{\frac{x^{n+1}-a^{n+1}}{a-x}dx}+a^{n+1}\int_0^1{\frac{1}{a-x}dx}\right)\\
&=&\frac{\H{a}1}{n+1}(1-a^{n+1})+\frac{1}{n+1}\int_0^1{a^n\sum_{i=0}^n\frac{x^{i}}{a^i}dx}\\
&=&\frac{\H{a}1}{n+1}(1-a^{n+1})+\frac{a^n}{n+1}\sum_{i=0}^n\frac{1}{a^i(i+1)}\\
&=&\frac{\S{1}{\frac{1}{a};\infty}}{n+1}(1-a^{n+1})+\frac{a^{n+1}}{n+1}\S{1}{\frac{1}{a};n+1}\\
&=&\frac{1}{n+1}\left(\frac{1}{n+1}+a^{n+1}\S{1}{\frac{1}{a};n}-(a^{n+1}-1)\S{1}{\frac{1}{a};\infty}  \right).
\end{eqnarray*}
\end{proof}

The higher depth results for $\M{\H{\ve m}{x}}n$ can now be obtained by recursion:
\begin{lemma}[Compare Lemma \ref{HSmelnotweighted}] For $n\in \N$, $a \in \R\setminus(0,1)$ and $\ve m \in (\R\setminus(0,1))^k$,
\begin{eqnarray*}
\M{\H{a,\ve m}{x}}n=\left\{ 
		\begin{array}{ll}
			\frac{(1-a^{n+1})\H{a,\ve m}1}{n+1}-\frac{a^{n}}{n+1}\sum_{i=0}^{n}{\frac{\M{\H{\ve m}{x}}i}{a^i}},& \textnormal{if } a < 0\\
			\frac{\H{0,\ve m}1}{n+1}-\frac{1}{n+1} \M{\H{\ve m}{x}}n,&  \textnormal{if } a=0  \\
			\frac{1}{n+1} \sum_{i=0}^n{\M{\H{\ve m}{x}}n},&  \textnormal{if } a=1  \\
			\frac{(1-a^{n+1})\H{a,\ve m}1}{n+1}+\frac{a^{n}}{n+1}\sum_{i=0}^{n}{\frac{\M{\H{\ve m}{x}}i}{a^i}},& \textnormal{if } a > 1 
		 \end{array} \right.
\end{eqnarray*}
where the arising constants are finite.
\label{SSmelnotweighted}
\end{lemma}
\begin{proof}
We get the following results by integration by parts. For $a<0$ we get:
\begin{eqnarray*}
&&\M{\H{a,\ve m}{x}}n
=\int_0^1{x^n\H{a,\ve m}x dx}=\left.\frac{x^{n+1}}{n+1}\H{a,\ve m}{x}\right|_0^1-\int_0^1{\frac{x^{n+1}\H{\ve m}x}{n+1}\frac{1}{\abs{a}+x}dx}\\
&&\hspace{2cm}=\frac{\H{a,\ve m}1}{n+1}-\frac{1}{n+1}\left(\int_0^1{\frac{x^{n+1}-a^{n+1}}{\abs{a}+x}\H{\ve m}x dx}+a^{n+1}\int_0^1{\frac{\H{\ve m}x}{\abs{a}+x}dx}\right)\\
&&\hspace{2cm}=\frac{\H{a,\ve m}1}{n+1}(1-a^{n+1})-\frac{1}{n+1}\int_0^1{a^n\sum_{i=0}^n\frac{x^{i}\H{\ve m}x}{a^i}dx}\\
&&\hspace{2cm}=\frac{\H{a,\ve m}1}{n+1}(1-a^{n+1})-\frac{a^n}{n+1}\sum_{i=0}^n\frac{1}{a^i}\M{\H{\ve m}x}i.
\end{eqnarray*}
For $a=0$ we get:
\begin{eqnarray*}
\int_0^1 x^n \H{0,\ve m}{x}dx &=& \left. \frac{x^{n+1}}{n+1}\H{0,\ve m}{x}\right|_0^1-\int_0^1 \frac{x^n}{n+1}\H{\ve m}{x}dx\\
				&=&\frac{\H{0,\ve m}1}{n+1}-\frac{1}{n+1} \M{\H{\ve m}{x}}n.
\end{eqnarray*}
For $a=1$ we it follows that
\begin{eqnarray}
\int_0^1 x^n \H{1,\ve m}{x}dx 
		&=&\lim_{\epsilon \rightarrow 1} \int_0^{\epsilon} x^n \H{1,\ve m}{x}dx \nonumber\\
		&=&\lim_{\epsilon \rightarrow 1}\left( \left. \frac{x^{n+1}}{n+1}\H{1,\ve m}{x}\right|_0^{\epsilon}-\int_0^{\epsilon}
			\frac{x^{n+1}}{(n+1)(1-x)}\H{\ve m}{x}dx \right) \nonumber\\
		&=&\frac{1}{n+1}\lim_{\epsilon \rightarrow 1}\left(\epsilon^{n+1}\H{1,\ve m}{\epsilon}-\int_0^{\epsilon}
			{\frac{x^{n+1}-1}{1-x}\H{\ve m}x dx}-\H{1,\ve m}{\epsilon} \right) \nonumber\\
		&=&\frac{1}{n+1}\left(\lim_{\epsilon \rightarrow 1}(\epsilon^{n+1}-1)\H{1,\ve m}{\epsilon}+\lim_{\epsilon \rightarrow 1}
			{\int_0^{\epsilon}\sum_{i=0}^n{x^i\H{\ve m}x}dx}\right)\nonumber\\
		&=&\frac{1}{n+1}\left(0+\int_0^{1}\sum_{i=0}^n{x^i\H{\ve m}x}dx\right)\nonumber\\
		&=&\frac{1}{n+1} \sum_{i=0}^n{\M{\H{\ve m}{x}}i}. \nonumber
\end{eqnarray}
For $a>1$ we conclude that
\begin{eqnarray*}
&&\M{\H{a,\ve m}{x}}n
=\int_0^1{x^n\H{a,\ve m}x dx}=\left.\frac{x^{n+1}}{n+1}\H{a,\ve m}{x}\right|_0^1-\int_0^1{\frac{x^{n+1}\H{\ve m}x}{n+1}\frac{1}{a-x}dx}\\
&&\hspace{2cm}=\frac{\H{a,\ve m}1}{n+1}-\frac{1}{n+1}\left(\int_0^1{\frac{x^{n+1}-a^{n+1}}{a-x}\H{\ve m}x dx}+a^{n+1}\int_0^1{\frac{\H{\ve m}x}{a-x}dx}\right)\\
&&\hspace{2cm}=\frac{\H{a,\ve m}1}{n+1}(1-a^{n+1})+\frac{1}{n+1}\int_0^1{a^n\sum_{i=0}^n\frac{x^{i}\H{\ve m}x}{a^i}dx}\\
&&\hspace{2cm}=\frac{\H{a,\ve m}1}{n+1}(1-a^{n+1})+\frac{a^n}{n+1}\sum_{i=0}^n\frac{1}{a^i}\M{\H{\ve m}x}i.
\end{eqnarray*}
\end{proof}

Using Lemma~\ref{SSmelweighted} together with Lemma~\ref{SSweight1mel} and Lemma~\ref{SSmelnotweighted} we are able to calculate the Mellin transform of 
multiple polylogarithms with indices in $\R \setminus (0,1).$ In addition, the polylogarithms can be weighted by $1/(a + x)$ or $1/(a-x)$ for $a\in\R.$
These Mellin transforms can be expressed using S-sums.
If we restrict the indices of the multiple polylogarithms to $\R \setminus ((-1,0)\cup (0,1)),$ it turns out that these Mellin transforms can be expressed using a subset 
of the S-sums. This subset consists of S-sums 
$\S{a_1,a_2,\ldots,a_l}{b_1,b_2,\ldots,b_k;n}$ with $(a_1,\ldots,a_k)\in \Z^k,$ $(b_2,\ldots,b_k)\in ([-1,1]\setminus \{0\})^{k-1}$ and $b_1\in \R\setminus \{0\}.$ 
We will call these sums $\bar{S}$-sums, see the following definition.

\begin{definition}[$\bar{S}$-sums and $\bar{H}$-multiple polylogarithms]
We define the set
\begin{eqnarray*}
&&\bar{S}:=\biggl\{\S{a_1,a_2,\ldots,a_k}{b_1,b_2,\ldots,b_k;n}\biggl|a_i\in \Z^* \textnormal{ for } 1\leq i\leq k ;b_1\in \R^*;\\
  &&\hspace{5.8cm} b_i \in [-1,1]\setminus \{0\}\textnormal{ for } 2\leq i\leq k\biggr\}
\end{eqnarray*}
and call the elements $\bar{S}$-sums and the set
$$
\bar{H}:=\left\{\H{m_1,m_2,\ldots,m_k}x\left| m_i \in \R \setminus ((-1,0)\cup (0,1)) \textnormal{ for } 1\leq i\leq k  \right.\right\}
$$
and call the elements $\bar{H}$-multiple polylogarithms.
\label{SSsubset}
\end{definition}

\begin{remark}
 Restricting to this class, we are able to construct the inverse Mellin transform.
\end{remark}

\subsection{The Inverse Mellin Transform of \texorpdfstring{$\bar{S}$}{S-bar}-sums}
\label{SSInvMellin}
As in the cases of harmonic sums, we can define an order on S-sums.
\begin{definition}[Order on S-sums]
Let $\S{\ve m_1}{\ve b_1; n}$ and $\S{\ve m_2}{\ve b_2;n}$ be S-sums with weights $w_1$, $w_2$ and depths $d_1$ and $d_2,$ respectively. Then
$$
		  	\begin{array}{ll}
						\S{\ve m_1}{\ve b_1; n} \prec \S{\ve m_2}{\ve b_2;n}, \ \textnormal{if} \ w_1<w_2, & \textnormal{or } \ (w_1=w_2 \ \textnormal{and} \ d_1<d_2).
			\end{array}
$$
We say that an S-sum $s_1$ is more complicated than an S-sum $s_2$ if $s_2 \prec s_1$.
For a set of S-sums we call an S-sum {\upshape most complicated} if it is a largest element with respect to $\prec$. 
\label{SSsord} 
\end{definition}

The following proposition guarantees that there is only one {\itshape most complicated} S-sum in the the Mellin transform of 
an $\bar{H}$-multiple polylogarithm.
\begin{prop}[Compare \cite{Ablinger2009}]
In the Mellin transform of a multiple polylogarithm $\H{\ve m}x \in \bar{H}$ weighted by $1/(c-x)$ or 
$1/(c+x)$ ($c\in\R^*$) there is only one {\upshape most complicated} 
S-sum $\S{\ve a}{\ve b; n}\in \bar{S},\ie$
\begin{eqnarray}\label{SSsinglmostcomp1}
\M{\frac{\H{\ve m}{x}}{c\pm x}}{n}=p\cdot\S{\ve a}{\ve b; n}+t
\end{eqnarray}
where $p\in\R^*$ and all S-sums in $t$ occur linearly and are smaller then $\S{\ve a}{\ve b; n}$.
\label{SSsinglmostcomp}
\end{prop}
The {\itshape most complicated} S-sum in the Mellin transform of a $\bar{H}$-multiple polylogarithm can be calculated using Algorithm \ref{SSmostcompsum}. Since all S-sums that occur in the 
Mellin transform are $\bar{S}$-sums, it is in fact a $\bar{S}$-sum. But even the reverse direction is possible: given an $\bar{S}$-sum, we can find an 
$\bar{H}$-multiple polylogarithm weighted by $1/(c-x)$ or $1/(c+x)$ such that the $\bar{S}$-sum is the {\itshape most complicated} S-sum in the Mellin transform of this weighted 
multiple polylogarithm. We will call this weighted multiple polylogarithm the {\itshape most complicated} weighted multiple polylogarithm in the inverse Mellin transform of the $\bar{S}$-sum. We can 
use Algorithm \ref{SSmostcomphlog} to compute the most-complicated multiple polylogarithm in the inverse Mellin transform of a
$\bar{S}$-sum. The proof of Proposition \ref{SSsinglmostcomp}, in particular the correctness of Algorithms \ref{SSmostcompsum} and \ref{SSmostcomphlog} are similar to the proof given in \cite{Ablinger2009} and is thus skipped here.

The computation of the inverse Mellin transform of a $\bar{S}$-sum $\S{a_1,a_2,\ldots,a_l}{b_1,b_2,\ldots,b_l;n}$ 
now is straightforward just as in Section \ref{HSInvMellin} (compare\cite{Ablinger2009,Remiddi2000}):
\begin{itemize}
	\item Locate the most complicated $\bar{S}$-sum.
	\item Construct the corresponding most complicated multiple polylogarithm.
	\item Add it and subtract it.
	\item Perform the Mellin transform to the subtracted version. This will cancel the original $\bar{S}$-sum.
	\item Repeat the above steps until there are no more $\bar{S}$-sums.
	\item Let $c$ be the remaining constant term; replace it by $\textnormal{M}^{-1}(c)$, or equivalently, multiply $c$ by $\delta(1-x)$.
\end{itemize}

\begin{algorithm}
\caption{The Most-Complicated S-Sum}
\label{SSmostcompsum}
\begin{algorithmic}
\State \bfseries input:\ \normalfont $\frac{\H{m_1,\ldots,m_k}x}{d+s*x},$ with $d\in \R^+$, $s\in\{-1,1\}$ and $\H{m_1,\ldots,m_k}x\in \bar{H}(x)$
\State \bfseries output:\ \normalfont $ p\cdot \S{\ve a}{\ve b;n}$ as in (\ref{SSsinglmostcomp1}) 
\Procedure{MostCompS}{$\frac{\H{m_1,\ldots,m_k}x}{d+s*x}$}
\State $(a_1,\ldots,a_k,a_{k+1})=(\sign{m_1},\ldots,\sign{m_k},1)$
\For{$i=k$ to $1$}
	\If{$a_i \neq 0$}
		\State $j=i-1$
		\While{$j>0$ and $a_j=0$}
			\State $j=j-1$
		\EndWhile
		\If{$j>0$ and $a_j<0$}
			\State $a_i=-a_i$
		\EndIf
		\If{$j=0$ and $s=1$}
			\State $a_i=-a_i$
		\EndIf
	\EndIf
\EndFor
\State delete the zeros in $(\abs{m_1},\ldots,\abs{m_k},1);$ let $(b_1,\ldots,b_{\bar{k}})$ be the result
\For{$i=\bar{k}$ to $2$}
	\State $b_i=\frac{b_i}{b_{i-1}}$
\EndFor
\State $b_1=\frac{b_1}{d}$
\State $l=$ number of entries $\leq 0$ in $(m_1,\ldots,m_k)$
\State let $\S{\bar{a}_1,\ldots,\bar{a}_{\bar{k}}}{n}$ be the sum that is denoted by the word $a_1 \cdots a_k$ (see Section \ref{HSalgrel})
\State $(b_1,\ldots,b_{\bar{k}})=(\sign{\bar{a}_1}b_1,\ldots,\sign{\bar{a}_{\bar{k}}}b_{\bar{k}})$
\State $(\bar{a}_1,\ldots,\bar{a}_{\bar{k}})=(\abs{\bar{a}_1},\ldots,\abs{\bar{a}_{\bar{k}}})$
\State \textbf{return} $s\cdot(-1)^l\cdot\S{\bar{a}_1,\ldots,\bar{a}_{\bar{k}}}{b_1,\ldots,b_{\bar{k}};n}$
\EndProcedure
\end{algorithmic}
\end{algorithm}

\begin{algorithm}
\caption{The Most-Complicated Harmonic Polylogarithm}
\label{SSmostcomphlog}
\begin{algorithmic}
\State \bfseries input:\ \normalfont $\S{a_1,\ldots,a_k}{b_1,\ldots,b_k;n},$ with $a_i \in \Z$ and $b_i \in \R^*$
\State \bfseries output:\ \normalfont the most-complicated multiple polylogarithm in the inverse Mellin transform of $\S{a_1,\ldots,a_k}{b_1,\ldots,b_k;n}$ together with a possibly
arising factor $(-1)^n$
\Procedure{MostCompH}{$\S{a_1,\ldots,a_k}{b_1,\ldots,b_k;n}$}
\State $(a_1,\ldots,a_k)=(\sign{b_1}a_1,\ldots,\sign{b_k}a_k)$
\State $l=$ number of entries $< 0$ in $(a_1,\ldots,a_k)$
\State $v=((-1)^l)^n$
\State $s=-(-1)^l$
\State let $m_1\cdots m_{\bar{k}}$ be the notation in the alphabet $\{-1,0,1\}$ of $\S{a_1,\ldots,a_k}{n}$ 
\State (see Section \ref{HSalgrel})
\State $j=1$
\If{$s=1$ and $\bar{k}\neq 1$}
	\While{$m_j=0$ and $j<\bar{k}-1$}
		\State $j=j+1$
	\EndWhile
	\State $m_j=-m_j$
\EndIf
\If{$\bar{k}\neq 1$}
	\State $d=\sign{m_j}$
\EndIf
\For{$i=j+1$ to $\bar{k}-1$}
	\If{$m_i\neq0$}
		\If{$d=-1$}
			\State $m_i=-m_i$
		\EndIf
		\State $d=\sign{m_i}$
	\EndIf
\EndFor
\State $l=$ number of entries $\leq 0$ in $(m_1,\ldots,m_{\bar{k}-1})$
\State $v=v\cdot s\cdot(-1)^l$
\State $b_k=\frac{1}{b_k}$
\For{$i=k-1$ to $1$}
	\State $b_i=\frac{b_{i+1}}{b_{i}}$
\EndFor
\State $d=b_1$
\State $j=1$
\For{$i=1$ to $\bar{k}-1$}
		\If{$m_i\neq 0$}
			\State $m_i=b_jm_i$
			\State $j=j+1$
		\EndIf
	\State $b_i=\frac{b_{i+1}}{b_{i}}$
\EndFor
\If{($d>1$ and $\bar{k}\neq 1$) or $s\neq -1$}
	\State $v=\frac{v}{d^n}$
\EndIf
\State \textbf{return} $v\cdot\frac{\H{m_1,\ldots,m_{\bar{k}-1}}x}{d+sx}$
\EndProcedure
\end{algorithmic}
\end{algorithm}

\begin{remark}
Again there is a second way to compute the inverse Mellin transform of an $\bar{S}$-sum: We can use the integral representation of Theorem \ref{SSintrep}. After repeated 
suitable integration by parts we can find the inverse Mellin transform.
\end{remark}

\section{Differentiation of S-Sums}
\label{SSdifferentiation}
We present two strategies to compute the differentiation of S-sums. The first approach follows the ideas carried out in Section \ref{HSdifferentiation} for harmonic sums. Since we rely here on the inverse 
Mellin transform, this method is restricted to $\bar{S}$-sums using our current technologies. Finally, we will develope a second approach that will work for S-sums in general.

\subsection{First Approach: Differentiation of \texorpdfstring{$\bar{S}$}{S-bar}-sums}
We are able to calculate the inverse Mellin transform of $\bar{S}$-sums $\S{a_1,\ldots,a_l}{b_1,\ldots,b_l;n}$. It turned out that they are linear combinations 
of multiple polylogarithms weighted by the factors $1/(a\pm x)$. By computing the Mellin transform of a $\bar{S}$-sum we find in fact an analytic continuation 
of the sum to $n\in\C$ up to countable many isolated points. As for harmonic sums (see Section \ref{HSdifferentiation}) this allows us to consider differentiation with respect to $n$, since we can differentiate the analytic 
continuation. Afterwards we may transform back to $\bar{S}$-sums applying the Mellin transformation.

Again differentiation turns out to be relatively easy if we represent the $\bar{S}$-sums using its inverse Mellin transform.
If we want to differentiate $\S{\ve a}{\ve b;n}$ with respect to $n$ we can proceed as follows:
\begin{itemize}
	\item Calculate the inverse Mellin transform of $\S{\ve a}{\ve b;n}.$
	\item Set the constants to zero and multiply the remaining terms of the inverse Mellin transform by $\H0x$. This is in fact differentiation with respect to $n$. 
	\item Calculate the Mellin transform of the multiplied inverse Mellin transform of $\S{\ve a}{\ve b;n}$ 
\end{itemize}

Looking in detail at this procedure, we see that first we express the $\bar{S}$-sum as a linear combination of expressions of the form $\M{\frac{\H{\ve m}x}{k\pm x}}n$ 
with $\H{\ve m}x \in \bar{H}(x).$ Next we have to differentiate these summands. For ``$+$'' and $k\in\R^+$ this leads to:
\begin{eqnarray*}
&&\frac{d}{d n}\M{\frac{\H{\ve m}x}{k+x}}n=\frac{d}{d n}\int_0^1\frac{x^n\H{\ve m}x}{k+x}dx=\int_0^1\frac{x^n\H{0}x\H{\ve m}x}{k+x}dx\\
&&\hspace{1cm}=\int_0^1\frac{x^n(\H{0,m_1,\ldots,m_l}x+\H{m_1,0,m_2\ldots,m_l}x+\cdots+\H{m_1,\ldots,m_l,0}x)}{k+x}dx\\
&&\hspace{1cm}=\M{\frac{\H{0,m_1,\ldots,m_l}x}{k+x}}n+\cdots+ \M{\frac{\H{m_1,\ldots,m_l,0}x}{k+x}}n.
\end{eqnarray*}
For ``$-$'' and $k\in (1,\infty)$ this leads to:
\begin{eqnarray*}
&&\frac{d}{d n}\M{\frac{\H{\ve m}x}{k-x}}n=\frac{d}{d n}\int_0^1\frac{x^n\H{\ve m}x}{k-x}dx=\int_0^1\frac{x^n\H{0}x\H{\ve m}x}{k-x}dx\\
&&\hspace{1cm}=\int_0^1\frac{x^n\H{0,m_1,\ldots,m_l}x+\H{m_1,0,m_2\ldots,m_l}x+\cdots+\H{m_1,\ldots,m_l,0}x}{k-x}dx\\
&&\hspace{1cm}=\M{\frac{\H{0,m_1,\ldots,m_l}x}{k- x}}n+\cdots+ \M{\frac{\H{m_1,\ldots,m_l,0}x}{k-x}}n.
\end{eqnarray*}
For ``$-$'' and $k\in(0,1]$ this leads to:
\begin{eqnarray}
\frac{d}{d n}\M{\frac{\H{\ve m}x}{k-x}}n&=&\frac{d}{d n}\int_0^1\frac{\left(\frac{x^n}{k^n}-1\right)\H{\ve m}x}{k-x}dx=
\int_0^1\frac{\frac{x^n}{k^n}\left(\H{0}x-\H{0}k\right)\H{\ve m}x}{k-x}dx\nonumber\\
&=&\int_0^1\frac{\left(\frac{x^n}{k^n}-1\right)\H{0}x\H{\ve m}x}{k-x}dx-\H{0}k\int_0^1\frac{\left(\frac{x^n}{k^n}-1\right)\H{\ve m}x}{k-x}dx\nonumber\\
&&+\int_0^1\frac{\H{0}{\frac{x}{k}}\H{\ve m}x}{k-x}dx.
\label{SSdiffconst1}
\end{eqnarray}
In the following lemmas we will see that $\int_0^1\frac{\Hma{0}{\frac{x}{k}}\H{\ve m}x}{k-x}dx$ is finite and in other words, the last integral in (\ref{SSdiffconst1}) is well defined. 
We will see how we can evaluate this constant in terms of multiple polylogarithms.
Hence in all three cases we eventually arrive at Mellin transforms of multiple polylogarithms.
\begin{lemma}
For a multiple polylogarithm $\H{\ve m}x$ with $\ve m=(m_1,m_2,\ldots,m_l),$ $m_i~\in~\R \setminus (0,1]$ and $i\in \N$ the integral
$$
\int_0^1\frac{\H{0}x^i\H{\ve m}x}{1-x}dx
$$ 
exists and can be expressed in terms of multiple polylogarithms.
\end{lemma}
\begin{proof}
Using integration by parts we get
\begin{eqnarray*}
\int_0^1\frac{\H{0}x^i\H{\ve m}x}{1-x}dx&:=&\lim_{b \rightarrow 1^-}\lim_{a \rightarrow 0^+}{\int_a^{b}\frac{\H{0}x^i\H{\ve m}x}{1-x}dx}\\
	&=&\lim_{b \rightarrow 1^-}\lim_{a \rightarrow 0^+}\left(\left.\H{0}x^i\H{1,\ve m}x\right|_a^{b}-
		i{\int_a^{b}\frac{\H{0}x^{i-1}\H{1,\ve m}x}{x}dx}\right)\\
	&=&-i{\int_0^{1}\frac{\H{0}x^{i-1}\H{1,\ve m}x}{x}dx}.
\end{eqnarray*}
After expanding the product $\H{0}x^{i-1}\H{1,\ve m}x$ and applying the integral we end up in multiple polylogarithms at one with leading zero. Hence the integral is finite.
\end{proof}

\begin{lemma}
For $k\in(0,1)$ and $\ve a=(a_1,a_2,\ldots,a_l),$ $m_i \in\R \setminus (0,1]$ we have that
\begin{eqnarray*}
\int_0^1\frac{\H{0}{\frac{x}{k}}\H{\ve a}x}{k-x}dx
\end{eqnarray*}
is finite and can be expressed in terms of multiple polylogarithms.
\end{lemma}
\begin{proof}
If $\H{\ve a}x$ has trailing zeroes, \ie $a_l=0,$ we first extract them and get a linear combination of expressions of the form $$\int_0^1\frac{\H{0}{\frac{x}{k}}\H{0}x^w\H{\ve m}x}{k-x}dx$$
with $w\in \N$. Now let us look at such an expression:
\begin{eqnarray*}
\int_0^1\frac{\H{0}{\frac{x}{k}}\H{0}x^w\H{\ve m}x}{k-x}dx&=&\int_0^{\frac{1}{k}}\frac{\H{0}{x}\H{0}{kx}^w\H{\ve m}{kx}}{k-kx}kdx\\
&=&\int_0^{\frac{1}{k}}\frac{\H{0}{x}(\H{0}{x}+\H{0}{k})^w\H{\frac{m_1}{k},\ldots,\frac{m_l}{k}}{x}}{1-x}dx\\ 
&=&\int_0^{\frac{1}{k}}\frac{\H{\frac{m_1}{k},\ldots,\frac{m_l}{k}}{x}\sum_{i=0}^w{\binom{w}{i}\H{0}{x}^{i+1}\H{0}{k}^{w-i}}} {1-x}dx\\
&=&\sum_{i=0}^w{\binom{w}{i}\H{0}{k}^{w-i}\underbrace{\int_0^{\frac{1}{k}}\frac{\H{\frac{m_1}{k},\ldots,\frac{m_l}{k}}{x}\H{0}{x}^{i+1}}{1-x}dx}_A}
\end{eqnarray*}
\begin{eqnarray*}
A&=&\int_0^{1}\frac{\H{\frac{m_1}{k},\ldots,\frac{m_l}{k}}{x}\H{0}{x}^{i+1}}{1-x}dx+\int_1^{\frac{1}{k}}\frac{\H{\frac{m_1}{k},\ldots,\frac{m_l}{k}}{x}\H{0}{x}^{i+1}}{1-x}dx\\
&=&\underbrace{\int_0^1\frac{\H{\frac{m_1}{k},\ldots,\frac{m_l}{k}}{x}\H{0}{x}^{i+1}}{x}dx}_B-\underbrace{\int_0^{\frac{1}{k}-1}\frac{\H{\frac{m_1}{k},\ldots,\frac{m_l}{k}}{x+1}\H{0}{x+1}^{i+1}}{x}dx}_C.
\end{eqnarray*}
The integral $B$ is finite and expressible in terms of multiple polylogarithms due to the previous lemma.
After expanding the product $\H{\frac{m_1}{k},\ldots,\frac{m_l}{k}}{x+1}\H{0}{x+1}^{i+1}$, applying Lemma \ref{SSeinplustrafo} and applying the integral we end up in multiple polylogarithms at one and at $\frac{1}k-1$ which are all finite. Hence C is finite.
Summarizing $$\int_0^1\frac{\H{0}{\frac{x}{k}}\H{\ve a}x}{k-x}dx$$ is finite, since we can write it as a sum of finite integrals. 
\end{proof}

\begin{example}
Let us differentiate $\S{2}{2;n}.$ We have:
\begin{eqnarray*}
\S{2}{2;n}&=&\M{\frac{\H{0}x}{\frac{1}{2}-x}}n=\int_0^1{\frac{(2^n x^n-1)\H{0}x}{\frac{1}{2}-x}dx}.
\end{eqnarray*}
Differentiating with respect to $n$ yields:
\begin{eqnarray*}
&&\frac{d}{d n}\int_0^1{\frac{(2^n x^n-1)\H{0}x}{\frac{1}{2}-x}dx}=\int_0^1\frac{2^nx^n\left(\H{0}x+\H{0}2\right)\H{0}x}{\frac{1}{2}-x}dx\\
&&\hspace{1cm}=2 \int_0^1\frac{\left(2^nx^n-1\right)\H{0,0}x}{\frac{1}{2}-x}dx+\H{0}2\int_0^1\frac{\left(2^nx^n-1\right)\H{0}x}{\frac{1}{2}-x}dx\\
&&\hspace{1.5cm}+\int_0^1\frac{\H{0}{2\,x}\H{0}x}{\frac{1}{2}-x}dx\\
&&\hspace{1cm}=2 \M{\frac{\H{0,0}x}{\frac{1}{2}-x}}n+\H{0}2\M{\frac{\H{0}x}{\frac{1}{2}-x}}n+\int_0^1\frac{\H{0}{2\, x}\H{0}x}{\frac{1}{2}-x}dx\\
&&\hspace{1cm}=-2\S{3}{2;n}+\H{0}2\S{2}{2;n}+\H{0,0,-1}1+2\H{0,0,1}1+\H{0,1,-1}1.
\end{eqnarray*}
\end{example}

\subsection{Second Approach: Differentiation of S-sums}
Using the first approach we were able to compute the derivative for $\bar{S}$-sums. However we are also interested in the derivative of S-sums which are 
not in this subset. For those sums we failed in the first approach since we were not able to compute their inverse Mellin transform. Therefore we look at a 
different approach, which we will first illustrate using the following detailed example.

\begin{example}[Differentiation of $\SS{1,2,1}{\frac{1}{2},2,\frac{1}{3}}n$]
Using Theorem \ref{SSintrep} leads to
\small
\begin{eqnarray*}
\SS{1,2,1}{\frac{1}{2},2,\frac{1}{3}}n&=&\int_0^{\frac{1}{3}}{\frac{1}{x-1}\int_{1}^{x}{\frac{1}{y}\int_0^{y}{\frac{1}{z-\frac{1}{2}}\int_{\frac{1}{2}}^{z}{\frac{w^{n}-1}{w-1}}dwdz}dy}dx}\\
&=&\int_0^{\frac{1}{3}}{\frac{1}{x-1}\int^{1}_{x}{-\frac{1}{y}\int_0^{y}{\frac{1}{\frac{1}{2}-z}\int_{\frac{1}{2}}^{z}{\frac{w^{n}-1}{w-1}}dwdz}dy}dx}.
\end{eqnarray*}
\normalsize
Differentiation of the integral representation with respect to $n$ yields
\small
\begin{eqnarray*}
\int_0^{\frac{1}{3}}{\frac{1}{x-1}\int^{1}_{x}{-\frac{1}{y}\int_0^{y}{\frac{1}{\frac{1}{2}-z}\int_{\frac{1}{2}}^{z}{\frac{w^{n}\H0w}{w-1}}dwdz}dy}dx}=:A.
\end{eqnarray*}
\normalsize
We want to rewrite $A$ in terms of S-sums at $n$ and finite S-sums at $\infty$ (or finite multiple polylogarithms). Therefore we first rewrite $A$ in the following form:
\small
\begin{eqnarray*}
A&=&\underbrace{\int_0^{\frac{1}{3}}{\frac{1}{x-1}\int^{1}_{x}{-\frac{1}{y}\int_0^{y}{\frac{1}{\frac{1}{2}-z}\int_{\frac{1}{2}}^{z}{\frac{w^{n}-1}{w-1}\H0w}dwdz}dy}dx}}_{B:=}+\\
&&\underbrace{\int_0^{\frac{1}{3}}{\frac{1}{x-1}\int^{1}_{x}{-\frac{1}{y}\int_0^{y}{\frac{1}{\frac{1}{2}-z}\int_{\frac{1}{2}}^{z}{\frac{\H0w}{w-1}}dwdz}dy}dx}}_{C:=}.
\end{eqnarray*}
\normalsize
Let us first consider $B$: we split the integral at the zeroes of the denominators and get:
\small
\begin{eqnarray*}
B&=&\int_0^{\frac{1}{3}}{\frac{1}{x-1}\int^{\frac{1}{2}}_{x}{-\frac{1}{y}\int_0^{y}{\frac{1}{z-\frac{1}{2}}\int^{\frac{1}{2}}_{z}{-\frac{w^{n}-1}{w-1}\H0w}dwdz}dy}dx}+\\
&&\int_0^{\frac{1}{3}}{\frac{1}{x-1}\int^{1}_{\frac{1}{2}}{-\frac{1}{y}\int_0^{\frac{1}{2}}{\frac{1}{z-\frac{1}{2}}\int^{\frac{1}{2}}_{z}{-\frac{w^{n}-1}{w-1}\H0w}dwdz}dy}dx}+\\
&&\int_0^{\frac{1}{3}}{\frac{1}{x-1}\int^{1}_{\frac{1}{2}}{-\frac{1}{y}\int_{\frac{1}{2}}^{y}{\frac{1}{z-\frac{1}{2}}\int_{\frac{1}{2}}^{z}{\frac{w^{n}-1}{w-1}\H0w}dwdz}dy}dx}\\
&=&B_1+B_2+B_3.
\end{eqnarray*}
\normalsize
Starting from the inner integral and integrating integral by integral leads to
\small
\begin{eqnarray*}
B_1&=&\int_0^{\frac{1}{3}}\frac{1}{x-1}\int^{\frac{1}{2}}_{x}-\frac{1}{y}\int_0^{y}\frac{1}{z-\frac{1}{2}}\biggl(\H{0}{\frac{1}{2}}\SS{1}{\frac{1}{2}}n-\H{0}z\SS{1}{z}n\biggr. \\
   & & \biggl.+\SS{2}{\frac{1}{2}}n-\SS{2}{z}n\biggr)dzdydx\\
   &=&\cdots =\\
   &=&-\textnormal{H}_3(1) \textnormal{H}_{0,1,0}(1) \textnormal{S}_1\left(\frac{1}{2};n\right)+\textnormal{H}_0(2) \textnormal{H}_{3,0,\frac{3}{2}}(1) \textnormal{S}_1\left(\frac{1}{2};n\right)-\textnormal{H}_0(3) \textnormal{H}_{3,0,\frac{3}{2}}(1)
   \textnormal{S}_1\left(\frac{1}{2};n\right)\\
    &&-\textnormal{H}_{0,3,0,\frac{3}{2}}(1) \textnormal{S}_1\left(\frac{1}{2};n\right)-2 \textnormal{H}_{3,0,0,\frac{3}{2}}(1)
   \textnormal{S}_1\left(\frac{1}{2};n\right)-\textnormal{H}_0(2) \textnormal{H}_3(1) \textnormal{S}_{1,2}\hspace{-0.2em}\left(\frac{1}{2},1;n\right)\\
    &&+\textnormal{H}_0(3) \textnormal{H}_3(1) \textnormal{S}_{1,2}\hspace{-0.2em}\left(\frac{1}{2},2;n\right)+\textnormal{H}_{0,3}(1)
   \textnormal{S}_{1,2}\hspace{-0.2em}\left(\frac{1}{2},2;n\right)-2 \textnormal{H}_3(1) \textnormal{S}_{1,3}\hspace{-0.2em}\left(\frac{1}{2},1;n\right)\\
    &&+2 \textnormal{H}_3(1) \textnormal{S}_{1,3}\hspace{-0.2em}\left(\frac{1}{2},2;n\right)-\textnormal{H}_3(1)
   \textnormal{S}_{2,2}\hspace{-0.2em}\left(\frac{1}{2},1;n\right)+\textnormal{H}_3(1) \textnormal{S}_{2,2}\hspace{-0.2em}\left(\frac{1}{2},2;n\right)\\
    &&-\textnormal{H}_0(3)
   \textnormal{S}_{1,2,1}\hspace{-0.2em}\left(\frac{1}{2},2,\frac{1}{3};n\right)-\textnormal{S}_{1,2,2}\hspace{-0.2em}\left(\frac{1}{2},2,\frac{1}{3};n\right)-2
   \textnormal{S}_{1,3,1}\hspace{-0.2em}\left(\frac{1}{2},2,\frac{1}{3};n\right)\\
    &&-\textnormal{S}_{2,2,1}\hspace{-0.2em}\left(\frac{1}{2},2,\frac{1}{3};n\right).
\end{eqnarray*}
\normalsize
Applying the same strategy to $B_2$ and $B_3$ leads to
\small
\begin{eqnarray*}
B&=&-\textnormal{H}_3(1) \textnormal{H}_{0,1,0}(2) \textnormal{S}_1\left(\frac{1}{2};n\right)+\textnormal{H}_0(2) \textnormal{H}_{3,0,\frac{3}{2}}(1) \textnormal{S}_1\left(\frac{1}{2};n\right)
  -\textnormal{H}_0(3) \textnormal{H}_{3,0,\frac{3}{2}}(1)\textnormal{S}_1\left(\frac{1}{2};n\right)\\
&&-\textnormal{H}_{0,3,0,\frac{3}{2}}(1) \textnormal{S}_1\left(\frac{1}{2};n\right)-2 \textnormal{H}_{3,0,0,\frac{3}{2}}(1)
   \textnormal{S}_1\left(\frac{1}{2};n\right)+\textnormal{H}_0(3) \textnormal{H}_3(1) \textnormal{S}_{1,2}\hspace{-0.2em}\left(\frac{1}{2},2;n\right)\\
&&+\textnormal{H}_{0,3}(1) \textnormal{S}_{1,2}\hspace{-0.2em}\left(\frac{1}{2},2;n\right)-\textnormal{H}_0(3)
   \textnormal{S}_{1,2,1}\hspace{-0.2em}\left(\frac{1}{2},2,\frac{1}{3};n\right)-\textnormal{S}_{1,2,2}\hspace{-0.2em}\left(\frac{1}{2},2,\frac{1}{3};n\right)\\
&&-2\textnormal{S}_{1,3,1}\hspace{-0.2em}\left(\frac{1}{2},2,\frac{1}{3};n\right)-\textnormal{S}_{2,2,1}\hspace{-0.2em}\left(\frac{1}{2},2,\frac{1}{3};n\right).
\end{eqnarray*}
\normalsize
Let us now consider $C$: again we split the integral at the zeroes of the denominators and get:
\small
\begin{eqnarray*}
C&=&\int_0^{\frac{1}{3}}{\frac{1}{x-1}\int^{\frac{1}{2}}_{x}{-\frac{1}{y}\int_0^{y}{\frac{1}{z-\frac{1}{2}}\int^{\frac{1}{2}}_{z}{-\frac{\H0w}{w-1}}dwdz}dy}dx}+\\
&&\int_0^{\frac{1}{3}}{\frac{1}{x-1}\int^{1}_{\frac{1}{2}}{-\frac{1}{y}\int_0^{\frac{1}{2}}{\frac{1}{z-\frac{1}{2}}\int^{\frac{1}{2}}_{z}{-\frac{\H0w}{w-1}}dwdz}dy}dx}+\\
&&\int_0^{\frac{1}{3}}{\frac{1}{x-1}\int^{1}_{\frac{1}{2}}{-\frac{1}{y}\int_{\frac{1}{2}}^{y}{\frac{1}{z-\frac{1}{2}}\int_{\frac{1}{2}}^{z}{\frac{\H0w}{w-1}}dwdz}dy}dx}\\
&=&C_1+C_2+C_3.
\end{eqnarray*}
\normalsize
Starting from the inner integral and integrating integral by integral leads to
\small
\begin{eqnarray*}
C_1&=&\int_0^{\frac{1}{3}}{\frac{1}{x-1}\int^{\frac{1}{2}}_{x}{-\frac{1}{y}\int_0^{y}{\frac{\H{1,0}{\frac{1}{2}}-\H{1,0}{z}}{z-\frac{1}{2}}dz}dy}dx}\\
   &=&\int_0^{\frac{1}{3}}{\frac{1}{x-1}\int^{\frac{1}{2}}_{x}{-\frac{-\H{\frac{1}{2}}{y}\H{1,0}{\frac{1}{2}}+\H{\frac{1}{2},1,0}{y}} {y}}dy}dx\\
   &=&\int_0^{\frac{1}{3}}{\frac{\H{1,0}{\frac{1}{2}}\left(\H{0,\frac{1}{2}}{\frac{1}{2}}-\H{0,\frac{1}{2}}{x}\right)
      -\H{0,\frac{1}{2},1,0}{\frac{1}{2}}+\H{0,\frac{1}{2},1,0}{x}}{x-1}dx}\\
   &=&\textnormal{H}_0(2) \textnormal{H}_2(1) \textnormal{H}_3(1) \textnormal{H}_{0,1}(1)+\textnormal{H}_3(1) \textnormal{H}_{0,2}(1) \textnormal{H}_{0,1}(1)-\textnormal{H}_0(2) \textnormal{H}_3(1) \textnormal{H}_{0,1,2}(1)\\
    &&-\textnormal{H}_0(2) \textnormal{H}_2(1) \textnormal{H}_{3,0,\frac{3}{2}}(1)-\textnormal{H}_{0,2}(1)
   \textnormal{H}_{3,0,\frac{3}{2}}(1)-2 \textnormal{H}_3(1) \textnormal{H}_{0,0,1,2}(1)-\textnormal{H}_3(1) \textnormal{H}_{0,1,0,2}(1)\\
&&+\textnormal{H}_0(3) \textnormal{H}_{3,0,\frac{3}{2},3}(1)+\textnormal{H}_{0,3,0,\frac{3}{2},3}(1)+2
   \textnormal{H}_{3,0,0,\frac{3}{2},3}(1)+\textnormal{H}_{3,0,\frac{3}{2},0,3}(1).
\end{eqnarray*}
\normalsize
Applying the same strategy to $C_2$ and $C_3$ leads to
\small
\begin{eqnarray*}
C&=&\textnormal{H}_3(1) \textnormal{H}_0(2){}^2 \textnormal{H}_{0,-1}(1)+\textnormal{H}_2(1) \textnormal{H}_3(1) \textnormal{H}_0(2) \textnormal{H}_{0,1}(1)+\textnormal{H}_3(1) \textnormal{H}_0(2) \textnormal{H}_{-1,0,1}(1)\\
&&+\textnormal{H}_3(1) \textnormal{H}_0(2) \textnormal{H}_{0,-1,1}(1)-\textnormal{H}_3(1) \textnormal{H}_0(2)\textnormal{H}_{0,1,2}(1)-\textnormal{H}_2(1) \textnormal{H}_0(2) \textnormal{H}_{3,0,\frac{3}{2}}(1)\\
&&+\textnormal{H}_3(1) \textnormal{H}_{0,1}(1) \textnormal{H}_{0,2}(1)-\textnormal{H}_{0,2}(1) \textnormal{H}_{3,0,\frac{3}{2}}(1)-\textnormal{H}_3(1) \textnormal{H}_{-1,0,1,-1}(1)-2\textnormal{H}_3(1) \textnormal{H}_{0,0,1,2}(1)\\
&&-\textnormal{H}_3(1) \textnormal{H}_{0,1,0,2}(1)+\textnormal{H}_0(3) \textnormal{H}_{3,0,\frac{3}{2},3}(1)+\textnormal{H}_{0,3,0,\frac{3}{2},3}(1)+2\textnormal{H}_{3,0,0,\frac{3}{2},3}(1)+\textnormal{H}_{3,0,\frac{3}{2},0,3}(1).
\end{eqnarray*}
\normalsize
Adding $B$ and $C$ leads to the final result:
\small
\begin{eqnarray*}
&&\frac{\partial \SS{1,2,1}{\frac{1}{2},2,\frac{1}{3}}{n}} {\partial n}=\\
&&\hspace{0.5cm}\textnormal{H}_0(2) \textnormal{H}_{3,0,\frac{3}{2}}\hspace{-0.2em}(1) \textnormal{S}_1\left(\frac{1}{2};n\right)-\textnormal{H}_3(1) \textnormal{H}_{0,1,0}\hspace{-0.2em}(2) \textnormal{S}_1\left(\frac{1}{2};n\right)-\textnormal{H}_0(3) \textnormal{H}_{3,0,\frac{3}{2}}\hspace{-0.2em}(1)
   \textnormal{S}_1\left(\frac{1}{2};n\right)\\
&&\hspace{0.5cm}-\textnormal{H}_{0,3,0,\frac{3}{2}}\hspace{-0.2em}(1) \textnormal{S}_1\left(\frac{1}{2};n\right)-2 \textnormal{H}_{3,0,0,\frac{3}{2}}\hspace{-0.2em}(1)
   \textnormal{S}_1\left(\frac{1}{2};n\right)+\textnormal{H}_0(3) \textnormal{H}_3(1) \textnormal{S}_{1,2}\hspace{-0.2em}\left(\frac{1}{2},2;n\right)\\
&&\hspace{0.5cm}+\textnormal{H}_{0,3}\hspace{-0.2em}(1) \textnormal{S}_{1,2}\hspace{-0.2em}\left(\frac{1}{2},2;n\right)-\textnormal{H}_0(3)
   \textnormal{S}_{1,2,1}\hspace{-0.2em}\left(\frac{1}{2},2,\frac{1}{3};n\right)-\textnormal{S}_{1,2,2}\hspace{-0.2em}\left(\frac{1}{2},2,\frac{1}{3};n\right)\\
&&\hspace{0.5cm}-2 \textnormal{S}_{1,3,1}\hspace{-0.2em}\left(\frac{1}{2},2,\frac{1}{3};n\right)-\textnormal{S}_{2,2,1}\hspace{-0.2em}\left(\frac{1}{2},2,\frac{1}{3};n\right)+\textnormal{H}_3(1) \textnormal{H}_0(2){}^2 \textnormal{H}_{0,-1}\hspace{-0.2em}(1)\\
&&\hspace{0.5cm}+\textnormal{H}_2(1)\textnormal{H}_3(1) \textnormal{H}_0(2) \textnormal{H}_{0,1}\hspace{-0.2em}(1)+\textnormal{H}_3(1) \textnormal{H}_0(2) \textnormal{H}_{-1,0,1}\hspace{-0.2em}(1)+\textnormal{H}_3(1) \textnormal{H}_0(2) \textnormal{H}_{0,-1,1}\hspace{-0.2em}(1)\\
&&\hspace{0.5cm}-\textnormal{H}_3(1) \textnormal{H}_0(2) \textnormal{H}_{0,1,2}\hspace{-0.2em}(1)-\textnormal{H}_2(1) \textnormal{H}_0(2)\textnormal{H}_{3,0,\frac{3}{2}}\hspace{-0.2em}(1)+\textnormal{H}_3(1) \textnormal{H}_{0,1}\hspace{-0.2em}(1) \textnormal{H}_{0,2}\hspace{-0.2em}(1)-\textnormal{H}_{0,2}\hspace{-0.2em}(1) \textnormal{H}_{3,0,\frac{3}{2}}\hspace{-0.2em}(1)\\
&&\hspace{0.5cm}-\textnormal{H}_3(1) \textnormal{H}_{-1,0,1,-1}\hspace{-0.2em}(1)-2 \textnormal{H}_3(1) \textnormal{H}_{0,0,1,2}\hspace{-0.2em}(1)-\textnormal{H}_3(1)
   \textnormal{H}_{0,1,0,2}\hspace{-0.2em}(1)\\
&&\hspace{0.5cm}+\textnormal{H}_0(3) \textnormal{H}_{3,0,\frac{3}{2},3}\hspace{-0.2em}(1)+\textnormal{H}_{0,3,0,\frac{3}{2},3}\hspace{-0.2em}(1)+2 \textnormal{H}_{3,0,0,\frac{3}{2},3}\hspace{-0.2em}(1)+\textnormal{H}_{3,0,\frac{3}{2},0,3}\hspace{-0.2em}(1).
\end{eqnarray*}
\normalsize
\end{example}

\subsubsection*{Details of the Method}
Let us now look in detail at this method to compute the differentiation of a S-sum $\S{\ve m}{\ve b; n}=\S{m_1,\ldots,m_k}{b_1,\ldots,b_l; n}$. We start with the integral representation 
of $\S{\ve m}{\ve b; n}$ 
given in Theorem \ref{SSintrep}. We split the integral at $0$ and $1$ and at all zeroes of the denominators, \ie at $0,1,b_1,b_1b_2,\ldots, b_1b_2\cdots b_{l-1}$ such that 
$0,1,b_1,b_1b_2,\ldots, b_1b_2\cdots b_{k-1}$ are always at the borders of the integration domains and never inside the integration domains. We end up with a sum of integrals of the form
\begin{eqnarray*}
&&\pm \int_a^b\frac{dx_k}{x_k-c_k} \int_{u_k}^{o_k}\frac{dx_{k-1}}{x_{k-1}-c_{k-1}}\int_{u_{k-1}}^{o_{k-1}}\frac{dx_{k-2}}{x_{k-2}-c_{k-2}}
      \cdots\int_{u_3}^{o_3}\frac{dx_2}{x_2-c_2}\int_{u_2}^{o_2}\frac{dx_1}{x_1-c_1}\\
&&\times \int_e^f\frac{dy_r}{y_r-g_r} \int_{s_r}^{t_r}\frac{dy_{r-1}}{y_{r-1}-g_{r-1}}\int_{s_{r-1}}^{t_{r-1}}\frac{dy_{r-2}}{y_{r-2}-g_{r-2}}
      \cdots\int_{s_3}^{t_3}\frac{dy_2}{y_2-g_2}\int_{s_2}^{t_2}\frac{{y_1}^n-1}{y_1-1}dy_1
\end{eqnarray*}
where $(u_i,o_i)=(a,x_i)\vee (u_i,o_i)=(x_i,b)\vee (u_i,o_i)=(a,b);\hspace{0.3em}b> a \geq 0 \vee 0\geq b> a;\hspace{0.3em}c_i~\in~\R\setminus(a,b);$
$a\neq c_1\neq b$ and $(s_i,t_i)=(e,y_i)\vee (s_i,t_i)=(y_i,f);\hspace{0.3em}f>e \geq 1 \vee 1\geq f>e\geq 0 \vee 0\geq f>e;$\hspace{0.3em} $g_i\in \R\setminus(e,f).$\\
Note that due to the structure of the integral representation we have in addition:
\begin{eqnarray*}
&\text{if } c_i=a \longrightarrow  &u_i=a\\
&\text{if } c_i=b \longrightarrow  &o_i=b\\
&\text{if } g_i=e \longrightarrow  &s_i=e\\
&\text{if } g_i=f \longrightarrow  &t_i=f.
\end{eqnarray*}
In the case that $0\geq b>a$ we apply the transform $x_i\rightarrow-x_i$ to the integral and get
$$
\pm \int_{\bar{a}}^{\bar{b}}\frac{dx_k}{x_k-\bar{c}_k} \int_{\bar{u}_k}^{\bar{o}_k}\frac{dx_{k-1}}{x_{k-1}-\bar{c}_{k-1}} \int_{\bar{u}_{k-1}}^{\bar{o}_{k-1}}\frac{dx_{k-2}}{x_{k-2}-\bar{c}_{k-2}}
      \cdots\int_{\bar{u}_3}^{\bar{o}_3}\frac{dx_2}{x_2-\bar{c}_2}\int_{\bar{u}_2}^{\bar{o}_2}\frac{dx_1}{x_1-\bar{c}_1}
$$
with $\bar{a}=-b; \hspace{0.3em} \bar{b}=-a$ and $\bar{c}_i=-c_i;$ where $(\bar{u}_i,\bar{o}_i)=(\bar{a},x_i)\vee(\bar{u}_i,\bar{o}_i)=(x_i,\bar{b})\vee (\bar{u}_i,\bar{o}_i)=(\bar{a},\bar{b}); \hspace{0.3em} \bar{b}>\bar{a}\geq 0;$
$\bar{c}_i\in\R\setminus(\bar{a},\bar{b}); \bar{a}\neq\bar{c}_1\neq \bar{b}$ and
\begin{eqnarray*}
&\text{if } \bar{c}_i=\bar{a} \longrightarrow  &\bar{u}_i=\bar{a}\\
&\text{if } \bar{c}_i=\bar{b} \longrightarrow  &\bar{o}_i=\bar{b}.
\end{eqnarray*}
Hence we can always assume $b>a\geq0$.\\
In the case that $0\geq f>e$ we apply the transform $y_i\rightarrow-y_i$ to the integral and get
$$
\pm \int_{\bar{e}}^{\bar{f}}\frac{dy_r}{y_r-\bar{g}_r} \int_{\bar{s}_r}^{\bar{t}_r}\frac{dy_{r-1}}{y_{r-1}-\bar{g}_{r-1}} \int_{\bar{s}_{r-1}}^{\bar{t}_{r-1}}\frac{dx_{r-2}}{y_{r-2}-\bar{g}_{r-2}}
      \cdots\int_{\bar{s}_3}^{\bar{t}_3}\frac{dx_2}{x_2-\bar{g}_2}\int_{\bar{s}_2}^{\bar{t}_2}\frac{(-y_1)^n-1}{y_1+1}dy_1
$$
with $\bar{e}=-f; \hspace{0.3em} \bar{f}=-e$ and $\bar{g}_i=-g_i;$ where $(\bar{s}_i,\bar{t}_i)=(\bar{e},y_i)\vee(\bar{s}_i,\bar{t}_i)=(y_i,\bar{f}); \hspace{0.3em} \bar{f}>\bar{e}\geq 0;$
$\bar{g}_i\in\R\setminus(\bar{e},\bar{f})$ and
\begin{eqnarray*}
&\text{if } \bar{g}_i=\bar{e} \longrightarrow  &\bar{s}_i=\bar{e}\\
&\text{if } \bar{g}_i=\bar{f} \longrightarrow  &\bar{t}_i=\bar{f}.
\end{eqnarray*}
Hence we can rewrite the integral representation of Theorem \ref{SSintrep} as a sum of integrals of the form
\begin{eqnarray}
&&\pm \int_a^b\frac{dx_k}{x_k-c_k} \int_{u_k}^{o_k}\frac{dx_{k-1}}{x_{k-1}-c_{k-1}}
      \cdots\int_{u_3}^{o_3}\frac{dx_2}{x_2-c_2}\int_{u_2}^{o_2}\frac{dx_1}{x_1-c_1}\nonumber\\
&&\times \int_e^f\frac{dy_r}{y_r-g_r} \int_{s_r}^{t_r}\frac{dy_{r-1}}{y_{r-1}-g_{r-1}}
      \cdots\int_{s_3}^{t_3}\frac{dy_2}{y_2-g_2}\int_{s_2}^{t_2}\frac{{(\pm y_1)}^n-1}{y_1 \mp 1}dy_1
\label{SSintrepsplitted}
\end{eqnarray}
where $(u_i,o_i)=(a,x_i)\vee (u_i,o_i)=(x_i,b);\hspace{0.3em}b> a \geq 0;\hspace{0.3em}c_i\in \R\setminus(a,b);$
$a\neq c_1\neq b$ and $(s_i,t_i)=(e,y_i)\vee (s_i,t_i)=(y_i,f);\hspace{0.3em}f>e \geq 1 \vee 1\geq f>e\geq 0;$\hspace{0.3em} $g_i\in \R\setminus(e,f)$ and
\begin{eqnarray*}
&\text{if } c_i=a \longrightarrow  &u_i=a\\
&\text{if } c_i=b \longrightarrow  &o_i=b\\
&\text{if } g_i=e \longrightarrow  &s_i=e\\
&\text{if } g_i=f \longrightarrow  &t_i=f.
\end{eqnarray*}
Differentiating of an integral of the form (\ref{SSintrepsplitted}) with respect to $n$ leads to (note that we did not differentiate the factor $(-1)^n$):
\begin{eqnarray*}
&&\pm \int_a^b\frac{dx_k}{x_k-c_k} \int_{u_k}^{o_k}\frac{dx_{k-1}}{x_{k-1}-c_{k-1}}
      \cdots\int_{u_3}^{o_3}\frac{dx_2}{x_2-c_2}\int_{u_2}^{o_2}\frac{dx_1}{x_1-c_1}\\
&&\times \int_e^f\frac{dy_r}{y_r-g_r} \int_{s_r}^{t_r}\frac{dy_{r-1}}{y_{r-1}-g_{r-1}}
      \cdots\int_{s_3}^{t_3}\frac{dy_2}{y_2-g_2}\int_{s_2}^{t_2}\frac{{(\pm y_1)}^n\H{0}{y_1}}{y_1\mp1}dy_1\\
&&=\pm \int_a^b\frac{dx_k}{x_k-c_k} \int_{u_k}^{o_k}\frac{dx_{k-1}}{x_{k-1}-c_{k-1}}
       \cdots\int_{u_3}^{o_3}\frac{dx_2}{x_2-c_2}\int_{u_2}^{o_2}\frac{dx_1}{x_1-c_1}\\
&&\times \int_e^f\frac{dy_r}{y_r-g_r} \int_{s_r}^{t_r}\frac{dy_{r-1}}{y_{r-1}-g_{r-1}}
       \cdots\int_{s_3}^{t_3}\frac{dy_2}{y_2-g_2}\int_{s_2}^{t_2}\frac{{(\pm y_1)}^n-1}{y_1\mp1}\H{0}{y_1}dy_1\\
&&\pm \int_a^b\frac{dx_k}{x_k-c_k} \int_{u_k}^{o_k}\frac{dx_{k-1}}{x_{k-1}-c_{k-1}}
       \cdots\int_{u_3}^{o_3}\frac{dx_2}{x_2-c_2}\int_{u_2}^{o_2}\frac{dx_1}{x_1-c_1}\\
&&\times \int_e^f\frac{dy_r}{y_r-g_r} \int_{s_r}^{t_r}\frac{dy_{r-1}}{y_{r-1}-g_{r-1}}
      \cdots\int_{s_3}^{t_3}\frac{dy_2}{y_2-g_2}\int_{s_2}^{t_2}\frac{\H{0}{y_1}}{y_1\mp1}dy_1.
\end{eqnarray*}
Now we want to transform these integrals back to S-sums and constants related to S-sums, \ie S-sums at infinity or multiple polylogarithms evaluated at constants. For this task we have to consider 
the following 5 different classes of integrals
\begin{description}
 \item[(a)] $\int_a^b\frac{dx_k}{x_k-c_k} \int_{u_k}^{o_k}\frac{dx_{k-1}}{x_{k-1}-c_{k-1}}\cdots\int_{u_3}^{o_3}\frac{dx_2}{x_2-c_2}\int_{u_2}^{o_2}\frac{dx_1}{x_1-c_1}$
 \item[(b)] $\int_e^f\frac{dy_r}{y_r-g_r} \int_{s_r}^{t_r}\frac{dy_{r-1}}{y_{r-1}-g_{r-1}}\cdots\int_{s_3}^{t_3}\frac{dy_2}{y_2-g_2}\int_{s_2}^{t_2}\frac{\H{0}{y_1}}{y_1-1}dy_1$
 \item[(c)] $\int_e^f\frac{dy_r}{y_r-g_r} \int_{s_r}^{t_r}\frac{dy_{r-1}}{y_{r-1}-g_{r-1}}\cdots\int_{s_3}^{t_3}\frac{dy_2}{y_2-g_2}\int_{s_2}^{t_2}\frac{\H{0}{y_1}}{y_1+1}dy_1$
 \item[(d)] $\int_e^f\frac{dy_r}{y_r-g_r} \int_{s_r}^{t_r}\frac{dy_{r-1}}{y_{r-1}-g_{r-1}}\cdots\int_{s_3}^{t_3}\frac{dy_2}{y_2-g_2}\int_{s_2}^{t_2}\frac{y_1{}^n-1}{y_1-1}\H{0}{y_1}dy_1$
 \item[(e)] $\int_e^f\frac{dy_r}{y_r-g_r} \int_{s_r}^{t_r}\frac{dy_{r-1}}{y_{r-1}-g_{r-1}}\cdots\int_{s_3}^{t_3}\frac{dy_2}{y_2-g_2}\int_{s_2}^{t_2}\frac{{(-y_1)}^n-1}{y_1+1}\H{0}{y_1}dy_1$
\end{description}
with $b>a\geq 0$ \hspace{0.3em}$f>e\geq1$ or $1\geq f >e \geq 0$\hspace{0.3em}$(u_i,o_i)=(a,x_i)\vee (u_i,o_i)=(x_i,b);\hspace{0.3em}b> a \geq 0;\hspace{0.3em}c_i\in \R\setminus(a,b);$
$a\neq c_1\neq b$ and $(s_i,t_i)=(e,y_i)\vee (s_i,t_i)=(y_i,f);\hspace{0.3em}f>e \geq 1 \vee 1\geq f>e\geq 0;$\hspace{0.3em} $g_i\in \R\setminus(e,f)$ and
\begin{eqnarray*}
&\text{if } c_i=a \longrightarrow  &u_i=a\\
&\text{if } c_i=b \longrightarrow  &o_i=b\\
&\text{if } g_i=e \longrightarrow  &s_i=e\\
&\text{if } g_i=f \longrightarrow  &t_i=f.
\end{eqnarray*}
In the following we will use the abbreviations 
\begin{eqnarray*}
\H{\ve m}{x|y}&:=&\H{\ve m}{x}-\H{\ve m}{y}\\
\S{m_1,\ldots,m_k}{b_1,\ldots,b_{k-1},b_k|b_{k+1};n}&:=&\S{m_1,\ldots,m_k}{b_1,\ldots,b_{k-1},b_k;n}\\&&-\S{m_1,\ldots,m_k}{b_1,\ldots,b_{k-1},b_{k+1};n}.
\end{eqnarray*}
We start with integrals of type \textbf{(a)}. We will show that we can represent it in the form
\begin{eqnarray}
 \sum_{\iota=1}^pz_{\iota}\H{\ve m_{\iota}}{b,a}+\sum_{\iota=1}^qv_{\iota}\H{\ve n_{\iota}}{b_i-c_j|a-c_j}
\label{SSintconstform}
\end{eqnarray}
with $z_{\iota},v_{\iota}$ products of multiple polylogarithms evaluated at constants and $j\in\{1,2,\ldots,k\}.$ 
If $k=1,$ we get
\begin{eqnarray*}
 \int_a^b\frac{dx_1}{x_1-c_1}=\left\{ 
		  	\begin{array}{ll}
						-\sign{c_1}\H{c_1}{a|b},&  \textnormal{if }c_1>b \vee 0\geq c_1\\
						\H{0}{a-c_1|b-c_1},& \textnormal{if }a>c_1>0. 
			\end{array}\right.
\end{eqnarray*}
Let now $1<i\leq k.$ We will see that we can express the first $i-1$ integrals in the form 
\begin{eqnarray*}
 \sum_{\iota=1}^pz_{\iota}\H{\ve m_{\iota}}{o_i|u_i}+\sum_{\iota=1}^qv_{\iota}\H{\ve n_{\iota}}{o_i-c_j|u_i-c_j}
\end{eqnarray*}
with $z_{\iota},v_{\iota}$ products of multiple polylogarithms evaluated at constants and $j\in\{1,2,\ldots,i-1\}$. For $i=2$ this is obvious. We will now perform the induction step $i\rightarrow i+1,$ therefore we will look at the 
integrals 
$$A:=\int_{u_{i+1}}^{o_{i+1}}\frac{\H{\ve m}{o_i|u_i}}{x_i-c_i}dx_i \text{ and }B:=\int_{u_{i+1}}^{o_{i+1}}\frac{\H{\ve m}{o_i-c_j|u_i-c_j}}{x_i-c_i}dx_i$$
where $u_{k+1}:=a,o_{k+1}:=b$ and distinguish several cases:
\begin{itemize}
 \item $(u_i,o_i)=(a,x_i);c_i>b\vee0>c_i:$
	\begin{eqnarray*}
	  A&=&\int_{u_{i+1}}^{o_{i+1}}\frac{\H{\ve m}{x_i|a}}{x_i-c_i}dx_i=-\sign{c_i}\H{c_i,\ve m}{o_{i+1}|u_{i+1}}\\
	      &&+\sign{c_i}\H{c_i}{o_{i+1}|u_{i+1}}\H{\ve m}{a}\\
	  B&=&\int_{u_{i+1}}^{o_{i+1}}\frac{\H{\ve m}{x_i-c_j|a-c_j}}{x_i-c_i}dx_i=\int_{u_{i+1}-c_j}^{o_{i+1}-c_j}\frac{\H{\ve m}{x_i|a-c_j}}{x_i+c_j-c_i}dx_i\\
	   &=&-\sign{c_i-c_j}\H{c_i-c_j,\ve m}{o_{i+1}-c_j|u_{i+1}-c_j}\\
	      &&+\sign{c_i-c_j}\H{c_i-c_j}{o_{i+1}-c_j|u_{i+1}-c_j}\H{\ve m}{a-c_j}.
	\end{eqnarray*}
	Hence in all of the three possible cases, \ie $(u_{i+1},o_{i+1})$ equals $(a,x_{i+1})$, $(x_{i+1},b)$ or $(a,b)$ we arrive at the desired results.
 \item $(u_i,o_i)=(x_i,b);c_i>b\vee0>c_i:$
	\begin{eqnarray*}
	  A&=&\int_{u_{i+1}}^{o_{i+1}}\frac{\H{\ve m}{b|x_i}}{x_i-c_i}dx_i=\sign{c_i}\H{c_i,\ve m}{o_{i+1}|u_{i+1}}\\
	      &&-\sign{c_i}\H{c_i}{o_{i+1}|u_{i+1}}\H{\ve m}{a}\\
	  B&=&\int_{u_{i+1}}^{o_{i+1}}\frac{\H{\ve m}{b-c_j|x_i-c_j}}{x_i-c_i}dx_i=\int_{u_{i+1}-c_j}^{o_{i+1}-c_j}\frac{\H{\ve m}{b-c_j|x_i}}{x_i+c_j-c_i}dx_i\\
	   &=&\sign{c_i-c_j}\H{c_i-c_j,\ve m}{o_{i+1}-c_j|u_{i+1}-c_j}\\
	      &&-\sign{c_i-c_j}\H{c_i-c_j}{o_{i+1}-c_j|u_{i+1}-c_j}\H{\ve m}{b-c_j}.
	\end{eqnarray*}
	Hence in all of the three possible cases, \ie $(u_{i+1},o_{i+1})$ equals $(a,x_{i+1})$, $(x_{i+1},b)$ or $(a,b)$ we arrive at the desired results.
 \item $(u_i,o_i)=(a,x_i);a>c_i\geq 0:$
	\begin{eqnarray*}
	  A&=&\int_{u_{i+1}-c_i}^{o_{i+1}-c_i}\frac{\H{\ve m}{x_i+c_i}-\H{\ve m}{a}}{x_i}dx_i\\
	  B&=&\int_{u_{i+1}-c_i}^{o_{i+1}-c_i}\frac{\H{\ve m}{x_i+c_i-c_j}-\H{\ve m}{a}}{x_i}dx_i.
	\end{eqnarray*}
	Applying transforms discussed in Section \ref{SSRelatedArguments} to $\H{\ve m}{x_i+c_i}$ and to $\H{\ve m}{x_i+c_i-c_j}$ leads to
	\begin{eqnarray*}
	  A&=&\sum_{\rho}\int_{u_{i+1}-c_i}^{o_{i+1}-c_i}\xi_{\rho}\frac{\H{\ve m_{\rho}}{x_i}}{x_i}-\H{0}{o_{i+1}-c_i|u_{i+1}-c_i}\H{\ve m}a\\
	   &=&\sum_{\rho}\xi_{\rho}\H{\ve m_{\rho}}{o_{i+1}-c_i|u_{i+1}-c_i}-\H{0}{o_{i+1}-c_i|u_{i+1}-c_i}\H{\ve m}a\\
	  B&=&\sum_{\rho}\int_{u_{i+1}-c_i}^{o_{i+1}-c_i}\xi_{\rho}\frac{\H{\ve m_{\rho}}{x_i}}{x_i}-\H{0}{o_{i+1}-c_i|u_{i+1}-c_i}\H{\ve m}a\\
	   &=&\sum_{\rho}\mu_{\rho}\H{\ve m_{\rho}}{o_{i+1}-c_i|u_{i+1}-c_i}-\H{0}{o_{i+1}-c_i|u_{i+1}-c_i}\H{\ve m}a,
	\end{eqnarray*}
	where $\xi_{\rho}$ and $\mu_\rho$ are multiple polylogarithms evaluated at constants.
	Hence in all of the three possible cases, \ie $(u_{i+1},o_{i+1})$ equals $(a,x_{i+1})$, $(x_{i+1},b)$ or $(a,b)$ we get the desired results.
 \item $(u_i,o_i)=(x_i,b);a>c_i\geq 0:$\\
	This case can be basically handled like the previous one.
 \item $(u_i,o_i)=(x_i,b);c_i=b:$
	\begin{eqnarray*}
	  A&=&\int_{u_{i+1}-b}^{o_{i+1}-b}\frac{\H{\ve m}{x_i+b}-\H{\ve m}{a}}{x_i}dx_i\\
	  B&=&\int_{u_{i+1}-b}^{o_{i+1}-b}\frac{\H{\ve m}{x_i+b-c_j}-\H{\ve m}{a}}{x_i}dx_i.
	\end{eqnarray*}
	Applying transforms discussed in Section \ref{SSRelatedArguments} to $\H{\ve m}{x_i+c_i}$ and to $\H{\ve m}{x_i+c_i-c_j}$ leads to
	\begin{eqnarray*}
	  A&=&\sum_{\rho}\int_{u_{i+1}-b}^{o_{i+1}-b}\xi_{\rho}\frac{\H{\ve m_{\rho}}{x_i}}{x_i}=\sum_{\rho}\xi_{\rho}\H{\ve m_{\rho}}{o_{i+1}-b|u_{i+1}-b}\\
	  B&=&\sum_{\rho}\int_{u_{i+1}-b}^{o_{i+1}-b}\xi_{\rho}\frac{\H{\ve m_{\rho}}{x_i}}{x_i}=\sum_{\rho}\mu_{\rho}\H{\ve m_{\rho}}{o_{i+1}-b|u_{i+1}-b},
	\end{eqnarray*}
	where $\xi_{\rho}$ and $\mu_\rho$ are multiple polylogarithms evaluated at constants and none of the leading indices of $\H{\ve m_{\rho}}{x_i}$ equals $o_{i+1}-b$.
	Hence in all of the three possible cases, \ie $(u_{i+1},o_{i+1})$ equals $(a,x_{i+1})$, $(x_{i+1},b)$ or $(a,b)$ we obain the desired results.
 \item $(u_i,o_i)=(a,x_i);c_i=a:$\\
	This case can be handled similarly to the previous one.
\end{itemize}
Combining all these cases we see that we can always express an integral of type \textbf{(a)} in the form:
\begin{eqnarray*}
 \sum_{\iota=1}^pz_{\iota}\H{\ve m_{\iota}}{b|a}+\sum_{\iota=1}^qv_{\iota}\H{\ve n_{\iota}}{b_i-c_j|a-c_j}.
\end{eqnarray*}
Let us now look at the inner integrals of \textbf{(b)} and \textbf{(c)}. Since we have
\begin{eqnarray*}
 \int_{s_2}^{t_2}\frac{\H{0}{y_1}}{y_1-1}dy_1=-\H{1,0}{t_2|s_2} \text{ and } \int_{s_2}^{t_2}\frac{\H{0}{y_1}}{y_1+1}dy_1&=&\H{-1,0}{t_2|s_2}
\end{eqnarray*}
we can use basically the same arguments for integrals of type \textbf{(b)} and \textbf{(c)} as we used for integrals of type \textbf{(a)} to see that we can express them in the same form, \ie in the form 
(\ref{SSintconstform}).

In the following we will show that we can represent integrals of type \textbf{(d)}  in the form
\begin{eqnarray}
&&\sum_{\iota=1}^{p_1}z_{\iota}\left(\H{0}{\frac{f}{\alpha_\iota}}\S{\ve m_{\iota}}{\ve b_{\iota}|\frac{f}{\alpha_\iota};n}-\H{0}{\frac{e}{\alpha_{\iota}}}\S{\ve m_{\iota}}{\ve b_{\iota}|\frac{e}{\alpha_\iota};n}\right)\nonumber\\
  &&+\sum_{\iota=1}^{p_2} \bar{z}_{\iota} \left(\S{\ve w_{\iota}}{\ve v_{\iota}|\frac{f}{\beta_{\iota}};n}-\S{\ve w_{\iota}}{\ve v_{\iota}|\frac{e}{\beta_{\iota}};n}\right)+\sum_{\iota=1}^{p_3} \hat{z}_{\iota}\; \S{\ve h}{\ve q| \frac{g_j}{\gamma_{\iota}}; n}
\label{SSintdiffform}
\end{eqnarray}
with $z_{\iota},\bar{z}_{\iota},\hat{z}_{\iota}$ products of multiple polylogarithms evaluated at constants and $\alpha_{\iota},\beta_{\iota},\gamma_{\iota}\in\R^*.$ 
If $k=1,$ we get
\begin{eqnarray}
 \int_e^f\frac{y_1{}^n-1}{y_1-1}dy_1=\H{0}{f}\S{1}{f;n}-\H{0}{e}\S{1}{e;n}-\S{2}{f;n}+\S{2}{e;n}.
\label{SSintdiff1form}
\end{eqnarray}
Let now $1<i\leq k.$ We will see that we can express the first $i-1$ integrals in the form 
\begin{eqnarray*}
&&\sum_{\iota=1}^{p_1}z_{\iota}\left(\H{0}{\frac{t_i}{\alpha_\iota}}\S{\ve m_{\iota}}{\ve b_{\iota}|\frac{t_i}{\alpha_\iota};n}-\H{0}{\frac{s_i}{\alpha_{\iota}}}\S{\ve m_{\iota}}{\ve b_{\iota}|\frac{s_i}{\alpha_\iota};n}\right)\\
  &&+\sum_{\iota=1}^{p_2} v_{\iota} \left(\S{\ve w_{\iota}}{\ve v_{\iota}|\frac{t_i}{\beta_{\iota}};n}-\S{\ve w_{\iota}}{\ve v_{\iota}|\frac{s_i}{\beta_{\iota}};n}\right)\\
  &&+\sum_{\iota=1}^{p_3} \hat{z}_{\iota}\;\H{1,0}{\frac{t_{i}}{g_{i-1}}|\frac{s_{i}}{g_{i-1}}}+\sum_{\iota=1}^{p_4} \tilde{z}_{\iota}\;\H{0}{t_{i}-g_{i-1}|s_{i}-g_{i-s}}
\end{eqnarray*}
with $z_{\iota},\bar{z}_{\iota},\hat{z}_{\iota},\tilde{z}_{\iota}$ products of multiple polylogarithms evaluated at constants and S-sums not dependent on a integration variable. 
For $i=2$ this is obvious (compare \ref{SSintdiff1form}). We will now perform the induction step 
$i\rightarrow i+1.$ Therefore we will look at the integrals 
\begin{eqnarray*}
A&:=&\int_{s_{i+1}}^{t_{i+1}}\frac{\H{0}{\frac{t_i}{\alpha}}\S{\ve m}{\ve b|\frac{t_i}{\alpha};n}-\H{0}{\frac{s_i}{\alpha}} \S{\ve m}{\ve b|\frac{s_i}{\alpha};n} }  {y_i-g_i} dy_i\\
B&:=&\int_{s_{i+1}}^{t_{i+1}}\frac{\S{\ve w}{\ve v|\frac{t_i}{\beta};n}-\S{\ve w}{\ve v|\frac{s_i}{\beta};n}} {y_i-g_i}dy_i\\
C&:=&\int_{s_{i+1}}^{t_{i+1}}\frac{\H{1,0}{\frac{t_{i}}{g_{i-1}}|\frac{s_{i}}{g_{i-1}}}}{y_i-g_i}dy_i\\
D&:=&\int_{s_{i+1}}^{t_{i+1}}\frac{\H{0}{t_{i}-g_{i-1}|s_{i}-g_{i-1}}}{y_i-g_i} dy_i
\end{eqnarray*}
where $s_{k+1}:=e,t_{k+1}:=f.$ 
Using the method as described for integrals of type (a) we can handle $D$ and using simple integral transforms we can handle $C$ in the same way.
For $A$ and $B$ we distinguish several cases:
\begin{itemize}
 \item $(s_i,t_i)=(e,y_i);g_i>f\vee0>g_i:$
	\begin{eqnarray*}
	  A&=&\int_{s_{i+1}}^{t_{i+1}}\frac{\H{0}{\frac{y_i}{\alpha}}\S{\ve m}{\ve b|\frac{y_i}{\alpha};n}-\H{0}{\frac{e}{\alpha}} \S{\ve m}{\ve b|\frac{e}{\alpha};n} }  {y_i-g_i} dy_i\\
	   &=&\H{0}{\abs{t_{i+1}-g_i}|\abs{s_{i+1}-g_i}}\left(\H{0}{\frac{g_i}{\alpha}}\S{\ve m}{\ve b|\frac{g_i}{\alpha};n}-\H{0}{\frac{e}{\alpha}}\S{\ve m}{\ve b|\frac{e}{\alpha};n}\right)\\
	    &&-\H{1,0}{\frac{t_{i+1}}{g_i}|\frac{s_{i+1}}{g_i}}\S{\ve m}{\ve b|\frac{g_i}{\alpha};n}\\
	    &&+\H{0}{\frac{g_i}{\alpha}}\left(\S{\ve m,1}{\ve b,\frac{g_i}{\alpha}|\frac{t_{i+1}}{g_i};n}-\S{\ve m,1}{\ve b,\frac{g_i}{\alpha}|\frac{s_{i+1}}{g_i};n}\right)\\
	    &&+\left(\H{0}{\frac{t_{i+1}}{g_i}}\S{\ve m,1}{\ve b,\frac{g_i}{\alpha}|\frac{t_{i+1}}{g_i};n}-\H{0}{\frac{s_{i+1}}{g_i}}\S{\ve m,1}{\ve b,\frac{g_i}{\alpha}|\frac{s_{i+1}}{g_i};n}\right)\\
	    &&-\left(\S{\ve m,2}{\ve b,\frac{g_i}{\alpha}|\frac{t_{i+1}}{g_i};n}-\S{\ve m,2}{\ve b,\frac{g_i}{\alpha}|\frac{s_{i+1}}{g_i};n}\right)\\
	  B&=&\int_{s_{i+1}}^{t_{i+1}}\frac{\S{\ve w}{\ve v|\frac{y_i}{\beta};n}-\S{\ve w}{\ve v|\frac{e}{\beta};n}} {y_i-g_i}dy_i\\
	   &=&\S{\ve w,1}{\ve v,\frac{g_i}{\beta}|\frac{t_{i+1}}{g_i};n}-\S{\ve w,1}{\ve v,\frac{g_i}{\beta}|\frac{s_{i+1}}{g_i};n}\\
	      &&+\H{0}{\abs{t_{i+1}-g_i}|\abs{s_{i+1}-g_i}}\left(\S{\ve w}{\ve v|\frac{g_i}{\beta};n}-\S{\ve w}{\ve v|\frac{e}{\beta};n}\right).
	\end{eqnarray*}
	Hence in all of the three possible cases, \ie $(s_{i+1},t_{i+1})$ equals $(e,y_{i+1})$, $(y_{i+1},f)$ or $(e,f),$ we arrive at the desired results.
 \item $(s_i,t_i)=(y_i,f);g_i>f\vee0>g_i:$
	This case can be handled similarly to the previous one.
 \item $(s_i,t_i)=(e,y_i);e>g_i\geq 0:$
	\begin{eqnarray*}
	  A&=&\int_{s_{i+1}}^{t_{i+1}}\frac{\H{0}{\frac{y_i}{\alpha}}\S{\ve m}{\ve b|\frac{y_i}{\alpha};n}-\H{0}{\frac{e}{\alpha}} \S{\ve m}{\ve b|\frac{e}{\alpha};n} }  {y_i-g_i} dy_i\\
	   &=&\H{0}{t_{i+1}-g_i|s_{i+1}-g_i}\left(\H{0}{\frac{g_i}{\alpha}}\S{\ve m}{\ve b|\frac{g_i}{\alpha};n}-\H{0}{\frac{e}{\alpha}}\S{\ve m}{\ve b|\frac{e}{\alpha};n}\right)\\
	    &&-\H{1,0}{\frac{t_{i+1}}{g_i}|\frac{s_{i+1}}{g_i}}\S{\ve m}{\ve b|\frac{g_i}{\alpha};n}\\
	    &&+\H{0}{\frac{g_i}{\alpha}}\left(\S{\ve m,1}{\ve b,\frac{g_i}{\alpha}|\frac{t_{i+1}}{g_i};n}-\S{\ve m,1}{\ve b,\frac{g_i}{\alpha}|\frac{s_{i+1}}{g_i};n}\right)\\
	    &&+\left(\H{0}{\frac{t_{i+1}}{g_i}}\S{\ve m,1}{\ve b,\frac{g_i}{\alpha}|\frac{t_{i+1}}{g_i};n}-\H{0}{\frac{s_{i+1}}{g_i}}\S{\ve m,1}{\ve b,\frac{g_i}{\alpha}|\frac{s_{i+1}}{g_i};n}\right)\\
	    &&-\left(\S{\ve m,2}{\ve b,\frac{g_i}{\alpha}|\frac{t_{i+1}}{g_i};n}-\S{\ve m,2}{\ve b,\frac{g_i}{\alpha}|\frac{s_{i+1}}{g_i};n}\right)\\
	  B&=&\int_{s_{i+1}}^{t_{i+1}}\frac{\S{\ve w}{\ve v|\frac{y_i}{\beta};n}-\S{\ve w}{\ve v|\frac{e}{\beta};n}} {y_i-g_i}dy_i\\
	   &=&\S{\ve w,1}{\ve v,\frac{g_i}{\beta}|\frac{t_{i+1}}{g_i};n}-\S{\ve w,1}{\ve v,\frac{g_i}{\beta}|\frac{s_{i+1}}{g_i};n}\\
	      &&+\H{0}{t_{i+1}-g_i|s_{i+1}-g_i}\left(\S{\ve w}{\ve v|\frac{g_i}{\beta};n}-\S{\ve w}{\ve v|\frac{e}{\beta};n}\right).
	\end{eqnarray*}
	Hence in all of the three possible cases, \ie $(s_{i+1},t_{i+1})$ equals $(e,y_{i+1})$, $(y_{i+1},f)$ or $(e,f),$ we arrive at the desired results.
 \item $(s_i,t_i)=(y_i,f);e>g_i\geq 0:$\\
	This case can be basically handled like the previous one.
 \item $(s_i,t_i)=(y_i,f);g_i=f:$
	\begin{eqnarray*}
	  A&=&\int_{s_{i+1}}^{t_{i+1}}\frac{\H{0}{\frac{f}{\alpha}}\S{\ve m}{\ve b|\frac{f}{\alpha};n}-\H{0}{\frac{y_i}{\alpha}} \S{\ve m}{\ve b|\frac{y_i}{\alpha};n} }  {y_i-f} dy_i\\
	   &=&\H{1,0}{\frac{t_{i+1}}{f}|\frac{s_{i+1}}{f}}\S{\ve m}{\ve b|\frac{f}{\alpha};n}\\
 	     &&-\H{0}{\frac{f}{\alpha}}\left(\S{\ve m,1}{\ve b,\frac{f}{\alpha}|\frac{t_{i+1}}{f};n}-\S{\ve m,1}{\ve b,\frac{f}{\alpha}|\frac{s_{i+1}}{f};n}\right)\\
 	     &&-\left(\H{0}{\frac{t_{i+1}}{f}}\S{\ve m,1}{\ve b,\frac{f}{\alpha}|\frac{t_{i+1}}{f};n}-\H{0}{\frac{s_{i+1}}{f}}\S{\ve m,1}{\ve b,\frac{f}{\alpha}|\frac{s_{i+1}}{f};n}\right)\\
 	     &&+\left(\S{\ve m,2}{\ve b,\frac{f}{\alpha}|\frac{t_{i+1}}{f};n}-\S{\ve m,2}{\ve b,\frac{f}{\alpha}|\frac{s_{i+1}}{f};n}\right)\\
	  B&=&\int_{s_{i+1}}^{t_{i+1}}\frac{\S{\ve w}{\ve v|\frac{f}{\beta};n}-\S{\ve w}{\ve v|\frac{y_i}{\beta};n}} {y_i-f}dy_i\\
	   &=&-\S{\ve w,1}{\ve v,\frac{f}{\beta}|\frac{t_{i+1}}{\beta};n}+\S{\ve w,1}{\ve v,\frac{f}{\beta}|\frac{s_{i+1}}{\beta};n}.
	\end{eqnarray*}
	Hence in all of the three possible cases, \ie $(s_{i+1},t_{i+1})$ equals $(e,y_{i+1})$, $(y_{i+1},f)$ or $(e,f),$ we arrive at the desired results.
 \item $(s_i,t_i)=(e,y_i);g_i=e:$\\
	This case can be handled similarly to the previous one.
\end{itemize}

Combining all these cases we see that we can always express an integral of type (d) in the form given in (\ref{SSintdiffform}).
Let us now look at the inner integral of \textbf{(e)}. Since we have
\begin{eqnarray*}
 \int_{s_2}^{t_2}\frac{(-y_1)^n-1}{y_1+1}dy_1=\H{0}{t_2}\S{1}{-t_2;n}-\H{0}{s_2}\S{1}{-s_2;n}-\S{2}{-t_2;n}+\S{2}{-s_2;n}
\end{eqnarray*}
we can use basically the same arguments for integrals of type (e) as we used for integrals of type (d) to see that we can express them in the same form, \ie in the form 
(\ref{SSintdiffform}). Since we are able to express all the integrals in the desired forms, we are able to transform the differentiated integral representation of an S-sum 
$\S{\ve m}{\ve b; n}$ back to expressions involving S-sums at $n$ and S-sums at infinity (or equivalently multiple polylogarithms at constants) and hence we can differentiate 
arbitrary S-sums with respect to the upper summation limit.

\section{Relations between S-Sums}
\subsection{Algebraic Relations}
\label{SSalgrel}
We already mentioned in the beginning of this chapter that S-sums form a quasi shuffle algebra. As showed in \cite{Moch2002}, ideas from Section \ref{HSalgrel} can be carried over to S-sums if we
consider an alphabet $A$, where pairs $(m,x)$ with $m\in\N$ and $x\in\R^*$ form the letters, \ie we identify a S-sum 
$$
\S{m_1,m_2,\ldots,m_k}{x_1,x_2,\ldots,x_k;n}
$$
with the word $(m_1,x_1)(m_2,x_2)\ldots(m_k,x_k)$.
We define the degree of a letter $(m,x) \in A$ as $\abs{(m,x)}=m$ and we order the letters for  $m_1,m_2\in\N$ and $x_1,x_2\in\R^*$ by
\begin{eqnarray*}
\begin{array}{llll} 
	(m_1,x_1)	&\prec (m_2,x_2) 	&\textnormal{if }&m_1<m_2\\
	(m_1,x_1)	&\prec (m_1,x_2) 	&\textnormal{if }&\abs{x_1}<\abs{x_1}\\
	(m_1,-x_1)	&\prec (m_1,x_1) 	&\textnormal{if }&x_1>0.
\end{array}
\end{eqnarray*}
We extend this order lexicographically to words. Using this order, it can be showed analogously as in \cite{Ablinger2009} (compare \cite{Hoffman}) that the S-sums form 
a quasi shuffle algebra which is the free polynomial algebra on the 
\textit{Lyndon} words with alphabet $A.$
Hence the number of algebraic independent sums in this algebra can be counted by counting the number of \textit{Lyndon} words.
If we consider for example an alphabet with $n$ letters and we lock for the number of basis sums with depth $d,$  
we can use the first Witt formula~(\ref{HSWitt1}):
  $$N_A(d) = \frac{1}{d}\sum_{q|d}{\mu\left(\frac{d}{q}\right)n^q}.$$
In the subsequent example we look at S-sums of depth 2 on alphabets of 4 letters. Hence we obtain
$$ \frac{1}{2}\sum_{q|2}{\mu\left(\frac{2}{q}\right)4^q}=6$$
basis sums.

We can use an analogous method of the method presented in \cite{Bluemlein2004,Ablinger2009} for harmonic sums to find the basis S-sums together with the relations for the 
dependent S-sums. 
\begin{example} We consider the letters $(1,\frac{1}{2}),(1,-\frac{1}{2}),(3,\frac{1}{2}),(3,-\frac{1}{2}).$ At depth $d=2$ we obtain $16$ sums with these letters.
Using the relations
\small
\begin{eqnarray*}
\textnormal{S}_{1,3}\left(\frac{1}{2},\frac{1}{2};n\right)&=& -\textnormal{S}_{3,1}\left(\frac{1}{2},\frac{1}{2};n\right)+\textnormal{S}_1\left(\frac{1}{2};n\right)\textnormal{S}_3\left(\frac{1}{2};n\right)+\textnormal{S}_4\left(\frac{1}{4};n\right)\\
\textnormal{S}_{1,3}\left(-\frac{1}{2},-\frac{1}{2};n\right)&=&-\textnormal{S}_{3,1}\left(-\frac{1}{2},-\frac{1}{2};n\right)+\textnormal{S}_1\left(-\frac{1}{2};n\right)\textnormal{S}_3\left(-\frac{1}{2};n\right)+\textnormal{S}_4\left(\frac{1}{4};n\right)\\
\textnormal{S}_{1,3}\left(\frac{1}{2},-\frac{1}{2};n\right)&=&-\textnormal{S}_{3,1}\left(-\frac{1}{2},\frac{1}{2};n\right)+\textnormal{S}_1\left(\frac{1}{2};n\right)\textnormal{S}_3\left(-\frac{1}{2};n\right)+\textnormal{S}_4\left(-\frac{1}{4};n\right)\\
\textnormal{S}_{1,3}\left(-\frac{1}{2},\frac{1}{2};n\right)&=&-\textnormal{S}_{3,1}\left(\frac{1}{2},-\frac{1}{2};n\right)+\textnormal{S}_1\left(-\frac{1}{2};n\right)\textnormal{S}_3\left(\frac{1}{2};n\right)+\textnormal{S}_4\left(-\frac{1}{4};n\right)\\
\textnormal{S}_{1,1}\left(-\frac{1}{2},-\frac{1}{2};n\right)&=& \frac{1}{2} \textnormal{S}_1\left(-\frac{1}{2};n\right)^2+\frac{1}{2}\textnormal{S}_2\left(\frac{1}{4};n\right)\\
\textnormal{S}_{1,1}\left(\frac{1}{2},\frac{1}{2};n\right)&=& \frac{1}{2} \textnormal{S}_1\left(\frac{1}{2};n\right)^2+\frac{1}{2}\textnormal{S}_2\left(\frac{1}{4};n\right)\\
\textnormal{S}_{1,1}\left(\frac{1}{2},-\frac{1}{2};n\right)&=&-\textnormal{S}_{1,1}\left(-\frac{1}{2},\frac{1}{2};n\right)+\textnormal{S}_1\left(-\frac{1}{2};n\right)\textnormal{S}_1\left(\frac{1}{2};n\right)+\textnormal{S}_2\left(-\frac{1}{4};n\right)\\
\textnormal{S}_{3,3}\left(-\frac{1}{2},-\frac{1}{2};n\right)&=& \frac{1}{2}\textnormal{S}_3\left(-\frac{1}{2};n\right)^2+\frac{1}{2} \textnormal{S}_6\left(\frac{1}{4};n\right)\\
\textnormal{S}_{3,3}\left(\frac{1}{2},\frac{1}{2};n\right)&=& \frac{1}{2}\textnormal{S}_3\left(\frac{1}{2};n\right)^2+\frac{1}{2} \textnormal{S}_6\left(\frac{1}{4};n\right)\\
\textnormal{S}_{3,3}\left(\frac{1}{2},-\frac{1}{2};n\right)&=&-\textnormal{S}_{3,3}\left(-\frac{1}{2},\frac{1}{2};n\right)+\textnormal{S}_3\left(-\frac{1}{2};n\right)\textnormal{S}_3\left(\frac{1}{2};n\right)+\textnormal{S}_6\left(-\frac{1}{4};n\right)
\end{eqnarray*}
\normalsize
we find the $6$ basis sums
\small
\begin{eqnarray*}
&&\textnormal{S}_{3,1}\left(\frac{1}{2},\frac{1}{2};n\right),\textnormal{S}_{3,1}\left(-\frac{1}{2},\frac{1}{2};n\right),\textnormal{S}_{3,1}\left(\frac{1}{2},-\frac{1}{2};n\right),\textnormal{S}_{3,1}\left(-\frac{1}{2},-\frac{1}{2};n\right),\\
&&\textnormal{S}_{1,1}\left(-\frac{1}{2},\frac{1}{2};n\right),\textnormal{S}_{3,3}\left(-\frac{1}{2},\frac{1}{2};n\right),
\end{eqnarray*}
\normalsize
in which all the other 10 sums of depth 2 can be expressed.
\label{SSRelationEx1}
\end{example}

\subsection{Differential Relations}
\label{SSdiffrel}
In Section \ref{SSdifferentiation} we described the differentiation of S-sums with respect to the upper summation limit. 
The differentiation leads to new relations for instance we find
$$
\frac{d}{d n}\S{2}{2;n}=-\S{3}{2;n}+\H{0}2\S{2}{2;n}+\H{0,0,-1}1+2\H{0,0,1}1+\H{0,1,-1}1.
$$
As for harmonic sums we collect the derivatives with respect to $n:$
%------------------------------------------------------------------------------------------------------------
\begin{eqnarray}
\S{a_1,\ldots,a_k}{b_1,\ldots,b_k;n}^{(D)}
= \left\{\frac{\partial^N}{\partial n^N}\S{a_1,\ldots,a_k}{b_1,\ldots,b_k;n};
N \in \N\right\}.
\end{eqnarray}
Continuing the Example \ref{SSRelationEx1} we get
\begin{example}[Example \ref{SSRelationEx1} continued]From differentiation we get the additional relations:
\small
\begin{eqnarray*}
\textnormal{S}_{3,1}\hspace{-0.2em}\left(\frac{1}{2},\frac{1}{2};n\right)&=& \frac{1}{12} \frac{\partial}{\partial n}\textnormal{S}_3\hspace{-0.2em}\left(\frac{1}{4};n\right)-\frac{1}{2}
   \frac{\partial}{\partial n}\textnormal{S}_{2,1}\hspace{-0.2em}\left(\frac{1}{2},\frac{1}{2};n\right)-\frac{1}{2} \textnormal{H}_{1,0}\hspace{-0.2em}\left(\frac{1}{2}\right)
   \textnormal{S}_2\hspace{-0.2em}\left(\frac{1}{2};n\right)\\&&+\textnormal{H}_0\hspace{-0.2em}\left(\frac{1}{2}\right) \textnormal{S}_{2,1}\hspace{-0.2em}\left(\frac{1}{2},\frac{1}{2};n\right)-\frac{1}{2}
   \textnormal{H}_{\frac{1}{2}}\hspace{-0.2em}\left(\frac{1}{4}\right) \textnormal{H}_{0,1,0}\hspace{-0.2em}\left(\frac{1}{2}\right)+\frac{1}{12} \textnormal{H}_{0,0,1,0}\hspace{-0.2em}\left(\frac{1}{4}\right)\\&&+\frac{1}{2}
   \textnormal{H}_{\frac{1}{2},0,1,0}\hspace{-0.2em}\left(\frac{1}{4}\right)-\frac{1}{12} \textnormal{H}_0\hspace{-0.2em}\left(\frac{1}{4}\right) \textnormal{S}_3\hspace{-0.2em}\left(\frac{1}{4};n\right)-\frac{1}{4}
   \textnormal{S}_2\hspace{-0.2em}\left(\frac{1}{2};n\right)^2\\
\textnormal{S}_{3,1}\hspace{-0.2em}\left(-\frac{1}{2},-\frac{1}{2};n\right)&=& \frac{1}{12}
   \frac{\partial}{\partial n}\textnormal{S}_3\hspace{-0.2em}\left(\frac{1}{4};n\right)-\frac{1}{2}
   \frac{\partial}{\partial n}\textnormal{S}_{2,1}\hspace{-0.2em}\left(-\frac{1}{2},-\frac{1}{2};n\right)+\frac{1}{2} \textnormal{H}_{-\frac{1}{2},0}\hspace{-0.2em}\left(\frac{1}{4}\right)
   \textnormal{S}_2\hspace{-0.2em}\left(-\frac{1}{2};n\right)\\&&+\frac{1}{2} \textnormal{H}_0\hspace{-0.2em}\left(\frac{1}{4}\right) \textnormal{S}_{2,1}\hspace{-0.2em}\left(-\frac{1}{2},-\frac{1}{2};n\right)-\frac{1}{2}
   \textnormal{H}_{-\frac{1}{2}}\hspace{-0.2em}\left(\frac{1}{4}\right) \textnormal{H}_{0,-1,0}\hspace{-0.2em}\left(\frac{1}{2}\right)\\&&-\frac{1}{2}
   \textnormal{H}_{-\frac{1}{2},0,1,0}\hspace{-0.2em}\left(\frac{1}{4}\right)+\frac{1}{12} \textnormal{H}_{0,0,1,0}\hspace{-0.2em}\left(\frac{1}{4}\right)-\frac{1}{2}
   \textnormal{H}_{-\frac{1}{2}}\hspace{-0.2em}\left(\frac{1}{4}\right) \textnormal{H}_0\hspace{-0.2em}\left(\frac{1}{2}\right) \textnormal{S}_2\hspace{-0.2em}\left(-\frac{1}{2};n\right)\\&&-\frac{1}{12} \textnormal{H}_0\hspace{-0.2em}\left(\frac{1}{4}\right)
   \textnormal{S}_3\hspace{-0.2em}\left(\frac{1}{4};n\right)-\frac{1}{4} \textnormal{S}_2\hspace{-0.2em}\left(-\frac{1}{2};n\right)^2\\
\textnormal{S}_{3,1}\hspace{-0.2em}\left(-\frac{1}{2},\frac{1}{2};n\right)&=& \frac{1}{6}
   \frac{\partial}{\partial n}\textnormal{S}_3\hspace{-0.2em}\left(-\frac{1}{4};n\right)-\frac{1}{2}
   \frac{\partial}{\partial n}\textnormal{S}_{2,1}\hspace{-0.2em}\left(-\frac{1}{2},\frac{1}{2};n\right)-\frac{1}{2}
   \frac{\partial}{\partial n}\textnormal{S}_{2,1}\hspace{-0.2em}\left(\frac{1}{2},-\frac{1}{2};n\right)\\&&-\frac{1}{2} \textnormal{H}_{1,0}\hspace{-0.2em}\left(\frac{1}{2}\right)
   \textnormal{S}_2\hspace{-0.2em}\left(-\frac{1}{2};n\right)+\frac{1}{2} \textnormal{H}_{-\frac{1}{2},0}\hspace{-0.2em}\left(\frac{1}{4}\right)
   \textnormal{S}_2\hspace{-0.2em}\left(\frac{1}{2};n\right)\\&&+\textnormal{H}_0\hspace{-0.2em}\left(\frac{1}{2}\right) \textnormal{S}_{2,1}\hspace{-0.2em}\left(-\frac{1}{2},\frac{1}{2};n\right)+\frac{1}{2}
   \textnormal{H}_0\hspace{-0.2em}\left(\frac{1}{4}\right) \textnormal{S}_{2,1}\hspace{-0.2em}\left(\frac{1}{2},-\frac{1}{2};n\right)\\&&-\textnormal{S}_{3,1}\hspace{-0.2em}\left(\frac{1}{2},-\frac{1}{2};n\right)+\frac{1}{2}
   \textnormal{H}_{\frac{1}{2}}\hspace{-0.2em}\left(\frac{1}{4}\right) \textnormal{H}_{0,-1,0}\hspace{-0.2em}\left(\frac{1}{2}\right)+\frac{1}{2} \textnormal{H}_{\frac{3}{4}}\hspace{-0.2em}\left(\frac{1}{4}\right)
   \textnormal{H}_{0,1,0}\hspace{-0.2em}\left(\frac{1}{2}\right)\\&&+\frac{1}{2} \textnormal{H}_{-\frac{1}{2},0,-1,0}\hspace{-0.2em}\left(\frac{1}{4}\right)-\frac{1}{6}
   \textnormal{H}_{0,0,-1,0}\hspace{-0.2em}\left(\frac{1}{4}\right)-\frac{1}{2} \textnormal{H}_{\frac{1}{2},0,-1,0}\hspace{-0.2em}\left(\frac{1}{4}\right)\\&&-\frac{1}{2} \textnormal{H}_0\hspace{-0.2em}\left(\frac{1}{2}\right)
   \textnormal{H}_{\frac{3}{4}}\hspace{-0.2em}\left(\frac{1}{4}\right) \textnormal{S}_2\hspace{-0.2em}\left(\frac{1}{2};n\right)-\frac{1}{6} \textnormal{H}_0\hspace{-0.2em}\left(\frac{1}{4}\right)
   \textnormal{S}_3\hspace{-0.2em}\left(-\frac{1}{4};n\right)\\&&-\frac{1}{2} \textnormal{S}_2\hspace{-0.2em}\left(-\frac{1}{2};n\right)
   \textnormal{S}_2\hspace{-0.2em}\left(\frac{1}{2};n\right).
\end{eqnarray*}
\normalsize
Using all the relations we can reduce the number of basis sums at depth $d=2$ to $3$ by introducing differentiation. In our case the basis sum are:
\small
\begin{eqnarray*}
\textnormal{S}_{3,1}\hspace{-0.2em}\left(\frac{1}{2},-\frac{1}{2};n\right),\textnormal{S}_{1,1}\hspace{-0.2em}\left(-\frac{1}{2},\frac{1}{2};n\right),\textnormal{S}_{3,3}\hspace{-0.2em}\left(-\frac{1}{2},\frac{1}{2};n\right)
\end{eqnarray*}
\normalsize
Note that we introduced the letters $(2,\pm \frac{1}{2})$ and $(3,\pm \frac{1}{4}),$ however these letters appear just in sums of depth 2 with weight $3$ and are only used to express sums of weight $4.$ 
\label{SSRelationEx2}
\end{example}

\subsection{Duplication Relations}
\label{SSduplrel}
As for harmonic sums we have a duplication relation:
\begin{thm}[Duplication Relation]
 For $a_i \in \N$, $b_i \in \R^+$ and $n\in \N$ we have
$$
\sum{\S{a_m,\ldots,a_1}{\pm b_m,\ldots,\pm b_1;2\;n}}=\frac{1}{2^{\sum_{i=1}^m{a_i}-m}}\S{a_m,\ldots,a_1}{b_m^2,\ldots,b_1^2;n}.
$$ 
where we sum on the left hand side over the $2^m$ possible combinations concerning $\pm$.
\end{thm}
\begin{proof}
We proceed by induction on $m.$ For $m=1$ we get
\begin{eqnarray*}
 \S{a}{-b,2n}
	&=&\sum_{i=1}^{2n}{\frac{(-b)^i}{i^a}}=\sum_{i=1}^{n}{\left(\frac{b^{2i}}{(2i)^a}-\frac{b^{2i-1}}{(2i-1)^a}\right)} 
		=\frac{1}{2^a}\S{a}{b^2;n}-\sum_{i=1}^n{\frac{b^{2i-1}}{(2i-1)^a}}\\
	&=&\frac{1}{2^a}\S{a}{b^2;n}-\sum_{i=1}^n{\left(\frac{b^{2i}}{(2i)^a}+\frac{b^{2i-1}}{(2i-1)^a}-\frac{b^{2i}}{(2i)^a}\right)}\\ 
	&=&\frac{1}{2^a}\S{a}{b^2;n}-\S{a}(b;2n)+\frac{1}{2^a}\S{a}{b^2;n}\\
	&=&\frac{1}{2^{a-1}}\S{a}{b^2,n}-\S{a}{b;2n}.
\end{eqnarray*}
Now suppose the theorem holds for $m:$
\begin{eqnarray*}
&&\sum{\S{a_{m+1}, \ldots,a_{1}}{\pm b_{m+1},\ldots, \pm b_1;2n}}=\\
	&&\hspace{2cm}= \sum_{i=1}^n\biggl(\frac{(\pm b_{m+1}^{2i})}{(2i)^{a_{m+1}}}\sum{\S{a_m,\ldots,a_1}{\pm b_m,\ldots,\pm b_1;2\;i}}\biggr.\\
		&&\hspace{2.5cm}+\biggl.\frac{\pm b_{m+1}^{2i-1}}{^(2i-1)^{a_{m+1}}}\sum{\S{a_m,\ldots,a_1}{\pm b_m,\ldots,\pm b_1;2\;i+1}}\biggr)\\
	&&\hspace{2cm}=2\sum_{i=1}^n\frac{(b^2)^i}{(2i)^{a_{m+1}}}\sum{\S{a_m,\ldots,a_1}{\pm b_m,\ldots,\pm b_1;2\;i}}\\
	&&\hspace{2cm}=2\sum_{i=1}^n\frac{(b^2)^i}{(2i)^{a_{m+1}}}\frac{1}{2^{\sum_{i=1}^m a_i-m}}{\S{a_m,\ldots,a_1}{b_m^2,\ldots,b_1^2;i}}\\
	&&\hspace{2cm}=\frac{1}{2^{\sum_{i=1}^{m+1} a_i-(m+1)}}\S{a_{m+1},a_m,\ldots,a_1}{\pm b_{m+1},\pm b_m,\ldots,\pm b_1;n}.
\end{eqnarray*}
\end{proof}

\begin{example}[Example \ref{SSRelationEx2} continued]The duplication relations are 
\small
\begin{eqnarray*}
\textnormal{S}_{3,3}\hspace{-0.2em}\left(\frac{1}{2},\frac{1}{2};2 n\right)&=& -\textnormal{S}_{3,3}\hspace{-0.2em}\left(-\frac{1}{2},-\frac{1}{2};2n\right)-\textnormal{S}_{3,3}\hspace{-0.2em}\left(-\frac{1}{2},\frac{1}{2};2 n\right)\\&&+\frac{1}{16}\textnormal{S}_{3,3}\hspace{-0.2em}\left(\frac{1}{4},\frac{1}{4};n\right)-\textnormal{S}_{3,3}\hspace{-0.2em}\left(\frac{1}{2},-\frac{1}{2};2 n\right)\\
\textnormal{S}_{1,1}\hspace{-0.2em}\left(\frac{1}{2},\frac{1}{2};2n\right)&=& -\textnormal{S}_{1,1}\hspace{-0.2em}\left(-\frac{1}{2},-\frac{1}{2};2 n\right)-\textnormal{S}_{1,1}\hspace{-0.2em}\left(-\frac{1}{2},\frac{1}{2};2n\right)\\&&+\textnormal{S}_{1,1}\left(\frac{1}{4},\frac{1}{4};n\right)-\textnormal{S}_{1,1}\left(\frac{1}{2},-\frac{1}{2};2n\right)\\
\textnormal{S}_{3,1}\hspace{-0.2em}\left(\frac{1}{2},\frac{1}{2};2 n\right)&=& -\textnormal{S}_{3,1}\hspace{-0.2em}\left(-\frac{1}{2},-\frac{1}{2};2n\right)-\textnormal{S}_{3,1}\hspace{-0.2em}\left(-\frac{1}{2},\frac{1}{2};2 n\right)\\&&+\frac{1}{4}\textnormal{S}_{3,1}\hspace{-0.2em}\left(\frac{1}{4},\frac{1}{4};n\right)-\textnormal{S}_{3,1}\hspace{-0.2em}\left(\frac{1}{2},-\frac{1}{2};2 n\right)\\
\textnormal{S}_{1,3}\hspace{-0.2em}\left(\frac{1}{2},\frac{1}{2};2n\right)&=& -\textnormal{S}_{1,3}\hspace{-0.2em}\left(-\frac{1}{2},-\frac{1}{2};2 n\right)-\textnormal{S}_{1,3}\hspace{-0.2em}\left(-\frac{1}{2},\frac{1}{2};2 n\right)\\&&+\frac{1}{4}\textnormal{S}_{1,3}\hspace{-0.2em}\left(\frac{1}{4},\frac{1}{4};n\right)-\textnormal{S}_{1,3}\hspace{-0.2em}\left(\frac{1}{2},-\frac{1}{2};2 n\right).
\end{eqnarray*}
\normalsize
Note that from duplication we do not get a further reduction in this case. We would have to introduce new sums of the same depth and weight to express the basis sums form Example \ref{SSRelationEx2}.
\label{SSRelationEx3}
\end{example}

\subsection{Examples for Specific Index Sets}
In this subsection we want to present the numbers of basis sums for specific index sets at special depths or weights. In the Tables \ref{SSdepth2table}, \ref{SSdepth3table} and \ref{SSdepth4table} 
we summarize the number of algebraic basis sums at the possible index sets at depths $2,3,4$ respectively. To illustrate how these tables are to be understood we take a closer look at the depth $d=3$ example with the 
index set $\{(a_1,a_1,a_2),(x_1,x_2,x_3)\};$ here, \eg $\{(a_1,a_1,a_2),(x_1,x_2,x_3)\}$ stands for a multiset and each element is allowed to be taken once to build all the possible S-sums. 
There we obtain the $18$ sums:
\small
\begin{eqnarray*}
&& \S{a_1, a_1, a_2}{x_1, x_2, x_3;n}, \S{a_1, a_1, a_2}{x_1, x_3, x_2;n}, \S{a_1, a_1, a_2}{x_2, x_1, x_3;n}, \\
&& \S{a_1, a_1, a_2}{x_2, x_3, x_1;n}, \S{a_1, a_1, a_2}{x_3, x_1, x_2;n}, \S{a_1, a_1, a_2}{x_3, x_2, x_1;n}, \\
&& \S{a_1, a_2, a_1}{x_1, x_2, x_3;n}, \S{a_1, a_2, a_1}{x_1, x_3, x_2;n}, \S{a_1, a_2, a_1}{x_2, x_1, x_3;n}, \\ 
&& \S{a_1, a_2, a_1}{x_2, x_3, x_1;n}, \S{a_1, a_2, a_1}{x_3, x_1, x_2;n}, \S{a_1, a_2, a_1}{x_3, x_2, x_1;n}, \\
&& \S{a_2, a_1, a_1}{x_1, x_2, x_3;n}, \S{a_2, a_1, a_1}{x_1, x_3, x_2;n}, \S{a_2, a_1, a_1}{x_2, x_1, x_3;n}, \\
&& \S{a_2, a_1, a_1}{x_2, x_3, x_1;n}, \S{a_2, a_1, a_1}{x_3, x_1, x_2;n}, \S{a_2, a_1, a_1}{x_3, x_2, x_1;n}.
\end{eqnarray*}
\normalsize
Using relations of the form 
\small
\begin{eqnarray*}
&&\S{a_1, a_1, a_2}{x_1, x_2, x_3;n} = \S{a_2}{x_3;n} \S{a_1, a_1}{x_1, x_2;n} + \S{a_1, a_1 + a_2}{x_1, x_2 x_3;n}\\
				&&\hspace{1cm} - \S{a_1}{x_1;n} \S{a_2, a_1}{x_3, x_2;n} - \S{a_2, 2 a_1}{x_3, x_1 x_2;n} + \S{a_2, a_1, a_1}{x_3, x_2, x_1;n}
\end{eqnarray*}
\normalsize
we can express all these $18$ sums using the following $6$ basis sums (with the price of introducing additional S-sums of lower depth subject to the relations given by the 
quasi shuffle algebra):
\small
\begin{eqnarray*}
&& \S{a_2, a_1, a_1}{x_1, x_2, x_3;n}, \S{a_2, a_1, a_1}{x_1, x_3, x_2;n}, \S{a_2, a_1, a_1}{x_2, x_1, x_3;n},\\
&& \S{a_2, a_1, a_1}{x_2, x_3, x_1;n}, \S{a_2, a_1, a_1}{x_3, x_1, x_2;n}, \S{a_2, a_1, a_1}{x_3, x_2, x_1;n}.
\end{eqnarray*}
\normalsize

\begin{table}\centering
\begin{tabular}{|c | r | r | r|}
\hline	
Index set& sums &basis sums & dependent sums\\
\hline	
  $\{a_1,a_1\},\{x_1,x_1\}$ &    1 &   0 &   1 \\
  $\{a_1,a_1\},\{x_1,x_2\}$ &    2 &   1 &   1 \\
\hline
  $\{a_1,a_2\},\{x_1,x_1\}$ &    2 &   1 &   1 \\
  $\{a_1,a_2\},\{x_1,x_2\}$ &    4 &   2 &   2 \\
\hline
\end{tabular}
\caption{\label{SSdepth2table}Number of basis sums for different index sets at depth 2.}
\end{table}

\begin{table}\centering
\begin{tabular}{|c | r | r | r|}
\hline	
Index set& sums &basis sums & dependent sums\\
\hline	
  $\{a_1,a_1,a_1\},\{x_1,x_1,x_1\}$ &    1 &   0 &   1  \\
  $\{a_1,a_1,a_1\},\{x_1,x_1,x_2\}$ &    3 &   1 &   2  \\
  $\{a_1,a_1,a_1\},\{x_1,x_2,x_3\}$ &    6 &   2 &   4  \\
\hline	
  $\{a_1,a_1,a_2\},\{x_1,x_1,x_1\}$ &    3 &   1 &   2  \\
  $\{a_1,a_1,a_2\},\{x_1,x_1,x_2\}$ &    9 &   3 &   6  \\
  $\{a_1,a_1,a_2\},\{x_1,x_2,x_3\}$ &   18 &   6 &  12  \\
\hline	
  $\{a_1,a_2,a_3\},\{x_1,x_1,x_1\}$ &    6 &   2 &   4  \\
  $\{a_1,a_2,a_3\},\{x_1,x_1,x_2\}$ &   18 &   6 &  12  \\
  $\{a_1,a_2,a_3\},\{x_1,x_2,x_3\}$ &   36 &  12 &  24  \\
\hline
\end{tabular}
\caption{\label{SSdepth3table}Number of basis sums for different index sets at depth 3.}
\end{table}

\begin{table}\centering
\begin{tabular}{|c | r | r | r|}
\hline	
Index set& sums &basis sums & dependent sums\\
\hline	
  $\{a_1,a_1,a_1,a_1\},\{x_1,x_1,x_1,x_1\}$ &    1 &   0 &   1  \\
  $\{a_1,a_1,a_1,a_1\},\{x_1,x_1,x_1,x_2\}$ &    4 &   1 &   3  \\
  $\{a_1,a_1,a_1,a_1\},\{x_1,x_1,x_2,x_2\}$ &    6 &   1 &   5  \\
  $\{a_1,a_1,a_1,a_1\},\{x_1,x_1,x_2,x_3\}$ &   12 &   3 &   9  \\
  $\{a_1,a_1,a_1,a_1\},\{x_1,x_2,x_3,x_4\}$ &   24 &   6 &  18  \\
\hline
  $\{a_1,a_1,a_1,a_2\},\{x_1,x_1,x_1,x_1\}$ &    4 &   1 &   3  \\
  $\{a_1,a_1,a_1,a_2\},\{x_1,x_1,x_1,x_2\}$ &   16 &   4 &  12  \\
  $\{a_1,a_1,a_1,a_2\},\{x_1,x_1,x_2,x_2\}$ &   24 &   6 &  18  \\
  $\{a_1,a_1,a_1,a_2\},\{x_1,x_1,x_2,x_3\}$ &   48 &  12 &  36  \\
  $\{a_1,a_1,a_1,a_2\},\{x_1,x_2,x_3,x_4\}$ &   96 &  24 &  72  \\
\hline	
  $\{a_1,a_1,a_2,a_2\},\{x_1,x_1,x_1,x_1\}$ &    6 &   1 &   5  \\
  $\{a_1,a_1,a_2,a_2\},\{x_1,x_1,x_1,x_2\}$ &   24 &   6 &  18  \\
  $\{a_1,a_1,a_2,a_2\},\{x_1,x_1,x_2,x_2\}$ &   36 &   8 &  28  \\
  $\{a_1,a_1,a_2,a_2\},\{x_1,x_1,x_2,x_3\}$ &   72 &  18 &  54  \\
  $\{a_1,a_1,a_2,a_2\},\{x_1,x_2,x_3,x_4\}$ &  144 &  36 & 108  \\
\hline	
  $\{a_1,a_1,a_2,a_3\},\{x_1,x_1,x_1,x_1\}$ &   12 &   3 &   9  \\
  $\{a_1,a_1,a_2,a_3\},\{x_1,x_1,x_1,x_2\}$ &   48 &  12 &  36  \\
  $\{a_1,a_1,a_2,a_3\},\{x_1,x_1,x_2,x_2\}$ &   72 &  18 &  54  \\
  $\{a_1,a_1,a_2,a_3\},\{x_1,x_1,x_2,x_3\}$ &  144 &  36 & 108  \\
  $\{a_1,a_1,a_2,a_3\},\{x_1,x_2,x_3,x_4\}$ &  288 &  72 & 216  \\
\hline	
  $\{a_1,a_2,a_3,a_4\},\{x_1,x_1,x_1,x_1\}$ &   24 &   6 &  18  \\
  $\{a_1,a_2,a_3,a_4\},\{x_1,x_1,x_1,x_2\}$ &   96 &  24 &  72  \\
  $\{a_1,a_2,a_3,a_4\},\{x_1,x_1,x_2,x_2\}$ &  144 &  36 & 108  \\
  $\{a_1,a_2,a_3,a_4\},\{x_1,x_1,x_2,x_3\}$ &  288 &  72 & 216  \\
  $\{a_1,a_2,a_3,a_4\},\{x_1,x_2,x_3,x_4\}$ &  576 & 144 & 432  \\
\hline	
\end{tabular}
\caption{\label{SSdepth4table}Number of basis sums for different index sets at depth 4.}
\end{table}

In Table \ref{SSweighttable1} we summarize the number of algebraic basis sums at specified weights for arbitrary indices in the $x_i$, while in Table \ref{SSweighttable2} we summarize 
the number of algebraic and differential bases sums for $x_i\in \{1,-1,1/2,-1/2,2,-2\}$ and where each of the indices $\{1/2,-1/2,2,-2\}$ is allowed to appear just once in each sum. 
To illustrate how these tables are to be understood we take a closer look at two examples.
At weight $w=3$ with index set $\{x_1,x_2,x_3\},$ we consider the $19$ sums:
\small
\begin{eqnarray*}
&& \S{3}{x_1 x_2 x_3;n}, \S{1, 2}{x_1, x_2 x_3;n}, \S{1, 2}{x_2, x_1 x_3;n}, \S{1, 2}{x_1 x_2, x_3;n}, \S{1, 2}{x_3, x_1 x_2;n},\\
&& \S{1, 2}{x_1 x_3, x_2;n}, \S{1, 2}{x_2 x_3, x_1;n}, \S{2, 1}{x_1, x_2 x_3;n}, \S{2, 1}{x_2, x_1 x_3;n}, \S{2, 1}{x_1 x_2, x_3;n},\\
&& \S{2, 1}{x_3, x_1 x_2;n}, \S{2, 1}{x_1 x_3, x_2;n}, \S{2, 1}{x_2 x_3, x_1;n}, \S{1, 1, 1}{x_1, x_2, x_3;n},  \S{1, 1, 1}{x_1, x_3, x_2;n},\\ 
&& \S{1, 1, 1}{x_2, x_1, x_3;n}, \S{1, 1, 1}{x_2, x_3, x_1;n}, \S{1, 1, 1}{x_3, x_1, x_2;n},  \S{1, 1, 1}{x_3, x_2, x_1;n}.
\end{eqnarray*}
\normalsize
Using relations of the form 
\small
\begin{eqnarray*}
\S{1, 2}{x_2 x_3, x_1;n} &=& \S{1}{x_2 x_3;n} \S{2}{x_1;n} + \S{3}{x_1 x_2 x_3;n} - \S{2, 1}{x_1, x_2 x_3;n}, \\
\S{1, 1, 1}{x_1, x_2, x_3;n} &=& -\S{1}{x_3;n} \S{1, 1}{x_2, x_1;n} + \S{1}{x_1;n} \S{1, 1}{x_2, x_3;n}\\
					&& + \S{2, 1}{x_1 x_2, x_3;n}- \S{2, 1}{x_2 x_3, x_1;n} + \S{1, 1, 1}{x_3, x_2, x_1;n}
\end{eqnarray*}
\normalsize
we can express all these $19$ sums using the following $9$ basis sums:
\small
\begin{eqnarray*}
&&\S{3}{x_1 x_2 x_3;n}, \S{2, 1}{x_1, x_2 x_3;n}, \S{2, 1}{x_2, x_1 x_3;n}, \S{2, 1}{x_1 x_2, x_3;n}, \S{2, 1}{x_3, x_1 x_2;n},\\
&&\S{2, 1}{x_1 x_3, x_2;n}, \S{2, 1}{x_2 x_3, x_1;n}, \S{1, 1, 1}{x_3, x_1, x_2;n}, \S{1, 1, 1}{x_3, x_2, x_1;n}.
\end{eqnarray*}
\normalsize

\begin{table}\centering
\begin{tabular}{|c | c | r | r | r|}
\hline	
weight&index set& sums &basis sums & dependent sums\\
\hline	
 2& $\{x_1,x_1\}$ &    2 &   1 &   1 \\
  & $\{x_1,x_2\}$ &    3 &   3 &   1 \\
\hline
 3& $\{x_1,x_1,x_1\}$ &    6 &   3 &   3 \\
  & $\{x_1,x_1,x_2\}$ &   12 &   6 &   6 \\
  & $\{x_1,x_2,x_3\}$ &   19 &   9 &  10 \\
\hline
 4& $\{x_1,x_1,x_1,x_1\}$ &   20 &   8 &  12 \\
  & $\{x_1,x_1,x_1,x_2\}$ &   50 &  20 &  30 \\
  & $\{x_1,x_1,x_2,x_2\}$ &   64 &  24 &  40 \\
  & $\{x_1,x_1,x_2,x_3\}$ &  106 &  40 &  66 \\
  & $\{x_1,x_2,x_3,x_4\}$ &  175 &  64 & 111 \\
\hline
 5& $\{x_1,x_1,x_1,x_1,x_1\}$ &   70 &  25 &   45 \\
  & $\{x_1,x_1,x_1,x_1,x_2\}$ &  210 &  70 &  140 \\
  & $\{x_1,x_1,x_1,x_2,x_2\}$ &  325 & 105 &  220 \\
  & $\{x_1,x_1,x_1,x_2,x_3\}$ &  555 & 175 &  380 \\
  & $\{x_1,x_1,x_2,x_2,x_3\}$ &  725 & 225 &  500 \\
  & $\{x_1,x_1,x_2,x_3,x_4\}$ & 1235 & 375 &  860 \\
  & $\{x_1,x_2,x_3,x_4,x_5\}$ & 2101 & 625 & 1476 \\
\hline
 6& $\{x_1,x_1,x_1,x_1,x_1,x_1\}$ &   252 &   75 &   177\\
  & $\{x_1,x_1,x_1,x_1,x_1,x_2\}$ &   882 &  252 &   630\\
  & $\{x_1,x_1,x_1,x_1,x_2,x_2\}$ &  1596 &  438 &  1158\\
  & $\{x_1,x_1,x_1,x_2,x_2,x_2\}$ &  1911 &  522 &  1389\\
  & $\{x_1,x_1,x_1,x_1,x_2,x_3\}$ &  2786 &  756 &  2030\\
  & $\{x_1,x_1,x_1,x_2,x_2,x_3\}$ &  4431 & 1176 &  3255\\
  & $\{x_1,x_1,x_2,x_2,x_3,x_3\}$ &  5886 & 1539 &  4347\\
  & $\{x_1,x_1,x_1,x_2,x_3,x_4\}$ &  7721 & 2016 &  5705\\
  & $\{x_1,x_1,x_2,x_2,x_3,x_4\}$ & 10251 & 2646 &  7605\\
  & $\{x_1,x_1,x_2,x_3,x_4,x_5\}$ & 17841 & 4536 & 13305 \\
  & $\{x_1,x_2,x_3,x_4,x_5,x_6\}$ & 31031 & 7776 & 23255\\
\hline
\end{tabular}
\caption{\label{SSweighttable1}Number of basis sums for different index sets up to weight 6.}
\end{table}

At weight $w=2$ with index set $\{1,-1,1/2,-1/2,2,-2\},$ where each of the indices $\{1/2,-1/2,2,-2\}$ is allowed to appear just once, we consider the $38$ sums:
\small
\begin{eqnarray*}
&&\S{-2}{n},\S{2}{n},\S{-1,-1}{n},\S{-1,1}{n},\S{1,-1}{n},\S{1,1}{n},\S{2}{-2;n},\S{2}{-\frac{1}{2};n},\\
&&\S{2}{\frac{1}{2};n},\S{2}{2;n},\S{1,1}{-2,-1;n},\S{1,1}{-2,-\frac{1}{2};n},\S{1,1}{-2,\frac{1}{2};n},\S{1,1}{-2,1;n},\\
&&\S{1,1}{-2,2;n},\S{1,1}{-1,-2;n},\S{1,1}{-1,-\frac{1}{2};n},\S{1,1}{-1,\frac{1}{2};n},\S{1,1}{-1,2;n},\\
&&\S{1,1}{-\frac{1}{2},-2;n},\S{1,1}{-\frac{1}{2},-1;n},\S{1,1}{-\frac{1}{2},\frac{1}{2};n},\S{1,1}{-\frac{1}{2},1;n},\S{1,1}{-\frac{1}{2},2;n},\\
&&\S{1,1}{\frac{1}{2},-2;n},\S{1,1}{\frac{1}{2},-1;n},\S{1,1}{\frac{1}{2},-\frac{1}{2};n},\S{1,1}{\frac{1}{2},1;n},\S{1,1}{\frac{1}{2},2;n},\\
&&\S{1,1}{1,-2;n},\S{1,1}{1,-\frac{1}{2};n},\S{1,1}{1,\frac{1}{2};n},\S{1,1}{1,2;n},\S{1,1}{2,-2;n},\S{1,1}{2,-1;n},\\
&&\S{1,1}{2,-\frac{1}{2};n},\S{1,1}{2,\frac{1}{2};n},\S{1,1}{2,1;n}.
\end{eqnarray*}
\normalsize
Using relations of the form 
\small
\begin{eqnarray*}
\S{1,1}{-\frac{1}{2},2;n}&=&-\frac{\partial}{\partial n}\S{-1}{n}+\H{-1,0}{1}+\S{1}{-\frac{1}{2};n}\;\S{1}{2;n}-\S{1,1}{2,-\frac{1}{2};n}\\
\S{1,1}{-\frac{1}{2},-2;n}&=&-\frac{\partial}{\partial n}\S{1}{n}-\H{1,0}{1}+\S{1}{-2;n}\;\S{1}{-\frac{1}{2};n}-\S{1,1}{-2,-\frac{1}{2};n}
\end{eqnarray*}
\normalsize
we can express all these $38$ sums using the following $17$ basis sums:
\small
\begin{eqnarray*}
&&\S{1,-1}{n},\S{1,1}{-2,-\frac{1}{2};n},\S{1,1}{-2,\frac{1}{2};n},\S{1,1}{-2,2;n},\S{1,1}{-1,-2;n},\\
&&\S{1,1}{-1,-\frac{1}{2};n},\S{1,1}{-1,\frac{1}{2};n},\S{1,1}{-1,2;n},\S{1,1}{-\frac{1}{2},\frac{1}{2};n},\S{1,1}{\frac{1}{2},-\frac{1}{2};n},\\
&&\S{1,1}{1,-2;n},\S{1,1}{1,-\frac{1}{2};n},\S{1,1}{1,\frac{1}{2};n},\S{1,1}{1,2;n},\S{1,1}{2,-2;n},\\
&&\S{1,1}{2,-\frac{1}{2};n},\S{1,1}{2,\frac{1}{2};n}.
\end{eqnarray*}
\normalsize
Note that we have
\begin{eqnarray*}
N_D(w)&=&N_S(w)-N_S(w-1)\\
N_{AD}(w)&=&N_A(w)-N_A(w-1)
\end{eqnarray*}
where $N_S(w),N_D(w)$ and $N_{AD}(w)$ are the number of sums, the number algebraic basis sums and the number of basis sums using algebraic and differential relations at weight $w$ respectively.

\begin{table}\centering
\begin{tabular}{|r | r | r | r | r|}
\hline	
weight& sums &a-basis sums & d-basis sums & ad-basis sums\\
\hline
 1&     6 &      6 &     6 &     6\\	
 2&    38 &     23 &    32 &    17\\
 3&   222 &    120 &   184 &    97\\
 4&  1206 &    654 &   984 &   543\\
 5&  6150 &   3536 &  4944 &  2882\\
 6& 29718 &  18280 & 23568 & 14744\\
\hline
\end{tabular}
\caption{\label{SSweighttable2}Number of basis sums up to weight 6 with index set $\{1,-1,1/2,-1/2,2,-2\}.$ Each of the indices $\{1/2,-1/2,2,-2\}$ is allowed to appear just once in each sum.}
\end{table}

\section{S-Sums at Infinity}

Of course not all S-sums are finite at infinity, since for example $\lim_{n\rightarrow \infty} \S{2}{2;n}$ does not exist.
In fact, we have the following theorem that extends Lemma \ref{HSconsumlem} which is used in the proof of the theorem:
\begin{thm}
Let $a_1, a_2, \ldots a_k \in \N$ and $x_1, x_2, \ldots x_k \in \R^*$ for $k \in \N.$
The S-sum $\S{a_1,a_2,\ldots,a_k}{x_1,x_2,\ldots,x_k;n}$ is absolutely convergent, when $n\rightarrow \infty$, if and only if one of the following conditions holds:
\begin{itemize}
 \item [1.] $\abs{x_1}<1 \wedge \abs{x_1 x_2}\leq 1 \wedge \ldots \wedge \abs{x_1 x_2 \cdots x_k}\leq 1$
 \item [2.] $a_1>1 \wedge \abs{x_1}=1 \wedge \abs{x_2}\leq 1 \wedge \ldots \wedge \abs{x_2 \cdots x_k}\leq 1.$
\end{itemize}
In addition the S-sum is conditional convergent (convergent but not absolutely convergent) if and only if
\begin{itemize}
 \item [3.] $a_1=1 \wedge x_1=-1 \wedge \abs{x_2}\leq 1 \wedge \ldots \wedge \abs{ x_2 \cdots x_k}\leq 1.$
\end{itemize}
\label{SSconsumthm}
\end{thm} 

\begin{proof}
First we want to show that the S-sum converges absolutely if either condition 1 or condition 2 holds. Therefore we proceed by induction on the depth $k.$ Here we prove the slightly more 
general case of S-sums, where we allow the last index $a_k$ to equal zero. Note that we extend the definition of S-sums in the obvious way. For $k=1$ the statement is obvious in both cases. 
Now let us look at the first condition and assume that the statement holds for depth $k-1:$
We distinguish two cases. First let $\abs{x_k}<1:$ 
\begin{eqnarray*}
\sum_{i=1}^n\abs{\frac{{x_1}^i}{i^{a_1}}\S{a_2,\ldots,a_k}{x_2,\ldots,x_k;i}} &\leq& \S{a_1,a_2,\ldots,a_k}{\abs{x_1},\abs{x_2},\ldots,\abs{x_k};n}\\
&\leq&  \S{a_1,a_2,\ldots,a_{k-1}-1}{\abs{x_1},\abs{x_2},\ldots,\abs{x_{k-1}};n}.
\end{eqnarray*}
If $a_{k-1}-1=0$ we extend the definition of S-sums in the obvious way. 
The last sum is of depth $k-1$ and it converges absolutely since $\abs{x_1}<1 \wedge \abs{x_1 x_2}\leq 1 \wedge \ldots \wedge \abs{x_1 x_2 \cdots x_{k-1}}\leq 1$ and due to the 
induction hypothesis. Thus $\S{a_1,a_2,\ldots,a_k}{x_1,x_2,\ldots,x_k;n}$ converges absolutely since the partial sums are bounded.\\
For $\abs{x_k}>1$ we get:
\begin{eqnarray*}
\sum_{i=1}^n\abs{\frac{{x_1}^i}{i^{a_1}}\S{a_2,\ldots,a_k}{x_2,\ldots,x_k;i}} &\leq& \S{a_1,a_2,\ldots,a_k}{\abs{x_1},\abs{x_2},\ldots,\abs{x_k};n}\\
&\leq&  \S{a_1,a_2,\ldots,a_{k-1}-1}{\abs{x_1},\abs{x_2},\ldots,\abs{x_{k-2}},\abs{x_{k-1}x_k};n}.
\end{eqnarray*}
The last sum is of depth $k-1$ and it converges absolutely since $\abs{x_1}<1 \wedge \abs{x_1 x_2}\leq 1 \wedge \ldots \wedge \abs{x_1 x_2 \cdots x_{k-2}}\leq~1 \wedge \abs{x_1 x_2 \cdots x_{k-1}x_k}\leq 1$ and due to
the induction hypothesis. Hence $\S{a_1,a_2,\ldots,a_k}{x_1,x_2,\ldots,x_k;n}$ converges absolutely.

Let us look at the second condition. For the special case $k=2$ and $a_1>2$ we get:
\begin{eqnarray*}
\sum_{i=1}^n\abs{\frac{1}{i^{a_1}}\S{a_2}{x_2;i}} &\leq& \S{a_1,a_2}{1,\abs{x_2};n} \leq \S{a_1-1}{1;n}=\S{a_1-1}{n}.
\end{eqnarray*}
$\SS{a_1,a_2}{1,x_2}{n}$ converges absolutely since $\S{a_1-1}{n}$ converges (base case).
For $k=2$ and $a_1=2$ we get
\begin{eqnarray*}
\sum_{i=1}^n\abs{\frac{1}{i^{2}}\S{a_2}{x_2;i}} &\leq& \S{2,a_2}{1,1;n}=\S{2,a_2}{n}.
\end{eqnarray*}
$\abs{\S{2,a_2}{1,x_2;n}}$ converges absolutely since $\S{2,a_2}{n}$ converges (see Lemma \ref{HSconsumlem}).
Assume now that the second condition holds for depth $k-1$ with $k\geq 3$:
We distinguish two cases. First let $\abs{x_k}<1:$ 
\begin{eqnarray*}
\sum_{i=1}^n\abs{\frac{{x_1}^i}{i^{a_1}}\S{a_2,\ldots,a_k}{x_2,\ldots,x_k;i}} &\leq& \S{a_1,a_2,\ldots,a_k}{1,\abs{x_2},\ldots,\abs{x_k};n}\\
&\leq&  \S{a_1,a_2,\ldots,a_{k-1}-1}{1,\abs{x_2},\ldots,\abs{x_{k-1}};n}.
\end{eqnarray*}
The last sum is of depth $k-1$ and it converges absolutely since $a_1>1 \wedge \abs{ x_2}\leq 1 \wedge \ldots \wedge \abs{x_2 \cdots x_{k-1}}\leq 1$ and due to the induction hypothesis. 
Therefore $\S{a_1,a_2,\ldots,a_k}{x_1,x_2,\ldots,x_k;n}$ converges absolutely.\\
For $\abs{x_k}>1$ we get:
\begin{eqnarray*}
\sum_{i=1}^n\abs{\frac{{x_1}^i}{i^{a_1}}\S{a_2,\ldots,a_k}{x_2,\ldots,x_k;i}} &\leq& \S{a_1,a_2,\ldots,a_k}{1,\abs{x_2},\ldots,\abs{x_k};n}\\
&\leq&  \S{a_1,a_2,\ldots,a_{k-1}-1}{1,\abs{x_2},\ldots,\abs{x_{k-2}},\abs{x_{k-1}x_k};n}.
\end{eqnarray*}
The last sum is of depth $k-1$ and it converges absolutely since $a_1>1  \wedge \abs{ x_2}\leq 1 \wedge \ldots \wedge \abs{x_2 \cdots x_{k-2}}\leq~1 \wedge \abs{x_2 \cdots x_{k-1}x_k}\leq 1$ and
due to the induction hypothesis. Consequently $\S{a_1,a_2,\ldots,a_k}{x_1,x_2,\ldots,x_k;n}$ converges absolutely.

Let us now look at the third condition: Note, that it can be seen easily, that the considered sums are not absolutely convergent. To show that they are conditional convergent, we again proceed by induction on the depth $k.$ For $k=1$ the statement is obvious. We assume that it holds for depth $k$ and we consider
$$
\S{1,a_1,\ldots,a_{k}}{-1,x_1,\ldots,x_{k};n}=\sum_{i=1}^n\frac{(-1)^i}{i}\S{a_1,\ldots,a_{k}}{x_1,\ldots,x_{k};i}.
$$
If $\abs{x_1}<1$ or $\abs{x_1}=1$ and $a_1>1$ the sum $\S{a_1,\ldots,a_{k}}{x_1,\ldots,x_{k};i}$ is convergent for $i\rightarrow \infty$
 due to condition 1 respectively 2, while if $x_1=-1$ and $a_1=1$ the sum $\S{a_1,\ldots,a_{k}}{x_1,\ldots,x_{k};i}$ is convergent for $i\rightarrow \infty$ due to the induction hypothesis and hence 
the sum $\S{1,a_1,\ldots,a_{k}}{-1,x_1,\ldots,x_{k};n}$ converges in these cases for $n\rightarrow \infty.$ 
The last case we have to consider is $x_1=1$ and $a_1=1:$ if $a_i=1$ and $x_i=1$ for all $1\leq i\leq k$ then we are in fact considering the sum 
$\S{-1,1,\ldots,1}n$ which is convergent for $n\rightarrow \infty$ due to Lemma \ref{HSconsumlem}. Now suppose that not for all $i$ we have $a_i=1$ and $x_i=1.$\\ 
Let $d$ be the smallest index such that $a_d\neq1$ or $x_d\neq1,$ \ie we consider the case
\begin{eqnarray*}
&&\S{1,1,\ldots,1,a_d,\ldots,a_{k}}{-1,1,\ldots,1,x_d,\ldots,x_{k};n}=\\
&&\hspace{2cm}\sum_{i=0}^n\frac{(-1)^i}{i}\S{1,\ldots,1,a_d,\ldots,a_{k}}{1,\ldots,1,x_d,\ldots,x_{k};i}.
\end{eqnarray*}
Similar as for harmonic sums we can extract the leading ones of the S-sum $\S{1,\ldots,1,a_d,\ldots,a_{k}}{1,\ldots,1,x_d,\ldots,x_{k};i}$ and end up in a univariate polynomial in $\S1n$ with coefficients in S-sums which are convergent 
for $i\rightarrow \infty$ due to condition 1 or 2 or due to the induction hypothesis (compare Remark \ref{HSextractleading1rem} and Section \ref{HSinfval}). Hence we can write
$$
\S{1,\ldots,1,a_d,\ldots,a_{k}}{1,\ldots,1,x_d,\ldots,x_{k};i}=\sum_{j=0}^{d-1}\S1i^jc_j(i),
$$
where the $c_j(i)$ are convergent $i\rightarrow \infty.$ For $p\in\N,$ $\S1i^p$ diverges for $i\rightarrow \infty$, but the divergence is of order $\log(i)^p$ (see Lemma \ref{HSexpandS1}). 
Hence $\frac{\S1i^j}{i}c_j(i)$ tends to zero for $i\rightarrow \infty$ and thus 
$$
\S{1,1,\ldots,1,a_d,\ldots,a_{k}}{-1,1,\ldots,1,x_d,\ldots,x_{k};n}=\sum_{j=0}^{d-1}\sum_{i=1}^{n}(-1)^i\frac{\S1i^j}{i}c_j(i)
$$
converges for $n\rightarrow \infty.$ 

If non of the three conditions holds 
\ie $$(a_1=1 \wedge x_1=1)\vee \abs{x_1}>1 \vee \abs{x_1 x_2}> 1 \vee \ldots \vee \abs{x_1 x_2 \cdots x_k}> 1$$ 
we show that the S-sum diverges.
To see that the S-sum diverges if non of the three conditions holds we use the integral representation of Theorem~\ref{SSintrep}. If $a_1=1 \wedge x_1=1,$ the innermost integral $$\int_0^1\frac{y^n-1}{y-1}dy$$ 
in the integral representation of Theorem~\ref{SSintrep} diverges for $n\rightarrow\infty$ and hence the S-sum diverges.\\
If  $\abs{x_1}>1 \vee \abs{x_1 x_2}> 1 \vee \ldots \vee \abs{x_1 x_2 \cdots x_k}> 1$ then the absolute value of $z$ the upper integration limit of the innermost integral
$$\int_0^{z}\frac{y^n-1}{y-1}$$
exceeds $1$ and hence the absolute value of the integral tends to infinity for $n\rightarrow\infty$ and the S-sum diverges.  
\end{proof}

\subsection{Relations between S-Sums at Infinity}\label{SSInfRelations}
In this section we will state several types of relations between the values of S-sums at infinity which are of importance in the following chapters.
\begin{itemize}
\item The first type of relations originates from the algebraic relations of S-sums, see Section \ref{SSalgrel}. These relations remain valid when we consider 
 them at infinity. We will refer to these relations as the quasi shuffle or stuffle relations.
\item The duplication relations from Section \ref{SSduplrel} remain valid if we consider sums which are finite at infinity, since it makes 
 no difference whether the argument is $\infty$ or $2\cdot\infty.$
\item We can generalize the relation form \cite{Vermaseren1998} (see Section \ref{HSInfRelations}) for harmonic sums to S-sums. For convergent sums we have:
\begin{eqnarray*}
&&\S{m_1,\ldots,m_p}{x_1,\ldots,x_p;\infty}\S{k_1,\ldots,k_q}{y_1,\ldots,y_p;\infty}=\\
	&&\hspace{2cm}\lim_{n \rightarrow \infty}\sum_{i=1}^n\frac{y_1^i \S{m_1,\ldots,m_p}{x_1,\ldots,x_p;n-i}\S{k_2,\ldots,k_q}{y_2,\ldots,y_p;i}} {i^{k_1}}.
\end{eqnarray*}
Using
\begin{eqnarray*}
&&\lim_{n \rightarrow \infty}\sum_{i=1}^n\frac{y_1^i \S{m_1,\ldots,m_p}{x_1,\ldots,x_p;n-i}\S{k_2,\ldots,k_q}{y_2,\ldots,y_p;i}} {i^{k_1}}=\\
&&\hspace{1cm}\sum_{a=1}^{k_1} \binom{k_1+m_1-1-a}{m_1-1}\sum_{i=1}^n\frac{x_1^i}{i^{m_1+k_1-a}}\\
&&\hspace{2cm}\sum_{j=1}^i\frac{\left(\frac{y_1}{x_1}\right)^j\S{m_2,\ldots,m_p}{x_2,\ldots,x_p;i-j}\S{k_2,\ldots,k_q}{y_2,\ldots,y_p;j}} {j^a}\\
&&\hspace{1cm}+\sum_{a=1}^{m_1} \binom{k_1+m_1-1-a}{m_1-1}\sum_{i=1}^n\frac{y_1^i}{i^{m_1+k_1-a}}\\
&&\hspace{2cm}\sum_{j=1}^i\frac{\left(\frac{x_1}{y_1}\right)^j\S{k_2,\ldots,k_q}{y_2,\ldots,y_p;i-j}\S{m_2,\ldots,m_p}{x_2,\ldots,x_p;j}}{j^a}
\end{eqnarray*}
we can rewrite the right hand side in terms of S-sums. We will refer to these relations as the shuffle relations since one could also obtain them from the shuffle algebra 
of multiple polylogarithms.
\end{itemize}

\subsection{S-Sums of Roots of Unity at Infinity}
\label{SSRootsofUnity}

In this section we look at the values of infinite S-sums at roots of unity (compare \cite{Ablinger2011}). We define
\begin{eqnarray*}
\sigma_{k_1,...,k_m}(x_1, ..., x_m):=\lim_{n \rightarrow \infty} S_{k_1,...,k_m}(x_1, ..., x_m; n),
\end{eqnarray*}
with $S_{k_1,...,k_m}(x_1, ..., x_m; n)$ a S-sum, $k_1 \neq 1 $ for $x_1 = 1$ and $x_j$ is a $l-$th root of unity ($l\geq 1$), 
\ie $x_j\in \left \{e_l| e_l^l = 1, e_l \in \mathbb{C} \right\}.$

We seek the relations between the sums of $w = 1,2$. They can be expressed in terms
of polylogarithms by~: 
%--------------------------------------------------------------------------------
\begin{eqnarray}
\sigma_w(x) &=& \Li_w(x),~~~~~ w \in \N, w\geq 1 \\
\sigma_1(x) =\Li_1(x) &=& - \ln(1-x)\\ 
\sigma_{1,1}(x,y) &=& \Li_2(x) + \frac{1}{2} \ln^2(1-x)+ \Li_2\left(- \frac{x(1-y)}{1-x}\right) \label{SSrurel1}
\end{eqnarray}
%--------------------------------------------------------------------------------
and $^*$ denotes complex conjugation.
At weight $w=1$ we use the suitable relations from Section \ref{SSInfRelations} together with relations due to complex conjugation (indicated by $^*$)
\begin{eqnarray*}
 \sigma_{1}(e_j^k)^*=\sigma_{1}(e_j^{j-k})=\sigma_{1}(e_j^k)-\frac{j-2k}{j}\pi.
\end{eqnarray*}
In addition we extend Proposition 2.3 of Ref.~\cite{GON1}, where we consider S-sums
$\S{k_1,...,k_m}{x_1, ..., x_m; n}$ with  $n \in \N_0, k_i \in \N, 
x_i \in 
\mathbb{C}, |x_i| \leq 1$. 
%--------------------------------------------------------------------------------
Let $l \in \N$ and 
%--------------------------------------------------------------------------------
\begin{eqnarray}
\label{SSPOW}
y_i^{l} = x_i 
\end{eqnarray}
%--------------------------------------------------------------------------------
then \cite[(6.12)]{Ablinger2011}
%--------------------------------------------------------------------------------
\begin{eqnarray}
\label{SSDISTR}
\S{k_1,...,k_m}{x_1, ..., x_m; n} = \prod_{j=1}^m l^{k_j-1} \sum_{y_i^l = x_i} \S{k_1,...,k_m}{y_1, ..., y_m; l n}~.
\end{eqnarray}
%--------------------------------------------------------------------------------
Here the sum is over the $l$th roots of $x_i$ for $i \in \{1,\ldots,m\}$.
Equation~(\ref{SSDISTR}) is called {\itshape distribution  relation}. It  follows from the Vieta's theorem for (\ref{SSPOW}) and properties of 
symmetric polynomials \cite{ALG1}. If $x_1 \neq 1$ for $k_1 = 1$ the limit
%--------------------------------------------------------------------------------
\begin{eqnarray}
\label{SSDISTR1}
\sigma_{k_1,...,k_m}(x_1, ..., x_m) = \lim_{N \rightarrow \infty}S_{k_1,...,k_m}(x_1, ..., x_m;  N)
\end{eqnarray}
%--------------------------------------------------------------------------------
exists.
One may apply (\ref{SSDISTR},\ref{SSDISTR1}) to roots of unity $x_j$ and $y_j$, 
\ie $x_j = \exp(2 \pi i n_j/m_j)$ and $y_{j}^k = \exp(2 \pi i k n_j/(m_j l)), k \in \{1, \ldots, (l-1)\}, n_j, m_j \in \N$ and $i$ the imaginary unit. 

We combined all these relations together with the relations form Section \ref{SSInfRelations} and solved the corresponding linear systems using computer algebra methods
and derived the representations for the dependent sums analytically in Table~\ref{TT:XXa} we summarize the number of basis elements.
%-----------------------------------------------------------------------------------------
\begin{table}
\centering
%\begin{tabular}{l*{6}{c}r}
\scalebox{0.80}{%
\begin{tabular}{|c|r|r|r|r|r|r|r|r|r|r|r|r|r|r|r|r|r|r|r|r|}
\hline
\multicolumn{1}{|c}{$l$}             &
\multicolumn{1}{|c}{1}                 &
\multicolumn{1}{|c}{2}                 &
\multicolumn{1}{|c}{3}                 &
\multicolumn{1}{|c}{4}                 &
\multicolumn{1}{|c}{5}                 &
\multicolumn{1}{|c}{6}                 &
\multicolumn{1}{|c}{7}                 &
\multicolumn{1}{|c}{8}                 &
\multicolumn{1}{|c}{9}                 &
\multicolumn{1}{|c}{10}                &
\multicolumn{1}{|c}{11}                &
\multicolumn{1}{|c}{12}                &
\multicolumn{1}{|c}{13}                &
\multicolumn{1}{|c}{14}                &

\multicolumn{1}{|c}{15}                &
\multicolumn{1}{|c}{16}                &
\multicolumn{1}{|c}{17}                &
\multicolumn{1}{|c}{18}                &
\multicolumn{1}{|c}{19}                &
\multicolumn{1}{|c|}{20}               \\
\hline
{\sf basis} 
&  0  & 1  & 2  & 2  & 3  & 3  & 4  & 3 & 4 & 4 &  6 & 4 & 7 & 5 & 6 & 5 & 9 & 5 & 10 & 6  
\\
{\sf new} 
&  0  & 1  & 2  & 0  & 2  & 0  & 3  & 1 & 2 & 0 &  5 & 1 & 6 & 0 & 2 & 2 & 8 & 0 & 9 & 2 
\\
\hline
\end{tabular}
}
\caption{\label{TT:XXa}The number of the basis elements spanning the $w=1$ S-sums at  
$l$th roots of unity up to 20.}
\end{table}
The new elements contributing at the respective level of cyclotomy
for $l \leq 20$ are~:\\

\begin{tabular}{l l c r l}
$l = 1$  & $-$ 						&\hspace{1cm} & $l = 11$ & $\left.\sigma_1(e_{11}^k)\right|_{k=1}^5$\\
$l = 2$  & $\sigma_1(e_2)$      			&\hspace{1cm} & $l = 12$ & $\sigma_1(e_{12})$\\
$l = 3$  & $\sigma_1(e_3), \pi$ 			&\hspace{1cm} & $l = 13$ & $\left.\sigma_1(e_{13}^k)\right|_{k=1}^6$\\
$l = 4$  & $-$ 						&\hspace{1cm} & $l = 14$ & $-$\\
$l = 5$  & $\sigma_1(e_5),  \sigma_1(e_5^2)$ 		&\hspace{1cm} & $l = 15$ & $\sigma_1(e_{15}),  \sigma_1(e_{15}^2)$\\
$l = 6$  & $-$ 						&\hspace{1cm} & $l = 16$ & $\sigma_1(e_{16}),  \sigma_1(e_{16}^3)$\\
$l = 7$  & $\left. \sigma_1(e_7^k)\right|_{k=1}^3$ 	&\hspace{1cm} & $l = 17$ & $\left.\sigma_1(e_{17}^k)\right|_{k=1}^8$\\
$l = 8$  & $\sigma_1(e_8)$ 				&\hspace{1cm} & $l = 18$ & $-$\\
$l = 9$  & $\sigma_1(e_9),\sigma_1(e_9^2)$ 		&\hspace{1cm} & $l = 19$ & $\left.\sigma_1(e_{19}^k)\right|_{k=1}^9$\\
$l = 10$ & $-$ 						&\hspace{1cm} & $l = 20$ & $\sigma_1(e_{20}),  \sigma_1(e_{20}^3)~.$
\end{tabular}

At weight $w=2$ we use again the suitable relations from Section \ref{SSInfRelations} together with the distribution relations (\ref{SSDISTR}) and the 
relations due to complex conjugation:
\begin{eqnarray*}
 \sigma_{2}(e_j^k)^*&=&\sigma_{2}(e_j^{j-k})=\sigma_{2}(e_j^k)-\frac{6k^2-6jk+j^2}{3j^2}\pi^2\\
\sigma_{1,1}(e_j^k,e_j^{j-k}) &=& \sigma_2(e_j^k) + \frac{1}{2} \sigma_1(e_j^k)^2.
\end{eqnarray*}
In addition we use the relations
\begin{eqnarray*}
 \sigma_{1,1}(x,1)=\sigma_2(x)+ \frac{1}{2} \sigma_1(x)^2
\end{eqnarray*}
and 
\begin{eqnarray*}
\sigma_{1,1}(e_j^k,e_j^l) &=& \sigma_2(e_j^k) + \frac{1}{2} \sigma_1(e_j^k)^2 + \sigma_2\left(- \frac{e_j^k(1-e_j^l)}{1-e_j^k}\right),
\end{eqnarray*}
for $-\frac{e_j^k(1-e_j^l)}{1-e_j^k}=1$ or $-\frac{e_j^k(1-e_j^l)}{1-e_j^k}=\frac{z}{1+z}$ for a $j$-th root of unity $z$, since then
$$
\sigma_2\left(\frac{z}{1+z}\right)=\sigma_2(z)-\frac{1}{2}\sigma_2(z^2)-\frac{1}{2}(\sigma_1(z)-\sigma_1(z^2))^2.
$$
The new basis elements spanning the dilogarithms of the $l$th roots of unity for $l \leq 20$ are~:

\begin{tabular}{l l c r l}
$l = 1$  & $-$ 						&\hspace{1cm} & $l = 11$ & $\left.\sigma_2(e_11^k)\right|_{k=1}^5$\\
$l = 2$  & $-$  		    			&\hspace{1cm} & $l = 12$ & $\sigma_2(e_{12})$\\
$l = 3$  & $\sigma_2(e_3)$ 				&\hspace{1cm} & $l = 13$ & $\left.\sigma_2(e_{13}^k)\right|_{k=1}^6$\\
$l = 4$  & $\sigma_2(e_4) $ 				&\hspace{1cm} & $l = 14$ & $-$\\
$l = 5$  & $\sigma_2(e_5),  \sigma_2(e_5^2)$ 		&\hspace{1cm} & $l = 15$ & $\sigma_2(e_{15})$\\
$l = 6$  & $-$ 						&\hspace{1cm} & $l = 16$ & $\sigma_2(e_{16}),  \sigma_2(e_{16}^3)$\\
$l = 7$  & $\left. \sigma_2(e_7^k)\right|_{k=1}^3$ 	&\hspace{1cm} & $l = 17$ & $\left.\sigma_2(e_{17}^k)\right|_{k=1}^8$\\
$l = 8$  & $\sigma_2(e_8)$ 				&\hspace{1cm} & $l = 18$ & $-$\\
$l = 9$  & $\sigma_2(e_9),\sigma_2(e_9^2)$ 		&\hspace{1cm} & $l = 19$ & $\left.\sigma_2(e_{19}^k)\right|_{k=1}^9 $\\
$l = 10$ & $-$ 						&\hspace{1cm} & $l = 20$ & $\sigma_2(e_{20})~.$
\end{tabular}

In Table~\ref{TT:XXb} we summarize the number of basis elements found at $w = 2$ using the mentioned relations. We also list the number of basis elements for the class of dilogarithms at roots
of unity, and in both cases the number of new elements beyond those being obtained at $w = 1$ at the 
same value of $l$.
%-----------------------------------------------------------------------------------------
\begin{table}
\centering
%\begin{tabular}{l*{6}{c}r}
\scalebox{0.80}{%
\begin{tabular}{|c|r|r|r|r|r|r|r|r|r|r|r|r|r|r|r|r|r|r|r|r|}
\hline
\multicolumn{1}{|c}{l}             &
\multicolumn{1}{|c}{1}                 &
\multicolumn{1}{|c}{2}                 &
\multicolumn{1}{|c}{3}                 &
\multicolumn{1}{|c}{4}                 &
\multicolumn{1}{|c}{5}                 &
\multicolumn{1}{|c}{6}                 &
\multicolumn{1}{|c}{7}                 &
\multicolumn{1}{|c}{8}                 &
\multicolumn{1}{|c}{9}                 &
\multicolumn{1}{|c}{10}                &
\multicolumn{1}{|c}{11}                &
\multicolumn{1}{|c}{12}                &
\multicolumn{1}{|c}{13}                &
\multicolumn{1}{|c}{14}                &

\multicolumn{1}{|c}{15}                &
\multicolumn{1}{|c}{16}                &
\multicolumn{1}{|c}{17}                &
\multicolumn{1}{|c}{18}                &
\multicolumn{1}{|c}{19}                &
\multicolumn{1}{|c|}{20}               \\
\hline
{\sf $\Li_2$~~basis} 
& 1 & 1 & 2 & 2 & 3 & 2 & 4 & 3 & 4 & 3 & 6 & 3 & 7 & 4 & 5 & 5 & 9 & 4 & 10 & 5   
\\
{\sf $\Li_2$~~new} 
& 1 & 0 & 1 & 1 & 2 & 0 & 3 & 1 & 2 & 0 & 5 & 0 & 6 & 0 & 1 & 2 & 8 & 0 & 9  & 1   
\\
{\sf basis} 
& 1 & 2 & 3 & 3 & 5  & 5 & 8  & 7  & 10  & 10  & 16  & 12  & 21  & 17  & 21  & 21  & 33  & 23  & 40   & 29   
\\
{\sf  new} 
& 1 & 0 & 1 & 1 & 2 & 1 & 4 & 3 & 5  & 4  & 10 & 5  & 14 & 8  & 12  & 12  & 24 & 11 & 30  & 16   
\\
\hline
\end{tabular}
}
\caption{\label{TT:XXb}The number of the basis elements spanning the dilogarithms respectively the weight $w = 2$ S-sums at $l$th roots of unity  up to $20$.}.
\end{table}

At  $w = 2$ the respective new basis elements are~:
%--------------------------------------------------------------------------------
\begin{eqnarray*}
{l = 1}   & & \sigma_2 \\
{l = 2}   & & - \\
{l = 3}   & & \sigma_2(e_3) \\
{l = 4}   & & \sigma_2(e_4) \\
{l = 5}   & & \sigma_2(e_5)\sigma_2(e_5^2) \\
{l = 6}   & & \sigma_{1,1}(e_3,e_2)  \\
{l = 7}   & & \left. \sigma_2(e_7^k)\right|_{k=1}^3,\sigma_{1,1}(e_7,e_7^2) \\
{l = 8}   & & \sigma_2(e_8), \sigma_{1,1}(e_8,e_4),\sigma_{1,1}(e_8,e_8^3)\\
{l = 9}   & & \sigma_2(e_9), \sigma_2(e_9^2),\sigma_{1,1}(e_9,e_9^2),\sigma_{1,1}(e_9,e_3),\sigma_{1,1}(e_9^2,e_3)\\
{l = 10}  & & \sigma_{1,1}(e_5,e_2),\sigma_{1,1}(e_5^2,e_2),\sigma_{1,1}(e_{10},e_5),\sigma_{1,1}(e_{10},e_{10}^3) \\
{l = 11}  & &  \left.\sigma_2(e_{11}^k)\right|_{k=1}^5, \left. \sigma_{1,1}(e_{11}, e_{11}^k)\right|_{k=2}^4,\left. \sigma_{1,1}(e_{11}^2, e_{11}^k)\right|_{k=3}^4,~\text{etc.} 
\end{eqnarray*}
%--------------------------------------------------------------------------------

\section{Asymptotic Expansion of S-Sums}
In Section \ref{HSExpansion} we introduced an algorithm to compute the asymptotic expansions of harmonic sums. In this section we seek expansions of S-sums and try to generalize the algorithm for harmonic sums to S-sums to the extent that is possible.
\subsection{Asymptotic Expansions of Harmonic Sums \texorpdfstring{$\S{1}{c;n}$}{S1c(n)} with \texorpdfstring{$c\geq 1$}{c<=1}}
\label{SSexpandS1b}
Since $\frac{1}{c-x}$ is not continuous in $[0,1]$ for $0<c\leq 1,$ we cannot use the strategy mentioned in Section \ref{HShexp} to find an asymptotic representation of $\int_0^1{\frac{\frac{x}{c}^n-1}{c- x}dx}=\S{1}{c;n}.$ For $\S{1}{n}$ see Lemma \ref{HSexpandS1}.
To compute the asymptotic expansion of $\S{1}{c;n}$ with $c>1$ one may represent it using the Euler-Maclaurin formula \cite{Euler1738,MacLaurin1742}:
\begin{eqnarray*}
\SS1cn&=& \bar{S}_{1}(c;n)+\S1n
\end{eqnarray*}
where
\begin{eqnarray*}
\bar{S}_{1}(c;n)&:=&\sum_{i=1}^n\frac{c^i-1}{i}\\
		    &=&-\frac{1}{2}\left[\log(c)+\frac{c^{n+1}-1}{n+1}\right]-\gamma-\log(\log(c))-\log(n+1)\\
		    & &+\textnormal{Ei}(\log(c)(n+1))+\sum_{j=1}^m{\frac{B_{2j}}{(2j)!}\left.\left[\frac{d^{2j-1}}{di^{2j-1}}\left(\frac{c^i-1}{i}\right)\right]\right|_0^{n+1}}+R_k(c,n);
\end{eqnarray*}
here $\textnormal{Ei}(z)$ is the exponential integral \cite{Euler1768,Nielsen1906a},
$$
\int_0^x\frac{c^t-1}{t}dt=\textnormal{Ei}(x\log(c))-\gamma-\log(x)-\log(\log(c)),
$$
$B_k$ are the Bernoulli numbers \cite{Bernoulli1713,Saalschuetz1893,Nielsen1923},
$$
B_n=\sum_{k=0}^n{\binom{n}{k}} B_k, \ \ \ B_0=1,
$$
and $R_k(c,n)$ stand for the rest term. The derivatives are given in closed form by
\begin{eqnarray*}
\frac{d^{j}} {di^{j}}\left(\frac{c^i-1}{i}\right)&=&\frac{c^i}{i}\sum_{k=0}^j\frac{\log^{j-k}(c)}{i^k}{(-1)}^k\frac{\Gamma(j+1)}{\Gamma(j+1-k)}-{(-1)}^j\frac{\Gamma(j+1)}{i^{j+1}}\\
\lim_{i\rightarrow0}\frac{d^{j}}{di^{j}}\left(\frac{c^i-1}{i}\right)&=&\frac{\log^{j+1}(c)}{j+1}.
\end{eqnarray*}
The asymptotic representation for $\S{1}{c;n}$ with $c>1$ can now be derived using
$$
\textnormal{Ei}(t) \sim \exp(t) \sum_{k=1}^{\infty}\frac{\Gamma(k)}{t^k}.
$$

\subsection{Computation of Asymptotic Expansions of S-Sums}
\label{SSExpansion}
In this section we extend the algorithm of Section \ref{HSExpansion} to S-Sums $\S{a_1,\ldots,a_k}{b_1,\ldots,b_k;n}$ with $b_i\in[-1,0)\cup(0,1].$ Analyzing 
Section~\ref{SSMellin} it is easy to see that these sums can be represented using Mellin transforms of the form
$$
\M{\frac{\H{m_1,m_2\ldots,m_k}x}{c\pm x}}{n}
$$
with $m_i\in \R \setminus ((-1,0)\cup(0,1))$ and $\abs{c}\geq1$. Analyzing the product of S-sums it is straightforward to see that the subclass of S-Sums we are considering here
are closed under multiplication, in the sense that the product can be expressed as a linear combination of S-sums of this subclass. 
In the following we state several necessary facts which are  more or less extensions of lemmas from Section \ref{HSExpansion}. The two subsequent lemmas are 
generalizations of Lemma \ref{HSanalytic1} and Lemma \ref{HSanalytic2}, respectively, and the proofs of the lemmas follow using similar arguments as used in the proof of Lemma \ref{HSanalytic1}.

\begin{lemma}
Let $\H{m_1,m_2,\ldots,m_k}x$ be a multiple polylogarithm with $m_i~\in~\R~\setminus(0,1]$ for $1~\leq~i~\leq k.$ Then 
$$\H{m_1,m_2,\ldots,m_k}x,\; \frac{\H{m_1,m_2,\ldots,m_k}x}{c+x}$$
and
$$\frac{\H{m_1,m_2,\ldots,m_k}x-\H{m_1,m_2,\ldots,m_k}1}{c-x}$$
are analytic for $\abs{c}\geq1$ and $x \in (0,1].$
\label{SSanalytic1}
\end{lemma}

\begin{lemma}
Let $\H{m_1,m_2,\ldots,m_k}x$ be a multiple polylogarithm with $m_i~\in~\R~\setminus~((-1,0)~\cup~(0,1))$ for $1~\leq~i~\leq~k,$ and $m_k~\neq~0.$ Then 
$$\frac{\H{m_1,m_2,\ldots,m_k}{1-x}}{c+x}$$
and
$$\frac{\H{m_1,m_2,\ldots,m_k}{1-x}}{c-x}$$
are analytic for $\abs{c}\geq1$ and $x \in (0,1].$
\label{SSanalytic2}
\end{lemma}
The proof of the subsequent lemma is analogue to the proof of Lemma \ref{HSmelexpconst}.
\begin{lemma}[Compare Lemma \ref{HSmelexpconst}]
Let $\H{\ve m}x=\H{m_1,m_2,\ldots,m_k}x$ and $\H{\ve b}x=\H{b_1,b_2,\ldots,b_l}x$ be a multiple polylogarithm with 
$m_i~\in~\R~\setminus(0,1]$ for $1~\leq~i~\leq~k$ and $b_i~\in~\R~\setminus~((-1,0)~\cup~(0,1))$ for $1~\leq~i~\leq~l$ where $b_l\neq 0.$ Then we have
$$
\M{\frac{\H{\ve m}x}{1-x}}{n} =\int_0^1{\frac{x^n(\H{\ve m}x-\H{\ve m}1)}{1-x}dx}-\int_0^1{\frac{\H{\ve m}x-\H{\ve m}1}{1-x}}dx-\S{1}n\H{\ve m}{1},
$$
and
$$
\M{\frac{\H{\ve b}{1-x}}{1-x}}{n}=\int_0^1{\frac{x^n\H{\ve b}{1-x}}{1-x}dx}-\H{0,\ve b}1
$$
where
$$\int_0^1{\frac{\H{\ve m}x-\H{\ve m}1}{1-x}}dx, \ \H{\ve m}{1} \textnormal{ and } \H{0,\ve b}1 $$
are finite constants.
\label{SSexpansioncorrection1}
\end{lemma}

\begin{remark}
\label{SSExpandableIntegrals}
Combining Lemma \ref{SSanalytic1}, Lemma \ref{SSanalytic2} and \ref{SSexpansioncorrection1} we are able to find the asymptotic expansion of Mellin transforms (see Section \ref{SSMellin}) of the form
$$\M{\frac{\H{m_1,m_2,\ldots,m_k}x}{c \pm x}}{n} \ \textnormal{ and } \M{\frac{\H{b_1,b_2,\ldots,b_k}{1-x}}{c \pm x}}{n}$$
where $\abs{c}>1,m_i~\in~\R~\setminus(0,1], b_i~\in~\R~\setminus~(-1,0)~\cup~(0,1))$ and $b_l\neq 0$ using the method presented in the beginning of Section \ref{HSExpansion}, 
or using repeated integration by parts.
\end{remark}

\begin{lemma}
\begin{enumerate}[Compare Lemma \ref{SSsinglmostcomp}]
\item If an $\bar{S}$-sum $\S{a_1,a_2,\ldots,a_k}{b_1,b_2,\ldots,b_k;n}$ has no trailing ones, \ie $a_k\neq1$ or $b_k\neq1$ then the most complicated multiple 
       polylogarithm in the inverse Mellin transform of $\S{a_1,a_2,\ldots,a_k}{b_1,b_2,\ldots,b_k;n}$ has no trailing ones.
\item If an $\bar{H}$-multiple polylogarithm $\H{m_1,m_2,\ldots,m_k}x$ has no trailing ones, \ie $m_k\neq 1$ then the most complicated S-sum 
      in the Mellin transform of $\S{a_1,a_2,\ldots,a_k}{b_1,b_2,\ldots,b_k;n}$ has no trailing ones.
\end{enumerate}
\label{SSnotrail1}
\end{lemma}
\begin{proof}
The statements follow immediately form Algorithm \ref{SSmostcomphlog} and Algorithm \ref{SSmostcompsum} respectively.
\end{proof}

\begin{remark}
According to the previous lemma we know that if $\H{m_1,m_2,\ldots,m_{k}}x$ is the most complicated multiple polylogarithm in the inverse Mellin 
transform of an S-sum  $\S{a_1,\ldots,a_k}{b_1,\ldots,b_k;n}$ with $b_i\in[-1,0)\cup(0,1]$ then $m_k\neq 1.$ Due to Remark \ref{SSExpandableIntegrals} we can expand the resulting integrals 
asymptotically (either we can handle it directly,
or we have to transform the argument of the multiple polylogarithm). A S-sum with trailing ones
might lead to a multiple polylogarithm with trailing ones, however $\frac{\Hma{m_1,m_2,\ldots,m_{k-1},1}x}{c\pm x}$ is not analytic at $1$ and hence we cannot use the strategy mentioned 
in Section \ref{HShexp} to find an asymptotic representation of such integrals. Fortunately we can always extract trailing ones such that we end up in a univariate polynomial 
in $\S1n$ with coefficients in the S-sums without trailing ones (this is possible due to the quasi shuffle algebra property). Hence we only need to deal with powers of $\S1n$ and harmonic sums without trailing ones.
\end{remark}

We are now prepared to extend the algorithm presented in Section \ref{HScompasy} to compute asymptotic expansions for harmonic sums to an algorithm which
computes asymptotic expansions of S-sums $\S{a_1,a_2,\ldots,a_k}{b_1,b_2,\ldots,b_k;n}$ with $b_i\in [-1,1]$ and $b_i\neq 0$.
We can proceed as follows:

\begin{itemize}
	\item if  $\S{a_1,a_2,\ldots,a_k}{b_1,b_2,\ldots,b_k;n}$ has trailing ones, \ie $a_k=b_k=1$, we first extract them such that we end up in a univariate 
	polynomial in $\S1n$ with coefficients in the S-sums without trailing ones; apply the 
	following items to each of the S-sums without trailing ones;
	\item suppose now $\S{a_1,a_2,\ldots,a_k}{b_1,b_2,\ldots,b_k;n}$ has no trailing ones, \ie $a_k\neq 1\neq b_k;$ let $\frac{\H{m_1,m_2,\ldots,m_l}x}{c+sx}$ be
	the most complicated weighted multiple polylogarithm in the inverse Mellin transform of $\S{a_1,a_2,\ldots,a_k}{b_1,b_2,\ldots,b_k;n};$ express 
	$\S{a_1,a_2,\ldots,a_k}{b_1,b_2,\ldots,b_k;n}$ as
	\begin{equation}\label{SSasyalg1}
		\S{a_1,a_2,\ldots,a_k}{b_1,b_2,\ldots,b_k;n}=\M{\frac{\H{m_1,m_2,\ldots,m_l}x}{c+sx}}{n}+T;
	\end{equation}
	note that $T$ is an expression in S-sums (which are less complicated than $\S{a_1,a_2,\ldots,a_k}{b_1,b_2,\ldots,b_k;n}$) and constants;
	\item we proceed by expanding $\M{\frac{\H{m_1,m_2,\ldots,m_l}x}{c+sx}}{n}$ in (\ref{SSasyalg1});  
	\begin{description}
		\item[all $m_i\neq 1$:] expand $\M{\frac{\H{m_1,m_2,\ldots,m_l}x}{c+sx}}{n}$ directly see Remark \ref{SSExpandableIntegrals}

		\item[not all $m_i\neq 1$:]
		\begin{itemize}\item[]
			\item transform $x\rightarrow 1-x$ in $\H{\ve m}x$; expand all products; now we can rewrite
				\begin{equation}\label{SSasyalg2}
					\M{\frac{\H{\ve m}x}{c+sx}}{n}=\sum_{i=1}^pd_i\M{\frac{\H{\ve b_i}{1-x}}{c+sx}}{n}+d \ \textnormal{  with } d,d_i\in\R
				\end{equation}
			 \item for each Mellin transform $\M{\frac{\H{b_1,\ldots,b_j}{1-x}}{c+sx}}{n}$ do
				\begin{description}\item[]
					\item[$b_j\neq 0:$] expand $\int_0^1{\frac{x^n\Hma{\ve b}{1-x}}{c+sx}}$ see Remark \ref{SSExpandableIntegrals}
					\item[$b_j=0:$] transform back $1-x\rightarrow x$ in $\H{\ve b}{1-x}$; expand all products; write
						\begin{equation}\label{SSasyalg3}
							\M{\frac{\H{\ve b}{1-x}}{c+sx}}{n}=\sum_{i=1}^pe_i\M{\frac{\H{\ve g_i}{x}}{c+sx}}{n}+e
						\end{equation}
						with $e,e_i\in\R$ and perform the Mellin transforms $\M{\frac{\Hma{\ve g_i}{x}}{c+sx}}{n}$
				\end{description}
		\end{itemize}
	\end{description}
	\item replace $\M{\frac{\Hma{m_1,m_2,\ldots,m_l}x}{c+sx}}{n}$ in equation (\ref{SSasyalg1}) by the result of this process
	\item expand the powers of $\S{1}n$ using Lemma \ref{HSexpandS1}
	\item for all harmonic sums that remain in equation (\ref{SSasyalg1}) apply the above points; since these harmonic sums are less complicated this process will terminate.
\end{itemize}

Some remarks are in place: Since $a_k\neq 1$ in equation (\ref{SSasyalg1}), we know due to Lemma \ref{SSnotrail1} that $m_l\neq 0$ in equation (\ref{SSasyalg1}).
If not all $m_i\neq 1$ in equation  (\ref{SSasyalg1}), we have to transform $x\rightarrow 1-x$ in $\H{m_1,m_2,\ldots,m_l}x$ from Section \ref{SSRelatedArguments}. 
It is easy to see that the single 
multiple polylogarithm at argument $x$ with weight $l,$ which will pop up, will not have trailing zeroes since $m_l\neq 1$. Therefore the S-sums which will appear
in equation (\ref{SSasyalg3}) are less complicated (\ie smaller in the sense of Definition \ref{SSsord}) than $\S{a_1,a_2,\ldots,a_k}{b_1,b_2,\ldots,b_k;n}$ of (\ref{SSasyalg1}) and 
hence this algorithm will eventually terminate.

\begin{example}
\small
\begin{eqnarray*}
\textnormal{S}_{2,1}\hspace{-0.2em}\left(\frac{1}{3},\frac{1}{2};n\right)&\sim&  -\textnormal{S}_{1,2}\hspace{-0.2em}\left(\frac{1}{6},3;\infty \right)
  +\textnormal{S}_1\hspace{-0.2em}\left(\frac{1}{2};\infty \right) \left(-\textnormal{S}_2\hspace{-0.2em}\left(\frac{1}{6};\infty \right)
+\textnormal{S}_2\hspace{-0.2em}\left(\frac{1}{3};\infty \right)+\frac{1}{6^n} \left(-\frac{1464}{625n^5}\right.\right.\\
&&\left.\left.+\frac{126}{125 n^4}-\frac{12}{25 n^3}+\frac{1}{5n^2}+33\frac{ 2^{n-1}}{n^5}-9\frac{2^{n-1}}{n^4}+3\frac{2^{n-1}}{n^3}-\frac{2^{n-1}}{n^2}\right)\right)\\
&&+\frac{1}{6^n} \left(-\frac{4074}{3125 n^5}+\frac{366}{625 n^4}-\frac{42}{125 n^3}+\frac{6}{25 n^2}-\frac{1}{5 n}\right)
   \textnormal{S}_2\hspace{-0.2em}\left(\frac{1}{2};\infty \right)\\
&&+\textnormal{S}_1\hspace{-0.2em}\left(\frac{1}{6};\infty \right) \textnormal{S}_2\hspace{-0.2em}\left(\frac{1}{2};\infty \right)+\textnormal{S}_3\hspace{-0.2em}\left(\frac{1}{2};\infty
   \right)+\frac{1}{6^n} \left(\frac{642}{125 n^5}-\frac{28}{25 n^4}+\frac{1}{5 n^3}\right)\\
&&+\frac{1}{6^n} \left(\frac{2037}{3125 n^5}-\frac{183}{625
   n^4}+\frac{21}{125 n^3}-\frac{3}{25 n^2}+\frac{1}{10 n}\right) \zeta_2+\frac{1}{6^n} \left(-\frac{2037}{3125 n^5}\right.\\
&&\left.+\frac{183}{625
   n^4}-\frac{21}{125 n^3}+\frac{3}{25 n^2}-\frac{1}{10 n}\right) \log ^2(2)+\frac{1}{6^n} \left(\frac{1464}{625 n^5}-\frac{126}{125 n^4}+\frac{12}{25
   n^3}\right.\\
&&\left.-\frac{1}{5 n^2}\right) \log (2), \ x \rightarrow \infty.
\end{eqnarray*}
\normalsize
\end{example}

In order to extend this algorithm to a slightly wider class, we state the following lemma.
\begin{lemma}
Let $\H{m_1,m_2,\ldots,m_k}x$ be a multiple polylogarithm with $m_i~\in~\R~\setminus~(-1, 1]~\cup~\{0\}$ for $1\leq i\leq k$ and $0<c<1.$ Then
\begin{eqnarray*}
&&\hspace{-1.5cm}\M{\frac{\H{m_1,m_2,\ldots,m_k}x}{c-x}}{n}=\\
					&&\frac{1}{c^n}\int_0^1{\frac{x^n\left(\H{m_1,m_2,\ldots,m_k}x-\H{m_1,m_2,\ldots,m_k}c\right)}{c-x}dx}\\
					&& + \H{m_1,c,m_2,\ldots,m_k}c+\H{m_1,m_2,c,\ldots,m_k}c+\cdots+\H{m_1,m_2,\ldots,m_k,c}c\\
					&& + \H{m_2,\ldots,m_k}c\H{0,m_1-c}{1-c} + \H{m_3,\ldots,m_k}c\H{0,m_1-c,m_2-c}{1-c}\\
					&& +\cdots+ \H{m_k}c\H{0,m_1-c,\ldots,m_{k-1}-c}{1-c} + \H{0,m_1-c,\ldots,m_k-c}{1-c} \\
					&&- S_{1,\{\frac{1}{c}\}}(n)\H{m_1,m_2,\ldots,m_k}c.
\end{eqnarray*}
\label{SSexpansioncorrection2}
\end{lemma}
\begin{proof}We have
\begin{eqnarray*}
\M{\frac{\H{m_1,m_2,\ldots,m_k}x}{c-x}}{n}&=&\int_0^1{\frac{((\frac{x}{c})^n-1)\H{m_1,m_2,\ldots,m_k}x}{c-x}dx}=\\
			&=&\int_0^1{\frac{(\frac{x}{c})^n(\H{m_1,m_2,\ldots,m_k}x-\H{m_1,m_2,\ldots,m_k}c)}{c-x}dx}\\
			&&+\int_0^1{\frac{(\frac{x}{c})^n\H{m_1,m_2,\ldots,m_k}c-\H{m_1,m_2,\ldots,m_k}x}{c-x}dx}\\
			&=&\int_0^1{\frac{(\frac{x}{c})^n(\H{m_1,m_2,\ldots,m_k}x-\H{m_1,m_2,\ldots,m_k}c)}{c-x}dx}\\
			&&+\underbrace{\H{m_1,m_2,\ldots,m_k}c\int_0^1{\frac{(\frac{x}{c})^n-1}{c-x}dx}}_{A:=}\\
			&&-\underbrace{\int_0^1{\frac{\H{m_1,m_2,\ldots,m_k}x-\H{m_1,m_2,\ldots,m_k}c}{c-x}dx}}_{B:=}.
\end{eqnarray*}
with
\begin{eqnarray*}
A&=&- \S{1}{\frac{1}{c};n}\H{m_1,m_2,\ldots,m_k}c\\
B&=&\underbrace{\int_0^c{\frac{\H{m_1,m_2,\ldots,m_k}x-\H{m_1,m_2,\ldots,m_k}c}{c-x}dx}}_{B_1:=}\\
			&&\hspace{2.5cm}+\underbrace{\int_c^1{\frac{\H{m_1,m_2,\ldots,m_k}x-\H{m_1,m_2,\ldots,m_k}c}{c-x}dx}}_{B_2:=}\\
\end{eqnarray*}
where
\begin{eqnarray*}
B_1&=&-\H{m_1,c,m_2,\ldots,m_k}c-\cdots-\H{m_1,m_2,\ldots,m_k,c}c\\
B_2&=&\int_0^{1-c}{\frac{\H{m_1,m_2,\ldots,m_k}{x+c}-\H{m_1,m_2,\ldots,m_k}c}{-x}dx}\\
			&=&-\int_0^{1-c}{\frac{\H{m_2,\ldots,m_k}c\H{m_1-c}{x} +\cdots + \H{m_1-c,\ldots,m_k-c}{x}}{x}dx}\\
			&=&-\H{m_2,\ldots,m_k}c\H{0,m_1-c}{1-c} - \H{m_3,\ldots,m_k}c\H{0,m_1-c,m_2-c}{1-c}\\
			&&-\cdots- \H{m_k}c\H{0,m_1-c,\ldots,m_{k-1}-c}{1-c} - \H{0,m_1-c,\ldots,m_k-c}{1-c};
\end{eqnarray*}
thus the lemma follows.

\end{proof}

\begin{remark}
Using this lemma we can extend the presented algorithm to $\bar{S}-$sums $\S{a_1,\ldots,a_k}{b_1,\ldots,b_k;n}$ provided that if $\abs{b_1}>1,$ we have that
$\abs{b_i}<1$ for $2<i<k$.
\end{remark}

\cleardoublepage 
 
%\chapter{Generalized Harmonic Sums: Cyclotomic Harmonic Sums}
\chapter{Cyclotomic Harmonic Sums}
\label{CSchapter}
In Chapter \ref{SSchapter} we already generalized harmonic sums and worked with S-sums, whereas in this chapter we want to extend harmonic sums into a different direction that leads to 
{\itshape cyclotomic harmonic sums} (compare \cite{Ablinger2011,Ablinger2012}). Most of the facts that we will present here were already published in \cite{Ablinger2011}, however here 
we will work out the theoretic background in detail. For instance the theorems are supplemented by proofs, the algorithms mentioned in \cite{Ablinger2011} are worked out in detail and more 
examples are given. Note that all methods presented in this chapter can be applied also for harmonic sums and harmonic polylogarithms.
\section{Definition and Structure of Cyclotomic Harmonic Sums}
\label{CSdef}
\begin{definition}[Cyclotomic Harmonic Sums]
Let $a_i,k \in \N,b_i,n\in\N_0$ and $c_i\in\Z^*$ for $i\in \{1,2,\ldots,k\}.$ We define
\begin{eqnarray*}
&&\S{(a_1,b_1,c_1),(a_2,b_2,c_2),\ldots,(a_k,b_k,c_k)}{n}=\\
&&\hspace{2cm}=\sum_{i_1 \geq i_2,\cdots i_k \geq 1}\frac{\sign{c_1}^{i_1}}{(a_1 i_1+b_1)^{\abs{c_1}}}\frac{\sign{c_2}^{i_2}}{(a_2 i_2+b_2)^{\abs{c_2}}}\cdots\frac{\sign{c_k}^{i_1}}{(a_k i_k+b_k)^{\abs{c_k}}}\\
&&\hspace{2cm}=\sum_{i_1=1}^n\frac{\sign{c_1}^{i_1}}{(a_1 i_1+b_1)^{\abs{c_1}}}\sum_{i_2=1}^{i_1}\frac{\sign{c_2}^{i_2}}{(a_2 i_2+b_2)^{\abs{c_2}}}\cdots\sum_{i_k=1}^{i_{k-1}}\frac{\sign{c_k}^{i_1}}{(a_k i_k+b_k)^{\abs{c_k}}}.
\end{eqnarray*}
$k$ is called the depth and $w=\sum_{i=0}^k\abs{c_i}$ is called the weight of the cyclotomic harmonic sum $\S{(a_1,b_1,c_1),(a_2,b_2,c_2),\ldots,(a_k,b_k,c_k)}{n}$.
\label{CSdefsum}
\end{definition}

\subsection{Product}
In the following we will see that as for harmonic sums and S-sums it is possible to express the product of cyclotomic harmonic sums with the same upper summation limit as linear 
combination of single cyclotomic harmonic sums. Hence cyclotomic harmonic sums form a quasi shuffle algebra. We start with the subsequent lemma.
\begin{lemma}
Let $a_1,a_2,c_1,c_2 \in \N,b_1,b_2,i \in \N_0.$ If $a_1 b_2\neq a_2 b_1,$ we have
\begin{eqnarray*}
\frac{1}{(a_1 i +b_1)^{c_1}(a_2 i +b_2)^{c_2}}&=&\\
&&\hspace{-2cm}(-1)^{c_1}\sum_{j=1}^{c_1}{ (-1)^j \binom{c_1+c_2-j-1}{c_2-1}\frac{a_1^{c_2}a_2^{c_1-j}}{(a_1 b_2-a_2 b_1)^{c_1+c_2-j}}\frac{1}{(a_1 i+b_1)^j}}\\
&&\hspace{-2cm}+(-1)^{c_2}\sum_{j=1}^{c_2}{ (-1)^j \binom{c_1+c_2-j-1}{c_1-1}\frac{a_1^{c_2-j}a_2^{c_1}}{(a_2 b_1-a_1 b_2)^{c_1+c_2-j}}\frac{1}{(a_2 i+b_2)^j}},
\end{eqnarray*}
and if $a_1 b_2= a_2 b_1,$ we have
\begin{eqnarray*}
\frac{1}{(a_1 i +b_1)^{c_1}(a_2 i +b_2)^{c_2}}&=&\left(\frac{a_1}{a_2}\right)^{c_2}\frac{1}{(a_1 i+b_1)^{c_1+c_2}}.
\end{eqnarray*}
\label{CSapart}
\end{lemma}

Having this lemma in mind we find the following product formulas.

\begin{thm}
Let $a_i,d_i,k,l,n \in \N,b_i,e_i \in \N_0$ and $c_i,f_i \in \Z^*.$
If $a_1 e_1\neq d_1 b_1$ we have
\begin{eqnarray*}
&&\S{(a_1,b_1,c_1),\ldots,(a_k,b_k,c_k)}{n}\S{(d_1,e_1,f_1),\ldots,(d_l,e_l,f_l)}{n}=\\
&&\sum_{i=1}^n\frac{\sign{c_1}^i}{(a_1 i+b_1)^{\abs{c_1}}}\S{(a_2,b_2,c_2),\ldots,(a_k,b_k,c_k)}{i}\S{(d_1,e_1,f_1),\ldots,(d_l,e_l,f_l)}{i}\\
&&+\sum_{i=1}^n\frac{\sign{f_1}^i}{(d_1 i+e_1)^{\abs{f_1}}}\S{(a_1,b_1,c_1),\ldots,(a_k,b_k,c_k)}{i}\S{(d_2,e_2,f_2),\ldots,(d_l,e_l,f_l)}{i}\\
&&-\sum_{i=1}^n \left( (-1)^{\abs{c_1}}\sum_{j_1=1}^{\abs{c_1}}  (-1)^j \binom{\abs{c_1}+\abs{f_1}-j-1}{\abs{f_1}-1}\frac{a_1^{\abs{f_1}}d_1^{\abs{c_1}-j}}{a_1 e_1-d_1 b_1}\frac{1}{(a_1 i+b_1)^j}\right.\\
&&\hspace{1cm}\left.+(-1)^{\abs{f_1}}\sum_{j=1}^{\abs{f_1}}{ (-1)^j \binom{\abs{c_1}+\abs{f_1}-j-1}{\abs{f_1}-1}\frac{a_1^{\abs{f_1}-j}d_1^{\abs{c_1}}}{d_1 b_1-a_1 e_1}\frac{1}{(d_1 i+e_1)^j}}\right)\\
&&\hspace{1cm}\S{(a_2,b_2,c_2),\ldots,(a_k,b_k,c_k)}{i}\S{(d_2,e_2,f_2),\ldots,(d_l,e_l,f_l)}{i},
\end{eqnarray*}
and if $a_1 e_1= d_1 b_1,$ we have
\begin{eqnarray*}
&&\S{(a_1,b_1,c_1),\ldots,(a_k,b_k,c_k)}{n}\S{(d_1,e_1,f_1),\ldots,(d_l,e_l,f_l)}{n}=\\
&&\sum_{i=1}^n\frac{\sign{c_1}^i}{(a_1 i+b_1)^{\abs{c_1}}}\S{(a_2,b_2,c_2),\ldots,(a_k,b_k,c_k)}{i}\S{(d_1,e_1,f_1),\ldots,(d_l,e_l,f_l)}{i}\\
&&+\sum_{i=1}^n\frac{\sign{f_1}^i}{(d_1 i+e_1)^{\abs{f_1}}}\S{(a_1,b_1,c_1),\ldots,(a_k,b_k,c_k)}{i}\S{(d_2,e_2,f_2),\ldots,(d_l,e_l,f_l)}{i}\\
&&-\left(\frac{a_1}{d_1}\right)^{\abs{f_1}}\sum_{i=1}^n\frac{(\sign{c_1} \sign{f_1})^i}{(a_1 i+b_1)^{\abs{c_1}+\abs{f_1}}}
	\S{(a_2,b_2,c_2),\ldots,(a_k,b_k,c_k)}{i}\S{(d_2,e_2,f_2),\ldots,(d_l,e_l,f_l)}{i}.
\end{eqnarray*}
\label{CSprod}
\end{thm}
\begin{proof}
From the identity $$\sum_{i=1}^n\sum_{j=1}^n a_{i j}=\sum_{i=1}^n\sum_{j=1}^i a_{i j}+\sum_{j=1}^n\sum_{i=1}^j a_{i j}-\sum_{i=1}^na_{i i}$$ we get immediately
\begin{eqnarray*}
&&\S{(a_1,b_1,c_1),\ldots,(a_k,b_k,c_k)}{n}\S{(d_1,e_1,f_1),\ldots,(d_l,e_l,f_l)}{n}=\\
&&\sum_{i=1}^n\frac{\sign{c_1}^i}{(a_1 i+b_1)^{\abs{c_1}}}\S{(a_2,b_2,c_2),\ldots,(a_k,b_k,c_k)}{i}\S{(d_1,e_1,f_1),\ldots,(d_l,e_l,f_l)}{i}\\
&&+\sum_{i=1}^n\frac{\sign{f_1}^i}{(d_1 i+e_1)^{\abs{f_1}}}\S{(a_1,b_1,c_1),\ldots,(a_k,b_k,c_k)}{i}\S{(d_2,e_2,f_2),\ldots,(d_l,e_l,f_l)}{i}\\
&&-\sum_{i=1}^n\frac{(\sign{c_1} \sign{f_1})^i}{(a_1 i+b_1)^{\abs{c_1}}(d_1 i+e_1)^{\abs{f_1}}}
	\S{(a_2,b_2,c_2),\ldots,(a_k,b_k,c_k)}{i}\S{(d_2,e_2,f_2),\ldots,(d_l,e_l,f_l)}{i}.
\end{eqnarray*}
Using Lemma \ref{CSapart} we get the results.
\end{proof}

\subsubsection*{Extracting Leading or Trailing Ones.}
\label{CSextracttrailing}
We say that a cyclotomic harmonic sum $\S{(a_1,b_1,c_1),\ldots,(a_k,b_k,c_k)}{n}$ has leading or trailing ones, if $c_1=1$ or $c_k=1,$ respectively. 
If $c_i=1$ for $1\leq i\leq j<k$ and $c_{j+1}\neq 1,$
we say that the sum has $j-$leading ones, analogously, if $c_i=1$ for $1<j\leq i\leq k$ and $c_{j-1}\neq 1,$
we say that the sum has $k-j+1-$trailing ones. In addition, we call a cyclotomic harmonic sum
$\S{(a_1,b_1,c_1),\ldots,(a_k,b_k,c_k)}{n}$ which has only ones \ie $c_i=1$ for $1\leq i\leq k$ a {\itshape linear} cyclotomic harmonic sum. 

As for harmonic sums we can use the product of the cyclotomic harmonic sums to 
single out leading or trailing ones: we consider the sum
$$
\S{(a_1,b_1,1),(a_2,b_2,1),\ldots,(a_j,b_j,1),(a_{j+1},b_{j+1},c_{j+1}),\ldots,(a_k,b_k,c_k)}{n}
$$
with $c_{j+1}\neq1,$ \ie this sum has $j-$leading ones. Let $p$ be the expansion of 
$$\S{(a_1,b_1,1),(a_2,b_2,1),\ldots,(a_j,b_j,1)}{n}\S{(a_{j+1},b_{j+1},c_{j+1}),\ldots,(a_k,b_k,c_k)}{n}$$ 
using Theorem \ref{CSprod}. The only cyclotomic harmonic sum with $j-$leading ones in $p$ is $\S{(a_1,b_1,1),\ldots,(a_j,b_j,1),(a_{j+1},b_{j+1},c_{j+1}),\ldots,(a_k,b_k,c_k)}{n},$ all 
the other sums in $p$ have less leading ones. Now define 
$$q:=p-\S{(a_1,b_1,1),\ldots,(a_j,b_j,1),(a_{j+1},b_{j+1},c_{j+1}),\ldots,(a_k,b_k,c_k)}{n},$$
\ie all the cyclotomic harmonic sums in $q$ have less than $j-$leading ones.
We can now write 
\begin{eqnarray*}
&&\S{(a_1,b_1,1),\ldots,(a_j,b_j,1),(a_{j+1},b_{j+1},c_{j+1}),\ldots,(a_k,b_k,c_k)}{n}\\
&&\hspace{2cm}=\S{(a_1,b_1,1),\ldots,(a_j,b_j,1)}{n}\S{(a_{j+1},b_{j+1},c_{j+1}),\ldots,(a_k,b_k,c_k)}{n}-q;
\end{eqnarray*}
in this way we can express $\S{(a_1,b_1,1),\ldots,(a_j,b_j,1),(a_{j+1},b_{j+1},c_{j+1}),\ldots,(a_k,b_k,c_k)}{n}$ using sums with less then $j-$leading ones and 
$\S{(a_1,b_1,1),\ldots,(a_j,b_j,1)}{n}$. Applying the same strategy to the
remaining sums with leading ones which are not {\itshape linear} cyclotomic harmonic sums we can eventually decompose the cyclotomic harmonic sum 
$\S{(a_1,b_1,1),(a_2,b_2,1),\ldots,(a_j,b_j,1),(a_{j+1},b_{j+1},c_{j+1}),\ldots,(a_k,b_k,c_k)}{n}$ in a multivariate polynomial in {\itshape linear} cyclotomic harmonic sum with coefficients 
in the cyclotomic harmonic sums without leading ones. 
\begin{example}
\begin{eqnarray*}
&&\textnormal{S}_{(2,1,1),(3,2,1),(1,0,2)}(n)=-\textnormal{S}_{(1,0,1)}(n)+\frac{1}{2} \textnormal{S}_{(1,0,2)}(n)+\frac{1}{2} \textnormal{S}_{(1,0,2)}(n) \textnormal{S}_{(2,1,1)}(n)\\
&&\hspace{2cm}+2 \textnormal{S}_{(2,1,1)}(n)-\frac{5}{2}\textnormal{S}_{(1,0,2),(2,1,1)}(n)-\textnormal{S}_{(2,1,1)}(n) \textnormal{S}_{(1,0,2),(3,2,1)}(n)\\
&&\hspace{2cm}+3 \textnormal{S}_{(1,0,2),(3,2,1)}(n)-\frac{3}{4} \textnormal{S}_{(2,1,1),(1,0,1)}(n)+\textnormal{S}_{(1,0,2)}(n)\textnormal{S}_{(2,1,1),(3,2,1)}(n)\\
&&\hspace{2cm}+\frac{9}{4} \textnormal{S}_{(2,1,1),(3,2,1)}(n)+\textnormal{S}_{(1,0,2),(3,2,1),(2,1,1)}(n).
\end{eqnarray*}
\end{example}

\begin{remark}
We can always decompose a cyclotomic harmonic sum $\S{(a_1,b_1,c_1),\ldots,(a_k,b_k,c_k)}{n}$ in a multivariate polynomial in {\itshape linear} cyclotomic harmonic sums 
with coefficients in the cyclotomic harmonic sums without leading ones. 
Note, that in a similar way it is possible to extract trailing ones. Hence we can also decompose a cyclotomic harmonic sum $\S{(a_1,b_1,c_1),\ldots,(a_k,b_k,c_k)}{n}$ in a multivariate 
polynomial in {\itshape linear} cyclotomic harmonic sums with coefficients in the cyclotomic harmonic sums without trailing ones. Clearly, restricting to harmonic sums we get just the 
results stated in Section \ref{HSdef}.
\label{CSextractleading1rem}
\end{remark}

\begin{example}
\begin{eqnarray*}
&&\textnormal{S}_{(1,0,2),(2,1,1),(3,2,1)}(n)=-\frac{3}{4} \textnormal{S}_{(1,0,1)}(n)+\frac{1}{2} \textnormal{S}_{(1,0,2)}(n)+\textnormal{S}_{(1,0,2)}(n) \textnormal{S}_{(3,2,1)}(n)\\
&&\hspace{2cm}+\frac{9}{4} \textnormal{S}_{(3,2,1)}(n)-2\textnormal{S}_{(1,0,1),(3,2,1)}(n)-\textnormal{S}_{(3,2,1)}(n) \textnormal{S}_{(2,1,1),(1,0,2)}(n)\\
&&\hspace{2cm}-2 \textnormal{S}_{(2,1,1),(1,0,2)}(n)+\textnormal{S}_{(1,0,2)}(n) \textnormal{S}_{(2,1,1),(3,2,1)}(n)+4\textnormal{S}_{(2,1,1),(3,2,1)}(n)\\
&&\hspace{2cm}+2 \textnormal{S}_{(3,2,1),(1,0,2)}(n)+\textnormal{S}_{(3,2,1),(2,1,1),(1,0,2)}(n).
\end{eqnarray*}
\end{example}

\subsection{Synchronization}
In this subsection we consider cyclotomic harmonic sums with upper summation limit $n+c,$ $kn$ and $kn+c$ for $c\in\Z$ and $k\in \N$.
\begin{lemma}
Let $a_i,n,k,c \in \N, b_i\in\N_0$ and $c_i\in \Z^*$ for $i\in (1,2,\ldots,k).$ Then for $n\geq 0,$
\begin{eqnarray*}
\S{(a_1,b_1,c_1),\ldots,(a_k,b_k,c_k)}{n+c}&=&\S{(a_1,b_1,c_1),\ldots,(a_k,b_k,c_k)}{n}\\
	&&+\sum_{j=1}^c{\frac{\sign{c_1}^{j+n}\S{(a_2,b_2,c_2),\ldots,(a_k,b_k,c_k)}{n+j}}{(a_1 (j+n)+b_1)^{\abs{c_1}}}},
\end{eqnarray*}
and for $n\geq c,$
\begin{eqnarray*}
\S{(a_1,b_1,c_1),\ldots,(a_k,b_k,c_k)}{n-c}&=&\S{(a_1,b_1,c_1),\ldots,(a_k,b_k,c_k)}{n}\\
	&&+\sum_{j=1}^c{\frac{\sign{c_1}^{j+n-c}\S{(a_2,b_2,c_2),\ldots,(a_k,b_k,c_k)}{n-c+j}}{(a_1(j+n-c)+b_1)^{\abs{c_1}}}}.	
\end{eqnarray*}
\end{lemma}

Given a cyclotomic harmonic sum of the form $\S{(a_1,b_1,c_1),\ldots,(a_k,b_k,c_k)}{n+c}$ with $c\in\Z,$ we can apply the previous lemma recursively in order to synchronize the upper summation limit of the 
arising cyclotomic harmonic sums to $n$.

\begin{lemma}
For $a, k \in \N, b\in\N_0$, $c\in \Z^*$, $k\geq 2$  :
\begin{eqnarray*}
\S{(a,b,c)}{k\cdot n}=\sum_{i=0}^{k-1}\sign{c}^i\S{(k\cdot a,b-a \cdot i, \sign{c}^k\abs{c})}{n}.
\end{eqnarray*}
\label{CSmultint1}
\end{lemma}

\begin{proof}
\begin{eqnarray*}
\S{(a,b,c)}{x,k\cdot n}&=&\sum_{j=1}^{k\cdot n}\frac{\sign{c}^j}{(a j +b)^{\abs{c}}}=\\
		&&\hspace{-3cm}=\sum_{j=1}^{n}\frac{\sign{c}^{kj}}{(a (k \cdot j) +b)^{\abs{c}}}+\sum_{j=1}^{n}\frac{\sign{c}^{kj-1}}{(a (k \cdot j - 1) +b)^{\abs{c}}}+\cdots+\sum_{j=1}^{n}\frac{\sign{c}^{kj-(k-1)}}{(a (k \cdot j-(k-1)) +b)^{\abs{c}}}\\
		&&\hspace{-3cm}=\sum_{j=1}^{n}\frac{\sign{c}^{kj}}{((a k) j +b)^{\abs{c}}}+\sum_{j=1}^{n}\frac{\sign{c}^{kj}\sign{c}^{-1}}{((a k) j + (b-a))^{\abs{c}}}+\cdots+\sum_{j=1}^{n}\frac{\sign{c}^{kj}\sign{c}^{-(k-1)}}{((a k) j + (b-(k-1)a))^{\abs{c}}}\\
		&&\hspace{-3cm}=\sum_{i=0}^{k-1}\sign{c}^i\S{(k\cdot a,b-a \cdot i, \sign{c}^k\abs{c})}{n}.
\end{eqnarray*}
\end{proof}

\begin{thm}
For $a_i, m \in \N, b_i, k \in \N_0$, $c_i\in \Z^*$, $k\geq 2$  :
\begin{eqnarray*}
&&\S{(a_m,b_m,c_m),(a_{m-1},b_{m-1}, c_{m-1}),\ldots,(a_1,b_1,c_1)}{k \cdot n}=\\
&&\hspace{3cm}\sum_{i=0}^{m-1}\sum_{j=1}^{n} \frac{\S{(a_{m-1},b_{m-1},c_{m-1}),\ldots,(a_1,b_1,c_1)}{k \cdot j-i}\sign{c_m}^{k\cdot j -i}} {(a_m (k\cdot j-i)+b_1)^{\abs{c_m}}}.
\end{eqnarray*}
\label{CSmultint}
\end{thm}

After applying Theorem \ref{CSmultint} we can synchronize the cyclotomic harmonic sums in the inner sum with upper summation limit $k \cdot j-i$ to upper the summation limit $k \cdot j.$ 
Now we can apply Theorem \ref{CSmultint} to these sums. Repeated application of this procedure leads eventually to cyclotomic harmonic sums with upper summation limit $n.$

\section{Definition and Structure of Cyclotomic Harmonic Polylogarithms}

In Definition \ref{HShlogdef} we defined harmonic polylogarithms and in Definition \ref{SShlogdef} we already extended harmonic polylogarithms to multiple polylogarithms by extending the considered 
alphabet. In this section we will extend the harmonic ploylogarithms from Definition \ref{HShlogdef} into an other direction, \ie we will end up with cyclotomic harmonic polylogarithms. We start by defining the following
auxiliary function: For $a \in \N$ and $b \in \N,$ $b < \varphi(a)$ (here $\varphi(b)$ denotes Euler's totient function \cite{TOTIENT1,TOTIENT2})  we define
\begin{eqnarray}
&&f_a^b:(0,1)\mapsto \R\nonumber\\
&&f_a^b(x)=\left\{ 
		\begin{array}{ll}
				\frac{1}{x}, &  \textnormal{if }a=b=0  \\
				\frac{x^b}{\Phi_a(x)}, & \textnormal{otherwise},
		\end{array} 
		\right.  \nonumber
\end{eqnarray}
where $\Phi_a(x)$ denotes the $a$th cyclotomic polynomial \cite{LANG}:
%----------------------------------------------------------------------------------------------
\begin{eqnarray}
\Phi_a(x) = \frac{x^a-1}{\ds \prod_{d|a, d < a} \Phi_d(x)}.
\end{eqnarray}
%----------------------------------------------------------------------------------------------
The first cyclotomic polynomials are given by
%----------------------------------------------------------------------------------------------
\begin{eqnarray*}
\Phi_1(x) &=& x - 1 \\
\Phi_2(x) &=& x + 1 \\
\Phi_3(x) &=& x^2 + x + 1 \\
\Phi_4(x) &=& x^2 + 1 \\
\Phi_5(x) &=& x^4 + x^3 + x^2 + x+ 1 \\
\Phi_6(x) &=& x^2 - x + 1 \\
\Phi_7(x) &=& x^6 + x^5 + x^4 + x^3 + x^2 + x+ 1 \\
\Phi_8(x) &=& x^4 + 1 \\
\Phi_9(x) &=& x^6 + x^3 + 1 \\
\Phi_{10}(x) &=& x^4 - x^3 + x^2 - x+ 1 \\
\Phi_{11}(x) &=& x^{10} + x^9 + x^8 + x^7 + x^6 + x^5 + x^4 + x^3 + x^2 + x+ 1 \\
\Phi_{12}(x) &=& x^4 - x^2 + 1,~~\text{etc.}
\end{eqnarray*}

Now we are ready to define cyclotomic polylogarithms:

\begin{definition}[Cyclotomic Harmonic Polylogarithms]
Let $m_i=(a_i,b_i) \in \N^2,$ $b_i<\varphi(a_i);$ we define for $x\in (0,1):$
\begin{eqnarray}
\H{}{x}&=&1,\nonumber\\
\H{m_1,\ldots,m_k}{x} &=&\left\{ 
		  	\begin{array}{ll}
						\frac{1}{k!}(\log{x})^k,&  \textnormal{if }m_i=(0,0) \textnormal{ for } 1\leq i \leq k\\
						  &\\
						\int_0^x{f_{a_1}^{b_1}(y) \H{m_2,\ldots,m_k}{y}dy},& \textnormal{otherwise}. 
					\end{array} \right.  \nonumber
\end{eqnarray}
The length $k$ of the vector $\ve m=(m_1,\cdots,m_k)$ is called the weight of the cyclotomic harmonic polylogarithm $\H{\ve m}x.$
\label{CShlogdef}
\end{definition}

\begin{example}
\begin{eqnarray*}
\H{(5,3)}x&=&\int_0^x\frac{y^3}{\Phi_5(y)}dy=\int_0^x\frac{y^3}{y^4 + y^3 + y^2 + y+ 1}dy\\
\vspace{1cm}\\
\H{(3,1),(2,0),(5,2)}x&=&\int_0^x\frac{y}{\Phi_3(y)}\int_0^y\frac{1}{\Phi_2(z)}\int_0^z\frac{w^2}{\Phi_5(w)}dwdzdy\\
\vspace{1cm}\\
\H{(0,0),(2,0),(1,0)}x&=&\int_0^x\frac{1}{y}\int_0^y\frac{1}{z+1}\int_0^z\frac{1}{w-1}dwdzdy=-\H{0,-1,1}x.
\end{eqnarray*}
\end{example}

A cyclotomic harmonic polylogarithm $\H{\ve m}x=\H{m_1,\ldots,m_w}x$ is an analytic functions for $x\in(0,1).$ For the limits $x\rightarrow 0$ and  $x\rightarrow 1$ we 
have (compare Section \ref{HShpldefsec}):
 \begin{itemize}
  \item It follows from the definition that if $m_v \neq (0,0)$ for all $v$ with $1\leq v\leq w$, then $\H{\ve m}0~=~0.$
  \item If $m_1\neq (1,0)$ or if $m_1=(1,0)$ and $m_v=(0,0)$ for all $v$ with $1<v\leq w,$ then $\H{\ve m}1$ is finite.
   \item If $m_1=(1,0)$ and $m_v\neq (0,0)$ for some $v$ with $1<v\leq w$, $\lim_{x\rightarrow 1^-} \H{\ve m}x$ behaves as a combination of powers of $\log(1-x).$
 \end{itemize}
 We define $\H{\ve m}0:=\lim_{x\rightarrow 0^+} \H{\ve m}x$ and $\H{\ve m}1:=\lim_{x\rightarrow 1^-} \H{\ve m}x$ if the limits exist.

\begin{remark}
For the derivatives we have for all $x\in (0,1)$ that $$ \frac{d}{d x} \H{(m_1,n_1),(m_2,n_2),\ldots,(m_k,n_k)}{x}=f_{m_1}^{n_1}(x)\H{(m_2,n_2),\ldots,(m_k,n_k)}{x}. $$ 
\end{remark}

\begin{remark}
For harmonic polylogarithms we considered the alphabet $\{-1,0,1\}.$ The letters $-1$ and $0$ correspond to $(2,0)$ and $(0,0)$ respectively. The letter $1$ corresponds to $(1,0).$ However $1$ indicates 
an iteration of $\frac{1}{1-x}$ while $(1,0)$ indicates an iteration of $\frac{1}{x-1};$ hence there is a change of sign. From now on we will sometimes mix both notations in the obvious way:
\begin{example}
\begin{eqnarray*}
 \H{(0,0),-1,(1,0),(3,1)}x&=&-\H{0,-1,1,(3,1)}x=\H{(0,0),(2,0),(1,0),(3,1)}x.
\end{eqnarray*}
\end{example}
\end{remark}

\begin{remark}
Again the product of two cyclotomic harmonic polylogarithms of the same argument can be expressed using the formula (compare (\ref{HShpro}))
\begin{equation}
\label{CShpro}
\H{\ve p}x\H{\ve q}x=\sum_{\ve r= \ve p \shuffle \ve q}\H{\ve r}x
\end{equation}
in which $\ve p \shuffle \ve q$ represent all merges of $\ve p$ and $\ve q$ in which the relative orders of the elements of $\ve p$ and $\ve q$ are preserved. Hence 
cyclotomic harmonic polylogarithms form a shuffle algebra.
\end{remark}

The number of basis elements spanning the shuffle algebra are given by
%----------------------------------------------------------------------------------------------
\begin{eqnarray}
N^{\rm basic}({w}) = \frac{1}{w} \sum_{d|{w}} \mu\left(\frac{w}{d}\right) M^d, {w} \geq 1
\end{eqnarray}
%----------------------------------------------------------------------------------------------
basis elements according to the 1st Witt formula \cite{Witt1937}. Here $\mu$ denotes the M\"obius function
\cite{MOEBIUS1,MOEBIUS2}. The number of basic cyclotomic harmonic polylogarithms in dependence of $w$ and $M$
is given in Table~\ref{CShreltab}. It applies to any alphabet containing $M$ letters.

%----------------------------------------------------------------------------------------------
\begin{table}[ht]\centering
\begin{tabular}{|r|r|r|r|r|r|r|r|}
\hline
\multicolumn{1}{|c}{          } &
\multicolumn{7}{|c|}{\sf Number of letters } \\
%\cline{2-8}
\multicolumn{1}{|c}{\sf weight} &
\multicolumn{1}{|c}{2} &
\multicolumn{1}{|c}{3} &
\multicolumn{1}{|c}{4} &
\multicolumn{1}{|c}{5} &
\multicolumn{1}{|c}{6} &
\multicolumn{1}{|c}{7} &
\multicolumn{1}{|c|}{8} \\
\hline
      1 &   2 &   3 &    4 &     5 &      6 &      7 &       8  \\
      2 &   1 &   3 &    6 &    10 &     15 &     21 &      28  \\
      3 &   2 &   8 &   20 &    40 &     70 &    112 &     168  \\
      4 &   3 &  18 &   60 &   150 &    315 &    588 &    1008  \\
      5 &   6 &  48 &  204 &   624 &   1554 &   3360 &    6552  \\
      6 &   9 & 116 &  670 &  2580 &   7735 &  19544 &   43596  \\
      7 &  18 & 312 & 2340 & 11160 &  39990 & 117648 &  299592  \\
      8 &  30 & 810 & 8160 & 48750 & 209790 & 729300 & 2096640  \\
\hline
\end{tabular}
\caption{\label{CShreltab}Number of basic cyclotomic harmonic polylogarithms in dependence of the
           number of letters and weight.}
\end{table}

\begin{example}
 We consider the alphabet $\{(0,0),(2,0),(4,1),(5,3)\}$ at weight two. We obtain the 6 basis cyclotomic harmonic polylogarithms
$$ \textnormal{H}_{(2,0),(0,0)}(x),\textnormal{H}_{(4,1),(0,0)}(x),\textnormal{H}_{(4,1),(2,0)}(x),\textnormal{H}_{(5,3),(0,0)}(x),
  \textnormal{H}_{(5,3),(2,0)}(x),\textnormal{H}_{(5,3),(4,1)}(x),$$
together with the relations
\begin{eqnarray*}
\textnormal{H}_{(0,0),(0,0)}(x)&=& \frac{1}{2} \textnormal{H}_{(0,0)}(x){}^2,\\
\textnormal{H}_{(0,0),(2,0)}(x)&=& \textnormal{H}_{(0,0)}(x) \textnormal{H}_{(2,0)}(x)-\textnormal{H}_{(2,0),(0,0)}(x),\\
\textnormal{H}_{(0,0),(4,1)}(x)&=& \textnormal{H}_{(0,0)}(x) \textnormal{H}_{(4,1)}(x)-\textnormal{H}_{(4,1),(0,0)}(x),\\
\textnormal{H}_{(0,0),(5,3)}(x)&=& \textnormal{H}_{(0,0)}(x) \textnormal{H}_{(5,3)}(x)-\textnormal{H}_{(5,3),(0,0)}(x),\\
\textnormal{H}_{(2,0),(2,0)}(x)&=& \frac{1}{2} \textnormal{H}_{(2,0)}(x){}^2,\\
\textnormal{H}_{(2,0),(4,1)}(x)&=& \textnormal{H}_{(2,0)}(x) \textnormal{H}_{(4,1)}(x)-\textnormal{H}_{(4,1),(2,0)}(x),\\
\textnormal{H}_{(2,0),(5,3)}(x)&=& \textnormal{H}_{(2,0)}(x) \textnormal{H}_{(5,3)}(x)-\textnormal{H}_{(5,3),(2,0)}(x),\\
\textnormal{H}_{(4,1),(4,1)}(x)&=& \frac{1}{2} \textnormal{H}_{(4,1)}(x){}^2,\\
\textnormal{H}_{(4,1),(5,3)}(x)&=& \textnormal{H}_{(4,1)}(x) \textnormal{H}_{(5,3)}(x)-\textnormal{H}_{(5,3),(4,1)}(x),\\
\textnormal{H}_{(5,3),(5,3)}(x)&=& \frac{1}{2} \textnormal{H}_{(5,3)}(x){}^2.
\end{eqnarray*}
\end{example}

\section{Identities between Cyclotomic Harmonic Polylogarithms of Related Arguments}
\label{CSRelatedArguments}
\subsection{\texorpdfstring{$\frac{1}{x}\rightarrow x$}{1/x->x}}
\label{CS1dxx}
First we consider just cyclotomic harmonic polylogarithms with indices not equal to $(1,0).$
Proceeding recursively on the weight $w$ of the cyclotomic harmonic polylogarithm we have for $0<x<1:$
\begin{eqnarray*}
\H{(0,0)}{\frac{1}{x}}&=&-\H{(0,0)}{x},% \label{CStrafo1dx1}
\end{eqnarray*}
 and in addition, for $a\in \N,b\in \N_0, a>1,b<\varphi{(a)}$ we can use
\begin{eqnarray}
\int_0^{\frac{1}{x}}\frac{y^b}{\Phi_a(y)}dy=\H{(a,b)}1-\int_{\frac{1}{x}}^1\frac{y^b}{\Phi_a(y)}dy=\H{(a,b)}1+\int_{x}^1\frac{\frac{1}{y^b}}{y^2 \Phi_a(\frac{1}{y})}dy.
\end{eqnarray}
At this point we can perform a partial fraction decomposition on the integrand and finally rewrite the resulting integrals in terms of cyclotomic harmonic polylogarithms at $x$ and $1$.
Now let us look at higher weights $w>1.$ We consider $\H{m_1,m_2,\ldots,m_w}{\frac{1}{x}}$ and suppose that we can already apply the transformation for cyclotomic harmonic polylogarithms 
of weight $<w.$ For $(a,b)\neq (1,0)$ we get:
\begin{eqnarray*}
\H{(0,0),m_2,\ldots,m_w}{\frac{1}{x}}&=&\H{(0,0),m_2,\ldots,m_w}1+\int_x^1\frac{1}{t^2(1/t)}\H{m_2,\ldots,m_w}{\frac{1}{t}}dt\\
				  &=&\H{(0,0),m_2,\ldots,m_w}1+\int_x^1\frac{1}{t}\H{m_2,\ldots,m_w}{\frac{1}{t}}dt\\
\H{(a,b),m_2,\ldots,m_w}{\frac{1}{x}}&=&\H{(a,b),m_2,\ldots,m_w}1+\int_x^1\frac{\frac{1}{y^b}}{y^2\Phi_a(\frac{1}{y})}\H{m_2,\ldots,m_w}{\frac{1}{y}}dy.
\end{eqnarray*}
At this point we again have to perform a partial fraction decomposition on the integrand and since we know the transformation for weights $<w$ we can apply it to $\H{m_2,\ldots,m_w}{\frac{1}{t}}$ 
and finally we obtain the required weight $w$ identity by using the definition of the cyclotomic harmonic polylogarithms.\\
The index $(1,0)$ in the index set leads to a branch point at $1$ and a branch cut discontinuity in the complex plane for $x\in(1,\infty).$ This corresponds to the branch point at $x=1$ 
and the branch cut discontinuity in the complex plane for $x\in(1,\infty)$ of $\log(1-x)=\H{(1,0)}x.$  However the analytic properties of the logarithm are well known and we
can set for $0<x<1$ for instance
\begin{eqnarray}
\H{(1,0)}{\frac{1}{x}}&=&\H{(1,0)}{x}-\H{(0,0)}{x}+i\pi \label{CStrafo1dx11}
\end{eqnarray}
by approaching the argument $\frac{1}{x}$ form the lower half complex plane.
The strategy now is as follows: if a cyclotomic harmonic polylogarithm has leading $(1,0),$ we remove them and end up with cyclotomic harmonic polylogarithms without leading $(1,0)$ and 
powers of $\H{(1,0)}{\frac{1}{x}}.$ We know how to deal 
with the cyclotomic harmonic polylogarithms without leading $(1,0)$ due to the previous part of this section and for the powers of $\H{(1,0)}{\frac{1}{x}}$ we can 
use~(\ref{CStrafo1dx11}).

\begin{example}
\begin{eqnarray*}
\textnormal{H}_{(3,1),(5,2)}(x)&=& \left(\textnormal{H}_{(5,0)}(1)+\textnormal{H}_{(5,2)}(1)\right)\left(-\textnormal{H}_{(0,0)}\left(\frac{1}{x}\right)\right)\\
&&-\left(\textnormal{H}_{(5,0)}(1)+\textnormal{H}_{(5,2)}(1)\right)\left(\textnormal{H}_{(3,0)}(1)-\textnormal{H}_{(3,0)}\left(\frac{1}{x}\right)\right)\\
&&-\left(\textnormal{H}_{(5,0)}(1)+\textnormal{H}_{(5,2)}(1)\right)\left(\textnormal{H}_{(3,1)}(1)-\textnormal{H}_{(3,1)}\left(\frac{1}{x}\right)\right)\\
&&+\textnormal{H}_{(0,0),(5,0)}\left(\frac{1}{x}\right)-\textnormal{H}_{(3,0),(5,0)}\left(\frac{1}{x}\right)-\textnormal{H}_{(3,1),(5,0)}\left(\frac{1}{x}\right)\\
&&-\textnormal{H}_{(0,0),(5,0)}(1)+\textnormal{H}_{(3,0),(5,0)}(1)+\textnormal{H}_{(3,1),(5,0)}(1)+\textnormal{H}_{(3,1),(5,2)}(1).
\end{eqnarray*}
\end{example}

\subsection{\texorpdfstring{$\frac{1-x}{1+x} \rightarrow x$}{(1-x)/(1+x)->x}}
\label{CS1x1x}
We restrict the index set now to letters out of $\{1,0,-1,(4,0),(4,1)\},$
and proceed recursively on the weight $w$ of the cyclotomic harmonic polylogarithm. For the base cases we have
\begin{eqnarray}
\H{-1}{\frac{1-x}{1+x}}&=&\H{-1}{1}-\H{-1}x\\
\H{0}{\frac{1-x}{1+x}}&=&-\H{1}{x}+\H{-1}x\\
\H{1}{\frac{1-x}{1+x}}&=&-\H{-1}{1}-\H{0}{x}+\H{-1}{x}\label{CStrafo1x1x}\\
\H{(4,0)}{\frac{1-x}{1+x}}&=&\H{(4,0)}{1}+\H{(4,0)}{x}\\
\H{(4,1)}{\frac{1-x}{1+x}}&=&-\H{(4,1)}{1}-\H{-1}{x}+\H{(4,1)}{x}.
\end{eqnarray}
Now let us look at higher weights $w>1.$ We consider $\H{m_1,m_2,\ldots,m_w}{\frac{1-x}{1+x}}$ with $m_i\in\{1,0,-1,(4,0),(4,1)\}$ and suppose that we can already apply the transformation for harmonic polylogarithms
 of weight $<w.$ If $m_1=1,$ we can remove leading ones and end up with harmonic polylogarithms without leading ones and powers of $\H{1}{\frac{1-x}{1+x}}.$ For the powers of $\H{1}{\frac{1-x}{1+x}}$ we
 can use (\ref{CStrafo1x1x}); therefore, only the cases in which the first index $m_1\neq 1$ are to be considered. We get (compare \cite{Remiddi2000} and Sections \ref{HS1x1x} and \ref{SS1x1x}):
\begin{eqnarray*}
\H{-1,m_2,\ldots,m_w}{\frac{1-x}{1+x}}&=&\H{-1,m_2,\ldots,m_w}1-\int_0^x\frac{1}{1+t}\H{m_2,\ldots,m_w}{\frac{1-t}{1+t}}dt\\
\H{0,m_2,\ldots,m_w}{\frac{1-x}{1+x}}&=&\H{0,m_2,\ldots,m_w}1-\int_0^x\frac{1}{1-t}\H{m_2,\ldots,m_w}{\frac{1-t}{1+t}}dt\\
		&&-\int_0^x\frac{1}{1-t}\H{m_2,\ldots,m_w}{\frac{1+t}{1+t}}dt\\
\H{(4,0),m_2,\ldots,m_w}{\frac{1-x}{1+x}}&=&\H{(4,0),m_2,\ldots,m_w}1-\int_0^x\frac{1}{1+t^2}\H{m_2,\ldots,m_w}{\frac{1-t}{1+t}}dt\\
\H{(4,1),m_2,\ldots,m_w}{\frac{1-x}{1+x}}&=&\H{(4,1),m_2,\ldots,m_w}1-\int_0^x\frac{1}{1+t}\H{m_2,\ldots,m_w}{\frac{1-t}{1+t}}dt\\
					& &+\int_0^x\frac{t}{1+t^2}\H{m_2,\ldots,m_w}{\frac{1-t}{1+t}}dt.
\end{eqnarray*}
Since we know the transform for weights $<w$ we can apply it to $\H{m_2,\ldots,m_w}{\frac{1+t}{1+t}}$ and finally we obtain the required weight $w$ identity by using the definition of the cyclotomic 
harmonic polylogarithms.

\section{Power Series Expansion of Cyclotomic Harmonic Polylogarithms}
In order to be able to compute the power series expansion of cyclotomic harmonic polylogarithms, we state the following lemma.
\begin{lemma}
 Let $\Phi_n(x)$ be the $n-$th cyclotomic polynomial and let $\sum_{i\geq 0}{f_i x^i}$ be the power series expansion of $f(x):=\frac{1}{\Phi_n(x)}$ about zero.
Then the sequence $(f_i)_{i\geq 0}$ is periodic with period $n$. For $0\leq i<n$ the $i$-th coefficient $f_i$ equals the coefficient of $x^i$ in $-\prod_{d|n, d < n} \Phi_d(x).$
\label{CScyclotopow}
\end{lemma}

\begin{proof}
We have that $$f(x)=\frac{1}{\Phi_n(x)} = \frac{\ds \prod_{d|n, d < n} \Phi_d(x)}{x^n-1},$$ hence 
\begin{eqnarray*}
\prod_{d|n, d < n} \Phi_d(x)&=&(x^n-1) f(x)=(x^n-1) \sum_{i\geq 0}{f_i x^i}=\sum_{i\geq 0}{f_i x^{n+i}}-\sum_{i\geq 0}{f_i x^i}\\
      &=&\sum_{i\geq 0}{f_i x^{n+i}}-\sum_{i=0}^{n-1}{f_i x^i}-\sum_{i\geq 0}{f_{n+i} x^{n+i}}\\
      &=&\sum_{i\geq 0}{(f_i-f_{n+i}) x^{n+i}}-\sum_{i=0}^{n-1}{f_i x^i}.
\end{eqnarray*}
Since the degree of $\prod_{d|n, d < n} \Phi_d(x)$ is $n-\varphi(n)<n,$ we get that $f_{n+i}-f_n=0$ and hence the sequence $(f_i)_{i\geq 0}$ is periodic with period $n.$
\end{proof}

In general, the cyclotomic harmonic polylogarithms $\H{\ve m}x$ do not have a regular Taylor series expansion. This is due to the
effect that trailing zeroes, $\ie$ the letter $(0,0)$ in the index set may cause powers of $\log(x).$ Hence the proper expansion
is one in terms of both $x$ and $\log(x)$. We first look at the cyclotomic harmonic polylogarithms of depth one. Let $a, b \in \N;$ and $0<x<1;$ due
 to the previous lemma we can write $${\frac{1}{\Phi_a(x)}}=\sum_{q=0}^{a-1}f_q\sum_{i=0}^{\infty}x^{a i +q}.$$
Hence we get
\begin{eqnarray*}
\H{(a,b)}x&=&\int_0^x{\frac{y^b}{\Phi_a(y)}}dy=\sum_{q=0}^{a-1}f_q\int_0^x\sum_{i=0}^{\infty}y^{a i+q+b}\\
	    &=&\sum_{q=0}^{a-1}f_q\sum_{i=0}^{\infty}\frac{ x^{a i+q+b+1}}{a i+q+b+1}\\
	    &=&\sum_{q=0}^{a-1}f_q\sum_{i=1}^{\infty}\frac{x^{a i+q-a+b+1}}{a i+q-a+b+1}.
\end{eqnarray*}
\begin{remark}
If $z=k a$ for some $k \in \N$ we can write this as well in the form:
\begin{eqnarray*}
\H{(a,b)}x &=&\sum_{j=1}^k\sum_{q=0}^{a-1}f_q\sum_{i=1}^{\infty}\frac{x^{z i+a(j-1)+q-z+b+1}}{z i+a(j-1)+q-z+b+1}.
\end{eqnarray*}
\label{CScyclotopowremark}
\end{remark}
\begin{lemma}
Let $\H{\ve m}x$ be a cyclotomic harmonic polylogarithm with depth $d.$ We assume that its power series expansion is of the form
$$\H{\ve m}x=\sum_{j=1}^w\sum_{i=1}^{\infty}\frac{x^{z i+c_j}}{(z i+c_j)^{g_j}}\S{\ve n_j}i$$
for $x \in (0,1)$, $w,g_j\in \N$ and $c_j\in \Z.$
We can get the depth $d+1$ expansion of $\H{(a,b),\ve m}x$ provided that $a=0$ or $a k=z$ for some $k \in \N$ using
\begin{eqnarray*}
\H{(0,0),\ve m}x&=&\sum_{i=1}^{\infty}\frac{x^{z i+c_j}}{(z i+c_j)^{g_j+1}}\S{\ve n_j}i\\
%\H{(a,b),\ve m}x&=&\sum_{j=1}^w{\sum_{p=0}^k\sum_{q=0}^{a-1}f_q\sum_{i=1}^{\infty}\frac{ x^{z i+a(i-1) +q+b+c_j+1}}{(z i+a(i-1)+q+b+c_j+1)^{g_j}}\S{\ve n_j}i}
\H{(a,b),\ve m}x&=&\sum_{j=1}^w{\sum_{q=0}^{z-1}f_q\sum_{i=1}^{\infty}\frac{ x^{z i+q+b+c_j+1}}{(z i+q+b+c_j+1)}\S{(z,c_j,g_j),\ve n_j}i}
\end{eqnarray*}
where $$\frac{1}{\Phi_a(x)}=\sum_{q=0}^{z-1}f_q\sum_{p=0}^{\infty}x^{z p +q}.$$
\end{lemma}
\begin{proof}
We start with $\H{(0,0),\ve m}x$. We have to consider $w$ integrals of the form
\begin{eqnarray*}
\int_0^x{\frac{1}{y}\sum_{i=1}^{\infty}\frac{y^{z i+c}}{(z i+c)^g}\S{\ve n}i}dy&=&\sum_{i=1}^{\infty}\frac{1}{(z i+c)^g}\S{\ve n}i\int_0^x{y^{z i+c-1}}dy\\
	&=&\sum_{i=1}^{\infty}\frac{1}{(z i+c)^g}\S{\ve n}i\frac{x^{z i+c}}{z i+c}dy\\
	&=&\sum_{i=1}^{\infty}\frac{x^{z i+c}}{(z i+c)^{g+1}}\S{\ve n}i.
\end{eqnarray*}
For $\H{(a,b),\ve m}x$ with $a k=z$ we have to consider $w$ integrals of the type
$$
\int_0^x{\frac{y^b}{\Phi_a(y)}\sum_{i=1}^{\infty}\frac{y^{z i+c}}{(z i+c)^g}\S{\ve n}i}dy.
$$
We use Lemma \ref{CScyclotopow} and write $$\frac{1}{\Phi_a(x)}=\sum_{q=0}^{z-1}f_q\sum_{p=0}^{\infty}x^{z p +q}.$$
Now we get
\begin{eqnarray*}
 &&\int_0^x{\frac{y^b}{\Phi_a(y)}\sum_{i=1}^{\infty}\frac{y^{z i+c}}{(z i+c)^g}\S{\ve n}i}dy\\
 &&\hspace{2cm}=\int_0^xy^b\sum_{q=0}^{z-1}f_q\sum_{p=0}^{\infty}y^{z p +q}\sum_{i=1}^{\infty}\frac{y^{z i+c}}{(z i+c)^g}\S{\ve n}idy\\
 &&\hspace{2cm}=\sum_{q=0}^{z-1}f_q\int_0^x\sum_{p=0}^{\infty}y^{z p +q+b}\sum_{i=0}^{\infty}\frac{y^{z i+z+c}}{(z i+z+c)^g}\S{\ve n}{i+1}dy\\
 &&\hspace{2cm}=\sum_{q=0}^{z-1}f_q\int_0^x\sum_{i=0}^{\infty}\sum_{p=0}^{i}y^{z(i-p) +q+b}\frac{y^{z p+z+c}}{(z p+z+c)^g}\S{\ve n}{p+1}dy\\
 &&\hspace{2cm}=\sum_{q=0}^{z-1}f_q\int_0^x\sum_{i=0}^{\infty}y^{z i +q+b+c+z}\S{(z,c,g),\ve n}{i+1}dy\\
 &&\hspace{2cm}=\sum_{q=0}^{z-1}f_q\sum_{i=0}^{\infty}\frac{x^{z i +q+b+c+z+1}}{z i +q+b+c+z+1}\S{(z,c,g),\ve n}{i+1}dy\\
 &&\hspace{2cm}=\sum_{q=0}^{z-1}f_q\sum_{i=1}^{\infty}\frac{x^{z i +q+b+c+1}}{z i +q+b+c+1}\S{(z,c,g),\ve n}{i}dy.\\
\end{eqnarray*}
\end{proof}

\begin{remark}
 Using the above strategy we can get the expansion of a cyclotomic harmonic polylogarithm $\H{(a_d,b_d),(a_{d-1},b_{d-1}),\ldots,(a_1,b_1)}x$
 up to arbitrary weight $d$. If we set z equal to the least common multiple of the $a_i$ which are not zero (this is possible due to Lemma 
\ref{CScyclotopowremark}) we can guarantee that the induction step works. 
\end{remark}

\begin{example}
\begin{eqnarray*}
\textnormal{H}_{(3,1),(1,0)}(x)&=& \sum _{i=1}^\infty \frac{x^{3 i+1} \textnormal{S}_{(3,1,1)}(i)}{3 i+1}-\sum _{i=1}^\infty \frac{x^{3 i+1} \textnormal{S}_{(3,2,1)}(i)}{3 i+1}
		+\sum _{i=1}^\infty \frac{x^{3 i+2}\textnormal{S}_{(3,2,1)}(i)}{3 i+2}\\
      &&-\frac{1}{3} \sum _{i=1}^\infty \frac{x^{3 i} \textnormal{S}_{(3,1,1)}(i)}{i}-\frac{1}{3} \sum _{i=1}^\infty \frac{\textnormal{S}_1(i) x^{3 i+2}}{3i+2}
		+\frac{1}{3} \sum _{i=1}^\infty \frac{\textnormal{S}_1(i) x^{3 i+3}}{3 i+3}\\
      &&+\frac{3}{2} \sum _{i=1}^\infty \frac{x^{3 i+1}}{3 i+1}-\sum _{i=1}^\infty \frac{x^{3i+1}}{(3 i+1)^2}-\sum _{i=1}^\infty \frac{x^{3 i+1}}{3 i+2}
		+\frac{1}{2} \sum _{i=1}^\infty \frac{x^{3 i+2}}{3 i+2}\\
      &&-\sum _{i=1}^\infty \frac{x^{3 i+2}}{(3i+2)^2}-\sum _{i=1}^\infty \frac{x^{3 i}}{3 i+1}.
\end{eqnarray*}
\end{example}

\subsection{Asymptotic Behavior of Extended Harmonic Polylogarithms}
\label{CSasybeh}
Combining Section \ref{CS1dxx} together with the power series expansion of cyclotomic harmonic polylogarithms we can determine the asymptotic behavior of cyclotomic harmonic polylogarithms. Let us look at the cyclotomic harmonic polylogarithm 
$\H{\ve m}x$ and define $y:=\frac{1}{x}.$ Using Section \ref{CS1dxx} on $\H{\ve m}{\frac{1}{y}}=\H{\ve m}x$ we can rewrite $\H{\ve m}x$ in terms of cyclotomic harmonic polylogarithms at argument $y$ together with some constants.
Now we can get the power series expansion of the harmonic polylogarithms at argument $y$ about 0 easily using the previous part of this Section. Since sending $x$ to infinity corresponds to sending $y$ to zero,
 we get the asymptotic behavior of $\H{\ve m}x.$
\begin{example}
\begin{eqnarray*}
\textnormal{H}_{(3,0),(3,1)}(x)&=&-\textnormal{H}_{(0,0)}\left(x\right) \left(\sum _{i=1}^\infty \frac{\left(\frac{1}{x}\right)^{3 i-2}}{3 i-2}
	-\sum _{i=1}^\infty\frac{\left(\frac{1}{x}\right)^{3 i-1}}{3 i-1}\right)\\
      &&-\textnormal{H}_{(3,0)}(1) \left(-\sum _{i=1}^\infty \frac{\left(\frac{1}{x}\right)^{3 i-2}}{3 i-2}+\sum_{i=1}^\infty \frac{\left(\frac{1}{x}\right)^{3 i-1}}{3 i-1}+\textnormal{H}_{(3,0)}(1)\right)\\
      &&+\frac{1}{3} \sum _{i=1}^\infty \frac{\left(\frac{1}{x}\right)^{3 i} S_{(3,1,1)}(i)}{i}-\sum _{i=1}^\infty \frac{\left(\frac{1}{x}\right)^{3 i-1} S_{(3,1,1)}(i)}{3 i-1}\\
      &&+\frac{1}{3} \sum _{i=1}^\infty \frac{S_1(i)\left(\frac{1}{x}\right)^{3 i+1}}{3 i+1}-\frac{1}{3} \sum _{i=1}^\infty \frac{S_1(i) \left(\frac{1}{x}\right)^{3 i+2}}{3 i+2}
	-\sum _{i=1}^\infty\frac{\left(\frac{1}{x}\right)^{3 i-2}}{(3 i-2)^2}\\
      &&-\frac{1}{2} \sum _{i=1}^\infty \frac{\left(\frac{1}{x}\right)^{3 i-1}}{3 i-1}+\sum _{i=1}^\infty\frac{\left(\frac{1}{x}\right)^{3 i-1}}{(3 i-1)^2}
	+\sum _{i=1}^\infty \frac{\left(\frac{1}{x}\right)^{3 i}}{3 i+1}-\frac{1}{2} \sum _{i=1}^\infty\frac{\left(\frac{1}{x}\right)^{3 i-1}}{3 i+1}\\
      &&-\textnormal{H}_{(3,0),(0,0)}(1)+\textnormal{H}_{(3,0),(3,0)}(1)+2 \textnormal{H}_{(3,0),(3,1)}(1).
\end{eqnarray*}
\end{example}

\subsection{Values of Cyclotomic Harmonic Polylogarithms at 1 Expressed by Cyclotomic Harmonic Sums at Infinity}

As worked out in the previous part of this section, the expansion of the cyclotomic harmonic polylogarithms without trailing zeros is a combination of sums of the form:
$$
\sum_{i=1}^\infty x^{ai+b} \frac{\S{\ve n}i}{(ai+b)^c}, \ a,b,c \in \N.
$$
For $x\rightarrow1$ these sums turn into cyclotomic harmonic sums at infinity if $c\neq1$:
$$
\sum_{i=1}^\infty x^{ai+b} \frac{ \S{\ve n}i}{(ai+b)^c}\rightarrow \S{(a,b,c),\ve n}{\infty}.
$$
Hence the values of cyclotomic harmonic polylogarithms at one are related to the values of the cyclotomic harmonic sums at infinity.
\begin{example}
$$
\H{(0,0),(3,2)}{1}=\frac{1}{9}\S{(1,0,2)}{\infty}-\S{(3,1,2)}{\infty}.
$$
\end{example}

If $c=1,$ these sums turn into
$$
\sum_{i=1}^\infty x^{ai+b} \frac{\S{\ve n}i}{(ai+b)}.
$$
Sending $x$ to one gives:
$$
\lim_{x\rightarrow1}\sum_{i=1}^\infty x^{ai+b} \frac{\S{\ve n}i}{(ai+b)}=\infty
$$
We see that these limits do not exist: this corresponds to the infiniteness of the cyclotomic harmonic sums with leading ones: $\lim_{k\rightarrow \infty}\S{(a,b,1),\ve n}{k}=\infty.$

\section{Integral Representation of Cyclotomic Harmonic Sums}
\label{CSintrep}
In this section we look at integral representations of cyclotomic harmonic sums. It will turn out that we can find representations in form of Mellin-type transforms of cyclotomic 
polylogarithms. The integral representation of cyclotomic harmonic sums of depth one can be derived easily: 
\begin{lemma}
Let $a,c\in\N,b\in \N_0,$ $d \in\{1,-1\}$ and $n\in\N;$ then
\begin{eqnarray*}
\SS{(a,b,1)}dn&=&\int_0^{1}{\frac{{x_1}^{a+b-1}\left(d^n{x_1}^{a n}-1\right)}{{x_1}^a-d}dx_1}\\
\SS{(a,b,2)}dn&=&\int_0^{1}{\frac{1}{x_2}\int_0^{x_2}{\frac{{x_1}^{a+b-1}\left( d^n {x_1}^{a n}-1\right)}{{x_1}^a-d}dx_1}dx_2}\\
\SS{(a,b,c)}dn&=&\int_0^1{\frac{1}{x_c}\int_0^{x_c}{\frac{1}{x_{c-1}} \cdots \int_0^{x_3}\frac{1}{x_2}{\int_0^{x_2}{\frac{{x_1}^{a+b-1}\left(d^n {x_1}^{a n}-1\right)}{{x_1}^a-d}dx_1}}\cdots}dx_c}.
\end{eqnarray*}
\label{CSintrep1}
\end{lemma}

Let us now look at the integral representation of cyclotomic harmonic sums of higher depths. We consider the sum 
$\S{(a_1,b_1,c_1),(a_2,b_2,c_2),\ldots,(a_k,b_k,c_k)}n$ and apply Lemma \ref{CSintrep1} to the innermost sum $(a = a_k, b = b_k, c = \abs{c_k}, d=\sign{c_k})$. 
One now may perform the next sum in the
same way, provided $a_{k-1} | a_k$. At this point we may need the fact:
\begin{eqnarray*}
\sum_{i=1}^n\frac{(d y^a)^i}{(ai+b)^c}=\frac{1}{y^b}\int_0^y{\frac{1}{x_c}\int_0^{x_c}{\frac{1}{x_{c-1}} \cdots \int_0^{x_3}\frac{1}{x_2}{
  \int_0^{x_2}{\frac{{x_1}^{a+b-1}\left(d^n {x_1}^{a n}-1\right)}{{x_1}^a-d}dx_1}}\cdots }dx_c}
\end{eqnarray*}
for $n,a,b,c,k\in \N,d=\pm1$ and $y\in\R.$ If $a_{k-1} \nmid a_k,$ one transforms the integration variables such that the next denominator can be generated, etc.
In this way, the sum $ \S{(a_1,b_1,c_1),(a_2,b_2,c_2),\ldots,(a_k,b_k,c_k)}n $
can be represented in terms of linear combinations of Poincar\'e-iterated integrals. Evidently, the representation of
the cyclotomic harmonic sum $\S{(a_1,b_1,c_1),(a_2,b_2,c_2),\ldots,(a_k,b_k,c_k)}n$ in terms of a (properly regularized) Mellin transform will be related to the
Mellin variable $k n$, with $k$ beingthe least common multiple of $a_1, ... ,a_k$.
 
Let us illustrate the principle steps in case of the following example (compare~\cite{Ablinger2011}):
%----------------------------------------------------------------------------------------------
\begin{eqnarray*}
\S{(3,2,2),(2,1,-1)}{n} = \sum_{k=1}^{n} \frac{1}{(3k+2)^2}\sum_{l=1}^{k} \frac{(-1)^l}{(2l+1)}~.
\end{eqnarray*}
%----------------------------------------------------------------------------------------------
The first sum yields
%----------------------------------------------------------------------------------------------
\begin{eqnarray*}
\S{(3,2,2),(2,1,-1)}{n} =  \sum_{k=1}^{n} \int_0^1 \frac{1}{(3k + 2)^2} \frac{x^2((-x^2)^{k} - 1)}{x^2+1}dx~.
\end{eqnarray*}
%----------------------------------------------------------------------------------------------
Setting $x = y^3$ one obtains
%----------------------------------------------------------------------------------------------
\begin{eqnarray}
\label{CSeq:ex1}
\S{(3,2,2),(2,1,-1)}{n} &=& 12 \int_0^1\frac{y^8}{y^6+1}
\sum_{k=1}^n \frac{(-y^6)^k-1}{(6k+4)^2}dy
\nonumber\\
&=& 12\int_0^1\frac{y^4}{y^6+1} \Biggl\{
\int_0^y
\frac{1}{z} \int_0^z t^9~\frac{(-t^6)^{n}-1}{t^6+1}dtdz
\nonumber\\ && \hspace*{2.9cm}
- y^4 \int_0^1
\frac{1}{z} \int_0^z t^9~\frac{t^{6n}-1}{t^6-1}dtdz
\Biggr\}dy \nonumber\\
&=& 12\int_0^1 \frac{y^4}{y^6+1}
\int_0^y
\frac{1}{z} \int_0^z t^9~\frac{(-t^6)^{n}-1}{t^6+1}dtdzdy
\nonumber\\ &&
- \left(4 -\pi \right) \int_0^1
\frac{1}{z} \int_0^z t^9~\frac{t^{6n}-1}{t^6-1}dtdz
\Biggr\}~.
\end{eqnarray}
%----------------------------------------------------------------------------------------------
 In general, the polynomials $$x_1^a \pm 1$$ in Lemma \ref{CSintrep1} are either cyclotomic or decompose into products of cyclotomic
 polynomials in other cases. All factors divide $(x^a)^l - 1$, resp. $(-x^a)^l - 1$. We remark that (\ref{CSeq:ex1}) is not yet written in terms of a Mellin
 transform. This can be achieved using partial fractioning and integration by parts (for details we refer to the next sections):

\begin{eqnarray*}
\label{CSeq:ex2}
&&\hspace{-1cm}\S{(3,2,2),(2,1,-1)}{n} =\\&& \frac{1}{6}(4 - \pi) \int_0^1 x^3(x^{6n}-1)
\left[6
+ f_1^0(x)
- f_2^0(x)
- 2 f_3^0(x) \right. \\ &&
\hspace*{2cm}
- f_3^1(x) \left.
- 2 f_6^0(x)
+ f_6^1(x)\right] \H{(0,0)}x dx\\ &&
-2 \int_0^1 x^3 \left[(-1)^n x^{6n} - 1\right]\left[3 - f_4^0(x) - 2 f_{12}^0(x) + 2
f_{12}^2(x) \right] \H{(0,0)}x dx
%\end{eqnarray}
%\begin{eqnarray}
\\ 
&&
-\frac{4}{3} \left[\H{(0,0),(4,0)}1-\H{(0,0),(12,0)}1+2\H{(0,0),(12,2)}1 \right]
\int_0^1  x^3 \left[(-1)^n x^{6n} - 1\right] \\
&&
\hspace*{5cm} \times \left[3 - f_4^0(x) - 2 f_{12}^0(x) + 2
f_{12}^2(x) \right]dx \\ &&
+\frac{4}{3}
\int_0^1 x^3 \left[(-1)^n x^{6n} - 1\right]
\left[\H{(0,0),(4,0)}x-\H{(0,0),(12,0)}x+2\H{(0,0),(12,2)}x \right]
\\
&&
\hspace*{5cm}
\times \left[3 - f_4^0(x) - 2 f_{12}^0(x) + 2
f_{12}^2(x) \right]dx~.
\end{eqnarray*}

\subsection{Mellin Transform of Cyclotomic Harmonic Polylogarithms}
\label{CSCyloMel}
In this subsection we look at the Mellin transformation of cyclotomic polylogarithms. Therefore we extend the Mellin transform to cyclotomic harmonic polylogarithms (compare \cite{Ablinger2011}).
\begin{definition}
Let $h(x)$ be a cyclotomic harmonic polylogarithm, $n,k\in\N;p\in\N^*$ with $p<k$ and $a\in\N,a>1$. We extend the Mellin-transform as follows:
\begin{eqnarray}
\M{h(x)}{k n+p}&=&\int_0^1{x^{k\;n+p}h(x)dx},\nonumber\\
\M{\frac{h(x)}{\Phi_a(x)}}{k\;n+p}&=&\int_0^1{\frac{x^{k n+p}h(x)}{\Phi_a(x)}dx},\nonumber\\
\M{\frac{h(x)}{x-1}}{k\;n+p}&=&\int_0^1{\frac{(x^{k n+p}-1)h(x)}{x-1}dx}.\nonumber
\label{CSabmellplus}
\end{eqnarray}
\end{definition}
We cannot compute the Mellin transformation
$$
\M{\H{(a_1,b_1),(a_2,b_2),\ldots}x}{k\;n+p}
$$ 
of a cyclotomic polylogarithm $\H{(a_1,b_1),(a_2,b_2),\ldots}x$ for general $n \in \N.$ However we will always be able to calculate the Mellin transformation of
$$
\M{\H{(a_1,b_1),(a_2,b_2),\ldots}x}{k\;n+p},
$$
where $k\in\N$ is a multiple of the least common multiple of $a_1,a_2,\ldots$ and $p\in\N.$
It will turn out that these Mellin transforms of cyclotomic polylogarithms can be expressed using cyclotomic harmonic sums.
The following lemma will be needed several times in this subsection.

\begin{lemma}Let $a,k,v,n \in \N, p\in \N_0$ and  $0<x<1.$ If $v a=k,$ we have
\begin{eqnarray*}
 \frac{x^{kn+p}-x^p}{\Phi_a(x)}&=&-\sum_{j=0}^{a-1}f_j\sum_{t=0}^{v-1}\sum_{s=0}^{n-1}x^{at+ks+j+p},
\end{eqnarray*}
and if $a=2k$ we have
\begin{eqnarray*}
 \frac{x^{kn+p}-(-1)^nx^p}{\Phi_{a}(x)}&=&-\sum_{j=0}^{k-1}f_j\sum_{s=0}^{n-1}(-1)^{n+s}x^{ks+j+p},
\end{eqnarray*}
where $f_j$ are the coefficients of the power series expansion of $f(x):=\frac{1}{\Phi_a(x)}$ about zero, \ie $f(x)=\sum_{j=0}^{\infty} f_j x^j$ (see Lemma \ref{CScyclotopow}). 
\label{CScycloapart}
\end{lemma}
\begin{proof}
 Due to Lemma \ref{CScyclotopow} we get
 \begin{eqnarray*}
 \frac{x^{kn+p}-x^p}{\Phi_a(x)}&=&(x^{kn+p}-x^p)\sum_{j=0}^{a-1}f_j\sum_{t=0}^{\infty}x^{at+j}\\
	&=&x^p\sum_{j=0}^{a-1}f_j\sum_{t=0}^{\infty}(x^{at+j+kn}-x^{at+j})\\
	&=&x^p\sum_{j=0}^{a-1}f_jx^j\sum_{t=0}^{\infty}(x^{a(t+vn)}-x^{at})\\
	&=&-x^p\sum_{j=0}^{a-1}f_jx^j\sum_{t=0}^{vn-1}x^{at}\\
	&=&-\sum_{j=0}^{a-1}f_j\sum_{t=0}^{v-1}\sum_{s=0}^{n-1}x^{at+ks+j+p}
 \end{eqnarray*}
and if $a=2k$ we get
 \begin{eqnarray*}
 \frac{x^{kn+p}-(-1)^nx^p}{\Phi_{a}(x)}&=&\sum_{j=0}^{\frac{a}{2}-1}f_j\sum_{s=0}^{\infty}x^{\frac{a}{2}i+j}(-1)^{i}(x^{kn+p}-(-1)^nx^{p})\\
	&=&-x^p\sum_{j=0}^{k-1}f_jx^j\sum_{s=0}^{\infty}(-1)^{i}(x^{k(n+s)}-(-1)^nx^{ks})\\
	&=&-\sum_{j=0}^{k-1}f_j\sum_{s=0}^{n-1}(-1)^{n+s}x^{ks+j+p}.
 \end{eqnarray*}.
\end{proof}

Now we can look at the depth one cases.
\begin{lemma} Let $k,a,n \in \N,b,p\in\N_0$ with $a>1, a|k.$ We have 
\begin{eqnarray*}
\M{\H{(0,0)}{x}}{k \;n+p}&=& -\frac{1}{(kn+p+1)^2},\\
\M{\H{(1,0)}{x}}{k \;n+p}&=& -\frac{\S1{kn+p+1}}{kn+p+1}\\&=&-\frac{1}{kn+p+1}\left(\S1{p+1}+\sum_{i=0}^{n-1}\sum_{t=0}^{k-1}\frac{1}{ki+t+p+2}\right),\\
\M{\H{(a,b)}{x}}{k \;n+p}&=&\frac{1}{kn+p+1}\Biggl(\H{(a,b)}{1}-\H{(a,p+b+1)}1\Biggr.\\
			  &&\Biggl.+\sum_{j=0}^{a-1}f_j\sum_{t=0}^{\frac{k}{a}-1}\sum_{s=0}^{n-1}\frac{1}{at+ks+j+p+b+2}\Biggr)
\end{eqnarray*}
where $f_j$ are the coefficients of the power series expansion of $f(x):=\frac{1}{\Phi_a(x)}$ about zero, \ie $f(x)=\sum_{j=0}^{\infty} f_j x^j$ (see Lemma \ref{CScyclotopow}). 
\label{CSweight1mel}
\end{lemma}
\begin{proof}We give a proof of the first and third identity:
\begin{eqnarray*}
 \M{\H{(0,0)}{x}}{k \;n+p}&=&\int_0^1x^{kn+p}\H{(0,0)}{x}dx=-\int_0^1\frac{x^{kn+p+1}}{kn+p+1}\frac{1}{x}dx\\
&=&-\frac{1}{(kn+p+1)^2}.
\end{eqnarray*}
Using integration by parts and Lemma \ref{CScycloapart} we get
\begin{eqnarray*}
\M{\H{(a,b)}{x}}{k \;n+p}&=&\int_0^1x^{kn+p}\H{(a,b)}{x}dx=\\
&=&\frac{1}{kn+p+1}\H{(a,b)}{1}-\int_0^1\frac{x^{kn+p+b+1}}{(kn+p+1)\Phi_a(x)}dx\\
&=&\frac{1}{kn+p+1}\Biggl(\H{(a,b)}{1}-\H{(a,p+b+1)}1\Biggr.\\
  &&\Biggl.-\int_0^1\frac{x^{kn+p+b+1}-x^{p+b+1}}{\Phi_a(x)}dx\Biggr)\\
&=&\frac{1}{kn+p+1}\Biggl(\H{(a,b)}{1}-\H{(a,p+b+1)}1\Biggr.\\
  &&\Biggl.+\sum_{j=0}^{a-1}f_j\sum_{t=0}^{\frac{k}{a}-1}\sum_{s=0}^{n-1}\int_0^1x^{at+ks+j+p+b+1}dx\Biggr)\\
&=&\frac{1}{kn+p+1}\Biggl(\H{(a,b)}{1}-\H{(a,p+b+1)}1\Biggr.\\
  &&\Biggl.+\sum_{j=0}^{a-1}f_j\sum_{t=0}^{\frac{k}{a}-1}\sum_{s=0}^{n-1}\frac{1}{at+ks+j+p+b+2}\Biggr).
\end{eqnarray*}
\end{proof}

For higher weights we need the following lemma.
\begin{lemma}Let $k,a \in \N,b,p\in\N_0$ with $a>1, a|k$ and $b<\varphi(a)$. We have 
\begin{eqnarray*}
\M{\H{(0,0),\ve m}x}{k \;n + p}&=&\frac{1}{k\;n+p+1}\Biggl(\H{(0,0),\ve m}1 -\M{\H{\ve m}x}{kn+p}\Biggr)\\
\M{\H{(1,0),\ve m}x}{k \;n + p}&=&\frac{-1}{k\;n+p+1}\Biggl(\int_0^1\frac{x^{p+1}-1}{x-1}\H{\ve m}xdx\Biggr.\\
				  &&\Biggl.+\sum_{i=0}^{n-1}\sum_{l=0}^{k-1}\M{\H{\ve m}x}{ki+l+p+1} \Biggr)\\
 \M{\H{(a,b),\ve m}x}{k \;n + p}&=&\frac{1}{k\;n+p+1}\Biggl(\H{(a,b),\ve m}1-\H{(a,p+b+1),\ve m}1 \Biggr.\\
				  &&\Biggl.+\sum_{i=0}^{n-1} \sum_{j=0}^{a-1} f_j \sum_{l=0}^{\frac{k}{a}-1}\M{\H{m}x}{k\;i+j+b+p+1+a\;l} \Biggr),
\end{eqnarray*}
where $f_j$ are the coefficients of the power series expansion of $f(x):=\frac{1}{\Phi_a(x)}$ about zero \ie $f(x)=\sum_{j=0}^{\infty} f_j x^j$ (see Lemma \ref{CScyclotopow}). 
\label{CSmelnotweighted}
\end{lemma}

\begin{proof}We just give a proof of the third identity using integration by parts and Lemma \ref{CScycloapart}:
\begin{eqnarray*}
  \M{\H{(a,b),\ve m}x}{k \;n + p}&=&\frac{1}{k\;n+p+1}\Biggl(\H{(a,b),\ve m}1-\int_0^1\frac{x^{kn+p+b+1}}{\Phi_a(x)}\H{\ve m}dx\Biggr)\\
      &=&\frac{1}{k\;n+p+1}\Biggl(\H{(a,b),\ve m}1-\int_0^1\frac{x^{p+b+1}}{\Phi_a(x)}\H{\ve m}xdx\Biggr.\\
      &&-\int_0^1\frac{x^{kn+p+b+1}-x^{p+b+1}}{\Phi_a(x)}\H{\ve m}xdx\Biggr)\\
      &=&\frac{1}{kn+p+1}\Biggl(\H{(a,b),\ve m}{1}-\H{(a,p+b+1),\ve m}1\Biggr.\\
	  &&\Biggl.+\sum_{j=0}^{a-1}f_j\sum_{t=0}^{\frac{k}{a}-1}\sum_{s=0}^{n-1}\frac{1}{at+ks+j+p+b+2}\Biggr).
\end{eqnarray*}
\end{proof}
We can even handle slightly more cases as mentioned in the beginning of this subsection:
\begin{lemma} Let $k,n \in \N,b,p\in \N_0.$ We have 
\begin{eqnarray*}
\M{\H{(2\;k ,b)}{x}}{k \;n+p}&=&\frac{1}{kn+p+1}\Biggl(\H{(2k,b)}{1}-(-1)^n\H{(2k,p+b+1)}1\Biggr.\\
		&&\Biggl.+(-1)^n\sum_{j=0}^{k-1}f_j\sum_{t=0}^{n-1}(-1)^t\frac{1}{k t+j+p+b+2}\Biggr)\\
\M{\H{(2\;k ,b),\ve m}{x}}{k \;n+p}&=&\frac{1}{kn+p+1}\Biggl(\H{(2k,b).\ve m}{1}-(-1)^n\H{(2k,p+b+1),\ve m}1\Biggr.\\
		&&\Biggl.+(-1)^n\sum_{j=0}^{k-1}f_j\sum_{t=0}^{n-1}(-1)^t\M{\H{\ve m}x}{k t+j+p+b+1}\Biggr)
\end{eqnarray*}
where $f_j$ are the coefficients of the power series expansion of $f(x):=\frac{1}{\Phi_{2k}(x)}$ about zero, \ie $f(x)=\sum_{j=0}^{\infty} f_j x^j$ (see Lemma \ref{CScyclotopow}). 
\label{CSmelnotweighted2k}
\end{lemma}
\begin{proof}
The proof follows using integration by parts and Lemma \ref{CScycloapart}.
\end{proof}

Combining the Lemmas \ref{CSweight1mel},\ref{CSmelnotweighted} and \ref{CSmelnotweighted2k} we can compute the desired Mellin transformations of cyclotomic polylogarithms up
to arbitrary weight. Using the following lemma we can as well compute the Mellin transformation of cyclotomic polylogarithms weighted by cyclotomic polynomials.
\begin{lemma}For $a,k,n \in \N,p\in\N_0$ we have
\begin{eqnarray*}
\M{\frac{\H{\ve m}{x}}{\Phi_1(x)}}{k\;n+p} &=&-\left(kn+p\right)\M{\H{(1,0),\ve m}{x}}{k\;n+p} \\
\M{\frac{\H{\ve m}{x}}{\Phi_a(x)}}{k\;n+p} &=& \H{(a,0),\ve m}1-\left(kn+p\right)\M{\H{(a,0),\ve m}x}{kn+p-1}.
\end{eqnarray*}
\label{CSmelweighted}
\end{lemma}
\begin{proof}
 The proof follows immediately by integration by parts.
\end{proof}

\begin{example}
\small
\begin{eqnarray*}
&&\M{\H{(4,1),(2,0)}x}{4n+1}=n\left(-\frac{ \textnormal{S}_{(1,0,2)}(\infty )}{32 n+16}+\frac{ \textnormal{S}_{(2,-1,1)}(\infty )}{16 n+8}-\frac{ \textnormal{S}_{(2,1,1)}(\infty )}{16 n+8}
	-\frac{ \textnormal{S}_{(2,1,2)}(\infty)}{8 n+4}\right.\\
&&\hspace{1cm}-\frac{ \textnormal{S}_{(4,-1,1)}(\infty )}{8 n+4}-\frac{7 n \textnormal{S}_{(4,1,1)}(\infty )}{12 n+6}+\frac{ \textnormal{S}_{(4,1,2)}(\infty )}{2 n+1}+\frac{17 n
   \textnormal{S}_{(4,3,1)}(\infty )}{24 n+12}+\frac{ \textnormal{S}_{(4,3,2)}(\infty )}{2 n+1}\\
&&\hspace{1cm}+\frac{ \textnormal{S}_{(1,0,1),(1,0,1)}(\infty )}{32 n+16}-\frac{
   \textnormal{S}_{(1,0,1),(2,1,1)}(\infty )}{16 n+8}-\frac{ \textnormal{S}_{(2,1,1),(1,0,1)}(\infty )}{16 n+8}+\frac{ \textnormal{S}_{(2,1,1),(2,1,1)}(\infty )}{8 n+4}\\
&&\hspace{1cm}-\frac{
   \textnormal{S}_{(4,1,1),(4,1,1)}(\infty )}{2 n+1}+\frac{ \textnormal{S}_{(4,1,1),(4,3,1)}(\infty )}{2 n+1}+\frac{ \textnormal{S}_{(4,3,1),(4,1,1)}(\infty )}{2 n+1}-\frac{
   \textnormal{S}_{(4,3,1),(4,3,1)}(\infty )}{2 n+1}\\
&&\hspace{1cm}-\frac{ \textnormal{S}_{(2,1,1)}(n)}{24 n+12}+\frac{ \textnormal{S}_{(4,1,1)}(n)}{4 n+2}-\frac{ \textnormal{S}_{(4,3,1)}(n)}{6 n+3}-\frac{
   \textnormal{S}_{(1,0,1),(2,1,1)}(n)}{32 n+16}+\frac{ \textnormal{S}_{(1,0,1),(4,1,1)}(n)}{16 n+8}\\
&&\hspace{1cm}+\frac{ \textnormal{S}_{(1,0,1),(4,3,1)}(n)}{16 n+8}+\frac{
   \textnormal{S}_{(2,1,1),(1,0,1)}(n)}{32 n+16}+\frac{ \textnormal{S}_{(2,1,1),(2,1,1)}(n)}{16 n+8}-\frac{ \textnormal{S}_{(2,1,1),(4,1,1)}(n)}{8 n+4}\\
&&\hspace{1cm}\left.-\frac{
   \textnormal{S}_{(2,1,1),(4,3,1)}(n)}{8 n+4}-\frac{ \textnormal{S}_{1,1}(n)}{64 n+32}+\frac{1}{12 n+6}\right).
\end{eqnarray*}
\normalsize
\end{example}

\subsection{Differentiation of Cyclotomic Harmonic Sums}
\label{CSdifferentiation}
We already considered the differentiation of harmonic sums with respect to the upper summation limit. In a similar way, we can differentiate cyclotomic harmonic sums 
since they can be represented as linear combinations of Mellin transforms of cyclotomic harmonic polylogarithms,
as has been illustrated in Section~\ref{CSCyloMel}.
Based on this representation, the differentiation of these sums can be expressed in the following format:
%------------------------------------------------------------------------------------------------------------
\begin{eqnarray}
\label{CSDIF1}
\frac{\partial^m}{\partial n^m}
\S{(a_1,b_1,c_1),\ldots,(a_k,b_k,c_k)}n = \sum_{i=1}^s e_i \int_0^1~x^{l n}
~l^m~
\H{0}x^m (f_{\alpha_i}^{\beta_i}(x))^{u_i}
\H{\ve g_i}xdx,
\end{eqnarray}
where $m,\alpha_i,\beta_i,u_i \in \N$ and $l,e_i\in \Z.$
%------------------------------------------------------------------------------------------------------------
The product $\H{0}x^m \H{\ve g_i}x$ may be transformed into a linear
combination of cyclotomic harmonic polylogarithms following Section \ref{CShpro}. Finally, using the inverse Mellin transform, the derivative (\ref{CSDIF1}) of a cyclotomic
harmonic sum w.r.t.\ $n$ is given as a polynomial expression in terms of cyclotomic harmonic sums and cyclotomic harmonic polylogarithms at $x=1$. Together with the previous 
sections, the derivative (\ref{CSDIF1}) can be expressed as a polynomial expression with rational coefficients in terms of cyclotomic harmonic sums and their values at
$n \rightarrow \infty$.
\begin{example}We start with the integral representation:
\begin{eqnarray*}
\S{(2,1,2)}n&=&-\frac{1}{2}\int_0^1\frac{(x^{2 n}-1)\H{(0,0)}{x}}{x-1}dx+\frac{1}{2} \int_0^1\frac{x^{2 n} \H{(0,0)}{x}}{x+1}dx\\
	    &&-\int_0^1x^{2 n} \H{(0,0)}{x}dx+\frac{1}{4} \S{2}{\infty}-1.
\end{eqnarray*}
Differentiation with respect to $n$ leads to
\begin{eqnarray*}
\frac{\partial}{\partial n}\S{(2,1,2)}n&=&-\frac{1}{2}\int_0^1\frac{2 x^{2 n}\H{(0,0)}{x}\H{(0,0)}{x}}{x-1}dx+\frac{1}{2} \int_0^1\frac{2 x^{2 n}\H{(0,0)}{x} \H{(0,0)}{x}}{x+1}dx\\
	    &&-\int_0^1 2 x^{2 n}\H{(0,0)}{x} \H{(0,0)}{x}dx+\frac{1}{4} \S{2}{\infty}-1.
\end{eqnarray*}
After expanding the products of the cyclotomic harmonic polylogarithms and performing the Mellin transforms we end up with
\begin{eqnarray*}
\frac{\partial}{\partial n}\S{(2,1,2)}n&=&-4 \S{(2,1,3)}{n}+\frac{7}{4} \S{(1,0,3)}{\infty}+2 \S{(2,1,3)}{\infty}-2.
\end{eqnarray*}
\end{example}

\section{Relations between Cyclotomic Harmonic Sums}
\label{CSRelations}

\subsection{Algebraic Relations}
\label{CSalgrel}

We already mentioned in Section \ref{CSdef} that cyclotomic harmonic sums form a quasi shuffle algebra. We will now take a closer look at the quasi shuffle algebra property. 
Therefore we will define the following set:
\begin{eqnarray}
\mathcal{C}(n)&=&\left\{q(s_1,\ldots,s_r)\left|r\in \N; \right. s_i \textnormal{ a cyclotomic harmonic sum at }n ;\ q\in \R[x_1,\ldots,x_r]\right\}. \label{C}\nonumber
\end{eqnarray}
From Section \ref{CSdef} we know that we can always expand a product of cyclotomic harmonic sums into a linear combination of cyclotomic harmonic sums. For two cyclotomic harmonic sums  $s_1$ and $s_2$ let
 $L(s_1,s_2)$ denote the expansion 
of $s_1s_2$ into a linear combination of cyclotomic harmonic sums.
Now we can define the ideal $\mathcal{I}$ on $\mathcal{C}(n)$:
\begin{eqnarray*}
\mathcal{I}(n)&=&\left\{s_1s_2-L(s_1,s_2)\left| s_i \textnormal{ a cyclotomic harmonic sum at }n\right.\right\}. \label{CSI}\\
\end{eqnarray*}
By construction $\mathcal{C}(n)/\mathcal{I}$  is a quasi shuffle algebra with alphabet
\begin{eqnarray*}
 A_C:=\left\{(a,b,\pm c)|a,c\in\N;b\in\N_0,b<a\right\}.
\end{eqnarray*}
We define the degree of a letter $(a,b,\pm c) \in A_C$ as $\abs{(a,b,\pm c)}:=c$. For \linebreak $a_1,a_2~\in~\N;~c_1,c_2~\in~\N;~b_1,~b_2~\in~\N_0,~b_i~<~a_i$ and $c_3\in\Z^*$ we order the letters in $A_C$ by 
\begin{eqnarray*}
\begin{array}{lll} 
	(a_1,b_1,\pm c_1)	&\prec (a_2,b_2,\pm c_2) 	&\textnormal{if } c_1<c_2\\
	(a_1,b_1,-c_1)		&\prec (a_2,b_2,c_1) 		&\\
	(a_1,b_1,c_3)		&\prec (a_2,b_2,c_3)		&\textnormal{if } a_1<a_2\\
	(a_1,b_1,c_3)		&\prec (a_1,b_2,c_3)		&\textnormal{if } b_1<b_2
\end{array}
\end{eqnarray*}
and extend this order lexicographically to words. Using this order it can be shown analogously to \cite{Ablinger2009,Hoffman} that the quasi shuffle algebra $\mathcal{C}(n)/\mathcal{I}$ is the free polynomial algebras on the 
\textit{Lyndon} words with alphabet $A_C.$
Hence the number of algebraic independent sums in $\mathcal{S}(n)/\mathcal{I}$ which we also call basis sums 
is the number of \textit{Lyndon} words.
If we consider for example an index set with 4 letters, say $$\{\alpha_1,\alpha_2,\alpha_3,\alpha_4\}$$ and we look for the number of basis sums where the index  
$\alpha_i$ appears $n_i$ times, we can use the second Witt formula (\ref{HSWitt2}):
 $$\frac{1}{n}\sum_{d|n}{\mu(d)\frac{(\frac{n}{d})!}{(\frac{n_1}{d})!\cdots(\frac{n_4}{d})!}}, \ \ \ n=\sum_{i=1}^4{n_i}.$$
For a specific alphabet we can count as well the number of basis sums at a certain weight. For the alphabet 
$$\left\{\frac{1}{k},\frac{(-1)^k}{k},\frac{1}{2k+1},\frac{(-1)^k}{2k+1}\right\}$$ 
the number of basis sums at $w\geq 1$ is given by
 $$N_A(w) = \frac{1}{w}\sum_{d|w}{\mu\left(\frac{w}{d}\right)5^d}.$$
This is due to the fact that the cyclotomic harmonic sums of weight $w$ with alphabet $\{(\pm 1)^k/k, (\pm 1)^k/(2k+1)\}$ can be viewed as the words 
of length $w$ out of an alphabet with the five letters $\{0,(1,0,1),(1,0,-1),(2,1,1),(2,1,-1)\}$ and that the number of \textit{Lyndon} words of length $n$ over 
an alphabet of length $q$ is given by the first Witt formula~(\ref{HSWitt1}).

We can use an analogous method of the method presented in \cite{Ablinger2009,Bluemlein2004} for harmonic sums to find the basis sums together with the relations for the 
dependent sums; compare Section \ref{HSalgrel}. Here we want to give an example for cyclotomic harmonic sums at weight $w=2$. 
\begin{example} We consider the alphabet $\{(\pm 1)^k/k, (\pm 1)^k/(2k+1)\}.$ At weight $w=2$ we have for instance the 10 basis sums:
\begin{eqnarray*}
&&\S{(1,0,-2)}{n},\S{(1,0,2)}{n},\S{(2,1,-2)}{n},\S{(2,1,2)}{n},\S{(1,0,-1),(1,0,1)}{n},\S{(1,0,-1),(2,1,1)}{n},\\
&&\S{(1,0,-1),(2,1,-1)}{n},\S{(1,0,1),(2,1,1)}{n},\S{(2,1,-1),(1,0,1)}{n},\S{(2,1,-1),(2,1,1)}{n}
\end{eqnarray*}
together with the relations for the remaining sums:
\begin{eqnarray*}
\S{(2,1,1),(2,1,1)}{n}&=& \frac{1}{2} \S{(2,1,1)}{n}^2+\frac{1}{2}\S{(2,1,2)}{n}\\
\S{(2,1,1),(1,0,1)}{n}&=& \S{(2,1,1)}{n} \S{(1,0,1)}{n}+\S{(1,0,1)}{n}-2\S{(2,1,1)}{n}\\&&-\S{(1,0,1),(2,1,1)}{n}\\
\S{(1,0,1),(1,0,1)}{n}&=& \frac{1}{2} \S{(1,0,1)}{n}^2+\frac{1}{2}\S{(1,0,2)}{n}\\
\S{(2,1,1),(2,1,-1)}{n}&=& \S{(2,1,-2)}{n}+\S{(2,1,-1)}{n}\S{(2,1,1)}{n}-\S{(2,1,-1),(2,1,1)}{n}\\
\S{(2,1,1),(1,0,-1)}{n}&=& \S{(2,1,1)}{n}\S{(1,0,-1)}{n}+\S{(1,0,-1)}{n}-2\S{(2,1,-1)}{n}\\&&-\S{(1,0,-1),(2,1,1)}{n}\\
\S{(1,0,1),(2,1,-1)}{n}&=& \S{(1,0,-1)}{n}+\S{(1,0,1)}{n}\S{(2,1,-1)}{n}-2 \S{(2,1,-1)}{n}\\&&-\S{(2,1,-1),(1,0,1)}{n}\\
\S{(1,0,1),(1,0,-1)}{n}&=&\S{(1,0,-2)}{n}+\S{(1,0,-1)}{n}\S{(1,0,1)}{n}-\S{(1,0,-1),(1,0,1)}{n}\\
\S{(2,1,-1),(2,1,-1)}{n}&=& \frac{1}{2} \S{(2,1,-1)}{n}^2+\frac{1}{2} \S{(2,1,2)}{n}\\
\S{(2,1,-1),(1,0,-1)}{n}&=&\S{(1,0,1)}{n}+\S{(1,0,-1)}{n} \S{(2,1,-1)}{n}-2\S{(2,1,1)}{n}\\&&-\S{(1,0,-1),(2,1,-1)}{n}\\
\S{(1,0,-1),(1,0,-1)}{n}&=& \frac{1}{2}\S{(1,0,-1)}{n}^2+\frac{1}{2} \S{(1,0,2)}{n}.
\end{eqnarray*}
Hence we can use the 10 basis sums together with sums of lower weight to express all sums of weight $w=2$. Note that the sums of lower weight in this example 
are not yet reduced to a basis.
\label{CSRelationEx1}
\end{example}

\subsection{Differential Relations}
\label{CSdiffrel}
In Section \ref{CSdifferentiation} we described the differentiation of cyclotomic harmonic sums with respect to the upper summation limit. 
The differentiation leads to new relation. For instance we find
$$
\frac{\partial}{\partial n}\S{(2,1,2)}n=-4 \S{(2,1,3)}{n}+\frac{7}{4} \S{(1,0,3)}{\infty}+2 \S{(2,1,3)}{\infty}-2.
$$
Continuing the Example \ref{CSRelationEx1} we get
\begin{example}[Example \ref{CSRelationEx1} continued]From differentiation we get the additional relations
\begin{eqnarray*}
\S{(1,0,-2)}{n}&=& -\frac{\partial}{\partial n}\S{(1,0,-1)}{n}+\frac{1}{4} \S{(1,0,2)}{\infty}-\S{(2,1,2)}{\infty}-1\\
\S{(1,0,2)}{n}&=& \S{(1,0,2)}{\infty}-\frac{\partial}{\partial n}\S{(1,0,1)}{n}\\
\S{(2,1,-2)}{n}&=& -\frac{1}{2}\frac{\partial}{\partial n}\S{(2,1,-1)}{n}+\S{(4,1,2)}{\infty}-\S{(4,3,2)}{\infty}-\frac{1}{9}\\
\S{(2,1,2)}{n}&=& -\frac{1}{2} \frac{\partial}{\partial n}\S{(2,1,1)}{n}+\frac{3}{8} \S{(1,0,2)}{\infty}+\frac{1}{2} \S{(2,1,2)}{\infty}-\frac{1}{2}.
\end{eqnarray*}
Hence we could reduce the number of basis sums at weight $w=2$ to $6$ by introducing differentiation. The basis sums are:
\begin{eqnarray*}
 &&\S{(1,0,-1),(1,0,1)}{n},\S{(1,0,-1),(2,1,-1)}{n},\S{(1,0,-1),(2,1,1)}{n},\S{(1,0,1),(2,1,1)}{n},\\
 &&\S{(2,1,-1),(1,0,1)}{n},\S{(2,1,-1),(2,1,1)}{n}.
\end{eqnarray*}
\label{CSRelationEx2}
\end{example}

Note that we collect the derivatives in 
%------------------------------------------------------------------------------------------------------------
\begin{eqnarray*}
\S{(a_1,b_1,c_1),...,(a_k,b_k,c_k)}{n}^{(D)}
= \left\{\frac{\partial^N}{\partial n^N}\S{(a_1,b_1,c_1),...,(a_k,b_k,c_k)}{n};
N \in \N\right\}.
\end{eqnarray*}
and identify an appearance of a derivative of a cyclotomic harmonic sum with the cyclotomic harmonic sum itself.
For a motivation see also Section \ref{HSdiffrel}.

\subsection{Multiple Argument Relations}
\label{CSmultargrel}
For harmonic sums we found duplication relations. For cyclotomic harmonic sums we can even look at upper summation limits of the form $kn$ with $k\in \N$, as the 
synchronization of Section \ref{CSdef} suggests: We can always
synchronize a cyclotomic harmonic sum with upper summation limit $kn$ with $k\in \N$ to sums with upper summation limit $n.$ This leads again to new relations
for example we have
$$
\S{(2,1,2)}{3 n}=\frac{1}{9} \S{(2,1,2)}{n}+\S{(6,1,2)}{n}+\S{(6,5,2)}{n}-\frac{1}{(6 n+3)^2}-\frac{1}{(6 n+5)^2}+\frac{34}{225}.
$$

In the following subsection we consider duplication relations in detail.

\subsubsection*{Duplication Relations}
\label{CSduplrel}
Unlike for harmonic sums, where we had just one type of duplication relations, there are two different types of duplication relations for cyclotomic harmonic sums. 
We summarize them in the following theorems; compare \cite{Ablinger2011}.
\begin{thm}
We have the following relation:
\begin{eqnarray*}
\sum{\S{(a_m,b_m,\pm c_m),(a_{m-1},b_{m-1}, \pm c_{m-1}),\ldots,(a_1,b_1, \pm c_1)}{2n}}=2^m \S{(2 a_m,b_m, c_m),\ldots,(2 a_1,b_1, c_1)}{n}
\label{halfint}
\end{eqnarray*}
where we sum on the left hand side over the $2^m$ possible combinations concerning $\pm$.
\end{thm}
\begin{proof}
We proceed by induction on $m.$ Let $m=1:$
\begin{eqnarray*}
\S{(a_1,b_1,c_1)}{2n}+\S{(a_1,b_1,-c_1)}{2n}&=&\sum_{i=1}^{2 n}\frac{1}{(a_1 i +b_1)^{c_1}}+\sum_{i=1}^{2 n}\frac{(-1)^i}{(a_1 i +b_1)^{c_1}}\\
		&=&\sum_{i=1}^{n}\left( \frac{1}{(a_1 2i +b_1)^{c_1}}+\frac{1}{(a_1 (2i -1) +b_1)^{c_1}}\right)\\
		&&+\sum_{i=1}^{n}\left( \frac{1}{(a_1 2i +b_1)^{c_1}}-\frac{1}{(a_1 (2i -1) +b_1)^{c_1}}\right)\\
		&=&2 \sum_{i=1}^{n}\frac{1}{(2 a_1 i +b_1)^{c_1}}=2 \S{(2 a_1,b_1,c_1)}{n}.\\
\end{eqnarray*}
In the following we use the abbreviation:
$$
A(n):=\sum{\S{(a_m,b_m,\pm c_m),\ldots,(a_1,b_1,\pm c_1)}{n}}.
$$
Now we assume that the theorem holds for $m:$
\begin{eqnarray*}
&&\sum{\S{(a_{m+1},b_{m+1},\pm c_{m+1}),\ldots,(a_1,b_1,\pm c_1)}{2n}}\\
	&&\hspace{1cm}=\sum_{i=1}^{2n}\frac{1}{(a_{m+1}i+b_{m+1})^{c_{m+1}}}\sum{\S{(a_m,b_m,\pm c_m),\ldots,(a_1,b_1,\pm c_1)}{i}}\\
	&&\hspace{1cm} \ +\sum_{i=1}^{2n}\frac{(-1)^i}{(a_{m+1}i+b_{m+1})^{c_{m+1}}}\sum{\S{(a_m,b_m,\pm c_m),\ldots,(a_1,b_1,\pm c_1)}{i}}\\
	&&\hspace{1cm}=\sum_{i=1}^{2n}\frac{A(i)}{(a_{m+1}i+b_{m+1})^{c_{m+1}}}+\sum_{i=1}^{2n}\frac{(-1)^i A(i)}{(a_{m+1}i+b_{m+1})^{c_{m+1}}}\\
	&&\hspace{1cm}=\sum_{i=1}^{n}\left(\frac{A(2i)}{(2 a_{m+1}i+b_{m+1})^{c_{m+1}}}+\frac{A(2 i-1)}{((2 i-1) a_{m+1}+b_{m+1})^{c_{m+1}}}\right)\\
	&&\hspace{1cm} \ +\sum_{i=1}^{n}\left(\frac{(-1)^{2i}A(2i)}{(2 a_{m+1}i+b_{m+1})^{c_{m+1}}}+\frac{(-1)^{2i-1}A(2 i-1)}{( (2 i-1) a_{m+1}+b_{m+1})^{c_{m+1}}}\right)\\
	&&\hspace{1cm}= 2 \sum_{i=1}^{n}\frac{A(2i)}{(2 a_{m+1}i+b_{m+1})^{c_{m+1}}}\\
	&&\hspace{1cm}= 2 \sum_{i=1}^{n}\frac{1}{(2 a_{m+1}i+b_{m+1})^{c_{m+1}}} 2^m \S{(2 a_m,b_m,c_m),\ldots,(2 a_1,b_1,c_1)}{i}\\
	&&\hspace{1cm}= 2^{m+1} \S{(2 a_{m+1},b_{m+1},c_{m+1}),\ldots,(2 a_1,b_1, c_1)}{n}.
\end{eqnarray*}
\end{proof}

Similar to the previous relation we have the following theorem.
\begin{thm}
Let $d_i \in \{-1,1\};$ we have the following relation:
\begin{eqnarray*}
&&\sum{d_m d_{m-1}\cdots d_1 \S{(a_m,b_m,d_m c_m),(a_{m-1},b_{m-1}, d_{m-1} c_{m-1}),\ldots,(a_1,b_1,d_1 c_1)}{2n}}\\
	&&\hspace{1cm}=2^m \S{(2 a_m,b_m-a_m, c_m),\ldots,(2 a_1,b_1-a_1, c_1)}{n},
\label{halfint2}
\end{eqnarray*}
where we sum on the left hand side over the $2^m$ possible combinations concerning $d_i$.
\end{thm}

\begin{proof}
We proceed by induction on $m.$ Let $m=1:$
\begin{eqnarray*}
\S{(a_1,b_1,c_1)}{2n}-\S{(a_1,b_1,-c_1)}{2n}&=&\sum_{i=1}^{2 n}\frac{1}{(a_1 i +b_1)^{c_1}}+\sum_{i=1}^{2 n}\frac{(-1)^i}{(a_1 i +b_1)^{c_1}}\\
		&=&\sum_{i=1}^{n}\left( \frac{1}{(a_1 2i +b_1)^{c_1}}+\frac{1}{(a_1 (2i -1) +b_1)^{c_1}}\right)\\
		&&-\sum_{i=1}^{n}\left( \frac{1}{(a_1 2i +b_1)^{c_1}}-\frac{1}{(a_1 (2i -1) +b_1)^{c_1}}\right)\\
		&=&2 \sum_{i=1}^{n}\frac{1}{((2i -1) a_1+b_1)^{c_1}}=2 \S{(2 a_1,b_1-a_1,c_1)}{n}.\\
\end{eqnarray*}
In the following we use the abbreviation:
$$
A(n):=\sum{d_m \cdots d_1 \S{(a_m,b_m,d_m c_m),\ldots,(a_1,b_1, d_1 c_1)}{n}}.
$$
Note that
\begin{eqnarray*}
A(2 n-1)&=&A(2n)-\sum{d_m \cdots d_1 \frac{d_m^{2n} \S{(a_{m-1},b_{m-1},d_{m-1} c_{m-1}),\ldots,(a_1,b_1, d_1 c_1)}{n} } {(a_m i+b_m)^{c_m}}}\\
	&=&A(2n)-\sum{d_{m-1} \cdots d_1 \frac{\S{(a_{m-1},b_{m-1},d_{m-1} c_{m-1}),\ldots,(a_1,b_1, d_1 c_1)}{n} } {(a_m i+b_m)^{c_m}}}\\
	&&+\sum{d_{m-1} \cdots d_1 \frac{\S{(a_{m-1},b_{m-1},d_{m-1} c_{m-1}),\ldots,(a_1,b_1, d_1 c_1)}{n} } {(a_m i+b_m)^{c_m}}}=A(2n).
\end{eqnarray*}
Now we assume that the theorem holds for $m.$ Then we conclude that
\begin{eqnarray*}
&&\sum{d_{m+1} \cdots d_1 \S{(a_{m+1},b_{m+1},d_{m+1} c_{m+1}),\ldots,(a_1,b_1,d_1 c_1)}{2n}}\\
	&&\hspace{1cm}=\sum_{i=1}^{2n}\frac{1}{(a_{m+1}i+b_{m+1})^{c_{m+1}}}\sum{d_m \cdots d_1\S{(a_m,b_m,d_m c_m),\ldots,(a_1,b_1,d_1 c_1)}{i}}\\
	&&\hspace{1cm} \ -\sum_{i=1}^{2n}\frac{(-1)^i}{(a_{m+1}i+b_{m+1})^{c_{m+1}}}\sum{d_m \cdots d_1\S{(a_m,b_m,d_m c_m),\ldots,(a_1,b_1,d_1 c_1)}{i}}\\
	&&\hspace{1cm}=\sum_{i=1}^{2n}\frac{A(i)}{(a_{m+1}i+b_{m+1})^{c_{m+1}}}-\sum_{i=1}^{2n}\frac{(-1)^i A(i)}{(a_{m+1}i+b_{m+1})^{c_{m+1}}}\\
	&&\hspace{1cm}=\sum_{i=1}^{n}\left(\frac{A(2i)}{(2 a_{m+1}i+b_{m+1})^{c_{m+1}}}+\frac{A(2 i-1)}{((2 i-1) a_{m+1}+b_{m+1})^{c_{m+1}}}\right)\\
	&&\hspace{1cm} \ -\sum_{i=1}^{n}\left(\frac{(-1)^{2i}A(2i)}{(2 a_{m+1}i+b_{m+1})^{c_{m+1}}}+\frac{(-1)^{2i-1}A(2 i-1)}{( (2 i-1) a_{m+1}+b_{m+1})^{c_{m+1}}}\right)\\
	&&\hspace{1cm}= 2 \sum_{i=1}^{n}\frac{A(2i-1)}{((2 i -1) a_{m+1}+b_{m+1})^{c_{m+1}}}= 2 \sum_{i=1}^{n}\frac{A(2i)}{((2 i -1) a_{m+1}+b_{m+1})^{c_{m+1}}}\\
	&&\hspace{1cm}= 2 \sum_{i=1}^{n}\frac{1}{(2 a_{m+1}i+b_{m+1}- a_{m+1})^{c_{m+1}}} 2^m \S{(2 a_m,b_m-a_m, c_m),\ldots,(2 a_1,b_1-a_1, c_1)}{n}\\
	&&\hspace{1cm}= 2^{m+1} \S{(2 a_{m+1},b_{m+1}-a_{m+1},c_{m+1}),\ldots,(2 a_1,b_1-a_1,c_1)}{n}.
\end{eqnarray*}
\end{proof}

Examples for these relations are
\begin{eqnarray*}
\S{(2,1,2)}{2 n}&=&2 \S{(4,1,2)}{n}-\S{(2,1,-2)}{2 n} \\
\S{(2,1,2)}{2 n}&=&\S{(2,1,-2)}{2 n}+2 \S{(4,3,2)}{n}-\frac{2}{(4 n+3)^2}+\frac{2}{9}.
\end{eqnarray*}

\begin{example}[Example \ref{HSRelationEx2} continued] At weight $w=2$ with letters $\{(\pm 1)^k/k, (\pm 1)^k/(2k+1)\}$ the duplication relations and the multiple 
integer relations do not lead to a further reduction of the basis.
\end{example}
\begin{remark}
 As for differentiation, when we are counting basis elements, we identify an appearance of a cyclotomic harmonic sum $\S{\ve a}{kn}$ with the cyclotomic harmonic sum $\S{\ve a}{n}$.
\end{remark}

\subsection{Number of Basis Elements for Specific Alphabets}

In this section we look at the different number of basis sums that we get using combinations of several of the relations discussed in the previous sections. First we restrict to the alphabet
$$\left\{\frac{(\pm 1)^k}{k}, \frac{(\pm 1)^k}{(2k+1)}\right\}$$
and look for the number of basis sums. We conjecture the following formulas which we checked up to weight 5 and we summarize the concrete numbers in Table \ref{cyclo2uptow5} up to weight 5:
\begin{eqnarray*}
{\sf N_S}(w)&=&4\cdot5^{w-1}\\
{\sf N_A}(w)&\overset{w>1}{=}& \frac{1}{w}\sum_{d|w}{\mu\left(\frac{w}{d}\right)5^d}\\
{\sf N_{D}}(w)&=&{\sf N_S}(w)-{\sf N_S}(w-1)\overset{w>1}{=}16\cdot5^{w-2}\\
{\sf N_{H_1}}(w)={\sf N_{H_2}}(w)={\sf N_{M}}(w)&=&{\sf N_S}(w)-2^{w-1}=4\cdot5^{w-1}-2^{w-1}\\
{\sf N_{H_1H_2}}(w)&=&{\sf N_S}(w)-(2\cdot2^{w-1}-1)=4\cdot5^{w-1}-(2\cdot2^{w-1}-1)\\
{\sf N_{H_1M}}(w)={\sf N_{H_2M}}(w)&=&{\sf N_S}(w)-2\cdot2^{w-1}=4\cdot5^{w-1}-2\cdot2^{w-1}\\
{\sf N_{H_1H_2M}}(w)&=&{\sf N_S}(w)-(3\cdot2^{w-1}-1)=4\cdot5^{w-1}-(3\cdot2^{w-1}-1)\\
{\sf N_{AD}}(w)&=&{\sf N_A}(w)-{\sf N_A}(w-1)\\
	&\overset{w>2}{=}&\frac{1}{w}\sum_{d|w}{\mu\left(\frac{w}{d}\right)5^d}-\frac{1}{w-1}\sum_{d|(w-1)}{\mu\left(\frac{w-1}{d}\right)5^d} \\
{\sf N_{AH_1H_2M}}(w)&\overset{w>1}{=}& \frac{1}{w}\sum_{d|w}{\mu\left(\frac{w}{d}\right)5^d}-\left(3\cdot\frac{1}{w}\sum_{d|w}{\mu\left(\frac{w}{d}\right)2^d}-1\right) \\
{\sf N_{DH_1H_2M}}(w)&=&{\sf N_{H_1H_2M}}(w)-{\sf N_{H_1H_2M}}(w-1)\\&\overset{w>1}{=}&16\cdot5^{w-2}-3\cdot2^{w-2} \\
{\sf N_{ADH_1H_2M}}(w)&=&{\sf N_{AH_1H_2M}}(w)-{\sf N_{AH_1H_2M}}(w-1) \\
	&\overset{w>2}{=}&\frac{1}{w}\sum_{d|w}{\mu\left(\frac{w}{d}\right)(5^d-3\cdot2^d)}\\&&-\frac{1}{w-1}\sum_{d|(w-1)}{\mu\left(\frac{w-1}{d}\right)(5^d-3\cdot2^d)}
\end{eqnarray*}
Here  {$\sf N_S$} counts the number of cyclotomic harmonic sums concerning the alphabet $\{(\pm 1)^k/k, (\pm 1)^k/(2k+1)\}$, while {$\sf N_A, N_D, N_{H_1},N_{H_2}$} and {$\sf N_M$} count the number 
of basis sums using algebraic, differential, the first dupplication, the second dupplication and the multiple integer relations respectively. All the other quantities count the number 
of basis sums with respect to combinations of different types of relations. For instance
{$\sf N_{ADH_1H_2M}$} gives the number of basis cyclotomic harmonic sum using all previously mentioned relations.

\begin{table}\centering
\scalebox{0.85}{%
\begin{tabular}{| r | r | r | r | r | r | r | r | r | r | r | r | }
\hline	
$w$& ${\sf N_S}$ & ${\sf N_A}$ & ${\sf N_D}$& ${\sf N_{H_1}}$& ${\sf N_{H_1H_2}}$& ${\sf N_{H_1M}}$& ${\sf N_{H_1H_2M}}$ &${\sf N_{AD}}$& ${\sf N_{AH_1H_2M}}$& ${\sf N_{DH_1H_2M}}$& ${\sf N_{ADH_1H_2M}}$  \\
\hline	
  1 &     4 &   4 &    4 &    3 &    3 &    2 &    2 &   4 &    2 &    2 &   2\\
  2 &    20 &  10 &   16 &   18 &   17 &   16 &   15 &   6 &    8 &   13 &   6\\
  3 &   100 &  40 &   80 &   96 &   93 &   92 &   89 &  30 &   35 &   74 &  27\\
  4 &   500 & 150 &  400 &  492 &  485 &  484 &  477 & 110 &  142 &  388 & 107\\
  5 &  2500 & 624 & 2000 & 2484 & 2469 & 2468 & 2453 & 474 &  607 & 1976 & 465\\ 
\hline
\end{tabular}
}
\caption{\label{cyclo2uptow5}Number of basis cyclotomic harmonic sums concerning the alphabet $\{(\pm 1)^k/k, (\pm 1)^k/(2k+1)\}$ and different relations up to weight 5.}
\end{table}

Now we consider the number of basic cyclotomic harmonic sums concerning the alphabet
$$\left\{\frac{(\pm 1)^k}{k}, \frac{(\pm 1)^k}{(3k+1)},\frac{(\pm 1)^k}{(3k+2)}\right\}$$
up to weight 4. We conjecture the following formulas which we checked up to weight 4 and we summarize the concrete numbers in Table \ref{cyclo3uptow4} up to weight 4: 
\begin{eqnarray*}
{\sf N_S}(w)&=&6\cdot7^{w-1}\\
{\sf N_A}(w)&\overset{w>1}{=}& \frac{1}{w}\sum_{d|w}{\mu\left(\frac{w}{d}\right)7^d}\\
{\sf N_{D}}(w)&=&{\sf N_S}(w)-{\sf N_S}(w-1)\overset{w>1}{=}36\cdot7^{w-2}\\
{\sf N_{H_1}}(w)={\sf N_{M}}(w)&=&{\sf N_S}(w)-2\cdot 3^{w-1}=6\cdot7^{w-1}-2\cdot 3^{w-1}\\
{\sf N_{H_2}}(w)&=&{\sf N_S}(w)-2^{w-1}=6\cdot7^{w-1}-2^{w-1}\\
{\sf N_{H_1H_2}}(w)&=&{\sf N_S}(w)-(2^{w-1}+2\cdot 3^{w-1})=6\cdot7^{w-1}-2\cdot 3^{w-1}-2^{w-1}\\
{\sf N_{H_1H_2M}}(w)&=&{\sf N_S}(w)-(2^{w-1}+4\cdot 3^{w-1}-1)=6\cdot7^{w-1}-4\cdot3^{w-1}-2^{w-1}+1\\
{\sf N_{AD}}(w)&=&{\sf N_A}(w)-{\sf N_A}(w-1)\\
	&\overset{w>2}{=}&\frac{1}{w}\sum_{d|w}{\mu\left(\frac{w}{d}\right)7^d}-\frac{1}{w-1}\sum_{d|(w-1)}{\mu\left(\frac{w-1}{d}\right)7^d} \\
{\sf N_{AH_1H_2M}}(w)&\overset{w>1}{=}& \frac{1}{w}\sum_{d|w}{\mu\left(\frac{w}{d}\right)(7^d-2\cdot3^d-2^d)}+1\\
{\sf N_{DH_1H_2M}}(w)&=&{\sf N_{H_1H_2M}}(w)-{\sf N_{H_1H_2M}}(w-1)\\&\overset{w>1}{=}&36\cdot7^{w-2}-8\cdot3^{w-2}-2^{w-2} \\
{\sf N_{ADH_1H_2M}}(w)&=&{\sf N_{AH_1H_2M}}(w)-{\sf N_{AH_1H_2M}}(w-1) \\
	&\overset{w>2}{=}&\frac{1}{w}\sum_{d|w}{\mu\left(\frac{w}{d}\right)(7^d-2\cdot 3^d- 2^d)}\\&&-\frac{1}{w-1}\sum_{d|(w-1)}{\mu\left(\frac{w-1}{d}\right)(7^d-2\cdot3^d-2^d)}.
\end{eqnarray*}

\begin{table}\centering
\scalebox{0.85}{%
\begin{tabular}{| r | r | r | r | r | r | r  | r | r | r | r | r| }
\hline	
$w$& ${\sf N_S}$ & ${\sf N_A}$ & ${\sf N_D}$& ${\sf N_{H_1}}$ & ${\sf N_{H_2}}$ & ${\sf N_{H_1H_2}}$ &${\sf N_{AD}}$&${\sf N_{H_1H_2M}}$& ${\sf N_{AH_1H_2M}}$& ${\sf N_{DH_1H_2M}}$& ${\sf N_{ADH_1H_2M}}$  \\
\hline	
  1 &     6 &   6 &    6 &    4 &    5 &    3 &    6 &    2 &    2 &    2 &    2 \\
  2 &    42 &  21 &   36 &   36 &   40 &   34 &   15 &   29 &   15 &   27 &   13 \\
  3 &   294 & 112 &  252 &  276 &  290 &  272 &   91 &  255 &   95 &  226 &   80 \\
  4 &  2058 & 588 & 1764 & 2004 & 2050 & 1996 &  476 & 1943 &  550 & 1688 &  455 \\
\hline
\end{tabular}
}
\caption{\label{cyclo3uptow4}Number of basis cyclotomic harmonic sums concerning the alphabet $\{(\pm 1)^k/k, (\pm 1)^k/(3k+1), (\pm 1)^k/(3k+2)\}$ and different relations up to weight~4.}
\end{table}

Finally in Table \ref{cyclouptol20w1} and Table \ref{cyclouptol20w2} we summarize the number of basis elements for different cyclotomies at weight $w=1$ and $w=2$ respectively. At 
weight $w=1$ we conjecture the following counting relation for the basis elements which was tested up to $l = 300.$ Let $p_i$ be pairwise distinct primes $>2,$ and let $k_i$ be positive integers.
The number ${\sf N}(l)$ of basis elements at $w = 1$ and cyclotomy $l$ are given by
%------------------------------------------------------------------------------------------------------------
\begin{eqnarray*}
{\sf N}(l)&=&\left\{
		  	\begin{array}{lll}
						l,\  & l=1 \text{ or } l=2^k\\
						\prod_{i=1}^n(p_i-1)p_i^{k_i-1},\  & l=\prod_{i=1}^np_i^{k_i}\\
						2^k\cdot \prod_{i=1}^n(p_i-1)p_i^{k_i-1},\  & l=2^k\prod_{i=1}^np_i^{k_i}.
			\end{array}
		\right.
\end{eqnarray*} 
For weight $w=2$ we guessed explicit formulas for some of the columns:
\begin{eqnarray*}
{\sf N_S}(l)&=&2\cdot l (2\cdot l + 1)\\
{\sf N_A}(l)&=&\frac{1}{2}\sum_{d|2}{\mu\left(\frac{2}{d}\right)(2l+1)^d}= l (2\cdot l + 1)\\
{\sf N_{D}}(l)&=&2\cdot l (2\cdot l + 1)-2\cdot l=4\cdot l^2\\
{\sf N_{H_1H_2}}(l)&=&\frac{1}{4} (-1)^l (l+1)+\frac{1}{4} (2\cdot l+1) (7\cdot l-1)\\
% {\sf N_{M}}(l)&=&\\
% {\sf N_{H_1H_2M}}(l)&=&\\
{\sf N_{AD}}(l)&=& l (2\cdot l + 1)-2\cdot l = 2\cdot l^2-l.
% {\sf N_{AH_1H_2M}}(l)&=&  \\
% {\sf N_{DH_1H_2M}}(l)&=& \\
% {\sf N_{ADH_1H_2M}}(l)&=& \\
\end{eqnarray*}

\begin{table}
\centering
%\begin{tabular}{l*{6}{c}r}
\scalebox{0.80}{
\begin{tabular}{|r|r|r|r|r|r|r|r|r|r|r|r|r|r|r|r|r|r|r|r|r|}
\hline
\multicolumn{1}{|c}{$l$}             &
\multicolumn{1}{|c}{1}                 &
\multicolumn{1}{|c}{2}                 &
\multicolumn{1}{|c}{3}                 &
\multicolumn{1}{|c}{4}                 &
\multicolumn{1}{|c}{5}                 &
\multicolumn{1}{|c}{6}                 &
\multicolumn{1}{|c}{7}                 &
\multicolumn{1}{|c}{8}                 &
\multicolumn{1}{|c}{9}                 &
\multicolumn{1}{|c}{10}                &
\multicolumn{1}{|c}{11}                &
\multicolumn{1}{|c}{12}                &
\multicolumn{1}{|c}{13}                &
\multicolumn{1}{|c}{14}                &
\multicolumn{1}{|c}{15}                &
\multicolumn{1}{|c}{16}                &
\multicolumn{1}{|c}{17}                &
\multicolumn{1}{|c}{18}                &
\multicolumn{1}{|c}{19}                &
\multicolumn{1}{|c|}{20}               \\
\hline
{\sf sums} &  2 & 4 & 6 & 8 & 10 & 12 & 14 & 16 & 18 & 20 & 22 & 24 & 26 & 28 & 30 & 32 & 34 & 36 & 38 & 40 \\
\hline
{\sf basis}
& 1 & 2 & 2 & 4 &  4 &  4 &  6 &  8 & 6 & 8 & 10 & 8 &12 & 12 & 8 & 16 & 16 & 12 & 18 & 16\\
\hline
{\sf new basis}
& 1 & 1 & 1 & 3 &  3 &  1 &  5 &  6 & 4 & 3 & 9 & 3 &11 & 5 & 3 & 12 & 15 & 4 & 17 & 9\\
{\sf sums}
& & & & & & & & & &
& & & & & & & & & & \\
\hline
\end{tabular}
}
\caption{\label{cyclouptol20w1}The number of the $w = 1$ cyclotomic harmonic sums up to $l =20$ and the corresponding number of basis elements at fixed value of $l$, and the new basis elements in ascending sequence.}
\end{table}

\begin{table}
\centering
\scalebox{0.80}{
\begin{tabular}{|r|r|r|r|r|r|r|r|r|r|}
\hline
 $l$& ${\sf N_S}$ & ${\sf N_A}$ & ${\sf N_D}$&  ${\sf N_{H_1H_2}}$& ${\sf N_{M}}$ &${\sf N_{AD}}$& ${\sf N_{AH_1H_2M}}$& ${\sf N_{DH_1H_2M}}$& ${\sf N_{ADH_1H_2M}}$  \\
\hline
  1 &   6 &   3 &   4 &   4 &   6 &   1 &   2  &   3 &   1\\
  2 &  20 &  10 &  16 &  17 &  18 &   6 &   8  &  13 &   6\\
  3 &  42 &  21 &  36 &  34 &  36 &  15 &  15  &  27 &  13\\
  4 &  72 &  36 &  64 &  62 &  66 &  28 &  29  &  52 &  25\\
  5 & 110 &  55 & 100 &  92 & 104 &  45 &  44  &  83 &  40\\
  6 & 156 &  78 & 144 & 135 & 126 &  66 &  55  & 103 &  51\\
  7 & 210 & 105 & 196 & 178 & 204 &  91 &  87  & 167 &  81\\
  8 & 272 & 136 & 256 & 236 & 252 & 120 & 110  & 208 & 102\\
  9 & 342 & 171 & 324 & 292 & 300 & 153 & 128  & 247 & 122\\
 10 & 420 & 210 & 400 & 365 & 372 & 190 & 162  & 311 & 154\\
 11 & 506 & 253 & 484 & 434 & 500 & 231 & 215  & 419 & 205\\
 12 & 600 & 300 & 576 & 522 & 492 & 276 & 213  & 412 & 205\\
 13 & 702 & 351 & 676 & 604 & 696 & 325 & 300  & 587 & 288\\
 14 & 812 & 406 & 784 & 707 & 738 & 378 & 321  & 623 & 309\\
 15 & 930 & 465 & 900 & 802 & 784 & 435 & 335  & 655 & 327\\
 16 &1056 & 528 &1024 & 920 & 984 & 496 & 428  & 832 & 412\\
 17 &1190 & 595 &1156 &1028 &1184 & 561 & 512  &1007 & 496\\
 18 &1332 & 666 &1296 &1161 &1098 & 630 & 474  & 927 & 462\\
 19 &1482 & 741 &1444 &1282 &1476 & 703 & 639  &1259 & 621\\
 20 &1640 & 820 &1600 &1430 &1464 & 780 & 635  &1244 & 619\\
\hline
\end{tabular}
}
\caption{\label{cyclouptol20w2}Number of basis elements of the $ w = 2$ cyclotomic harmonic sums up to cyclotomy $l = 20$.}
\end{table}

\section{Cyclotomic Harmonic Sums at Infinity}

Not all cyclotomic harmonic sums are finite at infinity, since for example $\lim_{n\rightarrow \infty} \S{(2,1,1)}n$ does not exist.
In fact, we have the following lemma, compare Lemma \ref{CSconsumlem}:
\begin{lemma}
Let $a_i, p \in \N,b_i\in\N_0$ and $c_i,\in \Z^*.$
The cyclotomic harmonic sum $\S{(a_1,b_1,c_1),\ldots,(a_p,b_p,c_p)}n$ is convergent, when $n\rightarrow \infty$, if and only if $c_1 \neq 1.$
\label{CSconsumlem}
\end{lemma}

\subsection{Relations between Cyclotomic Harmonic Sums at Infinity}
\label{CSInfRelations}

The values of the cyclotomic harmonic polylogarithms at argument
$x=1$ and, related to it, the associated cyclotomic harmonic sums at $n \rightarrow
\infty$ occur in various relations of the finite cyclotomic harmonic sums
and the Mellin transforms of cyclotomic harmonic polylogarithms; compare \cite{Ablinger2011}. In this
Section we investigate their relations and basis representations. The
infinite cyclotomic harmonic sums extend the Euler-Zagier and multiple
zeta values \cite{Broadhurst2010,MZV1,MZV2,MZV3} (see also Section \ref{HSInfRelations}) and are related at lower weight and depth to other
known special numbers. We define
$$
\sigma_{(a_1,b_1,c_1),\ldots,(a_l,b_l,c_l)}:=\lim_{n \rightarrow \infty} \S{(a_1,b_1,c_1),\ldots,(a_l,b_l,c_l)}n.
$$
We first consider the sums of weight $w = 1$ and $w = 2$ up to cyclotomy $l~=~20$.
Afterwards the relations of the infinite cyclotomic harmonic sums associated to
the summands 
\begin{eqnarray}
\label{eqSel}
\frac{1}{k^{l_1}},~~~~~~
\frac{(-1)^k}{k^{l_2}},~~~~~~
\frac{1}{(2k+1)^{l_3}},~~~~~~
\frac{(-1)^k}{(2k+1)^{l_4}},
\end{eqnarray}
with $l_i\in\N$ up to weight $ w = 6$ are worked out.

We consider the following relations:
\begin{itemize}
\item The algebraic relations of cyclotomic harmonic sums (see Section \ref{CSalgrel}) remain valid when we consider
 them at infinity. We will refer to these relations as the stuffle relations.
\item The duplication relations from Section \ref{CSduplrel} remain valid if we consider sums which are finite at infinity, since it makes no difference
 whether the argument is $\infty$ or $2\cdot\infty.$
\item We can generalize the relation form \cite{Vermaseren1998} (see Section \ref{HSInfRelations}) for harmonic sums to cyclotomic harmonic sums. For not both $c_1=1$ and $f_1=1$ we have (compare \cite{Ablinger2011}):
\begin{eqnarray*}
&&\S{(a_1,b_1,c_1),\ldots,(a_k,b_k,c_k)}{\infty}\S{(d_1,e_1,f_1),\ldots,(d_l,e_l,f_l)}{\infty}=\\
	&&\hspace{1cm}\lim_{n \rightarrow \infty}\sum_{i=1}^n\frac{\sign{f_1}^i \S{(a_1,b_1,c_1),\ldots,(a_k,b_k,c_k)}{n-i}\S{(d_2,e_2,f_2),\ldots,(d_l,e_l,f_l)}i} {(d_1 i+e_1)^{\abs{f_1}}}
\end{eqnarray*}
Using
\begin{eqnarray*}
&&\sum_{i=1}^n\frac{\S{(a_1,b_1,c_1),\ldots,(a_k,b_k,c_k)}{n-i}\S{(d_2,e_2,f_2),\ldots,(d_l,e_l,f_l)}{i} \sign{f_1}^i} {(d_1i+e_1)^{\abs{f_1}}}=\\
&&\hspace{0.5cm} \sum_{k=1}^{\abs{f_1}} \binom{\abs{c_1}+\abs{f_1}-k-1}{\abs{c_1}-1} a_1^{\abs{f_1}-k} d_1^{\abs{c_1}}\sum_{i=1}^n\frac{\sign{c_1}^i}{(d_1 b_1 + a_1 e_1 + i a_1 d_1)^{\abs{c_1} \abs{f_1} -k}} \\
&&\hspace{0.5cm} \hspace{0.5cm} \sum_{j=1}^{i-1}\frac{\S{(a_2,b_2,c_2),\ldots,(a_k,b_k,c_k)}{i-j}\S{(d_2,e_2,f_2),\ldots,(d_l,e_l,f_l)}{j} \left(\frac{\sign{f_1}}{\sign{c_1}}\right)^j} {(d_1 j +e_1)^{k}}\\
&&\hspace{0.5cm} +\sum_{k=1}^{\abs{c_1}} \binom{\abs{c_1}+\abs{f_1}-k-1}{c_2-1} a_1^{\abs{f_1}} d_1^{\abs{c_1}-k}\sum_{i=1}^n\frac{\sign{f_1}^i}{(d_1 b_1 + a_1 e_1 + i a_1 d_1)^{\abs{c_1} \abs{f_1} -k}} \\
&&\hspace{0.5cm} \hspace{0.5cm} \sum_{j=1}^{i-1}\frac{\S{(d_2,e_2,f_2),\ldots,(d_l,e_l,f_l)}{i-j}\S{(a_2,b_2,c_2),\ldots,(a_k,b_k,c_k)}{j} \left(\frac{\sign{c_1}}{\sign{f_1}}\right)^j} {(a_1 j +b_1)^{k}},
\end{eqnarray*}
we can rewrite the right hand side in terms of cyclotomic harmonic sums. We will refer to these relations as the shuffle relations since one could also obtain them from the shuffle algebra 
of cyclotomic harmonic polylogarithms.
\item The multiple argument relations from Section \ref{CSmultargrel} remain valid if we consider sums which are finite at infinity, since it makes no difference
 whether the argument is $\infty$ or $k\cdot\infty.$
\end{itemize}

\subsubsection*{Cyclotomic Harmonic Sums up to Weight \texorpdfstring{$w = 2$}{w=2} at Infinity}
%%%%%%%%%%%%%%%%%%%%%%%%%%%%%%%%%%%%%%%%%%%%%%%%%%%%%%%%%%%%%%%%%%%%%%%%%%

We study the sums up to weight $w = 2$ and cyclotomy $l = 20$~\footnote{Relations between colored
nested infinite harmonic sums have been investigated also in Refs.~\cite{DELIGNE,DG}
recently.}, based on the sets of sums using the letters
%----------------------------------------------------------------------------------------------
\begin{eqnarray} \label{eq:def1} \frac{(\pm 1)^k}{(lk+m)^n},~~~~1
\leq n \leq 2, 1 \leq l \leq 20, m
< l~. \end{eqnarray}
%----------------------------------------------------------------------------------------------
We use the stuffle (quasi-shuffle) relations for the sums, the shuffle relations on the side of the
associated cyclotomic harmonic polylogarithms, and the multiple argument relations for these
sums.
In some cases the latter request to include sums which are outside the above pattern. In these cases the
corresponding relations are not accounted for.

At $w = 1$ the respective numbers of basis elements is summarized in Table~\ref{tablew1l20}.
The independent sums at $w = 1$ up to $ l = 6$ are~:
%----------------------------------------------------------------------------------------------
\begin{eqnarray*}
&&\sigma_{(1, 0, 1)},
\sigma_{(1, 0, -1)},
\sigma_{(2, 1, -1)},
\sigma_{(3, 1,  1)},
\sigma_{(3, 1, -1)},
\sigma_{(4, 1, -1)},
\sigma_{(4, 3, -1)},
\sigma_{(5, 1,  1)},
\sigma_{(5, 1, -1)},
\nonumber\\ &&
\sigma_{(5, 2, -1)},
\sigma_{(5, 3, -1)},
\sigma_{(6, 1, -1)}~,
\label{CSbasw1}
\end{eqnarray*}
%----------------------------------------------------------------------------------------------
see \cite{Ablinger2011} for equivalent representations. The dependent sums up to $l = 6$ are
%----------------------------------------------------------------------------------------------
\begin{eqnarray*}
\label{Equ:Sigma1}
\sigma_{(2, 1, 1)} &=& -1 - \sigma_{(1, 0, -1)} + \frac{1}{2} \sigma_{(1, 0, 1)}
\\
%----
\sigma_{(3, 2, 1)} &=& - \frac{1}{2} - \frac{1}{3}\sigma_{(1, 0, -1)}
                         - \sigma_{(3, 1, -1)} + \sigma_{(3, 1, 1)}
\\
\sigma_{(3, 2, -1)} &=&   \frac{1}{2} + \frac{2}{3} \sigma_{(1, 0, -1)}
                         + \sigma_{(3, 1, -1)}
\\
%----
\sigma_{(4, 1, 1)} &=&
   - \frac{1}{2} - \frac{3}{4} \sigma_{(1, 0, -1)} + \frac{1}{4} \sigma_{(1, 0, 1)}
   + \frac{1}{2} \sigma_{(2, 1, -1)}
\\
\sigma_{(4, 3, 1)} &=&
-\frac{5}{6} - \frac{3}{4} \sigma_{(1, 0, -1)} + \frac{1}{4} \sigma_{(1, 0, 1)}
- \frac{1}{2} \sigma_{(2, 1, -1)}
\\
%----
\sigma_{(5, 2, 1)} &=&
\frac{1}{5} \sigma_{(1, 0, -1)}
+ \sigma_{(5, 1, 1)}
- \sigma_{(5, 2, -1)}
\\
\sigma_{(5, 3, 1)} &=&
- \frac{1}{3}
- \frac{1}{5} \sigma_{(1, 0, -1)}
- \sigma_{(5, 1, -1})
+ \sigma_{(5, 1, 1)}
\\
\sigma_{(5, 4, 1)} &=& -\frac{7}{12}
- \frac{2}{5} \sigma_{(1, 0, -1)}
- \sigma_{(5, 1, -1)}
+ \sigma_{(5, 1, 1)}
- \sigma_{(5, 3, -1)}
\\
\sigma_{(5, 4, -1}) &=& \frac{7}{12}
+ \frac{4}{5} \sigma_{(1, 0, -1)}
+ \sigma_{(5, 1, -1)}
- \sigma_{(5, 2, -1)}
+ \sigma_{(5, 3, -1)}
\\
%----
\sigma_{(6, 1, 1)} &=&
   - \frac{1}{6} \sigma_{(1, 0, - 1)}
   + \frac{1}{2} \sigma_{(3, 1, -1)}
   + \frac{1}{2} \sigma_{(3, 1, 1)}
\\
\sigma_{(6, 5, 1)} &=&
- \frac{7}{10} - \frac{2}{3} \sigma_{(1, 0, -1})
- \sigma_{(3, 1, -1)}
+ \frac{1}{2} \sigma_{(3, 1, 1)}
\\
\sigma_{(6, 5, -1)} &=&
\frac{2}{15}
+ \frac{4}{3} \sigma_{(2, 1, -1)}
- \sigma_{(6, 1, -1)},~\text{etc.}
\label{CSSigma2}
\end{eqnarray*}
%----------------------------------------------------------------------------------------------
\begin{table}
\centering
%\begin{tabular}{l*{6}{c}r}
\scalebox{0.80}{
\begin{tabular}{|r|r|r|r|r|r|r|r|r|r|r|r|r|r|r|r|r|r|r|r|r|}
\hline
\multicolumn{1}{|c}{$l$}             &
\multicolumn{1}{|c}{1}                 &
\multicolumn{1}{|c}{2}                 &
\multicolumn{1}{|c}{3}                 &
\multicolumn{1}{|c}{4}                 &
\multicolumn{1}{|c}{5}                 &
\multicolumn{1}{|c}{6}                 &
\multicolumn{1}{|c}{7}                 &
\multicolumn{1}{|c}{8}                 &
\multicolumn{1}{|c}{9}                 &
\multicolumn{1}{|c}{10}                &
\multicolumn{1}{|c}{11}                &
\multicolumn{1}{|c}{12}                &
\multicolumn{1}{|c}{13}                &
\multicolumn{1}{|c}{14}                &
\multicolumn{1}{|c}{15}                &
\multicolumn{1}{|c}{16}                &
\multicolumn{1}{|c}{17}                &
\multicolumn{1}{|c}{18}                &
\multicolumn{1}{|c}{19}                &
\multicolumn{1}{|c|}{20}               \\
\hline
{\sf sums} &  2 & 4 & 6 & 8 & 10 & 12 & 14 & 16 & 18 & 20 & 22 & 24 & 26 & 28 & 30 & 32 & 34 & 36 & 38 &
40
\\
\hline
{\sf basis}
& 2 & 3 & 4 & 5 &  6 &  6 &  8 &  9 & 8 & 10 & 12 & 10 &14 & 14 & 11 & 17 & 18 & 14 & 20 & 18\\
\hline
{\sf new basis}
& 2 & 1 & 2 & 2 &  4 &  1 &  6 &  4 & 4 & 3 & 10 & 2 &12 & 5 & 3 & 8 & 16 & 4 & 18 & 6\\
{\sf sums}
& & & & & & & & & &
& & & & & & & & & & \\
\hline
\end{tabular}
}
\caption{\label{tablew1l20}The number of the $w = 1$ cyclotomic harmonic sums at infinity up to $l =20$, the basis elements at fixed value of $l$, and the new basis elements in ascending sequence.}
\end{table}
The following counting relations for the basis elements were tested up to $l = 700$ using computer algebra
methods.
Let $p,p_i,q$ be pairwise distinct primes $>2,$ and let $k,k_i$ be positive integers.
We conjecture that the number ${\sf N}(l)$ of basis elements at $w = 1$ and cyclotomy $l$ are given by
%------------------------------------------------------------------------------------------------------------
\begin{eqnarray*}
{\sf N}(l)&=&\left\{
		  	\begin{array}{lll}
						l+1,\  & l=1 \text{ or } l=2^k\\
						(p-1)p^{k-1}+2,\  & l=p^k\\
						2\cdot {\sf N}\left(2^{k-1}\prod_{i=1}^np_i^{k_i}\right)-n-1,\  & l=2^k\prod_{i=1}^np_i^{k_i}\\
						(q-1)\cdot{\sf N}\left(\prod_{i=1}^np_i^{k_i}\right)-n(q-2)-q+3,\  & l=q\prod_{i=1}^np_i^{k_i}\\
						q\cdot{\sf N}\left(q^{k-1}\prod_{i=1}^np_i^{k_i}\right)
-(n+2)(q-1),\  & l=q^k\prod_{i=1}^np_i^{k_i},\ k>1~.\\
			\end{array}
		\right.
\end{eqnarray*}
%------------------------------------------------------------------------------------------------------------

Let us now consider the case $w = 2$. Applying the relations given in the beginning of this section
the results given in Table~\ref{tablew2l20} are obtained for the number of basis
elements. Again we solved the corresponding linear systems using computer algebra methods
and derived the representations for the dependent sums analytically.
%----------------------------------------------------------------------------------------------
\begin{table}
\centering
\scalebox{0.80}{
\begin{tabular}{|r|r|r|r|r|r|r|r|r|}
\hline
\multicolumn{1}{|c}{$l$ }                &
\multicolumn{1}{|c}{$\sf N_S$}              &
\multicolumn{1}{|c}{\sf SH}                &
\multicolumn{1}{|c}{\sf A}                &
\multicolumn{1}{|c}{\sf A + SH}           &
\multicolumn{1}{|c}{$\sf A + SH + H_1$}      &
\multicolumn{1}{|c}{$\sf A  + SH + H_1 +H_2$}      &
\multicolumn{1}{|c|}{$\sf A + SH + H_1 + H_2 +M$}\\
\hline
  1 &   6 &   4 &   3 &   1 &   1 &   1 &    1 \\
  2 &  20 &  13 &  10 &   3 &   3 &   2 &    1 \\
  3 &  42 &  27 &  21 &   7 &   6 &   6 &    5 \\
  4 &  72 &  46 &  36 &  12 &  11 &  10 &    3 \\
  5 & 110 &  70 &  55 &  19 &  17 &  17 &   16 \\
  6 & 156 &  99 &  78 &  27 &  25 &  24 &    5 \\
  7 & 210 & 133 & 105 &  37 &  34 &  34 &   33 \\
  8 & 272 & 172 & 136 &  48 &  45 &  44 &   12 \\
  9 & 342 & 216 & 171 &  61 &  57 &  57 &   52 \\
 10 & 420 & 265 & 210 &  75 &  71 &  70 &   22 \\
 11 & 506 & 319 & 253 &  91 &  86 &  86 &   85 \\
 12 & 600 & 378 & 300 & 108 & 103 & 102 &   21 \\
 13 & 702 & 442 & 351 & 127 & 121 & 121 &  120 \\
 14 & 812 & 551 & 406 & 147 & 141 & 140 &   49 \\
 15 & 930 & 585 & 465 & 169 & 162 & 162 &  145 \\
 16 &1056 & 664 & 528 & 192 & 185 & 184 &   50 \\
 17 &1190 & 748 & 595 & 217 & 209 & 209 &  208 \\
 18 &1332 & 837 & 666 & 243 & 235 & 234 &   63 \\
 19 &1482 & 931 & 741 & 271 & 262 & 262 &  261 \\
 20 &1640 &1030 & 820 & 300 & 291 & 290 &   74 \\
\hline
\end{tabular}
}
\caption{\label{tablew2l20}Number of basis elements of the $ w = 2$ cyclotomic harmonic sums at infinity up to cyclotomy $l = 20$ after applying the quasi-shuffle algebra of
the sums (A), the shuffle algebra of the cyclotomic harmonic polylogarithms (SH), and the
three multiple argument relations ($\sf H_1, H_2, M$) of the sums.}
\end{table}

The number of the weight $ w = 2$ infinite sums for cyclotomy $l$ is
%------------------------------------------------------------------------------------------------------------
\begin{eqnarray}
{\sf N_S}(l) = 2l (2l+1)~.
\end{eqnarray}
%------------------------------------------------------------------------------------------------------------
One may guess, based on the results for $l \leq 20$,  counting relations for the
length of the bases listed in Table~\ref{tablew2l20}. We found for all but the last column:
%------------------------------------------------------------------------------------------------------------
\begin{eqnarray*}
{\sf N_A}(l)  &=& l(2l+1)
%\end{eqnarray}
%\begin{eqnarray}
\\
{\sf N_{SH}}(l)  &=& \frac{(5l+3)l}{2}
\\
{\sf N_{A,SH}}(l) &=& \frac{6l^2+1-(-1)^l}{8}
\\
{\sf N_{A,SH,H_1}}(l) &=& \frac{6l^2-4l+7-(-1)^l}{8}
\\
%\label{eq:BROAD}
{\sf N_{A,SH,H_1,H_2}}(l) &=& \frac{6l^2-4l+3(1-(-1)^l)}{8}~.
\end{eqnarray*}

\subsubsection*{Cyclotomic Harmonic Sums of Higher Weights at Infinity}
First we consider the iterated summation of the terms 
\begin{eqnarray*}
\frac{1}{k^{l_1}},~~~~~~
\frac{(-1)^k}{k^{l_2}},~~~~~~
\frac{1}{(2k+1)^{l_3}},~~~~~~
\frac{(-1)^k}{(2k+1)^{l_4}},
\end{eqnarray*}
with $l_i\in\N.$ A corresponding sum $\S{(a_1,b_1,c_1),\ldots,(a_l,b_l,c_l)}n$
diverges if it has leading ones, \ie if $c_1~=~1$. However, these 
divergences
can be regulated by polynomials in $\sigma_{(1,0,1)}~=~\lim_{n\rightarrow \infty}\S{1}n$ and cyclotomic harmonic sums, which are convergent for
$n \rightarrow \infty$, very similar to the case of the usual harmonic sums; see Remark \ref{HSextractleading1rem}.
In Table~\ref{CSrell2w6} we present the number of basis elements obtained applying the
respective relations up to weight $w = 6$. The representation of all sums were
computed by means of computer algebra in explicit form.
We derive the cumulative basis, quoting only the new elements in the next
weight. For suitable integral representations
over known functions we refer to \cite{Ablinger2011}.
%----------------------------------------------------------------------------------------------
\begin{table}
\begin{center}
\scalebox{0.80}{
\begin{tabular}{| r | r | r | r | r | r | r | r | r|}
\hline
{\sf weight} & $\sf N_S$ & $\sf A$  & $\sf SH$ &  $\sf A + SH $ &  $\sf A+SH+H_1$   &
$\sf A + SH +H_1 + H_2$ & $\sf A + SH + H_1 + H_2 + M$\\
\hline
   1 &     4 &    4 &    4 &   4 &   4 &   3 &   3  \\
   2 &    20 &   10 &   13 &   3 &   3 &   2 &   1  \\
   3 &   100 &   40 &   46 &   6 &   6 &   5 &   3  \\
   4 &   500 &  150 &  163 &  10 &  10 &   9 &   6  \\
   5 &  2500 &  624 &  650 &  21 &  21 &  19 &  13  \\
   6 & 12500 & 2580 & 2635 &  36 &  36 &  34 &  25  \\
 \hline
\end{tabular}
}
\end{center}
\caption{\label{CSrell2w6}Basis representations of the infinite cyclotomic harmonic sums over the
alphabet $\{(\pm 1)^k/k, (\pm 1)^k/(2k+1)\}$ after applying the stuffle (A), shuffle
(SH) relations, their combination, and their application together with the three
multiple argument relations ($\sf H_1, H_2, M$).} 
\end{table}
%----------------------------------------------------------------------------------------------
One possible choice of basis elements is~:
\begin{eqnarray*}
w = 1: \hspace*{5cm}\\
&&\hspace*{-6cm}\sigma_{(1, 0, 1)},\sigma_{(1, 0, -1)}, \sigma_{(2, 1, -1)}\\
w = 2: \hspace*{5cm}\\
&&\hspace*{-6cm}\sigma_{(2, 1, -2)} \\
w = 3: \hspace*{5cm}\\
&&\hspace*{-6cm} \sigma_{(1, 0, 3)},\sigma_{(1, 0, -2), (2, 1, -1)}, \sigma_{(2, 1, -2), (1, 0, -1)} \\
w = 4: \hspace*{5cm}\\
&&\hspace*{-6cm}\sigma_{(1, 0, -1), (1, 0, 1), (1, 0, 1),(1, 0, 1)}, \sigma_{(2, 1, -4)}, \sigma_{(2, 1, -3), (2, 1, -1)}, \sigma_{(1, 0, -2), (1, 0, -1), (2, 1, -1)},\\
&&\hspace*{-6cm} \sigma_{(1, 0, -2), (2, 1, -1), (1, 0, 1)},\sigma_{(1, 0, -2), (2, 1, -1), (2, 1, 1)}\\
w = 5: \hspace*{5cm}\\
&& \hspace*{-6cm}\sigma_{(1, 0, 5)} \sigma_{(1, 0, -1), (1, 0, 1), (1, 0, 1), (1, 0, 1), (1, 0, 1)} \sigma_{(1, 0, -4), (2, 1, -1)} \sigma_{(1, 0, 4), (2, 1, -1)}\\
&& \hspace*{-6cm}\sigma_{(2, 1, -4), (1, 0, -1)},\sigma_{(1, 0, -3), (1, 0, -1), (2, 1, 1)},\sigma_{(1, 0, -3), (2, 1, -1), (2, 1, -1)},\sigma_{(1, 0, 3), (2, 1, -1), (2, 1, -1)},\\
&& \hspace*{-6cm}\sigma_{(2, 1, -3), (2, 1, -1), (2, 1, 1)},\sigma_{(1, 0, -2), (1, 0, -1), (1, 0, -1), (2, 1, -1)},\sigma_{(1, 0, -2), (1, 0, -1), (2, 1, -1), (1, 0, -1)}, \\
&& \hspace*{-6cm}\sigma_{(1, 0, -2), (2, 1, -1), (1, 0, 1), (1, 0, 1)},\sigma_{(1, 0, -2), (2, 1, -1), (2, 1, 1), (1, 0, 1)}\\
w = 6: \hspace*{5cm}\\
&& \hspace*{-6cm}
\sigma_{(1, 0, -5), (1, 0, -1)} ,\sigma_{(1, 0, -1), (1, 0, 1), (1, 0, 1), (1, 0, 1), (1, 0, 1),(1, 0, 1)} ,\sigma_{(2, 1, -6)}\\
&& \hspace*{-6cm}
\sigma_{(1, 0, -2), (2, 1, -1), (2, 1, 1), (1, 0, 1), (1, 0, 1)},
\sigma_{(1, 0, -2), (2, 1, -1), (1, 0, 1), (1, 0, 1), (1, 0, 1)},
\nonumber\\
&& \hspace*{-6cm}
\sigma_{(1, 0, -2), (1, 0, -1), (1, 0, -1), (2, 1, -1), (2, 1, 1)},
\sigma_{(1, 0, -2), (1, 0, -1), (1, 0, -1), (2, 1, -1), (1, 0, 1)},
\nonumber\\
&& \hspace*{-6cm}
\sigma_{(1, 0, -2), (1, 0, -1), (1, 0, -1), (1, 0, -1), (2, 1, -1)},
\sigma_{(1, 0, 3), (2, 1, -1), (2, 1, -1), (1, 0, 1)},
\nonumber\\
&& \hspace*{-6cm}
\sigma_{(1, 0, -3), (1, 0, -1), (2, 1, 1), (1, 0, 1)},
\sigma_{(1, 0, -3), (1, 0, -1), (1, 0, 1), (2, 1, 1)},
\nonumber\\
&& \hspace*{-6cm}
\sigma_{(2, 1, -3), (2, 1, -1), (2, 1, -1), (2, 1, -1)},
\sigma_{(1, 0, -3), (2, 1, -1), (2, 1, -1), (1, 0, -1)},
\nonumber\\
&& \hspace*{-6cm}
\sigma_{(1, 0, -3), (2, 1, -1), (1, 0, -1), (2, 1, -1)},
\sigma_{(1, 0, -3), (1, 0, -1), (2, 1, -1), (2, 1, -1)},
\sigma_{(2, 1, 4), (1, 0, -1), (2, 1, -1)},
\nonumber\\
&& \hspace*{-6cm}
\sigma_{(1, 0, 4), (2, 1, -1), (1, 0, -1)},
\sigma_{(1, 0, 4), (1, 0, -1), (2, 1, -1)},
\sigma_{(2, 1, -4), (1, 0, -1), (2, 1, 1)},
\nonumber\\
&& \hspace*{-6cm}
\sigma_{(2, 1, -4), (1, 0, -1), (1, 0, 1)},
\sigma_{(1, 0, -4), (2, 1, -1), (2, 1, 1)},
\sigma_{(1, 0, -4), (2, 1, -1), (1, 0, 1)},
\nonumber\\
&& \hspace*{-6cm}
\sigma_{(1, 0, -4), (1, 0, -1), (2, 1, -1)},
\sigma_{(2, 1, -4), (2, 1, -2)},
\sigma_{(2, 1, -5), (2, 1, -1)}~.
\end{eqnarray*}

As a further example of the infinite cyclotomic harmonic sums we consider the iterated summation of the terms 
\begin{eqnarray*}
\frac{1}{k^{l_1}},~~~~~~
\frac{(-1)^k}{k^{l_2}},~~~~~~
\frac{1}{(3k+1)^{l_3}},~~~~~~
\frac{(-1)^k}{(3k+1)^{l_4}},~~~~~~
\frac{1}{(3k+2)^{l_4}},~~~~~~
\frac{(-1)^k}{(3k+2)^{l_5}},
\end{eqnarray*}
with $l_i\in\N.$
In Table~\ref{CSrell3w4} we list the corresponding numbers of basis elements up to weight 4.
\begin{table}
\begin{center}
\scalebox{0.80}{
\begin{tabular}{| r | r | r | r | r | r | r | r | r|}
\hline
{\sf weight} & $\sf N_S$ & $\sf A$  & $\sf SH$ &  $\sf A + SH $ &  $\sf A+SH+H_1$   &
$\sf A + SH +H_1 + H_2$ & $\sf A + SH + H_1 + H_2 + M$\\
\hline
   1 &     6 &    6 &    6 &   6 &   5 &   4 &   4  \\
   2 &    42 &   21 &   27 &   7 &   6 &   6 &   5  \\
   3 &   294 &  112 &  131 &  17 &  16 &  16 &  15  \\
   4 &  2058 &  588 &  651 &  40 &  37 &  37 &  36  \\
 \hline
\end{tabular}
}
\end{center}
\caption{\label{CSrell3w4}Basis representations of the infinite cyclotomic harmonic sums over the
alphabet $\{(\pm 1)^k/k, (\pm 1)^k/(3k+1), (\pm 1)^k/(3k+2)\}$ after applying the stuffle (A), shuffle
(SH) relations, their combination, and their application together with the three
multiple argument relations ($\sf H_1, H_2, M$).} 
\end{table}

\section{Asymptotic Expansion of Cyclotomic Harmonic Sums}
\label{CSExpansion}
In this section we want to generalize the algorithm for the computation of the asymptotic expansion of harmonic sums to cyclotomic harmonic sums. Again the integral representation of these sums will 
be the starting point. 
In the algorithm for the computation of the asymptotic expansion of harmonic sums the possibility to express a harmonic polylogarithm $\H{\ve m}x$ in terms of harmonic polylogarithms at argument 
$1-x$ plays a decisive role. This remains true for cyclotomic harmonic polylogarithms. In the case of harmonic polylogarithms we had to extend the index set to $\{-1,0,1,2\}$ in order to be able 
to perform this transformation. Since in general we cannot express a cyclotomic harmonic polylogarithm $\H{\ve m}x$ in terms of cyclotomic harmonic polylogarithms at argument $1-x,$ we will introduce a new set 
of functions which is closely related to cyclotomic harmonic polylogarithms and which we will call {\itshape shifted cyclotomic harmonic polylogarithms}.

\subsection{Shifted Cyclotomic Harmonic Polylogarithms}
\label{CSshiftedhlogs}
In the following we present new additional aspects which can not be found in \cite{Ablinger2011}.
We start by defining the new auxiliary function $g:$ For $a \in \N$ and $b \in \N,$ $b < \varphi(a)$ we define
\begin{eqnarray}
&&g_a^b:(0,1)\mapsto \R\nonumber\\
&&g_a^b(x)=\left\{ 
		\begin{array}{ll}
				\frac{1}{1-x}, &  \textnormal{if }a=b=0\\
				\frac{(1-x)^b}{\Phi_a(1-x)}, & \textnormal{otherwise}.  
		\end{array} 
		\right.  \nonumber
\end{eqnarray}
Here $\Phi_a(x)$ denotes again the $a$th cyclotomic polynomial and $\varphi$ denotes Euler's totient function. Now we are ready to define {\itshape shifted cyclotomic harmonic polylogarithms}:
\begin{definition}[Shifted Cyclotomic Harmonic Polylogarithms]
Let $m_i~=~(a_i,b_i)$ where $a_i,b_i~\in~\N$ with $b_i<\varphi(a_i);$  we define for $x\in (0,1):$
\begin{eqnarray}
\Hs{}{x}&=&1,\nonumber\\
\Hs{m_1,\ldots,m_k}{x} &=&\left\{ 
		  	\begin{array}{ll}
						\frac{(-1)^k}{k!}(\log{x})^k,&  \textnormal{if } m_i=(1,0) \textnormal{ for } 1\leq i \leq k \\
									    &\\
						\int_0^x {g_{a_1}^{b_1}(y)} \Hs{m_2,\ldots,m_k}{y}dy,& \textnormal{otherwise}. 
					\end{array} \right.  \nonumber
\end{eqnarray}
The length $k$ of the vector $\ve m$ is called the weight of the shifted cyclotomic harmonic polylogarithm $\H{\ve m}x.$
\label{CShlogdefext}
\end{definition}

\begin{example}
\begin{eqnarray*}
\Hs{(5,3)}x&=&\int_0^x\frac{(1-y)^3}{\Phi_5(1-y)}dy=\int_0^x\frac{(1-y)^3}{y^4 - 5 y^3 +10 y^2 - 10 y+ 5}dy\\
\vspace{1cm}\\
\Hs{(3,1),(2,0),(5,2)}x&=&\int_0^x\frac{1-y}{\Phi_3(1-y)}\int_0^y\frac{1}{\Phi_2(1-z)}\int_0^z\frac{(1-w)^2}{\Phi_5(1-w)}dwdzdy\\
\vspace{1cm}\\
\Hs{(1,0),(0,0),(0,0)}x&=&\int_0^x\frac{1}{-y}\int_0^y\frac{1}{1-z}\int_0^z\frac{1}{1-w}dwdzdy=-\H{(0,0),(1,0),(1,0)}x.
\end{eqnarray*}
\end{example}

\begin{remark}Note that these functions are again analytic for $x\in(0,1)$ and 
for the derivatives we have for all $x\in (0,1)$ that $$ \frac{d}{d x} \Hs{(a_1,b_1),(a_2,b_2),\ldots,(a_k,b_k)}{x}=g_{a_1}^{b_1}(x)\Hs{(a_2,b_2),\ldots,(a_k,b_k)}{x}. $$ 
\end{remark}
In the following we will see that shifted cyclotomic harmonic polylogarithms at argument $1-x$ can be expressed in terms of cyclotomic harmonic polylogarithms at argument $x$ together with constants. 
We proceed recursively on the weight $w$ of the shifted cyclotomic harmonic polylogarithm. In the base case we have for $x \in (0,1)$  and $a,b\in \N$,$1<a$, $b<\varphi(a):$
\begin{eqnarray}
\Hs{(0,0)}{1-x}&=&-\H{(0,0)}{x}\label{CSshiftedtrafo0}\\
\Hs{(1,0)}{1-x}&=&-\H{(1,0)}{x}\label{CSshiftedtrafo1}\\
\Hs{(a,b)}{1-x}&=&\H{a,b}1-\H{(a,b)}{x}.
\end{eqnarray}
Now let us look at higher weights $w>1.$ We consider $\Hs{m_1,m_2,\ldots,m_w}{1-x}$ and suppose that we can already apply the transform for shifted cyclotomic harmonic 
polylogarithms of weight $<w.$ If $m_1=(0,0),$ we can remove leading $(0,0)$ and end up with shifted cyclotomic harmonic polylogarithms without leading $(0,0)$ and powers 
of $\Hs{(0,0)}{1-x}.$ For the powers of $\H{(0,0)}{1-x}$ we can use (\ref{CSshiftedtrafo0}); therefore, only the cases in which the first index $m_1\neq (0,0)$ are to 
be considered. If $m_i=(1,0)$ for all $1<i\leq w,$ we are in fact dealing with a power of $\Hs{(1,0)}{1-x}$ and hence we can use (\ref{CSshiftedtrafo1}) to transform these 
cyclotomic harmonic polylogarithms.
Now let us assume that not all $m_i=(1,0).$ For $m_1=(1,0)$ we get:
\begin{eqnarray*}
\Hs{(1,0),m_2,\ldots,m_w}{1-x}&=&\int_0^{1-x}{\frac{\Hs{m_2,\ldots,m_w}y}{-y}dy}\\
		&=&-\Hs{(1,0),m_2,\ldots,m_w}1+\int_{1-x}^{1}{\frac{\Hs{m_2,\ldots,m_w}y}{y}dy}\\
		&=&-\Hs{(1,0),m_2,\ldots,m_w}1+\int_{0}^{x}{\frac{\Hs{m_2,\ldots,m_w}{1-y}}{1-y}dy}.
\end{eqnarray*}
For $m_1=(a,b)$ with $a,b\in \N,1<a, b<\varphi(a)$ we get:
\begin{eqnarray*}
\Hs{(a,b),m_2,\ldots,m_w}{1-x}&=&\int_0^{1-x}{\frac{(1-y)^b\Hs{m_2,\ldots,m_w}y}{\Phi_a(1-y)}dy}\\
		&=&\Hs{(a,b),m_2,\ldots,m_w}1-\int_{1-x}^{1}{\frac{(1-y)^b\Hs{m_2,\ldots,m_w}y}{\Phi_a(1-y)}dy}\\
		&=&\Hs{(a,b),m_2,\ldots,m_w}1-\int_{0}^{x}{\frac{y^b\Hs{m_2,\ldots,m_w}{1-y}}{\Phi_a(y)}dy}.
\end{eqnarray*}
Since we know the transform for weights $<w$ we can apply it to $\Hs{m_2,\ldots,m_w}{1-y}$ and 
finally we obtain the required weight $w$ identity by using the definition of the cyclotomic harmonic polylogarithms.
\begin{example}
\begin{eqnarray*}
\Hs{(5,2),(3,1)}{1-x}&=&-\textnormal{H}_{(3,1)}(1) \textnormal{H}_{(5,2)}(x)+\textnormal{H}_{(5,2),(3,1)}(x)+\textnormal{H}_{(3,1)}(1) \textnormal{H}_{(5,2)}(1)\\
		    & &-\textnormal{H}_{(5,2),(3,1)}(1).
\end{eqnarray*}
\end{example}

\begin{remark}
For $a,b\in \N_0$ with $a>1$ and $b<\phi(a)$ we have $\Hs{(a,b)}1=\H{(a,b)}1.$ For the values of shifted cyclotomic harmonic polylogarithms at argument $1$ of higher weights
we can use the above procedure to express them by cyclotomic harmonic polylogarithms at argument $1.$
Note that we can as well perform the reverse direction, \ie we can find a combination of shifted cyclotomic harmonic 
polylogarithms at argument $1-x$ together with some constants to express a cyclotomic harmonic polylogarithm $\H{\ve m}x.$ 
\end{remark}
\begin{example}
\begin{eqnarray*}
 \H{(3,1),(5,2)}{x}&=&-\H{(5,2)}{1} \Hs{(3,1)}{1-x}+\Hs{(3,1),(5,2)}{1-x}+\H{(3,1)}{1}\H{(5,2)}{1}\\&&-\Hs{(3,1),(5,2)}{1}\\
\Hs{(3,1),(5,2)}{1}&=&\H{(3,1)}{1} \H{(5,2)}{1}-\H{(3,1),(5,2)}{1}.
\end{eqnarray*}
\end{example}

\subsection{Computation of Asymptotic Expansions of Cyclotomic Harmonic Sums}
Before we can state an algorithm to compute the asymptotic expansions for cyclotomic harmonic sums, we have to extend the Lemmas \ref{HSanalytic1}, \ref{HSanalytic2} and \ref{HSmelexpconst}. 
Since the proofs of these lemmas can be easily extended we will omit them here.
\begin{lemma}
Let $\H{m_1,m_2,\ldots,m_k}x$ be a cyclotomic harmonic polylogarithm with $m_i\neq~(1,0)$ for $1\leq i\leq k$ and let $a\in\N,$ with $a>1.$ Then 
$$\H{m_1,m_2,\ldots,m_k}x,\; \frac{\H{m_1,m_2,\ldots,m_k}x}{\Phi_a(x)}$$
and
$$\frac{\H{m_1,m_2,\ldots,m_k}x-\H{m_1,m_2,\ldots,m_k}1}{1-x}$$
are analytic for $x \in (0,\infty).$
\label{CSanalytic1}
\end{lemma}

\begin{lemma}
Let $\Hs{m_1,m_2,\ldots,m_k}x$ be a shifted cyclotomic harmonic polylogarithm with $m_k\neq (1,0)$ and let $a\in\N.$ Then 
$$\frac{\Hs{m_1,m_2,\ldots,m_k}{1-x}}{\Phi_a(x)}$$
is analytic for $x \in (0,2).$
\label{CSanalytic2}
\end{lemma}

\begin{lemma}
Let $\H{\ve m}x=\H{m_1,m_2,\ldots,m_k}x$ be a cyclotomic harmonic polylogarithm with $m_1\neq~(1,0)$ and let $\Hs{\ve b}x=\Hs{b_1,b_2,\ldots,b_l}x$ 
be a shifted cyclotomic harmonic polylogarithm with $b_l\neq~(1,0).$ Then we have 
\begin{eqnarray*}
\M{\frac{\H{\ve m}x}{x-1}}{in+p}&=&\int_0^1{\frac{x^{in+p}(\H{\ve m}x-\H{\ve m}1)}{x-1}dx}-\int_0^1{\frac{\H{\ve m}x-\H{\ve m}1}{x-1}}dx\\
	    &&-\S{1}{in+p}\H{\ve m}{1}
\end{eqnarray*}
and
\begin{eqnarray*}
\M{\frac{\Hs{\ve b}{1-x}}{x-1}}{jn+p}&=&\int_0^1{\frac{x^{jn+p}\Hs{\ve b}{1-x}}{x-1}dx}+\Hs{(1,0),\ve b}1
\end{eqnarray*}
where
$$\int_0^1{\frac{\H{\ve m}x-\H{\ve m}1}{x-1}}dx, \ \H{\ve m}{1} \textnormal{ and } \Hs{(1,0),\ve b}1 $$
are finite constants.
\label{CSmelexpconst}
\end{lemma}

\begin{remark}
\label{CSExpandableIntegrals}
Combining Lemma \ref{CSanalytic1}, Lemma \ref{CSanalytic2} and \ref{CSmelexpconst} we are able to expand Mellin transforms of the form
$$\M{\frac{\H{m_1,m_2,\ldots,m_k}x}{\Phi_a(x)}}{in+p} \ \textnormal{ and } \M{\frac{\Hs{b_1,b_2,\ldots,b_l}{1-x}}{\Phi_a(x)}}{jn+p}$$
where $m_r\neq (1,0)$ for $1\leq r\leq k$ and $b_l\neq (1,0)$ and suitable $i$ and $j$ (see Section \ref{CSCyloMel}) following the method presented in Section \ref{HShexp}, 
or using repeated integration by parts.
\end{remark}

\begin{remark}
\label{CSTrailing1rem}
Analyzing the method to compute the inverse Mellin transform of cyclotomic harmonic sums presented in Section \ref{CSintrep} we find out that the cyclotomic harmonic 
polylogarithms with highest weight in the inverse Mellin transform of a cyclotomic harmonic sum without trailing ones do not have trailing $(1,0),$ while the cyclotomic harmonic 
polylogarithms with highest weight in the inverse Mellin transform of a cyclotomic harmonic sum with trailing ones have trailing $(1,0).$\\
Using the method presented in Section \ref{CSshiftedhlogs} we can always express a cyclotomic harmonic polylogarithm without trailing $(1,0)$ at argument $x$ using shifted cyclotomic harmonic 
polylogarithms at argument $1-x$ where the shifted cyclotomic harmonic polylogarithms with highest weight do not have trailing $(1,0).$ 
\end{remark}

Now we are ready to state an algorithm to compute asymptotic expansions of cyclotomic harmonic sums assuming that we are able to expand linear cyclotomic harmonic sums (we will treat
linear cyclotomic harmonic sums in the subsequent subsection).
If we want to find the asymptotic expansions of a cyclotomic harmonic sum $\S{m_1,m_2,\ldots,m_k}n=\S{(a_1,b_1,c_1),\ldots,(a_k,b_k,c_k)}n,$ we can proceed as follows:

\begin{itemize}
	\item if $\S{m_1,m_2,\ldots,m_k}n$ has trailing ones, \ie $c_k=1$ we first extract them (see Section \ref{CSextracttrailing}); 
	treat linear cyclotomic harmonic sums in the subsequent subsection; apply the following items to each of the 
	cyclotomic harmonic sums without trailing ones;
	\item suppose now $\S{m_1,m_2,\ldots,m_k}n=\S{(a_1,b_1,c_1),\ldots,(a_k,b_k,c_k)}n$ has no trailing ones, \ie $c_k\neq 1;$ let 
	$\frac{\H{\ve q_1}x}{\Phi_{r_1}(x)},\ldots,\frac{\H{\ve q_v}x}{\Phi_{r_v}(x)}$ be the weighted cyclotomic harmonic polylogarithms of highest weight in the inverse Mellin 
	transform of $\S{m_1,m_2,\ldots,m_k}n$ and let $\M{\frac{\H{\ve q_1}x}{\Phi_{r_1}(x)}}{l_1n+p_1},\ldots, \M{\frac{\H{\ve q_v}x}{\Phi_{r_v}(x)}}{l_vn+p_v}$ be 
	the derived associated Mellin transforms; we can express $\S{m_1,m_2,\ldots,m_k}n$ as
	\begin{equation}\label{CSasyalg1}
		\S{m_1,m_2,\ldots,m_k}n=c_1\M{\frac{\H{\ve q_1}x}{\Phi_{r_1}(x)}}{l_1n+p_1}+\cdots +c_v\M{\frac{\H{\ve q_v}x}{\Phi_{r_v}(x)}}{l_vn+p_v}+T,
	\end{equation}
	where $c_i\in \R$ and $T$ is an expression in cyclotomic harmonic sums (which have smaller weight than $\S{m_1,m_2,\ldots,m_k}n$) and constants;
	\item we proceed by expanding each of the $\M{\frac{\H{\ve q_i}x}{\Phi_{r_i}(x)}}{l_i n+p_i}$; let $\M{\frac{\H{q_1,\ldots,q_s}x}{\Phi_{r}(x)}}{l n+p}$ be one of these 
	Mellin transforms:
	\begin{description}
		\item[all $q_i\neq (1,0)$:] expand $\M{\frac{\H{q_1,q_2,\ldots,q_s}x}{\Phi_{r}(x)}}{n}$ directly see Remark \ref{CSExpandableIntegrals}

		\item[not all $m_i\neq (1,0)$:]
		\begin{itemize}\item[]
			\item express $\H{\ve q}x$ in terms of shifted cyclotomic harmonic polylogarithms at argument $1-x$; expand all products; then we get
				\begin{equation}\label{CSasyalg2}
					\M{\frac{\H{\ve q}x}{\Phi_{r}(x)}}{ln+p}=\sum_{i=1}^tc_i\M{\frac{\Hs{\ve f_i}{1-x}}{\Phi_{r}(x)}}{ln+p}+c \ \textnormal{  with } c,c_i\in\R
				\end{equation}
			 \item for each Mellin transform $\M{\frac{\Hs{f_1,\ldots,f_j}{1-x}}{\Phi_{r}(x)}}{ln+p}$ do
				\begin{description}\item[]
					\item[$f_j\neq (1,0):$] expand $\M{\frac{\Hsma{\ve f}{1-x}}{\Phi_{r}(x)}}{ln+p}$ as given in Remark \ref{CSExpandableIntegrals}
					\item[$f_j=(1,0):$] express $\Hs{\ve f}{1-x}$  by cyclotomic harmonic polylogarithms at $x$ see Section~\ref{CSshiftedhlogs}; expand all products; hence we can write
						\begin{equation}\label{CSasyalg3}
							\M{\frac{\Hs{\ve f}{1-x}}{\Phi_{r}(x)}}{ln+p}=\sum_{i=1}^td_i\M{\frac{\H{\ve g_i}{x}}{\Phi_{r}(x)}}{ln+p}+d
						\end{equation}
						with $d,d_i\in\R$ and perform the Mellin transforms $\M{\frac{\H{\ve g_i}{x}}{\Phi_{r}(x)}}{ln+p};$ see Section~\ref{CSCyloMel}
				\end{description}
		\end{itemize}
	\end{description}
	\item replace the $\M{\frac{\H{\ve q_i}x}{\Phi_{r_i}(x)}}{l_i n+p_i}$ in equation (\ref{CSasyalg1}) by the results of this process
	\item for all nonlinear cyclotomic harmonic sums that remain in equation (\ref{CSasyalg1}) apply the above points; since these cyclotomic harmonic sums have smaller weights
	      than $\S{m_1,m_2,\ldots,m_k}n$ this process will terminate
\end{itemize}

Some remarks are in place: Since $c_k\neq 1$ in equation (\ref{CSasyalg1}), we know (see Remark \ref{CSTrailing1rem}) that the cyclotomic harmonic polylogarithms $\H{\ve q_i}x$ 
in equation (\ref{CSasyalg1}) do not have trailing index $(1,0)$. Since the cyclotomic harmonic polylogarithms $\H{\ve q_i}x$ do not have the trailing index $(1,0),$ the shifted
cyclotomic harmonic polylogarithms of highest weight in equation (\ref{CSasyalg2}) will not have the trailing index $(1,0)$ (see Remark \ref{CSTrailing1rem}). Therefore the 
nonlinear cyclotomic harmonic sums which will appear in equation (\ref{CSasyalg3}) have smaller weight than $\S{m_1,m_2,\ldots,m_k}n$ of (\ref{CSasyalg1}) and hence this 
algorithm will eventually terminate.

\begin{example}Using the algorithm from above we compute the asymptotic expansion of $\S{(3,1,2),(1,0,1)}n$ up to order $3$:
\small
\begin{eqnarray*}
&&\frac{1}{972} \biggl(
-\frac{1}{n^3}\bigl(120\textnormal{H}_{(3,0),(0,0)}(1)+60 \textnormal{H}_{(3,1),(0,0)}(1)-60 \zeta_2+67\bigr)
+\frac{1}{n^2}\bigl(168 \textnormal{H}_{(3,0),(0,0)}(1)\bigl.\\&&\bigr.+84 \textnormal{H}_{(3,1),(0,0)}(1)-84 \zeta_2+81\bigr)
+\frac{108}{n}\bigl(-2\textnormal{H}_{(3,0),(0,0)}(1)-\textnormal{H}_{(3,1),(0,0)}(1)+\zeta_2-1\bigr)\\
&&-162 \bigl(-12 \textnormal{H}_{(3,0)}(1)-6 \textnormal{H}_{(3,1)}(1)+12-6 \textnormal{H}_{(3,0),(0,0)}(1)-3 \textnormal{H}_{(3,1),(0,0)}(1)\bigr.\\
&&-4\textnormal{H}_{(3,0),(0,0),(1,0)}(1)-6 \textnormal{H}_{(3,0),(1,0),(0,0)}(1)+4 \textnormal{H}_{(3,0),(3,0),(0,0)}(1)+2 \textnormal{H}_{(3,0),(3,1),(0,0)}(1)\\
&&\bigl.-2 \textnormal{H}_{(3,1),(0,0),(1,0)}(1)-6 \textnormal{H}_{(3,1),(1,0),(0,0)}(1)+8
   \textnormal{H}_{(3,1),(3,0),(0,0)}(1)+4 \textnormal{H}_{(3,1),(3,1),(0,0)}(1)\\
&&+3 \zeta_2-2 \zeta_3-3\bigr)\biggr)+\frac{1}{3} \textnormal{S}_{(3,1,1)}(n) \left(2 \textnormal{H}_{(3,0),(0,0)}(1)+\textnormal{H}_{(3,1),(0,0)}(1)-\zeta_2-6\right)\\
&&+\frac{1}{3} \textnormal{S}_{(3,2,1)}(n) \left(2 \textnormal{H}_{(3,0),(0,0)}(1)+\textnormal{H}_{(3,1),(0,0)}(1)-\zeta_2+3\right)\\
&&+\frac{1}{162} \textnormal{S}_1(n) \left(-72 \textnormal{H}_{(3,0),(0,0)}(1)-36 \textnormal{H}_{(3,1),(0,0)}(1)-\frac{11}{n^3}+\frac{15}{n^2}-\frac{18}{n}+36 \zeta_2+54\right). 
\end{eqnarray*}
\normalsize
In order to expand the remaining linear cyclotomic sums, \ie $\textnormal{S}_{(3,1,1)}(n),\textnormal{S}_{(3,2,1)}(n)$ and $\S{1}n,$ we refer to the next section.
\label{CSasyexp1}
\end{example}

\subsubsection*{Asymptotic Expansion of {\itshape Linear} Cyclotomic Harmonic Sums}
So far we do not know how to compute the asymptotic expansion of Mellin transforms of the form
$$ 
\M{\frac{\H{m_1,m_2,\ldots,m_k}x}{\Phi_a(x)}}{in+p}
$$
with $m_k=(1,0).$\\
Since the cyclotomic harmonic polylogarithms with highest weight in the inverse Mellin transform of a cyclotomic harmonic sum with trailing ones have trailing $(1,0),$ we have to find 
a new method for these sums. From Section \ref{HShexp} we know how to compute the expansion of $\S{1,\ldots,1}n$ and in the following we will connect the computation of the asymptotic 
expansion of general linear cyclotomic sums to $\S{1,\ldots,1}n.$\\
We start with the weight $w=1$ case, \ie we consider the sum $\S{(a,b,1)}n.$ If $b=0,$ we have 
$$
\S{(a,b,1)}n=\frac{1}{a}\S{1}n,
$$
and since we know the expansion of $\S{1}n$ we can handle this case. Now assume that $b\neq 0.$ Then we have
$$
\S{(a,b,1)}n=\sum_{i=1}^{n}\left(\frac{1}{a i+b}-\frac{1}{a i}\right)+\frac{1}{a}\S{1}n.
$$
The integral representation of $$\sum_{i=1}^{n}\left(\frac{1}{a i+b}-\frac{1}{a i}\right)$$ yields (compare Lemma \ref{CSintrep1})
\begin{eqnarray*}
\sum_{i=1}^{n}\left(\frac{1}{a i+b}-\frac{1}{a i}\right)&=&\int_0^1\frac{x^{a-1} \left(x^{b}-1\right) \left(x^{a n}-1\right)}{x^{a}-1}dx\\
				&=&\int_0^1x^{a-1}\left(x^{a n}-1\right)\frac{\Phi_{b} \prod_{d|a, d < a}\Phi_{d}}{\Phi_{a}\prod_{d|b, d < b}\Phi_{d}}dx\\
				&=&\int_0^1x^{a-1}\left(x^{a n}-1\right)\frac{\Phi_{b} \prod_{d|a, 1 < d < a}\Phi_{d} }{\Phi_{a}\prod_{d|b, 1 < d < b}\Phi_{d}}dx.
\end{eqnarray*}
At this point we may perform a partial fraction decomposition (note that there is no factor of $(1-x)$ in the denominator) and split the integrals. We can already compute the
asymptotic expansions of these integrals (see Remark \ref{CSTrailing1rem}) and since we know the expansion of $\S1n$ we can handle the weight one case.\\
We will now assume that we are 
able to compute the asymptotic expansions of linear cyclotomic harmonic sums with weight $w<k$ and we will consider the weight $w=k$ case.
Note that we can always relate a cyclotomic harmonic sum with an index $(a_i,0,c_i)$ to a cyclotomic harmonic sum with index $(1,0,c_i)$ since
 $$\S{(a_1,b_1,c_1),\ldots,(a_i,0,c_i),\ldots,(a_k,b_k,c_k)}n=\frac{1}{{a_i}^{c_i}}\S{(a_1,b_1,c_1),\ldots,(1,0,c_i),\ldots,(a_k,b_k,c_k)}n;$$
hence in the following we can always assume that $a_i=1$ if an index $(a_i,0,c_i)$ is present.\\
Let now $\S{m_1,\ldots,m_k}n=\S{(a_1,b_1,1),\ldots,(a_k,b_k,1)}n$ be a linear cyclotomic harmonic sum with $a_i=1$ if $b_i=0$ and let $q$ be the number of indices $m_i$ that 
are different from $(1,0,1).$ 
If $m_k=(1,0,1),$ we extract the trailing $(1,0,1)$ and end up with a polynomial in $\S{1}n$ with coefficients in cyclotomic harmonic sums (these sums do not have to be linear 
cyclotomic sums) without trailing $(1,0,1).$ From the nonlinear cyclotomic harmonic sums we extract possible trailing ones (see Remark \ref{CSextractleading1rem}). We can handle 
the powers of $\S{1}n$ and for the sums without trailing ones we can use the algorithm of the previous section. What remains are linear sums without trailing $(1,0,1)$ having at 
most $q$ indices different form $(1,0,1).$\\
Assume now that $\S{m_1,\ldots,m_k}n=\S{(a_1,b_1,1),\ldots,(a_k,b_k,1)}n$ is a linear cyclotomic harmonic sum with $q$ indices $m_i$ different from $(1,0,1),$ where $m_k\neq (1,0,1).$
Note that we have
\begin{eqnarray*}
\S{m_1,\ldots,m_k}n&=&\overbrace{
\sum_{i_1=1}^n\frac{1}{a_1 i_1+b_1}\cdots\sum_{i_{k-1}=1}^{i_{k-2}}\frac{1}{a_{k-1} i_{k-1}+b_{k-1}}\sum_{i_{k}=1}^{i_{k-1}}\left(\frac{1}{a_{k} i_{k}+b_{k}}-\frac{1}{a_k i_k}\right)}^{=:A}\\
&&+\frac{1}{a_k}\S{m_1,\ldots,m_{k-1},(1,0,1)}n.
\end{eqnarray*}
We now look for an integral representation of $A.$ For the inner sum we get as before
\begin{eqnarray*}
\sum_{i_{k}=1}^{i_{k-1}}\left(\frac{1}{a_{k} i_{k}+b_{k}}-\frac{1}{a_k i_k}\right)&=&\int_0^1\frac{x^{a_k-1} \left(x^{b_k}-1\right) \left(x^{a_k i_{k-1}}-1\right)}{x^{a_k}-1}dx\\
				&=&\int_0^1x^{a_k-1}\left(x^{a_k i_{k-1}}-1\right)\frac{\Phi_{b_k} \prod_{d|a_k, d < a_k}\Phi_{d}}{\Phi_{a_k}\prod_{d|b_k, d < b_k}\Phi_{d}}dx\\
				&=&\int_0^1x^{a_k-1}\left(x^{a_k i_{k-1}}-1\right)\frac{\Phi_{b_k} \prod_{d|a_k, 1 < d < a_k}\Phi_{d} }{\Phi_{a_k}\prod_{d|b_k, 1 < d < b_k}\Phi_{d}}dx.
\end{eqnarray*}
At this point we may perform a partial fraction decomposition (note that there is no factor of $(1-x)$ in the denominator), split the integrals and proceed to get an integral 
representation of $A$ as described in Section \ref{CSintrep}. Finally we can use repeated integration by parts to get a Mellin type representation of $A.$ The cyclotomic 
harmonic polylogarithms at highest weight will not have trailing $(1,0)$ because of the absence of $(1-x)$ in the denominator and hence we can compute the asymptotic expansions
of the Mellin transform of these cyclotomic harmonic polylogarithms (see Remark \ref{CSTrailing1rem}). If there are cyclotomic harmonic polylogarithms of lower weights with 
trailing $(1,0),$ we can perform the Mellin transform which will lead to sums of lower weight which we can already expand asymptotically.\\
The last step is to compute the expansion of $\S{m_1,\ldots,m_{k-1},(1,0,1)}n.$ If $q=1,$ we have $\S{m_1,\ldots,m_{k-1},(1,0,1)}n=\S{1,\ldots,1}n$ and hence we already know the 
expansion. If $q>1,$ we reduced at least the number of indices different from $(1,0,1)$ by one, and we can now apply the described strategy to $\S{m_1,\ldots,m_{k-1},(1,0,1)}n.$ Summarizing, we
step by step reduce the number of indices different from $(1,0,1).$ Eventually we will arrive at a cyclotomic sum with all indices equal to $(1,0,1).$ This case we already handled.
\begin{example}We consider the sum $\S{(3,1,1),(3,2,1)}n.$ We have
$$
\S{(3,1,1),(3,2,1)}n=\overbrace{\sum_{i=1}^n\frac{1}{3i+1}\sum_{j=1}^i\left(\frac{1}{3j+2}-\frac{1}{3j}\right)}^{=:A}+\frac{1}{3}\S{(3,1,1),(1,0,1)}n.
$$
The integral representation of $A$ yields
\small
\begin{eqnarray*}
&&\frac{1}{6} \left(-2 H_{(3,0)}(1)+2 H_{(3,1)}(1)-1\right) \int_0^1 \frac{x^{3 n}-1}{x-1} \, dx+\frac{1}{3} \int_0^1 \frac{\left(x^{3 n}-1\right) H_{(3,0)}(x)}{x-1} \, dx\\
&&+\int_0^1 x^{3n} H_{(3,0)}(x) \, dx+\frac{2}{3} \left(H_{(3,0)}(1)-H_{(3,1)}(1)-1\right) \int_0^1 \frac{x^{3 n}-1}{x^2+x+1} \, dx\\
&&+\frac{1}{6} \left(2 H_{(3,0)}(1)-2 H_{(3,1)}(1)+1\right) \int_0^1\frac{x \left(x^{3 n}-1\right)}{x^2+x+1} \, dx-\frac{2}{3} \int_0^1 \frac{\left(x^{3 n}-1\right) H_{(3,0)}(x)}{x^2+x+1} \, dx\\
&&-\frac{1}{3} \int_0^1 \frac{x \left(x^{3 n}-1\right)H_{(3,0)}(x)}{x^2+x+1} \, dx+\frac{-2 H_{(3,0)}(1)+2 H_{(3,1)}(1)+2 n+1}{2 (3 n+1)}\\
&&-\int_0^1 x^{3 n+1} \, dx+\frac{1}{3} \int_0^1 \left(x^{3 n}-1\right) \, dx.
\end{eqnarray*}
\normalsize
We can handle all these integrals and the asymptotic expansion of $A$ up to order $3$ is
\small
\begin{eqnarray*}
&&\textnormal{H}_{(3,1)}(1) \left(\frac{1}{3} \textnormal{H}_{(2,0)}(2)+\frac{10}{243 n^3}-\frac{11}{108 n^2}+\frac{5}{18 n}+\frac{1}{3} (\log (n)+\gamma )-\frac{1}{6}\right)-\frac{2}{3}
   \textnormal{H}_{(3,0)}(1){}^2\\&&+\left(\frac{1}{3} \textnormal{H}_{(3,1)}(1)+\frac{2}{3}\right) \textnormal{H}_{(3,0)}(1)+\frac{1}{3} \textnormal{H}_{(3,1)}(1){}^2-\frac{1}{6} \textnormal{H}_{(2,0)}(2)+\frac{1}{3} \textnormal{H}_{(3,0),(1,0)}(1)\\
&&+\frac{2}{3}\textnormal{H}_{(3,0),(3,0)}(1)+\frac{1}{3} \textnormal{H}_{(3,1),(3,0)}(1)-\frac{445}{4374 n^3}+\frac{85}{648 n^2}-\frac{23}{108 n}+\frac{1}{6} (-\log (n)-\gamma ).
\end{eqnarray*}
\normalsize
Now we have to consider $\S{(3,1,1),(1,0,1)}n.$ Removing the trailing $(1,0,1)$ yields
$$
\S{(3,1,1),(1,0,1)}n=\textnormal{S}_{(3,1,1)}(n) \textnormal{S}_{(1,0,1)}(n)+\textnormal{S}_{(1,0,1)}(n)-3 \textnormal{S}_{(3,1,1)}(n)-\textnormal{S}_{(1,0,1),(3,1,1)}(n).
$$
The depth one sums are easy to handle and it remains to expand $\textnormal{S}_{(1,0,1),(3,1,1)}(n).$ We have
$$
\S{(1,0,1),(3,1,1)}n=\overbrace{\sum_{i=1}^n\frac{1}{i}\sum_{j=1}^i\left(\frac{1}{3j+1}-\frac{1}{3j}\right)}^{=:B}+\frac{1}{3}\S{1,1}n.
$$
We can handle $B$ as we have already worked with $A.$ The asymptotic expansion of $B$ up to order $3$ is
\small
\begin{eqnarray*}
&&\textnormal{H}_{(3,0)}(1) \left(\textnormal{H}_{(2,0)}(2)-\frac{1}{12 n^2}+\frac{1}{2 n}+\log (n)+\gamma -2\right)+\textnormal{H}_{(3,1)}(1) \biggl(\textnormal{H}_{(2,0)}(2)-\frac{1}{12 n^2}+\frac{1}{2 n}\biggr.
\\&&\biggl.+\log (n)+\gamma -1\biggr)-\textnormal{H}_{(2,0)}(2)+\textnormal{H}_{(3,0),(1,0)}(1)-\textnormal{H}_{(3,0),(3,0)}(1)-\textnormal{H}_{(3,0),(3,1)}(1)\\
&&+\textnormal{H}_{(3,1),(1,0)}(1)-2 \textnormal{H}_{(3,1),(3,0)}(1)-2 \textnormal{H}_{(3,1),(3,1)}(1)-\frac{101}{1458 n^3}+\frac{19}{108 n^2}-\frac{11}{18 n}\\&&-\log (n)-\gamma +3.
\end{eqnarray*}
\normalsize
Since we know the expansion of $\S{1,1}n$ (see Section \ref{HSExpansion}) we can combine all these expansions and end up at the asymptotic expansion of $\S{(3,1,1),(3,2,1)}n$ up to order $3:$
\small
\begin{eqnarray*}
&&\frac{\frac{29}{162}-\frac{1}{108} \textnormal{H}_{(2,0)}(2)}{n^2}+\textnormal{H}_{(3,0)}(1) \biggl(-\frac{1}{3} \textnormal{H}_{(2,0)}(2)+\frac{1}{3} \textnormal{H}_{(3,1)}(1)+\frac{1}{108 n^2}-\frac{1}{18 n}-\frac{1}{9}(\log
   (n)+\gamma)\biggr.\\
&&\biggl.+\frac{2}{3}\biggr)+(\log (n)+\gamma ) \left(\frac{1}{9} \textnormal{H}_{(2,0)}(2)+\frac{10}{729 n^3}-\frac{11}{324 n^2}+\frac{5}{54 n}-\frac{1}{6}\right)+\textnormal{H}_{(3,1)}(1)
   \biggl(\frac{10}{243 n^3}\biggr.\\
&&\biggl.-\frac{1}{12 n^2}+\frac{1}{6 n}+\frac{1}{9} (\log (n)+\gamma )-\frac{1}{6}\biggr)+\frac{\frac{1}{18} \textnormal{H}_{(2,0)}(2)-\frac{25}{108}}{n}-\frac{2}{3}
   \textnormal{H}_{(3,0)}(1){}^2+\frac{1}{3} \textnormal{H}_{(3,1)}(1){}^2\\
&&-\frac{1}{6} \textnormal{H}_{(2,0)}(2)+\textnormal{H}_{(3,0),(3,0)}(1)+\frac{1}{3} \textnormal{H}_{(3,0),(3,1)}(1)-\frac{1}{3}
   \textnormal{H}_{(3,1),(1,0)}(1)+\textnormal{H}_{(3,1),(3,0)}(1)\\
&&+\frac{2}{3} \textnormal{H}_{(3,1),(3,1)}(1)-\frac{2285}{17496 n^3}+\frac{1}{18} (\log (n)+\gamma )^2-\frac{\zeta_2}{18}.
\end{eqnarray*}
\normalsize
\end{example}

\begin{example}[Example \ref{CSasyexp1} continued]The asymptotic expansion of $\S{(3,1,2),(1,0,1)}n$ up to order $3$ yields
\small
\begin{eqnarray*}
&&-\frac{1}{3} \zeta_2 \textnormal{H}_{(2,0)}(2)+\frac{2}{3} \textnormal{H}_{(2,0)}(2) \textnormal{H}_{(3,0),(0,0)}(1)+\frac{1}{3} \textnormal{H}_{(2,0)}(2) \textnormal{H}_{(3,1),(0,0)}(1)\\
&&+\frac{2}{3}\textnormal{H}_{(3,0),(0,0),(1,0)}(1)+\textnormal{H}_{(3,0),(1,0),(0,0)}(1)-\frac{2}{3} \textnormal{H}_{(3,0),(3,0),(0,0)}(1)-\frac{1}{3} \textnormal{H}_{(3,0),(3,1),(0,0)}(1)\\
&&+\frac{1}{3}\textnormal{H}_{(3,1),(0,0),(1,0)}(1)+\textnormal{H}_{(3,1),(1,0),(0,0)}(1)-\frac{4}{3} \textnormal{H}_{(3,1),(3,0),(0,0)}(1)-\frac{2}{3} \textnormal{H}_{(3,1),(3,1),(0,0)}(1)\\
&&+\frac{47}{972 n^3}-\frac{1}{108 n^2}+\left(-\frac{11}{162n^3}+\frac{5}{54 n^2}-\frac{1}{9 n}\right) (\log (n)+\gamma )-\frac{1}{9 n}+\frac{\zeta_3}{3}.
\end{eqnarray*}
\normalsize
\end{example}
Note that the arising constants in the final result are not yet reduced with respect to the relations given in Section \ref{CSInfRelations}.

\cleardoublepage  

\chapter{Cyclotomic S-Sums}
\label{CSSchapter}

\def\firstcircle{(0,0) circle (3.5cm and 1.5cm)}
\def\secondcircle{(0:4cm) circle (3.5cm and 1.5cm)}
\def\thirdcircle{(0:2cm) circle (6cm and 2.5cm)}

\colorlet{circle area}{green!20}

\tikzset{filled/.style={fill=circle area, draw=circle edge, thick},
    outline/.style={draw=circle edge, thick}}

\colorlet{circle edge}{black!100}
\colorlet{circle area}{black!20}
\setlength{\parskip}{5mm}
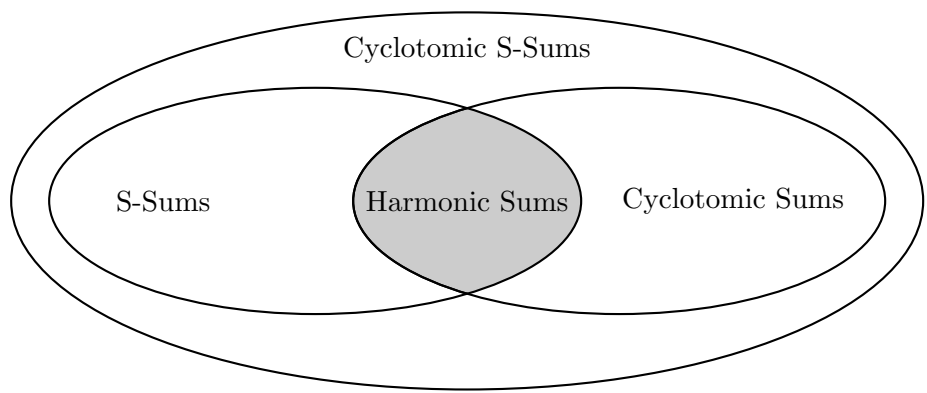
\begin{figure}
\centering
\begin{tikzpicture}
     \begin{scope}
         \clip \firstcircle;
         \fill[filled] \secondcircle;
     \end{scope}
     \draw[outline] \firstcircle;
     \draw[outline] \secondcircle;
     \draw[outline] \thirdcircle;
     \draw (2,0) node {Harmonic Sums};
     \draw (-2,0) node {S-Sums};
     \draw (5.5,0) node {Cyclotomic Sums};
     \draw (2,2) node {Cyclotomic S-Sums};
\end{tikzpicture}
\caption{\label{connectionfigure}Relations between the different extensions of harmonic sums.}
\end{figure}

In this chapter we will extend the definition of harmonic sums once more. We will consider {\itshape cyclotomic S-sums}, which will unify the two already treated extensions, 
\ie S-sums and cyclotomic harmonic sums are subsets of the {\itshape cyclotomic S-sums} as indicated in Figure \ref{connectionfigure}.

\section{Definition and Structure of Cyclotomic S-Sums}
\label{CSSdef}
\begin{definition}[Cyclotomic S-Sums]
Let $a_i,c_i,n,k \in \N,b_i\in\N_0,$ with $a_i>b_i$ and $x_i \in \R^*$ for $i\in \{1,2,\ldots,k\}.$ We define
\begin{eqnarray*}
&&\S{(a_1,b_1,c_1),(a_2,b_2,c_2),\ldots,(a_k,b_k,c_k)}{x_1,x_2,\ldots,x_k;n}=\\
&&\hspace{2cm}=\sum_{i_1 \geq i_2,\cdots i_k \geq 1}\frac{x_1^{i_1}}{(a_1 i_1+b_1)^{c_1}}\frac{x_2^{i_2}}{(a_2 i_2+b_2)^{c_2}}\cdots\frac{x_k^{i_1}}{(a_k i_k+b_k)^{c_k}}\\
&&\hspace{2cm}=\sum_{i_1=1}^n\frac{x_1^{i_1}}{(a_1 i_1+b_1)^{c_1}}\sum_{i_2=1}^{i_1}\frac{x_2^{i_2}}{(a_2 i_2+b_2)^{c_2}}\cdots\sum_{i_k=1}^{i_{k-1}}\frac{x_k^{i_1}}{(a_k i_k+b_k)^{c_k}}.
\end{eqnarray*}
$k$ is called the depth and $w=\sum_{i=0}^kc_i$ is called the weight of the cyclotomic S-sum $\S{(a_1,b_1,c_1),(a_2,b_2,c_2),\ldots,(a_k,b_k,c_k)}{x_1,x_2,\ldots,x_k;n}$.
\end{definition}

For later use we define the following set
\begin{eqnarray}
\mathcal{CS}(n)=\left\{q(s_1,\ldots,s_r)\left|r\in \N; \right. s_i \textnormal{ a cyclotomic S-sum at }n ;\ q\in \R[x_1,\ldots,x_r]\right\};
\label{CS}\nonumber
\end{eqnarray}
note that $\mathcal{CS}(n)\supseteq\mathcal{C}(n)\supseteq\mathcal{S}(n)$ as introduced on pages \pageref{C} and \pageref{S}.

\subsection{Product}
The product of harmonic sums can be generalized to {\itshape cyclotomic S-sums}. It turns out that they form again a quasi-shuffle algebra.

\begin{thm}
Let $a_i,c_i,d_i,f_i,k,l \in \N,b_i,e_i,n\in \N_0$ and $x_i,y_i \in \R^*.$
If $a_1 e_1\neq d_1 b_1,$ we have
\begin{eqnarray*}
&&\hspace{-1cm}\S{(a_1,b_1,c_1),\ldots,(a_k,b_k,c_k)}{x_1,x_2,\ldots,x_k;n}\S{(d_1,e_1,f_1),\ldots,(d_l,e_l,f_l)}{y_1,y_2,\ldots,y_l;n}=\\
&&\sum_{i=1}^n\frac{x_1^i\S{(a_2,b_2,c_2),\ldots,(a_k,b_k,c_k)}{x_2,\ldots,x_k;i}\S{(d_1,e_1,f_1),\ldots,(d_l,e_l,f_l)}{y_1,\ldots,y_l;i}}{(a_1 i+b_1)^{c_1}}\\
&&+\sum_{i=1}^n\frac{y_1^i\S{(a_1,b_1,c_1),\ldots,(a_k,b_k,c_k)}{x_1,\ldots,x_k;i}\S{(d_2,e_2,f_2),\ldots,(d_l,e_l,f_l)}{y_2,\ldots,y_l;i}}{(d_1 i+e_1)^{f_1}}\\
&&-\sum_{i=1}^n \left( (-1)^{c_1}\sum_{j_1=1}^{c_1}  (-1)^j \binom{c_1+f_1-j-1}{f_1-1}\frac{a_1^{f_1}d_1^{c_1-j}}{a_1 e_1-d_1 b_1}\frac{1}{(a_1 i+b_1)^j}\right.\\
&&\hspace{1cm}\left.+(-1)^{f_1}\sum_{j=1}^{f_1}{ (-1)^j \binom{c_1+f_1-j-1}{f_1-1}\frac{a_1^{f_1-j}d_1^{c_1}}{d_1 b_1-a_1 e_1}\frac{1}{(d_1 i+e_1)^j}}\right)\\
&&\hspace{1cm}\S{(a_2,b_2,c_2),\ldots,(a_k,b_k,c_k)}{x_1,\ldots,x_k;i}\S{(d_2,e_2,f_2),\ldots,(d_l,e_l,f_l)}{y_2,\ldots,y_l;i},
\end{eqnarray*}
and if $a_1 e_1= d_1 b_1$ we have
\begin{eqnarray*}
&&\hspace{-1cm}\S{(a_1,b_1,c_1),\ldots,(a_k,b_k,c_k)}{x_1,x_2,\ldots,x_k;n}\S{(d_1,e_1,f_1),\ldots,(d_l,e_l,f_l)}{y_1,y_2,\ldots,y_l;n}=\\
&&\sum_{i=1}^n\frac{x_1^i\S{(a_2,b_2,c_2),\ldots,(a_k,b_k,c_k)}{x_2,\ldots,x_k;i}\S{(d_1,e_1,f_1),\ldots,(d_l,e_l,f_l)}{y_1,\ldots,y_l;i}}{(a_1 i+b_1)^{c_1}}\\
&&+\sum_{i=1}^n\frac{y_1^i\S{(a_1,b_1,c_1),\ldots,(a_k,b_k,c_k)}{x_1,\ldots,x_k;i}\S{(d_2,e_2,f_2),\ldots,(d_l,e_l,f_l)}{y_2,\ldots,y_l;i}}{(d_1 i+e_1)^{f_1}}\\
&&-\frac{a_1^{f_1}}{d_1^{f_1}}\sum_{i=1}^n\frac{x_1^i y_1^i\S{(a_2,b_2,c_2),\ldots,(a_k,b_k,c_k)}{x_1,\ldots,x_k;i}\S{(d_2,e_2,f_2),\ldots,(d_l,e_l,f_l)}{y_2,\ldots,y_l;i}}{(a_1 i+b_1)^{c_1+f_1}}.
\end{eqnarray*}
\end{thm}
\begin{proof}
From the identity $$\sum_{i=1}^n\sum_{j=1}^n a_{i j}=\sum_{i=1}^n\sum_{j=1}^i a_{i j}+\sum_{j=1}^n\sum_{i=1}^j a_{i j}-\sum_{i=1}^na_{i i}$$ we get immediately
\begin{eqnarray*}
&&\hspace{-1cm}\S{(a_1,b_1,c_1),\ldots,(a_k,b_k,c_k)}{x_1,x_2,\ldots,x_k;n}\S{(d_1,e_1,f_1),\ldots,(d_l,e_l,f_l)}{y_1,y_2,\ldots,y_l;n}=\\
&&\sum_{i=1}^n\frac{x_1^i\S{(a_2,b_2,c_2),\ldots,(a_k,b_k,c_k)}{x_2,\ldots,x_k;i}\S{(d_1,e_1,f_1),\ldots,(d_l,e_l,f_l)}{y_1,\ldots,y_l;i}}{(a_1 i+b_1)^{c_1}}\\
&&+\sum_{i=1}^n\frac{y_1^i\S{(a_1,b_1,c_1),\ldots,(a_k,b_k,c_k)}{x_1,\ldots,x_k;i}\S{(d_2,e_2,f_2),\ldots,(d_l,e_l,f_l)}{y_2,\ldots,y_l;i}}{(d_1 i+e_1)^{f_1}}\\
&&-\sum_{i=1}^n\frac{(x_1 y_1)^i\S{(a_2,b_2,c_2),\ldots,(a_k,b_k,c_k)}{x_1,\ldots,x_k;i}\S{(d_2,e_2,f_2),\ldots,(d_l,e_l,f_l)}{y_2,\ldots,y_l;i}}{(a_1 i+b_1)^{c_1}(d_1 i+e_1)^{f_1}}.
\end{eqnarray*}
Using Lemma \ref{CSapart} we get the results.
\end{proof}

\subsection{Synchronization}
In this subsection we consider cyclotomic S-sums with upper summation limit $n+c,$ $kn$ and $kn+c$ for $c\in\Z$ and $k\in \N$.
\begin{lemma}
Let $a_i,c_i,n,k,c \in \N, b_i\in\N_0$ and $x_i \in \R^*$ for $i\in (1,2,\ldots,k).$ Then for $n\geq 0$
\begin{eqnarray*}
\S{(a_1,b_1,c_1),\ldots,(a_k,b_k,c_k)}{x_1,x_2,\ldots,x_k;n+c}&=&\S{(a_1,b_1,c_1),\ldots,(a_k,b_k,c_k)}{x_1,\ldots,x_k;n}\\
	&&\hspace{-4cm}+\sum_{j=1}^c{\frac{{x_1}^{j+n}\S{(a_2,b_2,c_2),\ldots,(a_k,b_k,c_k)}{x_2,\ldots,x_k;n+j}}{(a_1 (j+n)+b_1)^{c_1}}},
\end{eqnarray*}
and $n\geq c,$
\begin{eqnarray*}
\S{(a_1,b_1,c_1),\ldots,(a_k,b_k,c_k)}{x_1,x_2,\ldots,x_k;n-c}&=&\S{(a_1,b_1,c_1),\ldots,(a_k,b_k,c_k)}{x_1,\ldots,x_k;n}\\
	&&\hspace{-4cm}+\sum_{j=1}^c{\frac{{x_1}^{j+n-c}\S{(a_2,b_2,c_2),\ldots,(a_k,b_k,c_k)}{x_2,\ldots,x_k;n-c+j}}{(a_1(j+n-c)+b_1)^{c_1}}}.	
\end{eqnarray*}
\end{lemma}
Given a cyclotomic S-sum of the form $\S{(a_1,b_1,c_1),\ldots,(a_k,b_k,c_k)}{x_1,\ldots,x_k;n+c}$ with $c~\in~\Z,$ we can apply the previous lemma recursively in order to synchronize 
the upper summation limit of the arising cyclotomic S-sums to $n$.
\begin{lemma}
For $a, c, k \in \N,b\in\N_0$, $x\in \R^*$, $k\geq 2:$
\begin{eqnarray*}
\S{(a,b, c)}{x; k\cdot n}=\sum_{i=0}^{k-1}\frac{1}{x^i}\S{(k\cdot a,b-a \cdot i, c)}{x^k;n}.
\end{eqnarray*}
\label{CSSmultint1}
\end{lemma}

\begin{proof}
\begin{eqnarray*} 
\S{(a,b,c)}{x;k\cdot n}&=&\sum_{j=1}^{k\cdot n}\frac{x^j}{(a j +b)^{c}}=\\
		&&\hspace{-2.5cm}=\sum_{j=1}^{n}\frac{x^{kj}}{(a (k \cdot j) +b)^{c}}+\sum_{j=1}^{n}\frac{x^{kj-1}}{(a (k \cdot j - 1) +b)^{c}}+\cdots+\sum_{j=1}^{n}\frac{x^{kj-(k-1)}}{(a (k \cdot j-(k-1)) +b)^{c}}\\
		&&\hspace{-2.5cm}=\sum_{j=1}^{n}\frac{x^{kj}}{((a k) j +b)^{c}}+\sum_{j=1}^{n}\frac{x^{kj}x^{-1}}{((a k) j + (b-a))^{c}}+\cdots+\sum_{j=1}^{n}\frac{x^{kj}x^{-(k-1)}}{((a k) j + (b-(k-1)a))^{c}}\\
		&&\hspace{-2.5cm}=\sum_{i=0}^{k-1}\frac{1}{x^i}\S{(k\cdot a,b-a \cdot i, c)}{x^k;n}.
\end{eqnarray*}
\end{proof}

\begin{thm}
For $a_i, c_i, m, k \in \N,b_i,n\in\N_0$, $x\in \R^*$, $k\geq 2$  :
\begin{eqnarray*}
&&\S{(a_m,b_m,c_m),(a_{m-1},b_{m-1}, c_{m-1}),\ldots,(a_1,b_1,c_1)}{x_m,\ldots,x_1; k \cdot n}=\\
&&\hspace{3cm}\sum_{i=0}^{m-1}\sum_{j=1}^{n} \frac{\S{(a_{m-1},b_{m-1},c_{m-1}),\ldots,(a_1,b_1,c_1)}{x_{m-1}\ldots,x_{1}; k \cdot j-i}x_m^{k\cdot j -i}} {(a_m (k\cdot j-i)+b_1)^{c_1}}.
\end{eqnarray*}
\label{CSSmultint}
\end{thm}

After applying Theorem \ref{CSSmultint} we can synchronize the cyclotomic S-sums in the inner sum with upper summation limit $k \cdot j-i$ to the upper summation limit $k \cdot j.$ Now we 
can apply Theorem \ref{CSSmultint} to these sums. Repeated application of this procedure leads to cyclotomic S-sums with upper summation limit $n.$

\begin{remark}
 Like for cyclotomic harmonic sums the synchronization of the upper summation limit leads again to multiple argument relations and as a special case to duplication 
relations, and since cyclotomic S-sums form a quasi shuffle algebra, there are again algebraic relations.
\label{CSSrelrem}
\end{remark}

\section{Integral Representation of Cyclotomic Harmonic S-Sums}
\begin{lemma}
Let $a,c\in\N,b\in\N_0,$ $d\in\R^*$ and $n\in\N;$ then
\begin{eqnarray*}
\S{(a,b,1)}{d;n}&=&\int_0^{1}{\frac{{x_1}^{a+b-1}\left(d^n{x_1}^{a n}-1\right)}{{x_1}^a-\frac{1}{d}}dx_1}\\
\S{(a,b,2)}{d;n}&=&\int_0^{1}{\frac{1}{x_2}\int_0^{x_2}{\frac{{x_1}^{a+b-1}\left( d^n {x_1}^{a n}-1\right)}{{x_1}^a-\frac{1}{d}}dx_1}dx_2}\\
\S{(a,b,c)}{d;n}&=&\int_0^1{\frac{1}{x_c}\int_0^{x_c}{\frac{1}{x_{c-1}} \cdots \int_0^{x_3}\frac{1}{x_2}{\int_0^{x_2}{\frac{{x_1}^{a+b-1}\left(d^n {x_1}^{a n}-1\right)}{{x_1}^a-\frac{1}{d}}dx_1}dx_2}\cdots}dx_c}.
\end{eqnarray*}
\label{CSSintrep1}
\end{lemma}

Let us now look at the integral representation of cyclotomic S-sums of higher depths. We consider the sum 
$\S{(a_1,b_1,c_1),(a_2,b_2,c_2),\ldots,(a_k,b_k,c_k)}{x_1,x_2,\ldots,x_k;n}$ and apply Lemma~\ref{CSSintrep1} to the innermost sum $(a = a_k, b = b_k, c = c_k, d=x_k)$.
One now may perform the next sum in the
same way, provided $a_{k-1} | a_k$. At this point we may need the fact:
\begin{eqnarray*}
\sum_{i=1}^n\frac{(d y^a)^i}{(ai+b)^c}=\frac{1}{y^b}\int_0^y{\frac{1}{x_c}\int_0^{x_c}{\frac{1}{x_{c-1}} \cdots \int_0^{x_3}\frac{1}{x_2}{
  \int_0^{x_2}{\frac{{x_1}^{a+b-1}\left(d^n {x_1}^{a n}-1\right)}{{x_1}^a-\frac{1}{d}}dx_1}}\cdots }dx_c}
\end{eqnarray*}
for $n,a,b,c,k\in \N$ and $d,y\in\R.$ If $a_{k-1} \nmid a_k,$ one transforms the integration variables such that the next denominator can be generated, etc.
In this way, the sum $ \S{(a_1,b_1,c_1),(a_2,b_2,c_2),\ldots,(a_k,b_k,c_k)}{x_1,x_2,\ldots,x_k;n} $
can be represented in terms of linear combinations of iterated integrals.
 
Let us illustrate the principle steps in case of the following example~:
%----------------------------------------------------------------------------------------------
\begin{eqnarray*}
\S{(3,2,2),(2,1,1)}{\frac{1}{3},2;n} = \sum_{k=1}^{n} \frac{\frac{1}{3}^k}{(3k+2)^2}\sum_{l=1}^{k} \frac{2^l}{(2l+1)}~.
\end{eqnarray*}
%----------------------------------------------------------------------------------------------
The first sum yields
%----------------------------------------------------------------------------------------------
\begin{eqnarray*}
\S{(3,2,2),(2,1,1)}{\frac{1}{3},2;n} =  \sum_{k=1}^{n} \int_0^1 \frac{\frac{1}{3^k}}{(3k + 2)^2}\frac{x^2((2 x^2)^k - 1)}{x^2-\frac{1}{2}} dx ~.
\end{eqnarray*}
%----------------------------------------------------------------------------------------------
Setting $x = y^3$ one obtains
%----------------------------------------------------------------------------------------------
\begin{eqnarray}
\label{CSSeq:ex1}
\S{(3,2,2),(2,1,1)}{\frac{1}{3},2;n} &=& 12 \int_0^1\frac{y^8}{y^6-\frac{1}{2}}
\sum_{k=1}^n \frac{\left(\frac{2y^6}{3}\right)^k-\left(\frac{1}{3}\right)^k}{(6k+4)^2}dy
\nonumber\\
&=& 12\int_0^1\frac{y^4}{y^6-\frac{1}{2}} \Biggl\{
\int_0^y
\frac{1}{z} \int_0^z t^9~\frac{\left(\frac{2t^6}{3}\right)^{n}-1}{t^6-\frac{3}{2}}dtdz
\nonumber\\ && \hspace*{2.9cm}
- y^4 \int_0^1
\frac{1}{z} \int_0^z t^9~\frac{\left(\frac{t^6}{3}\right)^{n}-1}{t^6-3}dtdz
\Biggr\}dy. \nonumber
\end{eqnarray}
Due to the pole at $y^6=\frac{1}{2}$ we cannot split the integral as we did in (\ref{CSeq:ex1}).

\section{Cyclotomic S-Sums at Infinity}

As in the case of S-sums we want to give conditions for convergence of cyclotomic harmonic S-sums. Of course, not all cyclotomic S-sums are finite at infinity, since for example $\lim_{n\rightarrow \infty} \S{(2,1,1)}{2;n}$ does not exist.
In fact, we have the following theorem, compare Theorem \ref{SSconsumthm} and Lemma \ref{HSconsumlem}:
\begin{thm}
Let $a_1, \ldots a_k ,c_1, \ldots c_k \in \N, b_1, \ldots b_k\in\N_0$ and $x_1, x_2, \ldots x_k \in \R\setminus\{0\}$ for $k \in \N.$
The cyclotomic S-sum $\S{(a_1,b_1,c_1),(a_2,b_2,c_2),\ldots,(a_k,b_k,c_k)}{x_1,x_2,\ldots,x_k;n}$ is absolutely convergent, when $n\rightarrow \infty$, if and only if one of the following conditions holds:
\begin{itemize}
 \item [1.] $\abs{x_1}<1 \wedge \abs{x_1 x_2}\leq 1 \wedge \ldots \wedge \abs{x_1 x_2 \cdots x_k}\leq 1,$
 \item [2.] $c_1>1 \wedge \abs{x_1}=1 \wedge \abs{x_2}\leq 1 \wedge \ldots \wedge \abs{x_2 \cdots x_k}\leq 1.$
\end{itemize}
In addition the cyclotomic S-sum is conditional convergent (convergent but not absolutely convergent) if and only if
\begin{itemize}
 \item [3.] $c_1=1 \wedge x_1=-1 \wedge \abs{x_2}\leq 1 \wedge \ldots \wedge \abs{ x_2 \cdots x_k}\leq 1.$
\end{itemize}
\label{CSSconsumthm}
\end{thm} 
The proof of this theorem is in principle analogue to the proof of Theorem~\ref{SSconsumthm}.

\begin{remark}
The relations for cyclotomic harmonic sums at infinity carry over to cyclotomic S-sums: we have algebraic relations, duplication and multiple argument relations (see Remark \ref{CSSrelrem}), and even 
the shuffle algebra relations carry over using Theorem \ref{CSSsumtheo2}.
\label{CSSinfrelrem}
\end{remark}

\section{Summation of Cyclotomic S-Sums}
\label{Summation of Cyclotomic S-Sums}
With the summation techniques presented in \cite{Hasselhuhn2012}, definite multi-sums (in particular certain Feynman integrals transformed to such sums \cite{Bluemlein2011}) can be 
simplified to indefinite nested sum expressions; similarly, such Feynman integrals can be simplified to such formats by using the techniques presented in Chapter \ref{AZchapter}. Then 
one ot the crucial steps for further processing is the tranformation of these sums to harmonic sums, S-sums, cyclotomic sums or most generally to cyclotomic S-sums.\\
 This section follows the ideas of Chapter 5 of \cite{Ablinger2009}, however it is more general since it deals with cyclotomic S-sums instead of harmonic sums.
 With the package \ttfamily Sigma \rmfamily \cite{Schneider2007} we can simplify nested sums such that the nested depth is optimal and the degree of the denominator is minimal. This is possible due to 
a refined summation theory \cite{Schneider2007a,Schneider2008a} of $\Pi \Sigma$-fields \cite{Karr1981}. In the following we want to find sum representations of such simplified nested sums in terms of 
cyclotomic S-sums as much as it is possible. 
Inspired by \cite{Savio} for harmonic numbers and \cite{Moch2002,Vermaseren1998} we consider sums of the form
\begin{equation}
	\sum_{i=1}^n{s\,r(i)}
	\label{sumsum0}
\end{equation}
where $s\in \mathcal{CS}(i)$ or $s\in \mathcal{CS}(i)[x_1^i,\ldots,x_k^i]$ with $x_j\in \R$ and $r(i)$ is a rational function in $i.$ We can use the quasi-shuffle algebra property of the
 cyclotomic S-sums to split such sums into sums of the form
\begin{equation}
	\sum_{i=1}^n{x^ir(i)\S{\ve a}{\ve b; i}}
	\label{sumsum2}
\end{equation}
where $r(i)$ is a rational function in $i$ and $x\in\R.$\\

In this chapter we will present formulas to rewrite sums of the form (\ref{sumsum2}) and we will show how we can use these formulas in combination with \ttfamily Sigma\rmfamily. In the following we consider
 the problem:\\
\bfseries Given: \normalfont a sum $\sigma$ of the form (\ref{sumsum0}).\\
\bfseries Find: \normalfont as much as possible a representation of $\sigma$ in terms of cyclotomic S-sums.

\subsection{Polynomials in the Summand}
\label{polyinthesum}
First we consider the case that $r(i)$ is a polynomial in $\Z[i],$ \ie $r(i)=p_mi^m+\ldots+p_1i+p_0$ with $p_k \in \Z.$ If we are able to work out the sum for any power of $i$ times a cyclotomic 
S-sum (\ie $\sum_{i=1}^n x^i i^m \S{\ve a}{\ve b;i},$ $m \in \N$), we can work out the sum for each polynomial $r(i)$. 

\begin{thm}
\label{sumtheo3}
Let $k, \ n \in \N$ and $a_i,c_i \in \N,b_i\in\N_0$ with $a_i>b_i$ and $x,x_i\in \R^*.$ Then
\begin{eqnarray*}
&& \hspace{-2cm}\sum_{i=1}^n{x^ii^k\S{(a_1,b_1,c_1),(a_2,b_2,c_2),\ldots}{x_1,x_2,\ldots;i}}=\\
 && \hspace{2cm}\S{(a_1,b_1,c_1),(a_2,b_2,c_2),\ldots}{x_1,x_2,\ldots;n}\sum_{i=1}^n{x^i i^k}\\
 && \hspace{2cm} -\sum_{j=1}^n{\frac{x_1^j\S{(a_2,b_2,c_2),\ldots}{x_2,\ldots;i}}{(a_1j+b_1)^{c_1}}\sum_{i=1}^{j-1}{x^i i^k}}.
\end{eqnarray*}
\end{thm}

\begin{proof}
\begin{eqnarray}
&& \hspace{-2cm}\sum_{i=1}^n{x^ii^k\S{(a_1,b_1,c_1),(a_2,b_2,c_2),\ldots}{x_1,x_2,\ldots;i}}=\\
			&&\hspace{2cm}=\sum_{i=1}^n{x^ii^k \sum_{j=1}^i{\frac{x_1^j\S{(a_2,b_2,c_2),\ldots}{x_2,\ldots;j}}{(a_1j+b_1)^{c_1}}}}\nonumber\\
			&&\hspace{2cm}=\sum_{j=1}^n{\frac{x_1^j\S{(a_2,b_2,c_2),\ldots}{x_2,\ldots;j}}{(a_1j+b_1)^{c_1}}\sum_{i=j}^n{x^ii^k}}\nonumber\\
			&&\hspace{2cm}=\sum_{j=1}^n{\frac{x_1^j\S{(a_2,b_2,c_2),\ldots}{x_2,\ldots;j}}{(a_1j+b_1)^{c_1}}\left(\sum_{i=1}^n{x^ii^k}-\sum_{i=1}^{j-1}{x^ii^k}\right)}\nonumber\\
			&&\hspace{2cm}=\S{(a_1,b_1,c_1),(a_2,b_2,c_2),\ldots}{x_1,x_2,\ldots;n}\sum_{i=1}^n{x^ii^k}\nonumber\\
			&&\hspace{2.5cm}-\sum_{j=1}^n{\frac{x_1^j\S{(a_2,b_2,c_2),\ldots}{x_2,\ldots;j}} {(a_1j+b_1)^{c_1}}\sum_{i=1}^{j-1}{x^ii^k}}.\nonumber		  
\end{eqnarray}
\end{proof}
The sums $\sum_{i=1}^{j-1}{x^i i^k}$ can be expressed as a polynomial in $j$, together with a factor $(x)^j.$ Hence we can we can apply partial fraction decomposition on
$$
\frac{x_1^j\S{(a_2,b_2,c_2),\ldots}{x_2,\ldots;i}}{(a_1j+b_1)^{c_1}}\sum_{i=1}^{j-1}{x^i i^k}
$$
after applying Theorem \ref{sumtheo3}. The result can be rewritten in cyclotomic S-sums easily. Summarizing, we can always work out the 
sums $\sum_{i=1}^n{x^i r(i)\S{\ve a}{\ve b;i}}$ where $r(i)$ is a polynomial in $i$ and $x\in \R$; the result will be a combination of cyclotomic S-sums 
in the upper summation index $n$, rational functions in $n$ and the factor $x^n.$

\subsection{Rational Functions in the Summand}
If we want to work out sums of the form (\ref{sumsum2}) for a general rational function $r(n)=\frac{p(n)}{q(n)},$ where $p(n)$ and $q(n)$ are polynomials, we can 
proceed as follows:\\

If the degree of $p$ is greater than the degree of $q,$ we compute by polynomial division polynomials $\overline{r}(n)$ and $\overline{p}(n)$ such that $r(n)=\overline{r}(n)+\frac{\overline{p}(n)}{q(n)}$ and the degree of $\overline{p}(n)$ is smaller than the degree of $q.$
We split the sum into two parts, $\ie$ into 
$$\sum_{i=1}^n{x^ir(i)\S{\ve a}{\ve b; i}}=\sum_{i=1}^n{x^i\overline{r}(i)\S{\ve a}{\ve b; i}}+\sum_{i=1}^n{x^i\frac{\overline{p}(i)}{q(i)}\S{\ve a}{\ve b; i}}.$$
The first sum can be done using Theorems \ref{sumtheo3}. For the second sum we proceed as follows:
\begin{enumerate}				
	\item Factorize the denominator $q(i)$ over $\Q$.
	\item Let $A$ be the product of all factors of the form $(a i+b)^c$ with $a,c\in \N$ and $b \in \Z,$ and let $B$ be the product of all remaining factors such that $A(i)B(i)=q(i).$
	\item For example with the extended Euclidean algorithm we compute polynomials $s(i)$ and $t(i)$ such that $s(i)A(i)+t(i)B(i)=\overline{p}(i)$. This is always possible since $A$ and $B$ are relatively prime. Hence we get 
		$$\frac{\overline{p}}{q}=\frac{t}{A}+\frac{s}{B}.$$
	\item Split the sum into two sums, each sum over one fraction, $\ie$
		$$\sum_{i=1}^n{x^i\frac{\overline{p}(i)}{q(i)}\S{\ve a}{\ve b; i}}=\sum_{i=1}^n{x^i\frac{t}{A}\S{\ve a}{\ve b; i}}+\sum_{i=1}^n{x^i\frac{s}{B}\S{\ve a}{\ve b; i}}.$$											
	\item The sum with the denominator $B$ remains untouched (for details see Remark \ref{denomB}). We can now do a complete partial fraction decomposition to the first summand and split the sum such that we sum over each fraction separately.
	\item Each of these new sums can be expressed in terms of cyclotomic S-sums following Subsection \ref{polyinthesum}.
	\item We end up in a combination of rational functions in $n$, cyclotomic S-sums with upper index $n$, the factor $x^n$ and perhaps a sum over the fraction with denominator $B.$
\end{enumerate}

\begin{remark}
In our implementation the sum $\sum_{i=1}^n{\frac{s}{B}\S{\ve a}{i}}$ is passed further to Schneider's \ttfamily Sigma \rmfamily package. The underlying difference field 
and difference ring algorithms \cite{Schneider2007a,Schneider2008,Schneider2008a} can simplify those sums further to sum expressions where the denominator has minimal 
degree. The result of \ttfamily Sigma \rmfamily is again passed to the package \ttfamily HarmonicSums\rmfamily. If possible, it finds a closed form in terms of cyclotomic S-sums.
\label{denomB}
\end{remark}

\subsection{Two Cyclotomic S-sums in the Summand}

\begin{lemma}
Let $a_1,a_2,c_1,c_2,i \in \N,b_1,b_2\in \N_0.$ If $a_1 b_2\neq a_2 b_1,$ we have 
\begin{eqnarray*}
&&\frac{1}{(a_1 (i_1-i_2) +b_1)^{c_1}(a_2 i_2 +b_2)^{c_2}}=\\
&&\hspace{1cm}\sum_{j=1}^{c_1}{  \binom{c_1+c_2-j-1}{c_2-1}\frac{a_1^{c_2}a_2^{c_1-j}}{(a_1 b_2+a_2 b_1+i_1 a_1 a_2)^{c_1+c_2-j}}\frac{1}{(a_1 (i_1-i_2)+b_1)^j}}\\
&&\hspace{1cm}+\sum_{j=1}^{c_2}{  \binom{c_1+c_2-j-1}{c_1-1}\frac{a_1^{c_2-j}a_2^{c_1}}{(a_1 b_2+a_2 b_1+i_1 a_1 a_2)^{c_1+c_2-j}}\frac{1}{(a_2 i_2+b_2)^j}},
\end{eqnarray*}
and if $a_1 b_2= a_2 b_1$ we have
\begin{eqnarray*}
\frac{1}{(a_1 i +b_1)^{c_1}(a_2 i +b_2)^{c_2}}&=&\left(\frac{a_1}{a_2}\right)^{c_2}\frac{1}{(a_1 i+b_1)^{c_1+c_2}}.
\end{eqnarray*}
\label{CSSapart2}
\end{lemma}

\begin{thm}
\label{CSSsumtheo2}
Let $a_i,c_i,d_i,f_i,k,l, n \in \N,b_i,e_i\N_0,a_i>b_i,d_i>e_i,$ and $x_i, y_i \in \R^*.$ Then
\begin{eqnarray*}
&&\sum_{i=1}^n\frac{\S{(a_1,b_1,c_1),\ldots,(a_k,b_k,c_k)}{x_1,\ldots,x_k;n-i}\S{(d_2,e_2,f_2),\ldots,(d_l,e_l,f_l)}{y_2,\ldots,y_l;i} y_1^i} {(d_1i+e_1)^{f_1}}=\\
&&\hspace{0.5cm} \sum_{k=1}^{f_1} \binom{c_1+f_1-k-1}{c_1-1} a_1^{f_1-k} d_1^{c_1}\sum_{i=1}^n\frac{x_1^i}{(d_1 b_1 + a_1 e_1 + i a_1 d_1)^{c_1 f_1 -k}} \\
&&\hspace{0.5cm} \hspace{0.5cm} \sum_{j=1}^{i-1}\frac{\S{(a_2,b_2,c_2),\ldots,(a_k,b_k,c_k)}{x_2,\ldots,x_k;i-j}\S{(d_2,e_2,f_2),\ldots,(d_l,e_l,f_l)}{y_2,\ldots,y_l;j} \frac{y_1^j}{x_1^j}} {(d_1 j +e_1)^{k}}\\
&&\hspace{0.5cm} +\sum_{k=1}^{c_1} \binom{c_1+f_1-k-1}{c_2-1} a_1^{f_1} d_1^{c_1-k}\sum_{i=1}^n\frac{y_1^i}{(d_1 b_1 + a_1 e_1 + i a_1 d_1)^{c_1 f_1 -k}} \\
&&\hspace{0.5cm} \hspace{0.5cm} \sum_{j=1}^{i-1}\frac{\S{(d_2,e_2,f_2),\ldots,(d_l,e_l,f_l)}{y_2,\ldots,y_l;i-j}\S{(a_2,b_2,c_2),\ldots,(a_k,b_k,c_k)}{x_2,\ldots,x_k;j} \frac{x_1^j}{y_1^j}} {(a_1 j +b_1)^{k}}.
\end{eqnarray*}
\end{thm}

\begin{proof}
 We use the following abbreviations: 
\begin{eqnarray*}
S_1(n)&=&\S{(a_1,b_1,c_1),\ldots,(a_k,b_k,c_k)}{x_1,x_2,\ldots,x_k;n}\\
\bar{S}_1(n)&=&\S{(a_2,b_2,c_2),\ldots,(a_k,b_k,c_k)}{x_2,x_3,\ldots,x_k;n} \\
S_2(n)&=&\S{(d_2,e_2,f_2),\ldots,(d_l,e_l,f_l)}{y_2,y_3,\ldots,y_l;n}.
\end{eqnarray*}
We have
\begin{eqnarray*}
\sum_{i=1}^n\frac{S_1(n-i)S_2(i) y_1^i} {(d_1i+e_1)^{f_1}}&=&\sum_{i=1}^n\sum_{j=1}^{n-i}\frac{\bar{S}_1(j)S_2(i) y_1^i}{(a_1 j + b_1)^{c_1} (d_1i+e_1)^{f_1}}\\
	&=&\sum_{i=1}^n\sum_{j=i+1}^{n}\frac{\bar{S}_1(j-i)S_2(i) y_1^i}{(a_1 (j-i) + b_1)^{c_1} (d_1i+e_1)^{f_1}}\\
	&=&\sum_{j=1}^n\sum_{i=1}^{j-1}\frac{\bar{S}_1(j-i)S_2(i) y_1^i}{(a_1 (j-i) + b_1)^{c_1} (d_1i+e_1)^{f_1}}\\
	&=&\sum_{i=1}^n\sum_{j=1}^{i-1}\frac{\bar{S}_1(i-j)S_2(i) y_1^j}{(a_1 (i-j) + b_1)^{c_1} (d_1j+e_1)^{f_1}}.
\end{eqnarray*}
Now we use Lemma \ref{CSSapart2} and complete the proof:
\begin{eqnarray*}
&&\sum_{i=1}^n\sum_{j=1}^{i-1}\frac{\bar{S}_1(i-j)S_2(i) y_1^j}{(a_1 (i-j) + b_1)^{c_1} (d_1j+e_1)^{f_1}}=\\
&&\hspace{1cm}\sum_{k=1}^{f_1} \binom{c_1+f_1-k-1}{c_1-1} a_1^{f_1-k} d_1^{c_1}\sum_{i=1}^n\frac{x_1^i \hspace{1cm}\sum_{j=1}^{i-1}\frac{\bar{S}_1(i-j)S_2(j)\left(\frac{y_1}{x_1}\right)^j} {(d_1 j +e_1)^{k}}}{(d_1 b_1 + a_1 e_1 + i a_1 d_1)^{c_1 f_1 -k}}\\
&&\hspace{1cm}+\sum_{k=1}^{c_1} \binom{c_1+f_1-k-1}{c_2-1} a_1^{f_1} d_1^{c_1-k}\sum_{i=1}^n\frac{y_1^i \sum_{j=1}^{i-1}\frac{S_2(i-j)\bar{S}_1(j) \left(\frac{x_1}{y_1}\right)^j} {(a_1 j +b_1)^{k}}}{(d_1 b_1 + a_1 e_1 + i a_1 d_1)^{c_1 f_1 -k}}.
\end{eqnarray*}
\end{proof}

\section{Reducing the Depth of Cyclotomic S-Sums}
In this section we will define cyclotomic Euler-Sums (compare \cite{Flajolet1998} for the depth 2 case). We will show that we can use cyclotomic Euler-Sums to express cyclotomic S-sums and 
thereby we are able to reduce the nested depth.

\begin{definition}[Cyclotomic Euler-Sums]For $a,c,n,i,k\in \N; b\in\N_0$ and $x\in \R^*$ we define cyclotomic Euler-Sums of depth $d=1$ as
$$
\textnormal{E}_{(a,b,c,x)}(n)=\sum_{i=1}^n\frac{x^i}{(a i+b)^c}=\S{(a,b,c)}{x,n}
$$
and cyclotomic Euler-Sums of depth $d>1$ as
$$
\textnormal{E}_{(a,b,c,x)}(e_1,e_2,\ldots,e_k)(n)=\sum_{i=1}^n\frac{x^i e_1(i)e_2(i)\cdots e_k(i)}{(a i+b)^c}
$$
where $e_1(i),e_2(i),\ldots,e_k(i)$ are cyclotomic Euler-Sums of depth less than $d$ and at least one of the cyclotomic Euler-sums $e_1(i),e_2(i),\ldots,e_k(i)$ has depth $d-1.$
\end{definition}

\begin{example}
\begin{eqnarray*}
  \textnormal{E}_{(2, 1, 4, 2)}\left(\textnormal{E}_{(3, 1, 2, 3)}, \textnormal{E}_{(5, 2, 6, 4)}(\textnormal{E}_{(1, 2, 3, 5)})\right)(n)&=&
\sum _{i=1}^n \frac{2^{i} \sum _{j=1}^{i} \frac{3^{j}}{\left(3 j+1\right){}^2} \sum _{j=1}^{i} \frac{4^{j} \sum _{k=1}^{j} \frac{5^{k}}{\left(k+2\right){}^3}}{\left(5 j+2\right){}^6}}{\left(2 i+1\right){}^4}\\
\textnormal{E}_{(2, 1, 4, 2)}(\textnormal{E}_{(3, 1, 2, 3)}, (\textnormal{E}_{(5, 2, 6, 4)}), n)
	&=&\sum _{i=1}^n \frac{2^{i} \sum _{j=1}^{i} \frac{3^{j} \sum _{k=1}^{j} \frac{4^{k}}{\left(5 k+2\right){}^6}}{\left(3 j+1\right){}^2}}{\left(2 i+1\right){}^4}\\
	&=&\textnormal{S}_{(2,1,4),(3,1,2),(5,2,6)}(2,3,4;n).
\end{eqnarray*}
\end{example}

\begin{lemma}
Consider the functions $I:\N \mapsto \R$ and $T:\N \mapsto \R.$ Let $n,a_1,c_1,a_2,c_2\in\N,b_1,b_2\in\N_0.$ Then we have:
\begin{eqnarray*}
\sum_{i=1}^n{\frac{T(i)\sum_{j=1}^i{\frac{I(j)}{(a_2 j+b_2)^{c_2}}}}{(a_1 i+b_1)^{c_1}}}&=&
	\sum_{i=1}^n{\frac{T(i)I(i)}{(a_1 i+b_1)^{c_1}(a_2 i+b_2)^{c_2}}}-\sum_{i=1}^n{\frac{I(i)\sum_{j=1}^i{\frac{T(j)}{(a_1 j+b_1)^{c_1}}}}{(a_2 i+b_2)^{c_2}}}\\
	&&+\sum_{i=1}^n{\frac{I(i)}{(a_2 i+b_2)^{c_2}}}\sum_{i=1}^n{\frac{T(i)}{(a_1 i+b_1)^{c_1}}}.
\end{eqnarray*}
 \label{lem1}
\end{lemma}
\begin{proof}
This follows immediately from
\begin{eqnarray*}
\sum_{i=1}^n{\frac{I(i)}{(a_2 i+b_2)^{c_2}}}\sum_{i=1}^n{\frac{T(i)}{(a_1 i+b_1)^{c_1}}}&=&\sum_{i=1}^n{\sum_{j=1}^i{\frac{I(i)T(j)}{(a_2 i+b_2)^{c_2}(a_1 j+b_1)^{c_1}}}}\\
	&&+\sum_{j=1}^n{\sum_{i=1}^j{\frac{I(i)T(j)}{(a_2 i+b_2)^{c_2}(a_1 j+b_1)^{c_1}}}}\\&&-\sum_{i=1}^n{\frac{I(i)T(i)}{(a_1 i+b_1)^{c_1}(a_2 i+b_2)^{c_2}}}.
\end{eqnarray*}
\end{proof}

%TODO hier alles auf cyclo ssumen bringen!
\begin{thm}
\label{dephreducethm}
A cyclotomic S-sum of depth $k$ can be expressed as a polynomial in $\Q[s_1,s_2,\ldots,s_n]$ where the $s_i$ are cyclotomic Euler-sums with depth less or equal $\log_2{(k)}+1$.
\end{thm}
\begin{proof}For sake of simplicity and readability we will give the proof for harmonic sums only. The proof for cyclotomic S-sums would require Lemma \ref{CSapart} at some points, but in general the
proof follows analogously.
Consider the harmonic sum $\S{a_k,a_{k-1},\ldots,a_1}{n}.$ We proceed by induction on the depth $k.$\\
$k=1: \ \S{a_1}{n}$ has depth $1\leq\log_2{(1)}+1=1$.\\
$k=2: \ \S{a_2,a_1}{n}$ has depth $2\leq\log_2{(2)}+1=2$.\\
Now let $k>2$ and we assume that the theorem holds for all depths $\leq k$ (we will refer to this induction hypothesis by (IH1)). We want to show that the theorem holds for $k+1:$\\
\begin{eqnarray*}
\S{a_{k+1},a_{k},\ldots,a_1}{n}
	&=&\sum_{i=1}^n{\frac{\S{a_k,a_{k-1},\ldots,a_1}{i}}{i^{a_{k+1}}}}\\
	&\stackrel{\text{Lemma} \ \ref{lem1}}{=}&
		\sum_{i=1}^n{\frac{\S{a_{k-1},\ldots,a_1}{i}}{i^{a_{k+1}+a_k}}}
		+\sum_{i=1}^n{\frac{\S{a_{k-1},\ldots,a_1}{i}}{i^{a_{k}}}}\sum_{i=1}^n{\frac{1}{i^{a_{k+1}}}}\\
		&&-\sum_{i=1}^n{\frac{\S{a_{k-1},\ldots,a_1}{i}\sum_{j=1}^i{\frac{1}{j^{a_{k+1}}}}}{i^{a_k}}}\\
	&=&	\S{a_{k+1}+a_k,a_{k-1},\ldots,a_1}{n}+\S{a_{k},a_{k-1},\ldots,a_1}{n}\S{a_{k+1}}{n}\\
		&&-\sum_{i=1}^n{\frac{\S{a_{k-1},\ldots,a_1}{i}\sum_{j=1}^i{\frac{1}{j^{a_{k+1}}}}}{i^{a_k}}}.
\end{eqnarray*}
$\S{a_{k+1}}{n}$ has depth $1<\log_2{(k+1)}+1.$
The depth of $\S{a_{k+1}+a_k,a_{k-1},\ldots,a_1}{n}$ and $\S{a_{k},a_{k-1},\ldots,a_1}{n}$ is $k$. Using the induction hypothesis (IH1) we can express these sums by polynomials of sums 
with depth $\leq\log_2{(k)}+1$. Since $\log_2{(k)}+1<\log_2{(k+1)}+1$ it remains to show that we can express the sum
$$
\sum_{i=1}^n{\frac{\S{a_{k-1},\ldots,a_1}{i}\sum_{j=1}^i{\frac{1}{j^{a_{k+1}}}}}{i^{a_k}}}
$$
by polynomials of sums with depth $\leq\log_2{(k+1)}+1.$
Let therefore be $m\leq k.$ We will show the more general statement that for $r\geq1,$ $s\geq1,$ $s\geq r$ and $r+s=m$ the sum
$$
\sum_{i=1}^n{\frac{\S{a_{s},\ldots,a_1}{i}\S{b_{r},\ldots,b_1}{i}}{i^{c}}}
$$
can be express by polynomials of sums with depth $\leq\log_2{(m+1)}+1.$ We proceed by induction on $m.$\\
$m=2: \ \sum_{i=1}^n{\frac{\S{a_1}{i}\S{b_1}{i}}{i^{c}}}$ has depth $2\leq\log_2{(2+1)}+1$.\\
Now assume that the property holds for $m-1$ (we will refer to this induction hypothesis by (IH2)). We want to show that the property holds for $m.$ We will distinguish three different cases:\\

\begin{description}
 \item[s=r:] Using (IH1) the sum $$\sum_{i=1}^n{\frac{\S{a_{s},\ldots,a_1}{i}\S{b_{s},\ldots,b_1}{i}}{i^{c}}}$$ can be express by polynomials of sums with depth $\leq\log_2{(\frac{m}{2})}+1+1=\log_2{(m)}+1\leq\log_2{(m+1)}+1$. Hence the property holds.
 \item[s=r+1:] Using (IH1) the sum $$\sum_{i=1}^n{\frac{\S{a_{s},\ldots,a_1}{i}\S{b_{s-1},\ldots,b_1}{i}}{i^{c}}}$$ can be express by polynomials of sums with depth $\leq\log_2{(\frac{m+1}{2})}+1+1=\log_2{(m+1)}+1$. Hence the property holds.
 \item[s=r+d:] For $d=0$ and $d=1$ see the previous cases. Suppose now the property holds for $d<\hat{d}.$ We will show that the property holds for $\hat{d}$ (we will refer to this induction hypothesis by (IH3)):
\begin{eqnarray*}
\sum_{i=1}^n{\frac{\S{a_{s},\ldots,a_1}{i}\S{b_{s-\hat{d}},\ldots,b_1}{i}}{i^{c}}}&\stackrel{\text{Lemma} \ \ref{lem1}}{=}&
	\sum_{i=1}^n{\frac{\S{a_{s-1},\ldots,a_1}{i}\S{b_{s-\hat{d}},\ldots,b_1}{i}}{i^{c+a_s}}}\\
	&&+\S{a_s,a_{s-1},\ldots,a_1}{n}\S{b_{c,s-\hat{d}},\ldots,b_1}{n}\\
	&&-\sum_{i=1}^n{\frac{\S{a_{s-1},\ldots,a_1}{i}\S{c,b_{s-\hat{d}},\ldots,b_1}{i}}{i^{a_s}}}.
\end{eqnarray*}
Due to (IH2) the first sum can be express by polynomials of sums with depth $\leq\log_2{(m-1+1)}+1\leq\log_2{(m+1)}+1$.
Due to (IH1) the sum $\S{a_s,a_{s-1},\ldots,a_1}{n}$ can be express by polynomials of sums with depth $\leq\log_2{(s+1)}+1\leq\log_2{(m+1)}+1$.
Due to (IH1) the sum $\S{b_{c,s-\hat{d}},\ldots,b_1}{n}$ can be express by polynomials of sums with depth $\leq\log_2{(s-\hat{d}+1)}+1\leq\log_2{(m+1)}+1$.
Due to (IH3) the last sum can be express by polynomials of sums with depth $\leq\log_2{(m+1)}+1.$
Hence the property hods for $\hat{d}.$
\end{description}
We proved that we can express the sum
$$
\sum_{i=1}^n{\frac{\S{a_{s},\ldots,a_1}{i}\S{b_{r},\ldots,b_1}{i}}{i^{c}}}
$$
by polynomials of sums with depth $\leq\log_2{(m+1)}+1.$ Hence we can express the sum
$$
\sum_{i=1}^n{\frac{\S{a_{k-1},\ldots,a_1}{i}\sum_{j=1}^i{\frac{1}{j^{a_{k+1}}}}}{i^{a_k}}}=\sum_{i=1}^n{
\frac{\S{a_{k-1},\ldots,a_1}{i}\S{a_{k+1}}{i}}{i^{a_k}}}
$$
by polynomials of sums with depth $\leq\log_2{(k+1)}+1.$ This finishes the proof.
\end{proof}

\begin{example}
\begin{eqnarray*}
&&\S{(2,1,4),(3,1,2),(5,2,6)}{2,3,4,n}=
432 \sum _{i=1}^n \frac{6^{i} \sum _{j=1}^{i} \frac{4^{j}}{\left(5 j+2\right){}^6}}{2 i+1}
+108\sum _{i=1}^n \frac{6^{i} \sum _{j=1}^{i} \frac{4^{j}}{\left(5 j+2\right){}^6}}{\left(2 i+1\right){}^2}\\
&&\hspace{1cm}+24 \sum _{i=1}^n \frac{6^{i} \sum _{j=1}^{i} \frac{4^{j}}{\left(5 j+2\right){}^6}}{\left(2 i+1\right){}^3}
+4 \sum _{i=1}^n \frac{6^{i} \sum _{j=1}^{i} \frac{4^{j}}{\left(5 j+2\right){}^6}}{\left(2 i+1\right){}^4}
-648 \sum _{i=1}^n \frac{6^{i} \sum _{j=1}^{i} \frac{4^{j}}{\left(5 j+2\right){}^6}}{3 i+1}\\
&&\hspace{1cm}+\sum _{i=1}^n \frac{2^{i}}{\left(2 i+1\right){}^4} \sum _{i=1}^n \frac{3^{i} \sum _{j=1}^{i} \frac{4^{j}}{\left(5 j+2\right){}^6}}{\left(3 i+1\right){}^2}
+81 \sum _{i=1}^n \frac{6^{i} \sum _{j=1}^{i} \frac{4^{j}}{\left(5 j+2\right){}^6}}{\left(3 i+1\right){}^2}\\
&&\hspace{1cm}-\sum _{i=1}^n \frac{3^{i} \sum _{j=1}^{i} \frac{2^{j}}{\left(2 j+1\right){}^4} \sum _{j=1}^{i} \frac{4^{j}}{\left(5 j+2\right){}^6}}{\left(3 i+1\right){}^2}\\
&&\hspace{1cm}=432 \text{E}_{(2,1,1,6)}(\text{E}_{(5,2,6,4)})(n)+108 \text{E}_{(2,1,2,6)}(\text{E}_{(5,2,6,4)})(n)\\
&&\hspace{1cm}+24\text{E}_{(2,1,3,6)}(\text{E}_{(5,2,6,4)})(n)+4 \text{E}_{(2,1,4,6)}(\text{E}_{(5,2,6,4)})(n)\\
&&\hspace{1cm}-648\text{E}_{(3,1,1,6)}(\text{E}_{(5,2,6,4)})(n)+\text{E}_{(2,1,4,2)}(n)\text{E}_{(3,1,2,3)}(\text{E}_{(5,2,6,4)})(n)\\
&&\hspace{1cm}+81 \text{E}_{(3,1,2,6)}(\text{E}_{(5,2,6,4)})(n)-\text{E}_{(3,1,2,3)}(\text{E}_{(5,2,6,4)},\text{E}_{(2,1,4,2)})(n).
\end{eqnarray*}
\end{example}

\begin{remark}
 In \SigmaP\ indefinite nested sums and products are represented in so-called depth-optimal $\Pi\Sigma^*-$extensions \cite{Schneider2008a}, \ie the sums are represented with respect to 
optimal nesting depth \cite{Schneider2007a}. 
The underlying algorithm of Theorem \ref{dephreducethm} supports \SigmaP\ in this represantation, and might lead for various instances to more sufficient algorithms.
\end{remark}

\begin{remark}
 We mention that the reversed direction, \ie transforming Euler sums to cyclotomic S-sums can be achieved with the method presented in Section \ref{Summation of Cyclotomic S-Sums}.
\end{remark}

\cleardoublepage  

%\chapter{The Package \texttt{HarmonicSums}}
\chapter{The Package \ttfamily HarmonicSums \rmfamily}
\label{Packagechapter}
This chapter is dedicated to the presentation of the basic features of the package \ttfamily HarmonicSums \rmfamily which was developed \cite{Ablinger2009} and which was extended and generalized 
in the frame of this thesis. All the algorithms for 
harmonic sums, S-sums, cyclotomic harmonic sums, cyclotomic S-sums, harmonic polylogarithms, multiple polylogarithms and cyclotomic harmonic polylogarithms which were presented in the previous 
chapters of this thesis have been implemented.
This chapter contains a whole Mathematica session that runs throughout the sections. The inputs are given in the way how one has to type them into Mathematica and the outputs are
displayed as Mathematica gives them back.
We start the session by loading the package:
\begin{fmma}
{
\In \text{\bf <\hspace{-0.15cm} < HarmonicSums.m}\\
\fbox{\parbox{12cm}{\footnotesize HarmonicSums by Jakob Ablinger -RISC Linz- Version 1.0 (01/03/12)}}
}
\end{fmma}

\subsection*{Defintion of the Nested Sums}
Harmonic sums, S-sums, cyclotomic harmonic sums and cyclotomic S-sums are denoted by the letter \ttfamily S \rmfamily as we can see in the following examples.\\
The command \ttfamily ToHarmonicSumsSum \rmfamily yields the definition of the sums.
\begin{fmma}
{
\In \text{\bf S[1, 2, 3, 4, n]//ToHarmonicSumsSum}\\
\Out {\sum _{\iota _1=1}^\text{n} \frac{\sum _{\iota _2=1}^{\iota _1}\frac{\sum _{\iota _3=1}^{\iota _2} \frac{\sum _{\iota_4=1}^{\iota _3} \frac{1}{\iota _4^4}}{\iota _3^3}}{\iota_2^2}}{\iota _1}}\\
}
{
\In \text{\bf S[1, 2, 3, \{2, 1/3, 4\}, n]//ToHarmonicSumsSum}\\
\Out {\sum _{\iota _1=1}^\text{n} \frac{2^{\iota _1} \sum _{\iota _2=1}^{\iota _1} \frac{3^{-\iota _2} \sum _{\iota _3=1}^{\iota _2} \frac{4^{\iota _3}}{\iota _3^3}}{\iota _2^2}}{\iota _1}}\\
}
{
\In \text{\bf S[\{\{3, 2, 1\}, \{4, 1, 2\}, \{2, 0, -2\}\}, n]//ToHarmonicSumsSum}\\
\Out {\sum _{\iota _1=1}^\text{n} \frac{\sum _{\iota _2=1}^{\iota _1} \frac{\sum _{\iota _3=1}^{\iota _2} \frac{(-1)^{\iota _3}}{4 \iota _3^2}}{\left(4 \iota _2+1\right){}^2}}{3 \iota _1+2}}\\
}
{
\In \text{\bf S[\{\{3, 2, 1\}, \{4, 1, 2\}, \{2, 0, -2\}\}, \{2, 1/3, 4\}, n]//ToHarmonicSumsSum}\\
\Out {\sum _{\iota _1=1}^\text{n} \frac{2^{\iota _1} \sum _{\iota _2=1}^{\iota _1} \frac{3^{-\iota _2} \sum _{\iota _3=1}^{\iota _2} 4^{\iota _3+1} \iota _3^2}{\left(4 \iota _2+1\right){}^2}}{3
   \iota _1+2}}\\
}
\end{fmma}
Note that for internal reasons, sometimes the name \ttfamily CS \rmfamily is used to denote cyclotomic harmonic sums and cyclotomic S-sums.
\begin{fmma}
{
\In \text{\bf CS[\{\{3, 2, 1\}, \{4, 1, 2\}, \{2, 0, -2\}\}, n]//ToHarmonicSumsSum}\\
\Out {\sum _{\iota _1=1}^\text{n} \frac{\sum _{\iota _2=1}^{\iota _1} \frac{\sum _{\iota _3=1}^{\iota _2} \frac{(-1)^{\iota _3}}{4 \iota _3^2}}{\left(4 \iota _2+1\right){}^2}}{3 \iota _1+2}}\\
}
{
\In \text{\bf CS[\{\{3, 2, 1\}, \{4, 1, 2\}, \{2, 0, -2\}\}, \{2, 1/3, 4\}, n]//ToHarmonicSumsSum}\\
\Out {\sum _{\iota _1=1}^\text{n} \frac{2^{\iota _1} \sum _{\iota _2=1}^{\iota _1} \frac{3^{-\iota _2} \sum _{\iota _3=1}^{\iota _2} 4^{\iota _3+1} \iota _3^2}{\left(4 \iota _2+1\right){}^2}}{3
   \iota _1+2}}\\
}
\end{fmma}

\subsection*{Transformation to Nested Sums}
Using the command \ttfamily TransformToSSums \rmfamily the algorithm described in Section \ref{Summation of Cyclotomic S-Sums} is performed to rewrite nested sum expressions in terms of harmonic 
sums, S-sums, cyclotomic harmonic sums and cyclotomic S-sums.
\begin{fmma}
{
\In {\text{\bf{$\sum _{\iota _1=2}^{n-1} \frac{(-5)^{\iota _1} \left(\sum _{\iota _2=1}^{\iota _1} \frac{2^{\iota _2}}{\iota _2^3}\right) \sum _{\iota _2=1}^{\iota _1} \frac{1}{\iota _2+1}}{5 \iota
   _1+3}$}}\text{\bf//TransformToSSums}}\\
\Out {\frac{1}{40 (3 + 5 \text{n})}(75 + 125 \text{n} + 
  200 \text{n} \text{S}[\{\{5, 3, 1\}, \{1, 0, 3\}, \{1, 0, 1\}\}, \{-5, 2, 1\}, \text{n}] +\\ 
  300 \text{n} \text{S}[\{\{5, 3, 1\}, \{1, 0, 3\}\}, \{-5, 2\}, \text{n}] - 
 120 \text{S}[\{\{5, 3, 1\}, \{1, 0, 3\}, \{1, 0, 1\}\}, \{-5, 2, 1\}, \text{n}] - \\
  200 \text{n} \text{S}[\{\{5, 3, 1\}, \{1, 0, 4\}\}, \{-5, 2\}, \text{n}] + 
  120 \text{S}[\{\{5, 3, 1\}, \{1, 0, 1\}, \{1, 0, 3\}\}, \{-5, 1, 2\}, \text{n}] + \\
  200 \text{n} \text{S}[\{\{5, 3, 1\}, \{1, 0, 1\}, \{1, 0, 3\}\}, \{-5, 1, 2\}, \text{n}] + 
  120 \text{S}[\{\{5, 3, 1\}, \{1, 0, 4\}\}, \{-5, 2\}, \text{n}] + \\
   180 \text{S}[\{\{5, 3, 1\}, \{1, 0, 3\}\}, \{-5, 2\}, \text{n}]- 
  12 (-1)^\text{n} 5^{1 + \text{n}} \text{S}[3, \{2\}, \text{n}] - 12 \text{S}[4, \{-10\}, \text{n}] - \\
  20 \text{n} \text{S}[4, \{-10\}, \text{n}] + 8 (-1)^\text{n} 5^{1 + \text{n}} \text{S}[4, \{2\}, \text{n}] + 
  12 \text{S}[1, 3, \{-5, 2\}, \text{n}] + 20 \text{n} \text{S}[1, 3, \{-5, 2\}, \text{n}] - \\
  8 (-1)^\text{n} 5^{1 + \text{n}} \text{S}[1, 3, \{1, 2\}, \text{n}] - 
  8 (-1)^\text{n} 5^{1 + \text{n}} \text{S}[3, 1, \{2, 1\}, \text{n}])}\\
}
\end{fmma}
If in addition \SigmaP\ is loaded into Mathematica, sums can be handled whose denominators do not factor linearly over $\Q$ see Remark \ref{denomB}.

\subsection*{Synchronization}
After executing \ttfamily AutoSync[True] \rmfamily harmonic sums, S-sums, cyclotomic harmonic sums and cyclotomic S-sums at argument $a\cdot n+b$ for $a\in\N,b\in\N_0$ are synchronized automatically to 
argument $n$.
\begin{fmma}
{
\In \text{\bf S[1, 2, 2 n + 1]}\\
\Out {1 + 2 \text{S}[\{\{2, 1, 1\}\}, n] + \text{S}[\{\{2, 1, 2\}\}, n] + 
 \frac{1}{2} \text{S}[\{\{1, 0, 1\}, \{2, 1, 2\}\}, n] + \\
 \frac{1}{4} \text{S}[\{\{2, 1, 1\}, \{1, 0, 2\}\}, n] + \text{S}[\{\{2, 1, 1\}, \{2, 1, 2\}\}, n] + 
 \frac{1}{8} \text{S}[1, 2, n]}\\
}
{
\In \text{\bf S[\{\{2, 1, 2\}, \{3, 1, 2\}\}, 2 \text{n} + 1]}\\
\Out {\frac{1}{144}- \frac{18}{125} \text{S}[\{\{3, 2, 1\}\}, \text{n}] - \frac{9}{100} \text{S}[\{\{3, 2, 2\}\}, \text{n}] + 
 \frac{24}{125} \text{S}[\{\{4, 1, 1\}\}, \text{n}] - \frac{39}{400} \text{S}[\{\{4, 1, 2\}\}, \text{n}] + \\
 \frac{1}{16} \text{S}[\{\{4, 3, 2\}\}, \text{n}] + \frac{1}{4} \text{S}[\{\{4, 1, 2\}, \{3, 2, 2\}\}, \text{n}] + 
 \text{S}[\{\{4, 1, 2\}, \{6, 1, 2\}\}, \text{n}] +\\ \frac{1}{4} \text{S}[\{\{4, 3, 2\}, \{3, 2, 2\}\}, \text{n}] + 
 \text{S}[\{\{4, 3, 2\}, \{6, 1, 2\}\}, \text{n}]}\\
}
\end{fmma}

\subsection*{Definition of the Nested Integrals}
Harmonic polylogarithms, multiple polylogarithms and cyclotomic harmonic polylogarithms are denoted by the letter \ttfamily H \rmfamily as we can see in the following examples. The command 
\ttfamily ToHarmonicSumsIntegrate \rmfamily yields the definition of the nested integrals.
\begin{fmma}
{
\In \text{\bf H[0, 1, 0, -1, x]//ToHarmonicSumsIntegrate}\\
\Out {\int_0^\text{x} \frac{\int_0^{a_1} \frac{\int_0^{a_2}\frac{\int_0^{a_3} \frac{1}{a_4+1} \, da_4}{a_3} \,da_3}{a_2-1} \, da_2}{a_1} \, da_1}\\
}
{
\In \text{\bf H[1, 2, -3, 4, x]//ToHarmonicSumsIntegrate}\\
\Out {\int_0^\text{x} \frac{\int_0^{a_1} \frac{\int_0^{a_2}\frac{\int_0^{a_3} \frac{1}{a_4-4} \, da_4}{a_3+3} \, da_3}{a_2-2} \, da_2}{a_1-1} \, da_1}\\
}
{
\In \text{\bf H[\{3, 1\}, \{5, 1\}, \{2, 0\}, x]//ToHarmonicSumsIntegrate}\\
\Out {\int_0^\text{x} \frac{a_1 \left(\int_0^{a_1} \frac{a_2\left(\int_0^{a_2} \frac{1}{a_3+1} \,da_3\right)}{a_2^4+a_2^3+a_2^2+a_2+1} \,da_2\right)}{a_1^2+a_1+1} \, da_1}\\
}
\end{fmma}

\subsection*{Shuffle and Quasi-Shuffle Product}
We can use the functions \ttfamily LinearExpand \rmfamily and \ttfamily LinearHExpand \rmfamily to expand products of harmonic sum, S-sums, cyclotomic harmonic sums and cyclotomic S-sums and products 
of harmonic polylogarithms, multiple polylogarithms and cyclotomic harmonic polylogarithms, respectively.
\begin{fmma}
{
\In \text{\bf S[1,-2, \{1, 4\}, n] S[1, \{-3\}, n]//LinearExpand}\\
\Out {-\text{S}[1, -1, \{1, -12\}, \text{n}] - \text{S}[2, -2, \{-3, 4\}, \text{n}] + 
 \text{S}[1, -2, 1, \{1, 4, -3\}, \text{n}] + \text{S}[1, 1, -2, \{-3, 1, 4\}, \text{n}] + 
 \text{S}[1, 1, -2, \{1, -3, 4\}, \text{n}]}\\
}
{
\In \text{\bf S[\{\{3, 2, 1\}, \{2, 0, -2\}\}, \{1, 4\}, n] S[\{\{3, 1, 1\}\}, \{-3\}, n]//LinearExpand}\\
\Out {-\text{S}[\{\{3, 1, 1\}, \{2, 0, -2\}\}, \{-3, 4\}, \text{n}] + 
 \text{S}[\{\{3, 2, 1\}, \{2, 0, -2\}\}, \{-3, 4\}, \text{n}] +\\
 \text{S}[\{\{3, 1, 1\}, \{3, 2, 1\}, \{2, 0, -2\}\}, \{-3, 1, 4\}, \text{n}] + 
 \text{S}[\{\{3, 2, 1\}, \{2, 0, -2\}, \{3, 1, 1\}\}, \{1, 4, -3\}, \text{n}] + \\
 \text{S}[\{\{3, 2, 1\}, \{3, 1, 1\}, \{2, 0, -2\}\}, \{1, -3, 4\}, \text{n}]}\\
}
{
\In \text{\bf H[1, 2, x] H[3, 4, x]//LinearHExpand}\\
\Out {\text{H}[1, 2, 3, 4, \text{x}] + \text{H}[1, 3, 2, 4, \text{x}] + \text{H}[1, 3, 4, 2, \text{x}] + 
 \text{H}[3, 1, 2, 4, \text{x}] + \text{H}[3, 1, 4, 2, \text{x}] + \text{H}[3, 4, 1, 2, \text{x}]}\\
}
{
\In \text{\bf H[\{3, 1\}, \{5, 1\}, x] H[\{5, 2\}, \{2, 0\}, x]//LinearHExpand}\\
\Out {\text{H}[\{3, 1\}, \{5, 1\}, \{5, 2\}, \{2, 0\}, \text{x}] + 
 \text{H}[\{3, 1\}, \{5, 2\}, \{2, 0\}, \{5, 1\}, \text{x}] + \\
 \text{H}[\{3, 1\}, \{5, 2\}, \{5, 1\}, \{2, 0\}, \text{x}] + 
 \text{H}[\{5, 2\}, \{2, 0\}, \{3, 1\}, \{5, 1\}, \text{x}] + \\
 \text{H}[\{5, 2\}, \{3, 1\}, \{2, 0\}, \{5, 1\}, \text{x}] + 
 \text{H}[\{5, 2\}, \{3, 1\}, \{5, 1\}, \{2, 0\}, \text{x}]}\\
}
\end{fmma}

\subsection*{Transformation of the Argument of the Nested Integrals}
The function \ttfamily TransformH \rmfamily is used to transform the argument of harmonic polylogarithms, multiple polylogarithms and cyclotomic harmonic polylogarithms as described in the Sections 
\ref{HSRelatedArguments}, \ref{SSRelatedArguments} and \ref{CSRelatedArguments}. Some examples are:
\begin{fmma}
{
\In \text{\bf TransformH[H[-1, 0, (1 - x)/(1 + x)], x]}\\
\Out {\text{H}[-1, -1, \text{x}] + \text{H}[-1, 0, 1] + \text{H}[-1, 1, \text{x}]}\\
}
{
\In \text{\bf TransformH[H[-3, 0, -2, x + 2], x]}\\
\Out {\text{H}[-2, 2] \text{H}[-5, -2, \text{x}] + \text{H}[-5, \text{x}] \text{H}[0, -2, 2] + \text{H}[-5, -2, -4, \text{x}] + 
 \text{H}[-3, 0, -2, 2]}\\
}
{
\In \text{\bf TransformH[H[-3, 0, -2, 2 x], x]}\\
\Out {\text{H}[-3/2, 0, -1, \text{x}]}\\
}
{
\In \text{\bf TransformH[H[1, -2, 1 - x], x]}\\
\Out {-\text{H}[-2, 1] \text{H}[0, \text{x}] + \text{H}[-2, 0, 1] - \text{H}[-2, 1, 1] + \text{H}[0, -2, 1] + 
 \text{H}[0, 3, \text{x}]}\\
}
{
\In \text{\bf TransformH[H[{2, 1}, {2, 0}, x], 1/x]}\\
\Out {-\text{H}[\{0, -1\}, \{0, 0\}, 1] + \text{H}[\{0, -1\}, \{0, 0\}, 1/\text{x}] + 
 \text{H}[\{0, -1\}, \{2, 0\}, 1] - \text{H}[\{0, -1\}, \{2, 0\}, 1/\text{x}] + \\
 \text{H}[\{0, 0\}, \{0, 0\}, 1] - \text{H}[\{0, 0\}, \{0, 0\}, 1/\text{x}] - 
 \text{H}[\{0, 0\}, \{2, 0\}, 1] + \text{H}[\{0, 0\}, \{2, 0\}, 1/\text{x}] - \\
 \text{H}[\{2, 0\}, \{0, 0\}, 1] + \text{H}[\{2, 0\}, \{0, 0\}, 1/\text{x}] + 
 \text{H}[\{2, 0\}, \{2, 0\}, 1] - \text{H}[\{2, 0\}, \{2, 0\}, 1/\text{x}] + \\ \text{H}[\{2, 1\}, \{2, 0\}, 1]}\\
}
\end{fmma}

\subsection*{Power Series Expansions of the Nested Integrals}
The function \ttfamily HToS \rmfamily can be used to compute the power series expansions of harmonic polylogarithms, multiple polylogarithms and cyclotomic harmonic polylogarithms about $0$. 
\ttfamily SToH \rmfamily is used to perform the reverse direction.
\begin{fmma}
{
\In \text{\bf HToS[H[-1, 0, -1, x]]}\\
\Out {\sum _{\iota _1=1}^{\infty } \frac{\text{S}[2,\iota _1] (-\text{x})^{\iota _1}}{\iota
   _1}-\sum _{\iota _1=1}^{\infty } \frac{(-\text{x})^{\iota _1}}{\iota _1^3}}\\
}
{
\In \text{\bf SToH[$\sum _{\iota _1=1}^{\infty } \frac{(-x)^{\iota _1} \text{S[6,$\iota _1$]}}{\iota _1}$]}\\
\Out {\text{H}[-1, 0, -1, \text{x}] - \text{H}[0, 0, -1, \text{x}]}\\
}
{
\In \text{\bf HToS[H[-3, 0, -1/2, x]]}\\
\Out {\sum _{\iota _1=1}^{\infty } \frac{3^{-\iota _1} (-\text{x})^{\iota _1} \text{S}[2,\{6\},\iota
   _1]}{\iota _1}-\sum _{\iota _1=1}^{\infty } \frac{2^{\iota _1} (-\text{x})^{\iota
   _1}}{\iota _1^3}}\\
}
{
\In \text{\bf SToH[$\sum _{\iota _1=1}^{\infty } \frac{3^{-\iota _1} (-x)^{\iota _1} \text{S[2,\{6\},$\iota _1$]}}{\iota _1}$]}\\
\Out {\text{H}[-3, 0, -\frac{1}{2}, \text{x}] - \text{H}[0, 0, -\frac{1}{2}, \text{x}]}\\
}
{
\In \text{\bf HToS[H[{3, 1}, {1, 0}, x]]}\\
\Out {\sum _{\iota _1=1}^{\infty } \frac{\text{x}^{3 \iota _1+1} \text{S}[\{\{3,-2,1\}\},\iota _1]}{3 \iota _1+1}
-\sum _{\iota _1=1}^{\infty } \frac{\text{x}^{3 \iota _1+1}
   \text{S}[\{\{3,-1,1\}\},\iota _1]}{3 \iota _1+1}+\sum _{\iota _1=1}^{\infty }
   \frac{\text{x}^{3 \iota _1+2} \text{S}[\{\{3,-1,1\}\},\iota _1]}{3 \iota _1+2}-\sum
   _{\iota _1=1}^{\infty } \frac{\text{x}^{3 \iota _1+2} \text{S}[\{\{3,0,1\}\},\iota _1]}{3
   \iota _1+2}+\sum _{\iota _1=1}^{\infty } \frac{\text{x}^{3 \iota _1+3}
   \text{S}[\{\{3,0,1\}\},\iota _1]}{3 \iota _1+3}-\sum _{\iota _1=1}^{\infty }
   \frac{\text{x}^{3 \iota _1} \text{S}[\{\{3,-2,1\}\},\iota _1]}{3 \iota _1}}\\
}
\end{fmma}
The function \ttfamily HInfSeries \rmfamily can be used to compute the asymptotic behavior of harmonic polylogarithms, multiple polylogarithms and cyclotomic harmonic polylogarithms for $x \rightarrow \infty$.
\begin{fmma}
{
\In \text{\bf HInfSeries[H[-1, -2, x], x]}\\
\Out {\sum _{\iota _1=1}^{\infty } \frac{\left(-\frac{1}{x}\right)^{\iota _1}
   S_1\left(2;\iota _1\right)}{\iota _1}-\text{H}[0,\text{x}]\left(\sum _{\iota _1=1}^{\infty }
   \frac{\left(-\frac{1}{\text{x}}\right)^{\iota _1}}{\iota _1}\right)-\text{H}[0,2] \left(-\sum
   _{\iota _1=1}^{\infty } \frac{\left(-\frac{1}{\text{x}}\right)^{\iota _1}}{\iota
   _1}+\text{H}[0,x\text{x}]-\text{H}[-1,1]\right)-\sum _{\iota _1=1}^{\infty }
   \frac{\left(-\frac{1}{\text{x}}\right)^{\iota _1}}{\iota
   _1^2}+\text{H}[-1,-2,1]+\text{H}[-1,0,1]-\text{H}[-1,-\frac{1}{2},1]-\text{H}[0,0,1]+\text{H}[0,-\frac{1}{2},1]+\frac{1}{2} \text{H}[0,\text{x}]{}^2}\\
}
{
\In \text{\bf HInfSeries[H[\{3, 1\}, x], x]}\\
\Out {-\sum _{\iota _1=1}^{\infty } \frac{\left(\frac{1}{\text{x}}\right)^{3 \iota _1}}{3 \iota
   _1}+\sum _{\iota _1=1}^{\infty } \frac{\left(\frac{1}{\text{x}}\right)^{3 \iota _1-2}}{3
   \iota _1-2}-\text{H}[\{0,0\},\frac{1}{\text{x}}]+\text{H}[\{0,0\},1]-\text{H}[\{3,0\},1]}\\
}
\end{fmma}

Note that the function \ttfamily HarmonicSumsSeries \rmfamily can be used to compute the power series expansion of expressions involving harmonic polylogarithms, multiple polylogarithms and cyclotomic 
harmonic polylogarithms up to a specified order. As well it can be used to determine the asymptotic behavior of expressions involving harmonic polylogarithms, multiple polylogarithms and cyclotomic 
harmonic polylogarithms up to a specified order.
\begin{fmma}
{
\In \text{\bf HarmonicSumsSeries[$x^3$+H[-1, 0, -1, x]/x, x, 0, 5]}\\
\Out {\left(\frac{49 \text{x}^4}{1080}+\frac{\text{x}^3}{18}+\frac{\text{x}^2}{18}\right)
   \text{H}[0,\text{x}]-\frac{139 \text{x}^4}{2400}+\frac{11 \text{x}^3}{12}-\frac{11 \text{x}^2}{108}}\\
}
{
\In \text{\bf HarmonicSumsSeries[1/$x^3$+H[-1, -2, x]/x, x, $\infty$, 5]}\\
\Out {-\frac{131 \text{x}^4}{960}+\frac{\text{x}^3}{6}+\frac{1}{\text{x}^3}-\frac{5 \text{x}^2}{24}+\frac{\text{x}}{4}}\\
}
\end{fmma}

\subsection*{Transforming Nested Integrals at Special Values to Nested Sums at Infinity and Vice Versa}
The function \ttfamily HToSinf \rmfamily expresses values given by finite harmonic polylogarithms and multiple polylogarithms at real arguments using harmonic sums and S-sums at infinity while 
\ttfamily SinfToH \rmfamily transforms harmonic sums and S-sums at infinity to harmonic polylogarithms and multiple polylogarithms at real arguments. In addition \ttfamily HToSinf \rmfamily 
can be used to express values given by finite cyclotomic harmonic polylogarithms at argument $1$ using cyclotomic harmonic sums at infinity.
\begin{fmma}
{
\In \text{\bf HToSinf[H[-2, -1, 1]]}\\
\Out {-\text{S}[-2, \infty] + \text{S}[1, 1, \{-\frac{1}{2}, 2\}, \infty]}\\
}
{
\In \text{\bf SinfToH[S[1, 1, \{$-\frac{1}{2}$, 2\}, $\infty$]]}\\
\Out {\text{H}[-2, -1, 1] - \text{H}[0, -1, 1]}\\
}
{
\In \text{\bf HToSinf[H[\{2, 0\}, \{2, 1\}, 1]]}\\
\Out {\frac{1}{2} \text{S}[\{\{1,0,1\}\},\infty]+\frac{1}{4} \text{S}[\{\{1,0,2\}\},\infty]-\text{S}[\{\{2,1,1\}\},\infty]-\text{S}[\{\{2,1,2\}\},\infty]\\-\frac{1}{4} \text{S}[\{\{1,0,1\},\{1,0,1\}\},\infty]-\frac{1}{2}
   \text{S}[\{\{1,0,1\},\{2,1,1\}\},\infty]+\frac{1}{2} \text{S}[\{\{2,1,1\},\{1,0,1\}\},\infty]+\\ \text{S}[\{\{2,1,1\},\{2,1,1\}\},\infty]}\\
}
\end{fmma}

\subsection*{Mellin Transformation and Inverse Mellin Transformation}
To compute the Mellin transform of possibly weighted harmonic polylogarithms, multiple polylogarithms and cyclotomic harmonic polylogarithms \ttfamily hlog[x] \rmfamily we can use the command 
\ttfamily Mellin[hlog[x],x,n]\rmfamily. Using the option \ttfamily PlusFunctionDefinition$\rightarrow$2 \rmfamily the function \ttfamily Mellin \rmfamily computes 
the Mellin transform as it is defined in \cite{Ablinger2009} and \cite{Remiddi2000}.
\begin{fmma}
{
\In \text{\bf Mellin[H[1, 0, x]/(1 + x), x, n]}\\
\Out {(-1)^\text{n} (-\text{S}[-1,2,\text{n}]-2 \text{S}[3,\infty]+\text{S}[-2,-1,\infty]+\text{S}[-1,-2,\infty])}\\
}
{
\In \text{\bf Mellin[H[2, 3, x]/(2 + x), x, n]}\\
\Out {(-2)^\text{n} \biggl(\text{S}[-1,\text{n}] \text{S}[2,\left\{\frac{1}{3}\right\},\infty]-\text{S}[1,\left\{-\frac{1}{2}\right\},\text{n}]
   \text{S}[2,\left\{\frac{1}{3}\right\},\infty]+\text{S}[1,1,\left\{-1,\frac{1}{2}\right\},\text{n}] \text{S}[1,\left\{\frac{1}{3}\right\},\infty]
-\\ \text{S}[1,1,\left\{-1,\frac{3}{2}\right\},\text{n}] \text{S}[1,\left\{\frac{1}{3}\right\},\infty]-\text{S}[-1,\text{n}]
   \text{S}[1,1,\left\{\frac{1}{2},\frac{2}{3}\right\},\infty]+\\ \text{S}[1,\left\{-\frac{1}{2}\right\},\text{n}]
   \text{S}[1,1,\left\{\frac{1}{2},\frac{2}{3}\right\},\infty]
+\text{S}[1,1,1,\left\{-1,\frac{3}{2},\frac{1}{3}\right\},\text{n}]-\text{S}[3,\left\{\frac{1}{3}\right\},\infty]
+\\ \text{S}[1,2,\left\{-\frac{1}{2},-\frac{2}{3}\right\},\infty]+\text{S}[2,1,\left\{\frac{1}{2},\frac{2}{3}\right\},\infty]
-\text{S}[1,1,1,\left\{-\frac{1}{2},-1,\frac{2}{3}\right\},\infty]\biggr)}\\
}
{
\In \text{\bf Mellin[H[\{3, 1\}, \{3, 0\}, x], x, 3 n + 1]}\\
\Out {\frac{\text{n}}{36 \text{n}+24} \biggl(-24 \text{S}[\{\{3,2,1\}\},\text{n}] \text{S}[\{\{3,1,1\}\},\infty]
	+\bigl(4 \text{S}[\{\{3,1,1\}\},\text{n}] +4 \text{S}[\{\{3,2,1\}\},\text{n}]+\\ 6 \bigr)
  \text{S}[\{\{1,0,1\}\},\infty]
      +6 \text{S}[\{\{3,1,1\}\},\text{n}] \bigl(2 \text{S}[\{\{3,1,1\}\},\infty]-4 \text{S}[\{\{3,2,1\}\},\infty]-1\bigr)
       +\\ 12 \text{S}[\{\{3,2,1\}\},\text{n}] \text{S}[\{\{3,2,1\}\},\infty]+6 \text{S}[\{\{3,2,1\}\},\text{n}]-4 \text{S}[\{\{3,1,1\},\{1,0,1\}\},\text{n}]
       +\\12 \text{S}[\{\{3,1,1\},\{3,2,1\}\},\text{n}]-4 \text{S}[\{\{3,2,1\},\{1,0,1\}\},\text{n}]+12 \text{S}[\{\{3,2,1\},\{3,1,1\}\},\text{n}]
       +\\24 \text{S}[\{\{3,1,1\}\},\infty]+12 \text{S}[\{\{3,1,2\}\},\infty]-42 \text{S}[\{\{3,2,1\}\},\infty]-12 \text{S}[\{\{3,2,2\}\},\infty]
       +\\8 \text{S}[\{\{1,0,1\},\{3,1,1\}\},\infty]+4 \text{S}[\{\{1,0,1\},\{3,2,1\}\},\infty]
       -12 \text{S}[\{\{3,1,1\},\{3,1,1\}\},\infty]\\-24 \text{S}[\{\{3,1,1\},\{3,2,1\}\},\infty]-12 \text{S}[\{\{3,2,1\},\{3,1,1\}\},\infty]
       +12 \text{S}[\{\{3,2,1\},\{3,2,1\}\},\infty]-9\biggr)}\\
}
\end{fmma}
To compute the inverse Mellin transform of a harmonic sum or a $\bar{S}-$sum denoted by \ttfamily sum \rmfamily we can use the command \ttfamily InvMellin[sum,n,x]\rmfamily. Note that
\ttfamily Delta1x \rmfamily denotes the Dirac-$\delta$-distribution $\delta(1-x)\in D'[0,1]$. For cyclotomic 
harmonic sums and S-sums which are not $\bar{S}-$sums \ttfamily InvMellin \rmfamily yields an integral representations, where  \ttfamily Mellin[a[x,n]]\rmfamily$:=\int_0^1a(x,n)dx$ and 
\ttfamily Mellin[a[x,n],\{x,c,d\}]\rmfamily$:=\int_c^da(x,n)dx.$
\begin{fmma}
{
\In \text{\bf InvMellin[S[1, 2, n], n, x]}\\
\Out {\frac{\text{H}[1,0,\text{x}]}{1-\text{x}}}\\
}
{
\In \text{\bf InvMellin[S[1, 2, \{1, 1/3\}, n], n, x]}\\
\Out {\text{Delta1x} \biggl(-\text{S}[1,\left\{\frac{1}{3}\right\},\infty] \text{S}[2,\left\{\frac{1}{3}\right\},\infty]-2
   \text{S}[3,\left\{\frac{1}{3}\right\},\infty]+\text{S}[1,2,\left\{\frac{1}{3},1\right\},\infty]+\\ \text{S}[2,1,\left\{\frac{1}{3},1\right\},\infty]\biggr)
+\frac{3^{-\text{n}} \text{S}[2,\left\{\frac{1}{3}\right\},\infty]}{3-\text{x}}-\frac{\text{S}[2,\left\{\frac{1}{3}\right\},\infty]}{1-\text{x}}-\frac{3^{-\text{n}}\text{H}[3,0,\text{x}]}{\text{x}-3}}\\
}
{
\In \text{\bf InvMellin[S[\{\{3, 1, 2\}\}, n], n, x]}\\
\Out {-\text{Mellin}[\text{x}^{3 \text{n}} \text{H}[0,\text{x}]]-\frac{1}{3} \text{Mellin}[\frac{(\text{x}^{3 \text{n}}-1) \text{H}[0,\text{x}]}{\text{x}-1}]+\frac{1}{3} \biggl(2
   \text{Mellin}[\frac{(\text{x}^{3 \text{n}}-1) \text{H}[0,\text{x}]}{\text{x}^2+\text{x}+1}]+\\ \text{Mellin}[\frac{\text{x} (\text{x}^{3 \text{n}}-1) \text{H}[0,\text{x}]}{\text{x}^2+\text{x}+1}]\biggr)-1}\\
}
{
\In \text{\bf InvMellin[S[1, 1, \{2, 3\}, n], n, x]}\\
\Out {\text{H}[0,4] \text{Mellin}[\frac{\text{x}^\text{n}-1}{\text{x}-1},\{\text{x},2,6\}]-\text{Mellin}[\frac{\text{x}^\text{n} \text{H}[0,\text{x}-2]}{\text{x}-1},\{\text{x},2,6\}]\\+\text{Mellin}[\frac{\left(\text{x}^\text{n}-1\right)
   \text{H}[2,\text{x}]}{\text{x}-1},\{\text{x},0,2\}]+\text{Mellin}[\frac{\text{H}[0,\text{x}]}{\text{x}+1},\{\text{x},0,4\}]}\\
}
\end{fmma}

\subsection*{Differentiation of the Nested Sums}
In order to differentiate harmonic sums, S-sums or cyclotomic harmonic sums, we can use the function \ttfamily DifferentiateSSum\rmfamily. Note that \ttfamily z2, z3,... \rmfamily denote the values of 
the Riemann zeta function $\zeta(s)$ at $s=2, s=3,... $ while \ttfamily ln2 \rmfamily $=\log 2.$
\begin{fmma}
{
\In \text{\bf DifferentiateSSum[S[3, 1, n], n]}\\
\Out {-\text{S}[3,2,\text{n}]-3\;\text{S}[4,1,\text{n}]+\text{z2}\;\text{S}[3,\text{n}]-\text{z2}\;\text{z3}+\frac{9\;\text{z5}}{2}}\\
}
{
\In \text{\bf DifferentiateSSum[S[3, 1, \{2, 3\}, n], n]}\\
\Out {-\text{H}[1,0,3] \text{S}[3,\{2\},\text{n}]+\text{H}[0,2] \text{S}[3,1,\{2,3\},\text{n}]+\text{H}[0,3] \text{S}[3,1,\{2,3\},\text{n}]-\text{S}[3,2,\{2,3\},\text{n}]-\\3
   \text{S}[4,1,\{2,3\},\text{n}]+\text{H}[1,0,2] \text{H}[0,-2,-2,-1]-\text{H}[1,0,2] \text{H}[0,-2,-2,4]+\text{H}[0,-2,-1] \text{H}[0,1,0,2]\\-\text{H}[0,-2,4]
   \text{H}[0,1,0,2]-\text{H}[0,2] \text{H}[0,-2,-2,-1,-1]+\text{H}[0,2] \text{H}[0,-2,-2,-1,4]\\-\text{H}[2,1]
   \text{H}[0,0,1,0,2]-\text{H}[0,-2,-2,-1,-2,-1]+\text{H}[0,-2,-2,-1,-2,4]+\text{H}[2,0,0,1,0,1]}\\
}
{
\In \text{\bf DifferentiateSSum[S[\{\{2, 1, 1\}, \{2, 0, 1\}\}, n], n]}\\
\Out {\frac{1}{24} \biggl(-96\;\text{ln2}\;\text{S}[\{\{2,1,-1\}\},\infty]{}^2-192\;\text{ln2}\;\text{S}[\{\{2,1,-1\}\},\infty]
      +\\16 \text{S}[\{\{2,1,1\}\},\text{n}]\;\text{S}[\{\{2,1,-1\}\},\infty]{}^2+32\;\text{S}[\{\{2,1,1\}\},\text{n}] \text{S}[\{\{2,1,-1\}\},\infty]
      +6 \bigl(\text{S}[\{\{2,1,1\}\},\text{n}]+\\1\bigr) \text{S}[\{\{1,0,2\}\},\infty]+12 \bigl(\text{S}[\{\{2,1,2\}\},\text{n}]+1\bigr) \text{S}[\{\{1,0,1\}\},\infty]
      -12 \text{S}[1,\infty]\;\text{S}[\{\{2,1,2\}\},\text{n}]+\\4\;\text{S}[\{\{2,1,1\}\},\text{n}]
      -12 \text{S}[\{\{2,1,1\},\{1,0,2\}\},\text{n}]-24 \text{S}[\{\{2,1,2\},\{1,0,1\}\},\text{n}]+\\16 \text{S}[\{\{2,1,-1\}\},\infty]{}^2
      +32 \text{S}[\{\{2,1,-1\}\},\infty]-96 \text{ln2}-12\;\text{S}[1,\infty]+21\;\text{z3}+4\biggr)}\\
}
\end{fmma}

\subsection*{Series Expansions of the Nested Sums}
The function \ttfamily TaylorSeries[sum,x,n,ord] \rmfamily can be used to calculate the Taylor series expansion around \ttfamily x \rmfamily of a harmonic sum, S-sum or cyclotomic harmonic 
sum \ttfamily sum \rmfamily at argument \ttfamily n \rmfamily up to order \ttfamily ord\rmfamily.
\begin{fmma}
{
\In \text{\bf TaylorSeries[S[2, 1, n], 0, n, 3]}\\
\Out {\text{n}^3 \left(\frac{74\; \text{z2}^3}{105}-\text{z3}^2\right)+\text{n}^2 \left(2\;\text{z2}\;\text{z3}-\frac{11\;\text{z5}}{2}\right)+\frac{7\;\text{n}\;\text{z2}^2}{10}}\\
}
{
\In \text{\bf TaylorSeries[S[\{\{3, 1, 2\}\}, n], 0, n, 2]}\\
\Out {\frac{1}{2}\;\text{n}^2 \left(36\;\text{H}[\{3,0\},0,0,0,1]+18 \text{H}[\{3,1\},0,0,0,1]-\frac{36\;\text{z2}^2}{5}+54\right)+\text{n} (4\;\text{H}[\{3,0\},0,0,1]+\\2 \text{H}[\{3,1\},0,0,1]+2
   \;\text{z3}-6)}\\
}
\end{fmma}

The function \ttfamily SExpansion[sum,n,ord] \rmfamily can be used to calculate the asymptotic expansion of a harmonic sum, $\bar{S}$-sum or a cyclotomic harmonic 
sum \ttfamily sum \rmfamily at argument \ttfamily n \rmfamily up to order \ttfamily ord\rmfamily.
\begin{fmma}
{
\In \text{\bf SExpansion[S[3, 1, n], n, 3]}\\
\Out {\left(\frac{1}{2 \text{n}^3}-\frac{1}{2 \text{n}^2}\right) \text{L}[\text{n}]-\frac{1}{6 \text{n}^3}-\frac{1}{4
   \text{n}^2}+\frac{\text{z2}^2}{2}}\\
}
{
\In \text{\bf SExpansion[S[3, 1, \{1/2, 1/3\}, n], n, 3]}\\
\Out {\text{S}[1,\left\{\frac{1}{3}\right\},\infty] \left(-\text{S}[3,\left\{\frac{1}{6}\right\},\infty]+\text{S}[3,\left\{\frac{1}{2}\right\},\infty]
+6^{-\text{n}} \left(\frac{1}{5 \text{n}^3}-\frac{3^\text{n}}{\text{n}^3}\right)\right)+\\6^{-\text{n}} \left(\frac{12}{25 \text{n}^3}-\frac{1}{5 \text{n}^2}\right) \text{S}[2,\left\{\frac{1}{3}\right\},\infty]
+6^{-\text{n}} \left(\frac{42}{125 \text{n}^3}-\frac{6}{25 \text{n}^2}+\frac{1}{5 \text{n}}\right) \text{S}[3,\left\{\frac{1}{3}\right\},\infty]
+\\ \text{S}[2,\left\{\frac{1}{6}\right\},\infty] \text{S}[2,\left\{\frac{1}{3}\right\},\infty]-\text{S}[1,\left\{\frac{1}{6}\right\},\infty]
 \text{S}[3,\left\{\frac{1}{3}\right\},\infty]-\text{S}[4,\left\{\frac{1}{3}\right\},\infty]+\text{S}[1,3,\left\{\frac{1}{6},2\right\},\infty]
+\\6^{-\text{n}} \left(\frac{1}{5 \text{n}^2}-\frac{12}{25 \text{n}^3}\right) \text{H}[0,3,1]+6^{-\text{n}} \left(-\frac{42}{125 \text{n}^3}+\frac{6}{25 \text{n}^2}-\frac{1}{5 \text{n}}\right) \text{H}[0,0,3,1]-\frac{\text{H}[3,1] 6^{-\text{n}}}{5
   \text{n}^3}}\\
}
{
\In \text{\bf SExpansion[S[\{\{3, 1, 1\}, \{2, 1, 1\}\}, n], n, 3]}\\
\Out {\text{H}[\{6,0\},1] \left(-\frac{1}{3} \text{H}[-1,5]-\frac{\text{L}[\text{n}]}{3}+\text{ln2}-\frac{10}{243 \text{n}^3}+\frac{11}{108 \text{n}^2}-\frac{5}{18 \text{n}}+\frac{5}{3}\right)+\\ \text{H}[\{6,1\},1]
   \left(\frac{2}{3} (\text{H}[-1,5]-2)+\frac{2 \text{L}[\text{n}]}{3}-\text{ln2}+\frac{20}{243 \text{n}^3}-\frac{11}{54 \text{n}^2}+\frac{5}{9 \text{n}}\right)+\\ \text{H}[\{3,1\},1]
   \left(-\text{H}[\{6,0\},1]+\text{H}[\{6,1\},1]-\frac{1}{2} \text{H}[-1,2]-\frac{\text{L}[\text{n}]}{3}+\frac{1}{36 \text{n}^2}-\frac{1}{6 \text{n}}+\frac{8}{3}\right)+\\ \text{H}[\{3,0\},1]
   \left(\text{H}[\{6,0\},1]-\text{H}[\{6,1\},1]-\frac{1}{2} \text{H}[-1,2]-\frac{\text{L}[\text{n}]}{6}+\frac{1}{72 \text{n}^2}-\frac{1}{12
   \text{n}}+\frac{1}{3}\right)-\text{H}[\{6,0\},1]^2+\\ \text{H}[\{6,1\},1]^2-\frac{2}{3} \text{H}[-1,\{6,0\},1]+\frac{1}{3} \text{H}[-1,\{6,1\},1]-\frac{1}{3}
   \text{H}[\{3,0\},-1,1]+\frac{1}{2} \text{H}[\{3,0\},1,1]+\\ \frac{1}{2} \text{H}[\{3,0\},\{3,0\},1]+\frac{1}{2} \text{H}[\{3,0\},\{3,1\},1]-\frac{2}{3}
   \text{H}[\{3,0\},\{6,0\},1]+\frac{1}{3} \text{H}[\{3,0\},\{6,1\},1]+\\ \frac{1}{3} \text{H}[\{3,1\},-1,1]+\frac{1}{2}
   \text{H}[\{3,1\},1,1]+\text{H}[\{3,1\},\{3,0\},1]+\text{H}[\{3,1\},\{3,1\},1]+\frac{2}{3} \text{H}[\{3,1\},\{6,0\},1]-\\ \frac{1}{3} \text{H}[\{3,1\},\{6,1\},1]+\frac{1}{3}
   \text{H}[\{6,0\},-1,1]-\frac{2}{3} \text{H}[\{6,0\},1,1]+\frac{2}{3} \text{H}[\{6,0\},\{6,0\},1]-\\ \frac{1}{3} \text{H}[\{6,0\},\{6,1\},1]+\frac{1}{3}
   \text{H}[\{6,1\},-1,1]+\frac{1}{3} \text{H}[\{6,1\},1,1]+\frac{2}{3} \text{H}[\{6,1\},\{6,0\},1]-\\ \frac{1}{3} \text{H}[\{6,1\},\{6,1\},1]+\text{ln2} \left(\frac{1}{3}
   (\text{H}[-1,5]-5)+\frac{\text{L}[\text{n}]}{3}+\frac{10}{243 \text{n}^3}-\frac{11}{108 \text{n}^2}+\frac{5}{18 \text{n}}\right)\\+\text{L}[\text{n}] \left(\frac{1}{6} (\text{H}[-1,2]-2)+\frac{5}{243 \text{n}^3}-\frac{11}{216
   \text{n}^2}+\frac{5}{36 \text{n}}\right)+\frac{\frac{25}{96}-\frac{1}{72} \text{H}[-1,2]}{\text{n}^2}+\frac{\frac{1}{12} \text{H}[-1,2]-\frac{7}{18}}{\text{n}}-\\ \frac{1}{3}
   \text{H}[-1,5]+\frac{\text{L}[\text{n}]^2}{12}-\frac{3905}{23328 \text{n}^3}-\frac{\text{z2}}{4}}\\
}
\end{fmma}

Note that the function \ttfamily HarmonicSumsSeries[expr,n,p,ord] \rmfamily can be used to compute series expansions about the point \ttfamily x=p \rmfamily of expressions \ttfamily expr \rmfamily involving 
harmonic sums, $\bar{S}$-sums, cyclotomic harmonic sums, harmonic polylogarithms, multiple polylogarithms and cyclotomic 
harmonic polylogarithms up to a specified order \ttfamily ord. \rmfamily

\begin{fmma}
{
\In \text{\bf HarmonicSumsSeries[n*S[2, n] + n*H[-2, n], n, 0, 4] // ReduceConstants}\\
\Out {\text{n}^4 \left(4\; \text{z5}+\frac{1}{24}\right)+\text{n}^3 \left(-\frac{6\; \text{z2}^2}{5}-\frac{1}{8}\right)+\text{n}^2 \left(2\;\text{z3}+\frac{1}{2}\right)}\\
}
{
\In \text{\bf HarmonicSumsSeries[n*S[2, n] + n*H[-2, n], n, 1, 4]}\\
\Out {(\text{n}-1) \left(\text{H}[-2,1]+2\; \text{z3}-\frac{2}{3}\right)+\text{H}[-2,1]+(\text{n}-1)^4 \left(-\frac{8\; \text{z2}^3}{7}+4\;
   \text{z5}+\frac{109}{108}\right)\\+(\text{n}-1)^2 \left(-\frac{6\; \text{z2}^2}{5}+2\;
   \text{z3}+\frac{23}{18}\right)+(\text{n}-1)^3 \left(-\frac{6\; \text{z2}^2}{5}+4\;
   \text{z5}-\frac{169}{162}\right)+1}\\
}
{
\In \text{\bf HarmonicSumsSeries[n*S[2, n] + n*H[-2, n], n, $\infty$, 4] // ReduceConstants}\\
\Out {-\text{n}\; \text{H}[0,2] +\text{n}\; \text{H}[0,\text{n}]-\frac{4}{\text{n}^3}+\frac{5}{2\; \text{n}^2}+\text{n}\; \text{z2}-\frac{3}{2 \text{n}}+1}\\
}
\end{fmma}

\subsection*{Basis Representations}
In order to look for basis representations of harmonic sums, S-sums, cyclotomic harmonic sums, harmonic polylogarithms and multiple polylogarithms, the functions \ttfamily ComputeHSumBasis, ComputeSSumBasis, ComputeCSumBasis \rmfamily and 
\ttfamily ComputeHLogBasis \rmfamily are provided.
\begin{itemize}
 \item \ttfamily ComputeHSumBasis[w,n] \rmfamily computes a basis and the corresponding relations for harmonic sums at weight \ttfamily w\rmfamily. With the options 
    \ttfamily UseDifferentiation \rmfamily and \ttfamily UseHalfInteger \rmfamily it can be specified whether relations due to differentiation and argument duplication should be used.
 \item \ttfamily ComputeSSumBasis[w,x,n] \rmfamily computes a basis and the corresponding relations for S-sums at weight \ttfamily w \rmfamily where the allowed $``x``-$ indices are defined in the list 
    \ttfamily x\rmfamily. With the options \ttfamily UseDifferentiation \rmfamily and \ttfamily UseHalfInteger \rmfamily it can be specified whether relations due to 
    differentiation and argument duplication should be used.
 \item \ttfamily ComputeCSumBasis[w,{let},n] \rmfamily computes a basis and the corresponding relations for cyclotomic harmonic sums at weight \ttfamily w \rmfamily with letters \ttfamily let\rmfamily. 
    With the options \ttfamily UseDifferentiation, UseMultipleInteger \rmfamily and \ttfamily UseHalfInteger \rmfamily it can be specified whether relations due to differentiation and 
    argument multiplication should be used.
\item \ttfamily ComputeHLogBasis[w,n] \rmfamily computes a basis and the corresponding relations for multiple polylogarithms at weight \ttfamily w\rmfamily. 
    The option \ttfamily Alphabet->a \rmfamily and \ttfamily IndexStructure->ind \rmfamily can be used to specify an alphabet or a special index structure respectively.
\end{itemize}

\begin{fmma}
{
\In {\text{\bf ComputeCSumBasis[2, \{\{2, 1\}\}, n, UseDifferentiation -> False,}\\
    \text{\bf UseMultipleInteger -> False, UseHalfInteger -> False}}\\
\Out {\biggl\{
\bigl\{\text{S}[\{\{2,1,-2\}\},\text{n} ],\text{S}[\{\{2,1,2\}\},\text{n} ],\text{S}[\{\{2,1,-1\},\{2,1,1\}\},\text{n} ]\bigr\}, \\
\bigl\{\text{S}[\{\{2,1,1\},\{2,1,1\}\},\text{n} ]\to \frac{1}{2} \text{S}[\{\{2,1,1\}\},\text{n} ]{}^2+\frac{1}{2} \text{S}[\{\{2,1,2\}\},\text{n} ], \\
	\text{S}[\{\{2,1,1\},\{2,1,-1\}\},\text{n} ]\to \text{S}[\{\{2,1,-2\}\},\text{n} ]+\text{S}[\{\{2,1,-1\}\},\text{n} ] \text{S}[\{\{2,1,1\}\},\text{n} ]\\-\text{S}[\{\{2,1,-1\},\{2,1,1\}\},\text{n} ],
	\text{S}[\{\{2,1,-1\},\{2,1,-1\}\},\text{n} ]\to \frac{1}{2} \text{S}[\{\{2,1,-1\}\},\text{n} ]{}^2+\\ \frac{1}{2} \text{S}[\{\{2,1,2\}\},\text{n} ]
\bigr\}
\biggr\}}\\
}
\end{fmma}
In order to look for relations for harmonic sums, S-sums and and cyclotomic harmonic sums at infinity we can use the functions \ttfamily ComputeHSumInfBasis, ComputeSSumInfBasis \rmfamily and 
\ttfamily ComputeCSumInfBasis\rmfamily.
\begin{fmma}
{
\In \text{\bf ComputeCSumInfBasis[2, \{\{2, 1\}\}]}\\
\Out {\biggl\{
\bigl\{\text{S}[\{\{2,1,-2\}\},\infty ],\text{S}[\{\{2,1,2\}\},\infty ],\text{S}[\{\{2,1,-1\},\{2,1,1\}\},\infty ]\bigr\}, \\
\bigl\{\text{S}[\{\{2,1,1\},\{2,1,1\}\},\infty ]\to \frac{1}{2} \text{S}[\{\{2,1,1\}\},\infty ]{}^2+\frac{1}{2} \text{S}[\{\{2,1,2\}\},\infty ], \\
	\text{S}[\{\{2,1,1\},\{2,1,-1\}\},\infty ]\to \text{S}[\{\{2,1,-2\}\},\infty ]+\text{S}[\{\{2,1,-1\}\},\infty ] \text{S}[\{\{2,1,1\}\},\infty ]\\-\text{S}[\{\{2,1,-1\},\{2,1,1\}\},\infty ],
	\text{S}[\{\{2,1,-1\},\{2,1,-1\}\},\infty ]\to \frac{1}{2} \text{S}[\{\{2,1,-1\}\},\infty ]{}^2+\\ \frac{1}{2} \text{S}[\{\{2,1,2\}\},\infty ]
\bigr\}
\biggr\}}\\
}
\end{fmma}
For harmonic sums and cyclotomic harmonic sums tables with relations are provided. These tables can be applied using the command \ttfamily ReduceToBasis\rmfamily . With the options \ttfamily UseDifferentiation \rmfamily 
and \ttfamily UseHalfInteger \rmfamily it is possible to specify whether relations due to differentiation and argument duplication should be used.
\ttfamily ReduceToBasis[expr,n,Dynamic->True] \rmfamily computes relations between harmonic sums, S-sums and cyclotomic harmonic sums in \ttfamily expr \rmfamily from scratch 
and applies them.
\ttfamily ReduceToHBasis \rmfamily uses precomputed tables with relations between harmonic polylogarithms and applies them to expressions involving harmonic polylogarithms.
\ttfamily ReduceConstants \rmfamily uses precomputed tables with relations between harmonic polylogarithms at argument $1$ and harmonic sums at infinity to reduce the appearing 
constants as far as possible.
\begin{fmma}
{
\In \text{\bf ReduceToBasis[S[2, 1, n] + S[1, 2, n], n]}\\
\Out {\text{S}[1, \text{n}] \text{S}[2, \text{n}] + \text{S}[3, \text{n}]}\\
}
{
\In \text{\bf ReduceToBasis[S[5, 5, {3, 3}, n], n, Dynamic -> True]}\\
\Out {\frac{1}{2} \left(\text{S}[5,\{3\},\text{n}]^2+\text{S}[10,\{9\},\text{n}]\right)}\\
}
{
\In \text{\bf ReduceToHBasis[H[1, 0, x] + H[0, 1, x]]}\\
\Out {\text{H}[0,\text{x}] \text{H}[1,\text{x}]}\\
}
{
\In \text{\bf ReduceConstants[S[2, $\infty$] + 2 H[1, 0, 1] + H[0, 1, -1, 1]]}\\
\Out {-\frac{3}{2} \text{S}[-1,\infty] \text{S}[2,\infty]-\text{S}[2,\infty]-\text{S}[3,\infty]}\\
}
\end{fmma}
\subsection*{Depth Reduction}
To reduce the depth of a harmonic sum, S-sum, cyclotomic harmonic sum or cyclotomic S-sum as described in the proof of Theorem \ref{dephreducethm} the function \ttfamily ReduceDepth \rmfamily is provided.
\begin{fmma}
{
\In \text{\bf ReduceDepth[S[1, 2, 3, 4, n]] // ToHarmonicSumsSum}\\
\Out {\sum _{\iota _1=1}^\text{n} \frac{\left(\sum _{\iota _2=1}^{\iota _1} \frac{\sum _{\iota _3=1}^{\iota _2} \frac{1}{\iota _3^4}}{\iota _2^3}\right) \sum _{\iota _2=1}^{\iota _1}
   \frac{1}{\iota _2^2}}{\iota _1}-\left(\sum _{\iota _1=1}^\text{n} \frac{\left(\sum _{\iota _2=1}^{\iota _1} \frac{1}{\iota _2^4}\right) \sum _{\iota _2=1}^{\iota _1} \frac{1}{\iota
   _2^2}}{\iota _1^3}\right) \sum _{\iota _1=1}^\text{n} \frac{1}{\iota _1}+\\ \sum _{\iota _1=1}^\text{n} \frac{\left(\sum _{\iota _2=1}^{\iota _1} \frac{1}{\iota _2^4}\right) \left(\sum _{\iota
   _2=1}^{\iota _1} \frac{1}{\iota _2^2}\right) \sum _{\iota _2=1}^{\iota _1} \frac{1}{\iota _2}}{\iota _1^3}-\sum _{\iota _1=1}^\text{n} \frac{\left(\sum _{\iota _2=1}^{\iota _1}
   \frac{1}{\iota _2^4}\right) \sum _{\iota _2=1}^{\iota _1} \frac{1}{\iota _2^2}}{\iota _1^4}-\\ \sum _{\iota _1=1}^\text{n} \frac{\left(\sum _{\iota _2=1}^{\iota _1} \frac{1}{\iota
   _2^4}\right) \sum _{\iota _2=1}^{\iota _1} \frac{1}{\iota _2}}{\iota _1^5}+\left(\sum _{\iota _1=1}^\text{n} \frac{1}{\iota _1}\right) \sum _{\iota _1=1}^\text{n} \frac{\sum _{\iota _2=1}^{\iota
   _1} \frac{1}{\iota _2^4}}{\iota _1^5}+\sum _{\iota _1=1}^\text{n} \frac{\sum _{\iota _2=1}^{\iota _1} \frac{1}{\iota _2^4}}{\iota _1^6}}\\
}
\end{fmma}

\cleardoublepage 

\chapter{Multi-Variable Almkvist Zeilberger Algorithm and Feynman Integrals}
\label{AZchapter}
In \cite{Bluemlein2011} integrals emerging in renormalizable Quantum Field Theories, like Quantum Electrodynamics or Quantum Chromodynamics are transformed by means of symbolic computation 
to hypergeometric multi-sums. The very general class of Feynman integrals which are considered in \cite{Bluemlein2011} are of relevance for many physical processes at high energy colliders, 
such as the Large Hadron Collider, LHC, and others.
The considered integrals are two--point
Feynman integrals in $D$-dimensional Minkowski space with one time- and $(D-1)$
Euclidean space dimensions, $\ep = D - 4$ and $\ep \in {\mathbb R}$ with
$|\ep| \ll 1$ of the
following structure:
%--------------------------------------------------------------------------
\begin{eqnarray}
{\cal I}(\ep,N,p) = \int \frac{d^D p_1}{(2\pi)^D} \ldots \int \frac{d^D
p_k}{(2\pi)^D}
\frac{{\cal N}(p_1, \ldots p_k; p; m_1 \ldots
m_k;
\Delta, N)}{(-p_1^2 + m_1^2)^{l_1} \ldots (-p_k^2 + m_k^2)^{l_k}}
\prod_V \delta_V~.
\label{eq:A7}
\end{eqnarray}
%---------------------------------------------------------------------------
They can be shown to obey diffence equations with respect to $N,$ see, \eg \cite{BKKS}.
In (\ref{eq:A7}) the external momentum $p$ and the loop momenta $p_i$ denote $D$-dimensional
vectors, $m_i > 0 , m_i \in
{\mathbb R}$ are scalars
(masses),
$m_i \in \{0, M\}$,
$k, l_i \in {\mathbb N}$, $k \geq 2,l_i \geq 1$, and $\Delta$ is a light-like
$D$-vector,
$\Delta.\Delta = 0$. The numerator function ${\cal N}$ is a polynomial in the
scalar products $p.p_i,~p_i.p_k$ and of monomials $(\Delta.p_{(i)})^{n_i}$, $n_i \in {\mathbb N},
n_i \geq 0$. $N \in {\mathbb N}$ denotes the spin of a local operator stemming from
the light cone expansion, see, e.g., \cite{Frishman1971} and references therein, which
contributes to the numerator function ${\cal N}$ with a polynomial in $\Delta.p_i$ of
maximal degree $N$, cf. \cite{Bierenbaum2009}. Furthermore it is assumed for simplicity that only
one of the loops is formed of massive lines. The $\delta_V$ occurring in (\ref{eq:A7}) are shortcuts for Dirac delta distributions in $D$ dimensions 
$\delta_V = \delta^{(D)}\left(\sum_{l=1}^k a_{V,l}p_l\right), a_{V,l}\in\Q.$\\
These integrals are mathematically well defined and in \cite{Bluemlein2011} it is showed how they can be mapped onto
integrals on the $m$-dimensional unit cube with the following structure:
\begin{eqnarray}
{\cal I}(\ep,N) = C(\ep, N, M) \int_0^1 dy_1 \ldots \int_0^1 dy_m
\frac{\sum_{i=1}^k \prod_{{l}=1}^{r_i}
[P_{i,l}(y)]^{\alpha_{i,l}(\ep,N)}}{[Q(y)]^{\beta(\ep)}}~,
\label{eq:A9}
\end{eqnarray}
%---------------------------------------------------------------------------
with $k\in\set N$, $r_1,\dots,r_k\in\set N$ and
where $\beta(\ep)$ is  given by a rational function in $\ep$, i.e., $\beta(\ep)\in\set R(\ep)$, and similarly
$\alpha_{i,l}(\ep,N) = n_{i,l} N + \overline{\alpha}_{i,l}$ for some $n_{i,l} \in \{0,1\}$ and $\overline{\alpha}_{i,l}\in\set R(\ep)$, see also \cite{Bogner2010}
in the case no local operator insertions are present.
$C(\ep, N, M)$ is a factor, which depends on the dimensional parameter $\ep$,
the integer parameter $N$ and the mass $M$.
$P_i(y), Q(y)$ are polynomials in the remaining Feynman parameters $y=(y_1,\dots,y_m)$ written in multi-index notation.
In (\ref{eq:A9}) all terms
which stem from local operator insertions were geometrically resumed; see~\cite{Bierenbaum2009}.
In \cite{Bluemlein2011} it was already mentioned that after splitting the integral (\ref{eq:A9}), the integrands fit into the input class of the multivariate Almkvist-Zeilberger algorithm. Hence, if 
the split integrals are properly defined, they obey homogenous recurrence relations in $N$ due to the theorems in \cite{AlmZeil}. In \cite{Bluemlein2011} the integrals of (\ref{eq:A9}) are transformed
further to a multi-sum representation, while in this chapter we want to tackle them directly by looking on integrals of the form
\begin{eqnarray}
\label{AZhypexpint}
{\cal I}(\ep,N) = \int_{u_d}^{o_d} \dots\int_{u_1}^{o_1}F(n;x_1, \dots, x_d;\ep) dx_1 \dots dx_d,
\end{eqnarray}
with $d,N \in \N$, $F(n;x_1, \dots, x_d;\ep)$ a hyperexponential term, $\ep>0$ a real parameter and $u_i,o_i \in \R\cup \{-\infty,\infty\}.$
Here we will use our package \ttfamily MultiIntegrate \rmfamily that can be considered as an enhenced implementation of the multivariate Almkvist Zeilberger algorithm to compute recurrences for the integrands and integrals we will subsequently take a closer look onto it. 
Subsequently, $\set K$ denotes a field with $\set Q\subseteq\set K$ (\eg $\set K =\Q(\ep)$ forms a rational function field) in which the usual operations can be computed.

\section{Finding a Recurrence for the Integrand}
\label{AZintegrandrec}
In order to find a recurrence for the integrand we will use the following theorem given in \cite{AlmZeil} which gives rise to the multi-variable Almkvist Zeilberger algorithm.
\begin{thm}[mAZ; see \cite{AlmZeil}]
Let
\begin{equation}
F(n;x_1, \dots , x_d)=POL(n;x_1, \dots, x_d) \cdot H(n;x_1, \dots, x_d), \label{mAZintegrand}
\end{equation}
where  $POL(n;x_1, \dots, x_d) \in \set K[n,x_1, \dots, x_d]$,  and
$$
H(n;x_1, \dots, x_d)=
e^{a(x_1, \dots,x_d)/b(x_1, \dots, x_d)} \cdot
\left ( \prod_{p=1}^P {S_p(x_1, \dots, x_d)}^{\alpha_p} \right )\cdot
\left ( { \frac{s(x_1, \dots, x_d)}{t(x_1, \dots, x_d)} } \right )^n,
$$
where $a(x_1, \dots, x_d),b(x_1, \dots, x_d)$,
$s(x_1, \dots, x_d),t(x_1, \dots, x_d)$ and
$S_p(x_1, \dots, x_d) \in \set K[x_1, \dots, x_d]$ ($1 \leq p \leq P$), and $\alpha_p\in \set K$. Note that such a sequence $F(n;x_1, \dots , x_d)$ is called hyperexponential.
Then there exists a non-negative integer $L$ and
there exist $e_0(n),e_1(n), \dots , e_L(n)\in\set K[n]$,
{\it not all zero}, and  there also exist $R_i(n;x_1, \dots, x_d)\in\set K(n,x_1, \dots, x_d)$
($i=1, \dots ,d$) such that
$$
G_i(n;x_1, \dots, x_d):=R_i(n;x_1, \dots,x_d)F(n;x_1, \dots, x_d)
$$
satisfy
\begin{equation}
\sum_{i=0}^L e_i(n) F(n+i;x_1, \dots, x_d)= \sum_{i=1}^d D_{x_i} G_i(n;x_1, \dots, x_d). \label{mAZrec}
\end{equation}
\label{mAZthm}
\end{thm}

In the proof of the theorem given in \cite{AlmZeil} the following expression is defined for a non-negative integer $L$
$$
\overline{H} (n;x_1, \dots, x_d)
    =e^{a(x_1, \dots,x_d)/b(x_1, \dots, x_d)} \cdot\left (\prod_{p=1}^P {S_p(x_1, \dots, x_d)}^{\alpha_p} \right )\cdot{ \frac{s(x_1, \dots, x_d)^n}{t(x_1, \dots, x_d)^{n+L}} }.
$$
From the logarithmic derivatives (for $i=1,\dots,d$) of $\overline{H} (n;x_1, \dots, x_d)$
$$
{ \frac{D_{x_i} \overline{H}(n;x_1, \dots, x_d)}{\overline{H}(n;x_1, \dots, x_d)} }={\frac{q_i(x_1, \dots, x_d)}{r_i(x_1, \dots, x_d)}}
$$
the ansatz for $i=1, \dots, d$,
\begin{equation}
G_i(n;x_1, \dots, x_d)=\overline{H}(n;x_1, \dots, x_d) \cdot r_i(x_1,\dots, x_d ) \cdot X_i(x_1, \dots, x_d), \label{mAZansatz}
\end{equation}
where $X_i(x_1, \dots , x_d)\in \set K[x_1, \dots, x_d]$ to be determined, is built. With this ansatz (\ref{mAZrec})
is equivalent to
\begin{eqnarray}
&&\sum_{i=1}^d [D_{x_i}r_i(x_1, \dots, x_d)+q_i(x_1, \dots, x_d)] \cdot X_i(x_1, \dots, x_d) \nonumber \\
&&\hspace{1cm}+r_i(x_1, \dots, x_d)\cdot D_{x_i}X_i(x_1, \dots, x_d)= h(x_1, \dots, x_d). \label{mAZrec2}
\end{eqnarray}

The general algorithm now is straightforward: Given an integrand of the form (\ref{mAZintegrand}), we can set $L=0,$ look for degree bounds for $X_i(x_1, \dots , x_d)$ and try to find a 
solution of (\ref{mAZrec2}) by coefficient comparison. If we do not find a solution of (\ref{mAZrec2}) with not all $e_i(n)$'s equal to zero, we increase $L$ by one, look for new degree bounds 
for $X_i(x_1, \dots , x_d)$ and try again to find a solution of (\ref{mAZrec2}). Again, if we do not find a solution  with not all $e_i(n)$'s equal to zero, we increase $L$ by one and repeat 
the process. Since according to \cite{AlmZeil} the existence of a solution of (\ref{mAZrec2}) with not all the $e_i(n)$'s equal to zero is guaranteed for sufficiently large $L,$ this process will 
eventually terminate.\\
A crucial point in this algorithm is the determination of degree bounds for the $X_i(x_1, \dots , x_d),$ in \ttfamily MultiIntegrate \rmfamily we use similar considerations as in the proof of Theorem mZ of \cite{AlmZeil} 
to determine such bounds.

During the implementation of our package \ttfamily MultiIntegrate \rmfamily it turned out that setting up the system (\ref{mAZrec2}) can be done usually rather fast. However the bottle neck for complicated examples usually is the solving of
(\ref{mAZrec2}) especially when there are several symbolic parameters around. Hence we put special emphasis on this step. In the following we want to describe a method using homomorphic images to speed up 
the solving of (\ref{mAZrec2}) and at the same time the whole algorithm: Note that parts of these speed ups have been exploited the first time in \ttfamily MultiSum \rmfamily \cite{Wegschaider,Riese}. Similar techniques are used in 
\ttfamily Sigma\rmfamily.
  
\begin{itemize}
 \item Replace all parameters (especially $n$ and $\alpha_p$) in (\ref{mAZrec2}) which pop up in (\ref{mAZintegrand}) by primes $q_j$ of a size which can be handled by hardware arithmetic and denote the 
	resulting equation by Eq$_h.$
 \item Use coefficient comparison with respect to $x_1,\ldots,x_d$ in Eq$_h;$ this leads to a system of equations SEq$_h$ and let $A_h$ be the corresponding matrix
 \item Solve the system of equations SEq$_h$ modulo a prime $q$ distinct from the $q_j$, by computing the null space of $A_h$ and denote the obtained basis by NS$_h;$  in our implementation we use 
	\ttfamily Nullspace \rmfamily of Mathematica to accomplish this task.
	\begin{itemize}
		\item If NS$_h$ is the empty set, conclude, that there is no solution of (\ref{mAZrec2}) and hence increase $L$ and compute a new ansatz
		\item If NS$_h$ is not empty, sort the elements of NS$_h$ by the number of contained zeros (most zeroes on top) and denote them by $l_1,\ldots,l_k.$
	\end{itemize}
 \item Start with $l_1;$ set the variables which correspond to zero entries in $l_1$ to zero; in this way Eq$_h$ (and Eq) gets simpler (if there are zero entries in $l_1$) and coefficient 
	comparison with respect to $x_1,\ldots,x_d$ leads to a smaller system of equations SEq$_h;$ let $A_h$ now be the corresponding matrix; note that in comparison to the original $A_h$ the new one has 
	less columns; in the next step the number of rows is reduced (if possible).
 \item Compute the null space of the transpose of $A_h$ modulo $q;$ read off the rows, which can be removed (these rows correspond to equations in SEq$_h$ which can be removed); 
	using \ttfamily Nullspace \rmfamily of Mathematica this means interpreting the output as a matrix and looking for the last non zero element in each row.
 \item Apply coefficient comparison in Eq and remove the equations which correspond to the row numbers which we found in the previous step.
 \item Solve the system of equations SEq (which is usually now much smaller as in the beginning). If there are solution with not all $e_i(n)$'s equal to zero, take them and plug them into (\ref{mAZrec2}) to test whether it is a solution of (\ref{mAZrec});
	if none of these solutions is a solution of (\ref{mAZrec2}), try the next $l_i$; if none of the $l_i$ leads to a solution of (\ref{mAZrec}), increase $L$ and compute a new ansatz.
\end{itemize}

\begin{remark}
By replacing the parameters by primes we might introduce fake solutions, however this is quite unlikely and it has so far never happened in all the examples we have calculated so far. If we would run 
into such a fake solution this will not harm at all since we can always check if our solution solves the original problem. Although we might introduce new fake solutions we will not loose solutions
of the original problem, since we just add structure and do not remove it.\\
Removing columns and and rows might as well introduce fake solutions however in the end we will always have a certificate which guarantees the correctness of our result.
\end{remark}

\section{Finding a Recurrence for the Integral}
\label{AZintegralrec}
From \cite{AlmZeil} we know that if $F(n; \pm \infty)=0$ (and hence $G(n; \pm \infty)=0$) in Theorem \ref{mAZthm}
then it follows, by integrating over $[-\infty,\infty]^d$, that
$$
a(n):=\int_{-\infty}^{\infty} \dots\int_{-\infty}^{\infty}F(n;x_1, \dots, x_d) dx_1 \dots dx_d
$$
satisfies a homogenous linear recurrence equation with polynomial coefficients of the form
\begin{equation}\label{AZlinrec}
\sum_{i=0}^L e_i(n) a(n+i)= 0.
\end{equation}
We are interested in different integration limits and of course this can be generalized to other limits of integration: if $F(n;\dots,x_{i-1},u_i,x_{i+1},\dots)=0$ and $F(n;\dots,x_{i-1},o_i,x_{i+1},\dots)=0$ 
(and hence $G(n;\dots,x_{i-1},u_i,x_{i+1},\dots)=0$ and $G(n;\dots,x_{i-1},o_i,x_{i+1},\dots)=0$ ) in Theorem \ref{mAZthm} then
$$
a(n):=\int_{u_d}^{o_d} \dots\int_{u_1}^{o_1}F(n;x_1, \dots, x_d) dx_1 \dots dx_d,
$$
satisfies the homogenous linear recurrence equation with polynomial coefficients (\ref{AZlinrec}).
However, in general the integrand in (\ref{AZhypexpint}) does not necessarily vanish at the limits of integration. Therefore we have to extend the above sketched algorithm.
Subsequently we are going to present two ways of extending the algorithm such that we can handle integrals where $F(n;x_1, \dots, x_d)$ (and hence $G(n;x_1, \dots, x_d)$) does not vanish at the 
integration limits.

The first extension consists in forcing the $G_i$ to vanish at the integration bounds by modifying the ansatz (\ref{mAZansatz}). For $u_i,o_i\in\R$ we can for example use the ansatz 
(for multi-sums we refer, \eg to \cite{Wegschaider})
\begin{equation}
G_i(n;x_1, \dots, x_d)=\overline{H}(n;x_1, \dots, x_d) \cdot r_i(x_1,\dots, x_d ) \cdot X_i(x_1, \dots, x_d)(x_i-u_i)(x_i-o_i), \label{mAZansatz2}
\end{equation}
In principle, this new ansatz works, however the drawback is that it might increase the order of the recurrence that we find drastically. It might even happen that we would find a recurrence (suggested by homomorphic image testing) using 
ansatz (\ref{mAZansatz}), but due to time and space limitations we do not explicitly find a recurrence using ansatz (\ref{mAZansatz2}).

\begin{example}
As a straightforward example we consider the integral
$$
I(\ep,n)=\int_0^1\int_0^1\frac{(1+x_1\cdot x_2)^n}{(1+x_1)^\ep}dx_1dx_2.
$$
Note that the integrand does not vanish at the integration bounds, hence we apply the algorithm using ansatz (\ref{mAZansatz2}); this leads to the recurrence
\begin{eqnarray*}
&&2 (n+1) (\ep-n-2) I(\ep,n)-(n+2) (5 \ep-5 n-13) I(\ep,n+1)\\
&&+(n+3) (4 \ep-4 n-13) I(\ep,n+2)-(n+4) (\ep-n-4) I(\ep,n+3)=0.
\end{eqnarray*}
Using \SigmaP, the solution of the recurrence together with the initial values of the integral leads to the result
$$
I(\ep,n)=\frac{1}{n+1}\left(\sum _{i=1}^n \frac{1}{-i+\ep -1}-2^{1-\ep} \sum _{i=1}^n \frac{2^i}{-i+\ep -1}+\frac{2^{-\ep }\left(2^{\ep }-2\right)}{\ep -1}\right)
$$ 
\end{example}

The second extension keeps the ansatz (\ref{mAZansatz}), however it deals with the inhomogeneous recurrence similar to \cite{Bluemlein2011}; this will give rise to a recursive method. We consider the 
integral
$$
a(n):=\int_{u_d}^{o_d} \cdots\int_{u_1}^{o_1}F(n;x_1, \dots, x_d) dx_1 \dots dx_d.
$$
Suppose that we found
\begin{equation}
\sum_{i=0}^L e_i(n) F(n+i;x_1, \dots, x_d)= \sum_{i=1}^d D_{x_i} G_i(n;x_1, \dots, x_d)
\end{equation}
where at least one $G_i(n;x_1, \dots, x_d)$ does not vanish at the integration limits. By integration with respect to $x_1,\ldots,x_d$ we can deduce that $a(n)$ satisfies the inhomogeneous linear recurrence 
equation
\begin{eqnarray*}
\sum_{i=0}^L e_i(n) a(n+i)&=&\sum_{i=1}^d \int_{u_d}^{o_d} \cdots \int_{u_{i-1}}^{o_{i-1}}\int_{u_{i+1}}^{o_{i+1}}\cdots \int_{u_1}^{o_1}O_i(n)dx_1\dots dx_{i-1}dx_{i+1}\dots dx_d\\
			  & &-\sum_{i=1}^d \int_{u_d}^{o_d} \cdots \int_{u_{i-1}}^{o_{i-1}}\int_{u_{i+1}}^{o_{i+1}}\cdots \int_{u_1}^{o_1}U_i(n)dx_1\dots dx_{i-1}dx_{i+1}\dots dx_d
\end{eqnarray*}
with
\begin{eqnarray*}
U_i(n)&:=&G_i(n;x_1,\dots,x_{i-1},o_i,x_{i+1} \dots, x_d)\\
O_i(n)&:=&G_i(n;x_1,\dots,x_{i-1},u_i,x_{i+1} \dots, x_d).
\end{eqnarray*}
Note that the inhomogeneous part of the above recurrence equation is a sum of $2\cdot d$ integrals of dimension $d-1,$ which fit again into the input class of the Almkvist-Zeilberger algorithm. 
Hence we can apply the algorithms to the $2\cdot d$ integrals recursively until we arrive at the base case of one-dimensional integrals for which we have to solve an inhomogeneous linear recurrence relation 
where the inhomogeneous part is free of integrals. Given the solutions for the one-dimensional integrals we can step by step find the solutions of higher dimensional integrals until we finally find the 
solution for $a(n)$ by solving again an inhomogeneous linear recurrence equation and combining it with the initial values. Note that we have to calculate initial values with respect to $n$ for all the integrals 
arising in this process.\\ 
To compute the solutions of the arising recurrences we exploit algorithms from~\cite{Petkovsek1992,Abramov1994,Schneider2001,Schneider2005} which can constructively decide if a solution with certain initial values 
is expressible in terms of indefinite nested products and sums. To be more precise, we use the summation package~\SigmaP\ in which algorithms are implemented with which one can solve the 
following problem (see \cite{Bluemlein2011}).

\medskip

\begin{ProblemSpec}{\textbf{Problem \ProblemRS}: \textbf{R}ecurrence \textbf{S}olver for indefinite nested product-sum expressions.}
\textbf{Given} $a_0(N),\dots,a_d(N)\in\set K[N]$; given $\mu\in\set N$ such that $a_d(k)\neq0$ for all $k\in\set N$ with $N\geq\mu$; given an expression $h(N)$ in terms of indefinite nested product-sum expressions which can be evaluated for all $N\in\set N$ with $N\geq\mu$; given the initial values $(c_{\mu},\dots,c_{\mu+d-1})$ which produces the sequence $(c_i)_{i\geq\mu}\in{\set K}^{\set N}$ by the defining recurrence relation
$$a_0(N)c_N+a_1(N)c_{N+1}+\dots+a_d(N)c_{N+d}=h(N)\quad\forall N\geq\mu.$$
\textbf{Find}, if possible, $\lambda\in\set N$ with $\lambda\geq\mu$ and an indefinite nested product-sum expression $g(N)$ such that $g(k)=c_k$ for all $k\geq\lambda$.
\end{ProblemSpec}

Summarizing, with these algorithms we use the following strategy (note that we assume that we are able to compute the initial values for the arising integrals); compare \cite{Bluemlein2011}:

\noindent\textbf{Divide and conquer strategy}
\begin{enumerate}
\item BASE CASE: If $\mathcal{I}(N)$ has no integration quantifiers, return $\mathcal{I}(N).$

\item DIVIDE: As worked out in above, compute a recurrence relation
\begin{equation}\label{Equ:IntRec}
a_0(N)\mathcal{I}(N)+\dots+a_d(N)\mathcal{I}(N+d)=h(N)
\end{equation}
with polynomial coefficients
$a_i(N)\in\set K[N]$, $a_m(N)\neq0$ and the right side $h(N)$ containing a linear
combination of hyperexponential multi-integrals each with less than $d$ summation
quantifiers.
\item CONQUER: Apply the strategy recursively to the simpler integrals in
$h(N)$. This results in an indefinite nested product-sum expressions $\tilde{h}(N)$ with 
\begin{equation}\label{Equ:hsol}
\tilde{h}(N)=h(N)
\end{equation}
if the method fails to find the $\tilde{h}(N)$ in terms of indefinite
nested product-sum expressions, STOP.

\item COMBINE: Given~\eqref{Equ:IntRec} with~\eqref{Equ:hsol},
compute, if possible, $\tilde{\mathcal{I}}(N)$ in terms of nested product-sum expressions such that
\begin{equation}
\tilde{\mathcal{I}}(N)=\mathcal{I}(N)
\end{equation}
by solving problem \ProblemRS.
\end{enumerate}
\normalsize
\medskip

\begin{example}
Again we consider the integral
$$
I(\ep,n)=\int_0^1\int_0^1\underbrace{\frac{(1+x_1\cdot x_2)^n}{(1+x_1)^\ep}}_{F(n;x_1,x_2):=}dx_1dx_2,
$$
however now we will use the second proposed extension of the algorithm. Applying the algorithm using ansatz (\ref{mAZansatz}) and choosing $\set K=\Q(\ep)$ leads to
\begin{eqnarray*}
-(n+1)F(n;x_1,x_2)+(n+2)F(n+1;x_1,x_2)=D_{x_1}0+D_{x_2}\frac{x_2(x_1\cdot x_2+1)^{n+1}}{(1+x_1)^{\ep}}
\end{eqnarray*}
and hence it follows by integration
\begin{eqnarray}
-(n+1)I(\ep,n)+(n+2)I(\ep,n+1)=\underbrace{\int_0^1(x_1+1)^{n+1-\ep}dx_1}_{I_1(n)}-\int_0^1 0 dx_1. \label{mAZExrec1}
\end{eqnarray}
In the next step apply the algorithm to $I_1(n);$ we find
$$
I_1(\ep,n)=\frac{4\cdot 2^n-2^{\ep}}{2^{\ep}(n+2-\ep)}.
$$
Plugging in this result into (\ref{mAZExrec1}) yields
$$
-(n+1)I(\ep,n)+(n+2)I(\ep,n+1)=\frac{4\cdot 2^n-2^{\ep}}{2^{\ep}(n+2-\ep)}.
$$
The solution of this recurrence together with the initial values of the integral yields again the result
$$
I(\ep,n)=\frac{1}{n+1}\left(\sum _{i=1}^n \frac{1}{-i+\ep -1}-2^{1-\ep} \sum _{i=1}^n \frac{2^i}{-i+\ep -1}+\frac{2^{-\ep }\left(2^{\ep }-2\right)}{\ep -1}\right).
$$ 
\end{example}

\begin{remark}
 We can as well combine both extensions and thereby balance the size of the inhomogeneous part and the order of the recurrence.
 Although we could handle a much larger class of integrals using the second extension we again ran into integrals which we could 
not process due to time and space limitations. Note that another bottle neck
might be the calculation of the initial values of all the integrals arising.
\end{remark}

\section{Finding \texorpdfstring{$\ep$}{Epsilon}-Expansions of the Integral}\label{AZexpint}
As already mentioned we ran into integrals which we could not process due to time and space limitations using the methods described in the previous section. Therefore, inspired 
by \cite{Bluemlein2011}, we develop subsequently another method which computes $\ep$-expansions of integrals of the form (\ref{AZhypexpint}). In the following we assume that the
integral ${\cal I}(\ep,N)$ from (\ref{AZhypexpint}) has a Laurent expansion in $\ep$ for each $N\in\N$ with $N\geq\lambda$ for some $\lambda\in\N$ and 
thus it is an analytic function in $\ep$ throughout an annular region centered by $0$ where the pole at $\ep=0$ has some order $L.$ Hence we can write it in the form
\begin{equation}
 {\cal I}(\ep,N) = \sum_{l=-L}^{\infty} \ep^l I_l(N).
\end{equation}
In the following we try to find the first coefficients $I_t(N),I_{t+1}(N),\ldots,I_{u}(N)$ in terms of indefinite nested product-sum expressions of the expansion
\begin{equation}\label{Equ:FExp2}
{\cal I}(\ep, N) = I_t(N)\ep^t+I_{t+1}(N)\ep^{t+1}+I_{t+2}(N)\ep^{t+2}+\dots
\end{equation}
with $t=-L,\ (t \in \Z).$

Restricting the $\mathcal{O}$-notation to formal Laurent series $f=\sum_{i=r}^{\infty}f_i\ep^i$ and $g=\sum_{i=s}^{\infty}g_i\ep^i$ the notation
$$f=g+O(\ep^t)$$
for some $t\in\set Z$ means that the order of $f-g$ is larger or equal to $t$, i.e., $f-g=\sum_{i=t}^{\infty}h_i\ep^i$.
We start by computing a recurrence for  ${\cal I}(\ep,N)$ in the form
\begin{equation}\label{Equ:ExpansionEquMod}
a_0(\ep,N)T(\ep,N)+\dots+a_d(\ep,N)T(\ep,N+d)=h_0(N)+\dots,+h_u(N)\ep^u+O(\ep^{u+1});
\end{equation}
in order to accomplish this task we can use the methods presented in the previous two sections.
Given the recurrence we exploit an algorithm from \cite{Bluemlein2011} which can constructively decide if a formal Laurent series solution with certain initial values is expressible (up to a certain order) 
in terms of indefinite nested products and sums. To be more precise, with the algorithm \FLSR\ presented in \cite{Bluemlein2011} and implemented in \SigmaP\ we can solve the following problem.

\begin{ProblemSpec}{\textbf{Problem \FLSR}: \textbf{F}ormal \textbf{L}aurent \textbf{S}eries solutions of linear \textbf{R}ecurrences.}
\textbf{Given} $\mu\in\set N$; $a_0(\ep,N),\dots,a_d(\ep,N)\in\set K[\ep,N]$ such that $a_d(0,k)\neq0$ for all $k\in\set N$ with $k\geq\mu$; 
indefinite nested product-sum expressions $h_t(N)$, $\dots$, $h_u(N)$ ($t,u\in\set Z$ with $t\leq u$) which can be evaluated for all $N\in\set N$ with $N\geq \mu$;
$c_{i,j}\in\set K$ with $t\leq i\leq u$ and $\mu\leq j<\mu+d$\\
\textbf{Find,} if possible, $(r,\lambda,\tilde{T}(N)),$
where $r\in\{t-1,t,\dots,u\}$ is the maximal number such that for the unique solution  $T(N)=\sum_{i=t}^uF_i(N)\ep^i$ with $F_i(k)=c_{i,k}$ for all $\mu\leq k<\mu+d$ 
and with the relation~\eqref{Equ:ExpansionEquMod} the following holds: there are indefinite nested product-sum expressions $\tilde{F}_t(N),\dots,\tilde{F}_r(N)$ that compute the $F_t(N),\dots,F_r(N)$ for 
all $N\geq\lambda$ for some $\lambda\geq\mu$; if $r\geq t$, $\tilde{T}(N)=\sum_{i=t}^r\tilde{F}_i(N)\ep^i$.
\end{ProblemSpec}

Summarizing the considerations leads to the following theorem (compare \cite{Bluemlein2011}).
\begin{thm}\label{AZAlgForMultiInt}
Let ${\cal I}(\ep,N)$ be an integral of the form \eqref{AZhypexpint} which forms an analytic function in $\ep$ throughout an annular region centered 
by $0$  with the Laurent expansion ${\cal I}(\ep,N)=\sum_{i=t}^{\infty}f_i(N)\ep^i$ for some $t\in\set Z$
for each non-negative $N$; let $u\in\set N$. Then there is an algorithm (provided that we can compute initial values for the arising integrals) which finds the maximal $r\in\{t-1,0,\dots,u\}$ such that 
the $f_t(N),\dots,f_r(N)$ are expressible in terms of indefinite nested product-sums; it outputs such expressions $F_t(N),\dots,F_r(N)$ 
and $\lambda\in\set N$ s.t.\ $f_i(k)=F_i(k)$ for all $0\leq i\leq r$ and all $k\in\set N$ with $k\geq\lambda$.
\end{thm}

Let ${\cal I}(\ep,N)$ be a multi-integral of the form~\eqref{AZhypexpint} and assume that ${\cal I}(\ep,N)$ has a series expansion~\eqref{Equ:FExp2} for all $N\geq\lambda$ for 
some $\lambda\in\set N$. Combining the methods of the previous sections we obtain the following general method (compare \cite{Bluemlein2011}) to compute the first coefficients, say $F_t(N),\dots,F_u(N)$ of~\eqref{Equ:FExp2}. Note that we
assume that we can handle the initial values.

\medskip
\noindent\textbf{Divide and conquer strategy}
\begin{enumerate}
\item BASE CASE: If $\mathcal{I}(\ep, N)$ has no integration quantifiers, compute the expansion by standard methods.
\item DIVIDE: As worked out in Section~\ref{AZintegralrec}, compute a recurrence relation
\begin{equation}\label{Equ:IntRecurrence}
a_0(\ep,N)\mathcal{I}(\ep, N)+\dots+a_d(\ep,N)\mathcal{I}(\ep, N+d)=h(\ep,N)
\end{equation}
with polynomial coefficients
$a_i(\ep,N)\in\set K[\ep,N]$, $a_m(\ep,N)\neq0$ and the right side $h(\ep,N)$ containing a linear
combination of hyperexponential multi-integrals each with less than $d$ integration
quantifiers.

\item CONQUER: Apply the strategy recursively to the simpler integrals in
$h(\ep,N)$. This results in an expansion of the form
\begin{equation}\label{Equ:hExpansion}
h(\ep,
N)=h_t(N)+h_1(N)\ep+\dots+h_u(N)\ep^u+O(\ep^{u+1});
\end{equation}
if the method fails to find the $h_t(N),\dots,h_u(N)$ in terms of indefinite
nested product-sum expressions, STOP.

\item COMBINE: Given~\eqref{Equ:IntRecurrence} with ~\eqref{Equ:hExpansion},
compute, if possible, the $F_t(N),\dots,F_u(N)$ of~\eqref{Equ:FExp2} in terms of
nested product-sum expressions by executing Algorithm~\FLSR\ of \SigmaP.
\end{enumerate}
\normalsize
\medskip

\begin{example}
Again we consider the integral
$$
I(\ep,n)=\int_0^1\int_0^1\underbrace{\frac{(1+x_1\cdot x_2)^n}{(1+x_1)^\ep}}_{F(n;x_1,x_2):=}dx_1dx_2,
$$
however now we will use the second proposed extension of the algorithm. Applying the algorithm using ansatz (\ref{mAZansatz}) leads to
\begin{eqnarray*}
-(n+1)F(n;x_1,x_2)+(n+2)F(n+1;x_1,x_2)=D_{x_1}0+D_{x_2}\frac{x_2(x_1\cdot x_2+1)^{n+1}}{(1+x_1)^{\ep}}
\end{eqnarray*}
and hence it follows by integration
\begin{eqnarray}
-(n+1)I(\ep,n)+(n+2)I(\ep,n+1)=\underbrace{\int_0^1(x_1+1)^{n+1-\ep}dx_1}_{I_1(n)}-\int_0^1 0 dx_1. \label{mAZExexprec1}
\end{eqnarray}
In the next step apply the method to $I_1(n);$ we find
\begin{eqnarray*}
I_1(\ep,n)&=&\frac{2^{n+2}-1}{n+2}+\frac{\ep \left(-2^{n+2} \left(\textnormal{H}_{-1}(1) (n+2)-1\right)-1\right)}{(n+2)^2}\\
	   && +\frac{\ep^2 \left(2^{n+1}\left(\textnormal{H}_{-1}(1){}^2 (n+2)^2-2 \textnormal{H}_{-1}(1) (n+2)+2\right)-1\right)}{(n+2)^3}+O\left(\ep^3\right).
\end{eqnarray*}
Plugging in this result into (\ref{mAZExrec1}) and solving the resulting recurrence by means of \FLSR\ and combining the solution with the initial values of the integral yields the result
\begin{eqnarray*}
I(\ep,n)&=&\frac{\textnormal{S}_1(2;n)}{n+1}-\frac{\textnormal{S}_1(n)}{n+1}+\frac{2 (n+1)^3+4 \left(-n+2^n-1\right) (n+1)^2+2 n (n+1)^2}{2 (n+1)^4}\\
	&&+\ep \Biggl( \textnormal{H}_{-1}(1) \left(\frac{-2^{n+2} (n+1)-2^{n+2} n (n+1)}{2(n+1)^4}-\frac{\textnormal{S}_1(2;n)}{n+1}\right)\\
	&&+\frac{\left(2 (n+1) n^2+4 (n+1) n+2 (n+1)\right) \textnormal{S}_2(2;n)}{2(n+1)^4}-\frac{\textnormal{S}_2(n)}{n+1}\\
	&&+\frac{2^{n+2} (n+1)-2 (n+1)}{2 (n+1)^4}\Biggr)+\ep^2 \Biggl( \textnormal{H}_{-1}(1){}^2 \left(\frac{\textnormal{S}_1(2;n)}{2 (n+1)}+\frac{2^{n+1}( n^2+2 n+1)}{2 (n+1)^4}\right)\\
	&&+\textnormal{H}_{-1}(1)\left(\frac{-2^{n+2} n-2^{n+2}}{2 (n+1)^4}-\frac{\textnormal{S}_2(2;n)}{n+1}\right)+\frac{\textnormal{S}_3(2;n)}{n+1}
	  -\frac{\textnormal{S}_3(n)}{n+1}+\frac{2^{n+2}-2}{2(n+1)^4}\Biggr)\\&&+O\left(\ep^3\right).
\end{eqnarray*}
\end{example}

\section{The Package \ttfamily MultiIntegrate \rmfamily}
\setcounter{mmacnt}{0}
This section is dedicated to the presentation of the basic features of the package \ttfamily MultiIntegrate, \rmfamily which was developed in the frame of this thesis.
In order to use the package \ttfamily MultiIntegrate \rmfamily C. Schneider's packages \ttfamily Sigma \rmfamily \cite{Schneider2007} and \ttfamily EvaluateMultiSum \rmfamily (see \cite{Hasselhuhn2012}) as well as 
the package \ttfamily HarmonicSums \rmfamily have to be loaded.
\begin{fmma2}
{
\In \text{\bf <\hspace{-0.15cm} < Sigma.m}\\
\text{\hspace{0.88cm}\footnotesize\bf<\hspace{-0.15cm} < HarmonicSums.m}\\
\text{\hspace{0.88cm}\footnotesize\bf<\hspace{-0.15cm} < EvaluateMultiSums.m}\\
\text{\hspace{0.88cm}\footnotesize\bf<\hspace{-0.15cm} < MultiIntegrate.m}\\
\fbox{\parbox{12cm}{\footnotesize Sigma by Carsten Schneider -RISC Linz- Version 1.0}}
\fbox{\parbox{12cm}{\footnotesize HarmonicSums by Jakob Ablinger -RISC Linz- Version 1.0 (01/03/12)}}
\fbox{\parbox{12cm}{\footnotesize EvaluateMultiSums by Carsten Schneider -RISC Linz- Version 1.0}}
\fbox{\parbox{12cm}{\footnotesize MultiIntegrate by Jakob Ablinger -RISC Linz- Version 1.0 (01/03/12)}}
}
\end{fmma2}
The function \ttfamily mAZ \rmfamily which is provided by the package \ttfamily MultiIntegrate \rmfamily performs the multivariate Almkvist-Zeilberger and hence finds recurrences for
hyperexponentional integrands:
\begin{fmma2}
{
\In {\text{\bf mAZ$\Big[\frac{(1 + x_1\; x_2)^n}{(1 + x_1)^\ep}, n, \{x_1,x_2\},f\Big]$}}\\
\Out {\Bigg\{\Bigg\{1,\frac{\left(x_1 x_2+1\right){}^n}{\left(x_1+1\right){}^2}\Bigg\},\Bigg\{\frac{\left(x_1
   x_2+1\right){}^n}{\left(x_1+1\right){}^2},\left\{0,x_1\right\},\left\{x_1 x_2^2+x_2,x_2\right\},(-n-1) f(n)+(n+2) f(n+1)\Bigg\}\Bigg\}}\\
}
\end{fmma2}

The 3 different integration methods discussed in Sections \ref{AZintegralrec} and \ref{AZexpint} are implemented in the package \ttfamily MultiIntegrate \rmfamily.
Suppose we want to compute the integral
$$
\int_{u_d}^{o_d} \cdots\int_{u_1}^{o_1}F[n;x_1, \dots, x_d] dx_1 \dots dx_d
$$
whrere $F[n;x_1, \dots, x_d]$ is of the form (\ref{mAZintegrand}) then the first method described in Section \ref{AZintegralrec} corresponds to the function call 
$$\text{\ttfamily mAZDirectIntegrate}[F[n;x_1, \dots, x_d],n,\{\{x_1,u_1,o_1\},\ldots,\{x_d,u_d,o_d\}\}] $$ 
whereas the second method corresponds to 
$$\text{\ttfamily mAZIntegrate}[F[n;x_1, \dots, x_d],n,\{\{x_1,u_1,o_1\},\ldots,\{x_d,u_d,o_d\}\}].$$
\begin{fmma2}
{
\In {\text{\bf mAZDirectIntegrate$\Big[\frac{(1 + x_1\; x_2)^n}{(1 + x_1)^\ep}$, n, \{\{$x_1$, 0, 1\}, \{$x_2$, 0, 1\}\}$\Big]$}}\\
\Out {\frac{\sum _{\iota _1=1}^n \frac{1}{-\ep +\iota _1+1}}{-n-1}-\frac{2^{1-\ep } \sum _{\iota _1=1}^n \frac{2^{\iota _1}}{\ep -\iota
    _1-1}}{n+1}+\frac{2^{-\ep } \left(2^{\ep }-2\right)}{(n+1) (\ep -1)}}\\
}
{
\In {\text{\bf mAZIntegrate$\Big[\frac{(1 + x_1\; x_2)^n}{(1 + x_1)^\ep}$, n, \{\{$x_1$, 0, 1\}, \{$x_2$, 0, 1\}\}$\Big]$}}\\
\Out {\frac{\sum _{\iota _1=1}^n \frac{1}{-\ep +\iota _1+1}}{-n-1}-\frac{2^{1-\ep } \sum _{\iota _1=1}^n \frac{2^{\iota _1}}{\ep -\iota
    _1-1}}{n+1}+\frac{2^{-\ep } \left(2^{\ep }-2\right)}{(n+1) (\ep -1)}}\\
}
\end{fmma2}
If more discrete variables \ie $n_1,\ldots,n_k$ are present one can use 
$$\text{\ttfamily mAZDirectIntegrate}[F[n_1,\ldots,n_k;x_1, \dots, x_d],\{n_1,\ldots,n_k\},\{\{x_1,u_1,o_1\},\ldots\}] $$ 
or
$$\text{\ttfamily mAZIntegrate}[F[n_1,\ldots,n_k;x_1, \dots, x_d],\{n_1,\ldots,n_k\},\{\{x_1,u_1,o_1\},\ldots\}].$$
\begin{fmma2}
{
\In {\text{\bf mAZIntegrate$\Big[\frac{(1 + x_1\; x_2)^{n_1}x_2^{n_2}}{(1 + x_1)^\ep}$, \{$n_1,n_2$\}, \{\{$x_1$, 0, 1\}, \{$x_2$, 0, 1\}\}$\Big]$}}\\
\Out {\frac{2^{-\ep } \left(n_2 2^{\ep }-2 n_2-2 \ep +2^{\ep }\right) \left(\prod _{\iota _1=1}^{n_1} \frac{\iota _1}{n_2+\iota_1+1}\right)}{\left(n_2+1\right) (\ep -1) \left(n_2+\ep \right)}
    +\frac{\left(\prod _{\iota _1=1}^{n_1} \frac{\iota _1}{n_2+\iota _1+1}\right) \sum_{\iota _1=1}^{n_1} \frac{\prod _{\iota _2=1}^{\iota _1} \frac{n_2+\iota _2+1}{\iota _2}}{-\ep +\iota _1+1}}{-n_2-\ep }
    -\frac{2^{1-\ep }\left(\prod _{\iota _1=1}^{n_1} \frac{\iota _1}{n_2+\iota _1+1}\right)}{n_2+\ep}
      \left(
	  \sum _{\iota _1=1}^{n_1} \frac{2^{\iota _1} \prod _{\iota _2=1}^{\iota _1} \frac{n_2+\iota_2+1}{\iota _2}}{n_2+\iota _1+1}
	 +\sum _{\iota_1=1}^{n_1} \frac{2^{\iota _1} \prod _{\iota _2=1}^{\iota _1} \frac{n_2+\iota _2+1}{\iota _2}}{\ep -\iota _1-1}
      \right)
    +\frac{1}{\left(n_2+1\right)\left(n_2+\ep \right)}}\\
}
\end{fmma2}
In order to find $\ep$-expansions of hyperexponential integrals (see Section \ref{AZexpint}) the function \ttfamily mAZExpandedIntegrate \rmfamily is provided.
Suppose we want to compute the coefficients $\ep^p,\ep^{p+1},\ldots,\ep^{q}$ of in the $\ep$-expansion of the integral
$$
\int_{u_d}^{o_d} \cdots\int_{u_1}^{o_1}F[\ep;n;x_1, \dots, x_d] dx_1 \dots dx_d
$$
we can use the function call
$$\text{\ttfamily mAZExpandedIntegrate}[F[\ep;n;x_1, \dots, x_d],n,\{\ep,p,q\},\{\{x_1,u_1,o_1\},\ldots,\{x_d,u_d,o_d\}\}].$$
\begin{fmma2}
{
\In {\text{\bf mAZExpandedIntegrate$\Big[\frac{(1 + x_1\; x_2)^n}{(1 + x_1)^\ep}$, n, \{$\ep$, 0, 2\}, \{\{$x_1$, 0, 1\}, \{$x_2$, 0, 1\}\}$\Big]$}}\\
\Out {\Bigg\{\Bigg\{\frac{\sum _{\iota _1=1}^n \frac{1}{\iota _1+1}}{-n-1}+\frac{2 \sum _{\iota _1=1}^n \frac{2^{\iota _1}}{\iota _1+1}}{n+1}+\frac{1}{n+1},-\frac{2
   \text{ln2} \sum _{\iota _1=1}^n \frac{2^{\iota _1}}{\iota _1+1}}{n+1}+\frac{\sum _{\iota _1=1}^n \frac{1}{\left(\iota _1+1\right){}^2}}{-n-1}+\frac{2 \sum _{\iota
   _1=1}^n \frac{2^{\iota _1}}{\left(\iota _1+1\right){}^2}}{n+1}+\frac{1-2 \text{ln2}}{n+1},\frac{\text{ln2}^2 \sum _{\iota _1=1}^n \frac{2^{\iota _1}}{\iota
   _1+1}}{n+1}-\frac{2 \text{ln2} \sum _{\iota _1=1}^n \frac{2^{\iota _1}}{\left(\iota _1+1\right){}^2}}{n+1}+\frac{\sum _{\iota _1=1}^n \frac{1}{\left(\iota
   _1+1\right){}^3}}{-n-1}+\frac{2 \sum _{\iota _1=1}^n \frac{2^{\iota _1}}{\left(\iota _1+1\right){}^3}}{n+1}+\frac{(\text{ln2}-1)^2}{n+1}\Bigg\},0,2\Bigg\}}\\
}
{
\In {\text{\bf $\%[[1]].\text{Table}\left[\ep ^i,\{i,\%[[2]],\%[[3]]\}\right]$}}\\
\Out {\ep ^2 \left(\frac{\text{ln2}^2 \sum _{\iota _1=1}^n \frac{2^{\iota _1}}{\iota _1+1}}{n+1}-\frac{2 \text{ln2} \sum _{\iota _1=1}^n \frac{2^{\iota
   _1}}{\left(\iota _1+1\right){}^2}}{n+1}+\frac{\sum _{\iota _1=1}^n \frac{1}{\left(\iota _1+1\right){}^3}}{-n-1}+\frac{2 \sum _{\iota _1=1}^n \frac{2^{\iota
   _1}}{\left(\iota _1+1\right){}^3}}{n+1}+\frac{(\text{ln2}-1)^2}{n+1}\right)+\ep  \left(-\frac{2 \text{ln2} \sum _{\iota _1=1}^n \frac{2^{\iota _1}}{\iota
   _1+1}}{n+1}+\frac{\sum _{\iota _1=1}^n \frac{1}{\left(\iota _1+1\right){}^2}}{-n-1}+\frac{2 \sum _{\iota _1=1}^n \frac{2^{\iota _1}}{\left(\iota
   _1+1\right){}^2}}{n+1}+\frac{1-2 \text{ln2}}{n+1}\right)+\frac{\sum _{\iota _1=1}^n \frac{1}{\iota _1+1}}{-n-1}+\frac{2 \sum _{\iota _1=1}^n \frac{2^{\iota
   _1}}{\iota _1+1}}{n+1}+\frac{1}{n+1}}\\
}
\end{fmma2}
Note that we can use the package \ttfamily HarmonicSums \rmfamily to rewrite the occuring sums in terms of S-sums.

We conclude this chapter with the following integral that occurs in the direct computation of a 3-loop diagram of the ladder-type which emerges for local quarkonic twist-2 operator matrix elements \cite{Ablinger2012b}.
\begin{eqnarray*}
&&\hspace*{-0.5cm}\int _0^1\int _0^1\left(\frac{(s (x-1)+t (u-1)+1)^{n}}{(w-1) (z-1) (s x-s+t u-t-u+1) (s x-s+t u-t-x+1)}\right.\\
&&+\frac{1}{(z-1) (s x-s+t u-t-x+1)}\\
&&\times\frac{(z (-s+t u-t+1)+x ((s-1) z+1))^{n}}{-s w x+s w+s x z-s z-t u w+t u z+t w-t z+u w-u-w-x z+x+z}\\
&&+\frac{1}{(w-1) (s x-s+t u-t-u+1)}\\
&&\times\left.\frac{(u ((t-1) w+1)-w (s (-x)+s+t-1))^{n}}{s w x-s w-s x z+s z+t u w-t u z-t w+t z-u w+u+w+x z-x-z}\right)dudx.
\end{eqnarray*}
All three summands of the integrand are of the form 
$$
\frac{(a+b x+c u)^n}{(d+e x+f u) (g+h x+i u)}.
$$
Therefore we solve the integral
$$
\int _0^1\int _0^1 \frac{(a+b x+c u)^n}{(d+e x+f u) (g+h x+i u)}dudx
$$
for general $a,b,c,d,e,f,g,h,i$ using \ttfamily mAZIntegrate\rmfamily:
\begin{fmma2}
{
\In {\text{\bf mAZIntegrate$\Big[\frac{(a+b x+c u)^n}{(d+e x+f u) (g+h x+i u)}$, n, \{\{$x$, 0, 1\}, \{$u$, 0, 1\}\}$\Big]$}}\\
\Out {\scriptstyle \frac{\left(\frac{-a e i+a f h+b d i-b f g-c d h+c e g}{f h-e i}\right)^n}{f h-e i}
  \Bigg(
   \log (d+e) \log \left(\frac{d f h-e f g}{e (d+e) i-e f (g+h)}\right)-\log (d+e+f) \log \left(\frac{f ((d+f) h-e (g+i))}{e (d+e) i-e f (g+h)}\right)+\log (d+f) \log
   \left(\frac{f (e (g+i)-(d+f) h)}{e (f g-d i)}\right)-\log (g+i) \log \left(\frac{i (e (g+i)-(d+f) h)}{h (f g-d i)}\right)-\log (d) \log \left(\frac{e f g-d f
   h}{e f g-d e i}\right)+\log (g) \log \left(\frac{e g i-d h i}{f g h-d h i}\right)-\log (g+h) \log \left(\frac{(e g-d h) i}{f h (g+h)-(d+e) h i}\right)+\log
   (g+h+i) \log \left(\frac{i (e (g+i)-(d+f) h)}{f h (g+h)-(d+e) h i}\right)+\text{Li}_2\left(\frac{(d+e) (e i-f h)}{e (d+e) i-e f
   (g+h)}\right)-\text{Li}_2\left(\frac{(d+e+f) (e i-f h)}{e (d+e) i-e f (g+h)}\right)+\text{Li}_2\left(\frac{(d+f) (f h-e i)}{e (f g-d
   i)}\right)-\text{Li}_2\left(\frac{(g+i) (f h-e i)}{h (f g-d i)}\right)-\text{Li}_2\left(\frac{d f h-d e i}{e f g-d e i}\right)+\text{Li}_2\left(\frac{f g h-e g
   i}{f g h-d h i}\right)-\text{Li}_2\left(\frac{(g+h) (f h-e i)}{f h (g+h)-(d+e) h i}\right)+\text{Li}_2\left(\frac{(g+h+i) (f h-e i)}{f h (g+h)-(d+e) h
   i}\right)+(\log (d+e)-\log (d)) \text{S}[1,\left\{\frac{\left(a-\frac{b d}{e}\right) (f h-e i)}{c e g-b f g-c d h+a f h+b d i-a e i}\right\},n]+(\log
   (d)-\log (d+f)) \text{S}[1,\left\{\frac{\left(a-\frac{c d}{f}\right) (f h-e i)}{c e g-b f g-c d h+a f h+b d i-a e i}\right\},n]+(\log (g)-\log (g+h))
   \text{S}[1,\left\{\frac{\left(a-\frac{b g}{h}\right) (f h-e i)}{c e g-b f g-c d h+a f h+b d i-a e i}\right\},n]+(\log (g+i)-\log (g))
   \text{S}[1,\left\{\frac{\left(a-\frac{c g}{i}\right) (f h-e i)}{c e g-b f g-c d h+a f h+b d i-a e i}\right\},n]+(\log (d+f)-\log (d+e+f))
   \text{S}[1,\left\{\frac{((a+c) e-b (d+f)) (f h-e i)}{e (c e g-b f g-c d h+a f h+b d i-a e i)}\right\},n]+(\log (d+e+f)-\log (d+e))
   \text{S}[1,\left\{\frac{((a+b) f-c (d+e)) (f h-e i)}{f (c e g-b f g-c d h+a f h+b d i-a e i)}\right\},n]+(\log (g+h+i)-\log (g+i)) \text{S}[1,\left\{\frac{(f
   h-e i) ((a+c) h-b (g+i))}{h (c e g-b f g-c d h+a f h+b d i-a e i)}\right\},n]+(\log (g+h)-\log (g+h+i)) \text{S}[1,\left\{\frac{((a+b) i-c (g+h)) (f h-e
   i)}{i (c e g-b f g-c d h+a f h+b d i-a e i)}\right\},n]
-\text{S}[1,1,\left\{\frac{\left(a-\frac{b d}{e}\right) (f h-e i)}{c e g-b f g-c d h+a f h+b d i-a ei},\frac{a e}{a e-b d}\right\},n]
+\text{S}[1,1,\left\{\frac{\left(a-\frac{b d}{e}\right) (f h-e i)}{c e g-b f g-c d h+a f h+b d i-a e i},\frac{(a+b) e}{a e-bd}\right\},n]
-\text{S}[1,1,\left\{\frac{((a+c) e-b (d+f)) (f h-e i)}{e (c e g-b f g-c d h+a f h+b d i-a e i)},\frac{(a+b+c) e}{(a+c) e-b(d+f)}\right\},n]
-\text{S}[1,1,\left\{\frac{\left(a-\frac{c d}{f}\right) (f h-e i)}{c e g-b f g-c d h+a f h+b d i-a e i},\frac{(a+c) f}{a f-cd}\right\},n]
+\text{S}[1,1,\left\{\frac{((a+b) f-c (d+e)) (f h-e i)}{f (c e g-b f g-c d h+a f h+b d i-a e i)},\frac{(a+b+c) f}{(a+b) f-c(d+e)}\right\},n]
-\text{S}[1,1,\left\{\frac{\left(a-\frac{b g}{h}\right) (f h-e i)}{c e g-b f g-c d h+a f h+b d i-a e i},\frac{(a+b) h}{a h-bg}\right\},n]
+\text{S}[1,1,\left\{\frac{((a+b) i-c (g+h)) (f h-e i)}{i (c e g-b f g-c d h+a f h+b d i-a e i)},\frac{(a+b) i}{(a+b) i-c(g+h)}\right\},n]
+\text{S}[1,1,\left\{\frac{\left(a-\frac{c g}{i}\right) (f h-e i)}{c e g-b f g-c d h+a f h+b d i-a e i},\frac{(a+c) i}{a i-cg}\right\},n]
+\text{S}[1,1,\left\{\frac{((a+c) e-b (d+f)) (f h-e i)}{e (c e g-b f g-c d h+a f h+b d i-a e i)},\frac{(a+c) e}{(a+c) e-b(d+f)}\right\},n]
+\text{S}[1,1,\left\{\frac{\left(a-\frac{c d}{f}\right) (f h-e i)}{c e g-b f g-c d h+a f h+b d i-a e i},\frac{a f}{a f-cd}\right\},n]
-\text{S}[1,1,\left\{\frac{((a+b) f-c (d+e)) (f h-e i)}{f (c e g-b f g-c d h+a f h+b d i-a e i)},\frac{(a+b) f}{(a+b) f-c(d+e)}\right\},n]
+\text{S}[1,1,\left\{\frac{\left(a-\frac{b g}{h}\right) (f h-e i)}{c e g-b f g-c d h+a f h+b d i-a e i},\frac{a h}{a h-bg}\right\},n]
-\text{S}[1,1,\left\{\frac{(f h-e i) ((a+c) h-b (g+i))}{h (c e g-b f g-c d h+a f h+b d i-a e i)},\frac{(a+c) h}{(a+c) h-b(g+i)}\right\},n]
+\text{S}[1,1,\left\{\frac{(f h-e i) ((a+c) h-b (g+i))}{h (c e g-b f g-c d h+a f h+b d i-a e i)},\frac{(a+b+c) h}{(a+c) h-b(g+i)}\right\},n]
-\text{S}[1,1,\left\{\frac{\left(a-\frac{c g}{i}\right) (f h-e i)}{c e g-b f g-c d h+a f h+b d i-a e i},\frac{a i}{a i-cg}\right\},n]
-\text{S}[1,1,\left\{\frac{((a+b) i-c (g+h)) (f h-e i)}{i (c e g-b f g-c d h+a f h+b d i-a e i)},\frac{(a+b+c) i}{(a+b) i-c(g+h)}\right\},n]
  \Bigg)}\\
}
\end{fmma2}
Using this general solution we arrive after simplification at the following solution for the whole integral:
\footnotesize
\begin{eqnarray*}
&&\frac{1}{(w-1) (z-1) (s+t-1)}\Biggl(
\text{S}_{1,1}\left(-\frac{(s+t-1) w (z-1)}{s w-s z+z-1},\frac{s w-s z+z-1}{z-1};n\right)\\
&&-\text{S}_{1,1}\left(-\frac{(s+t-1) w (z-1)}{s w-s z+z-1},\frac{(t-1) (s w-s z+z-1)}{(s+t-1) (z-1)};n\right)+\text{S}_{1,1}\left(\frac{s+t-1}{t-1},(1-t) w;n\right)\\
&&-\text{S}_{1,1}\left(-\frac{(s+t-1) w (z-1)}{s w-s z+z-1},\frac{z (s w-s z+z-1)}{w (z-1)};n\right)+\text{S}_{1,1}\left(\frac{s+t-1}{t-1},-\frac{(s-1) (t-1)}{s+t-1};n\right)\\
&&+\text{S}_{1,1}\left(-\frac{(s+t-1) w (z-1)}{s w-s z+z-1},\frac{(s w-s z+z-1) (t z-1)}{(s+t-1) w (z-1)};n\right)-\text{S}_{1,1}\left(\frac{s+t-1}{t-1},1-t;n\right)\\
&&+\text{S}_{1,1}\left(\frac{(s+t-1) (w-1) z}{(t-1) w-t z+1},\frac{-t w+w+t z-1}{w-1};n\right)-\text{S}_{1,1}\left(\frac{s+t-1}{t-1},-\frac{(t-1) (s w-1)}{s+t-1};n\right)\\
&&-\text{S}_{1,1}\left(\frac{(s+t-1) (w-1) z}{(t-1) w-t z+1},\frac{(s-1) (-t w+w+t z-1)}{(s+t-1) (w-1)};n\right)-\text{S}_{1,1}\left(\frac{s+t-1}{s-1},1-s;n\right)\\
&&-\text{S}_{1,1}\left(\frac{(s+t-1) (w-1) z}{(t-1) w-t z+1},-\frac{w ((t-1) w-t z+1)}{(w-1) z};n\right)-\text{S}_{1,1}\left(\frac{s+t-1}{s-1},-\frac{(s-1) (t z-1)}{s+t-1};n\right)\\
&&+\text{S}_{1,1}\left(\frac{(s+t-1) (w-1) z}{(t-1) w-t z+1},-\frac{(s w-1) ((t-1) w-t z+1)}{(s+t-1) (w-1) z};n\right)+\text{S}_{1,1}\left(\frac{s+t-1}{s-1},(1-s) z;n\right)\\
&&+\text{S}_{1,1}\left(\frac{s+t-1}{s-1},-\frac{(s-1)(t-1)}{s+t-1};n\right)-\text{S}_{1,1}(1,(1-s) z;n)+\text{S}_{1,1}(1,z-s z;n)+\text{S}_2((1-s) z;n)\\
&&-\text{S}_{1,1}(1,(1-t) w;n)
+\text{S}_{1,1}(1,w-t w;n)
-\text{S}_2((-s-t+1) w;n)
-\text{S}_2((-s-t+1) z;n)\\
&&+2 \text{S}_2(-s-t+1;n)
-\text{S}_2(1-s;n)
+\text{S}_2((1-t)w;n)
-\text{S}_2(1-t;n)\Biggr).
\end{eqnarray*}
\normalsize

\nocite{AlmZeilonedim}

\cleardoublepage
\phantomsection
\addcontentsline{toc}{chapter}{Bibliography}
\bibliographystyle{plain}

\cleardoublepage 
\includepdf{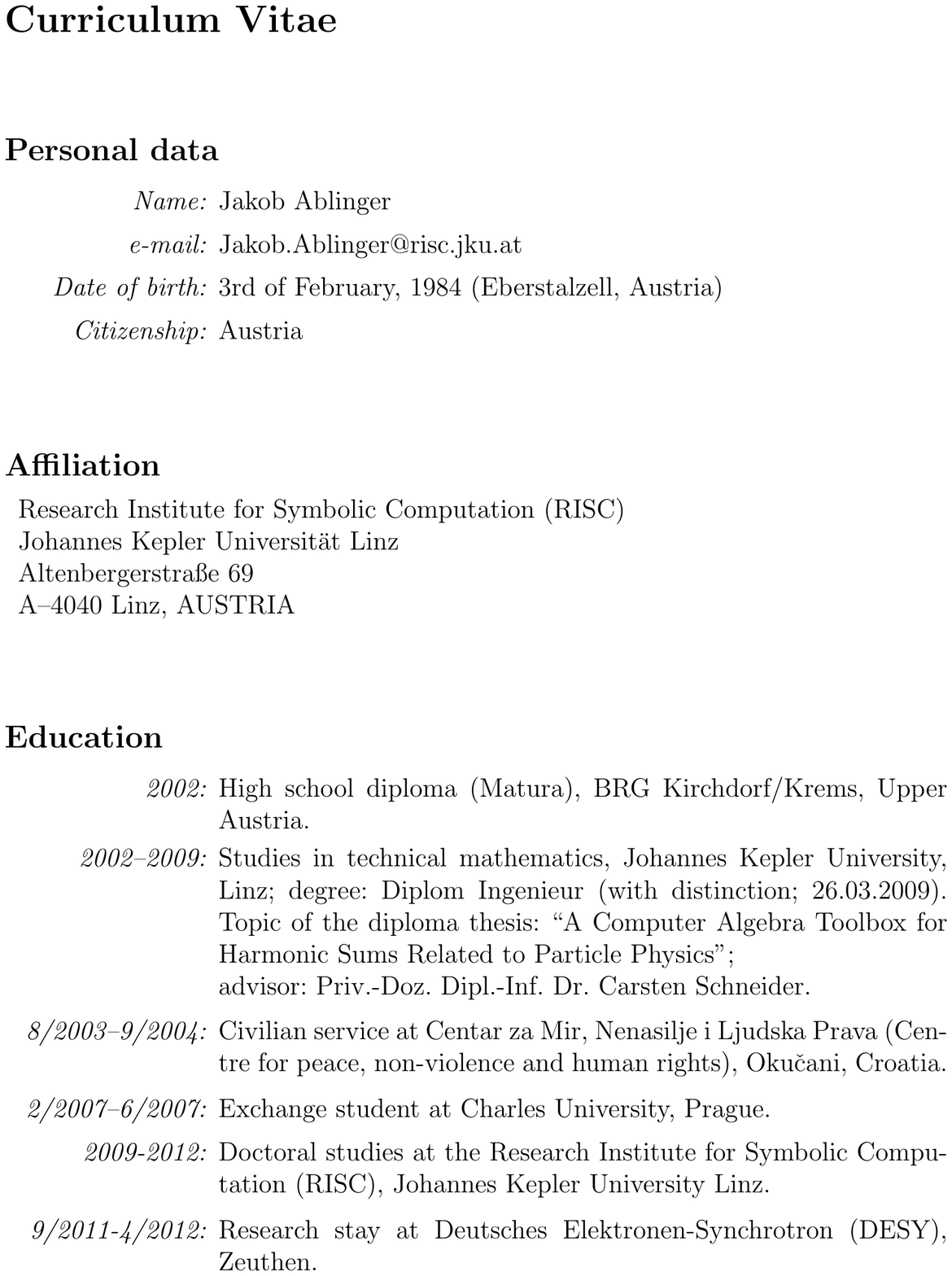}

\end{document}